\documentclass[letterpaper,longbibliography,11pt]{yalephd}

\usepackage{geometry} 
\usepackage{graphicx} 
\usepackage{dcolumn}
\usepackage{bm}
\usepackage{amsmath}
\usepackage{amsfonts}
\usepackage{amssymb}
\usepackage{appendix}
\usepackage{comment}
\usepackage{cite}
\usepackage{gensymb}

\usepackage{enumitem}

\usepackage{boldline}
\usepackage{arydshln}

\usepackage[sort&compress,numbers]{natbib}
\usepackage{doi}
\usepackage{hyperref}

\usepackage[dvipsnames]{xcolor}

\newcommand\myshade{85}

\colorlet{mylinkcolor}{violet}
\colorlet{mycitecolor}{Aquamarine}
\colorlet{myurlcolor}{YellowOrange}

\hypersetup{
  linkcolor  = mylinkcolor!\myshade!black,
  citecolor  = mycitecolor!\myshade!black,
  urlcolor   = myurlcolor!\myshade!black,
  colorlinks = true,
}

\usepackage{amsthm}
\newtheoremstyle{mystyle}
  {}                                      
  {}                                      
  {\itshape}                              
  {}                                      
  {\bfseries}                             
  {.}                                     
  { }                                     
  {\thmname{#1}\thmnumber{ #2}\thmnote{ (#3)}}%

\theoremstyle{mystyle}
\newtheorem{theorem}{Theorem}
\newtheorem{definition}[theorem]{Definition} 
 
\newtheorem{lemma}[theorem]{Lemma}

\newcommand{\nth}{\bar{n}_{\scriptsize\textrm{th}}}
\newcommand{\nthtiny}{\bar{n}_{\tiny\textrm{th}}}

\newcommand{\reg}{\scriptsize\mathrm{reg}}

\begin{document}

\title{Quantum Computation and Communication in Bosonic Systems}
\author{Kyungjoo Noh}
\advisor{Liang Jiang}
\date{May 2020} 

\frontmatter

\begin{abstract}
Quantum computation and communication are important branches of quantum information science. However, noise in realistic quantum devices fundamentally limits the utility of these quantum technologies. A conventional approach towards large-scale and fault-tolerant quantum information processing is to use multi-qubit quantum error correction (QEC), that is, to encode a logical quantum bit (or a logical qubit) redundantly over many physical qubits such that the redundancy can be used to detect errors. The required resource overhead associated with the use of conventional multi-qubit QEC schemes, however, is too high for these schemes to be realized at scale with currently available quantum devices. Recently, bosonic (or continuous-variable) quantum error correction has risen as a promising hardware-efficient alternative to multi-qubit QEC schemes. 

In this thesis, I provide an overview of bosonic QEC and present my contributions to the field. Specifically, I present the benchmark and optimization results of various single-mode bosonic codes against practically relevant excitation loss errors. I also demonstrate that fault-tolerant bosonic QEC is possible by concatenating a single-mode bosonic code with a multi-qubit error-correcting code. Moreover, I discuss the fundamental aspects of bosonic QEC using the framework of quantum communication theory. In particular, I present improved bounds on important communication-theoretic quantities such as the quantum capacity of bosonic Gaussian channels. Furthermore, I provide explicit bosonic error correction schemes that nearly achieve the fundamental performance limit set by the quantum capacity. I conclude the thesis with discussions on the importance of non-Gaussian resources for continuous-variable quantum information processing.     
\end{abstract}

\maketitle
\makecopyright{2020} 
\tableofcontents
\listoffigures 
\listoftables 

\clearpage
\begin{center}
    \thispagestyle{empty}
    \vspace*{\fill}
    \large{\textit{To my partner Sunnie S. Y. Kim}}
    \vspace*{\fill}
\end{center}
\clearpage

\chapter{Acknowledgments} 

I would like to thank my advisor Liang Jiang for his guidance throughout my PhD studies. His broad research interests have allowed me to explore diverse fields in quantum information science and tackle various research questions from many different angles. I also thank him for being flexible with me freely pursuing what interests me the most.  

I am greatly indebted to my partner Sunnie S. Y. Kim for her incredible support throughout my PhD journey. Doing research has sometimes been unavoidably stressful, but it has never been difficult thanks to her love and support. In fact, she has made doing a PhD so much fun that I am even willing to repeat it infinitely many times if it is with her. (Of course, I would not otherwise!)

I acknowledge all the collaborators who have helped me learn new things. Among all, I want to especially thank Victor V. Albert and Christopher Chamberland. If it were not for them, I would not haven known what I now know about bosonic quantum error correction and fault-tolerant quantum computing.

I would also like to thank all members of the YQI community. Professor Steve Girvin has always been supportive and his insights and guidance have been instrumental in my research. All the exciting experiments from the groups of Professors Michel Devoret and Robert Schoelkopf have inspired me, let me stay motivated, and keep me grounded in reality. Being around such amazing and enthusiastic experimentalists is not a privilege that every theorist enjoys. I also want to thank all members of Liang's group for creating a welcoming and collaborative work environment. Special thanks go to Professors Meng Cheng, Michel Devoret, Steve Girvin, and Arne Grimsmo for serving as my thesis committee members.

Lastly, I would like to express my sincere gratitude to my parents for all their hard work that allowed me to enjoy all the opportunities that they did not necessarily have themselves.

\mainmatter

\chapter{Introduction and motivation}

\section{The field of study}

\subsection{Quantum computation}
\label{subsection:Quantum computation}

Quantum computers \cite{Nielsen2000} are a fundamentally different kind of computers than conventional computers as they take advantage of the unique quantum mechanical properties such as quantum superposition, interference, and quantum entanglement to process information. As an example, while the integer factorization problem is believed not to be solvable in polynomial time (in the size of the input) by using classical computers, quantum computers can factor large integers efficiently in polynomial time by using the Shor's factoring algorithm \cite{Shor1994}. Quantum computers can thus have a significant impact on the field of cryptography since the security of RSA encryption \cite{Rivest1978}, a widely used cryptographic method, is based on the assumption that factoring large integers is practically impossible. This assumption is not valid any more if reliable quantum computers can be built. 

In addition to factoring large integers, quantum computers can efficiently simulate the real-time dynamics of large quantum systems \cite{Lloyd1996}. Since various physical, chemical, and biological phenomena are inherently quantum, quantum computers can be used to simulate these various natural phenomena at scale. Thus, quantum computers can help us make new scientific discoveries in a more guided manner.   


\subsection{Quantum communication}
\label{subsection:Quantum communication}

Quantum communication \cite{Holevo2012,Wilde2013,Hayashi2016,Watrous2018} is another branch of quantum information science wherein the unique features of quantum mechanics can be harnessed to achieve tasks that are otherwise unachievable. As an example, quantum key distribution (QKD) \cite{Bennett1984,Pirandola2019} allows secure classical communication where the security is guaranteed by the validity of the laws of quantum mechanics, not by the computational intractability of certain mathematical problems. Hence, quantum communication provides an alternative cryptographic solution to RSA encryption that is secure against attacks by quantum computers.

More generally, quantum communication theory has a richer structure than classical communication theory. This is because in quantum communication, we can also consider transmitting a quantum bit or a qubit instead of a usual classical bit. Moreover, quantum entanglement plays an important role in quantum communication theory as it has an interesting interplay with quantum and classical information transmission via the quantum super-dense coding protocol \cite{Bennett1992} and the quantum teleportation protocol \cite{Bennett1993}.   

Quantum communication also plays an essential role in the broader context of quantum information science, including quantum computation. To be more specific, while a single quantum processor may be able to perform an interesting computational task that is intractable even by the most powerful conventional computer, it may not support sufficiently many qubits that are needed for a useful quantum computing application. In this case, it will be essential to connect distant quantum processors via quantum communication such that multiple quantum processors can be operated in concert. Hence, quantum communication goes hand in hand with quantum computation.

\section{A brief historical context}

\subsection{Conventional approach towards fault-tolerance}

While ideal quantum computers can efficiently factor large integers and simulate large quantum systems, realistic quantum devices are noisy and thus do not produce a reliable computational outcome. Also, noise in realistic quantum communication channels corrupts transmitted information. Thus, quantum error correction (QEC) \cite{Shor1995,Gottesman1997} is absolutely essential for realizing reliable and large-scale quantum information processing. 

For the past two decades, there has been significant theoretical progress in quantum error correction. In particular, it has been established that fault-tolerant quantum information processing is possible if the noise strength is below a certain threshold \cite{Shor1996,Gottesman2009}. The conventional approach towards fault-tolerance is to encode an error-corrected logical qubit redundantly over many qubits such that the redundancy can be used to detect errors. 

Topological quantum error-correcting codes \cite{Bravyi1998,Dennis2002,Bombin2006,Bombin2007,Fowler2012} such as the surface code and the color code are leading candidates for achieving fault-tolerance since they can be implemented by using only nearest-neighbor interactions and have relatively high fault-tolerance thresholds. However, the resource overhead associated with the use of these conventional approaches is too high for these schemes to be realized at scale with currently available quantum devices. For instance, it is estimated that roughly $10^{3}$ qubits are needed to encode a single reliable logical qubit given a physical error rate $p\sim 10^{-3}$ \cite{Fowler2012}. On the other hand, state-of-the-art quantum processors currently support about $50$ qubits with a physical error rate $p\sim 5\times 10^{-3}$ \cite{Arute2019}.    

While building a fully fault-tolerant quantum computer is still a distant goal, there has been significant progress in the experimental realization of QEC. In particular, it has been demonstrated that the lifetime of a qubit can be extended by using a quantum error-detecting code and post-selection \cite{Vuillot2017,Linke2017,Kraglund2019}. Moreover, there have been various theoretical proposals such as flag-qubit schemes for reducing the resource overhead required for fault-tolerant quantum error correction \cite{Chamberland2018a,Chao2018,Chao2018b,Chamberland2019,
Chao2019,Chamberland2020,Chamberland2020a}. Despite these recent progress, however, achieving fault-tolerance using is still very challenging.

\subsection{Bosonic quantum error correction (focus of the thesis)}

\begin{figure}[t!]
\centering
\includegraphics[width=5.5in]{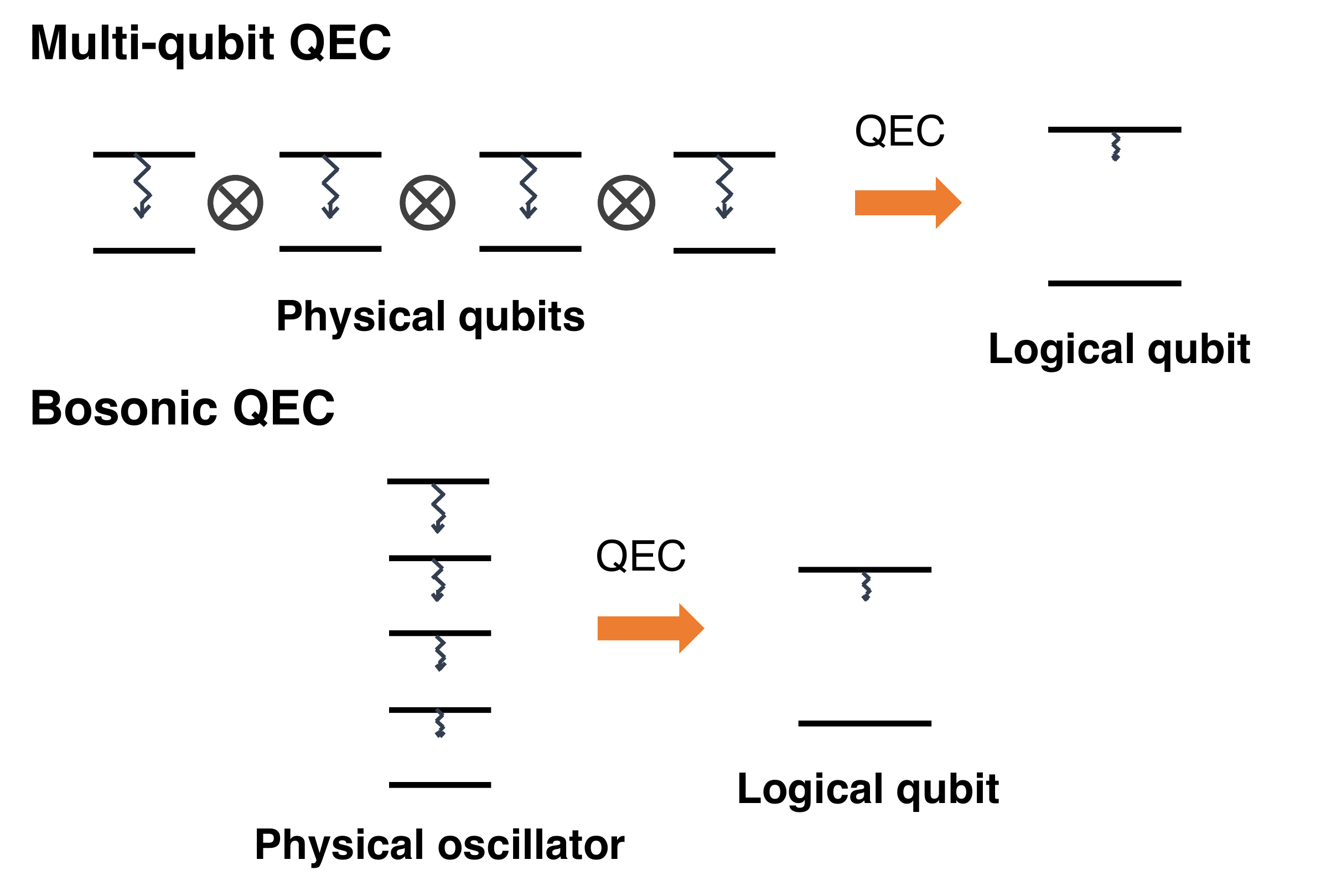}
\caption{Schematic illustration of multi-qubit QEC and bosonic QEC. In multi-qubit QEC, a logical qubit is redundantly encoded over multiple qubits. In bosonic QEC, on the other hand, a logical qubit is redundantly encoded by using multiple levels in a single bosonic mode.  }
\label{fig:bosonic QEC}
\end{figure}

Bosonic QEC \cite{Albert2018} has recently risen as a hardware-efficient alternative to the conventional multi-qubit QEC. Bosonic QEC takes advantage of the fact that even a single bosonic mode consists of infinitely many quantum states. In bosonic QEC, an error-corrected logical qubit is redundantly encoded by using multiple levels in a single bosonic mode (see Fig.\ \ref{fig:bosonic QEC} for an illustration). Because a single bosonic mode is sufficient for providing redundancy needed for detecting relevant errors, bosonic QEC is hardware efficient. Indeed, it was demonstrated that it is possible to extend the lifetime of a qubit by using a bosonic cat code \cite{Ofek2016}. In particular, only a single bosonic mode, an ancillary transmon qubit, and a readout cavity mode were needed to achieve the break-even error correction performance.  

Bosonic QEC is also relevant to quantum communication. This is because quantum communication is typically implemented by using light modes, which are described by an infinite-dimensional bosonic Hilbert space. In particular, it was demonstrated that an error-correctable bosonic code can be sent through a microwave quantum communication channel in a circuit QED system \cite{Axline2018}. Such a long-distance quantum state transfer can be used to connect distant quantum processors. 

\section{Summary and reading guide}

In Chapter \ref{chapter:Bosonic quantum error correction}, I provide a pedagogic review of bosonic quantum error correction. I discuss various error models that are relevant to realistic bosonic systems. Also, various single-mode bosonic codes such as cat, binomial, and GKP codes are reviewed. The rest of the thesis is not based on the cat and the binomial codes but they are reviewed for completeness and their experimental relevance. The square-lattice GKP code will be reviewed comprehensively in Subsection \ref{subsection:The square-lattice Gottesman-Kitaev-Preskill (GKP) code} as they are referenced in Chapters \ref{chapter:Benchmarking and optimizing single-mode bosonic codes} and \ref{chapter:Fault-tolerant bosonic quantum error correction}. Generalizations of the GKP codes to a more general lattice structure and multi-mode cases are reviewed in Subsection \ref{subsection:Generalization of GKP codes}. Among these generalizations, the hexagonal-lattice GKP code will be referenced in Chapter \ref{chapter:Benchmarking and optimizing single-mode bosonic codes} and the symplectic lattice codes will be used in Chapter \ref{chapter:Achievable communication rates with bosonic codes}.  

In Chapter \ref{chapter:Benchmarking and optimizing single-mode bosonic codes}, I provide benchmarking results for the performance of various single-mode bosonic codes against practically relevant excitation loss errors. It turns out that the GKP code families outperform many other bosonic code families in correcting excitation loss errors, despite the fact the GKP code families are not specifically designed to correct loss errors. Furthermore, I apply a biconvex optimization technique to search for an optimal single-mode bosonic code for correcting excitation loss errors. Surprisingly again, the hexagonal-lattice GKP code emerges as an optimal code from many independent random Haar initial codes. To explain the excellent, if not optimal, performance of the GKP codes against excitation loss errors, I provide an explicit decoding strategy for the GKP codes against loss errors. Fault-tolerance is not addressed in Chapter \ref{chapter:Benchmarking and optimizing single-mode bosonic codes}. That is, I assume that our attempts to correct errors are noiseless. I adopt this idealized assumption in the interest of simplicity and to focus on the intrinsic error-correcting capabilities of various bosonic codes.  

In Chapter \ref{chapter:Fault-tolerant bosonic quantum error correction}, I consider more realistic situations where even our attempts to correct for errors can be erroneous. In particular, motivated by the excellent performance of the GKP code in the idealized situation, I focus on the fault-tolerance properties of the GKP code. One of the main messages is that fault-tolerant bosonic quantum error correction is possible with the surface-GKP code, i.e., concatenation of the GKP code with the surface code. Also, I demonstrate that the additional information gathered during the GKP error correction can be used to significantly boost the performance of the outer multi-qubit surface code. 

In Chapter \ref{chapter:Quantum capacities of bosonic Gaussian channels}, I consider fundamental aspects of bosonic quantum error correction. More specifically, I study an important quantum communication-theoretic quantity, i.e., the quantum capacity of bosonic Gaussian channels, as it determines the fundamental performance limits of bosonic codes. In Section \ref{section:Upper bound of the Gaussian thermal-loss channel capacity}, I provide an improved upper bound of the Gaussian thermal-loss channel capacity based on a data-processing argument. In Section \ref{section:Lower bound of the Gaussian thermal-loss channel capacity}, I provide an improved lower bound of the energy-constrained quantum capacity of Gaussian thermal-loss channels and show that higher quantum state transmission rates can be achieved than previously believed. By doing so, I prove that Gaussian thermal-loss channels are superadditive with respect to Gaussian input states. 

In Chapter \ref{chapter:Achievable communication rates with bosonic codes}, I study explicit bosonic quantum error correction schemes in the context of quantum communication. Specifically, I investigate the achievable quantum state transmission rates of multi-mode symplectic lattice codes and of numerically optimized single-mode codes. I show they these codes nearly achieve the fundamental limits set by the quantum capacity which are studied in Chapter \ref{chapter:Quantum capacities of bosonic Gaussian channels}.   

In Chapter \ref{chapter:Non-Gaussian resources for bosonic quantum information processing}, I discuss the importance of non-Gaussian resources in continuous-variable quantum information processing. In Section \ref{section:GKP state as a non-Gaussian resource}, I construct a non-Gaussian oscillator-into-oscillators code, namely the GKP-two-mode-squeezing code, and show that the GKP states can be used as a valuable non-Gaussian resource to enable error-corrected bosonic quantum information processing. In Section \ref{section:Cubic phase state as a non-Gaussian resource}, I consider cubic phase states which are analogous to magic states for multi-qubit quantum information processing. In particular, I address the question of whether noisy cubic phase states can be distilled by using only Gaussian states, operations, and homodyne measurements. I show that direct translations of the conventional magic state distillation schemes to bosonic systems do not work.

In Appendix \ref{appendix:Gaussian states, unitaries, and channels}, fundamentals of bosonic systems and Gaussian operations are reviewed. The concepts such as Gaussian states, unitaries, channels, and measurements are used throughout the thesis. In particular, the material in Chapter \ref{chapter:Quantum capacities of bosonic Gaussian channels} relies heavily on the materials in Appendix \ref{appendix:Gaussian states, unitaries, and channels}.  

In Chapters \ref{chapter:Benchmarking and optimizing single-mode bosonic codes}--\ref{chapter:Non-Gaussian resources for bosonic quantum information processing}, I discuss open questions and possible future research directions at the end of each chapter.  

\chapter{Bosonic quantum error correction}
\label{chapter:Bosonic quantum error correction}

In this chapter, I will provide a pedagogic review of bosonic quantum error correction. The main goal of this chapter is to introduce various bosonic codes (i.e., cat, binomial, and GKP codes) and study their error-correcting capabilities against practically relevant errors.   

In Section \ref{section:Fundamentals of quantum error correction}, I will briefly review fundamentals of quantum error correction to introduce the notation and terminology. In Section \ref{section:Relevant error models in bosonic systems}, I will introduce several error models for bosonic systems, namely, excitation loss errors, Gaussian random shift errors in the phase space, and bosonic dephasing errors which are either experimentally relevant or theoretically important for understanding certain bosonic error-correcting codes. In Section \ref{section:Rotation-symmetric bosonic codes}, I will review rotation-symmetric bosonic codes such as cat and binomial codes. In Section \ref{section:Translation-symmetric bosonic codes}, I will review translation-symmetric bosonic codes, namely, Gottesman-Kitaev-Preskill (GKP) codes.  

\section{Fundamentals of quantum error correction}
\label{section:Fundamentals of quantum error correction}

In this section, we will briefly review fundamentals of quantum error correction (QEC). The main goal of this section is to introduce notations and terminologies rather than to give a comprehensive introduction to QEC. For a more through introduction to QEC, see for example Ref.\ \cite{Nielsen2000}.    

\subsection{Error-correcting code, correctable error set, and recovery}

An error-correcting code $\mathcal{C} = \textrm{span}\lbrace |\mu_{L}\rangle\rbrace_{\mu=0,\cdots, d-1}$ is a subspace of a physical Hilbert space $\mathcal{H}$. The orthonormalized basis states $|\mu_{L}\rangle$ of the code space $\mathcal{C}$ are called the logical codewords. Ideally, an error-correcting code $\mathcal{C}$ should be robust against relevant errors in the physical Hilbert space. A code space $\mathcal{C}$ can also be uniquely identified by using the projection operator to the code space: 
\begin{align}
\hat{P}_{\mathcal{C}} &\equiv \sum_{\mu=0}^{d-1}|\mu_{L}\rangle\langle\mu_{L}|. 
\end{align}

\begin{definition}[Correctable error set]
Consider a completely-positive noise map $\mathcal{N}(\hat{\rho}) = \sum_{k}\hat{N}_{k}\hat{\rho}\hat{N}_{k}^{\dagger}$. The error set $\lbrace \hat{N}_{k} \rbrace$ is said to be correctable by a code $\mathcal{C}$ if there exists a CPTP recovery map $\mathcal{R}$ such that 
\begin{align}
\mathcal{R}\cdot \mathcal{N} (\hat{\rho}) \propto \hat{\rho}, 
\end{align}  
for all density matrices $\hat{\rho} \in \mathcal{D}(\mathcal{C}) = \lbrace \hat{\rho}\in \mathcal{L}(\mathcal{C}) | \hat{\rho}^{\dagger}=\hat{\rho} \succeq 0,\mathrm{Tr}[\hat{\rho}]=1 \rbrace$. Here, $\mathcal{L}(\mathcal{C})$ is the space of linear operators on the code space $\mathcal{C}$. 
\end{definition}

Note that we used the proportionality $\propto$ instead of the equality $=$ in the definition. This is because the noise map $\mathcal{N}$ may not be trace-preserving. For example, we may take only the first few leading Kraus operators of an entire CPTP noise channel $\mathcal{N}'$ to define $\mathcal{N}$. In this case the noise map $\mathcal{N}$ may be trace-decreasing. However, the recovery map $\mathcal{R}$ has to be trace-preserving because here we only consider error correction protocols that succeed with probability $1$. Below, we review the Knill-Laflamme condition \cite{Knill1997} that allows us to directly check whether the error set $\lbrace \hat{N}_{k} \rbrace$ is correctable by the code space $\mathcal{C}$ or not. 

\subsection{Knill-Laflamme condition}

\begin{theorem}[Knill-Laflamme condition \cite{Knill1997}]
An error set $\lbrace \hat{N}_{k} \rbrace$ is correctable by a code $\mathcal{C}$ if and only if the following Knill-Laflamme condition is satisfied: 
\begin{align}
\hat{P}_{\mathcal{C}} \hat{N}_{k}^{\dagger}\hat{N}_{k'}\hat{P}_{\mathcal{C}}  = \alpha_{kk'}\hat{P}_{\mathcal{C}}
\end{align}
for all $k,k'$, where $\hat{P}_{\mathcal{C}}$ is the projection operator to the code space and $\alpha_{kk'}$ are the elements of a complex hermitian matrix $\alpha$.  
\label{theorem:knill-laflamme condition}
\end{theorem}

The proof of Theorem \ref{theorem:knill-laflamme condition} is given in, e.g., Refs.\ \cite{Knill1997,Nielsen2000}. The Knill-Laflamme condition will be frequently used in the following sections to probe the error-correcting capabilities of various bosonic codes. Note that the Knill-Laflamme condition can also be expressed as
\begin{align}
\langle \mu_{L} | \hat{N}_{k}^{\dagger}\hat{N}_{k'} |\nu_{L}\rangle = \alpha_{kk'}\delta_{\mu\nu}
\end{align} 
for all $k,k'$ and $\mu,\nu$, where $|\mu_{L}\rangle$ and $|\nu_{L}\rangle$ are orthonormal basis states of the code space $\mathcal{C}$.       

\section{Relevant error models in bosonic systems}
\label{section:Relevant error models in bosonic systems}

Here, we review three important error models for bosonic systems, namely, excitation loss errors, Gaussian random shift errors, and bosonic dephasing errors.  See Tables \ref{table:excitation loss errors}, \ref{table:random shift errors}, and \ref{table:bosonic dephasing errors} for a summary.  

\subsection{Excitation loss errors}
\label{subsection:Excitation loss errors}

In this subsection, we consider excitation loss errors which are ubiquitous in many realistic bosonic systems. Specifically, we will discuss four different ways to describe excitation loss errors (see Table \ref{table:excitation loss errors}). 

\begin{table}[h!]
  \centering
     \def\arraystretch{2}
  \begin{tabular}{ V{3} c V{1.5} c V{3} }
   \hlineB{3}  
     \textbf{Representation} & \textbf{Excitation loss errors}  \\  \hlineB{3} 
      Lindblad equation & $\frac{d\hat{\rho}(t)}{dt} = \kappa \mathcal{D}[\hat{a}](\hat{\rho}(t)) =\kappa \big{(} \hat{a}\hat{\rho}(t)\hat{a}^{\dagger}-\frac{1}{2}\hat{a}^{\dagger}\hat{a}\hat{\rho}(t)-\frac{1}{2}\hat{\rho}(t)\hat{a}^{\dagger}\hat{a} \big{)} $  \\   
      (Eqs.\ \eqref{eq:excitation loss Lindbladian equation}, \eqref{eq:dissipation superoperator}) & $\rightarrow \hat{\rho}(t) =  e^{\kappa t \mathcal{D}[\hat{a}]}  \hat{\rho}(0)$.    \\  \hlineB{1.5}    
      Kraus representation & $\hat{\rho}(t) =  e^{\kappa t \mathcal{D}[\hat{a}]}  \hat{\rho}(0) = \sum_{\ell=0}^{\infty} \hat{N}_{\ell}\hat{\rho}(0)\hat{N}_{\ell}^{\dagger}$ where  \\ 
       (Eqs.\ \eqref{eq:excitation loss CPTP map explicit}, \eqref{eq:excitation loss Kraus operators})& $\hat{N}_{\ell} \equiv \sqrt{ \frac{(1-e^{-\kappa t})^{\ell}}{\ell!} } e^{-\frac{\kappa t}{2} \hat{n} } \hat{a}^{\ell}$.   \\ \hlineB{1.5}   
       Heisenberg picture & $\hat{q}(t) = \hat{q}e^{-\frac{\kappa t}{2}}$, $\hat{p}(t) = \hat{p}e^{-\frac{\kappa t}{2}}$,  \\ 
       (Eq.\ \eqref{eq:excitation loss Heisenberg picture summarized})  & $\hat{q}^{2} (t)  = \hat{q}^{2} e^{-\kappa t} + \frac{1}{2}(1-e^{-\kappa t})$,   \\ 
        & $\frac{1}{2}( \hat{q}\hat{p} + \hat{p}\hat{q} )(t) = \frac{1}{2}( \hat{q}\hat{p} + \hat{p}\hat{q} )e^{-\kappa t}$,   \\ 
        & $\hat{p}^{2} (t)  = \hat{p}^{2} e^{-\kappa t} + \frac{1}{2}(1-e^{-\kappa t})$.  \\ \hlineB{1.5}   
        Gaussian channels & $e^{\kappa t \mathcal{D}[\hat{a}]} = \mathcal{N}[\eta = e^{-\kappa t}, 0]$  \\  
      (Eq.\ \eqref{eq:excitation loss errors equal pure loss channels})   & $ \leftrightarrow (\boldsymbol{T},\boldsymbol{N},\boldsymbol{d}) = ( e^{-\frac{\kappa t}{2}} \boldsymbol{I}_{2}, \frac{1}{2}(1-e^{-\kappa t})\boldsymbol{I}_{2} , 0 )$.  \\ \hlineB{3}   
  \end{tabular}
  \caption{Various representations of the excitation loss errors.}
  \label{table:excitation loss errors}
\end{table}

\subsubsection{Lindblad equation}

Realistic bosonic modes are typically subject to excitation loss errors which are described by the following Lindblad equation: 
\begin{align}
\frac{d\hat{\rho}(t)}{dt} = \kappa \mathcal{D}[\hat{a}](\hat{\rho}(t)). \label{eq:excitation loss Lindbladian equation}
\end{align}
Here, $\hat{a}$ is the annihilation operator of the bosonic mode and $\mathcal{D}[\hat{A}]$ is the dissipation superoperator (mapping a density matrix to another density matrix) defined as 
\begin{align}
\mathcal{D}[\hat{A}](\hat{\rho}) \equiv \hat{A}\hat{\rho}\hat{A}^{\dagger} -\frac{1}{2}\lbrace \hat{A}^{\dagger}\hat{A} , \hat{\rho} \rbrace. \label{eq:dissipation superoperator}
\end{align}
$\lbrace \hat{A} , \hat{B} \rbrace \equiv \hat{A}\hat{B} + \hat{B}\hat{A}$ is the anti-commutator. By solving Eq.\ \eqref{eq:excitation loss Lindbladian equation}, we get a completely-positive and trace-preserving (CPTP) map \cite{Choi1975} 
\begin{align}
\hat{\rho}(0) \rightarrow \hat{\rho}(t) =  e^{\kappa t \mathcal{D}[\hat{a}]}  \hat{\rho}(0) .  \label{eq:excitation loss CPTP map formal}
\end{align}
While being concise, the expression in Eq.\ \eqref{eq:excitation loss CPTP map formal} does not provide any useful information about how to evaluate $\hat{\rho}(t)$ given an initial density matrix $\hat{\rho}(0)$. 

\subsubsection{Kraus representation}

A more useful way to represent the CPTP map in Eq.\ \eqref{eq:excitation loss CPTP map formal} is to use the Kraus representation. Here, we will show that the CPTP map $e^{\kappa t \mathcal{D}[\hat{a}]}$ is explicitly given by the Kraus form  
\begin{align}
\hat{\rho}(t) =  e^{\kappa t \mathcal{D}[\hat{a}]}  \hat{\rho}(0) = \sum_{\ell=0}^{\infty} \hat{N}_{\ell}\hat{\rho}(0)\hat{N}_{\ell}^{\dagger}, \label{eq:excitation loss CPTP map explicit}
\end{align}
where the Kraus operators are given by 
\begin{align}
\hat{N}_{\ell} &\equiv \sqrt{ \frac{(1-e^{-\kappa t})^{\ell}}{\ell!} } e^{-\frac{\kappa t}{2} \hat{n} } \hat{a}^{\ell}. \label{eq:excitation loss Kraus operators}
\end{align}
Here, $\hat{n}\equiv \hat{a}^{\dagger}\hat{a}$ is the excitation number operator. The derivation of Eqs.\ \eqref{eq:excitation loss CPTP map explicit} and \eqref{eq:excitation loss Kraus operators} is given below. Note that the Kraus operator $\hat{N}_{\ell}$ corresponds to an $\ell$-excitation loss event. Often times, the decay term $e^{-\frac{\kappa t}{2} \hat{n} }$ is referred to as the no-jump evolution term and the excitation-number-decreasing term $\hat{a}^{\ell}$ is referred to as the jump term. These terminologies are motivated by the quantum trajectory picture which we use to derive Eqs.\ \eqref{eq:excitation loss CPTP map explicit} and \eqref{eq:excitation loss Kraus operators} below.   

Recall the Lindblad equation in Eq.\ \eqref{eq:excitation loss Lindbladian equation}: 
\begin{align}
\frac{d\hat{\rho}(t)}{dt} = \kappa \hat{a}\hat{\rho}(t)\hat{a}^{\dagger} - \frac{\kappa}{2}\hat{n}\hat{\rho}(t) - \frac{\kappa}{2}\hat{\rho}(t)\hat{n}.  \label{eq:excitation loss Lindbladian equation more explicit}
\end{align}
The first term on the right hand side is referred to as the jump term and the second and the third terms are referred to as the back-action terms. The back-action terms can be intuitively understood as terms that are generated by a non-hermitian Hamiltonian $\hat{H} = - i\kappa\hat{n}/2$. Then, it is useful to define an interaction picture that takes the back-action terms as the unperturbed Lindbladian and the jump term as the perturbative Lindbladian. That is, we define 
\begin{align}
\hat{\rho}_{I}(t) \equiv e^{\frac{\kappa t}{2} \hat{n} } \hat{\rho}(t) e^{\frac{\kappa t}{2} \hat{n} }. 
\end{align} 
Then, the density matrix in the interaction picture $\hat{\rho}_{I}(t)$ evolves under the following equation:   
\begin{align}
\frac{d\hat{\rho}_{I}(t)}{dt} &= e^{\frac{\kappa t}{2} \hat{n} } \Big{[} \frac{\kappa}{2}\hat{n} \hat{\rho}(t) + \frac{d\hat{\rho}(t)}{dt}  + \frac{\kappa}{2}\hat{\rho}(t)\hat{n} \Big{]} e^{\frac{\kappa t}{2} \hat{n} }
\nonumber\\
&= \kappa e^{\frac{\kappa t}{2} \hat{n} } \hat{a} \hat{\rho}(t) \hat{a}^{\dagger}  e^{\frac{\kappa t}{2} \hat{n} }
\nonumber\\
&= \kappa e^{\frac{\kappa t}{2} \hat{n} } \hat{a} e^{-\frac{\kappa t}{2} \hat{n} } \hat{\rho}_{I}(t) e^{-\frac{\kappa t}{2} \hat{n} } \hat{a}^{\dagger}  e^{\frac{\kappa t}{2} \hat{n} }
\nonumber\\
&= \kappa e^{-\kappa t} \hat{a}\hat{\rho}_{I}(t)\hat{a}^{\dagger}. \label{eq:excitation loss jump evolution interaction picture}
\end{align}
Here, we used
\begin{align}
e^{\theta \hat{n}} \hat{a} e^{-\theta \hat{n}} &=  e^{-\theta} \hat{a}, 
\nonumber\\
e^{-\theta \hat{n}} \hat{a}^{\dagger} e^{\theta \hat{n}} &=  e^{-\theta} \hat{a}^{\dagger}.   
\end{align}
to derive the last equation. By iteratively integrating the both sides of Eq.\ \eqref{eq:excitation loss jump evolution interaction picture}, we find 
\begin{align}
\hat{\rho}_{I}(t) &= \hat{\rho}_{I}(0) + \int_{0}^{t}dt_{1} \kappa e^{-\kappa t_{1}}\hat{a}\hat{\rho}_{I}(t_{1})\hat{a}^{\dagger} 
\nonumber\\
&=\hat{\rho}_{I}(0) + \int_{0}^{t}dt_{1} \kappa e^{-\kappa t_{1}}\hat{a}\hat{\rho}_{I}(0)\hat{a}^{\dagger} + \int_{0}^{t}dt_{1}\int_{0}^{t_{1}}dt_{2} \kappa^{2} e^{-\kappa (t_{1} +t_{2} )}\hat{a}^{2}\hat{\rho}_{I}(t_{2}) ( \hat{a}^{\dagger} )^{2} 
\nonumber\\
&=\hat{\rho}_{I}(0) +  \sum_{\ell=1}^{\infty}  \kappa^{\ell} \int_{0}^{t} dt_{1}\cdots \int_{0}^{t_{\ell-1}} dt_{\ell} e^{-\kappa(t_{1}+\cdots +t_{\ell})} \hat{a}^{\ell}  \hat{\rho}_{I}(0) (\hat{a}^{\dagger})^{\ell} . \label{eq:excitation loss jump evolution intermediate}
\end{align} 
One can show by mathematical induction that the time integral in Eq.\ \eqref{eq:excitation loss jump evolution intermediate} is given by 
\begin{align}
\kappa^{\ell} \int_{0}^{t} dt_{1}\cdots \int_{0}^{t_{\ell-1}} dt_{\ell} e^{-\kappa(t_{1}+\cdots +t_{\ell})} = \frac{(1-e^{-\kappa t})^{\ell}}{\ell!}. \label{eq:excitation loss relevant time integral}
\end{align}
Thus, combining Eqs.\ \eqref{eq:excitation loss jump evolution intermediate} and \eqref{eq:excitation loss relevant time integral}, we get 
\begin{align}
\hat{\rho}_{I}(t) &= \sum_{\ell=0}^{\infty} \frac{(1-e^{-\kappa t})^{\ell}}{\ell!} \hat{a}^{\ell} \hat{\rho}_{I}(0)(\hat{a}^{\dagger})^{\ell} . \label{eq:excitation loss jump final}
\end{align}
Recall Eq.\ \eqref{eq:excitation loss jump evolution interaction picture} and observe that the time-evolution of $\hat{\rho}_{I}(t)$ is solely generated by the jump term (modulo the additional time-dependent factor $e^{-\kappa t}$). This is why the resulting term $\hat{a}^{\ell}$ is called the jump term. 

On the other hand, the actual density matrix in the laboratory frame $\hat{\rho}(t)$ is also affected by the back-action terms (or the non-Hermitian Hamiltonian terms) as $\hat{\rho}(t)$ and $\hat{\rho}_{I}(t)$ are related by the equation 
\begin{align}
\hat{\rho}(t) &= e^{-\frac{\kappa t}{2}\hat{n}} \hat{\rho}_{I}(t)e^{-\frac{\kappa t}{2}\hat{n}}. \label{eq:excitation loss no jump evolution}
\end{align}
Note that the decay term $e^{-\frac{\kappa t}{2} \hat{n}}$ is solely generated by the back-action terms. This is why the decay term $e^{-\frac{\kappa t}{2} \hat{n}}$ is called the no-jump evolution term. Finally, combining Eqs.\ \eqref{eq:excitation loss jump final} and \eqref{eq:excitation loss no jump evolution} and using $\hat{\rho}_{I}(0) = \hat{\rho}(0)$, we end up with the desired result. 
\begin{align}
\hat{\rho}(t) &= \sum_{\ell =0}^{\infty} \frac{(1-e^{-\kappa t})^{\ell}}{\ell!} e^{-\frac{\kappa t}{2}\hat{n}} \hat{a}^{\ell} \hat{\rho}(0)(\hat{a}^{\dagger})^{\ell} e^{-\frac{\kappa t}{2}\hat{n}}. 
\end{align}

\subsubsection{Heisenberg picture}

Now, we will try to understand the excitation loss errors by inspecting how the expectation value of an observable changes over time under the excitation loss. Recall the Lindblad equation for excitation loss errors (i.e., Eq.\ \eqref{eq:excitation loss Lindbladian equation}). 
\begin{align}
\frac{d\hat{\rho}(t)}{dt} &= \kappa \Big{[} \hat{a}\hat{\rho}(t)\hat{a}^{\dagger} - \frac{1}{2}\hat{a}^{\dagger}\hat{a}\hat{\rho}(t) - \frac{1}{2}\hat{\rho}(t)\hat{a}^{\dagger}\hat{a}  \Big{]}. 
\end{align}
Then, the expectation value of an observable $\hat{A}$ (i.e., $\langle \hat{A}\rangle_{t} \equiv \mathrm{Tr}[\hat{\rho}(t)\hat{A}]$ evolves under the following equation: 
\begin{align}
\frac{d\langle \hat{A}\rangle_{t}}{dt} &= \kappa \Big{[} \langle \hat{a}^{\dagger}\hat{A}\hat{a} \rangle_{t} - \frac{1}{2} \langle \hat{A}\hat{a}^{\dagger}\hat{a} \rangle_{t} - \frac{1}{2} \langle \hat{a}^{\dagger}\hat{a}\hat{A} \rangle_{t} \Big{]}
\nonumber\\
&= \frac{\kappa}{2} \Big{[} \langle [\hat{a}^{\dagger},\hat{A}]\hat{a} \rangle_{t} + \langle \hat{a}^{\dagger} [\hat{A},\hat{a}] \rangle_{t} \Big{]}. \label{eq:time evolution of observables under excitation loss}
\end{align}

Let us first consider the annihilation and creation operators $\hat{a}$ and $\hat{a}^{\dagger}$. 
Note that the relevant commutation relations are given by 
\begin{align}
&[\hat{a}^{\dagger} , \hat{a} ] = -1, \quad [\hat{a},\hat{a}] = 0,  \quad [\hat{a}^{\dagger} , \hat{a}^{\dagger} ] = 0 .  
\end{align}
Specializing Eq.\ \eqref{eq:time evolution of observables under excitation loss} to $\hat{A} = \hat{a}$ and $\hat{A} = \hat{a}^{\dagger}$, we find 
\begin{align}
\frac{d\langle \hat{a} \rangle_{t}}{dt} &= -\frac{\kappa}{2}\langle \hat{a} \rangle_{t},   \quad \frac{d\langle \hat{a}^{\dagger} \rangle_{t}}{dt} = -\frac{\kappa}{2}\langle \hat{a}^{\dagger} \rangle_{t}, 
\end{align}
yielding
\begin{align}
\langle \hat{a} \rangle_{t} &= \langle \hat{a} \rangle_{0}e^{-\frac{\kappa t}{2}}, \quad \langle \hat{a}^{\dagger} \rangle_{t} = \langle \hat{a}^{\dagger} \rangle_{0}e^{-\frac{\kappa t}{2}} .  \label{eq:excitation loss Heisenberg picture anni and cre first order}
\end{align}
Since the position and momentum operators are defined as 
\begin{align}
\hat{q} &\equiv \frac{1}{\sqrt{2}}(\hat{a}^{\dagger}+\hat{a}), \quad \hat{p} \equiv \frac{i}{\sqrt{2}}(\hat{a}^{\dagger}-\hat{a}), 
\end{align}
the expectation values of the position and momentum operators also decrease exponentially over time and eventually converge to zero as $t\rightarrow \infty$: 
\begin{align}
\langle \hat{q} \rangle_{t} &= \langle \hat{q} \rangle_{0}e^{-\frac{\kappa t}{2}}, \quad \langle \hat{p} \rangle_{t} = \langle \hat{p} \rangle_{0}e^{-\frac{\kappa t}{2}} .  \label{eq:excitation loss Heisenberg picture first moments}
\end{align}

Let us now move on to the second-order operators in $\hat{a}$ and $\hat{a}$, i.e., $\hat{a}^{2}$, $\hat{a}^{\dagger}\hat{a}$, and $(\hat{a}^{\dagger})^{2}$.
Then, the relevant commutation relations are given by
\begin{alignat}{2}
&[\hat{a}^{\dagger},\hat{a}^{2}] = -2\hat{a}, &\quad &[\hat{a}^{2},\hat{a}] =0, 
\nonumber\\
&[\hat{a}^{\dagger}, \hat{a}^{\dagger}\hat{a} ] = -\hat{a}^{\dagger}, &\quad &[\hat{a}^{\dagger}\hat{a} ,\hat{a}] =-\hat{a}, 
\nonumber\\
&[\hat{a}^{\dagger}, (\hat{a}^{\dagger})^{2} ] = 0, &\quad &[(\hat{a}^{\dagger})^{2} ,\hat{a}] =-2\hat{a}^{\dagger}.  
\end{alignat}
Specializing Eq.\ \eqref{eq:time evolution of observables under excitation loss} to $\hat{A} = \hat{a}^{2}$, $\hat{A} = \hat{a}^{\dagger}\hat{a}$, and $\hat{A} = (\hat{a}^{\dagger})^{2}$, we find 
\begin{align}
\frac{d\langle \hat{a}^{2} \rangle_{t}}{dt} &= -\kappa \langle \hat{a}^{2} \rangle_{t}, \quad \frac{d\langle \hat{a}^{\dagger}\hat{a} \rangle_{t}}{dt} = -\kappa \langle \hat{a}^{\dagger}\hat{a} \rangle_{t}, \quad \frac{d\langle (\hat{a}^{\dagger})^{2} \rangle_{t}}{dt} = -\kappa \langle (\hat{a}^{\dagger})^{2} \rangle_{t},
\end{align}
yielding 
\begin{align}
\langle \hat{a}^{2} \rangle_{t} &= \langle \hat{a}^{2} \rangle_{0}e^{-\kappa t}, \quad \langle \hat{a}^{\dagger}\hat{a} \rangle_{t} = \langle \hat{a}^{\dagger}\hat{a} \rangle_{0}e^{-\kappa t}, \quad \langle (\hat{a}^{\dagger})^{2} \rangle_{t} = \langle (\hat{a}^{\dagger})^{2} \rangle_{0}e^{-\kappa t}. \label{eq:excitation loss Heisenberg picture second moments cre and ani}
\end{align}
Note that the second equation shows that the average excitation number (or energy) $\langle \hat{a}^{\dagger}\hat{a}\rangle_{t}$ decreases exponentially over time and eventually converges to zero as $t\rightarrow \infty$.  

Now consider the second-order operators in $\hat{q}$ and $\hat{p}$: 
\begin{align}
\hat{q}^{2} &= \frac{1}{2}(\hat{a}^{\dagger}+\hat{a})(\hat{a}^{\dagger}+\hat{a}) = \frac{1}{2}( (\hat{a}^{\dagger})^{2} + 2\hat{a}^{\dagger}\hat{a} + \hat{a}^{2} +1 ), 
\nonumber\\
\frac{1}{2}(\hat{q}\hat{p} + \hat{p}\hat{q}) &= \frac{i}{2} \Big{[} (\hat{a}^{\dagger}+\hat{a})(\hat{a}^{\dagger}-\hat{a}) + (\hat{a}^{\dagger}-\hat{a})(\hat{a}^{\dagger}+\hat{a}) \Big{]} = \frac{i}{2}( (\hat{a}^{\dagger})^{2} - \hat{a}^{2} ), 
\nonumber\\
\hat{p}^{2} &= -\frac{1}{2}(\hat{a}^{\dagger}-\hat{a})(\hat{a}^{\dagger}-\hat{a}) = \frac{1}{2}( -(\hat{a}^{\dagger})^{2} +2\hat{a}^{\dagger}\hat{a} -\hat{a}^{2} +1 ). \label{eq:second order operators in q and p} 
\end{align}
Then, combining Eqs.\ \eqref{eq:excitation loss Heisenberg picture second moments cre and ani} and \eqref{eq:second order operators in q and p}, we get  
\begin{align}
\langle \hat{q}^{2} \rangle_{t} &= \langle \hat{q}^{2} \rangle_{0}e^{-\kappa t} + \frac{1}{2}(1-e^{-\kappa t}), 
\nonumber\\
\Big{\langle} \frac{1}{2}( \hat{q}\hat{p} + \hat{p}\hat{q} ) \Big{\rangle}_{t} &= \Big{\langle} \frac{1}{2}( \hat{q}\hat{p} + \hat{p}\hat{q} ) \Big{\rangle}_{0}e^{-\kappa t}, 
\nonumber\\
\langle \hat{p}^{2} \rangle_{t} &= \langle \hat{p}^{2} \rangle_{0}e^{-\kappa t} + \frac{1}{2}(1-e^{-\kappa t}). \label{eq:excitation loss Heisenberg picture second moments}
\end{align}
In the Heisenberg picture, operators evolve over time and states remain unchanged. Thus, from Eqs.\ \eqref{eq:excitation loss Heisenberg picture first moments} and \eqref{eq:excitation loss Heisenberg picture second moments}, we can conclude 
\begin{align}
&\hat{q}(t) = \hat{q}e^{-\frac{\kappa t}{2}}, \quad \hat{p}(t) = \hat{p}e^{-\frac{\kappa t}{2}}, \quad \hat{q}^{2} (t)  = \hat{q}^{2} e^{-\kappa t} + \frac{1}{2}(1-e^{-\kappa t})
\nonumber\\
&\frac{1}{2}( \hat{q}\hat{p} + \hat{p}\hat{q} )(t) = \frac{1}{2}( \hat{q}\hat{p} + \hat{p}\hat{q} )e^{-\kappa t} ,  \quad \hat{p}^{2} (t)  = \hat{p}^{2} e^{-\kappa t} + \frac{1}{2}(1-e^{-\kappa t}) . 
  \label{eq:excitation loss Heisenberg picture summarized}
\end{align}
We remark that the time evolution of an operator is governed by the adjoint master equation and, in general, its solution cannot be directly inferred from the time evolution of the expectation value of the operator. However, in the case of photon loss, one can verify that the solution of the adjoint master equation can be directly read off from the expectation value by plugging in the expectation-value-inspired solutions in Eq.\ \eqref{eq:excitation loss Heisenberg picture summarized} to the adjoint master equation. Note also that we did not consider higher than second order operators in $\hat{q}$ and $\hat{p}$. However, since an excitation loss error is a Gaussian channel as we will show below, it is sufficient to understand the first and the second-order operators in $\hat{q}$ and $\hat{p}$.

\subsubsection{Gaussian channels}

Lastly, the CPTP map $e^{\kappa t \mathcal{D}[\hat{a}]}$ due to excitation loss is equivalent to a bosonic pure-loss channel $\mathcal{N}[\eta, 0]$ with transmissivity $\eta = e^{-\kappa t}$: 
\begin{align}
e^{\kappa t \mathcal{D}[\hat{a}]} = \mathcal{N}[\eta = e^{-\kappa t}, 0] \leftrightarrow (\boldsymbol{T},\boldsymbol{N},\boldsymbol{d}) = ( e^{-\frac{\kappa t}{2}} \boldsymbol{I}_{2}, \frac{1}{2}(1-e^{-\kappa t})\boldsymbol{I}_{2} , 0 ) . \label{eq:excitation loss errors equal pure loss channels} 
\end{align}
The bosonic pure-loss channel $\mathcal{N}[\eta, 0]$ is a Gaussian channel characterized by $(\boldsymbol{T},\boldsymbol{N},\boldsymbol{d}) = ( \sqrt{\eta} \boldsymbol{I}_{2}, \frac{1}{2}(1-\eta)\boldsymbol{I}_{2} , 0 )$ (see also Definition \ref{definition:Bosonic pure loss channels}). An introduction to Gaussian states, unitaries, and channels is given in Appendix \ref{appendix:Gaussian states, unitaries, and channels}. Specifically, the definitions of Gaussian channels and their characterization $(\boldsymbol{T},\boldsymbol{N},\boldsymbol{d})$ are given in Appendix \ref{section:Gaussian channels}. 

The characterization $(\boldsymbol{T},\boldsymbol{N},\boldsymbol{d}) = ( \sqrt{\eta} \boldsymbol{I}_{2}, \frac{1}{2}(1-\eta)\boldsymbol{I}_{2} , 0 )$ implies that the expectation values of the first and the second-order operators in $\hat{q}$ and $\hat{p}$ are transformed via the bosonic pure-loss channel $\mathcal{N}[\eta,0]$ as follows: 
\begin{align}
\langle \hat{q}\rangle &\rightarrow \langle \hat{q}\rangle' = \sqrt{\eta}\langle \hat{q}\rangle, 
\nonumber\\
\langle \hat{p}\rangle &\rightarrow \langle \hat{p}\rangle' = \sqrt{\eta}\langle \hat{p}\rangle, 
\nonumber\\
\langle \hat{q}^{2}\rangle &\rightarrow \langle \hat{q}^{2}\rangle' = \eta\langle \hat{q}^{2}\rangle + \frac{1}{2}(1-\eta),
\nonumber\\
\Big{\langle} \frac{1}{2}(\hat{q}\hat{p} + \hat{p}\hat{q}) \Big{\rangle} &\rightarrow\Big{\langle} \frac{1}{2}(\hat{q}\hat{p} + \hat{p}\hat{q}) \Big{\rangle}' = \eta \Big{\langle} \frac{1}{2}(\hat{q}\hat{p} + \hat{p}\hat{q}) \Big{\rangle}, 
\nonumber\\ 
\langle \hat{p}^{2}\rangle &\rightarrow \langle \hat{p}^{2}\rangle' = \eta\langle \hat{p}^{2}\rangle + \frac{1}{2}(1-\eta). \label{eq:excitation loss expectation value change under pure loss channel}
\end{align}
Note that Eq.\ \eqref{eq:excitation loss expectation value change under pure loss channel} is consistent with Eq.\ \eqref{eq:excitation loss Heisenberg picture summarized} if one sets $\eta = e^{-\kappa t}$. Hence, we have Eq.\ \eqref{eq:excitation loss errors equal pure loss channels}.  



\subsection{Gaussian random shift errors in the phase space}
\label{subsection:Random shift errors in the phase space}

In this subsection, we consider Gaussian random shift errors in the phase space. This error model is also known as a Gaussian random displacement error or an additive Gaussian noise error. Typically, natural errors in realistic bosonic systems are not described by random shift errors. However, random shift errors are important for understanding GKP codes. Here, we will discuss three different ways to describe Gaussian random shift errors (see Table \ref{table:random shift errors}).

\begin{table}[h!]
  \centering
     \def\arraystretch{2}
  \begin{tabular}{ V{3} c V{1.5} c V{3} }
   \hlineB{3}  
     \textbf{Representation} & \textbf{Gaussian random shift errors in the phase space}  \\  \hlineB{3} 
      Kraus representation & $\mathcal{N}_{B_{2}}[\sigma](\hat{\rho}) \equiv \frac{1}{\pi \sigma^{2}} \int d^{2}\alpha \exp\big{[} -\frac{|\alpha|^{2}}{\sigma^{2}}\big{]} \hat{D}(\alpha) \hat{\rho}\hat{D}^{\dagger}(\alpha)$  \\ 
       (Eqs.\ \eqref{eq:random shift Kraus in terms of displacement}, \eqref{eq:random shift Kraus in terms of shift}) & $=\frac{1}{2\pi\sigma^{2}} \int_{-\infty}^{\infty}d\xi_{q}\int_{-\infty}^{\infty}d\xi_{p} \exp\big{[} -\frac{\xi_{q}^{2}+\xi_{p}^{2}}{2\sigma^{2}} \big{]} e^{i(\xi_{p}\hat{q}-\xi_{q}\hat{p})} \hat{\rho} e^{-i(\xi_{p}\hat{q}-\xi_{q}\hat{p})}$.    \\ \hlineB{1.5}   
       Heisenberg picture & $\hat{q} \rightarrow \hat{q}'  = \hat{q}  + \xi_{q}$ and $\hat{p} \rightarrow \hat{p}'  = \hat{p}  + \xi_{p}$    \\ 
       (Eq.\ \eqref{eq:random shift Heisenberg picture})  & where $\xi_{q},\xi_{p} \sim \mathcal{N}(0,\sigma^{2})$.    \\  \hlineB{1.5}   
        Gaussian channels & $\mathcal{N}_{B_{2}}[\sigma] \leftrightarrow (\boldsymbol{T} , \boldsymbol{N}, \boldsymbol{d}) = ( \boldsymbol{I}_{2} , \sigma^{2}\boldsymbol{I}_{2} , 0  )$.  \\  
        (Eq.\ \eqref{eq:random shift errors Gaussian channels}) &   \\   \hlineB{3}   
  \end{tabular}
  \caption{Various representations of the Gaussian random shift errors in the phase space.}
  \label{table:random shift errors}
\end{table}

\subsubsection{Kraus representation}

Gaussian random shift errors are defined as follows: 
\begin{align}
\mathcal{N}_{B_{2}}[\sigma](\hat{\rho}) &\equiv \frac{1}{\pi \sigma^{2}} \int d^{2}\alpha \exp\Big{[} -\frac{|\alpha|^{2}}{\sigma^{2}}\Big{]} \hat{D}(\alpha) \hat{\rho}\hat{D}^{\dagger}(\alpha).   \label{eq:random shift Kraus in terms of displacement}
\end{align}
Here, $\hat{D}(\alpha) \equiv \exp[\alpha \hat{a}^{\dagger} -\alpha^{*}\hat{a} ]$ is the displacement operator and $\sigma$ is the standard deviation of the random displacement. The expression in Eq.\ \eqref{eq:random shift Kraus in terms of displacement} can be understood as a continuous Kraus representation, i.e., 
\begin{align}
\mathcal{N}_{B_{2}}[\sigma](\hat{\rho}) = \int d^{2}\alpha \hat{E}(\alpha)\hat{\rho}\hat{E}^{\dagger}(\alpha), \,\,\, \textrm{where}\,\,\,  \hat{E}(\alpha) = \sqrt{ \frac{1}{\pi\sigma^{2}} \exp\Big{[} -\frac{|\alpha|^{2}}{\sigma^{2}} \Big{]}} \hat{D}(\alpha). 
\end{align}
Equivalently, one can also write 
\begin{align}
\mathcal{N}_{B_{2}}[\sigma](\hat{\rho}) &= \frac{1}{2\pi\sigma^{2}} \int_{-\infty}^{\infty}d\xi_{q}\int_{-\infty}^{\infty}d\xi_{p} \exp\Big{[} -\frac{\xi_{q}^{2}+\xi_{p}^{2}}{2\sigma^{2}} \Big{]} e^{i(\xi_{p}\hat{q}-\xi_{q}\hat{p})} \hat{\rho} e^{-i(\xi_{p}\hat{q}-\xi_{q}\hat{p})} . \label{eq:random shift Kraus in terms of shift}
\end{align}

\subsubsection{Heisenberg picture}

The definitions in Eqs.\ \eqref{eq:random shift Kraus in terms of displacement} and \eqref{eq:random shift Kraus in terms of shift} clearly show that the Gaussian random shift error $\mathcal{N}_{B_{2}}[\sigma]$ makes the system drift in the phase space. More explicitly, the action of $\mathcal{N}_{B_{2}}[\sigma]$ transforms the position and momentum operators as follows in the Heisenberg picture: 
\begin{align}
\hat{q} \rightarrow \hat{q}'  = \hat{q}  + \xi_{q} , \quad  \hat{p} \rightarrow \hat{p}'  = \hat{p}  + \xi_{p} , \,\,\, \textrm{where} \,\,\, \xi_{q},\xi_{p} \sim \mathcal{N}(0,\sigma^{2}). \label{eq:random shift Heisenberg picture}
\end{align}
Here, $\mathcal{N}(0,\sigma^{2})$ is the Gaussian normal distribution with zero mean and variance $\sigma^{2}$. To see why this is the case, let us consider an observable $\hat{A}$ and its expectation value $\langle \hat{A} \rangle = \mathrm{Tr}[\hat{\rho}\hat{A}]$. Under the Gaussian random shift error $\mathcal{N}_{B_{2}}[\sigma]$, $\langle \hat{A} \rangle$ is transformed as follows: 
\begin{align}
\langle \hat{A} \rangle \rightarrow \langle \hat{A} \rangle'  &= \mathrm{Tr}\Big{[} \frac{1}{2\pi\sigma^{2}} \int_{-\infty}^{\infty}d\xi_{q}\int_{-\infty}^{\infty}d\xi_{p} \exp\Big{[} -\frac{\xi_{q}^{2}+\xi_{p}^{2}}{2\sigma^{2}} \Big{]} e^{i(\xi_{p}\hat{q}-\xi_{q}\hat{p})} \hat{\rho} e^{-i(\xi_{p}\hat{q}-\xi_{q}\hat{p})}  \hat{A}   \Big{]}
\nonumber\\
&=  \frac{1}{2\pi\sigma^{2}} \int_{-\infty}^{\infty}d\xi_{q}\int_{-\infty}^{\infty}d\xi_{p} \exp\Big{[} -\frac{\xi_{q}^{2}+\xi_{p}^{2}}{2\sigma^{2}} \Big{]}  \Big{\langle }  e^{-i(\xi_{p}\hat{q}-\xi_{q}\hat{p})}  \hat{A}  e^{i(\xi_{p}\hat{q}-\xi_{q}\hat{p})} \Big{\rangle }
\end{align}
Note that for $\hat{A} = \hat{q}$ and $\hat{A} = \hat{p}$, we have
\begin{align}
e^{-i(\xi_{p}\hat{q}-\xi_{q}\hat{p})}  \hat{q}  e^{i(\xi_{p}\hat{q}-\xi_{q}\hat{p})} &= \hat{q}  +\xi_{q}, \quad e^{-i(\xi_{p}\hat{q}-\xi_{q}\hat{p})}  \hat{p}  e^{i(\xi_{p}\hat{q}-\xi_{q}\hat{p})} = \hat{p}  +\xi_{p}, 
\end{align}
and thus the expectation values of the quadrature operators $\hat{q}$ and $\hat{p}$ do not change, i.e.,
\begin{align}
\langle \hat{q}\rangle' = \langle \hat{q} +\xi_{q}\rangle = \langle \hat{q}\rangle, \quad  \langle \hat{p}\rangle' = \langle \hat{p} +\xi_{p}\rangle = \langle \hat{p}\rangle, \label{eq:random shift observable change first order}
\end{align}
where we used $\langle \xi_{q} \rangle =0$ and $\langle \xi_{p}\rangle =0$. 

Let us now consider the second-order operators in $\hat{q}$ and $\hat{p}$, i.e., $\hat{A} = \hat{q}^{2}$, $\hat{A} = \frac{1}{2}(\hat{q}\hat{p} + \hat{p}\hat{q})$, and $\hat{A}=\hat{p}^{2}$. Note that 
\begin{align}
&e^{-i(\xi_{p}\hat{q}-\xi_{q}\hat{p})}  \hat{q}^{2}  e^{i(\xi_{p}\hat{q}-\xi_{q}\hat{p})} = (\hat{q}  +\xi_{q})^{2} , 
\nonumber\\
&e^{-i(\xi_{p}\hat{q}-\xi_{q}\hat{p})}  \frac{1}{2}(\hat{q}\hat{p}+\hat{p}\hat{q})  e^{i(\xi_{p}\hat{q}-\xi_{q}\hat{p})} = \frac{1}{2}\Big{[} (\hat{q}+\xi_{q})(\hat{p}+\xi_{p}) + (\hat{p}+\xi_{p})(\hat{q}+\xi_{q})  \Big{]}  , 
\nonumber\\
&e^{-i(\xi_{p}\hat{q}-\xi_{q}\hat{p})}  \hat{p}^{2}  e^{i(\xi_{p}\hat{q}-\xi_{q}\hat{p})} = (\hat{p}  +\xi_{p})^{2}. 
\end{align}
Thus, the expectation values of the second-order operators change under the Gaussian random shift error $\mathcal{N}_{B_{2}}[\sigma]$ as follows:   
\begin{align}
\langle \hat{q}^{2}\rangle' &= \langle (\hat{q}+\xi_{q})^{2}\rangle = \langle \hat{q}^{2}\rangle+ 2\langle\hat{q}\rangle \langle \xi_{q}\rangle + \langle \xi_{q}^{2}\rangle = \langle \hat{q}^{2}\rangle  + \sigma^{2} , 
\nonumber\\
\Big{\langle} \frac{1}{2}(\hat{q}\hat{p} + \hat{p}\hat{q} ) \Big{\rangle}' &= \Big{\langle} \frac{1}{2}(\hat{q}\hat{p} + \hat{p}\hat{q} ) \Big{\rangle} + \langle \hat{q}\rangle\langle\xi_{p}\rangle  + \langle \hat{p}\rangle\langle \xi_{q}\rangle + \langle \xi_{q}\rangle\langle\xi_{p}\rangle = \Big{\langle} \frac{1}{2}(\hat{q}\hat{p} + \hat{p}\hat{q} ) \Big{\rangle} , 
\nonumber\\
\langle \hat{p}^{2}\rangle' &= \langle (\hat{p}+\xi_{p})^{2}\rangle = \langle \hat{p}^{2}\rangle+ 2\langle\hat{p}\rangle \langle \xi_{p}\rangle + \langle \xi_{p}^{2}\rangle = \langle \hat{p}^{2}\rangle  + \sigma^{2}, \label{eq:random shift observable change second order}
\end{align}
where we used $\langle \xi_{q}\rangle = \langle \xi_{p}\rangle = 0$ and $\langle\xi_{q}^{2}\rangle =\langle\xi_{p}^{2}\rangle = \sigma^{2}$. 

\subsubsection{Gaussian channels} 

The Gaussian random shift error $\mathcal{N}_{B_{2}}[\sigma]$ is a Gaussian channel with the following characterization: 
\begin{align}
\mathcal{N}_{B_{2}}[\sigma] \leftrightarrow (\boldsymbol{T} , \boldsymbol{N}, \boldsymbol{d}) = ( \boldsymbol{I}_{2} , \sigma^{2}\boldsymbol{I}_{2} , 0  ).  \label{eq:random shift errors Gaussian channels}
\end{align}
Note that the characterization $(\boldsymbol{T} , \boldsymbol{N}, \boldsymbol{d}) = ( \boldsymbol{I}_{2} , \sigma^{2}\boldsymbol{I}_{2} , 0  )$ yields
\begin{align}
\langle \hat{q} \rangle &\rightarrow \langle \hat{q} \rangle' = \langle \hat{q} \rangle, 
\nonumber\\
\langle \hat{p} \rangle &\rightarrow \langle \hat{p} \rangle' = \langle \hat{p} \rangle, 
\nonumber\\
\langle \hat{q}^{2} \rangle &\rightarrow \langle \hat{q}^{2} \rangle' = \langle \hat{q}^{2} \rangle + \sigma^{2} , 
\nonumber\\
\Big{\langle} \frac{1}{2}(\hat{q}\hat{p} + \hat{p}\hat{q} ) \Big{\rangle} &\rightarrow \Big{\langle} \frac{1}{2}(\hat{q}\hat{p} + \hat{p}\hat{q} ) \Big{\rangle}' = \Big{\langle} \frac{1}{2}(\hat{q}\hat{p} + \hat{p}\hat{q} ) \Big{\rangle}, 
\nonumber\\
\langle \hat{p}^{2} \rangle &\rightarrow \langle \hat{p}^{2} \rangle' = \langle \hat{p}^{2} \rangle + \sigma^{2} . \label{eq:random shift errors obvservable change Gaussian channel}
\end{align}
That is, the channel adds a noise variance $\sigma^{2}$ only to the diagonal elements of the covariance matrix. Note that Eq.\ \eqref{eq:random shift errors obvservable change Gaussian channel} is consistent with Eqs.\ \eqref{eq:random shift observable change first order} and \eqref{eq:random shift observable change second order}. Thus, Eq.\ \eqref{eq:random shift errors Gaussian channels} follows. 

\subsection{Bosonic dephasing errors}
\label{subsection:Bosonic dephasing errors}

Realistic bosonic modes are sometimes subject to bosonic dephasing errors as well as excitation loss errors. Here, we provide three different ways to describe bosonic dephasing errors (see Table \ref{table:bosonic dephasing errors}). 

\begin{table}[h!]
  \centering
     \def\arraystretch{2}
  \begin{tabular}{ V{3} c V{1.5} c V{3} }
   \hlineB{3}  
     \textbf{Representation} & \textbf{Bosonic dephasing errors}  \\  \hlineB{3} 
      Lindblad equation & $\frac{d\hat{\rho}(t)}{dt} = \kappa_{\phi} \mathcal{D}[\hat{a}^{\dagger}\hat{a}] (\hat{\rho}(t)) \rightarrow \hat{\rho}(t) = e^{\kappa_{\phi} t \mathcal{D}[\hat{a}^{\dagger}\hat{a}]} \hat{\rho}(0)$  \\ 
       (Eqs.\ \eqref{eq:bosonic dephasing Lindblad equation}, \eqref{eq:bosonic dephasing CPTP map},  & $e^{\kappa_{\phi} t \mathcal{D}[\hat{a}^{\dagger}\hat{a}]} \hat{\rho}(0)  =  \sum_{m,n=0}^{\infty}\rho_{mn} e^{-\frac{1}{2}(m-n)^{2}\kappa_{\phi}t} |m\rangle \langle n|$,    \\
       and \eqref{eq:bosonic dephasing CPTP map Fock basis explicit}) & where $\rho_{mn} = \langle m|\hat{\rho}(0)|n\rangle$. \\ \hlineB{1.5}   
       Continuous Kraus  & $e^{\kappa_{\phi}t \mathcal{D}[\hat{a}^{\dagger}\hat{a}] } = \mathcal{N}_{D}[\sigma = \sqrt{ \kappa_{\phi} t} ]$ where    \\ 
          representation & $\mathcal{N}_{D}[\sigma](\hat{\rho}) \equiv \frac{1}{\sqrt{2\pi \sigma^{2}}} \int_{-\infty}^{\infty} d\phi e^{-\frac{\phi^{2}}{2\sigma^{2}}} e^{i\phi \hat{n}} \hat{\rho} e^{-i\phi\hat{n}}$.    \\ 
       (Eqs.\ \eqref{eq:bosonic dephasing continuous Kraus representation}, \eqref{eq:bosonic dephasing equals random rotation}) & \\ \hlineB{1.5}   
        Discrete Kraus & $\mathcal{N}_{D}[\sigma](\hat{\rho}) =  \sum_{k=0}^{\infty} \hat{N}_{k} \hat{\rho}\hat{N}_{k}^{\dagger}$ where  \\  
        representation & $\hat{N}_{k} = \sqrt{ \frac{\sigma^{2k}}{k!} }  e^{-\frac{\sigma^{2}}{2} \hat{n}^{2} }\hat{n}^{k}$.  \\ 
      (Eq.\ \eqref{eq:bosonic dephasing discrete Kraus representation})  & \\\hlineB{3}   
  \end{tabular}
  \caption{Various representations of the bosonic dephasing errors.}
  \label{table:bosonic dephasing errors}
\end{table}

\subsubsection{Lindblad equation}

Bosonic dephasing errors are described by the following Lindblad equation: 
\begin{align}
\frac{d\hat{\rho}(t)}{dt} &= \kappa_{\phi} \mathcal{D}[\hat{a}^{\dagger}\hat{a}] (\hat{\rho}(t)). \label{eq:bosonic dephasing Lindblad equation}
\end{align}
Note that the jump operator is given by the excitation number operator $\hat{n} = \hat{a}^{\dagger}\hat{a}$ for the dephasing errors, whereas it is given by the annihilation operator $\hat{a}$ for the excitation loss errors. By solving Eq.\ \eqref{eq:bosonic dephasing Lindblad equation}, we get the following CPTP map
\begin{align}
\hat{\rho}(0) \rightarrow \hat{\rho}(t) = e^{\kappa_{\phi} t \mathcal{D}[\hat{a}^{\dagger}\hat{a}]} \hat{\rho}(0) .  \label{eq:bosonic dephasing CPTP map}
\end{align}
In the Fock basis, the Lindblad equation in Eq.\ \eqref{eq:bosonic dephasing Lindblad equation} is explicitly given by
\begin{align}
\langle m| \frac{d\hat{\rho}(t)}{dt} | n\rangle &= \kappa_{\phi}\langle m| \Big{[} \hat{n} \hat{\rho}(t)\hat{n} - \frac{1}{2}\hat{n}^{2}\hat{\rho}(t) - \frac{1}{2}\hat{\rho}(t)\hat{n}^{2} \Big{]} |n\rangle. 
\end{align}
Thus, we have
\begin{align}
\frac{d\rho_{mn}(t)}{dt} &= \kappa_{\phi}\Big{(} mn-\frac{1}{2}m^{2} - \frac{1}{2}n^{2} \Big{)} \rho_{mn}(t)  = -\frac{\kappa_{\phi}}{2} (m-n)^{2} \rho_{mn}(t), 
\end{align}
yielding $\rho_{mn}(t) = \rho_{mn}(0)e^{-\frac{1}{2}(m-n)^{2}\kappa_{\phi}t}$ and 
\begin{align}
\hat{\rho}(t) = \sum_{m,n=0}^{\infty}|m\rangle \langle m|\hat{\rho}(0)|n\rangle \langle n| e^{-\frac{1}{2}(m-n)^{2}\kappa_{\phi}t} . \label{eq:bosonic dephasing CPTP map Fock basis explicit}
\end{align}

\subsubsection{Kraus representation (continuous)}

Bosonic dephasing errors can also be understood as a random phase rotation error. Consider the error channel 
\begin{align}
\mathcal{N}_{D}[\sigma](\hat{\rho}) &\equiv \frac{1}{\sqrt{2\pi \sigma^{2}}} \int_{-\infty}^{\infty} d\phi e^{-\frac{\phi^{2}}{2\sigma^{2}}} e^{i\phi \hat{n}} \hat{\rho} e^{-i\phi\hat{n}},  \label{eq:bosonic dephasing continuous Kraus representation}
\end{align}
which is in a continuous Kraus representation where the Kraus operators are given by a rotation operator $\hat{E}(\phi) \propto e^{i\phi \hat{n}}$. Here, the random rotation angle $\phi$ follows the Gaussian normal distribution with zero mean and variance $\sigma^{2}$, i.e., $\phi \sim \mathcal{N}(0, \sigma^{2})$. 

To show that the CPTP map in Eq.\ \eqref{eq:bosonic dephasing continuous Kraus representation} is equivalent to the CPTP map generated by the Lindblad equation in Eq.\ \eqref{eq:bosonic dephasing Lindblad equation}, we explicitly write down Eq.\ \eqref{eq:bosonic dephasing continuous Kraus representation} in the Fock basis: 
\begin{align}
\mathcal{N}_{D}[\sigma](\hat{\rho}) &= \frac{1}{\sqrt{2\pi \sigma^{2}}} \int_{-\infty}^{\infty} d\phi e^{-\frac{\phi^{2}}{2\sigma^{2}}} e^{i\phi \hat{n}} \Big{(} \sum_{m=0}^{\infty}|m\rangle\langle m|\Big{)}  \hat{\rho}  \Big{(} \sum_{n=0}^{\infty}|n\rangle\langle n| \Big{)} e^{-i\phi\hat{n}}
\nonumber\\
&= \sum_{m,n=0}^{\infty} |m\rangle \langle m|\hat{\rho}(0)|n\rangle \langle n| \frac{1}{\sqrt{2\pi\sigma^{2}}} \int_{-\infty}^{\infty} d\phi e^{-\frac{\phi^{2}}{2\sigma^{2}}} e^{i\phi (m-n)} 
\nonumber\\
&= \sum_{m,n=0}^{\infty} |m\rangle \langle m|\hat{\rho}(0)|n\rangle \langle n| e^{-\frac{1}{2}(m-n)^{2}\sigma^{2} }.   \label{eq:bosonic dephasing continuous Kraus representation Fock basis}
\end{align} 
Comparing Eq.\ \eqref{eq:bosonic dephasing continuous Kraus representation Fock basis} with Eq.\ \eqref{eq:bosonic dephasing CPTP map Fock basis explicit}, we can conclude 
\begin{align}
e^{\kappa_{\phi}t \mathcal{D}[\hat{a}^{\dagger}\hat{a}] } = \mathcal{N}_{D}[\sigma = \sqrt{ \kappa_{\phi} t} ]  . \label{eq:bosonic dephasing equals random rotation} 
\end{align}

\subsubsection{Kraus representation (discrete)} 

Lastly, we provide a discrete Kraus representation of the bosonic dephasing channel $\mathcal{N}_{D}[\sigma]$. To do so, recall Eq.\ \eqref{eq:bosonic dephasing continuous Kraus representation Fock basis} and note that 
\begin{align}
\mathcal{N}_{D}[\sigma](\hat{\rho}) &= \sum_{m,n=0}^{\infty} |m\rangle \langle m|\hat{\rho}(0)|n\rangle \langle n| e^{-\frac{1}{2}(m-n)^{2}\sigma^{2} } 
\nonumber\\
&= \sum_{m,n=0}^{\infty} |m\rangle \langle m|\hat{\rho}(0)|n\rangle \langle n| e^{-\frac{1}{2}\sigma^{2} m^{2} } e^{-\frac{1}{2}\sigma^{2} n^{2} }  e^{mn\sigma^{2}}
\nonumber\\
&= \sum_{m,n=0}^{\infty} |m\rangle \langle m|\hat{\rho}(0)|n\rangle \langle n| e^{-\frac{1}{2}\sigma^{2} m^{2} } e^{-\frac{1}{2}\sigma^{2} n^{2} } \sum_{k=0}^{\infty} \frac{1}{k!} (mn\sigma^{2})^{k} 
\nonumber\\
&=  \sum_{k=0}^{\infty} \frac{\sigma^{2k}}{k!} e^{-\frac{\sigma^{2}}{2}\hat{n}^{2}} \hat{n}^{k} \hat{\rho}  \hat{n}^{k} e^{-\frac{\sigma^{2}}{2}\hat{n}^{2}}. 
\end{align}
Thus, the bosonic dephasing channel $\mathcal{N}_{D}[\sigma]$ can also be expressed in the following discrete Kraus form: 
\begin{align}
\mathcal{N}_{D}[\sigma](\hat{\rho}) &=  \sum_{k=0}^{\infty} \hat{N}_{k} \hat{\rho}\hat{N}_{k}^{\dagger} , \,\,\, \textrm{where} \,\,\, \hat{N}_{k} = \sqrt{ \frac{\sigma^{2k}}{k!} }  e^{-\frac{\sigma^{2}}{2} \hat{n}^{2} }\hat{n}^{k} . \label{eq:bosonic dephasing discrete Kraus representation}
\end{align}

\section{Rotation-symmetric bosonic codes}
\label{section:Rotation-symmetric bosonic codes}

In this section, we review rotation-symmetric bosonic codes \cite{Grimsmo2019} such as the cat codes and the binomial codes that are invariant under a discrete set of rotations. 

\subsection{Cat codes}
\label{subsection:Cat codes}

Here, we review the basic properties of the two-component and the four-component cat codes and discuss their experimental implementations. See Table \ref{table:cat codes} for a summary. At the end of this subsection, we also briefly review the recent developments in the cat-code-based bosonic QEC.  

\begin{table}[h!]
  \centering
     \def\arraystretch{1.5}
  \begin{tabular}{ V{3} c V{1.5} c V{1.5} c V{3} }
   \hlineB{3}  
     & \textbf{Two-component}& \textbf{Four-component} \\
     & \textbf{cat codes} \cite{Cochrane1999} & \textbf{cat codes} \cite{Leghtas2013} \\  \hlineB{3} 
     Logical states & $|0_{2-\textrm{cat}}^{(\alpha)}\rangle \propto |\alpha\rangle + |-\alpha\rangle$ & $|0_{4-\textrm{cat}}^{(\alpha)}\rangle \propto |\alpha\rangle + |i\alpha\rangle + |-\alpha\rangle + |-i\alpha\rangle$  \\
     &  $|1_{2-\textrm{cat}}^{(\alpha)}\rangle \propto |\alpha\rangle - |-\alpha\rangle$ & $|1_{4-\textrm{cat}}^{(\alpha)}\rangle \propto |\alpha\rangle - |i\alpha\rangle + |-\alpha\rangle - |-i\alpha\rangle$   \\ \hlineB{1.5}   
    Correctable errors & Dephasing errors & Loss and dephasing errors \\ \hlineB{1.5}
    Active QEC & Teleportation-based  & Parity measurement and  \\ 
     & error correction \cite{Grimsmo2019}  & amplitude recovery \cite{Li2017}  \\  \hdashline
     Experiments & N/A & Refs.\ \cite{Sun2014,Ofek2016,Rosenblum2018}  \\  \hlineB{1.5}     
    Autonomous QEC & Engineered two-photon & Engineered four-photon  \\  
    & dissipation \cite{Mirrahimi2014} & dissipation \cite{Mundhada2017}  \\  \hdashline
    Experiments & Refs.\ \cite{Leghtas2015,Touzard2018,Lescanne2019} & Ref.\ \cite{Mundhada2019}  \\  \hlineB{3}     
  \end{tabular}
  \caption{Basic properties of the two-component and the four-component cat codes. Note that the teleportation-based error correction scheme in Ref.\ \cite{Grimsmo2019} works for any rotation-symmetric bosonic codes. }
  \label{table:cat codes}
\end{table}

\subsubsection{Two-component cat codes}

Logical states of the two-component cat code $\mathcal{C}_{2-\textrm{cat}}^{(\alpha)}$ \cite{Cochrane1999} are given by Schr\"odinger's cat states: 
\begin{align}
|0^{(\alpha)}_{2-\textrm{cat}}\rangle &= \frac{1}{\sqrt{2N_{0}^{2-\textrm{cat}}(\alpha)}} ( |\alpha\rangle + |-\alpha\rangle  ), 
\nonumber\\ 
|1^{(\alpha)}_{2-\textrm{cat}}\rangle &= \frac{1}{\sqrt{2N_{1}^{2-\textrm{cat}}(\alpha)}} ( |\alpha\rangle - |-\alpha\rangle ).  
\end{align}
Here, the normalization constant $N_{\mu}^{2-\textrm{cat}}(\alpha)$ is given by
\begin{align}
N_{\mu}^{2-\textrm{cat}}(\alpha)&\equiv 1 +(-1)^{\mu} e^{-2|\alpha|^{2}} , \,\,\, \mu\in \lbrace 0,1 \rbrace. 
\end{align}
Note that $|0^{(\alpha)}_{2-\textrm{cat}}\rangle$ and $|1^{(\alpha)}_{2-\textrm{cat}}\rangle$ have even and odd excitation numbers, respectively. Wigner functions of the logical states of the two-component cat code are shown in Fig.\ \ref{fig:wigner functions two component cat code}. 

\begin{figure}[t!]
\centering
\includegraphics[width=6.0in]{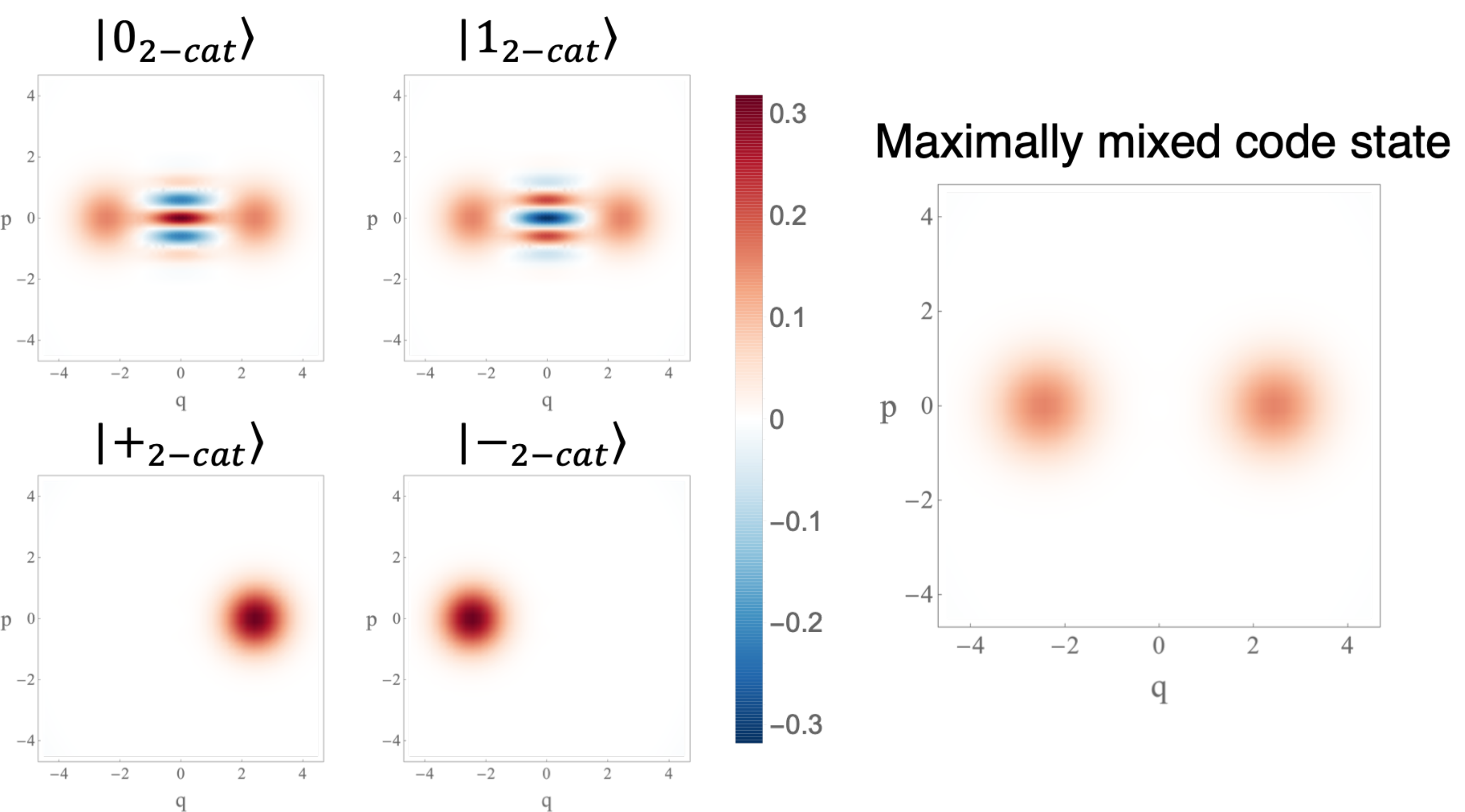}
\caption{Wigner functions of the logical states of the two-component cat code $\mathcal{C}_{2-\textrm{cat}}^{(\alpha)}$ with $\alpha = \sqrt{3}$. The maximally mixed code state is defined as the projection operator to the code space divided by $2$.  }
\label{fig:wigner functions two component cat code}
\end{figure}

Two-component cat codes with sufficiently large $\alpha$ are capable of correcting bosonic dephasing errors (see Subsection \ref{subsection:Bosonic dephasing errors}). To see why this is the case, let us recall the discrete Kraus operators of the bosonic dephasing errors (Eq.\ \eqref{eq:bosonic dephasing discrete Kraus representation}): 
\begin{align}
\hat{N}_{k} = \sqrt{ \frac{\sigma^{2k}}{k!} } e^{-\frac{\sigma^{2}}{2}\hat{n}^{2}} \hat{n}^{k}. 
\end{align}
Note that the first two Kraus operators are given by 
\begin{align}
\hat{N}_{0} = \hat{I} + \mathcal{O}( \sigma^{2} ) , \quad \hat{N}_{1} = \sigma \hat{a}^{\dagger}\hat{a} + \mathcal{O}( \sigma^{2} ). 
\end{align}
The relevant first-order dephasing error set is thus given by $\lbrace \hat{I} , \hat{a}^{\dagger}\hat{a} \rbrace$, i.e., $\hat{N}_{0} = \hat{I}$ (no error) and $\hat{N}_{1} = \hat{a}^{\dagger}\hat{a}$ (single dephasing), where the prefactor $\sigma$ in $\hat{N}_{1}$ is omitted. 

Since these dephasing error operators cannot change the parity of the excitation number, we have 
\begin{align}
\langle \mu^{(\alpha)}_{2-\textrm{cat}} | \hat{N}_{k}^{\dagger}\hat{N}_{k'} |\nu^{(\alpha)}_{2-\textrm{cat}}\rangle = 0,  \,\,\, \textrm{for all} \,\,\,  k\in \lbrace 0,1\rbrace \,\,\, \textrm{and} \,\,\, \mu\neq \nu.  \label{eq:two component cat code dephasing KL condition off diagonal}
\end{align}   
Also for $\mu=\nu$, $\langle \mu^{(\alpha)}_{2-\textrm{cat}} | \hat{N}_{k}^{\dagger}\hat{N}_{k'} |\mu^{(\alpha)}_{2-\textrm{cat}}\rangle$ is given by 
\begin{align}
&(k,k')=(0,0) : 
\nonumber\\
&\quad   \langle \mu^{(\alpha)}_{2-\textrm{cat}} | \mu^{(\alpha)}_{2-\textrm{cat}}\rangle = 1, 
\nonumber\\ 
&(k,k')=(0,1),(1,0) :   
\nonumber\\
&\quad \langle \mu^{(\alpha)}_{2-\textrm{cat}} | \hat{a}^{\dagger}\hat{a} | \mu^{(\alpha)}_{2-\textrm{cat}}\rangle =  \frac{N_{1-\mu}^{2-\textrm{cat}}(\alpha)}{N_{\mu}^{2-\textrm{cat}}(\alpha)} |\alpha|^{2} =  \frac{1-(-1)^{\mu} e^{-2|\alpha|^{2}}}{1+(-1)^{\mu} e^{-2|\alpha|^{2}} } |\alpha|^{2} , 
\nonumber\\
&(k,k')=(1,1) : 
\nonumber\\
&\quad \langle \mu^{(\alpha)}_{2-\textrm{cat}} | (\hat{a}^{\dagger}\hat{a})^{2} | \mu^{(\alpha)}_{2-\textrm{cat}}\rangle =  |\alpha|^{4} + \frac{N_{1-\mu}^{2-\textrm{cat}}(\alpha)}{N_{\mu}^{2-\textrm{cat}}(\alpha)} |\alpha|^{2} =  |\alpha|^{4} +  \frac{1-(-1)^{\mu} e^{-2|\alpha|^{2}}}{1+(-1)^{\mu} e^{-2|\alpha|^{2}} } |\alpha|^{2} . \label{eq:two component cat code dephasing KL condition diagonal}
\end{align}
Thus, if $e^{-2|\alpha|^{2}} \ll 1$, $\langle \mu^{(\alpha)}_{2-\textrm{cat}} | \hat{N}_{k}^{\dagger}\hat{N}_{k'} |\mu^{(\alpha)}_{2-\textrm{cat}}\rangle$ is independent of $\mu$ for all $k,k'\in\lbrace 0,1\rbrace$. Hence, the two-component cat code $\mathcal{C}_{2-\textrm{cat}}^{(\alpha)}$ satisfies the Knill-Laflamme condition for the first-order dephasing error set $\lbrace \hat{I} ,\hat{a}^{\dagger}\hat{a} \rbrace$ if $e^{-2|\alpha|^{2}} \ll 1$. This implies that two-component cat codes with sufficiently large $\alpha$ can correct dephasing errors. 

\subsubsection{Autonomous QEC of two-component cat codes} 

Recall Theorem \ref{theorem:knill-laflamme condition} and note that there will be a recovery method for the two-component cat code subject to a single dephasing error, since two-component cat codes with sufficiently large $\alpha$ satisfy the Knill-Laflamme condition for the first-order dephasing error set. A simple way to implement the dephasing recovery for the two-component cat code is to use an engineered two-photon dissipation to autonomously stabilize the code space \cite{Mirrahimi2014}:  
\begin{align}
\frac{d\hat{\rho}(t)}{dt} &= \kappa_{2\textrm{ph}} \mathcal{D}[\hat{a}^{2} - \alpha^{2}]  ( \hat{\rho}(t) ) . \label{eq:two component cat code engineered two photon dissipation}
\end{align}
Here, $\kappa_{2\textrm{ph}}$ is the engineered two-photon dissipation rate. Since the basis states of the two-component cat code $\mathcal{C}_{2-\textrm{cat}}^{(\alpha)}$ are annihilated by the jump operator $\hat{F}_{2\textrm{ph}} \equiv \hat{a}^{2}-\alpha^{2}$, i.e., 
\begin{align}
\hat{F}_{2\textrm{ph}} |\mu^{(\alpha)}_{2-\textrm{cat}}\rangle &= (\hat{a}^{2}-\alpha^{2}) |\mu^{(\alpha)}_{2-\textrm{cat}}\rangle = (\hat{a}^{2}-\alpha^{2}) \frac{1}{\sqrt{ 2N_{\mu}^{2-\textrm{cat}}(\alpha) }} ( |\alpha\rangle +(-1)^{\mu}|-\alpha\rangle ) =0, 
\end{align}
the two-photon dissipation $\mathcal{D}[\hat{a}^{2} - \alpha^{2}]$ stabilizes the two-component cat code manifold.

Let us add a bosonic dephasing error on top of the engineered dissipation:  
\begin{align}
\frac{d\hat{\rho}(t)}{dt} &= \Big{(} \kappa_{2\textrm{ph}} \mathcal{D}[\hat{a}^{2} - \alpha^{2}]   + \kappa_{\phi}\mathcal{D}[\hat{a}^{\dagger}\hat{a}] \Big{)}(\hat{\rho}(t)) . 
\end{align}
Here, $\kappa_{\phi}$ is the dephasing rate. Since dephasing errors cannot change the parity of the excitation number, it cannot induce logical bit-flip errors between $|0^{(\alpha)}_{2-\textrm{cat}}\rangle$ and $|1^{(\alpha)}_{2-\textrm{cat}}\rangle$ which have even and odd excitation number parity, respectively. The absence of logical bit-flip errors is directly related to the fact that the Knill-Laflamme conditions in Eq.\ \eqref{eq:two component cat code dephasing KL condition off diagonal} are exactly satisfied. However, since the two-component cat code $\mathcal{C}_{2-\textrm{cat}}^{(\alpha)}$ does not exactly satisfy the Knill-Laflamme conditions in Eq.\ \eqref{eq:two component cat code dephasing KL condition diagonal} (for any finite value of $\alpha$), there can be logical phase-flip errors. The logical phase-flip rate is computed in Ref.\ \cite{Mirrahimi2014} and is given by
\begin{align}
\gamma_{\textrm{phase-flip}} \xrightarrow{ \kappa_{\phi} \ll \kappa_{2\textrm{ph}} } \kappa_{\phi} \frac{|\alpha|^{2}}{ \sinh (2|\alpha|^{2}) } =\kappa_{\phi} \frac{2|\alpha|^{2}}{ e^{ 2|\alpha|^{2}} - e^{ -2|\alpha|^{2}} } .  
\end{align}  
Thus, the logical phase-flip rate decreases exponentially as $\alpha$ increases. In particular, if $e^{-2|\alpha|^{2}} \ll 1$, the logical phase-flip rate is negligible. This is consistent with the fact that the Knill-Laflamme conditions in Eq.\ \eqref{eq:two component cat code dephasing KL condition diagonal} is satisfied if $e^{-2|\alpha|^{2}} \ll 1$. While two-component cat codes can correct bosonic dephasing errors, they cannot correct excitation loss errors. One can readily see this by observing that a single-excitation loss maps an even cat state to an odd cat state and vice versa: 
\begin{align}
\hat{a} |0^{(\alpha)}_{2-\textrm{cat}}\rangle &\propto   \hat{a} ( |\alpha\rangle + |-\alpha\rangle ) =  \alpha (|\alpha\rangle - |-\alpha\rangle ) \propto \hat{a} |1^{(\alpha)}_{2-\textrm{cat}}\rangle, 
\nonumber\\
\hat{a} |1^{(\alpha)}_{2-\textrm{cat}}\rangle &\propto   \hat{a} ( |\alpha\rangle - |-\alpha\rangle ) =  \alpha (|\alpha\rangle + |-\alpha\rangle ) \propto \hat{a} |0^{(\alpha)}_{2-\textrm{cat}}\rangle. 
\end{align}
Thus, excitation loss errors cause logical bit-flip errors to two-component cat codes. See Subsection \ref{subsection:Concatenation of a cat code with a multi-qubit code} for a review of several recent proposals for dealing with the residual bit-flip errors due to excitation loss errors.

Let us now move on to the experimental realization of the engineered two-photon dissipation $\mathcal{D}[\hat{a}^{2}-\alpha^{2}]$. Note that 
\begin{align}
\frac{d\hat{\rho}(t)}{dt} &= \kappa_{2\textrm{ph}}\mathcal{D}[\hat{a}^{2}-\alpha^{2}]( \hat{\rho}(t) ) = \frac{\kappa_{2\textrm{ph}}}{2} \Big{[}  \alpha^{2}(\hat{a}^{\dagger})^{2} - \alpha^{*2}\hat{a}^{2}  , \hat{\rho}(t) \Big{]} + \kappa_{2\textrm{ph}} \mathcal{D}[ \hat{a}^{2} ]( \hat{\rho}(t) ). 
\end{align}
The first term can be implemented by using a two-photon driving Hamiltonian $\hat{H} = i\kappa_{2\textrm{ph}} (\alpha^{2}(\hat{a}^{\dagger})^{2} - \alpha^{*2}\hat{a}^{2}  ) /2$, which is a generator of the single-mode squeezing operation. Since the second term $\mathcal{D}[\hat{a}^{2}]$ is dissipative, it cannot be generated by using only Hamiltonian interactions. Instead, we need a fast-decaying ancilla system. To be more concrete, consider an ancilla bosonic mode which is described by the annihilation and creation operators $\hat{b}$ and $\hat{b}^{\dagger}$. Then, let us assume that we can engineer the following Hamiltonian interaction between the mode $\hat{a}$ and the ancilla mode $\hat{b}$: 
\begin{align}
\hat{H}_{\textrm{int}} = g  \hat{a}^{2}\hat{b}^{\dagger} + g^{*} (\hat{a}^{\dagger})^{2}\hat{b}, 
\end{align}
and thus the time evolution of the joint system is described by 
\begin{align}
\frac{d\hat{\rho}_{T}(t)}{dt} &= -i[ \hat{H}_{\textrm{int}} , \hat{\rho}_{T}(t) ]  +\kappa_{b} \mathcal{D}[\hat{b}]( \hat{\rho}_{T}(t) ). 
\end{align}
Here, $\hat{\rho}_{T}(t)$ is the density matrix of the joint system of mode $a$ and $b$. Also, $\kappa_{b}$ is the decay rate of the fast-decaying ancilla mode $b$. If the coupling strength $|g|$ is much smaller than the decay rate $\kappa_{b}$ of the ancilla mode (i.e., $|g|\ll \kappa_{b}$), one can show by using adiabatic elimination \cite{Verstraete2009,Reiter2012,Zanardi2016} that $\hat{\rho}_{T}(t)$ is approximately given by $\hat{\rho}_{T}(t) = \hat{\rho}(t)\otimes |0\rangle\langle 0|_{b}$ where the system density matrix $\hat{\rho}(t)$ evolves under the desired two-photon dissipation. 
\begin{align}
\frac{d\hat{\rho}(t)}{dt} &= \frac{4|g|^{2}}{\kappa_{b}} \mathcal{D}[\hat{a}^{2}]( \hat{\rho}(t) ) . 
\end{align}

Putting all these components together, the engineered two-photon dissipation $\mathcal{D}[\hat{a}^{2} - \alpha^{2}]$ was realized experimentally in circuit QED systems \cite{Leghtas2015,Touzard2018,Lescanne2019}. In particular, Ref.\ \cite{Touzard2018} demonstrated a coherent quantum oscillation between the protected logical states of the two-component cat code. Also, Ref.\ \cite{Lescanne2019} demonstrated an exponential suppression of the phase-flip error in the stabilized two-component cat code manifold.

\subsubsection{Four-component cat codes} 

While two-component cat codes cannot correct excitation loss errors, four-component cat codes \cite{Leghtas2013} are capable of correcting excitation loss errors. Logical states of the four-component cat code $\mathcal{C}_{4-\textrm{cat}}^{(\alpha)}$ are given by
\begin{align}
|0^{(\alpha)}_{\textrm{4-\textrm{cat}}}\rangle &= \frac{1}{\sqrt{ 4N_{0}^{4-\textrm{cat}}(\alpha) }} ( |\alpha\rangle + |i\alpha\rangle + |-\alpha\rangle + |-i\alpha\rangle ), 
\nonumber\\
|1^{(\alpha)}_{\textrm{4-\textrm{cat}}}\rangle &= \frac{1}{\sqrt{ 4N_{1}^{4-\textrm{cat}}(\alpha) }} ( |\alpha\rangle - |i\alpha\rangle +|-\alpha\rangle  - |-i\alpha\rangle ). 
\end{align}
Here, the normalization constant $N_{\mu}^{4-\textrm{cat}}(\alpha)$ is given by
\begin{align}
N_{\mu}^{4-\textrm{cat}}(\alpha) &= 1 + e^{-2|\alpha|^{2}} +(-1)^{\mu} 2e^{-|\alpha|^{2}}\cos|\alpha|^{2}. 
\end{align}
Wigner functions of the logical states of the four-component cat code are shown in Fig.\ \ref{fig:wigner functions four component cat code}. 

\begin{figure}[t!]
\centering
\includegraphics[width=6.0in]{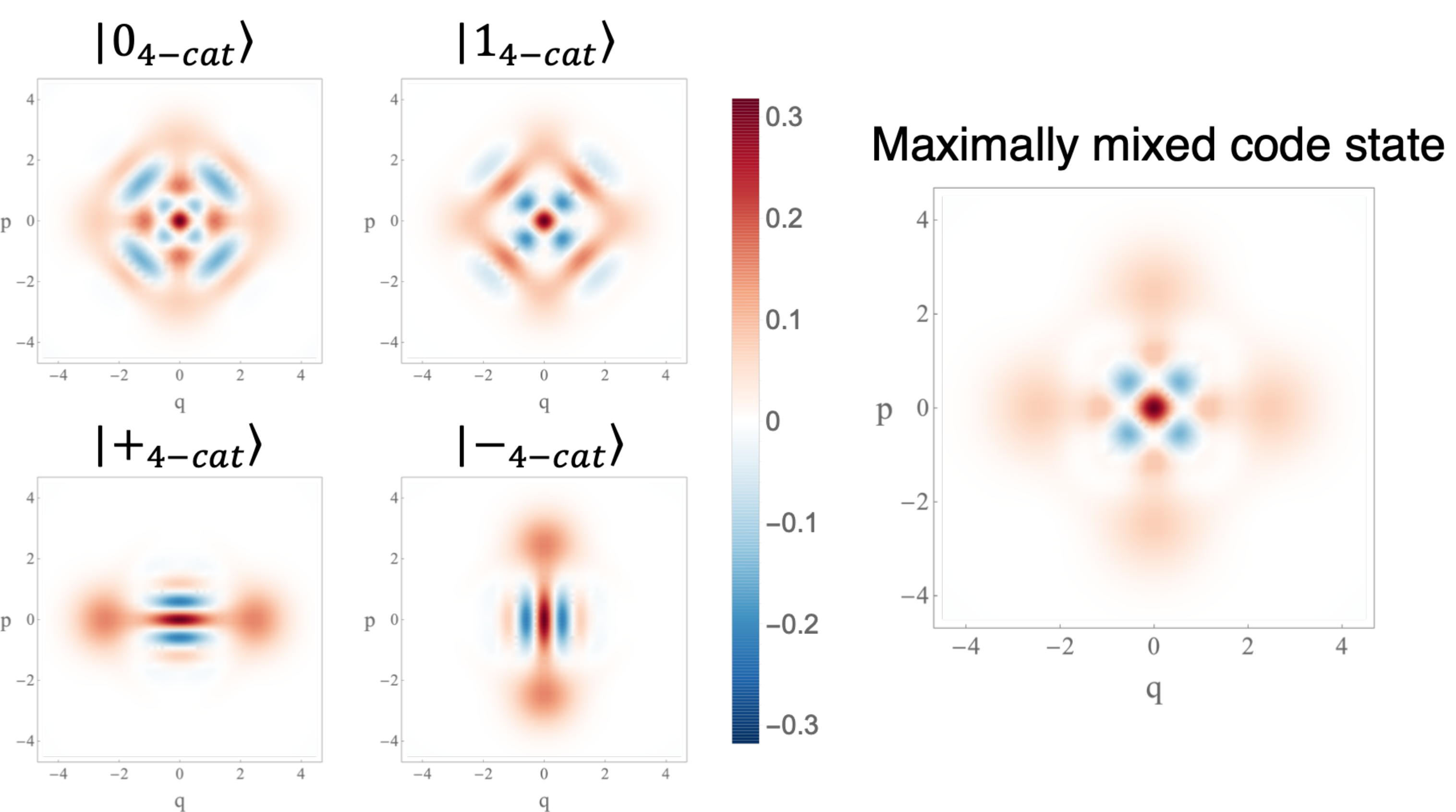}
\caption{Wigner functions of the logical states of the four-component cat code $\mathcal{C}_{4-\textrm{cat}}^{(\alpha)}$ with $\alpha = \sqrt{3}$. The maximally mixed code state is defined as the projection operator to the code space divided by $2$.  }
\label{fig:wigner functions four component cat code}
\end{figure}

Unlike the two-component cat code, both logical states of the four-component cat code have even excitation number parity. Also, this means that the logical basis states of the four-component cat code is invariant under the $180\degree$ rotation, i.e., 
\begin{align}
e^{i\pi \hat{n}} |\mu^{(\alpha)}_{\textrm{4-\textrm{cat}}}\rangle, \,\,\, \textrm{for all} \,\,\, \mu\in\lbrace 0,1 \rbrace, 
\end{align}  
and thus the four-component cat code $\mathcal{C}_{4-\textrm{cat}}^{(\alpha)}$ is an example of rotation-symmetric bosonic codes \cite{Grimsmo2019}. Note also that $|0^{(\alpha)}_{\textrm{4-\textrm{cat}}}\rangle$ has $0$ excitations mod $4$ and  $|1^{(\alpha)}_{\textrm{4-\textrm{cat}}}\rangle$ has $2$ excitations mod $4$, so they are clearly orthogonal to each other.

Four-component cat codes with sufficiently large $\alpha$ are capable of correcting excitation loss errors (see Subsection \ref{subsection:Excitation loss errors}), which are dominant error sources in many realistic bosonic modes. To see why this is the case, let us recall the Kraus representation of the excitation loss errors (Eq.\ \eqref{eq:excitation loss Kraus operators}), i.e., 
\begin{align}
\hat{N}_{\ell} = \sqrt{ \frac{(1-e^{-\kappa t})^{\ell} }{\ell!}} e^{-\frac{\kappa t}{2} \hat{n} } \hat{a}^{\ell}. 
\end{align}   
Note that the first two Kraus operators are given by
\begin{align}
\hat{N}_{0} &= \hat{I} + \mathcal{O}(\kappa t) , \quad \hat{N}_{1} = \sqrt{\kappa t} \hat{a} + \mathcal{O}(\kappa t). 
\end{align}
The relevant first-order excitation loss error set is thus given by $\lbrace \hat{I},\hat{a} \rbrace$, i.e., $\hat{N}_{0} = \hat{I}$ (no error) and $\hat{N}_{1} = \hat{a}$ (single-excitation loss), where the prefactor $\sqrt{\kappa t}$ in $\hat{N}_{1}$ is omitted. 

Upon a single-excitation loss, the logical states of the four-component cat code are transformed into the following error states: 
\begin{align}
\hat{a}|0^{(\alpha)}_{\textrm{4-\textrm{cat}}}\rangle &= \frac{\alpha}{\sqrt{ 4N_{0}^{4-\textrm{cat}}(\alpha) }}  ( |\alpha\rangle + i |i\alpha\rangle - |-\alpha\rangle -i |-i\alpha\rangle ),
\nonumber\\
\hat{a}|1^{(\alpha)}_{\textrm{4-\textrm{cat}}}\rangle &= \frac{\alpha}{\sqrt{ 4N_{1}^{4-\textrm{cat}}(\alpha) }}  ( |\alpha\rangle - i |i\alpha\rangle - |-\alpha\rangle +i |-i\alpha\rangle ) . 
\end{align}
While the code states have even excitation numbers, their corresponding error states have odd excitation numbers. Thus, we can distinguish the no error event $\hat{N}_{0} = \hat{I}$ from the single-excitation loss event $\hat{N}_{1} = \hat{a}$ by measuring the excitation number parity of the system. 

Let us now investigate the error correction capability of four-component cat codes by inspecting the Knill-Laflamme condition for the first-order excitation loss error set $\lbrace \hat{I},\hat{a} \rbrace$. Observe that the error state derived from the logical zero state has $3$ excitations mod $4$ and the error state derived from the logical one state has $1$ excitations mod $4$. Thus, all the relevant states $|0^{(\alpha)}_{\textrm{4-\textrm{cat}}}\rangle$, $|1^{(\alpha)}_{\textrm{4-\textrm{cat}}}\rangle$, $\hat{a}|0^{(\alpha)}_{\textrm{4-\textrm{cat}}}\rangle$, and $\hat{a}|1^{(\alpha)}_{\textrm{4-\textrm{cat}}}\rangle$ have different excitation numbers modulo $4$ so are mutually orthogonal. Thus, we have
\begin{align}
\langle \mu_{4-\textrm{cat}}^{(\alpha)} | \hat{N}_{\ell}^{\dagger}\hat{N}_{\ell'} |\nu_{4-\textrm{cat}}^{(\alpha)}\rangle =0, \,\,\, \textrm{for all} \,\,\, \ell,\ell'\in\lbrace 0,1 \rbrace \,\,\, \textrm{and} \,\,\, \mu\neq \nu, 
\end{align}
and the Knill-Laflamme condition is satisfied for all $\mu\neq \nu$. 

To analyze the $\mu=\nu$ case, it is convenient to define the following normalized error states 
\begin{align}
|0_{4-\textrm{cat},e}^{(\alpha)}\rangle &= \frac{1}{\sqrt{ 4N_{0,e}^{4-\textrm{cat}}(\alpha) }} ( |\alpha\rangle + i |i\alpha\rangle - |-\alpha\rangle -i |-i\alpha\rangle ), 
\nonumber\\
|1_{4-\textrm{cat},e}^{(\alpha)}\rangle &= \frac{1}{\sqrt{ 4N_{1,e}^{4-\textrm{cat}}(\alpha) }} ( |\alpha\rangle - i |i\alpha\rangle - |-\alpha\rangle +i |-i\alpha\rangle ), 
\end{align}
where the normalization constant $N_{\mu,e}^{4-\textrm{cat}}(\alpha)$ is given by
\begin{align}
N_{0,e}^{4-\textrm{cat}}(\alpha) &= 1-e^{-2|\alpha|^{2}} - 2(-1)^{\mu} e^{-|\alpha|^{2}}\sin|\alpha|^{2} . 
\end{align}
Then, the unnormalized error state $\hat{a}|\mu_{4-\textrm{cat}}^{(\alpha)}\rangle$ is given by
\begin{align}
\hat{a}|\mu_{4-\textrm{cat}}^{(\alpha)}\rangle &= \alpha \sqrt{ \frac{N_{\mu,e}^{4-\textrm{cat}}(\alpha)}{N_{\mu}^{4-\textrm{cat}}(\alpha)} } |\mu_{4-\textrm{cat},e}^{(\alpha)}\rangle. 
\end{align} 
The relevant $\mu=\nu$ terms in the Knill-Laflamme condition, i.e., $\langle \mu_{4-\textrm{cat}}^{(\alpha)} | \hat{N}_{\ell}^{\dagger}\hat{N}_{\ell'} |\mu_{4-\textrm{cat}}^{(\alpha)}\rangle$, are given by 
\begin{align}
&(\ell,\ell')=(0,0):
\nonumber\\
&\quad \langle \mu_{4-\textrm{cat}}^{(\alpha)} | \mu_{4-\textrm{cat}}^{(\alpha)}\rangle = 1, 
\nonumber\\
&(\ell,\ell')=(0,1),(1,0):
\nonumber\\
&\quad \langle \mu_{4-\textrm{cat}}^{(\alpha)} | \hat{a} | \mu_{4-\textrm{cat}}^{(\alpha)}\rangle = 0, 
\nonumber\\
&(\ell,\ell')=(1,1):
\nonumber\\
&\quad \langle \mu_{4-\textrm{cat}}^{(\alpha)} | \hat{a}^{\dagger}\hat{a} | \mu_{4-\textrm{cat}}^{(\alpha)}\rangle =  \frac{N_{\mu,e}^{4-\textrm{cat}}(\alpha)}{N_{\mu}^{4-\textrm{cat}}(\alpha)}  |\alpha|^{2} = \frac{1-e^{-2|\alpha|^{2}} - 2(-1)^{\mu}e^{-|\alpha|^{2}}\sin|\alpha|^{2}}{1+e^{-2|\alpha|^{2}} + 2(-1)^{\mu}e^{-|\alpha|^{2}}\cos|\alpha|^{2} } |\alpha|^{2} 
\nonumber\\
&\quad\qquad\qquad\qquad\qquad\qquad\qquad\qquad\qquad = \frac{ \sinh|\alpha|^{2} - (-1)^{\mu}\sin|\alpha|^{2}}{\cosh|\alpha|^{2} + (-1)^{\mu}\cos|\alpha|^{2} } |\alpha|^{2}  . 
\end{align}      
Similarly as in the case of the two-component cat code, $\langle \mu_{4-\textrm{cat}}^{(\alpha)} | \hat{a}^{\dagger}\hat{a} | \mu_{4-\textrm{cat}}^{(\alpha)}\rangle$ is $\mu$-independent if $\alpha$ is sufficiently large such that $e^{-|\alpha|^{2}} \ll 1$. 

What distinguishes the four-component cat codes from the two-component cat codes is that in the former case, there are values of $\alpha$ where $\langle \mu_{4-\textrm{cat}}^{(\alpha)} | \hat{a}^{\dagger}\hat{a} | \mu_{4-\textrm{cat}}^{(\alpha)}\rangle$ is precisely $\mu$-independent, i.e., 
\begin{align}
&\frac{ \sinh|\alpha|^{2} -\sin|\alpha|^{2}}{\cosh|\alpha|^{2} + \cos|\alpha|^{2} }  = \frac{ \sinh|\alpha|^{2} + \sin|\alpha|^{2}}{\cosh|\alpha|^{2} - \cos|\alpha|^{2} } 
\nonumber\\
&\leftrightarrow -\cosh|\alpha|^{2} \sin|\alpha|^{2} = \sinh|\alpha|^{2}\cos|\alpha|^{2} 
\nonumber\\
&\leftrightarrow \tan|\alpha|^{2} = -\tanh|\alpha|^{2} . 
\end{align}
The first three non-trivial solutions to $\tan|\alpha|^{2} = -\tanh|\alpha|^{2}$ are given by 
\begin{align}
|\alpha^{\star}_{1}| &= 1.538, \quad |\alpha^{\star}_{2}| = 2.345 , \quad   |\alpha^{\star}_{3}| =2.939. 
\end{align} 
These values are also called the sweet spots of the four-component cat code \cite{Li2017,Albert2018}.  

\subsubsection{Active QEC of four-component cat codes}

As was implied earlier, active QEC of four-component cat codes can be done by measuring the excitation number parity operator 
\begin{align}
\hat{\Pi}_{2} &\equiv e^{i \pi \hat{n}} = \sum_{n=0}^{\infty} |n\rangle\langle n| \times \begin{cases}
+1 & n \textrm{ even}\\
-1 & n \textrm{ odd}
\end{cases} . 
\end{align}
Note that the parity operator $\hat{\Pi}_{2}$ is equivalent to the phase rotation by $180\degree$. The parity operator $\hat{\Pi}_{2}$ is also a stabilizer of the four-component cat code, in the sense that any logical state of the four-component cat code is stabilized by $\hat{\Pi}_{2}$, i.e., 
\begin{align}
\hat{\Pi}_{2} |\psi_{4-\textrm{cat}}^{(\alpha)}\rangle = |\psi_{4-\textrm{cat}}^{(\alpha)}\rangle , \,\,\,  \textrm{for all} \,\,\, |\psi_{4-\textrm{cat}}^{(\alpha)}\rangle\in \mathcal{C}_{4-\textrm{cat}}^{(\alpha)}. 
\end{align}
Thus, whenever the parity measurement yields an odd parity measurement outcome (i.e., $\hat{\Pi}_{2} = -1$), we are alerted that some error has happened.  

To see how parity measurement can be used to correct excitation loss errors, let us consider the excitation loss error $e^{\kappa t \mathcal{D}[\hat{a}]} \hat{\rho} = \sum_{\ell =0}^{\infty} \hat{N}_{\ell} \hat{\rho} \hat{N}_{\ell}^{\dagger}$. When there is no excitation loss error (i.e., $\ell=0$), the logical state $|\mu_{4-\textrm{cat}}^{(\alpha)}\rangle$ undergoes the no-jump evolution. 
\begin{align}
\hat{N}_{0}|\mu_{4-\textrm{cat}}^{(\alpha)}\rangle &=  \frac{1}{ \sqrt{4N_{\mu}^{4-\textrm{cat}}(\alpha) } } e^{-\frac{\kappa t}{2} \hat{n} } \Big{(} |\alpha\rangle +(-1)^{\mu}|i\alpha\rangle +|-\alpha\rangle +(-1)^{\mu}|-i\alpha\rangle  \Big{)}
\nonumber\\
&= \frac{1}{ \sqrt{4N_{\mu}^{4-\textrm{cat}}(\alpha) } }  \Big{(} |\alpha e^{-\frac{\kappa t}{2}} \rangle +(-1)^{\mu}|i\alpha e^{-\frac{\kappa t}{2}} \rangle +|-\alpha e^{-\frac{\kappa t}{2}} \rangle +(-1)^{\mu}|-i\alpha e^{-\frac{\kappa t}{2}} \rangle  \Big{)} 
\nonumber\\
&= \sqrt{ \frac{N_{\mu}^{4-\textrm{cat}}(\alpha e^{-\frac{\kappa t}{2}} )}{N_{\mu}^{4-\textrm{cat}}(\alpha)} } |\mu_{4-\textrm{cat}}^{ ( \alpha e^{-\frac{\kappa t}{2}} ) }\rangle . 
\end{align}    
That is, due to the decay term $e^{-\frac{\kappa t}{2} \hat{n} }$, we end up with a code state with smaller amplitude $\alpha' = \alpha e^{-\frac{\kappa t}{2}}$. Also in this case, the states are still in the even-parity subspace. Note that since $\sqrt{ N_{\mu}^{4-\textrm{cat}}(\alpha e^{-\frac{\kappa t}{2}} )/ N_{\mu}^{4-\textrm{cat}}(\alpha) }$ is $\mu$-independent in the $\kappa t \rightarrow 0$ limit, we have
\begin{align}
\hat{N}_{0} ( c_{0}|0_{4-\textrm{cat}}^{(\alpha)}\rangle + c_{1}|1_{4-\textrm{cat}}^{(\alpha)}\rangle ) &\xrightarrow{ \kappa t \rightarrow 0 } c_{0}|0_{4-\textrm{cat}}^{(\alpha)}\rangle + c_{1}|1_{4-\textrm{cat}}^{(\alpha)}\rangle, 
\end{align}
i.e., no distortion of the encoded logical information. 

Let us now consider the one-excitation loss event ($\ell =1$). Upon a single-excitation loss, the logical state $|\mu_{4-\textrm{cat}}^{(\alpha)}\rangle$ is mapped to 
\begin{align}
\hat{N}_{1} |\mu_{4-\textrm{cat}}^{(\alpha)}\rangle &= \sqrt{ 1-e^{-\kappa t} }   \frac{ e^{-\frac{\kappa t}{2} \hat{n} } \hat{a} }{ \sqrt{4N_{\mu}^{4-\textrm{cat}}(\alpha) } } \Big{(} |\alpha\rangle +(-1)^{\mu}|i\alpha\rangle +|-\alpha\rangle +(-1)^{\mu}|-i\alpha\rangle  \Big{)} 
\nonumber\\
&=    \frac{ \sqrt{ 1-e^{-\kappa t} } \alpha }{ \sqrt{4N_{\mu}^{4-\textrm{cat}}(\alpha) } }  \Big{(} |\alpha e^{-\frac{\kappa t}{2}} \rangle +i(-1)^{\mu}|i\alpha e^{-\frac{\kappa t}{2}} \rangle -|-\alpha e^{-\frac{\kappa t}{2}} \rangle -i(-1)^{\mu}|-i\alpha e^{-\frac{\kappa t}{2}} \rangle  \Big{)}  
\nonumber\\
&= \sqrt{ 1-e^{-\kappa t} } \alpha \sqrt{ \frac{N_{\mu,e}^{4-\textrm{cat}}(\alpha e^{-\frac{\kappa t}{2}})}{N_{\mu}^{4-\textrm{cat}}(\alpha)} } |\mu_{4-\textrm{cat},e}^{( \alpha e^{-\frac{\kappa t}{2}} )} \rangle. 
\end{align}
Again, due to the decay term $e^{-\frac{\kappa t}{2} \hat{n} }$, we end up with an error state with smaller amplitude $\alpha' = \alpha e^{-\frac{\kappa t}{2}}$. Also, due to the annihilation operator $\hat{a}$, the states are now in the odd-parity subspace and therefore the single-excitation loss event can be flagged by measuring the parity operator $\hat{\Pi}_{2} = e^{i\pi \hat{n}}$. Note also that at the sweet spots, $\sqrt{ N_{\mu,e}^{4-\textrm{cat}}(\alpha e^{-\frac{\kappa t}{2}})/ N_{\mu}^{4-\textrm{cat}}(\alpha) }$ is $\mu$-independent in the $\kappa t \rightarrow 0$ limit. Hence, we have 
\begin{align}
&\hat{N}_{1} ( c_{0}|0_{4-\textrm{cat}}^{(\alpha)}\rangle + c_{1}|1_{4-\textrm{cat}}^{(\alpha)}\rangle ) 
\nonumber\\
&\xrightarrow{ \kappa t \rightarrow 0 } \sqrt{\kappa t} \alpha  \Big{(} c_{0}\sqrt{ \frac{N_{0,e}^{4-\textrm{cat}}(\alpha )}{N_{0}^{4-\textrm{cat}}(\alpha)}  } |0_{4-\textrm{cat},e}^{( \alpha  )} \rangle + c_{1}\sqrt{ \frac{N_{1,e}^{4-\textrm{cat}}(\alpha )}{N_{1}^{4-\textrm{cat}}(\alpha)}  } |1_{4-\textrm{cat},e}^{( \alpha  )} \rangle  \Big{)} 
\nonumber\\
&\propto c_{0}  |0_{4-\textrm{cat},e}^{( \alpha  )} \rangle + c_{1} |1_{4-\textrm{cat},e}^{( \alpha  )} \rangle, 
\end{align}
i.e., no distortion of the encoded logical information at sweet spots (or when $\alpha$ satisfies $\tan|\alpha|^{2} = -\tanh|\alpha|^{2}$). 

Lastly, when there is a two-excitation loss event ($\ell =2$), the logical state is mapped to 
\begin{align}
\hat{N}_{2} |\mu_{4-\textrm{cat}}^{(\alpha)}\rangle &= \sqrt{ \frac{ (1-e^{-\kappa t})^{2} }{2} } \frac{ e^{-\frac{\kappa t}{2} \hat{n} } \hat{a}^{2} }{ \sqrt{4N_{\mu}^{4-\textrm{cat}}(\alpha) } }   \Big{(} |\alpha\rangle +(-1)^{\mu}|i\alpha\rangle +|-\alpha\rangle +(-1)^{\mu}|-i\alpha\rangle  \Big{)} 
\nonumber\\
&= \frac{ (1-e^{-\kappa t})\alpha^{2} }{ \sqrt{8N_{\mu}^{4-\textrm{cat}}(\alpha) } }   \Big{(} |\alpha e^{-\frac{\kappa t}{2}} \rangle -(-1)^{\mu}|i\alpha e^{-\frac{\kappa t}{2}}\rangle +|-\alpha e^{-\frac{\kappa t}{2}} \rangle -(-1)^{\mu}|-i\alpha e^{-\frac{\kappa t}{2}} \rangle  \Big{)} 
\nonumber\\
&= \frac{(1-e^{-\kappa t}) \alpha^{2} }{\sqrt{2}} \sqrt{ \frac{N_{1-\mu}^{4-\textrm{cat}}(\alpha e^{-\frac{\kappa t}{2}})}{N_{\mu}^{4-\textrm{cat}}(\alpha)} } |(1-\mu)_{4-\textrm{cat}}^{( \alpha e^{-\frac{\kappa t}{2}} ) } \rangle .
\end{align}
Similarly as above, the decay term $e^{-\frac{\kappa t}{2} \hat{n} }$ reduces the amplitude from $\alpha$ to $\alpha ' = \alpha e^{-\frac{\kappa t}{2}}$. Also, the states are in the even-parity subspace after a two-excitation loss event. In this case, however, the logical zero state is mapped to a logical one state with a smaller amplitude, and the logical one state is mapped to a logical zero state with a smaller amplitude. Two-excitation loss events thus cause a logical bit-flip error. Since these bit-flip events are not flagged by the parity measurement, the performance of the four-component cat code is ultimately limited by the two-excitation (or more) loss events. 

It is clear by now that in the error recovery process, we should take care of both the overall amplitude damping due to the decay term $e^{-\frac{\kappa t}{2} \hat{n} }$ and excitation losses due to the loss term $\hat{a}$. Note that single-excitation loss events can be addressed by measuring the parity operator $\hat{\Pi}_{2} = e^{i\pi\hat{n}}$ in a non-destructive way. In particular, all single-excitation loss events will be flagged this way. However, two-excitation (or more) loss events will not be detected. 

In addition to measuring the parity operator, we should also recover the reduced amplitude $\alpha ' = \alpha e^{-\frac{\kappa t}{2}}$ back to $\alpha$. Such an amplitude recovery can be done by performing a unitary operation $\hat{U}_{\alpha' \rightarrow \alpha}$ that has the following property: 
\begin{align}
\hat{U}_{\alpha' \rightarrow \alpha} |0_{4-\textrm{cat}}^{(\alpha')}\rangle &= |0_{4-\textrm{cat}}^{(\alpha)}\rangle,
\nonumber\\
\hat{U}_{\alpha' \rightarrow \alpha} |1_{4-\textrm{cat}}^{(\alpha')}\rangle &= |1_{4-\textrm{cat}}^{(\alpha)}\rangle, 
\nonumber\\
\hat{U}_{\alpha' \rightarrow \alpha} |0_{4-\textrm{cat},e}^{(\alpha')}\rangle &= |0_{4-\textrm{cat},e}^{(\alpha)}\rangle,
\nonumber\\
\hat{U}_{\alpha' \rightarrow \alpha} |1_{4-\textrm{cat},e}^{(\alpha')}\rangle &= |1_{4-\textrm{cat},e}^{(\alpha)}\rangle. \label{eq:four component cat code amplitude recovery}
\end{align}

Note that one could consider coherently mapping the amplitude-recovered error states $|0_{4-\textrm{cat},e}^{(\alpha)}\rangle$ and $|1_{4-\textrm{cat},e}^{(\alpha)}\rangle$ back to the code states $|0_{4-\textrm{cat}}^{(\alpha)}\rangle$ and $|1_{4-\textrm{cat}}^{(\alpha)}\rangle$ by using a unitary operator when the parity measurement yields an odd parity outcome. However, we remark that it is not necessary to physically map these error states back to the code states because we can simply keep track of the classical data of the parity measurement outcomes and interpret the quantum data appropriately in reference to the parity measurement outcomes. 

A thorough analysis of the above error recovery process (i.e., parity measurement and amplitude recovery) is given in Ref.\ \cite{Li2017} and we do not review it here. Instead, we will discuss below how the parity measurement and an amplitude recovery operation can be implemented experimentally in circuit QED systems.  

\subsubsection{Experimental realization of the parity measurement and amplitude recovery}

Note that the parity operator $\hat{\Pi}_{2} = e^{i\pi\hat{n}}$ is a unitary operator and satisfies $(\hat{\Pi}_{2})^{2} = \hat{I}$. In general, any unitary operator $\hat{U}$ satisfying $\hat{U}^{2} = \hat{I}$ can be measured in a non-destructive way by using an ancilla qubit (see Fig.\ \ref{fig:one bit phase estimation circuit}). More explicitly, 
\begin{align}
|\psi\rangle  |+\rangle = \frac{1}{\sqrt{2}} |\psi\rangle  ( |0\rangle + |1\rangle )   &\xrightarrow{ \textrm{controlled-}\hat{U} } \frac{1}{\sqrt{2}} |\psi\rangle|0\rangle + \frac{1}{\sqrt{2}}\hat{U}|\psi\rangle|1\rangle 
\nonumber\\
&\qquad\qquad\quad = \frac{1}{2}(\hat{I}+\hat{U})|\psi\rangle|+\rangle + \frac{1}{2}(\hat{I}-\hat{U})|\psi\rangle|-\rangle . 
\end{align}
One can readily see that $(\hat{I} \pm \hat{U}) / 2$ is the projection operator to the $\hat{U}=\pm 1$ subspace. Thus, if we measure the $|\pm \rangle$ state at the end of the circuit, we are projecting the system to the $\hat{U}=\pm 1$ subspace, hence measuring the unitary operator $\hat{U}$ in a non-destructive way. Specializing the circuit in Fig.\ \ref{fig:one bit phase estimation circuit} to the case of $\hat{U} = \hat{\Pi}_{2} = e^{i\pi \hat{n}}$, we can realize that we need an ancilla qubit prepared in the $|+\rangle$ state, the ability to perform controlled $180\degree$ rotation
\begin{align}
e^{ i\pi\hat{n} |1\rangle\langle 1| } = \hat{I} \otimes  |0\rangle\langle 0| + e^{i\pi\hat{n}}\otimes |1\rangle\langle 1|, 
\end{align} 
and the ability to measure the ancilla qubit in the $X$ basis. In general, the most challenging step is to perform the rotation of a bosonic mode conditioned on the ancilla qubit state. 

\begin{figure}[t!]
\centering
\includegraphics[width=2.5in]{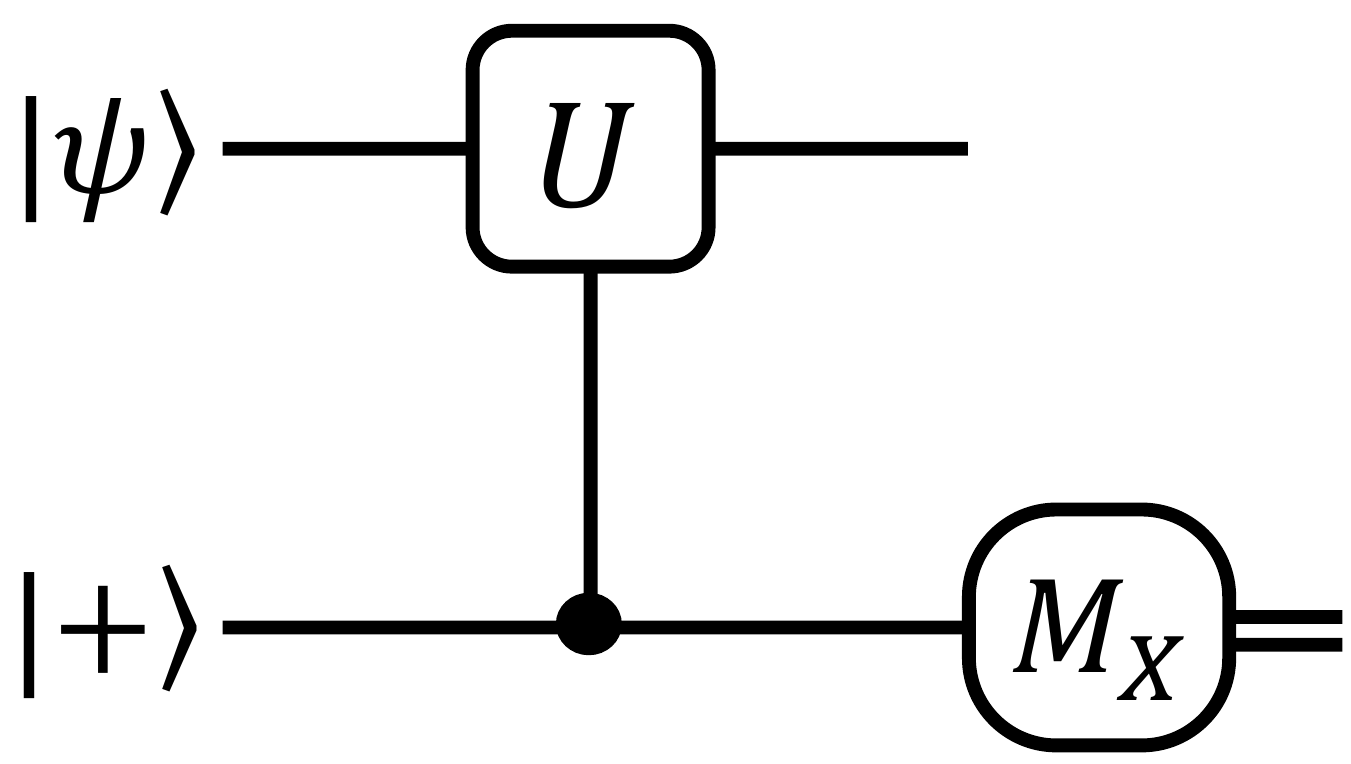}
\caption{A circuit for quantum non-demolition measurement of a unitary operator $\hat{U}$ that satisfies $\hat{U}^{2} = \hat{I}$. $|+\rangle \equiv \frac{1}{\sqrt{2}} (|0\rangle + |1\rangle)$ and $\mathcal{M}_{X}$ represents the Pauli X measurement (or measurement of a qubit in the $|+\rangle,|-\rangle$ basis). Note that this circuit is equivalent to the $1$-bit phase estimation circuit given in Fig.\ \ref{fig:n bit phase estimation circuit}(a).  }
\label{fig:one bit phase estimation circuit}
\end{figure}

In circuit QED systems, one can use a microwave cavity mode as a bosonic mode and a transmon qubit as an ancilla qubit to implement bosonic QEC. In the dispersive coupling regime, a cavity mode and a transmon qubit are coupled via the dispersive coupling, i.e., 
\begin{align}
\hat{H} = \omega_{C}\hat{a}^{\dagger}\hat{a} + \omega_{T} |e\rangle\langle e| -\chi \hat{a}^{\dagger}\hat{a} |e\rangle\langle e| . \label{eq:cavity-transmon Hamiltonian}
\end{align}
Here, $\omega_{C}$ is the frequency of the cavity mode, $\omega_{T}$ is the frequency of the transmon, and $\chi$ is the strength of the dispersive coupling. Also, $\hat{a}$ and $\hat{a}^{\dagger}$ are the annihilation and creation operators of the cavity mode and $|e\rangle\langle e|$ is the projection operator onto the excited state of the transmon qubit. Thus, in the interaction picture, the joint system of the cavity mode and the transmon system undergoes the desired controlled rotation: 
\begin{align}
\hat{U}(t) = e^{i \chi t \hat{a}^{\dagger}\hat{a} |e\rangle\langle e|} = \hat{I} \otimes |g\rangle\langle g| +  e^{i\chi t \hat{n}} \otimes |e\rangle\langle e|.
\end{align}
Here, $|g\rangle$ is the ground state of the transmon qubit. Hence, by letting the system evolve for the time interval $\Delta t = \pi /\chi$, one can implement the desired qubit-controlled $180\degree$ rotation of a bosonic mode, which is precisely what we need for the parity measurement. 

The parity measurement scheme discussed above was realized experimentally in a circuit QED system \cite{Sun2014}. Subsequently, the active QEC of the four-component cat code was implemented in a similar experimental setup based on the above parity measurement scheme \cite{Ofek2016}. However, note that in Ref.\ \cite{Ofek2016}, an amplitude recovery operation $\hat{U}_{\alpha'\rightarrow \alpha}$ (see Eq.\ \eqref{eq:four component cat code amplitude recovery}) was not implemented so the overall amplitude still decays over time exponentially, i.e., $\alpha \rightarrow \alpha' = \alpha e^{-\frac{\kappa t}{2}}$. We remark that the experiment in Ref.\ \cite{Ofek2016} nevertheless achieved the ``break-even'' point (i.e., outperforming the best error-uncorrected physical qubit element by using QEC). This is because the amplitude decay is deterministic unlike excitation losses which are stochastic. Thus in Ref.\ \cite{Ofek2016}, the adverse effects of the deterministic amplitude decay were taken into account in the classical decoding process and therefore were mitigated. However, an amplitude recovery operation $\hat{U}_{\alpha'\rightarrow \alpha}$ is still essential if we want to go way beyond the break-even point. This is because otherwise the amplitude $\alpha$ will almost vanish as the elapsed time $t$ becomes much larger than the natural lifetime of the cavity mode $1/\kappa$ and thus the classical post-processing will yield diminishing returns.  

Having discussed the parity measurement and the necessity of an amplitude recovery, let us now move on to the implementation of an amplitude recovery unitary operation $\hat{U}_{\alpha'\rightarrow \alpha}$ in circuit QED systems. Note that the dispersive coupling $-\chi \hat{a}^{\dagger}\hat{a} |e\rangle\langle e|$ in Eq.\ \eqref{eq:cavity-transmon Hamiltonian} is a non-linear interaction as it is cubic in $\hat{a}$, $\hat{a}^{\dagger}$, and $\hat{\sigma}_{z} = |e\rangle\langle e| - |g\rangle\langle g|$, going beyond the quadratic Hamiltonian. At the conceptual level, one could immediately infer at this point that this non-linear interaction could be used to implement an arbitrary unitary operation on the joint cavity-transmon system, including a desired amplitude recovery operation $\hat{U}_{\alpha'\rightarrow \alpha}$. 

In Ref.\ \cite{Krastanov2015}, it was shown that one can implement a selective number-dependent arbitrary phase (SNAP) gate $\hat{U}_{\textrm{SNAP}}(\vec{\theta})$ on a microwave cavity bosonic mode by using the dispersive coupling to a transmon qubit. Explicitly, the SNAP gate $\hat{U}_{\textrm{SNAP}}(\vec{\theta})$ is defined as
\begin{align}
\hat{U}_{\textrm{SNAP}}(\vec{\theta}) &\equiv \sum_{n=0}^{\infty} e^{i\theta_{n}} |n\rangle\langle n|, 
\end{align}
where the $\theta_{n}$ is the number-dependent phase which can take an arbitrary value. The key underlying idea behind the implementation of the SNAP gate is that the oscillation frequency between the transmon qubit states $|g\rangle$ and $|e\rangle$ are dependent on the photon number $n$ in the cavity mode, i.e., 
\begin{align}
\omega_{ge} (n) = \omega_{T} - n\chi, 
\end{align}
where $\chi$ is the strength of the dispersive coupling. This means that all these frequencies can be addressed selectively and thus we can control each Fock state $|n\rangle$ of the cavity mode in a selective manner. See Refs.\ \cite{Krastanov2015,Heeres2015} for more details.  

Note that SNAP gates are generally non-Gaussian. For example, a unitary operation generated by a self-Kerr nonlinearity (which is non-Gaussian) 
\begin{align}
e^{-i\frac{Kt}{2}(\hat{a}^{\dagger})^{2}\hat{a}^{2}} = \sum_{n=0}^{\infty} e^{-\frac{Kt}{2} n(n-1) } |n\rangle\langle n| 
\end{align}
is a specific instance of the general SNAP gates with $\theta_{n} = -n(n-1)Kt/2$. Since it was shown in Ref.\ \cite{Lloyd1999} that self-Kerr nonlinearity combined with Gaussian operations are universal, we can guess that SNAP gates should also be very useful for universal quantum control. Indeed in Ref.\ \cite{Krastanov2015}, it was shown that an arbitrary unitary operation on a bosonic mode can be implemented by combining displacement operations $\hat{D}(\alpha) \equiv \exp[ \alpha\hat{a}^{\dagger} - \alpha^{*}\hat{a} ]$ (easily realizable by a linear drive) and SNAP gates $\hat{U}_{\textrm{SNAP}}(\vec{\theta})$ which can be implemented by using dispersive coupling as discussed above. Thus, it is in principle possible to implement an amplitude recovery operation $\hat{U}_{\alpha'\rightarrow \alpha}$ in circuit QED systems. On the other hand, note that using the SNAP gates may not be the most practical way to implement the amplitude recovery operation. It thus remains to be answered whether there is a more tailored method for the amplitude recovery operation that can perform better in practice than the generic SNAP gate approach.

\subsubsection{Autonomous QEC of four-component cat codes} 

Similar to the case of the two-component cat code, it is possible to autonomously stabilize the four-component cat code by using an engineered dissipation. Note that in the case of the four-component cat code, we want to stabilize the space spanned by the four coherent states $|\alpha\rangle$, $|i\alpha\rangle$, $|-\alpha\rangle$, and $|-\alpha\rangle$. This can be done by using the following engineered four-photon dissipation: 
\begin{align}
\frac{d\hat{\rho}(t)}{dt} &= \kappa_{4\textrm{ph}} \mathcal{D}[ \hat{a}^{4}-\alpha^{4} ] ( \hat{\rho}(t) ) . \label{eq:four component cat code engineered dissipation}
\end{align} 
Note that all the four coherent states given above are annihilated by the engineered jump operator $\hat{F}_{4\textrm{ph}}\equiv \hat{a}^{4}-\alpha^{4}$. Thus, the engineered four-photon dissipation $\mathcal{D}[ \hat{a}^{4} - \alpha^{4} ]$ stabilizes the four-component cat code manifold.  

Similar to the case of the two-component cat codes, the engineered four-photon dissipation in Eq.\ \eqref{eq:four component cat code engineered dissipation} protect the four-component cat code manifold against bosonic dephasing errors. However, the engineered four-photon dissipation $\mathcal{D}[ \hat{a}^{4}-\alpha^{4} ]$ does not protect the code space against photon loss errors. This can be readily seen by observing that the engineered four-photon dissipation does not change the photon number parity of the system. Thus, when there is a single-photon loss, although the encoded logical information is well preserved in the odd-parity subspace, we cannot recover the encoded information if we only use the engineered four-photon dissipation. 

To address this issue, one might think of a hybrid approach where one actively measures the photon number parity by using the measurement circuit in Fig.\ \ref{fig:one bit phase estimation circuit} while autonomously stabilizing the code space by using the engineered four-photon dissipation. Unfortunately, however, the parity measurement circuit in Fig.\ \ref{fig:one bit phase estimation circuit} and the engineered four-photon dissipation $\mathcal{D}[\hat{a}^{4}-\alpha^{4}]$ do not commute with each other and thus they cannot be implemented simultaneously. To be more precise, while the parity operator $\hat{\Pi}_{2} = e^{i\pi\hat{n}}$ commutes with the jump operator $\hat{F}_{4\mathrm{ph}} = \hat{a}^{4}-\alpha^{4}$, the generator of the parity operator (or $180\degree$ rotation) $\hat{n}$ does not commute with $\hat{F}_{4\mathrm{ph}}$. It means that the engineered jump operator $\hat{F}_{4\mathrm{ph}}$ does not commute with the qubit-state-conditional rotation process at all times (or at all angles). Hence, we should turn off the engineered four-photon dissipation while measuring the photon number parity and then turn it on again while we wait for the next round of the parity measurement. Note that in this alternating scheme, one can view the engineered four-photon dissipation as a non-unitary amplitude recovery operation that maps the contracted cat code space $\mathcal{C}_{4-\textrm{cat}}^{(\alpha  e^{-\frac{\kappa t}{2}})}$ back to the original code space $\mathcal{C}_{4-\textrm{cat}}^{(\alpha )}$. We also remark that the non-commutativity of the engineered dissipation and the parity measurement is specific to the parity measurement scheme based on the circuit in Fig.\ \ref{fig:one bit phase estimation circuit}. An alternative parity measurement scheme that is compatible with the engineered four-photon dissipation was proposed in Ref.\ \cite{Cohen2017}. 

Similarly as in the case of the engineered two-photon dissipation $\mathcal{D}[ \hat{a}^{2}-\alpha^{2} ]$, the engineered four-photon dissipation $\mathcal{D}[\hat{a}^{4}-\alpha^{4}]$ can be decomposed into a four-photon drive and a four-photon dissipation: 
\begin{align}
\frac{d\hat{\rho}(t)}{dt} &= \kappa_{4\textrm{ph}} \mathcal{D}[\hat{a}^{4}-\alpha^{4}]( \hat{\rho}(t) ) =  \frac{\kappa_{4\textrm{ph}}}{2} \Big{[}  \alpha^{4}(\hat{a}^{\dagger})^{4} - \alpha^{*4}\hat{a}^{4}  , \hat{\rho}(t) \Big{]} + \kappa_{4\textrm{ph}} \mathcal{D}[ \hat{a}^{4} ]( \hat{\rho}(t) ).  
\end{align}
In Ref.\ \cite{Mundhada2017}, it was proposed that one can realize the four-photon dissipation $\mathcal{D}[\hat{a}^{4}]$ by coupling the system to a fast-decaying ancilla mode via an interaction Hamiltonian $\hat{H}_{\textrm{int}}  = g( \hat{a}^{4}|f\rangle\langle g| + (\hat{a}^{\dagger})^{4}|g\rangle\langle f|)$. Such a sixth-order interaction was subsequently realized experimentally in a circuit QED system \cite{Mundhada2019}.

\subsubsection{Generalization to higher order error correction}

Recall that the four-component cat code cannot correct two-excitation loss events. It is possible however to generalize the cat codes by adding more coherent state components such that they are robust against $\ell$-excitation loss errors with some $\ell \ge 2$ \cite{Li2017}. For example, we can define six-component cat codes $\mathcal{C}_{6-\textrm{cat}}^{(\alpha)}$ as follows: 
\begin{align}
|0_{6-\textrm{cat}}^{(\alpha)}\rangle &= \frac{1}{ \sqrt{ 6N_{0}^{6-\textrm{cat}}(\alpha) } }\Big{(} |\alpha\rangle + |e^{i\frac{\pi}{3}}\alpha\rangle + |e^{i\frac{2\pi}{3}}\alpha\rangle + |e^{i\pi}\alpha\rangle + |e^{i\frac{4\pi}{3}}\alpha\rangle + |e^{i\frac{5\pi}{3}}\alpha\rangle \Big{)} , 
\nonumber\\
|1_{6-\textrm{cat}}^{(\alpha)}\rangle &= \frac{1}{ \sqrt{ 6N_{1}^{6-\textrm{cat}}(\alpha) } }\Big{(} |\alpha\rangle - |e^{i\frac{\pi}{3}}\alpha\rangle + |e^{i\frac{2\pi}{3}}\alpha\rangle - |e^{i\pi}\alpha\rangle + |e^{i\frac{4\pi}{3}}\alpha\rangle - |e^{i\frac{5\pi}{3}}\alpha\rangle \Big{)}  . 
\end{align}
The normalization constants $N_{\mu}^{6-\textrm{cat}}(\alpha)$ are defined as
\begin{align}
N_{\mu}^{6-\textrm{cat}}(\alpha) &\equiv \frac{1}{6} \sum_{k,\ell =0}^{5} e^{-i ( k -\ell)\pi} \langle e^{i\frac{k\pi}{3}} \alpha | e^{i\frac{\ell\pi}{3}} \alpha\rangle  . 
\end{align}
Most importantly, the logical zero state $|0_{6-\textrm{cat}}^{(\alpha)}\rangle$ has $0$ excitations mod $6$ and the logical one state $|0_{6-\textrm{cat}}^{(\alpha)}\rangle$ has $3$ excitations mod $6$. This implies that all the logical states have $0$ excitations mod $3$ and thus the six-component cat code $\mathcal{C}_{6-\textrm{cat}}^{(\alpha)}$ is stabilized by the $120\degree$ phase rotation. 
\begin{align}
\hat{\Pi}_{3} \equiv e^{i\frac{\pi}{3}\hat{n}} =  \sum_{n=0}^{\infty}|n\rangle\langle n|\times  \begin{cases}
1 & n=0 \textrm{ mod }3 \\
e^{i\frac{\pi}{3}} & n=1 \textrm{ mod }3 \\
e^{i\frac{2\pi}{3}} & n=2 \textrm{ mod }3
\end{cases} .  
\end{align}
Hence, the six-component cat code is an example of rotation-symmetric bosonic codes \cite{Grimsmo2019}. 

Since the logical states of the six-component cat code have $0$ excitations mod $3$, they will be mapped via single-excitation loss to some error states with $2$ excitations mod $3$, and similarly via two-excitation loss to some error states with $1$ excitations mod $3$. Thus, by measuring the stabilizer of the six-component cat code $\hat{\Pi}_{3} \equiv e^{i\frac{\pi}{3}\hat{n}}$ (or equivalently, the excitation number modulo $3$), we can detect any single-excitation and two-excitation loss events. Thus, the six-component cat codes are robust against two-excitation loss events. More generally, one can define a $2d$-component cat code by using $2d$ coherent state components that is robust against all $\ell$-excitation loss events for $\ell \le d$ (see Ref.\ \cite{Li2017} for more details).

\subsubsection{Recent developments}  

Recall that two-component cat codes $\mathcal{C}_{2-\textrm{cat}}^{(\alpha)}$ can be realized by using an engineered two-photon dissipation $\mathcal{D}[\hat{a}^{2}-\alpha^{2}]$. An alternative way to implement the two-component cat code is to use the self-Kerr nonlinearity and two-photon drive \cite{Puri2017}. Specifically, the scheme in Ref.\ \cite{Puri2017} is based on the fact that the Hamiltonian 
\begin{align}
\hat{H} = -K( \hat{a}^{\dagger} )^{2}\hat{a}^{2} + ( \epsilon_{p}(\hat{a}^{\dagger})^{2} + \epsilon_{p}^{*}\hat{a}^{2} ) = -K \Big{(} (\hat{a}^{\dagger})^{2} -\frac{\epsilon_{p}^{*}}{K} \Big{)}\Big{(} \hat{a}^{2} - \frac{\epsilon_{p}}{K} \Big{)} + \frac{|\epsilon_{p}|^{2}}{K} 
\end{align}
has the two coherent states $|\pm \alpha \rangle$ with $\alpha = \sqrt{\epsilon_{p}/K}$ as its degenerate ground states. Hence, if the system is described by the above Hamiltonian, one can stabilize the two-component cat code manifold simply by cooling the system to its ground state manifold. This scheme was recently realized experimentally in a circuit QED system \cite{Grimm2019}.  

Note that the two-component cat code is not robust against excitation loss errors regardless of how it is implemented. Specifically, excitation loss errors cause logical bit-flip errors in the two-component cat code manifold. One way to make the cat code robust against excitation loss errors is to use the four-component cat code as discussed above. On the other hand, it is also possible to concatenate the two-component cat code with a conventional multi-qubit error-correcting code to correct the residual bit-flip errors in the two-component cat code. For instance, concatenation of the two-component cat code with a repetition code (i.e., repetition-cat code) was explored in Refs.\ \cite{Cohen2017a,Guillaud2019}. 

Recently, it has been observed that one might be able to reduce the required resource overhead associated with the use of conventional multi-qubit error-correcting codes by using the two-component cat qubits. This is because the two-component cat qubits are subject predominantly to bit-flip errors due to excitation loss errors but not phase-flip errors (since phase-flip errors are suppressed exponentially in the size of the cat code). Thus, the next layer of the multi-qubit error-correcting code can be tailored to such biased-noise models. Along this line, it has been shown recently that if each qubit in the surface code is subject to a biased noise, the fault-tolerance thresholds for the surface code can be significantly relaxed by using a tailored decoding scheme for the biased-noise model \cite{Tuckett2018,Tuckett2019a,Tuckett2019}. Also, various schemes for bias-preserving gates for the two-component cat code have been proposed \cite{Guillaud2019,Puri2019} so that the noise bias can be maintained even during the application of quantum operations. See also Subsection \ref{subsection:Concatenation of a cat code with a multi-qubit code} for more discussions.  

Getting back to the single-mode bosonic QEC, recall that we can directly deal with the excitation loss errors by using the four-component cat code and measuring the excitation number parity. In all the circuit QED implementations, an ancilla transmon qubit was used to measure the photon number parity of a microwave cavity bosonic mode. However, note that the ancilla transmon qubits used in the parity measurement scheme are noisy. For example, the excited state of a transmon qubit $|e\rangle$ may decay the ground state $|g\rangle$ during the parity measurement. Note that coherence times of a transmon qubit are typically given by $10–50\mu$s. On the other hand, the parity measurement based on the qubit-conditional $180\degree$ phase rotation takes $\Delta t = \pi /\chi \sim 1\mu$s where $\chi$ is the strength of the dispersive coupling between a cavity mode and a transmon qubit. Thus, each parity measurement causes additional errors to the bosonic cavity mode with an error rate roughly given by $0.01–0.1$. This was the limiting factor in the previous circuit QED implementations of the parity measurement \cite{Sun2014,Ofek2016}. 

A simple way to address the transmon decay during the parity measurement is to use higher excited states of the transmon qubit. For instance, one could use the second excited state of a transmon qubit $|f\rangle$ instead of the first excited state $|e\rangle$ to perform the parity measurement. In this case, the states $|g\rangle$ and $|f\rangle$ form the basis of the ancilla qubit and the $|e\rangle$ state serves as a buffer state. Then, although the second excited state $|f\rangle$ may decay to the first excited state $|e\rangle$ during the parity measurement, this decay event can be detected by measuring the buffer state $|e\rangle$ at the end of the parity measurement circuit. Thus, the parity measurement scheme can be made robust against the single transmon decay error by discarding all the measurement runs that ended in the $|e\rangle$ state. However, this simple scheme will not be scalable because the success probability will decrease exponentially as we repeat the parity measurements. 

Recently, an improved alternative scheme has been proposed and implemented experimentally \cite{Rosenblum2018}. In this more sophisticated scheme, one carefully engineers the dispersive coupling between a cavity mode and a transmon qubit such that
\begin{align}
\hat{H}_{\textrm{int}} &= -\chi_{e}\hat{a}^{\dagger}\hat{a}|e\rangle\langle e|-\chi_{f}\hat{a}^{\dagger}\hat{a}|f\rangle\langle f| = -\chi \hat{a}^{\dagger}\hat{a} ( |e\rangle\langle e| + |f\rangle\langle f| ), 
\end{align}    
i.e., $\chi_{e} = \chi_{f} = \chi$. Using this ``$\chi$-matching'' technique, one can ensure that the cavity state is not decohered even when the qubit state is measured in the $|e\rangle$ state due to the transmon decay. Thus, one does not need to discard the measurement runs with the $|e\rangle$ state and instead can simply reset the qubit and retry the parity measurement (see Ref.\ \cite{Rosenblum2018} for more details). A similar technique also proved to be useful for improving the fidelity of the SNAP gates \cite{Reinhold2019,Ma2019}.  

There has also been several progress on the autonomous QEC of cat codes. Recall that the engineered four-photon dissipation for the four-component cat code $\mathcal{D}[\hat{a}^{4} - \alpha^{4}]$ does not correct excitation loss errors. Thus parity measurements are necessary if we want to fully benefit from the error correction capability of the four-component cat code. As discussed above, however, the usual parity measurement scheme based on the circuit in Fig.\ \ref{fig:one bit phase estimation circuit} is not compatible with the engineered four-photon dissipation. In Ref.\ \cite{Cohen2017}, an alternative parity measurement scheme that is compatible with the engineered four-photon dissipation was proposed. Furthermore, the pair-cat code \cite{Albert2019} has recently been proposed as an alternative to the four-component cat code.

\subsection{Binomial codes} 

Here, we review the binomial codes \cite{Michael2016}. Note that the cat codes are composed of multiple components of the coherent states. Since a coherent coherent state $|\alpha\rangle$ occupies the entire infinite-dimensional bosonic Hilbert space, i.e.,
\begin{align}
|\alpha\rangle &= e^{-\frac{1}{2}|\alpha|^{2}} \sum_{n=0}^{\infty}\frac{\alpha^{n}}{\sqrt{n!}}|n\rangle, 
\end{align}   
we need a large Hilbert space dimension to faithfully describe cat code states. For instance, for the smallest sweet-spot value of the four-component cat code $|\alpha_{1}^{\star}|=1.538$, we need to have $n_{\textrm{cut}}\ge 9$ to capture more than $99.9\%$ of the state's total population.

\begin{table}[t!]
  \centering
     \def\arraystretch{1.5}
  \begin{tabular}{ V{3} c V{1.5} c V{3} }
   \hlineB{3}  
     & \textbf{$(1,1)$-binomial code} \cite{Michael2016}  \\  \hlineB{3} 
     Logical states & $|0_{\textrm{bin}}^{(1,1)}\rangle = \frac{1}{\sqrt{2}} ( |0\rangle + |4\rangle ) $   \\
     &  $\!\!\!\!\!\!\!\!\!\!\!\!\!\!\!\!\!\!\!\!\!\!\!\!\!\!\! |1_{\textrm{bin}}^{(1,1)}\rangle = |2\rangle$  \\ \hlineB{1.5}   
    Correctable errors & Single-excitation loss  \\ \hlineB{1.5}
    Active QEC & Parity measurement and   \\ 
     & recovery unitaries \cite{Michael2016} \\  \hdashline
     Experiment & Ref.\ \cite{Hu2019}   \\  \hlineB{1.5}     
    Autonomous QEC & Engineered dissipation \cite{Lihm2018} \\ \hdashline
    Experiment & Ref.\ \cite{Ma2019a}  \\  \hlineB{3}     
  \end{tabular}
  \caption{Basic properties of the $(1,1)$-binomial code. }
  \label{table:binomial codes}
\end{table}

In many aspects, binomial codes are similar to cat codes, especially in the sense that they are both rotation-symmetric \cite{Grimsmo2019}. However, binomial codes are distinguished from cat codes because binomial codes occupy only a finite-dimensional subspace with at most $n_{\textrm{cut}} < \infty$ excitations. Below, we review the properties of the binomial codes and discuss their experimental implementations. See Table \ref{table:binomial codes} for a summary.

\subsubsection{The $(1,1)$-binomial code}

Logical states of the smallest non-trivial binomial code $\mathcal{C}_{\textrm{bin}}^{(1,1)}$ are given by 
\begin{align}
|0_{\textrm{bin}}^{(1,1)}\rangle &= \frac{1}{\sqrt{2}} ( |0\rangle + |4\rangle ), 
\nonumber\\
|1_{\textrm{bin}}^{(1,1)}\rangle &= |2\rangle. 
\end{align}
The superscript $(1,1)$ is due to the fact that the above binomial code is a special instance of the general binomial code $\mathcal{C}_{\textrm{bin}}^{(N,S)}$, where the two parameters $N$ and $S$ are given by $N=S=1$ (see below for more details about the general binomial code). From now on, we will refer to this binomial code as the $(1,1)$-binomial code.

\begin{figure}[t!]
\centering
\includegraphics[width=6.0in]{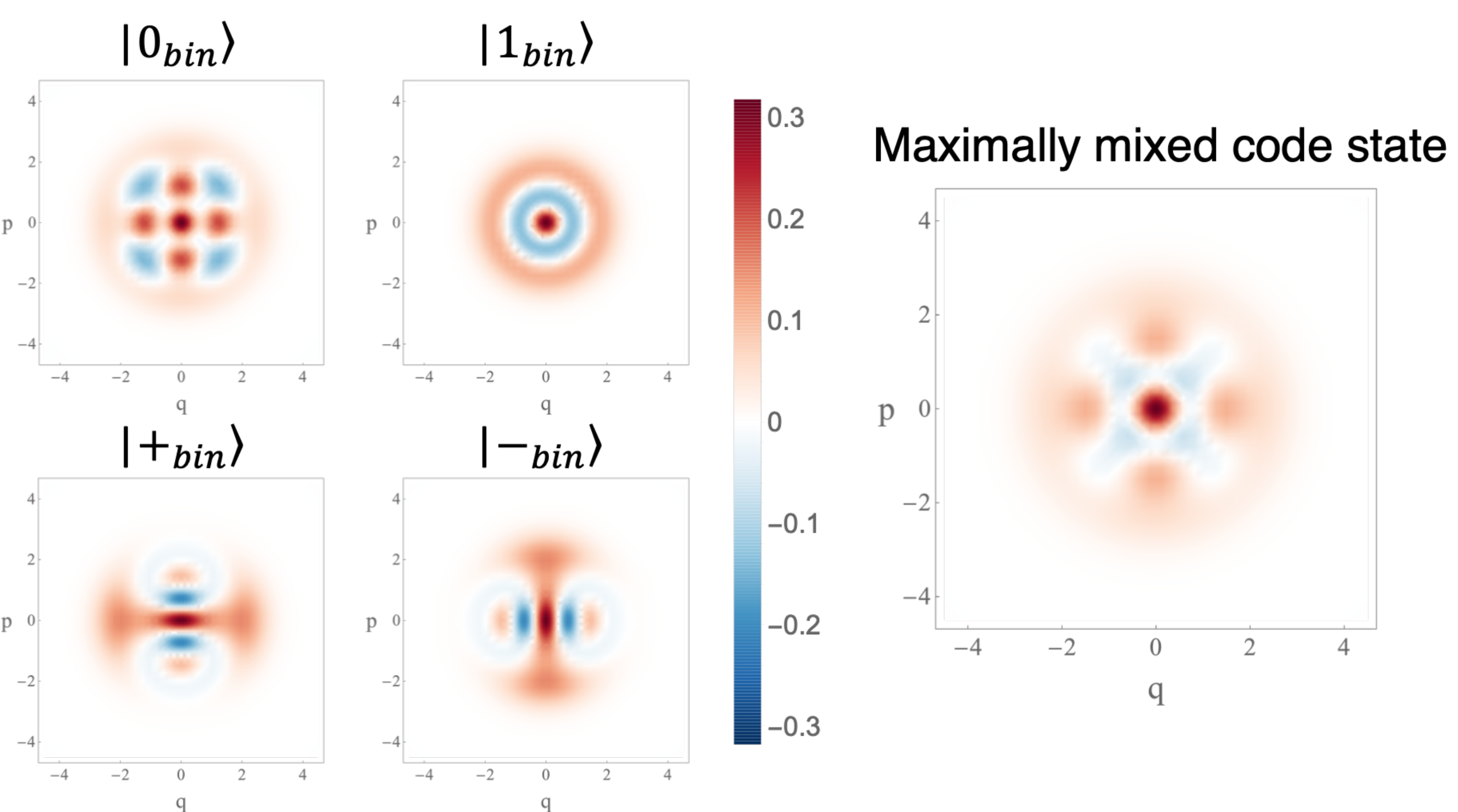}
\caption{Wigner functions of the logical states of the $(1,1)$-binomial code $\mathcal{C}_{\textrm{bin}}^{(1,1)}$. The maximally mixed code state is defined as the projection operator to the code space divided by $2$. }
\label{fig:wigner functions binomial code}
\end{figure}

Note that the logical states of the $(1,1)$-binomial code consist of even excitation number states. Thus, the $(1,1)$-binomial code is stabilized by the parity operator $\hat{\Pi}_{2} = e^{i\pi\hat{n}}$, i.e.,
\begin{align}
\hat{\Pi}_{2} |\psi_{\textrm{bin}}^{(1,1)}\rangle &= e^{i\pi\hat{n}} |\psi_{\textrm{bin}}^{(1,1)}\rangle = |\psi_{\textrm{bin}}^{(1,1)}\rangle, \,\,\, \textrm{for all}\,\,\, |\psi_{\textrm{bin}}^{(1,1)}\rangle\in \mathcal{C}_{\textrm{bin}}^{(1,1)}
\end{align}
and is invariant under the $180\degree$ rotation. Hence, the $(1,1)$-binomial code is an instance of rotation-symmetric bosonic codes \cite{Grimsmo2019}. Note also that $|0_{\textrm{bin}}^{(1,1)}\rangle$ has $0$ excitations mod $4$ and $|1_{\textrm{bin}}^{(1,1)}\rangle$ has $2$ excitations mod $4$, so they are clearly orthogonal to each other. Logical states of the $(1,1)$-binomial code are visualized in Fig.\ \ref{fig:wigner functions binomial code}.

The $(1,1)$-binomial code is capable of correcting the excitation loss errors to the first order. That is, the $(1,1)$-binomial code $\mathcal{C}_{\textrm{bin}}^{(1,1)}$ satisfies the Knill-Laflamme condition for the first-order excitation loss error set $\lbrace \hat{I},  \hat{a} \rbrace$, i.e., $\hat{N}_{0}=\hat{I}$ (no error) and $\hat{N}_{1}=\hat{a}$ (single-excitation loss). Before thoroughly checking the Knill-Laflamme condition, let us consider an arbitrary encoded state $|\psi_{\textrm{bin}}^{(1,1)}\rangle = c_{0} |0_{\textrm{bin}}^{(1,1)}\rangle + c_{1} |1_{\textrm{bin}}^{(1,1)}\rangle$ to see why the $(1,1)$-binomial code works in an intuitive way. Note that upon a single-excitation loss event, the encoded state is mapped to the following corrupted state 
\begin{align}
\hat{a}|\psi_{\textrm{bin}}^{(1,1)}\rangle &= \hat{a} \Big{[} c_{0} \frac{1}{\sqrt{2}}\Big{(}|0\rangle + |4\rangle\Big{)} +c_{1} |2\rangle \Big{]} = \sqrt{2}\Big{[} c_{0} |3\rangle + c_{1} |1\rangle \Big{]}. 
\end{align} 
Since the corrupted state is now in the odd excitation parity subspace whereas the encoded state is in the even excitation parity subspace, we can detect the loss event by measuring the excitation number parity $\hat{\Pi}_{2} = e^{i\pi\hat{n}}$. Moreover, since the error states $\hat{a}|0_{\textrm{bin}}^{(1,1)}\rangle = \sqrt{2}|3\rangle$ and $\hat{a}|1_{\textrm{bin}}^{(1,1)}\rangle = \sqrt{2}|1\rangle$ have the same normalization constant (i.e., $\sqrt{2}$) and are orthogonal to each other, the corrupted state $\hat{a}|\psi_{\textrm{bin}}^{(1,1)}\rangle = \sqrt{2} ( c_{0} |3\rangle + c_{1} |1\rangle )$ contains the same logical quantum information as the uncorrupted state $|\psi_{\textrm{bin}}^{(1,1)}\rangle$. That is, by mapping the error states back to the logical states, i.e., 
\begin{align}
|3\rangle &\rightarrow \frac{1}{\sqrt{2}} ( |0\rangle +|4\rangle ), 
\nonumber\\
|1\rangle &\rightarrow |2\rangle, \label{eq:binomial code mapping error states back to code states}
\end{align}   
we can recover the original encoded state: 
\begin{align}
\hat{a}|\psi_{\textrm{bin}}^{(1,1)}\rangle = \sqrt{2} ( c_{0} |3\rangle + c_{1} |1\rangle ) \rightarrow \sqrt{2} \Big{[} c_{0} \frac{1}{\sqrt{2}}\Big{(}|0\rangle + |4\rangle\Big{)} +c_{1} |2\rangle \Big{]}  \propto |\psi_{\textrm{bin}}^{(1,1)}\rangle. 
\end{align}
Note that the mapping in Eq.\ \eqref{eq:binomial code mapping error states back to code states} can be implemented by a unitary operation because all the four states that are involved are mutually orthogonal.  

Let us now explicitly check the Knill-Laflamme condition. As mentioned above, the error states are given by $\hat{a}|0_{\textrm{bin}}^{(1,1)}\rangle = \sqrt{2}|3\rangle$ and $\hat{a}|1_{\textrm{bin}}^{(1,1)}\rangle = \sqrt{2}|1\rangle$ and thus all the relevant states $|0_{\textrm{bin}}^{(1,1)}\rangle$, $|1_{\textrm{bin}}^{(1,1)}\rangle$, $\hat{a}|0_{\textrm{bin}}^{(1,1)}\rangle$, and $\hat{a}|1_{\textrm{bin}}^{(1,1)}\rangle$ are mutually orthogonal. Hence, we have 
\begin{align}
\langle \mu_{\textrm{bin}}^{(1,1)} | \hat{N}_{\ell}^{\dagger}\hat{N}_{\ell'} |\nu_{\textrm{bin}}^{(1,1)} \rangle =0, \,\,\, \textrm{for all}\,\,\, \ell,\ell'\in \lbrace 0,1 \rbrace \,\,\, \textrm{and}\,\,\, \mu\neq \nu, 
\end{align}
and the Knill-Laflamme condition is satisfied for all $\mu\neq \nu$. The relevant $\mu=\nu$ terms in the Knill-Laflamme condition, i.e., $\langle \mu_{\textrm{bin}}^{(1,1)} | \hat{N}_{\ell}^{\dagger}\hat{N}_{\ell'} |\mu_{\textrm{bin}}^{(1,1)} \rangle$, are given by 
\begin{align}
&(\ell,\ell')= (0,0):
\nonumber\\
&\quad \langle \mu_{\textrm{bin}}^{(1,1)}  |\mu_{\textrm{bin}}^{(1,1)} \rangle =1,
\nonumber\\ 
&(\ell,\ell')= (0,1),(1,0): 
\nonumber\\
&\quad \langle \mu_{\textrm{bin}}^{(1,1)}  | \hat{a} |\mu_{\textrm{bin}}^{(1,1)} \rangle =0, 
\nonumber\\
&(\ell,\ell')= (1,1): 
\nonumber\\
&\quad \langle 0_{\textrm{bin}}^{(1,1)}  | \hat{a}^{\dagger}\hat{a} |0_{\textrm{bin}}^{(1,1)} \rangle = 2\langle 3|3\rangle = 2, \,\,\, \textrm{and}\,\,\, \langle 1_{\textrm{bin}}^{(1,1)}  | \hat{a}^{\dagger}\hat{a} |1_{\textrm{bin}}^{(1,1)} \rangle = 2\langle 1|1\rangle = 2. 
\end{align}
Thus, $\langle \mu_{\textrm{bin}}^{(1,1)} | \hat{N}_{\ell}^{\dagger}\hat{N}_{\ell'} |\mu_{\textrm{bin}}^{(1,1)} \rangle$ is $\mu$-independent for all $\ell,\ell'\in\lbrace 0,1 \rbrace$ and the $(1,1)$-binomial code can correct the excitation loss errors to the first order. Note that the coefficients of the logical zero state $|0_{\textrm{bin}}^{(1,1)}\rangle = \frac{1}{\sqrt{2}}( |0\rangle + |4\rangle )$ (i.e., $\frac{1}{\sqrt{2}}$ and $\frac{1}{\sqrt{2}}$) are carefully chosen such that $\langle \mu_{\textrm{bin}}^{(1,1)}  | \hat{a}^{\dagger}\hat{a} |\mu_{\textrm{bin}}^{(1,1)} \rangle$ is $\mu$-independent. Also, such $\mu$-independence is essential to ensure that the corrupted state $\hat{a}|\psi_{\textrm{bin}}^{(1,1)}\rangle$ contains the same logical quantum information as the original encoded state $|\psi_{\textrm{bin}}^{(1,1)}\rangle$.

\subsubsection{Active QEC of the $(1,1)$-binomial code}

Similarly to the case of four-component cat codes, the most important ingredient of the active QEC of the $(1,1)$-binomial code is the quantum non-demolition measurement of the excitation number parity operator
\begin{align}
\hat{\Pi}_{2} \equiv e^{i\pi\hat{n}} = \sum_{n=0}^{\infty} |n\rangle\langle n| \times \begin{cases}
+1 & n\textrm{ even}\\
-1 & n\textrm{ odd}
\end{cases}. 
\end{align} 
Since the logical states of the $(1,1)$-binomial code is in the even excitation parity subspace, we can infer that there was no excitation loss when we measure the even parity ($\hat{\Pi}_{2} = +1$) and a single-excitation loss error when we measure the odd parity ($\hat{\Pi}_{2} = -1$). As was briefly explained above, when we measure the odd parity, we need to apply a recovery unitary operator $\hat{U}_{\textrm{rec}}^{(\textrm{odd})}$ such that the error states $|0_{\textrm{bin},e}^{(1,1)}\rangle=|3\rangle$ and $|1_{\textrm{bin},e}^{(1,1)}\rangle = |1\rangle$ are mapped back to the code states $|0_{\textrm{bin}}^{(1,1)}\rangle$ and $|1_{\textrm{bin}}^{(1,1)}\rangle$: 
\begin{align}
\hat{U}_{\textrm{rec}}^{(\textrm{odd})}|0_{\textrm{bin},e}^{(1,1)}\rangle &= |0_{\textrm{bin}}^{(1,1)}\rangle = \frac{1}{\sqrt{2}} ( |0\rangle + |4\rangle ), 
\nonumber\\
\hat{U}_{\textrm{rec}}^{(\textrm{odd})}|1_{\textrm{bin},e}^{(1,1)}\rangle &= |1_{\textrm{bin}}^{(1,1)}\rangle = |2\rangle.  \label{eq:binomial code recovery unitary odd}
\end{align}   
When we measure the even parity, we might be tempted to think that no recovery operation is needed because the system did not lose any excitations. However, this is not true because the system is disturbed by the no-jump evolution term $\hat{N}_{0} = e^{-\frac{\kappa t}{2} \hat{n} }$ even when it did not lose any excitations. Below, we discuss this subtlety further. 

Let us take a closer look into the no-excitation loss case (i.e., $\ell=0$). Consider an arbitrary encoded state $|\psi_{\textrm{bin}}^{(1,1)}\rangle = c_{0}|0_{\textrm{bin}}^{(1,1)}\rangle + c_{1} |1_{\textrm{bin}}^{(1,1)}\rangle$. Upon the no-jump evolution $\hat{N}_{0}$, the state is mapped to
\begin{align}
\hat{N}_{0}|\psi_{\textrm{bin}}^{(1,1)}\rangle &= e^{-\frac{\kappa t}{2} \hat{n} } \Big{[} c_{0}\frac{1}{\sqrt{2}} \Big{(} |0\rangle + |4\rangle \Big{)} + c_{1}|2\rangle \Big{]}
\nonumber\\
&= c_{0}\frac{1}{\sqrt{2}} \Big{(} |0\rangle + e^{-2\kappa t}|4\rangle \Big{)} + c_{1}e^{-\kappa t}|2\rangle 
\nonumber\\
&= c_{0}\frac{1}{\sqrt{2}} \Big{(} |0\rangle + |4\rangle - 2\kappa t |4\rangle \Big{)} + c_{1} (1-\kappa t) |2\rangle  + \mathcal{O}\Big{(} (\kappa t)^{2} \Big{)}
\nonumber\\
&= c_{0} \Big{(} (1-\kappa t)|0_{\textrm{bin}}^{(1,1)}\rangle + \kappa t \frac{1}{\sqrt{2}} ( |0\rangle - |4\rangle )  \Big{)}+ c_{1} (1-\kappa t) |1_{\textrm{bin}}^{(1,1)}\rangle  + \mathcal{O}\Big{(} (\kappa t)^{2} \Big{)}. \label{eq:binomial code no-jump evolution populates the residual state}
\end{align}
Note that this state would have been proportional to the original code state $|\psi_{\textrm{bin}}^{(1,1)}\rangle$ to the order $\kappa t$ if it were not for the term $\kappa t \frac{1}{\sqrt{2}}(|0\rangle -|4\rangle)$. The emergence of the residual state $|\phi\rangle = \frac{1}{\sqrt{2}}(|0\rangle -|4\rangle)$ is due to the non-trivial no-jump evolution term $e^{-\frac{\kappa t}{2}\hat{n}}$. This residual state has to be removed since it has non-trivial effects on the encoded information to the first order in $\kappa t$ whereas we want to suppress the excitation loss errors to the second order in $\kappa t$. 

To counter the non-trivial effects of the no-jump evolution term $e^{-\frac{\kappa t}{2}\hat{n}}$, we need to apply a recovery unitary $\hat{U}_{\textrm{rec}}^{(\textrm{even})}$ such that
\begin{align}
\hat{U}_{\textrm{rec}}^{(\textrm{even})} |0_{\textrm{bin}}^{(1,1)}\rangle &= \cos(\kappa t) |0_{\textrm{bin}}^{(1,1)}\rangle - \sin(\kappa t)|\phi\rangle, 
\nonumber\\
\hat{U}_{\textrm{rec}}^{(\textrm{even})} |\phi\rangle &= \sin(\kappa t)|0_{\textrm{bin}}^{(1,1)}\rangle + \cos(\kappa t) |\phi\rangle, 
\nonumber\\
\hat{U}_{\textrm{rec}}^{(\textrm{even})} |1_{\textrm{bin}}^{(1,1)}\rangle &=|1_{\textrm{bin}}^{(1,1)}\rangle. \label{eq:binomial code recovery unitary even}
\end{align}
Then, it follows that
\begin{align}
\hat{U}_{\textrm{rec}}^{(\textrm{even})} \hat{N}_{0}|\psi_{\textrm{bin}}^{(1,1)}\rangle &= (1-\kappa t) |\psi_{\textrm{bin}}^{(1,1)}\rangle + \mathcal{O}\Big{(} (\kappa t)^{2} \Big{)}, \label{eq:binomial code no excitation loss final}
\end{align}
and thus the encoded quantum information is preserved to the first order in $\kappa t$ as desired. 

Let us move on to the single-excitation loss event (i.e., $\ell = 1$). Upon a single-excitation loss $\hat{N}_{1} = \sqrt{ \kappa t } e^{-\frac{\kappa t}{2} \hat{n} } \hat{a}$, the encoded state $|\psi_{\textrm{bin}}^{(1,1)}\rangle = c_{0}|0_{\textrm{bin}}^{(1,1)}\rangle + c_{1} |1_{\textrm{bin}}^{(1,1)}\rangle$ is mapped to 
\begin{align}
\hat{N}_{1}|\psi_{\textrm{bin}}^{(1,1)}\rangle &= \sqrt{ \kappa t } e^{-\frac{\kappa t}{2} \hat{n} } \hat{a} \Big{[} c_{0}\frac{1}{\sqrt{2}} \Big{(} |0\rangle + |4\rangle \Big{)} + c_{1}|2\rangle \Big{]}
\nonumber\\
&= \sqrt{ \kappa t } e^{-\frac{\kappa t}{2} \hat{n} } \Big{[} c_{0} \sqrt{2}|3\rangle + c_{1} \sqrt{2} |1\rangle \Big{]}
\nonumber\\
&= \sqrt{ 2\kappa t }  \Big{[} c_{0} e^{-\frac{3\kappa t}{2}  }|3\rangle + c_{1} e^{-\frac{\kappa t}{2}  } |1\rangle \Big{]}
\nonumber\\
&= \sqrt{ 2\kappa t }  \Big{[} c_{0} |3\rangle + c_{1}  |1\rangle \Big{]} + \mathcal{O}\Big{(} (\kappa t)^{\frac{3}{2}} \Big{)}. 
\end{align}
Thus, by applying the recovery operation $\hat{U}_{\textrm{rec}}^{(\textrm{odd})}$, we can map the error states $|3\rangle$ and $|1\rangle$ back to the code states $|0_{\textrm{bin}}^{(1,1)}\rangle$ and $|1_{\textrm{bin}}^{(1,1)}\rangle$ (see Eq.\ \eqref{eq:binomial code recovery unitary odd}). Hence, we have 
\begin{align}
\hat{U}_{\textrm{rec}}^{(\textrm{odd})}\hat{N}_{1}|\psi_{\textrm{bin}}^{(1,1)}\rangle &= \sqrt{ 2\kappa t } |\psi_{\textrm{bin}}^{(1,1)}\rangle + \mathcal{O}\Big{(} (\kappa t)^{\frac{3}{2}} \Big{)} , \label{eq:binomial code single excitation loss final}
\end{align} 
as desired. 

To summarize, no-error events and the single-excitation loss events can be corrected by using the $(1,1)$-binomial code and the following recovery map 
\begin{align}
\mathcal{R}_{\textrm{bin}} (\hat{\rho}) &\equiv \hat{R}_{\textrm{even}}\hat{\rho}\hat{R}_{\textrm{even}}^{\dagger} + \hat{R}_{\textrm{odd}}\hat{\rho}\hat{R}_{\textrm{odd}}^{\dagger}, 
\end{align}
where $\hat{R}_{\textrm{even}}$ and $\hat{R}_{\textrm{odd}}$ are defined as 
\begin{align}
\hat{R}_{\textrm{even}} &\equiv \hat{U}_{\textrm{rec}}^{(\textrm{even})}\hat{P}_{\textrm{even}} = \hat{U}_{\textrm{rec}}^{(\textrm{even})} \frac{1}{2} ( \hat{I} + \hat{\Pi}_{2} ) , 
\nonumber\\ 
\hat{R}_{\textrm{odd}} &\equiv \hat{U}_{\textrm{rec}}^{(\textrm{odd})}\hat{P}_{\textrm{odd}} = \hat{U}_{\textrm{rec}}^{(\textrm{odd})} \frac{1}{2} ( \hat{I} - \hat{\Pi}_{2} ). 
\end{align}
In particular, by considering the error channel $\mathcal{N} = e^{\kappa \mathcal{D}[\hat{a}] t}$ and putting everything together, we find 
\begin{align}
\mathcal{R}_{\textrm{bin}} \cdot \mathcal{N} ( |\psi_{\textrm{bin}}^{(1,1)}\rangle\langle \psi_{\textrm{bin}}^{(1,1)} | ) &=\mathcal{R}_{\textrm{bin}} \Big{(}  \sum_{\ell=0}^{\infty}  \hat{N}_{\ell} |\psi_{\textrm{bin}}^{(1,1)}\rangle\langle \psi_{\textrm{bin}}^{(1,1)} | \hat{N}_{\ell}^{\dagger} \Big{)} 
\nonumber\\
&= \hat{U}_{\textrm{rec}}^{(\textrm{even})} \hat{N}_{0} |\psi_{\textrm{bin}}^{(1,1)}\rangle\langle \psi_{\textrm{bin}}^{(1,1)} | \hat{N}_{0}^{\dagger} ( \hat{U}_{\textrm{rec}}^{(\textrm{even})} )^{\dagger} 
\nonumber\\
&\quad + \hat{U}_{\textrm{rec}}^{(\textrm{odd})} \hat{N}_{1} |\psi_{\textrm{bin}}^{(1,1)}\rangle\langle \psi_{\textrm{bin}}^{(1,1)} | \hat{N}_{1}^{\dagger} ( \hat{U}_{\textrm{rec}}^{(\textrm{odd})} )^{\dagger}  + \mathcal{O} \Big{(} (\kappa t)^{2} \Big{)}
\nonumber\\
&= (1-\kappa t)^{2}  |\psi_{\textrm{bin}}^{(1,1)}\rangle\langle \psi_{\textrm{bin}}^{(1,1)} | +  2\kappa t |\psi_{\textrm{bin}}^{(1,1)}\rangle\langle \psi_{\textrm{bin}}^{(1,1)} | + \mathcal{O} \Big{(} (\kappa t)^{2} \Big{)}  
\nonumber\\
&= |\psi_{\textrm{bin}}^{(1,1)}\rangle\langle \psi_{\textrm{bin}}^{(1,1)} | +  \mathcal{O} \Big{(} (\kappa t)^{2} \Big{)}. 
\end{align}
Here, we used Eqs.\ \eqref{eq:binomial code no excitation loss final} and \eqref{eq:binomial code single excitation loss final} to derive the third equality. Hence, the logical error probability is suppressed to the second order in $\kappa t$, i.e., 
\begin{align}
1- \langle \psi_{\textrm{bin}}^{(1,1)} | \mathcal{R}_{\textrm{bin}} \cdot \mathcal{N} ( |\psi_{\textrm{bin}}^{(1,1)}\rangle\langle \psi_{\textrm{bin}}^{(1,1)} | )|\psi_{\textrm{bin}}^{(1,1)}\rangle = \mathcal{O}\Big{(} (\kappa t)^{2} \Big{)}, 
\end{align}
for any encoded input state $|\psi_{\textrm{bin}}^{(1,1)}\rangle \in \mathcal{C}_{\textrm{bin}}^{(1,1)}$.

\subsubsection{Autonomous QEC of the $(1,1)$-binomial code}

Instead of actively measuring the excitation number parity, we can autonomously stabilize the $(1,1)$-binomial code space by using a strong engineered dissipation \cite{Lihm2018}. More explicitly, we consider the following Lindblad master equation: 
\begin{align}
\frac{d\hat{\rho}(t)}{dt} &= \kappa_{\textrm{eng}} \sum_{j=1}^{J} \mathcal{D}[ \hat{F}_{\textrm{eng},j} ] ( \hat{\rho}(t) ) + \kappa \mathcal{D}[\hat{a}] ( \hat{\rho}(t) ). 
\end{align} 
Here, the dissipation superoperator $\mathcal{D}[\hat{A}](\hat{\rho})$ is defined as $\mathcal{D}[\hat{A}](\hat{\rho}) \equiv \hat{A}\hat{\rho}\hat{A}^{\dagger}-\frac{1}{2} \lbrace \hat{A}^{\dagger}\hat{A},\hat{\rho} \rbrace$. Note that the second term on the right hand side represents the excitation loss error with a loss rate $\kappa$. We aim to protect the encoded logical information against such an excitation loss error by using an engineered dissipation which is represented in the first term on the right hand side. $\kappa_{\textrm{eng}}$ is the engineered dissipation rate which ideally has to be much larger than the natural dissipation rate $\kappa$. Also, $\hat{F}_{\textrm{eng},j}$ is the engineered jump operator which should be designed carefully such that the natural loss errors can be corrected. 

Recall that the active QEC of the $(1,1)$-binomial code consists of a parity measurement followed by a recovery unitary operation. In particular, conditioned on measuring the odd parity, we apply the recovery unitary $\hat{U}_{\textrm{rec}}^{(\textrm{odd})}$ that maps the error states $|0_{\textrm{bin},e}^{(1,1)}\rangle = |3\rangle$ and $|1_{\textrm{bin},e}^{(1,1)}\rangle = |1\rangle$ back to the code states $|0_{\textrm{bin}}^{(1,1)}\rangle = \frac{1}{\sqrt{2}} ( |0\rangle + |4\rangle )$ and $|1_{\textrm{bin}}^{(1,1)}\rangle = |2\rangle$ (see Eq.\ \eqref{eq:binomial code recovery unitary odd}). In autonomous QEC, one does not need to perform the excitation parity measurement which should followed by a recovery unitary. Instead, one can implement the following engineered jump operator to correct for the single-excitation loss.   
\begin{align}
\hat{F}_{\textrm{eng},1} &= |0_{\textrm{bin}}^{(1,1)}\rangle\langle 0_{\textrm{bin},e}^{(1,1)}| + |1_{\textrm{bin}}^{(1,1)}\rangle\langle 1_{\textrm{bin},e}^{(1,1)}| 
\nonumber\\
&= \frac{1}{\sqrt{2}} ( |0\rangle +|4\rangle )\langle 3| + |2\rangle\langle 1| , \label{eq:binomial code corrective jump}
\end{align}
Note that this engineered jump operator maps the error states back to the code states similarly to the recovery unitary operation $\hat{U}_{\textrm{rec}}^{(\textrm{odd})}$. Importantly, the engineered jump operator $\hat{F}_{\textrm{eng},1}$ is triggered only when the system is in the error space, i.e., $\textrm{span}\lbrace |1\rangle, |3\rangle \rbrace$. Hence, the excitation number parity measurement is not needed in the case of autonomous QEC because the engineered jump operator $\hat{F}_{\textrm{eng},1}$ can stay turned on continuously regardless of the parity of the system.    

Similarly to the case of active QEC, we might be tempted to think that the engineered jump operator in Eq.\ \eqref{eq:binomial code corrective jump} is the only thing that we need. However, this is not true because the no-jump evolution $e^{-\frac{\kappa t}{2}\hat{n}}$ can populate the residual state $|\phi\rangle = \frac{1}{\sqrt{2}}( |0\rangle -|4\rangle )$ even when the system did not lose any excitations (see Eq.\ \eqref{eq:binomial code no-jump evolution populates the residual state}). In the Lindbladian picture, this is due to the non-trivial effects of the back-action term $-\frac{1}{2}\lbrace \hat{a}^{\dagger}\hat{a},\hat{\rho}(t) \rbrace$ in the superoperator $\mathcal{D}[\hat{a}] ( \hat{\rho}(t) )$. In the case of active QEC, the non-trivial effects of the no-jump evolution (or the back-action) are countered by applying a recovery unitary operation $\hat{U}_{\textrm{rec}}^{(\textrm{even})}$ conditioned on measuring the even parity (see Eq.\ \eqref{eq:binomial code recovery unitary even}). In the case of autonomous QEC, one can counter the adverse effects of the no-jump evolution simply by emptying the population in the residual state $|\phi\rangle$. That is, any engineered jump operator of the following form would work: 
\begin{align}
\hat{F}_{\textrm{eng},2} &= |\Phi\rangle\langle \phi| , \label{eq:binomial code preventive jump} 
\end{align}     
where $|\Phi\rangle$ is a state in the relevant physical Hilbert space $\mathcal{H} = \textrm{span}\lbrace |0\rangle,\cdots,|4\rangle \rbrace$ which is perpendicular to the residual state $|\phi\rangle =  \frac{1}{\sqrt{2}} ( |0\rangle - |4\rangle )$, i.e., $\langle \Phi|\phi\rangle =0$. Since the jump operator $\hat{F}_{\textrm{eng},2}$ is triggered only when the system has a non-zero population in the residual state $|\phi\rangle$, it can be turned on continuously regardless of the parity of the system. Thus, active excitation number parity measurement is not needed. 

Compared to the active QEC, autonomous QEC has more flexibility in dealing with the no-jump evolution term. In the case of active QEC, the recovery unitary operation $\hat{U}_{\textrm{rec}}^{(\textrm{even})}$ has to be fine-tuned: That is, the rotation angle between the states $|0_{\textrm{bin}}^{(1,1)}\rangle =  \frac{1}{\sqrt{2}}( |0\rangle+|4\rangle )$ and $|\phi\rangle = \frac{1}{\sqrt{2}} (|0\rangle -|4\rangle)$ has to be precisely $\kappa t + o(\kappa t)$. Hence, a precise knowledge of the loss rate $\kappa$ is needed in the case of active QEC. On the other hand, in the case of autonomous QEC, precise knowledge of the loss rate $\kappa$ is not needed. Moreover, $|\Phi\rangle\in \textrm{span}\lbrace |0\rangle, \cdots, |4\rangle \rbrace$ in Eq.\ \eqref{eq:binomial code preventive jump} can be chosen arbitrarily as long as it is perpendicular to the residual state $|\phi\rangle$. This flexibility can be used to make the experimental implementation more feasible. For instance, by choosing $|\Phi\rangle = |2\rangle$, we have 
\begin{align}
\hat{F}_{\textrm{eng},2} &= \frac{1}{\sqrt{2}}|2\rangle ( \langle 0| - \langle 4| ). 
\end{align}
In this case, at most two-excitation exchanges (i.e., $|0\rangle \leftrightarrow |2\rangle$ and $|2\rangle \leftrightarrow |4\rangle$) are needed. For other choices of $|\Phi\rangle$, more than three-excitation exchanges are needed to empty the residual state $|\phi\rangle = \frac{1}{\sqrt{2}}( |0\rangle - |4\rangle )$. 

At glance, it might seem that the jump operator $\hat{F}_{\textrm{eng},2} = \frac{1}{\sqrt{2}}|2\rangle ( \langle 0| - \langle 4| )$ will cause a logical bit-flip error because it maps the residual state $\frac{1}{\sqrt{2}}( |0\rangle-|4\rangle)$, which is derived via the back-action term from the logical zero state $\frac{1}{\sqrt{2}}( |0\rangle+|4\rangle)$, to the logical one state $|2\rangle$. However, this is not the case because the back-action terms merely induce an undesirable coherence between the logical zero state and the residual state to the first order. In other words, the back-action terms do not immediately cause any population transfer to the residual state to the first order. Such a population transfer is a second order effect in the back-action terms. Therefore, the only role of the secondary jump operator $\hat{F}_{\textrm{eng},2}$ is to prevent any population transfer to the residual state and thus any secondary jump operator that empties the residual state suffices. See Ref.\ \cite{Lihm2018} for more details on the flexibility of the secondary jump operator $\hat{F}_{\textrm{eng},2}$ and the performance of the autonomous QEC of the $(1,1)$-binomial code. We remark that active QEC of the $(1,1)$-binomial code was demonstrated experimentally in Ref.\ \cite{Hu2019} and autonomous QEC of the $(1,1)$-binomial code was demonstrated experimentally in Ref.\ \cite{Ma2019a}.

\subsubsection{Generalization to higher-order error correction} 

As discussed above, the $(1,1)$-binomial code can correct single-excitation loss errors. In Ref.\ \cite{Michael2016}, it was shown that the $(1,1)$-binomial code can be generalized to a higher order so that the logical error probability can be suppressed to a higher order than $\mathcal{O}\big{(} (\kappa t)^{2} \big{)}$. More explicitly, we can define the logical states of the $(N,S)$-binomial code as follows:
\begin{align}
|0_{\textrm{bin}}^{(N,S)}\rangle &= \frac{1}{\sqrt{2^{N}}} \sum_{ p\textrm{ even} }^{[0,N+1]} \sqrt{ \binom{N+1}{p} } |p(S+1)\rangle, 
\nonumber\\
|1_{\textrm{bin}}^{(N,S)}\rangle &= \frac{1}{\sqrt{2^{N}}} \sum_{ p\textrm{ odd} }^{[0,N+1]} \sqrt{ \binom{N+1}{p} } |p(S+1)\rangle. 
\end{align}  
Here, $\binom{N+1}{p}$ is the binomial coefficient and this is the reason why this code family is referred to as the binomial code. Ref.\ \cite{Michael2016} showed that the $(L,L)$-binomial code can correct any $\ell$-excitation loss errors if $\ell \le L$. 

For instance, the $(2,2)$-binomial code can correct two-excitation loss errors and the logical states of the $(2,2)$-binomial code are explicitly given by 
\begin{align}
|0_{\textrm{bin}}^{(2,2)}\rangle &= \frac{|0\rangle + \sqrt{3}|6\rangle  }{2}, 
\nonumber\\
|1_{\textrm{bin}}^{(2,2)}\rangle &=\frac{\sqrt{3}|3\rangle + |9\rangle  }{2} . 
\end{align}
Since the logical states of the $(2,2)$-binomial code has $0$ excitations modulo $3$, the $(2,2)$-binomial code is stabilized by the $120\degree$ phase rotation $\hat{\Pi}_{3} \equiv e^{i\frac{\pi}{3}\hat{n}}$. Hence, active QEC of the $(2,2)$-binomial code can be implemented by measuring the excitation number modulo $3$, or equivalently the stabilizer $\hat{\Pi}_{3} \equiv e^{i\frac{\pi}{3}\hat{n}}$, and then apply an appropriate recovery unitary operation conditioned on the parity measurement outcome.  

We also remark that there are multi-mode variants of the binomial code based on $\chi^{(2)}$ non-linear interactions \cite{Niu2018,Niu2018b}. 

\subsection{Logical gates on rotation-symmetric bosonic codes}
\label{subsection:Logical gates on rotation-symmetric bosonic codes}

Certain logical gates on rotation-symmetric codes can be straightforwardly constructed by taking advantage of the rotation-symmetric structure. Consider, for instance, an even-parity bosonic code that is invariant under the $180\degree$ phase rotation. Then, the computational basis states have the following parity structure: 
\begin{align}
|0_{\textrm{even codes}}\rangle &: 0\textrm{ excitations modulo }4, 
\nonumber\\ 
|1_{\textrm{even codes}}\rangle &: 2\textrm{ excitations modulo }4. 
\end{align}    
Then, as shown in Ref.\ \cite{Grimsmo2019}, we can implement the logical Z, phase S, and T gates on these even-parity codes by using the following operations: 
\begin{align}
\hat{Z}_{\textrm{even codes}} &= \exp\Big{[} i\frac{\pi}{2} \hat{n} \Big{]} = \sum_{n=0}^{\infty} |n\rangle\langle n| \times \begin{cases}
1 & n=0\textrm{ mod }4 \\
-1 & n=2\textrm{ mod }4
\end{cases}, 
\nonumber\\
\hat{S}_{\textrm{even codes}} &= \exp\Big{[} i\frac{\pi}{8} \hat{n}^{2} \Big{]} = \sum_{n=0}^{\infty} |n\rangle\langle n| \times \begin{cases}
1 & n=0\textrm{ mod }4 \\
i & n=2\textrm{ mod }4
\end{cases},  
\nonumber\\
\hat{T}_{\textrm{even codes}} &= \exp\Big{[} i\frac{\pi}{64} \hat{n}^{4} \Big{]} = \sum_{n=0}^{\infty} |n\rangle\langle n| \times \begin{cases}
1 & n=0\textrm{ mod }4 \\
\exp[ i\frac{\pi}{4} ] & n=2\textrm{ mod }4
\end{cases}. 
\end{align}
Note that these operations impart desired phases on the logical code space of even-parity codes: 
\begin{align}
\hat{Z}_{\textrm{even codes}} &\rightarrow \begin{bmatrix}
1 & 0\\
0 & -1
\end{bmatrix}, \quad \hat{S}_{\textrm{even codes}} \rightarrow \begin{bmatrix}
1 & 0\\
0 & i
\end{bmatrix} , \quad \hat{T}_{\textrm{even codes}} \rightarrow \begin{bmatrix}
1 & 0\\
0 & e^{i\frac{\pi}{4}}
\end{bmatrix} . 
\end{align} 
These operations are special cases of single-qubit rotations along the Z axis. Similarly, we can implement the logical controlled-Z operation by using the following controlled-rotation operation: 
\begin{align}
\textrm{CZ}_{\textrm{even codes}} = \exp\Big{[}  i \frac{\pi}{2}\hat{n}_{1}\hat{n}_{2} \Big{]} . 
\end{align}
We can see that this controlled-rotation operation imparts the desired phase on the even-parity code space, i.e., 
\begin{align}
\exp\Big{[}  i \frac{\pi}{2} n_{1}n_{2} \Big{]} &= \begin{cases}
1 & (n_{1},n_{2}) = (0,0)\textrm{ mod }4 \\
1 & (n_{1},n_{2}) = (0,2)\textrm{ mod }4\\
1 & (n_{1},n_{2}) = (2,0)\textrm{ mod }4\\
-1 & (n_{1},n_{2}) = (2,2)\textrm{ mod }4
\end{cases}
\end{align}

In circuit QED systems, we can directly use a SNAP gate \cite{Krastanov2015} to implement any desired single-logical-qubit rotation along the Z axis. Recall that a general SNAP gate imparts an arbitrary phase on each excitation number state: 
\begin{align}
\hat{U}_{\textrm{SNAP}} (\vec{\theta}) &= \sum_{n=0}^{\infty} e^{i\theta_{n}}|n\rangle\langle n| ,  
\end{align} 
where $\theta_{n}$ is the number-dependent phase. Then, to implement any logical operation of the form 
\begin{align}
\begin{bmatrix}
1 & 0 \\
0 & e^{i\theta}
\end{bmatrix}, 
\end{align}
one can simply use a SNAP gate with 
\begin{align}
\theta_{n} = \begin{cases}
0 & n=0\textrm{ mod }4 \\
\theta & n=2\textrm{ mod }4
\end{cases} . 
\end{align}

While single-logical-qubit rotations along the Z axis on the even-parity codes can be readily realized by using unitary operations that are diagonal in the excitation number basis, implementing the logical Hadamard operation, i.e., 
\begin{align}
\frac{1}{\sqrt{2}}\begin{bmatrix}
1 & 1 \\
1 & -1
\end{bmatrix} , 
\end{align}
on the even-parity codes is relatively more challenging. Ref.\ \cite{Grimsmo2019} proposed a teleportation-based method for implementing the logical Hadamard gate. However, the scheme requires an additional bosonic mode encoded in an even-parity code and the use of a Pauli frame \cite{Knill2005,DiVincenzo2007,Terhal2015,Chamberland2018} to keep track of an undesired Pauli X operation that occurs with $50\%$ probability during the teleportation.         

\section{Translation-symmetric bosonic codes}
\label{section:Translation-symmetric bosonic codes}

In this section, we review the Gottesman-Kitaev-Preskill (GKP) codes \cite{Gottesman2001,Harrington2001,Harrington2004}, which are invariant under a discrete set of translations.

\subsection{The square-lattice Gottesman-Kitaev-Preskill (GKP) code} 
\label{subsection:The square-lattice Gottesman-Kitaev-Preskill (GKP) code}

GKP codes are designed to correct random shift errors in the phase space. Since shift errors can occur both in the position and the momentum directions, it is essential to measure both the position and the momentum operators to correct for the random shift errors. However, since the position and the momentum operators do not commute with each other (i.e., $[\hat{q},\hat{p}] = i$), they cannot be measured simultaneously, as implied by the Heisenberg uncertainty principle. The key idea behind the design of the square-lattice GKP code is that the following two displacement operators commute with each other nevertheless: 
\begin{align}
\hat{S}_{q} &= e^{i2\sqrt{\pi}\hat{q}}, \quad \hat{S}_{p} = e^{-i2\sqrt{\pi}\hat{p}}.  
\end{align}
Thus, these two displacement operators can be measured simultaneously. Measuring the displacement operator $\hat{S}_{q} = e^{i2\sqrt{\pi}\hat{q}}$ (or $\hat{S}_{p} = e^{-i2\sqrt{\pi}\hat{p}}$) is equivalent to measuring its phase angle $2\sqrt{\pi}\hat{q}$ (or $-2\sqrt{\pi}\hat{p}$) modulo $2\pi$. Hence, the commutativity of the two displacement operators $\hat{S}_{q}$ and $\hat{S}_{p}$ implies that we can simultaneously measure both the position and the momentum operators $\hat{q}$ and $\hat{p}$ modulo $\sqrt{\pi}$. In this subsection, we provide a detailed review of the square-lattice GKP states. See Table \ref{table:square-lattice GKP code} for a summary.

\begin{table}[t!]
  \centering
     \def\arraystretch{1.5}
  \begin{tabular}{ V{3} c V{1.5} c V{3} }
   \hlineB{3}  
     & \textbf{The square-lattice GKP code} \cite{Gottesman2001}  \\  \hlineB{3} 
     Logical states & $|0_{\textrm{gkp}}^{(\textrm{sq})}\rangle = \sum_{n\in\mathbb{Z}} |\hat{q} = 2n\sqrt{\pi} \rangle $   \\
     &  $|1_{\textrm{gkp}}^{(\textrm{sq})}\rangle = \sum_{n\in\mathbb{Z}} |\hat{q} = (2n+1)\sqrt{\pi} \rangle  $  \\ \hlineB{1.5}   
    Correctable errors & Random shift errors in the phase space   \\
    & $e^{i( \xi_{p}\hat{q} - \xi_{q}\hat{p}) }$ with $|\xi_{q}|,|\xi_{p}| < \frac{\sqrt{\pi}}{2}$ \\     \hlineB{1.5}
    Stabilizers & $\hat{S}_{q} = e^{i2\sqrt{\pi} \hat{q} }$ \\
     & $\hat{S}_{p} = e^{-i2\sqrt{\pi} \hat{p} }$ \\ \hlineB{1.5}
     Logical Pauli operations & $\hat{Z}_{\textrm{gkp}}  = (\hat{S}_{q})^{\frac{1}{2}} = e^{i\sqrt{\pi} \hat{q} }$ \\
    & $\hat{X}_{\textrm{gkp}}  = (\hat{S}_{p})^{\frac{1}{2}} = e^{-i\sqrt{\pi} \hat{p} }$ \\ \hlineB{1.5}
    Logical Pauli measurement & Homodyne measurement of a quadrature operator  \\ \hlineB{1.5}
    Logical Clifford operations & $\hat{S}_{\textrm{gkp}}  = e^{ i \hat{q}^{2}/2 }$ \\
    & $\hat{H}_{\textrm{gkp}}  = e^{ i (\pi/2)\hat{a}^{\dagger}\hat{a} }$ \\
    & $\textrm{CNOT}_{\textrm{gkp}}^{j\rightarrow k}  = \textrm{SUM}_{j\rightarrow k} = e^{-i\hat{q}_{j}\hat{p}_{k}}$ \\ \hlineB{1.5}
    Active QEC & Measurement of the stabilizers $\hat{S}_{q}$ and $\hat{S}_{p}$  \\ \hdashline
    Experiments & Refs.\ \cite{Fluhmann2018,Fluhmann2019,Fluhmann2019b,Campagne2019}  \\  \hlineB{1.5}     
    Magic state preparation & Stabilizer measurements on the vacuum state $|0\rangle$ \cite{Terhal2016,Baragiola2019} \\  \hlineB{3}     
  \end{tabular}
  \caption{Basic properties of the square-lattice GKP code. }
  \label{table:square-lattice GKP code}
\end{table}

\subsubsection{Stabilizers of the square-lattice GKP code}

The square-lattice GKP code $\mathcal{C}_{\textrm{gkp}}^{(\textrm{sq})}$ is defined as the space of the state vectors $|\psi_{\textrm{gkp}}^{(\textrm{sq})}\rangle$ that are stabilized by the two stabilizers $\hat{S}_{q}$ and $\hat{S}_{p}$, i.e., 
\begin{align}
\hat{S}_{q}|\psi_{\textrm{gkp}}^{(\textrm{sq})}\rangle = \hat{S}_{p}|\psi_{\textrm{gkp}}^{(\textrm{sq})}\rangle = |\psi_{\textrm{gkp}}^{(\textrm{sq})}\rangle, \,\,\, \textrm{for all}\,\,\, |\psi_{\textrm{gkp}}^{(\textrm{sq})}\rangle \in \mathcal{C}_{\textrm{gkp}}^{(\textrm{sq})}. 
\end{align}  
Since $|\psi_{\textrm{gkp}}^{(\textrm{sq})}\rangle$ is stabilized by $\hat{S}_{q}$ and $\hat{S}_{p}$, it satisfies $\hat{q} = \hat{p} =0$ mod $\sqrt{\pi}$. The square-lattice GKP code $\mathcal{C}_{\textrm{gkp}}^{(\textrm{sq})}$ is two-dimensional and in the computational basis, the two logical states are explicitly given by 
\begin{align}
|0_{\textrm{gkp}}^{(\textrm{sq})}\rangle &= \sum_{n\in\mathbb{Z}} |\hat{q} = 2n\sqrt{\pi}\rangle, 
\nonumber\\
|1_{\textrm{gkp}}^{(\textrm{sq})}\rangle &= \sum_{n\in\mathbb{Z}} |\hat{q} = (2n+1)\sqrt{\pi}\rangle.  
\end{align}
These logical states are visualized in Fig.\ \ref{fig:wigner functions square-lattice GKP code}.

\begin{figure}[t!]
\centering
\includegraphics[width=6.0in]{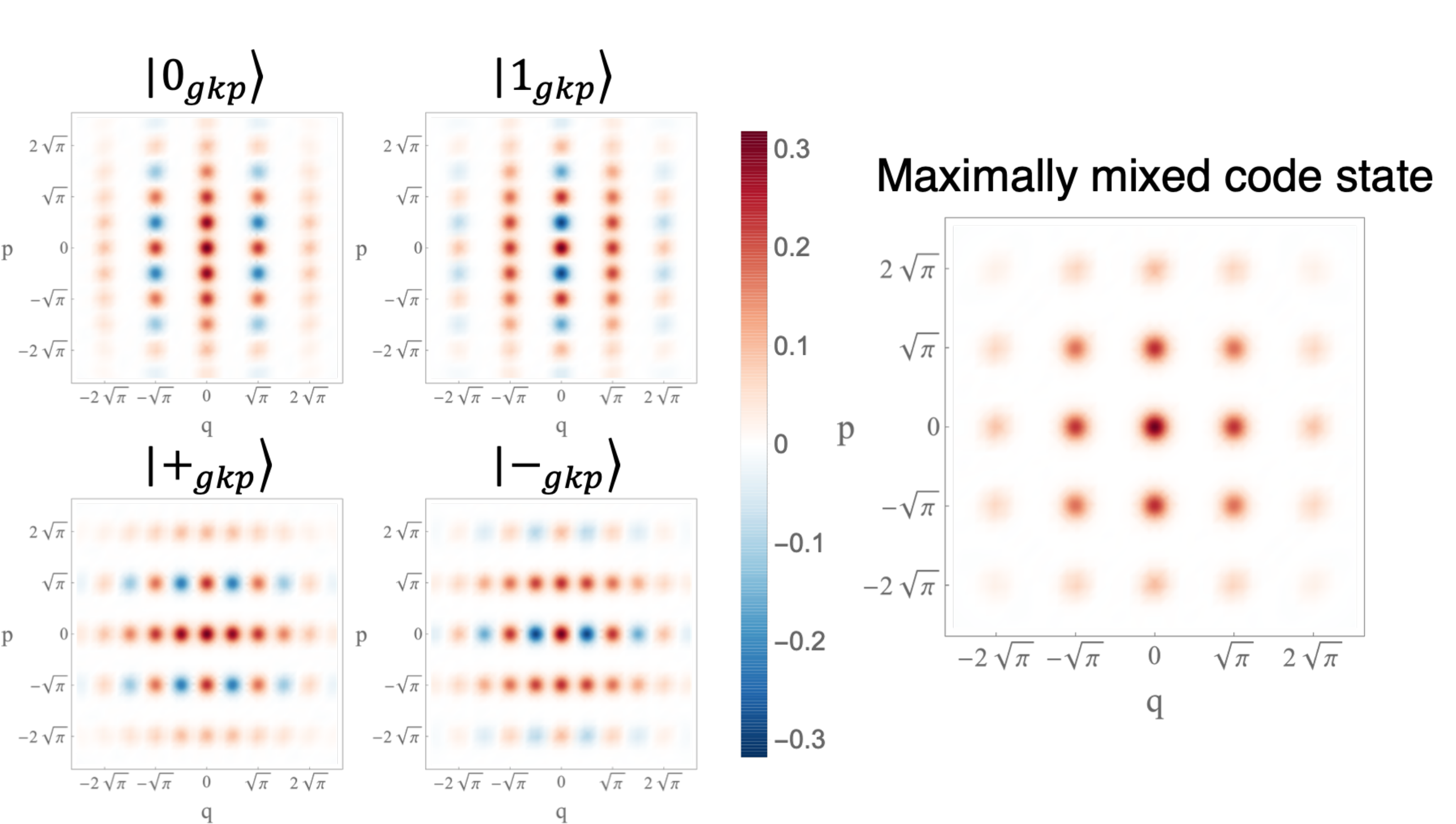}
\caption{Wigner functions of the logical states of the square-lattice GKP code $\mathcal{C}_{\textrm{gkp}}^{(\textrm{sq})}$ with an average photon number $\bar{n}=5$. The maximally mixed code state is defined as the projection operator to the code space divided by $2$. }
\label{fig:wigner functions square-lattice GKP code}
\end{figure}

In the computational basis, it is apparent that the logical states satisfy $\hat{q}=0$ mod $\sqrt{\pi}$ and thus are stabilized by $\hat{S}_{q}$. Also, one can explicitly see 
\begin{align}
\hat{S}_{q}|\mu_{\textrm{gkp}}^{(\textrm{sq})}\rangle &= e^{i2\sqrt{\pi}\hat{q}} \sum_{n\in\mathbb{Z}} |\hat{q} = (2n+\mu)\sqrt{\pi}\rangle 
\nonumber\\
&= \sum_{n\in\mathbb{Z}} e^{i2\pi(2n+\mu)}  |\hat{q} = (2n+\mu)\sqrt{\pi}\rangle  = \sum_{n\in\mathbb{Z}}  |\hat{q} = (2n+\mu)\sqrt{\pi}\rangle = |\mu_{\textrm{gkp}}^{(\textrm{sq})}\rangle. 
\end{align}
While it is less apparent, these logical states also satisfy $\hat{p}=0$ mod $\sqrt{\pi}$ and thus are also stabilized by $\hat{S}_{p}$. Indeed, one can explicitly confirm that  
\begin{align}
\hat{S}_{p} |\mu_{\textrm{gkp}}^{(\textrm{sq})}\rangle &= e^{-i2\sqrt{\pi}\hat{p}} \sum_{n\in\mathbb{Z}} |\hat{q} = (2n+\mu)\sqrt{\pi}\rangle 
\nonumber\\
&= \sum_{n\in\mathbb{Z}} |\hat{q} = (2(n+1)+\mu)\sqrt{\pi}\rangle = \sum_{n\in\mathbb{Z}} |\hat{q} = (2n+\mu)\sqrt{\pi}\rangle = |\mu_{\textrm{gkp}}^{(\textrm{sq})}\rangle, 
\end{align}
for all $\mu \in \lbrace 0,1 \rbrace$. 

In the complementary basis, logical states of the square-lattice GKP code are given by 
\begin{align}
|+_{\textrm{gkp}}^{(\textrm{sq})}\rangle &\equiv \frac{1}{\sqrt{2}}( |0_{\textrm{gkp}}^{(\textrm{sq})}\rangle + |1_{\textrm{gkp}}^{(\textrm{sq})}\rangle ) =  \sum_{n\in\mathbb{Z}} |\hat{p} = 2n\sqrt{\pi}\rangle, 
\nonumber\\
|-_{\textrm{gkp}}^{(\textrm{sq})}\rangle &\equiv \frac{1}{\sqrt{2}}( |0_{\textrm{gkp}}^{(\textrm{sq})}\rangle - |1_{\textrm{gkp}}^{(\textrm{sq})}\rangle ) = \sum_{n\in\mathbb{Z}} |\hat{p} = (2n+1)\sqrt{\pi}\rangle.  
\end{align}
In this basis, it is more apparent that the logical states satisfy $\hat{p}=0$ mod $\sqrt{\pi}$ and thus are stabilized by $\hat{S}_{p}$. While it is less apparent, $|\pm_{\textrm{gkp}}^{(\textrm{sq})}\rangle$ is also stabilized by $\hat{S}_{q}$ (hence satisfying $\hat{q}=0$ mod $\sqrt{\pi}$) because $|\pm_{\textrm{gkp}}^{(\textrm{sq})}\rangle$ is a linear combination of $|0_{\textrm{gkp}}^{(\textrm{sq})}\rangle$ and $|1_{\textrm{gkp}}^{(\textrm{sq})}\rangle$ which are stabilized by $\hat{S}_{q}$. 

\subsubsection{Logical Pauli operators of the square-lattice GKP code}

Logical Pauli operators on the square-lattice GKP code can be readily implemented by using displacement operations. More specifically, the logical Z and X operators on the square-lattice GKP code are given by the square root of the stabilizers $\hat{S}_{q}$ and $\hat{S}_{p}$, i.e.,
\begin{align}
\hat{Z}_{\textrm{gkp}} &= ( \hat{S}_{q} )^{\frac{1}{2}}  = e^{i\sqrt{\pi}\hat{q}}, 
\nonumber\\
\hat{X}_{\textrm{gkp}} &= ( \hat{S}_{p} )^{\frac{1}{2}}  = e^{-i\sqrt{\pi}\hat{p}}. 
\end{align}  
Indeed, one can explicitly check that these displacement operators act as a logical Pauli operation on the square-lattice GKP code in the computational basis: 
\begin{align}
\hat{Z}_{\textrm{gkp}} |\mu_{\textrm{gkp}}^{(\textrm{sq})}\rangle &= e^{i\sqrt{\pi}\hat{q}} \sum_{n\in\mathbb{Z}} |\hat{q} = (2n+\mu)\sqrt{\pi}\rangle 
\nonumber\\
&= \sum_{n\in\mathbb{Z}} e^{i\pi(2n+\mu)}  |\hat{q} = (2n+\mu)\sqrt{\pi}\rangle = (-1)^{\mu} |\mu_{\textrm{gkp}}^{(\textrm{sq})}\rangle, 
\nonumber\\
\hat{X}_{\textrm{gkp}} |\mu_{\textrm{gkp}}^{(\textrm{sq})}\rangle &= e^{-i\sqrt{\pi}\hat{p}} \sum_{n\in\mathbb{Z}} |\hat{q} = (2n+\mu)\sqrt{\pi}\rangle 
\nonumber\\
&= \sum_{n\in\mathbb{Z}}  |\hat{q} = (2n+\mu+1)\sqrt{\pi}\rangle =  |(\mu\oplus 1)_{\textrm{gkp}}^{(\textrm{sq})}\rangle, 
\end{align}
for all $\mu\in\lbrace 0,1 \rbrace$ as desired. Here, $\oplus$ is the addition modulo $2$. 

Note that Pauli measurements for the square-lattice GKP code can be implemented by using a homodyne measurement. For instance, the Pauli Z measurement on the GKP code can be implemented by performing a homodyne measurement of the position quadrature $\hat{q}$. Similarly, the Pauli X measurement can be done by performing a homodyne measurement of the momentum quadrature $\hat{p}$. In the case of the Pauli Z measurement, the logical zero (one) state should ideally yield a measurement outcome that is an even (odd) integer multiple of $\sqrt{\pi}$. However, if the measurement outcome is noisy, the measurement outcome may take a value that is not precisely an integer multiple $\sqrt{\pi}$. In general, therefore, if the measurement outcome $z$ lies in the range $|z-n\sqrt{\pi}| < \frac{\sqrt{\pi}}{2}$ for some even (odd) $n$, we conclude that $\hat{Z}_{\textrm{gkp}} = 1$ ($\hat{Z}_{\textrm{gkp}} = -1$). At this point, we can already see that the GKP code has some robustness against shift errors in the phase space. We will make this even clearer below. We also remark that because the Pauli measurements may be destructive, modular quadrature measurements are not necessary.

\subsubsection{Error correction capability of the square-lattice GKP code}

Let us now discuss the error correction capability of the square-lattice GKP code. Recall that the logical states of the square-lattice GKP code are stabilized by $\hat{S}_{q} = e^{i2\sqrt{\pi}\hat{q}}$ and $\hat{S}_{p} = e^{-i2\sqrt{\pi}\hat{p}}$ and thus satisfy $\hat{q} = \hat{p} = 0$ mod $\sqrt{\pi}$. Thus, a natural way to detect errors acting on the GKP code is to measure its stabilizers $\hat{S}_{q}$ and $\hat{S}_{p}$, or equivalently, the position and the momentum quadrature operators modulo $\sqrt{\pi}$. If the modular quadrature measurement outcomes deviate from the desired result $\hat{q} = \hat{p} =  0$ mod $\sqrt{\pi}$, we can infer that there was an error. 

Since the GKP code works by measuring the quadrature operators, it is perfectly suited for correcting random shift errors in the phase space. To make the discussion more concrete, let us consider a random shift error that adds random noise to the position and the momentum quadrature operators: 
\begin{align}
\hat{q} &\rightarrow \hat{q} + \xi_{q}, 
\nonumber\\
\hat{p} &\rightarrow \hat{q} + \xi_{p} .  
\end{align}
Here, $\xi_{q}$ and $\xi_{p}$ are the position and the momentum quadrature noise, respectively, and are random variables drawn from a probability distribution. For example in the case of the Gaussian random shift error $\mathcal{N}_{B_{2}}[\sigma]$ (see Subsection \ref{subsection:Random shift errors in the phase space} for the definition and more details on $\mathcal{N}_{B_{2}}[\sigma]$), the stochastic variables $\xi_{q}$ and $\xi_{p}$ are drawn from an independent and identically distributed Gaussian distribution with zero mean and variance $\sigma^{2}$, i.e., $\xi_{q},\xi_{p} \sim_{\textrm{iid}}\mathcal{N}(0,\sigma^{2})$. 

If the square-lattice GKP code undergoes a random shift error $\boldsymbol{\xi} = (\xi_{q},\xi_{p})$, the logical states will be displaced and have $\hat{q} = \xi_{q}$ and $\hat{p} = \xi_{p}$ mod $\sqrt{\pi}$. Then, by measuring the position and the momentum operators modulo $\sqrt{\pi}$, we can extract the values of the random shifts $\xi_{q}$ and $\xi_{p}$ modulo $\sqrt{\pi}$. That is, after the stabilizer measurements, we know that $\xi_{q}$ and $\xi_{p}$ are given by 
\begin{align}
\xi_{q} &= R_{\sqrt{\pi}}( \xi_{q} ) + n_{q}\sqrt{\pi}, 
\nonumber\\
\xi_{p} &= R_{\sqrt{\pi}}( \xi_{p} ) + n_{p}\sqrt{\pi}, \label{eq:shift error candidates}
\end{align}   
for some integers $n_{q},n_{p}\in \mathbb{Z}$. Here, $R_{s}( z )$ is defined as 
\begin{align}
R_{s}(z) &\equiv z - n^{\star}s \,\,\, \textrm{where}\,\,\, n^{\star} = \textrm{argmin}_{n\in\mathbb{Z}} |z-ns| .  \label{eq:definition of the R function}
\end{align}
Equivalently, $R_{s}(z)$ is a periodic function with period $s$ such that $R_{s}(z) = z$ for $|z| < \frac{s}{2}$. Note that in Eq.\ \eqref{eq:shift error candidates}, we are given with multiple candidates of $\xi_{q}$ and $\xi_{p}$ that are compatible with the modular measurement outcome. In the maximum likelihood estimation, we decide that $\xi_{q}$ and $\xi_{p}$ are the ones with the largest likelihood (or probability) among the compatible error candidates. Typically, smaller shift errors are more likely to occur than larger shift errors. For instance, this is certainly true in the case of Gaussian random shift errors. In this case, the maximum likelihood estimation is equivalent to inferring that $\xi_{q}$ and $\xi_{p}$ are the ones with the smallest size among the compatible error candidates. Hence, assuming smaller shifts are more likely, we infer that $\xi_{q}$ and $\xi_{p}$ are given by 
\begin{align}
\bar{\xi}_{q} &= R_{\sqrt{\pi}}( \xi_{q} ), 
\nonumber\\
\bar{\xi}_{p} &= R_{\sqrt{\pi}}( \xi_{p} ).  
\end{align}  
Then, to correct for the shifts, we apply the counter displacement operations $\exp[ -i\bar{\xi}_{p}\hat{q} ]$ and $\exp[i\bar{\xi}_{q}\hat{p}]$ based on the estimates $\bar{\xi}_{q}$ and $\bar{\xi}_{p}$. 

If the shift errors $\xi_{q}$ and $\xi_{p}$ are small enough to be contained in the square $|\xi_{q}|,|\xi_{p}| < \frac{\sqrt{\pi}}{2}$, we have 
\begin{align}
\bar{\xi}_{q} &= R_{\sqrt{\pi}}( \xi_{q} ) = \xi_{q}, 
\nonumber\\
\bar{\xi}_{p} &= R_{\sqrt{\pi}}( \xi_{p} ) = \xi_{p}. 
\end{align} 
Thus in this case, the maximum likelihood estimation succeeds and we can completely remove the shift errors. Thus, any small shift errors such that $|\xi_{q}|,|\xi_{p}| < \frac{\sqrt{\pi}}{2}$ can be corrected. 

On the other hand, let us consider the case where the position shift is small $|\xi_{q}| < \frac{\sqrt{\pi}}{2}$ but the momentum shift is not, e.g., $\frac{\sqrt{\pi}}{2}<\xi_{p} < \frac{3\sqrt{\pi}}{2}$. In this case, the inferred shifts are given by 
\begin{align}
\bar{\xi}_{q} &= R_{\sqrt{\pi}}( \xi_{q} ) = \xi_{q}, 
\nonumber\\
\bar{\xi}_{p} &= R_{\sqrt{\pi}}( \xi_{p} ) = \xi_{p} - \sqrt{\pi}. 
\end{align}    
Hence, the maximum likelihood estimation fails in the momentum direction. In particular, the momentum shift error will be under-corrected by the counter displacement operation. Thus, we are left with the following overall shift 
\begin{align}
e^{i(\xi_{p}-\bar{\xi}_{p})\hat{q}} &= e^{i\sqrt{\pi}\hat{q}} = \hat{Z}_{\textrm{gkp}}   
\end{align}
and the GKP code will undergo a Pauli Z error at the end of the error correction protocol. Similarly, if the momentum shift error is small $|\xi_{p}| < \frac{\sqrt{\pi}}{2}$ but the position shift is not, e.g., $\frac{\sqrt{\pi}}{2}<\xi_{q} < \frac{3\sqrt{\pi}}{2}$, the inferred shifts are given by 
\begin{align}
\bar{\xi}_{q} &= R_{\sqrt{\pi}}( \xi_{q} ) = \xi_{q} - \sqrt{\pi}, 
\nonumber\\
\bar{\xi}_{p} &= R_{\sqrt{\pi}}( \xi_{p} ) = \xi_{p} .  
\end{align}    
Hence, the maximum likelihood estimation fails in the position direction and we under-correct the position shift error. As a result, we are left with the overall shift
\begin{align}
e^{-i(\xi_{q}-\bar{\xi}_{q})\hat{p}} &= e^{-i\sqrt{\pi}\hat{p}} = \hat{X}_{\textrm{gkp}}   
\end{align}
and the GKP code will undergo a Pauli X error at the end of the error correction protocol. One can similarly show that a Pauli Y error occurs after the error correction protocol if both $\xi_{q}$ and $\xi_{p}$ lie in the range $[\frac{\sqrt{\pi}}{2},\frac{3\sqrt{\pi}}{2}]$.  

In the most general case, given the shifts $\xi_{q}$ and $\xi_{p}$, we can find the integers $n_{q}$ and $n_{p}$ that satisfy
\begin{align}
\Big{(}n_{q} -\frac{1}{2}\Big{)}\sqrt{\pi}< \xi_{q} < \Big{(}n_{q} +\frac{1}{2}\Big{)}\sqrt{\pi}, 
\nonumber\\
\Big{(}n_{p} -\frac{1}{2}\Big{)}\sqrt{\pi}< \xi_{p} < \Big{(}n_{p} +\frac{1}{2}\Big{)}\sqrt{\pi} . 
\end{align}
In this case, $R_{\sqrt{\pi}}(\xi_{q})$ and $R_{\sqrt{\pi}}(\xi_{p})$ are given by $R_{\sqrt{\pi}}(\xi_{q}) = \xi_{q} - n_{q}\sqrt{\pi}$ and $R_{\sqrt{\pi}}(\xi_{p}) = \xi_{p} - n_{p}\sqrt{\pi}$. Then, after the syndrome measurement and the error correction, we are left with the overall shifts 
\begin{align}
e^{i(\xi_{p} - R_{\sqrt{\pi}}(\xi_{p}) )\hat{q}} &= e^{i n_{p} \sqrt{\pi}\hat{q}} = ( \hat{Z}_{\textrm{gkp}} )^{n_{p}}, 
\nonumber\\
e^{-i(\xi_{q} - R_{\sqrt{\pi}}(\xi_{q}) )\hat{p}} &= e^{-i n_{q}\sqrt{\pi} \hat{p}} = ( \hat{X}_{\textrm{gkp}} )^{n_{q}} .  
\end{align}
Thus, the GKP code will undergo the following error depending on the parity of $n_{q}$ and $n_{p}$:  
\begin{align}
\begin{cases}
\textrm{no error} & (n_{q},n_{p}) =(\textrm{even},\textrm{even}) \\
\textrm{Pauli Z error} & (n_{q},n_{p}) =(\textrm{even},\textrm{odd}) \\
\textrm{Pauli X error} & (n_{q},n_{p}) =(\textrm{odd},\textrm{even}) \\
\textrm{Pauli Y error} & (n_{q},n_{p}) =(\textrm{odd},\textrm{odd}) 
\end{cases}  . 
\end{align}
Hence, if we consider the Gaussian random shift error $\mathcal{N}_{B_{2}}[\sigma]$ with the noise standard deviation $\sigma$, the success probability of the error correction protocol is given by 
\begin{align}
p_{\textrm{succ}}^{(\textrm{sq})}(\sigma) &= \sum_{n_{q},n_{p} \in 2\mathbb{Z} } \int_{(n_{q}-\frac{1}{2})\sqrt{\pi}}^{(n_{q}+\frac{1}{2})\sqrt{\pi}}d\xi_{q}\int_{(n_{p}-\frac{1}{2})\sqrt{\pi}}^{(n_{p}+\frac{1}{2})\sqrt{\pi}}d\xi_{p} \frac{1}{2\pi\sigma^{2}} \exp\Big{[} -\frac{\xi_{q}^{2} + \xi_{p}^{2} }{2\sigma^{2}} \Big{]} , 
\end{align}
where $2\mathbb{Z}$ is the set of even integers. The failure probability $p_{\textrm{fail}}^{(\textrm{sq})}(\sigma)$ is then defined as $p_{\textrm{fail}}^{(\textrm{sq})}(\sigma) = 1- p_{\textrm{succ}}^{(\textrm{sq})}(\sigma)$. 

\begin{figure}[t!]
\centering
\includegraphics[width=5.0in]{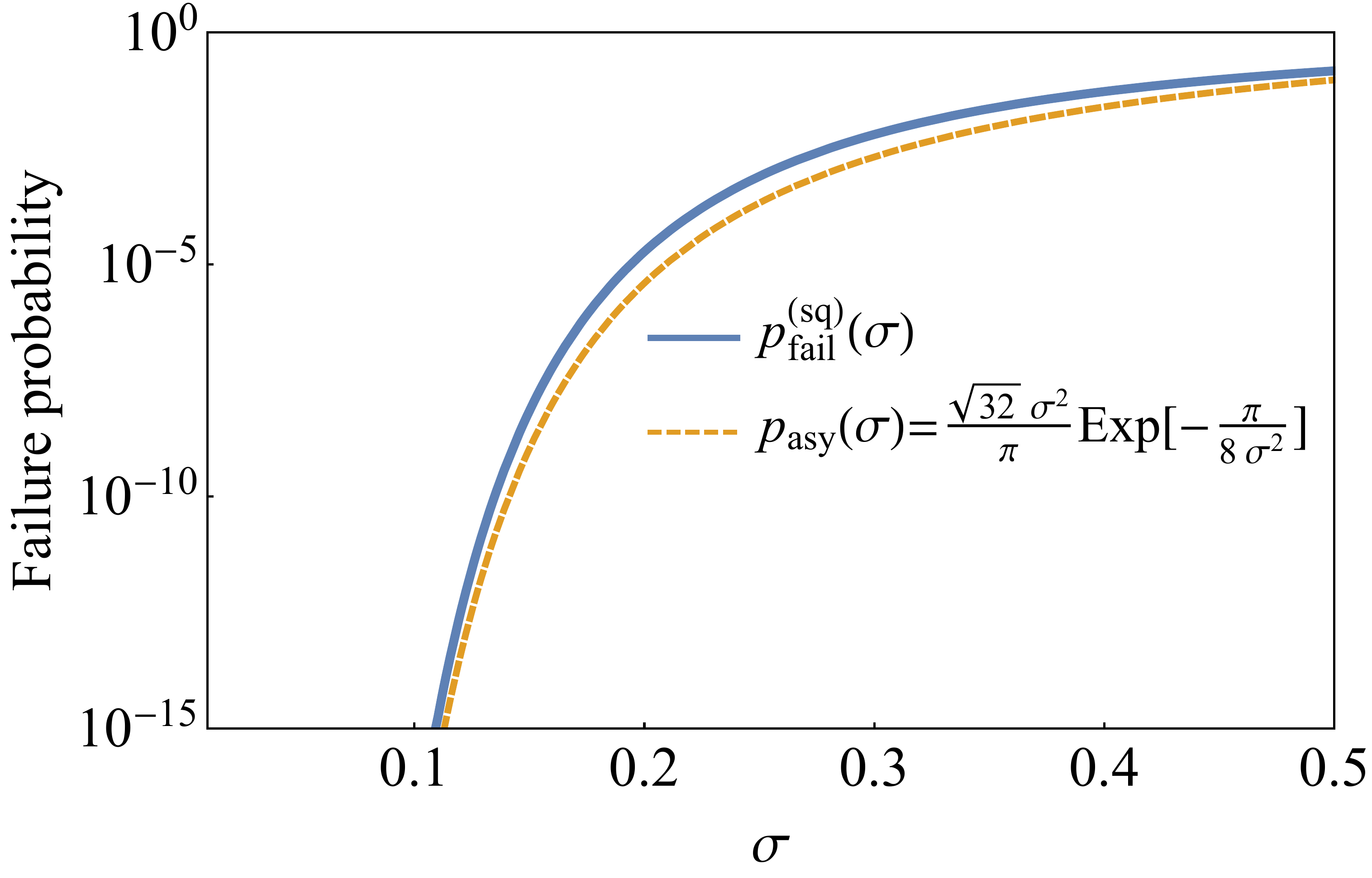}
\caption{Failure probability of the square-lattice GKP code subject to the Gaussian random shift error $\mathcal{N}_{B_{2}}[\sigma]$, i.e., $p_{\textrm{fail}}^{(\textrm{sq})}(\sigma) \equiv 1- p_{\textrm{succ}}^{(\textrm{sq})}(\sigma)$ (solid blue line). The asymptotic expression $p_{\textrm{asy}}(\sigma) = \frac{\sqrt{32}\sigma}{\pi} \exp[ -\frac{\pi}{8\sigma^{2}} ]$ is represented by the dashed orange line. Note that the asymptotic expression agrees well with the exact result in the $\sigma \ll \sqrt{\pi}$ limit.  }
\label{fig:square lattice GKP code failure probability}
\end{figure}

An important spacial case is when the noise standard deviation $\sigma$ is much smaller than the spacing of the square-lattice GKP code $\sqrt{\pi}$ (i.e., $\sigma \ll \sqrt{\pi}$). In this case, the $(n_{q},n_{p}) = (0,0)$ term dominates and we have 
\begin{align}
p_{\textrm{succ}}^{(\textrm{sq})}(\sigma) &\xrightarrow{\sigma \ll \sqrt{\pi} }  \int_{-\frac{\sqrt{\pi}}{2} }^{ \frac{\sqrt{\pi}}{2} }d\xi_{q} \frac{1}{\sqrt{2\pi\sigma^{2}}} \exp\Big{[} -\frac{\xi_{q}^{2} }{2\sigma^{2}} \Big{]}  \int_{-\frac{\sqrt{\pi}}{2} }^{ \frac{\sqrt{\pi}}{2} } d\xi_{p} \frac{1}{\sqrt{ 2\pi\sigma^{2}} } \exp\Big{[} -\frac{ \xi_{p}^{2} }{2\sigma^{2}} \Big{]} 
\nonumber\\
&= \Big{[} \textrm{erf}\Big{(}  \frac{\sqrt{ \pi }}{ \sqrt{8} \sigma} \Big{)} \Big{]}^{2}  \xrightarrow{\sigma \ll \sqrt{\pi} } \Big{[} 1- \frac{\sqrt{8}\sigma}{\pi} \exp\Big{[} -\frac{\pi}{8\sigma^{2}} \Big{]} \Big{]}^{2} \xrightarrow{\sigma \ll \sqrt{\pi} } 1- \frac{\sqrt{32}\sigma}{\pi} \exp\Big{[} -\frac{\pi}{8\sigma^{2}} \Big{]}. 
\end{align}   
Here, $\textrm{erf}(x)$ is the error function defined as $\textrm{erf}(x) \equiv \frac{1}{\sqrt{\pi}} \int_{-x}^{x}dt e^{-t^{2}} $. To derive the third line, we used $\textrm{erfc}(x)\equiv 1-\textrm{erf}(x) \xrightarrow{x\gg 1} \frac{e^{-x^{2}}}{x\sqrt{\pi}}$. Thus, the failure probability of the square-lattice GKP code is given by 
\begin{align}
p_{\textrm{fail}}^{(\textrm{sq})}(\sigma)  &\equiv 1-p_{\textrm{succ}}^{(\textrm{sq})}(\sigma) \xrightarrow{\sigma \ll \sqrt{\pi} } \frac{\sqrt{32}\sigma}{\pi} \exp\Big{[} -\frac{\pi}{8\sigma^{2}} \Big{]}
\end{align} 
and decreases very rapidly as the noise standard deviation $\sigma$ decreases. See Fig.\ \ref{fig:square lattice GKP code failure probability} for the illustration.    

\subsubsection{Measurement of the stabilizers of the square-lattice GKP code}

Note that it is essential to measure the position and the momentum operators modulo $\sqrt{\pi}$ (or equivalently, the stabilizers $\hat{S}_{q}$ and $\hat{S}_{p}$) to achieve the excellent performance of the square-lattice GKP code shown in Fig.\ \ref{fig:square lattice GKP code failure probability}. It is relatively straightforward to measure a quadrature operator (e.g., the position or the momentum operator) via a homodyne measurement which is a Gaussian measurement. However, measurement of a quadrature operator modulo some spacing $s$ (e.g., $\hat{q}$ and $\hat{p}$ modulo $\sqrt{\pi}$) is a non-Gaussian measurement. Thus, performing a modulo quadrature measurement is more challenging than performing a homodyne measurement.    

One of many desirable properties of the GKP code is that we can measure its stabilizers by using Gaussian operations and consuming a GKP state $|0_{\textrm{gkp}}^{(\textrm{sq})}\rangle$ or $|+_{\textrm{gkp}}^{(\textrm{sq})}\rangle$. That is, preparation of a GKP state is the only non-Gaussian resource needed for implementing the stabilizer measurements. Since the required GKP states can be prepared offline and then supplied to the stabilizer measurement chain whenever they are needed, no online non-Gaussian operations (e.g., cubic phase gate \cite{Gottesman2001} or Kerr nonlinearity \cite{Lloyd1999}) are needed to implement the GKP error correction. More details on the preparation of the GKP states will be provided later in the chapter. Here, we assume that the GKP states $|0_{\textrm{gkp}}^{(\textrm{sq})}\rangle$ and $|+_{\textrm{gkp}}^{(\textrm{sq})}\rangle$ are already available and discuss how they can be used to measure the stabilizers of the GKP code.      

\begin{figure}[t!]
\centering
\includegraphics[width=3.5in]{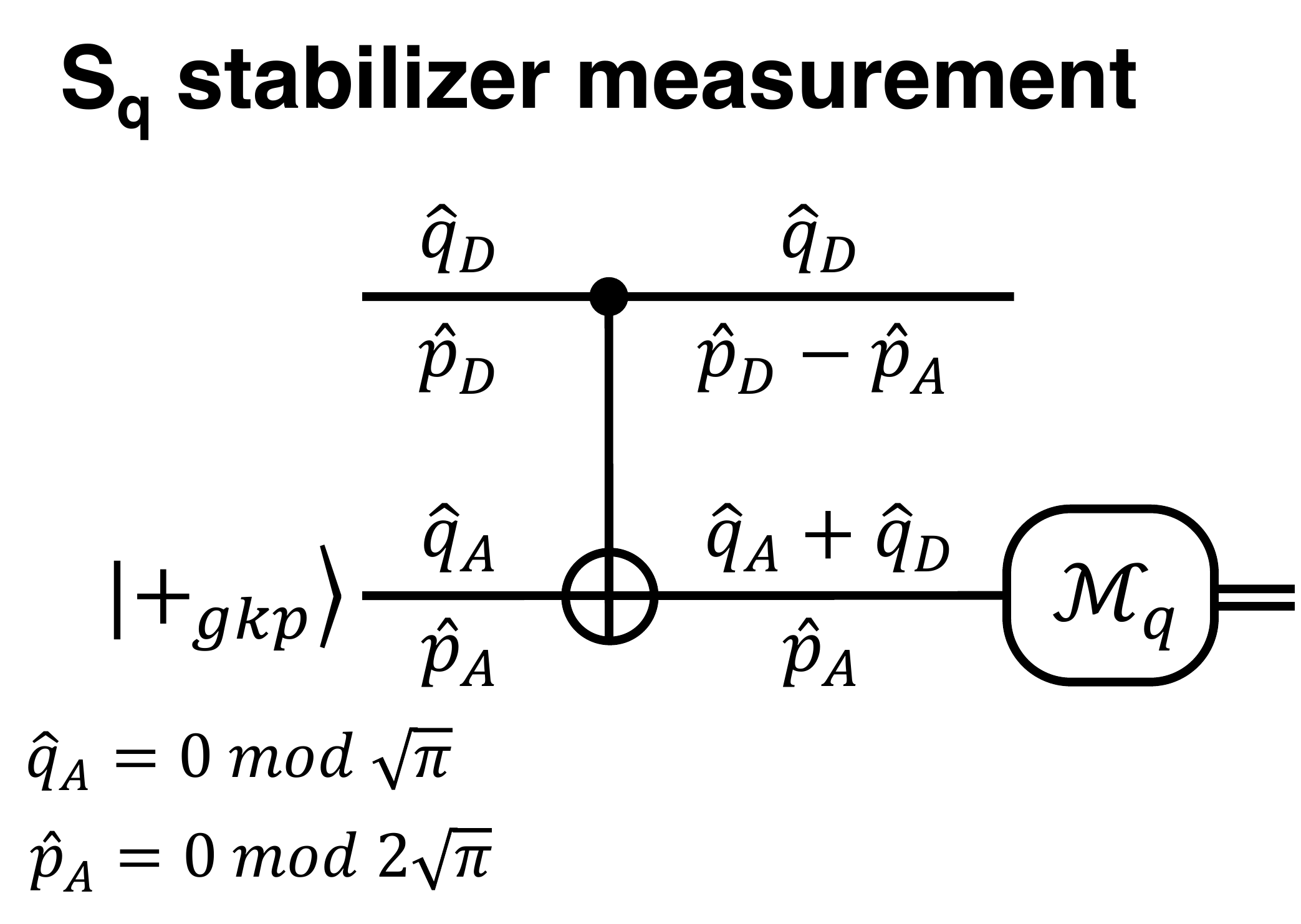}
\caption{Measurement of the $\hat{S}_{q}= e^{i2\sqrt{\pi}\hat{q}}$ stabilizer of the square-lattice GKP code. The upper and the lower lines represent a data and an ancilla mode, respectively. The controlled-$\oplus$ symbol represents the SUM gate $\textrm{SUM}_{D\rightarrow A} = \exp[ -i\hat{q}_{D}\hat{p}_{A} ]$ and $\mathcal{M}_{q}$ represents the homodyne measurement of the position operator.   }
\label{fig:GKP stabilizer measurement position}
\end{figure}

The circuits for the measurement of the GKP stabilizers $\hat{S}_{q}$ and $\hat{S}_{p}$ are given in Figs.\ \ref{fig:GKP stabilizer measurement position} and \ref{fig:GKP stabilizer measurement momentum}, respectively. Let us first consider the $\hat{S}_{q}$ stabilizer measurement, i.e., the measurement of the position operator of the data mode $\hat{q}_{D}$ modulo $\sqrt{\pi}$. Note that the GKP state in the ancilla mode $|+_{\textrm{gkp}}^{(\textrm{sq})}\rangle$ is given by 
\begin{align}
|+_{\textrm{gkp}}^{(\textrm{sq})}\rangle &= \sum_{n\in \mathbb{Z}} |\hat{p}_{A} = 2\sqrt{\pi} n\rangle = \frac{1}{\sqrt{2}} ( |0_{\textrm{gkp}}^{(\textrm{sq})}\rangle + |1_{\textrm{gkp}}^{(\textrm{sq})}\rangle ) = \frac{1}{\sqrt{2}} \sum_{n\in\mathbb{Z}} |\hat{q}_{A} = n\sqrt{\pi}\rangle. 
\end{align}  
Thus, the quadrature operators of the ancilla mode satisfy 
\begin{align}
\hat{q}_{A} &=0 \,\,\, \textrm{mod} \,\,\, \sqrt{\pi}, 
\nonumber\\
\hat{p}_{A} &=0 \,\,\, \textrm{mod} \,\,\, 2\sqrt{\pi} . \label{eq:position and momentum quadrature of the plus GKP state}
\end{align} 
One can also see this by observing that the state $|+_{\textrm{gkp}}^{(\textrm{sq})}\rangle$ is stabilized by the stabilizer $\hat{S}_{q}^{(A)} = e^{i2\sqrt{\pi}\hat{q}_{A}}$ and the Pauli X operator $\hat{X}_{\textrm{gkp}}^{(A)} = e^{-i\sqrt{\pi}\hat{p}_{A}}$, i.e., $\hat{S}_{q}^{(A)} |+_{\textrm{gkp}}^{(\textrm{sq})}\rangle =  \hat{X}_{\textrm{gkp}}^{(A)}  |+_{\textrm{gkp}}^{(\textrm{sq})}\rangle =  |+_{\textrm{gkp}}^{(\textrm{sq})}\rangle$. Then in the Heisenberg picture, the SUM gate 
\begin{align}
\textrm{SUM}_{D\rightarrow A} \equiv \exp[ -i\hat{q}_{D}\hat{p}_{A} ]
\end{align}
transforms the quadrature operators of the data and the ancilla modes as follows: 
\begin{align}
\hat{q}_{D} &\rightarrow \hat{q}'_{D}  = \hat{q}_{D}, 
\nonumber\\
\hat{p}_{D} &\rightarrow \hat{p}'_{D}  = \hat{p}_{D} - \hat{p}_{A},
\nonumber\\
\hat{q}_{A} &\rightarrow \hat{q}'_{A}  = \hat{q}_{A} + \hat{q}_{D},
\nonumber\\ 
\hat{p}_{A} &\rightarrow \hat{p}'_{A}  = \hat{p}_{A}, 
\end{align}
where $\hat{X} ' \equiv \textrm{SUM}_{D\rightarrow A}^{\dagger} \cdot    \hat{X}  \cdot \textrm{SUM}_{D\rightarrow A}$. Most importantly, the position operator of the data mode $\hat{q}_{D}$ is transferred via the SUM gate to the ancilla mode, i.e., $\hat{q}'_{A} = \hat{q}_{A} + \hat{q}_{D}$. The transformed ancilla position operator $\hat{q}'_{A}$ is then measured by the homodyne measurement $\mathcal{M}_{q}$. Since $\hat{q}_{A} =0$ mod $\sqrt{\pi}$ (see Eq.\ \eqref{eq:position and momentum quadrature of the plus GKP state}), measuring $\hat{q}'_{A} = \hat{q}_{A} + \hat{q}_{D}$ is equivalent to measuring $\hat{q}_{D}$ modulo $\sqrt{\pi}$. Thus, the circuit in Fig.\ \ref{fig:GKP stabilizer measurement position} implements the measurement of the stabilizer of the data mode $\hat{S}_{q}^{(D)} = e^{i2\sqrt{\pi}\hat{q}_{D}}$.   

It is also important to note that in the $\hat{S}_{q}^{(D)}$ stabilizer measurement circuit, the momentum operator of the ancilla mode $\hat{p}_{A}$ is transferred via the SUM gate to the data mode, i.e., $\hat{p}'_{D}  = \hat{p}_{D} - \hat{p}_{A}$. That is, the data mode is displaced in the momentum direction by $\hat{p}_{A}$. However, since $\hat{p}_{A} =0$ mod $2\sqrt{\pi}$ (see Eq.\ \eqref{eq:position and momentum quadrature of the plus GKP state}), the size of the momentum shifts in the data mode is an integer multiple of $2\sqrt{\pi}$. Since the logical states of the square-lattice GKP code are stabilized by $\hat{S}_{q} = e^{i2\sqrt{\pi}\hat{q}}$, they are invariant under the momentum shifts of the size an integer multiple of $2\sqrt{\pi}$. Thus, the propagation of the ancilla momentum operator to the data mode does not really impact the data mode. Hence, the $\hat{S}_{q}^{(D)}$ stabilizer measurement circuit in Fig.\ \ref{fig:GKP stabilizer measurement position} is non-destructive as desired.

\begin{figure}[t!]
\centering
\includegraphics[width=3.5in]{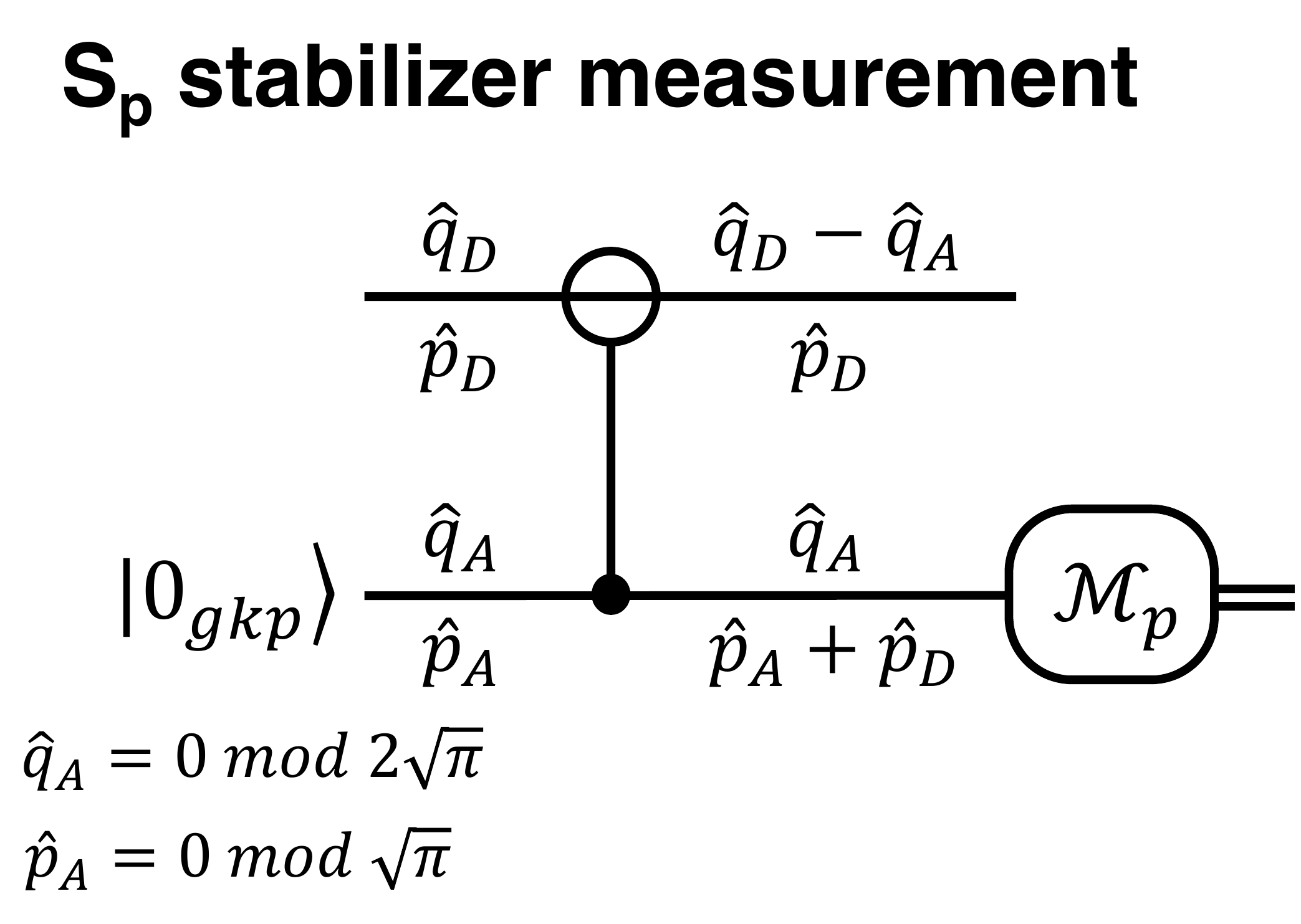}
\caption{Measurement of the $\hat{S}_{p}= e^{-i2\sqrt{\pi}\hat{p}}$ stabilizer of the square-lattice GKP code. The upper and the lower lines represent a data and an ancilla mode, respectively. The controlled-$\ominus$ symbol represents the inverse-SUM gate $\textrm{SUM}_{A\rightarrow D}^{\dagger} = \exp[ i\hat{q}_{A}\hat{p}_{D} ]$ and $\mathcal{M}_{p}$ represents the homodyne measurement of the momentum operator.   }
\label{fig:GKP stabilizer measurement momentum}
\end{figure}

Let us now move on the the $\hat{S}_{p}$ stabilizer measurement, i.e., the measurement of the momentum operator of the data mode $\hat{p}_{D}$ modulo $\sqrt{\pi}$. The circuit for the $\hat{S}_{p}$ stabilizer measurement is given in Fig.\ \ref{fig:GKP stabilizer measurement momentum}. Note that the GKP state in the ancilla mode $|0_{\textrm{gkp}}^{(\textrm{sq})}\rangle$ is given by 
\begin{align}
|0_{\textrm{gkp}}^{(\textrm{sq})}\rangle &= \sum_{n\in \mathbb{Z}} |\hat{q}_{A} = 2\sqrt{\pi} n\rangle = \frac{1}{\sqrt{2}} ( |+_{\textrm{gkp}}^{(\textrm{sq})}\rangle + |-_{\textrm{gkp}}^{(\textrm{sq})}\rangle ) = \frac{1}{\sqrt{2}} \sum_{n\in\mathbb{Z}} |\hat{p}_{A} = n\sqrt{\pi}\rangle. 
\end{align}  
Thus, the quadrature operators of the ancilla mode satisfy 
\begin{align}
\hat{q}_{A} &=0 \,\,\, \textrm{mod} \,\,\, 2\sqrt{\pi}, 
\nonumber\\
\hat{p}_{A} &=0 \,\,\, \textrm{mod} \,\,\, \sqrt{\pi} . \label{eq:position and momentum quadrature of the zero GKP state}
\end{align} 
Similarly as above, one can also see this by observing that the state $|+_{\textrm{gkp}}^{(\textrm{sq})}\rangle$ is stabilized by the stabilizer $\hat{S}_{p}^{(A)} = e^{-i2\sqrt{\pi}\hat{p}_{A}}$ and the Pauli Z operator $\hat{Z}_{\textrm{gkp}}^{(A)} = e^{i\sqrt{\pi}\hat{q}_{A}}$. Then in the Heisenberg picture, the inverse-SUM gate in Fig.\ \ref{fig:GKP stabilizer measurement momentum}
\begin{align}
\textrm{SUM}_{A\rightarrow D}^{\dagger} \equiv \exp[ i\hat{q}_{A}\hat{p}_{D} ]
\end{align}
transforms the quadrature operators of the data and the ancilla modes as follows: 
\begin{align}
\hat{q}_{D} &\rightarrow \hat{q}'_{D}  = \hat{q}_{D} - \hat{q}_{A}, 
\nonumber\\
\hat{p}_{D} &\rightarrow \hat{p}'_{D}  = \hat{p}_{D}, 
\nonumber\\
\hat{q}_{A} &\rightarrow \hat{q}'_{A}  = \hat{q}_{A},
\nonumber\\ 
\hat{p}_{A} &\rightarrow \hat{p}'_{A}  = \hat{p}_{A} + \hat{p}_{D}, 
\end{align}
where $\hat{X} ' \equiv \textrm{SUM}_{A\rightarrow D} \cdot    \hat{X}  \cdot \textrm{SUM}_{A\rightarrow D}^{\dagger}$. Most importantly, the momentum operator of the data mode $\hat{p}_{D}$ is transferred via the inverse-SUM gate to the ancilla mode, i.e., $\hat{p}'_{A} = \hat{p}_{A} + \hat{p}_{D}$. The transformed ancilla position operator $\hat{p}'_{A}$ is then measured by the homodyne measurement $\mathcal{M}_{p}$. Since $\hat{p}_{A} =0$ mod $\sqrt{\pi}$ (see Eq.\ \eqref{eq:position and momentum quadrature of the zero GKP state}), measuring $\hat{p}'_{A} = \hat{p}_{A} + \hat{p}_{D}$ is equivalent to measuring $\hat{p}_{D}$ modulo $\sqrt{\pi}$. Thus, the circuit in Fig.\ \ref{fig:GKP stabilizer measurement momentum} implements the measurement of the stabilizer of the data mode $\hat{S}_{p}^{(D)} = e^{-i2\sqrt{\pi}\hat{p}_{D}}$.   

Similarly to the case of the $\hat{S}_{q}^{(D)}$ stabilizer measurement, the $\hat{S}_{p}^{(D)}$ stabilizer measurement circuit in Fig.\ \ref{fig:GKP stabilizer measurement momentum} is non-destructive. Note that the position operator of the ancilla mode $\hat{q}_{A}$ is transferred via the inverse-SUM gate to the data mode, i.e., $\hat{q}'_{D}  = \hat{q}_{D} - \hat{q}_{A}$. That is, the data mode is displaced in the position direction by $\hat{q}_{A}$. However, since $\hat{q}_{A} =0$ mod $2\sqrt{\pi}$ (see Eq.\ \eqref{eq:position and momentum quadrature of the zero GKP state}), the size of the position shifts in the data mode is an integer multiple of $2\sqrt{\pi}$. Since the logical states of the square-lattice GKP code are stabilized by $\hat{S}_{p} = e^{-i2\sqrt{\pi}\hat{p}}$, they are invariant under the position shifts of the size an integer multiple of $2\sqrt{\pi}$. Thus, the propagation of the ancilla position operator to the data mode does not impact the data mode.

\subsubsection{Approximate GKP states with a finite squeezing}

It is clear by now that the ability to prepare a GKP state $|0_{\textrm{gkp}}^{(\textrm{sq})}\rangle$ or $|+_{\textrm{gkp}}^{(\textrm{sq})}\rangle$ is a useful non-Gaussian resource for implementing the GKP error correction. On the other hand, it is also very important to realize that the logical states of the square-lattice GKP code are not realistic because they are superpositions of infinitely many (i.e., $\sum_{n\in\mathbb{Z}}$) infinitely-squeezed states (i.e., $|\hat{q} = (2n+\mu)\sqrt{\pi}\rangle$ or $|\hat{p} = (2n+\mu)\sqrt{\pi}\rangle$). Nevertheless, one can define an approximate GKP state by replacing the infinitely-squeezed states with finitely-squeezed states and then introducing an overall Gaussian envelope function. Since the finitely-squeezed approximate GKP states have a bounded energy, they can be realized experimentally. Indeed, there have been many proposals for preparing an approximate GKP state in various experimental platforms \cite{Gottesman2001,Travaglione2002,Pirandola2004,Pirandola2006,
Vasconcelos2010,Terhal2016,Motes2017,Weigand2018,Arrazola2019,
Su2019,Eaton2019,Shi2019,Weigand2019,Hastrup2019}. Notably, the proposal in Ref.\ \cite{Travaglione2002} has recently been realized in a trapped ion system \cite{Fluhmann2018,Fluhmann2019,Fluhmann2019b} and a variation of the scheme in Ref.\ \cite{Terhal2016} has recently been realized in a circuit QED system \cite{Campagne2019}. Among these proposals, we will review the phase estimation method \cite{Terhal2016} below. Here, we instead focus more on the mathematical descriptions of an approximate GKP state.  

A comprehensive review of various representations of an approximate GKP state and their relations is given in Ref.\ \cite{Matsuura2019}. Here, we only review the representations that will be referenced later in the thesis. One simple way to represent an approximate GKP state is to apply a non-unitary envelope operator $\exp[-\Delta^{2} \hat{n}]$ to the ideal GKP state $|\psi_{\textrm{gkp}}^{(\textrm{sq})}\rangle$, i.e., 
\begin{align}
|\psi_{\textrm{gkp},\Delta}^{(\textrm{sq})}\rangle = \exp[-\Delta^{2} \hat{n}] |\psi_{\textrm{gkp}}^{(\textrm{sq})}\rangle, 
\end{align}
where $\Delta$ characterizes the width of each peak in the Wigner function of an approximate GKP state. While an ideal GKP state $ |\psi_{\textrm{gkp}}^{(\textrm{sq})}\rangle,$ is not normalizable, the approximate GKP state $|\psi_{\textrm{gkp},\Delta}^{(\textrm{sq})}\rangle$ is normalizable thanks to the Gaussian envelope operator $ \exp[-\Delta^{2} \hat{n}] $. Using the fact that any bounded operator on the bosonic Hilbert space can be expanded in terms of displacement operators, one can expand the envelope operator $\exp[-\Delta^{2} \hat{n}]$ as follows: 
\begin{align}
\exp[-\Delta^{2} \hat{n}] &= \int_{\alpha \in \mathbb{C} } \frac{d^{2}\alpha}{\pi} \mathrm{Tr}\big{[} \exp[-\Delta^{2}\hat{n}] \hat{D}^{\dagger}(\alpha)\big{]} \hat{D}(\alpha)  
\nonumber\\
&\propto \int_{\alpha \in \mathbb{C} } \frac{d^{2}\alpha}{\pi} \exp\Big{[} -\frac{|\alpha|^{2}}{2\sigma_{\textrm{gkp}}^{2}} \Big{]}  \hat{D}(\alpha)  , 
\end{align}
where $\sigma_{\textrm{gkp}}^{2} = (1-e^{-\Delta^{2}}) / (1+e^{-\Delta^{2}}) \xrightarrow{\Delta\ll 1} \Delta^{2} /2$. To derive the second line, we used
\begin{align}
\mathrm{Tr}\big{[} \exp[-\Delta^{2}\hat{n}] \hat{D}^{\dagger}(\alpha)\big{]} &= \sum_{n=0}^{\infty}e^{-\Delta^{2} n} \langle n| \hat{D}^{\dagger}(\alpha)|n\rangle 
\nonumber\\
&= \exp\Big{[} -\frac{|\alpha|^{2}}{2} \Big{]} \sum_{n=0}^{\infty} e^{-\Delta^{2} n} L_{n}(|\alpha|^{2}) 
\nonumber\\
&= \exp\Big{[} -\frac{|\alpha|^{2}}{2} \Big{]} \frac{1}{1-e^{-\Delta^{2}}} \exp\Big{[}-\frac{ e^{-\Delta^{2}} }{ 1-e^{-\Delta^{2}} } |\alpha|^{2} \Big{]} 
\nonumber\\
&= \frac{1}{1-e^{-\Delta^{2}}} \exp\Big{[} -\frac{ 1+e^{-\Delta^{2}} }{2( 1-e^{-\Delta^{2}}) } |\alpha|^{2} \Big{]}. \label{eq:approximate GKP laguerre proof}
\end{align}  
Thus, we can understand an approximate GKP state as the state that results from applying coherent superpositions of displacement operations with a Gaussian envelope to an ideal GKP state, i.e., 
\begin{align}
|\psi_{\textrm{gkp},\Delta}^{(\textrm{sq})}\rangle &\propto \int_{\alpha\in\mathbb{C}} d^{2}\alpha \exp\Big{[} -\frac{|\alpha|^{2}}{ 2\sigma_{\textrm{gkp}}^{2} }  \Big{]} \hat{D}(\alpha) |\psi_{\textrm{gkp}}^{(\textrm{sq})} \rangle . \label{eq:approximate GKP state coherent displacement}
\end{align}
      
For the purpose of efficiently simulating error correction schemes involving approximate GKP states, it is often useful to make the approximate GKP state $|\psi_{\textrm{gkp},\Delta}^{(\textrm{sq})}\rangle$ a bit more noisy by using a noise twirling technique. The purpose of the twirling is to transform the coherent superposition of displacement errors in Eq.\ \eqref{eq:approximate GKP state coherent displacement} into an incoherent mixture of displacement errors. More concretely, by applying random shifts of the size integer multiples of $2\sqrt{\pi}$ in both the position and the momentum directions, we can convert the approximate GKP state $|\psi_{\textrm{gkp},\Delta}^{(\textrm{sq})}\rangle$ into 
\begin{align}
\hat{\psi}_{\textrm{gkp},\Delta}^{(\textrm{sq})} &\propto \sum_{n_{1},n_{2} \in \mathbb{Z} } (\hat{S}_{q})^{n_{1}}(\hat{S}_{p})^{n_{2}} |\psi_{\textrm{gkp},\Delta}^{(\textrm{sq})}\rangle \langle \psi_{\textrm{gkp},\Delta}^{(\textrm{sq})}| (\hat{S}_{p}^{\dagger})^{n_{2}}(\hat{S}_{q}^{\dagger})^{n_{1}} 
\nonumber\\
&= \sum_{n_{1},n_{2}\in \mathbb{Z}} \int_{\alpha,\beta\in\mathbb{C}} d^{2}\alpha d^{2}\beta \exp\Big{[} -\frac{ |\alpha|^{2} + |\beta|^{2} }{ 2\sigma_{\textrm{gkp}}^{2} } \Big{]}  
\nonumber\\
&\qquad\qquad\qquad\qquad\qquad \times (\hat{S}_{q})^{n_{1}} (\hat{S}_{p})^{n_{2}} \hat{D}(\alpha) |\psi_{\textrm{gkp}}\rangle\langle \psi_{\textrm{gkp}}| \hat{D}^{\dagger}(\beta)  (\hat{S}_{p}^{\dagger})^{n_{2}} (\hat{S}_{q}^{\dagger})^{n_{1}} 
\nonumber\\
&\propto \sum_{k_{1},k_{2}\in \mathbb{Z}}   \exp\Big{[} -\frac{ \pi |k_{1}+ ik_{2}|^{2}      }{ 2\sigma_{\textrm{gkp}}^{2} } \Big{]} 
\nonumber\\
&\quad\times \int_{\alpha\in\mathbb{C}} d^{2}\alpha \exp\Big{[} -\frac{ |\alpha - \sqrt{\frac{\pi}{2}} (k_{1} + ik_{2}) |^{2}  }{\sigma_{\textrm{gkp}}^{2} } \Big{]}  \hat{D}(\alpha) |\psi_{\textrm{gkp}}\rangle\langle \psi_{\textrm{gkp}}|   \hat{D}^{\dagger}(\alpha -\sqrt{2\pi} (k_{1}+ik_{2}) ). 
\end{align}  
See Appendix A of Ref.\ \cite{Noh2020} for the derivation of the last proportionality. In the small noise limit (i.e., $\sigma_{\textrm{gkp}} \ll \sqrt{\pi}$), we can neglect all the $(k_{1},k_{2}) \neq (0,0)$ terms due to the exponentially decaying prefactor $ \exp[ -\frac{ \pi |k_{1}+ ik_{2}|^{2}      }{ 2\sigma_{\textrm{gkp}}^{2} } ] $ and get the noise model
\begin{align}
\hat{\psi}_{\textrm{gkp},\Delta}^{(\textrm{sq})}  &\propto   \int_{\alpha\in\mathbb{C}} \frac{d^{2}\alpha }{ \pi\sigma_{\textrm{gkp}}^{2} } \exp\Big{[} -\frac{ |\alpha  |^{2}  }{\sigma_{\textrm{gkp}}^{2} } \Big{]}  \hat{D}(\alpha) |\psi_{\textrm{gkp}}^{(\textrm{sq})}\rangle\langle \psi_{\textrm{gkp}}^{(\textrm{sq})}|   \hat{D}^{\dagger}(\alpha  )   . \label{eq:approximate GKP state incoherent displacement}
\end{align}
Thus, the coherent displacement error model in Eq.\ \eqref{eq:approximate GKP state coherent displacement} is transformed into an incoherent displacement error model. In particular, the incoherent noise model in Eq.\ \eqref{eq:approximate GKP state incoherent displacement} is equivalent to the Gaussian random shift error $\mathcal{N}_{B_{2}}[\sigma_{\textrm{gkp}}]$ with a noise standard deviation $\sigma_{\textrm{gkp}}$. The incoherent noise model is easy to work with numerically because we can simply sample the shift errors from a classical Gaussian random distribution $\mathcal{N}(0,\sigma_{\textrm{gkp}}^{2})$. Indeed, all the previous works on the large-scale simulation of the GKP code \cite{Menicucci2014,Wang2017,Fukui2018a,Vuillot2019,Fukui2019,Noh2020} were performed by assuming the incoherent noise model in Eq.\ \eqref{eq:approximate GKP state incoherent displacement}.

The quality of an approximate GKP state is typically measured by the GKP squeezing
\begin{align}
s_{\textrm{gkp}} \equiv -10\log_{10}(2\sigma_{\textrm{gkp}}^{2}). 
\end{align} 
Note that the GKP squeezing $s_{\textrm{gkp}}$ quantifies how much an approximate GKP state is squeezed in both the position and the momentum quadrature in comparison to the vacuum noise variance $1/2$. We also remark that the squeezing of the experimentally realized GKP states ranges from $5.5$dB to $9.5$dB \cite{Fluhmann2019,Campagne2019}. 

The natural question is then whether the finitely-squeezed GKP states can be used, for example, to realize fault-tolerant quantum error correction and computation. Over the past few years, it has been shown that the answer is affirmative in the case of fault-tolerant quantum error correction \cite{Menicucci2014,Wang2017,Fukui2018a,Vuillot2019,Fukui2019,Noh2020}. Detailed issues related to the use of finitely-squeezed GKP states will be addressed in Chapter \ref{chapter:Fault-tolerant bosonic quantum error correction}.      

\subsubsection{Preparation of a GKP state via phase estimation}

Phase estimation algorithms are used to measure an arbitrary unitary operation. One way to prepare a GKP state is to use a phase estimation algorithm to measure the stabilizers of the GKP code, which are unitary operations \cite{Terhal2016}. While there are many variants of phase estimation algorithms, we only introduce an adaptive phase estimation algorithm developed by Kitaev \cite{Kitaev1996} due to its simplicity. 

The goal of the phase estimation is to measure the phase $\theta$ of a unitary operator $\hat{U}$, i.e., 
\begin{align}
\hat{U}|\psi_{\theta}\rangle= e^{i\theta}|\psi_{\theta}\rangle, 
\end{align}
where $|\psi_{\theta}\rangle$ is an eigenstate of the unitary operator $\hat{U}$. More specifically, given an input state $|\psi\rangle$, an ideal implementation of the phase estimation of the unitary operator $\hat{U}$ should yield 
\begin{itemize}
\item an eigenvalue $e^{i\theta}$ (or phase $\theta$) of the unitary operator $\hat{U}$ as a measurement outcome
\item with probability $P_{\theta}=\langle \psi |\hat{P}_{\theta}|\psi\rangle$, where $\hat{P}_{\theta}$ is the projection operator to the eigenspace $\mathcal{H}_{\theta}  \equiv \lbrace |\chi\rangle : \hat{U}|\chi\rangle = e^{i\theta}|\chi\rangle\rbrace$
\item and the state $|\psi\rangle$ should collapse to $\hat{P}_{\theta}|\psi\rangle / \sqrt{\langle \psi| \hat{P}_{\theta} |\psi\rangle}$ after the measurement. 
\end{itemize}
The simplest case is when $\hat{U}^{2} = \hat{I}$ where $\theta$ can only take the values $0$ and $\pi$. In this case, $1$-bit precision is enough and the measurement circuit shown in Fig.\ \ref{fig:one bit phase estimation circuit} (or in Fig.\ \ref{fig:n bit phase estimation circuit}(a)) implements the desired phase estimation in $1$-bit precision. In a more general case where a unitary operator $\hat{U}$ satisfies $\hat{U}^{2^{n}}=\hat{I}$ for some natural number $n$, the unitary operator $\hat{U}$ can take $2^{n}$ phase values, i.e., 
\begin{align}
\frac{2\pi k }{2^{n}}  \,\,\, \textrm{for}\,\,\, k\in \lbrace 0,1,\cdots, 2^{n}-1 \rbrace. 
\end{align} 
Thus in this case, we need an $n$-bit precision phase estimation circuit shown in Fig.\ \ref{fig:n bit phase estimation circuit} (for $n=1,2,3$). 

\begin{figure}[t!]
\centering
\includegraphics[width=5.4in]{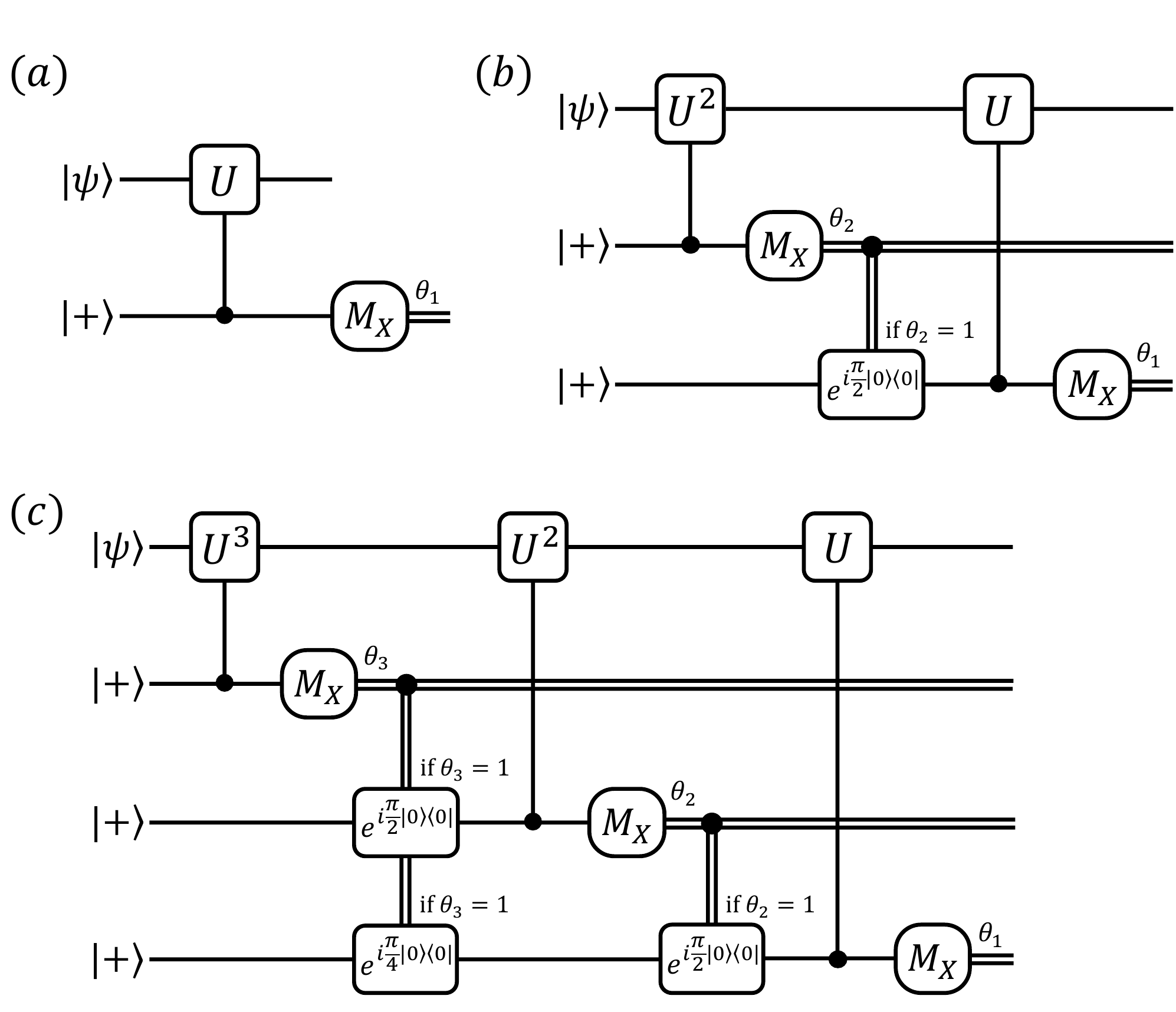}
\caption{An adaptive phase estimation circuit for measuring a unitary operator $\hat{U}$ in $n$-bit precision for (a) $n=1$, (b) $n=2$, and (c) $n=3$. The $1$-bit precision measurement circuit in (a) is the same as the circuit in Fig.\ \ref{fig:one bit phase estimation circuit}. $\exp[i\phi |0\rangle\langle 0|] \equiv e^{i\phi}|0\rangle\langle 0| + |1\rangle\langle 1|$ imparts a phase $e^{i\phi}$ to the computational zero state of the qubit. $\mathcal{M}_{X}$ represents the Pauli X measurement. 	The estimated phase $\theta$ of the unitary operator (i.e., $\hat{U} = e^{i\theta}$) is given by $\theta = 2\pi \times 0.\theta_{1}\cdots\theta_{n}$ in binary representation. That is, $\theta = 2\pi \sum_{k=1}^{n} \theta_{k}  2^{-k}$.   }
\label{fig:n bit phase estimation circuit}
\end{figure}

Let us now get back to the task of preparing a computational zero GKP state, i.e., $|0_{\textrm{gkp}}^{(\textrm{sq})}\rangle$. Note that the target state $|0_{\textrm{gkp}}^{(\textrm{sq})}\rangle$ is the unique state (up to an overall phase and normalization) that is stabilized by the following two commuting displacement operators.   
\begin{align}
\hat{S}_{p} = e^{-i2\sqrt{\pi} \hat{p} } , \quad \hat{Z}_{\textrm{gkp}} = e^{i\sqrt{\pi} \hat{q} } . \label{eq:stabilizers for the GKP computational zero phase estimation}
\end{align}   
Since the quadrature operators $\hat{q}$ and $\hat{p}$ can take any real number, the phases of these two unitary operators can take any real value between $-\pi$ and $\pi$. Thus, to measure these two unitary operators precisely, we need to perform an $n$-bit phase estimation with an infinitely large $n$. In practice, however, we can use an $n$-bit precision phase estimation circuit with some finite $n$ to measure the two unitary operators in Eq.\ \eqref{eq:stabilizers for the GKP computational zero phase estimation} approximately. Then, correcting for the shift errors based on the measurement outcomes of these two unitary operators, we can prepare an approximate computational zero state of the GKP code \cite{Terhal2016}.       

Note that the phase estimation circuit in Fig.\ \ref{fig:n bit phase estimation circuit} is sensitive to any measurement errors in the ancilla qubits for all $n\ge 2$. This is because there are qubit rotations that are applied conditioned on the ancilla qubit measurement outcomes in the previous rounds. Thus, non-adaptive phase estimation schemes can be more robust against such ancilla qubit errors than the adaptive scheme shown in Fig.\ \ref{fig:n bit phase estimation circuit}. We remark that Ref.\ \cite{Shi2019} provided a fault-tolerant scheme for preparing an approximate GKP state by making a non-adaptive phase estimation circuit fault-tolerant using flag qubits.

\subsection{Universal set of logical gates on the square-lattice GKP code}

Recall that the GKP code has a desirable property that the preparation of a GKP state $|0_{\textrm{gkp}}^{(\textrm{sq})}\rangle$ or $|+_{\textrm{gkp}}^{(\textrm{sq})}\rangle$ is the only non-Gaussian resource needed to implement the error correction protocol. Surprisingly, the GKP code has an even more desirable property, that is, the ability to prepare a GKP state $|0_{\textrm{gkp}}^{(\textrm{sq})}\rangle$ or $|+_{\textrm{gkp}}^{(\textrm{sq})}\rangle$ enables universal logical operations on the GKP code, when combined with Gaussian operations. Here, we will review that any Clifford operations on the GKP code can be implemented by using a Gaussian operation \cite{Gottesman2001}. Furthermore, we will also review how a GKP state $|0_{\textrm{gkp}}^{(\textrm{sq})}\rangle$ or $|+_{\textrm{gkp}}^{(\textrm{sq})}\rangle$ can be consumed to prepare a magic state encoded in the GKP code \cite{Terhal2016,Baragiola2019}. Since Clifford operations and magic states allow universal quantum computation \cite{Bravyi2005}, the ability to prepare a GKP state and perform Gaussian operations is sufficient for realizing any logical operation on the GKP code.   

\subsubsection{Logical Clifford operations on the square-lattice GKP qubit}

The set of $n$-qubit Clifford operations is defined as the set of all $n$-qubit operations that maps an $n$-qubit Pauli operator into another $n$-qubit Pauli operator under conjugation \cite{Gottesman1999}, i.e.,
\begin{align}
\mathcal{C}^{(n)} \equiv \lbrace \hat{U} | \hat{U}^{\dagger} \hat{P} \hat{U} \in \mathcal{P}^{(n)} \textrm{ for all } \hat{P}\in\mathcal{P}^{(n)} \rbrace . 
\end{align}   
Here, $\mathcal{P}^{(n)}$ is the $n$-qubit Pauli group generated by $\lbrace \pm 1, \pm i, \hat{X}_{j},\hat{Z}_{j},\hat{Y}_{j} | j\in \lbrace 1,\cdots, n \rbrace \rbrace $, where $\hat{X}_{j},\hat{Z}_{j},\hat{Y}_{j}$ are the Pauli X, Z, Y operators acting on the $j^{\textrm{th}}$ qubit. Pauli operators are explicitly given by 
\begin{align}
\hat{X} = \begin{bmatrix}
0 & 1\\
1 & 0
\end{bmatrix}, \quad \hat{Z} = \begin{bmatrix}
1 & 0\\
0 & -1
\end{bmatrix}, \quad \hat{Y} = i\hat{X}\hat{Z} = \begin{bmatrix}
0 & -i\\
i & 0
\end{bmatrix} 
\end{align} 
in the computational basis. 

Precisely due to the property that Clifford operations map a Pauli operator to another Pauli operator, their actions on quantum circuits can be efficiently simulated by a classical computer \cite{Gottesman1998,Aaronson2004}. Nevertheless, Clifford operations are extremely useful for conventional multi-qubit quantum error correction. In fact, most, if not all, of the leading multi-qubit error-correcting codes are a stabilizer code \cite{Gottesman1997} which can be implemented by using only computational zero states $|0\rangle$, Clifford operations, and Pauli measurements.  

It is known that the set of all Clifford operations is generated by the phase gate, the Hadamard gate, and the CNOT gate, i.e., 
\begin{align}
\hat{S} &: \quad \hat{S} |0\rangle = |0\rangle, \quad  \hat{S} |1\rangle = i|1\rangle, 
\nonumber\\
\hat{H} &: \quad \hat{H} |0\rangle = |+\rangle, \quad  \hat{H} |1\rangle = |-\rangle, 
\nonumber\\
\textrm{CNOT}_{j\rightarrow k} &: \quad \textrm{CNOT}_{j\rightarrow k} |\mu\rangle_{j} |\nu\rangle_{k} =  |\mu\rangle_{j} |\nu\oplus \mu\rangle_{k} , 
\end{align}
for all $\mu,\nu\in \lbrace 0,1 \rbrace $. Here, $|0\rangle$ and $|1\rangle$ are the computational basis states, $|\pm\rangle = \frac{1}{\sqrt{2}} (|0\rangle \pm |1\rangle)$ are the complementary basis states, and $\oplus$ is the addition modulo $2$. Gaussian operations are analogous to Clifford operations in the sense that Gaussian operations map a displacement operator to another displacement operator under conjugation (see Appendix \ref{appendix:Gaussian states, unitaries, and channels}). For the GKP code, any logical Clifford operations can be implemented by using Gaussian operations. More specifically, the phase gate, the Hadamard gate, and the CNOT gate on the square-lattice GKP code can be realized by using the following Gaussian operations: 
\begin{align}
\hat{S}_{\textrm{gkp}} &= e^{i\frac{\hat{q}^{2}}{2}} , 
\nonumber\\
\hat{H}_{\textrm{gkp}} &= e^{i\frac{\pi}{2}\hat{a}^{\dagger}\hat{a} } , 
\nonumber\\
\textrm{CNOT}_{\textrm{gkp}}^{j\rightarrow k} &= \textrm{SUM}_{j\rightarrow k} = e^{ - i \hat{q}_{j}\hat{p}_{k} }. \label{eq:square lattice GKP code Clifford gates}
\end{align}
Indeed, we can explicitly check that 
\begin{align}
\hat{S}_{\textrm{gkp}} |0_{\textrm{gkp}}^{(\textrm{sq})} \rangle &= e^{i\frac{\hat{q}^{2}}{2}} \sum_{ n\in\mathbb{Z} } | \hat{q} = 2n\sqrt{\pi}  \rangle = \sum_{ n\in\mathbb{Z} } e^{i2 n^{2}\pi }   | \hat{q} = 2n\sqrt{\pi}  \rangle  =|0_{\textrm{gkp}}^{(\textrm{sq})} \rangle, 
\nonumber\\ 
\hat{S}_{\textrm{gkp}} |1_{\textrm{gkp}}^{(\textrm{sq})} \rangle &= e^{i\frac{\hat{q}^{2}}{2}} \sum_{ n\in\mathbb{Z} } | \hat{q} = (2n+1)\sqrt{\pi}  \rangle  = \sum_{ n\in\mathbb{Z} } e^{i \frac{1}{2} (2n+1)^{2}\pi }   | \hat{q} = (2n+1)\sqrt{\pi}  \rangle  = i |1_{\textrm{gkp}}^{(\textrm{sq})} \rangle, \label{eq:square lattice GKP logical phase check}
\end{align}
and
\begin{align}
\hat{H}_{\textrm{gkp}} |0_{\textrm{gkp}}^{(\textrm{sq})} \rangle &= e^{i\frac{\pi}{2} \hat{a}^{\dagger}\hat{a} } \sum_{ n\in\mathbb{Z} } | \hat{q} = 2n\sqrt{\pi}  \rangle = \sum_{ n\in\mathbb{Z} } | \hat{p} = 2n\sqrt{\pi}  \rangle  =|+_{\textrm{gkp}}^{(\textrm{sq})} \rangle, 
\nonumber\\
\hat{H}_{\textrm{gkp}} |1_{\textrm{gkp}}^{(\textrm{sq})} \rangle &= e^{i\frac{\pi}{2} \hat{a}^{\dagger}\hat{a} } \sum_{ n\in\mathbb{Z} } | \hat{q} = (2n+1)\sqrt{\pi}  \rangle = \sum_{ n\in\mathbb{Z} } | \hat{p} = (2n+1)\sqrt{\pi}  \rangle  =|-_{\textrm{gkp}}^{(\textrm{sq})} \rangle, \label{eq:square lattice GKP logical Hadamard check}
\end{align}
and
\begin{align}
\textrm{CNOT}_{\textrm{gkp}}^{j\rightarrow k} |\mu_{\textrm{gkp}}^{(\textrm{sq})} \rangle|\nu_{\textrm{gkp}}^{(\textrm{sq})} \rangle &= \textrm{SUM}_{j\rightarrow k} |\mu_{\textrm{gkp}}^{(\textrm{sq})} \rangle|\nu_{\textrm{gkp}}^{(\textrm{sq})} \rangle  
\nonumber\\
&= e^{-i\hat{q}_{j}\hat{p}_{k}} \sum_{m,n\in\mathbb{Z}} |\hat{q}_{j} = (2m+\mu)\sqrt{\pi}\rangle|\hat{q}_{k} = (2n+\nu)\sqrt{\pi}\rangle
\nonumber\\
&= \sum_{m,n\in\mathbb{Z}} |\hat{q}_{j} = (2m+\mu)\sqrt{\pi}\rangle|\hat{q}_{k} = (2(m+n)+\nu+\mu)\sqrt{\pi}\rangle
\nonumber\\
&= \sum_{m,n\in\mathbb{Z}} |\hat{q}_{j} = (2m+\mu)\sqrt{\pi}\rangle|\hat{q}_{k} = (2n+\nu+\mu)\sqrt{\pi}\rangle  
\nonumber\\
&= |\mu_{\textrm{gkp}}^{(\textrm{sq})} \rangle|(\nu\oplus\mu)_{\textrm{gkp}}^{(\textrm{sq})} \rangle, \label{eq:square lattice GKP logical CNOT check}
\end{align}
as desired. Here, we used $e^{i\frac{\pi}{2} \hat{a}^{\dagger}\hat{a} } |\hat{q} = q\rangle = |\hat{p} = q\rangle$ and $e^{-i\hat{q}_{j}\hat{p}_{j}} |\hat{q}_{j} = q\rangle|\hat{q}_{k} = q'\rangle = |\hat{q}_{j} = q\rangle|\hat{q}_{k} = q'+q\rangle$ to derive Eqs.\ \eqref{eq:square lattice GKP logical Hadamard check} and \eqref{eq:square lattice GKP logical CNOT check}, respectively.

\subsubsection{H-type magic state for the GKP qubit}

As discussed above, Clifford operations can be efficiently simulated by using a classical computer. For a quantum circuit to be classically intractable, it should have a non-Clifford resource in it. While Clifford operations themselves are trivial in terms of computational power, their eigenstates may be non-trivial. Consider the following state:  
\begin{align}
|H\rangle &\equiv \cos\Big{(} \frac{\pi}{8} \Big{)} |0\rangle + \sin\Big{(} \frac{\pi}{8} \Big{)} |1\rangle . 
\end{align} 
The state $|H\rangle$ is called an H-type magic state \cite{Bravyi2005}, because it is a $+
1$ eigenstate of the Hadamard operator $\hat{H}$, i.e., 
\begin{align}
\hat{H}  |H\rangle &=  \cos\Big{(} \frac{\pi}{8} \Big{)} |+\rangle + \sin\Big{(} \frac{\pi}{8} \Big{)} |-\rangle 
\nonumber\\
&= \frac{1}{\sqrt{2}} \Big{[} \cos\Big{(} \frac{\pi}{8} \Big{)} + \sin\Big{(} \frac{\pi}{8} \Big{)} \Big{]} |0\rangle  + \frac{1}{\sqrt{2}} \Big{[} \cos\Big{(} \frac{\pi}{8} \Big{)} - \sin\Big{(} \frac{\pi}{8} \Big{)} \Big{]} |1\rangle 
\nonumber\\
&= \cos\Big{(} \frac{\pi}{4} - \frac{\pi}{8} \Big{)} |0\rangle + \sin\Big{(} \frac{\pi}{4} - \frac{\pi}{8} \Big{)} |1\rangle  =  \cos\Big{(} \frac{\pi}{8} \Big{)} |0\rangle + \sin\Big{(} \frac{\pi}{8} \Big{)} |1\rangle  = |H\rangle. 
\end{align}
Note that the H-type magic state $|H\rangle$ is not exactly the same as the more commonly known magic state $|0\rangle + e^{i\frac{\pi}{4}}|1\rangle$. However, these two states are equivalent in the sense that they can be mapped to each other via a Clifford operation.   

The H-type magic state $|H\rangle$ can be generated by applying a non-Clifford gate $\hat{T}$ to the computational basis state $|0\rangle$, i.e., 
\begin{align}
|H\rangle = \hat{T} |0\rangle, 
\end{align}
where the non-Clifford operation $\hat{T}$ is defined as 
\begin{align}
\hat{T} \equiv \hat{Y}\Big{(} \frac{\pi}{4} \Big{)}  = e^{  -i \frac{\pi}{8} \hat{Y} } &= \cos\Big{(} \frac{\pi}{8} \Big{)} \hat{I} -i\sin\Big{(} \frac{\pi}{8} \Big{)}  \hat{Y} =  \begin{bmatrix}
\cos \frac{\pi}{8} & -\sin \frac{\pi}{8} \\
\sin \frac{\pi}{8} & \cos \frac{\pi}{8} 
\end{bmatrix} . 
\end{align}
Note that the $\hat{T}$ gate maps the Pauli $X$ operator to the Hadamard operator $\hat{H}$ under conjugation, i.e., 
\begin{align}
\hat{T}^{\dagger}\hat{X}\hat{T} = \hat{H} = \frac{1}{\sqrt{2}} (\hat{X} + \hat{Z}).  
\end{align} 
Hence, the $\hat{T}$ gate is non-Clifford as it maps a Pauli operator to a non-Pauli operator. Note also that the term ``T gate'' is frequently used to refer to another non-Clifford gate $|0\rangle\langle 0| + e^{i\frac{\pi}{4}}|1\rangle\langle 1|$. Similarly as in the case of the magic state, the $\hat{T}$ gate we introduced above is equivalent to the more commonly known non-Clifford gate $|0\rangle\langle 0| + e^{i\frac{\pi}{4}}|1\rangle\langle 1|$ up to Clifford operations.        

\begin{figure}[t!]
\centering
\includegraphics[width=4.5in]{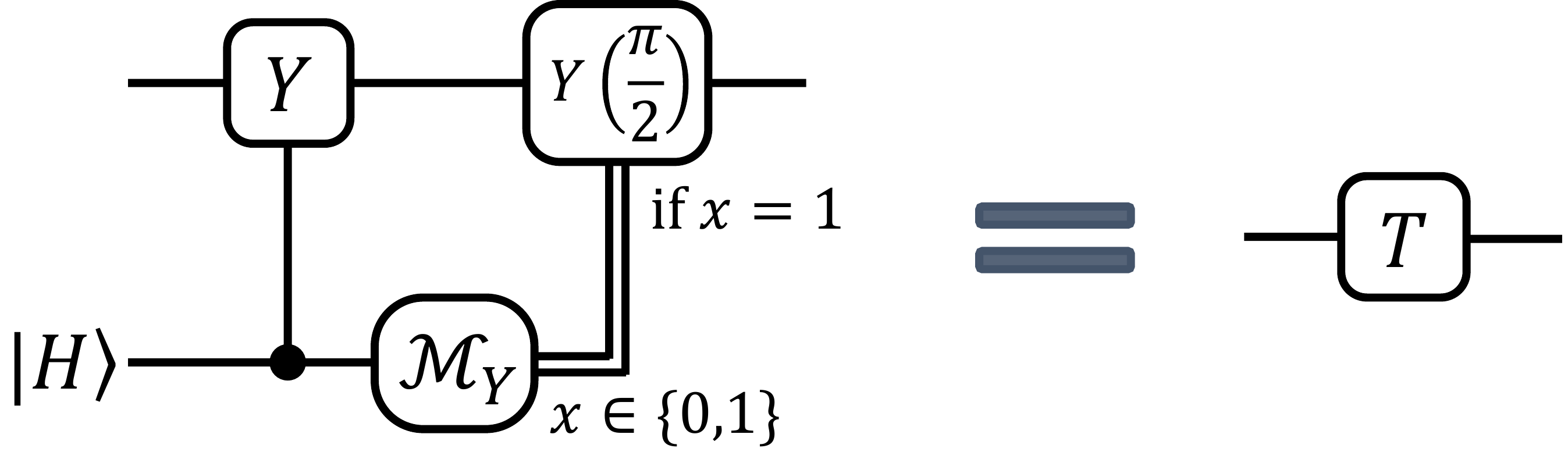}
\caption{A circuit for magic state injection. To implement the non-Clifford $\hat{T}$ gate on the data qubit, an ancilla H-type magic state $|H\rangle$ is supplied to the circuit. The controlled-Y symbol represents the controlled-Y gate and $\mathcal{M}_{Y}$ represents the Pauli Y measurement. $Y(\frac{\pi}{2})$ is defined as $Y(\frac{\pi}{2}) = e^{-i\frac{\pi}{4}\hat{Y}}$ and is explicitly given by $Y(\frac{\pi}{2}) = \frac{1}{\sqrt{2}} (\hat{I}-i\hat{Y})$. The $Y(\frac{\pi}{2})$ gate is applied only when the Pauli Y measurement outcome $x$ is $1$.   }
\label{fig:magic state injection}
\end{figure}

It might appear that the H-type magic state $|H\rangle$ is a weaker non-Clifford resource than the $\hat{T}$ gate since $|H\rangle$ can be generated by applying the $\hat{T}$ gate to the computational basis state $|0\rangle$. Remarkably, however, the converse is also true. That is, the H-type magic state $|H\rangle$ is a strong enough non-Clifford resource that it can be used to realize the non-Clifford $\hat{T}$ gate via a magic state injection protocol \cite{Bravyi2005}, when it is assisted by Clifford operations and Pauli measurements. 

The magic state injection circuit is given in Fig.\ \ref{fig:magic state injection}. To see how this works in more detail, let us consider an arbitrary input state $|\psi_{D}\rangle$ to the data qubit and an ancilla H-type magic state, i.e., 
\begin{align}
|\Psi_{0}\rangle &= |\psi_{D}\rangle |H_{A}\rangle = \cos\Big{(} \frac{\pi}{8} \Big{)} |\psi_{D}\rangle |0_{A}\rangle + \sin\Big{(} \frac{\pi}{8} \Big{)} |\psi_{D}\rangle |1_{A}\rangle . 
\end{align}     
After the controlled-Y gate, the state is transformed into 
\begin{align}
|\Psi_{1}\rangle &= \textrm{CY}_{A\rightarrow D} |\Psi_{0}\rangle = \cos\Big{(} \frac{\pi}{8} \Big{)} |\psi_{D}\rangle |0_{A}\rangle + \sin\Big{(} \frac{\pi}{8} \Big{)} \hat{Y}_{D} |\psi_{D}\rangle |1_{A}\rangle  . 
\end{align}
Note that the eigenstates of the Pauli Y operator are given by $|\pm Y\rangle = \frac{1}{\sqrt{2}} ( |0\rangle \pm i|1\rangle )$. Thus, we have $|0\rangle = \frac{1}{\sqrt{2}} ( |+Y\rangle + |-Y\rangle )$ and $|1\rangle = -\frac{i}{\sqrt{2}} ( |+Y\rangle - |-Y\rangle )$ and 
\begin{align}
|\Psi_{1}\rangle &= \frac{1}{\sqrt{2}}\Big{[} \cos\Big{(} \frac{\pi}{8} \Big{)}\hat{I}_{D} -i\sin\Big{(} \frac{\pi}{8} \Big{)} \hat{Y}_{D} \Big{]}  |\psi_{D}\rangle |+Y_{A}\rangle 
\nonumber\\  
&\quad + \frac{1}{\sqrt{2}}\Big{[} \cos\Big{(} \frac{\pi}{8} \Big{)}\hat{I}_{D} +i\sin\Big{(} \frac{\pi}{8} \Big{)} \hat{Y}_{D} \Big{]}  |\psi_{D}\rangle |-Y_{A}\rangle 
\nonumber\\
&= \frac{1}{\sqrt{2}} \hat{Y}\Big{(} \frac{\pi}{4} \Big{)}_{D} |\psi_{D}\rangle |+Y_{A}\rangle  + \frac{1}{\sqrt{2}} \hat{Y}\Big{(} -\frac{\pi}{4} \Big{)}_{D} |\psi_{D}\rangle |+Y_{A}\rangle . 
\end{align} 
Thus, conditioned on having $|\pm Y_{A}\rangle$ state in the ancilla qubit, the data qubit undergoes an evolution by the $\hat{Y}( \pm \frac{\pi}{4} )$ gate. Therefore, by measuring the Pauli Y operator of the ancilla qubit and conditioned on obtaining $\hat{Y}_{A} = (-1)^{x} $ (where $ x\in \lbrace 0,1 \rbrace$), we have 
\begin{align}
|\bar{\Psi}_{2}(x)\rangle &= \begin{cases}
 \langle +Y _{A} |\Psi_{1}\rangle  & x = 0  \\ 
 \langle -Y_{A} |\Psi_{1}\rangle  & x = 1
\end{cases} 
\nonumber\\
&= \begin{cases}
\frac{1}{\sqrt{2}} \hat{Y}( \frac{\pi}{4} )_{D}|\psi_{D}\rangle & x = 0  \\ 
\frac{1}{\sqrt{2}} \hat{Y}( -\frac{\pi}{4} )_{D}|\psi_{D}\rangle  & x = 1
\end{cases}  = \frac{1}{\sqrt{2}} \hat{Y}\Big{(} (-1)^{x} \frac{\pi}{4} \Big{)}_{D}|\psi_{D}\rangle . 
\end{align}  
The prefactor $\frac{1}{\sqrt{2}}$ indicates that both the $x=0$ and the $x=1$ outcomes happen with $50\%$ probability. Also, the prefactor $\frac{1}{\sqrt{2}}$ disappears once we normalize the state, i.e., $|\Psi_{2}(x)\rangle = \sqrt{2} |\bar{\Psi}_{2}(x)\rangle$. 

Conditioned on measuring $x=0$ or the $|+Y_{A}\rangle$ state in the ancilla qubit, the non-Clifford gate $\hat{T} = \hat{Y}(\frac{\pi}{4})$ is applied to the data qubit as desired. So in this case, we do not need to do anything further. On the other hand, if we measure $x=1$ or the $|-Y_{A}\rangle$ state in the ancilla qubit, another non-Clifford gate $\hat{T}^{\dagger} =\hat{Y}(-\frac{\pi}{4})$ is applied to the data qubit. In this case, we can further apply $\hat{Y}(\frac{\pi}{2})$ to the data qubit to get the following desired result: 
\begin{align}
|\Psi_{3}(x)\rangle = \begin{cases}
|\Psi_{3}(x)\rangle  & x=0 \\
\hat{Y}(\frac{\pi}{2})_{D}|\Psi_{3}(x)\rangle  & x=1
\end{cases}  = \hat{Y}\Big{(} \frac{\pi}{2} \Big{)}_{D} |\psi_{D}\rangle. 
\end{align}   
Note that $\hat{Y}(\frac{\pi}{2})$ is a Clifford operation. Therefore, the only non-Clifford resource needed in the above magic state injection scheme is the preparation of an H-type magic state $|H\rangle$. This then implies that the preparation of an H-type magic state is the only required non-Clifford resource to implement the universal quantum computation. This is due to a well-known result stating that the generators of the Clifford operations and any non-Clifford gate form a universal gate set (see, e.g., Ref.\ \cite{Nielsen2000} for more details).    

Getting back to the GKP code, we have so far realized that the GKP stabilizer measurements can be done by using Gaussian operations, homodyne measurements, and an ancilla GKP state $|0_{\textrm{gkp}}^{(\textrm{sq})}\rangle$ or $|+_{\textrm{gkp}}^{(\textrm{sq})}\rangle = \hat{H}_{\textrm{gkp}}|0_{\textrm{gkp}}^{(\textrm{sq})}\rangle$ (See Figs.\ \ref{fig:GKP stabilizer measurement position} and \ref{fig:GKP stabilizer measurement momentum}). Furthermore, any Clifford operation on the GKP code can be implemented by using Gaussian operations (Eq.\ \eqref{eq:square lattice GKP code Clifford gates}). The only remaining piece for implementing the universal gate set on the GKP code is, for instance, to prepare an H-type magic state encoded in the GKP code, i.e., 
\begin{align}
|H_{\textrm{gkp}}^{(\textrm{sq})}\rangle = \cos\Big{(} \frac{\pi}{8} \Big{)} |0_{\textrm{gkp}}^{(\textrm{sq})}\rangle + \sin\Big{(} \frac{\pi}{8} \Big{)} |1_{\textrm{gkp}}^{(\textrm{sq})}\rangle . 
\end{align}  

One way to prepare the H-type magic state is to measure the Hadamard operator in a non-destructive way. More specifically, assuming that the state is already in the code space (i.e., $|\psi\rangle \in \mathcal{C}_{\textrm{gkp}}^{(\textrm{sq})}$), if we measure the Hadamard operator non-destructively and get $\hat{H}_{\textrm{gkp}} = +1$, we are guaranteed to be left with the desired H-type magic state $|H_{\textrm{gkp}}^{(\textrm{sq})}\rangle$ as it is the only state (up to an overall phase) that is stabilized by the GKP stabilizers $\hat{S}_{q}$ and $\hat{S}_{p}$, and the logical Hadamard operator $\hat{H}_{\textrm{gkp}}$. On the other hand, if the measurement outcome is $\hat{H}_{\textrm{gkp}} = -1$, we get 
\begin{align}
|-H_{\textrm{gkp}}^{(\textrm{sq})}\rangle &= \hat{Y}_{\textrm{gkp}} |H_{\textrm{gkp}}^{(\textrm{sq})}\rangle, 
\end{align}
where $\hat{Y}_{\textrm{gkp}}$ is the Pauli Y operator on the GKP code. In this case, we can simply apply the Pauli Y operator (via a displacement operation) to the state $|-H_{\textrm{gkp}}^{(\textrm{sq})}\rangle$ and end up with the desired magic state $|H_{\textrm{gkp}}^{(\textrm{sq})}\rangle$. 

Recall that the logical Hadamard operator of the square-lattice GKP code is given by 
\begin{align}
\hat{H}_{\textrm{gkp}} = e^{ i\frac{\pi}{2}\hat{n} } , 
\end{align} 
i.e., a phase rotation by $90\degree$. Thus, the measurement of the Hadamard operator for the GKP code can be implemented by measuring the excitation number operator $\hat{n}$ modulo $4$ \cite{Gottesman2001}. Such a modular measurement of the excitation number is a non-Gaussian resource. In circuit QED systems, this can be done by using the $2$-bit phase estimation circuit in Fig.\ \ref{fig:n bit phase estimation circuit}(b). Similarly as in the case of the excitation parity measurement, we need a qubit-state-conditional displacement operation $\hat{I} \otimes |g\rangle \langle g| + \hat{D}(\alpha) |e\rangle\langle e|$. In circuit QED systems, the qubit-state-conditional displacement operation can be realized by using dispersive coupling $\hat{H} = -\chi \hat{a}^{\dagger}\hat{a} |e\rangle\langle e|$ between a microwave cavity mode and a transmon qubit.        

However, not all physical systems are equipped with a qubit-state-conditional displacement operation and thus measuring the excitation number modulo $4$ can be challenging. On the other hand, if one ever wants to implement the GKP code, it is necessary to have an ability to prepare a non-Gaussian GKP state $|0_{\textrm{gkp}}^{(\textrm{sq})}\rangle$ or  $|+_{\textrm{gkp}}^{(\textrm{sq})}\rangle = \hat{H}_{\textrm{gkp}} |0_{\textrm{gkp}}^{(\textrm{sq})}\rangle$ in any case. Therefore, it would be ideal if one can leverage such an ability to prepare a non-Gaussian GKP state and use it to prepare a magic state encoded in the GKP code. Surprisingly, this is indeed the case. That is, we can prepare an H-type magic state in the GKP code by using a vacuum state $|0\rangle$, Gaussian operations, homodyne measurements, and the GKP states $|0_{\textrm{gkp}}^{(\textrm{sq})}\rangle$ and $|+_{\textrm{gkp}}^{(\textrm{sq})}\rangle$ \cite{Terhal2016,Baragiola2019}. Below, we will review why this is the case. 


It was first realized in Ref.\ \cite{Terhal2016} that one can avoid the use of excitation number measurement modulo $4$ if one starts with a state that is already invariant under the $90\degree$ rotation. More specifically, the vacuum state $|0\rangle$, which is a Gaussian state, is already stabilized by the logical Hadamard operator, i.e., 
\begin{align}
\hat{H}_{\textrm{gkp}} |0\rangle =  e^{i\frac{\pi}{2}\hat{n}}|0\rangle = |0\rangle. 
\end{align} 
On the other hand, the vacuum state $|0\rangle$ is apparently not stabilized by the GKP stabilizers $\hat{S}_{q}$ and $\hat{S}_{p}$, and therefore is not a valid logical state. We can address this by measuring the GKP stabilizers $\hat{S}_{q}$ and $\hat{S}_{p}$ on the vacuum state in a non-destructive way. In particular, because the stabilizers commute with the logical Hadamard operator $\hat{H}_{\textrm{gkp}}$, the state after the stabilizer measurement will still be stabilized by the logical Hadamard operator. This then implies that if we post-select the outcome with $\hat{S}_{q} = \hat{S}_{p}=1$, or $\hat{q} = \hat{p} =0$ modulo $\sqrt{\pi}$, we will end up with a state
\begin{align}
|\psi(0,0)\rangle &\propto \hat{\Pi}_{\hat{S}_{q} = 1} \hat{\Pi}_{\hat{S}_{p} = 1} |0\rangle, 
\end{align} 
where $\hat{\Pi}_{\hat{S}_{q} = 1} $ and $\hat{\Pi}_{\hat{S}_{p} = 1} $ are the projection operators to the subspace defined by $\hat{S}_{q} = 1$ and $\hat{S}_{p} = 1$, respectively. Since the post-selected state $|\psi(0,0)\rangle$ is now stabilized by the logical Hadamard operator as well as all the GKP stabilizers, it is equivalent to the encoded H-type magic state $|H_{\textrm{gkp}}^{(\textrm{sq})}\rangle$ up to an overall phase, i.e., 
\begin{align}
|\psi(0,0)\rangle &\propto |H_{\textrm{gkp}}^{(\textrm{sq})}\rangle. \label{eq:magic GKP state origin from words}
\end{align}
Note that the only non-Gaussian resource needed in the above magic state preparation scheme is the ability to perform the $\hat{S}_{q}$ and the $\hat{S}_{p}$ GKP stabilizer measurements, which can be implemented by using GKP states as a non-Gaussian resource. Therefore, preparation of a GKP state $|0_{\textrm{gkp}}^{(\textrm{sq})}\rangle$ or $|+_{\textrm{gkp}}^{(\textrm{sq})}\rangle$ is sufficient for realizing universal quantum computation with the GKP code. 

Note, however, that the above magic state preparation scheme is non-deterministic. In fact, the success probability of the scheme is zero as the only accepted measurement outcome is $\hat{q} = \hat{p} =0$ mod $\sqrt{\pi}$. Nevertheless, this is not a fundamental problem because, as observed in Ref.\ \cite{Terhal2016}, we would obtain a GKP state that is close to the desired H-type magic state if the measurement outcome is close to desired outcome $\hat{q} = \hat{p} =0$ mod $\sqrt{\pi}$ (i.e., if $\hat{S}_{q} \sim 1$ and $\hat{S}_{p}\sim 1$). This intuition has been rigorously verified in Ref.\ \cite{Baragiola2019}. In particular, Ref.\ \cite{Baragiola2019} showed post-selection is not really necessary because we obtain a non-trivial state that can be distilled to the ideal H-type magic with $100\%$ probability. We review this result below.    
 
To understand the general case where the measurement outcome is not necessarily given by $\hat{S}_{q} = \hat{S}_{p}=1$, let us unpack the stabilizer measurement circuits in Figs.\ \ref{fig:GKP stabilizer measurement position} and \ref{fig:GKP stabilizer measurement momentum} in more detail. Specifically, we will find the Kraus operator associated with each stabilizer measurement outcome. Recall the $\hat{S}_{q}$ stabilizer circuit in Fig.\ \ref{fig:GKP stabilizer measurement position} and consider the action of the SUM gate $\textrm{SUM}_{D\rightarrow A}$ on an arbitrary input state in the data mode $|\psi_{D}\rangle = \int_{-\infty}^{\infty} dq \psi(q) |\hat{q}_{D} = q \rangle$ and the ancilla GKP state $|+_{\textrm{gkp}}^{(\textrm{sq})}\rangle$: 
\begin{align}
\textrm{SUM}_{D\rightarrow A} |\psi_{D}\rangle |+_{\textrm{gkp}}^{(\textrm{sq})}\rangle &= \frac{1}{\sqrt{2}}\sum_{n\in\mathbb{Z}} \int_{-\infty}^{\infty} dq \psi(q)  e^{-i\hat{q}_{a}\hat{p}_{b}}  |\hat{q}_{D} = q \rangle    |\hat{q}_{A} = n\sqrt{\pi}\rangle 
\nonumber\\
&= \frac{1}{\sqrt{2}}\sum_{n\in\mathbb{Z}} \int_{-\infty}^{\infty} dq \psi(q)   |\hat{q}_{D} = q \rangle    |\hat{q}_{A} = n\sqrt{\pi} + q\rangle . 
\end{align}
Now suppose that we measured $\hat{q}_{A} = z_{q}$ in the ancilla mode via a homodyne measurement of the ancilla position operator. Then, we are left with the state 
\begin{align}
\langle \hat{q}_{A} = z_{q}| \textrm{SUM}_{a\rightarrow b} |\psi_{D}\rangle |+_{\textrm{gkp}}^{(\textrm{sq})}\rangle &= \frac{1}{\sqrt{2}}\sum_{n\in\mathbb{Z}}  \psi(z_{q}-n\sqrt{\pi})   |\hat{q}_{D} = z_{q}-n\sqrt{\pi} \rangle   
\nonumber\\
&= \Big{[} \frac{1}{\sqrt{2}}\sum_{n\in\mathbb{Z}}    |\hat{q}_{D} = z_{q}-n\sqrt{\pi} \rangle \langle \hat{q}_{D} = z_{q}-n\sqrt{\pi}| \Big{]} |\psi\rangle 
\nonumber\\
&= \Big{[} \frac{1}{\sqrt{2}}\sum_{n\in\mathbb{Z}}    |\hat{q}_{D} = z_{q}+n\sqrt{\pi} \rangle \langle \hat{q}_{D} = z_{q}+n\sqrt{\pi}| \Big{]} |\psi\rangle. 
\end{align}
Note that after correcting for the shift $z_{q}$ by applying a counter displacement operation $e^{i z_{q}  \hat{p}_{D}}$, we get 
\begin{align}
e^{i z_{q} \hat{p}_{D}} \langle \hat{q}_{A} = z_{q} | \textrm{SUM}_{D\rightarrow A} |\psi_{D}\rangle |+_{\textrm{gkp}}^{(\textrm{sq})}\rangle &= \Big{[} \frac{1}{\sqrt{2}}\sum_{n\in\mathbb{Z}}    |\hat{q}_{D} = n\sqrt{\pi}\rangle \langle \hat{q}_{D} = z_{q}+n\sqrt{\pi}| \Big{]} |\psi\rangle. \label{eq:GKP stabilizer measurement and shift correction for position}
\end{align} 
Thus, it is clear that after the $\hat{S}_{q}$ stabilizer measurement and the shift correction, the final state satisfies $\hat{q}_{D} =0$ mod $\sqrt{\pi}$ and is indeed stabilized by the stabilizer $\hat{S}_{q}^{(D)} = e^{i2\sqrt{\pi}\hat{q}_{D}}$ regardless of the input state $|\psi\rangle$. Note that the correction shift is given by $z_{q}$ instead of $R_{\sqrt{\pi}}(z_{q})$ as in the case of the usual GKP error correction. This is not a problem since the input state $|0\rangle$ is known. From Eq.\ \eqref{eq:GKP stabilizer measurement and shift correction for position}, we can see that the Kraus operator $\hat{K}_{\textrm{EC}}^{(q)}(z_{q})$ associated with the measurement outcome $z_{q}$ is given by 
\begin{align}
\hat{K}_{\textrm{EC}}^{(q)}(z_{q})&= \frac{1}{\sqrt{2}}\sum_{n\in\mathbb{Z}}    |\hat{q}_{D} = n\sqrt{\pi} \rangle \langle \hat{q}_{D} = z_{q}+ n\sqrt{\pi}|. 
\end{align}
One can similarly show that for the $\hat{S}_{p}$ stabilizer measurement (see Fig.\ \ref{fig:GKP stabilizer measurement momentum}) followed by a shift correction (by a counter displacement operation $e^{-i z_{p} \hat{q}_{D} }$), the Kraus operator associated with the measurement outcome $z_{p}$ is given by 
\begin{align}
\hat{K}_{\textrm{EC}}^{(p)}(z_{p}) &= \frac{1}{\sqrt{2}}\sum_{n\in\mathbb{Z}}    |\hat{p}_{D} = n\sqrt{\pi} \rangle \langle \hat{p}_{D} = z_{p}+n\sqrt{\pi}|. 
\end{align}

With all the tools ready, let us apply these Kraus operators to the vacuum state $|0\rangle$ sequentially (the $\hat{S}_{p}$ measurement first and then the $\hat{S}_{q}$ measurement) to get the output state conditioned on measuring $( z_{q} , z_{p} )$: 
\begin{align}
|\psi(z_{q} , z_{p} ) \rangle &\equiv \hat{K}_{\textrm{EC}}^{(q)}(z_{q}) \hat{K}_{\textrm{EC}}^{(p)}(z_{p})  |0\rangle 
\nonumber\\
&\propto \sum_{m,n\in \mathbb{Z}} |\hat{q} = n\sqrt{\pi} \rangle  \times \langle \hat{q} = z_{q} + n \sqrt{\pi} | \hat{p} = m\sqrt{\pi} \rangle \times \langle  \hat{p} = z_{p} + m \sqrt{\pi}  | 0\rangle 
\nonumber\\
&\propto \sum_{m,n\in \mathbb{Z}} |\hat{q} = n\sqrt{\pi} \rangle  \times  e^{i (z_{q} + n\sqrt{\pi})m\sqrt{\pi} }  \times e^{-\frac{1}{2} ( z_{p} + m \sqrt{\pi} )^{2} } 
\nonumber\\
&= \sum_{m,n\in \mathbb{Z}} |\hat{q} = n\sqrt{\pi} \rangle  \times e^{imn \pi}  \times  e^{i z_{q} m\sqrt{\pi} }  \times e^{-\frac{1}{2} ( z_{p} + m \sqrt{\pi} )^{2} } . 
\end{align} 
Note that for $m,n\in\mathbb{Z}$, $e^{imn\pi}$ is given by $1$ for even $n$ and $(-1)^{m}$ for odd $n$. Thus, we have 
\begin{align}
|\psi(z_{q} , z_{p} ) \rangle &\propto   \Big{[} \sum_{m\in \mathbb{Z}}  e^{i z_{q} m\sqrt{\pi} }  e^{-\frac{1}{2} ( z_{p} + m \sqrt{\pi} )^{2} } \Big{]}  \times \sum_{n\in \mathbb{Z}} |\hat{q} = 2n\sqrt{\pi} \rangle 
\nonumber\\
&+ \Big{[} \sum_{m\in \mathbb{Z}}  (-1)^{m} e^{i z_{q} m\sqrt{\pi} }  e^{-\frac{1}{2} ( z_{p} + m \sqrt{\pi} )^{2} } \Big{]}  \times \sum_{n\in \mathbb{Z}} |\hat{q} = (2n+1)\sqrt{\pi} \rangle 
\nonumber\\
&=  \sum_{\mu \in \lbrace 0,1 \rbrace} c_{\mu}(z_{q},z_{p}) | \mu_{\textrm{gkp}}^{(\textrm{sq})}\rangle , 
\end{align}
where the unnormalized coefficient $c_{\mu}$ is defined as 
\begin{align}
c_{\mu}(z_{q},z_{p}) \equiv \sum_{m\in \mathbb{Z}}  (-1)^{m\mu} e^{i z_{q} m\sqrt{\pi} }  e^{-\frac{1}{2} ( z_{p} + m \sqrt{\pi} )^{2} }, 
\end{align}
for $\mu \in \lbrace 0,1 \rbrace$. For the desired measurement outcome $(z_{q},z_{p})=(0,0)$, we have 
\begin{align}
c_{0}(0,0) &= \sum_{m\in\mathbb{Z}} e^{-\frac{\pi}{2} m^{2} }, 
\nonumber\\
c_{1}(0,0) &= \sum_{m\in\mathbb{Z}} (-1)^{m} e^{-\frac{\pi}{2} m^{2} }  =  \sum_{m\in 2\mathbb{Z}} e^{-\frac{\pi}{2} m^{2} } - \sum_{m\in 2\mathbb{Z} +1 } e^{-\frac{\pi}{2} m^{2} } . 
\end{align}
Note that 
\begin{align}
\sum_{m\in 2\mathbb{Z}}  e^{-\frac{\pi}{2} m^{2} } &= \sum_{m\in\mathbb{Z}} e^{-2\pi m^{2}}
\nonumber\\
&= \sum_{k\in\mathbb{Z}} \int_{-\infty}^{\infty} dx e^{i2\pi k x} e^{-2\pi x^{2}} = \frac{1}{\sqrt{2}} \sum_{k\in\mathbb{Z}} e^{-\frac{\pi}{2} k^{2} } = \frac{1}{\sqrt{2}} \sum_{m\in\mathbb{Z}} e^{-\frac{\pi}{2} m^{2} } ,  
\end{align}
where we used the Poisson summation formula to get the second equality. Also, $\sum_{m\in 2\mathbb{Z} + 1}  e^{-\frac{\pi}{2} m^{2} }$ is given by 
\begin{align}
\sum_{m\in 2\mathbb{Z} + 1 }  e^{-\frac{\pi}{2} m^{2} } &= \sum_{m\in \mathbb{Z}}  e^{-\frac{\pi}{2} m^{2} } - \sum_{m\in 2\mathbb{Z}}  e^{-\frac{\pi}{2} m^{2} } = \Big{(} 1 - \frac{1}{\sqrt{2}} \Big{)} \sum_{m\in \mathbb{Z}}  e^{-\frac{\pi}{2} m^{2} } . 
\end{align}
Putting everything together, we find that 
\begin{align}
\frac{ c_{1}(0,0) }{c_{0}( 0,0 )} &= \frac{ \sum_{m\in 2\mathbb{Z}}  e^{-\frac{\pi}{2} m^{2} } - \sum_{m\in 2\mathbb{Z} +1 }  e^{-\frac{\pi}{2} m^{2} } }{ \sum_{m\in \mathbb{Z}}  e^{-\frac{\pi}{2} m^{2} }  } = \frac{1}{\sqrt{2}} - \Big{(} 1 - \frac{1}{\sqrt{2}} \Big{)} = \sqrt{2} - 1 = \tan\Big{(} \frac{\pi}{8} \Big{)}, 
\end{align}
and thus 
\begin{align}
|\psi(0,0)\rangle &\propto \cos\Big{(} \frac{\pi}{8} \Big{)} |0_{\textrm{gkp}}^{(\textrm{sq})}\rangle + \sin\Big{(} \frac{\pi}{8} \Big{)} |1_{\textrm{gkp}}^{(\textrm{sq})}\rangle = |H_{\textrm{gkp}}^{(\textrm{sq})}\rangle, 
\end{align}
confirming the conclusion in Eq.\ \eqref{eq:magic GKP state origin from words} by an explicit calculation. 

\begin{figure}[t!]
\centering
\includegraphics[width=3.8in]{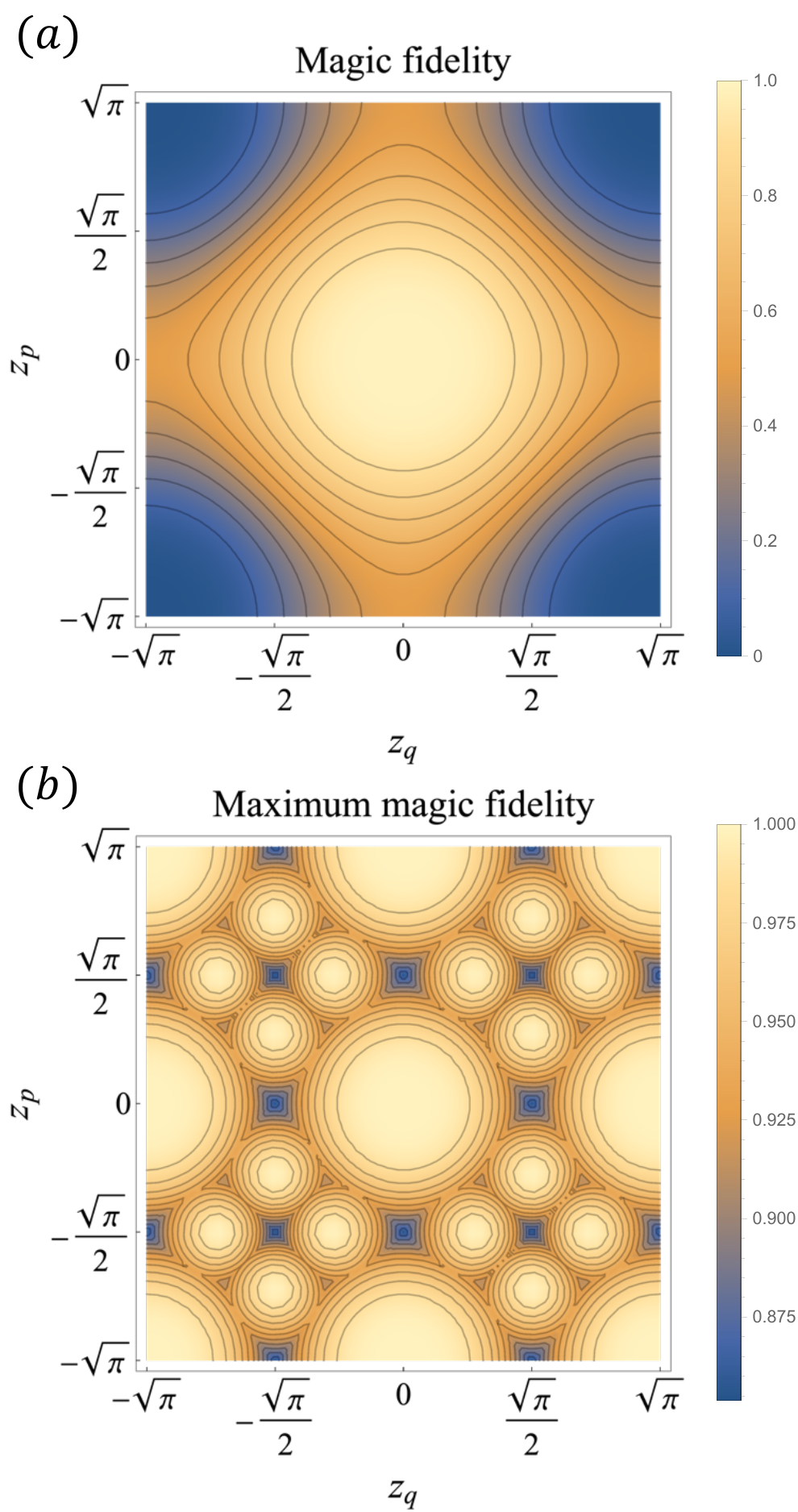}
\caption{(a) Fidelity of between the output state $|\psi(z_{q},z_{p})\rangle = \hat{K}_{\textrm{EC}}^{(q)}(z_{q}) \hat{K}_{\textrm{EC}}^{(p)}(z_{p})  |0\rangle $ and the ideal H-type magic state $|H\rangle = \cos(\frac{\pi}{8} ) |0\rangle + \sin( \frac{\pi}{8} )|1\rangle$. (b) [Reproduction of Fig.\ 2(a) in PRL \textbf{123}, 200502 (2019)] Maximum fidelity between the output state $|\psi(z_{q},z_{p})\rangle$ and $12$ different magic states that are equivalent to the magic state $|H\rangle$ up to a Clifford operation. The $12$ equivalent H-type magic states are $ \hat{S}^{n} |H\rangle$, $ \hat{S}^{n} \hat{X} |H\rangle$, $ \hat{S}^{n} \hat{H} \hat{S}^{\dagger} |H\rangle$  with $n\in \lbrace 0,1,2,3\rbrace$, where $\hat{S}$ is the phase gate, $\hat{X}$ is the Pauli X operator, and $\hat{H}$ is the Hadamard operator.      }
\label{fig:GKP magic fidelity}
\end{figure}

In the most general case where $(z_{q},z_{p})$ is not necessarily given by $(0,0)$, we can compute the fidelity between the output state $|\psi(z_{q},z_{p})\rangle$ with the ideal H-type magic state $|H_{\textrm{gkp}}^{(\textrm{sq})}\rangle$, i.e., 
\begin{align}
F_{H}(z_{q},z_{p}) &\equiv \frac{ |\langle \psi(z_{q},z_{p}) | H_{\textrm{gkp}}^{(\textrm{sq})}\rangle |^{2}  }{ |\langle \psi(z_{q},z_{p}) | \psi(z_{q},z_{p})\rangle|^{2} }= \frac{ |\cos(\frac{\pi}{8}) c_{0}(z_{q},z_{p}) + \sin(\frac{\pi}{8}) c_{1}(z_{q},z_{p})  |^{2} }{ |c_{0}(z_{q},z_{p})|^{2} + |c_{1}(z_{q},z_{p})|^{2} } . 
\end{align}  
The numerically evaluated magic fidelity $F_{H}(z_{q},z_{p})$ is plotted in Fig.\ \ref{fig:GKP magic fidelity}(a). While one might expect that the magic fidelity $F_{H}(z_{q},z_{p})$ is periodic in $z_{q}$ and $z_{p}$ with period $\sqrt{\pi}$, the period is in fact given by $2\sqrt{\pi}$. The reason for this is that in the shift correction, the sizes of the correction shifts were given by $z_{q}$ and $z_{p}$, instead of $R_{\sqrt{\pi}}(z_{q})$ and $R_{\sqrt{\pi}}(z_{p})$. Note also that the magic fidelity $F_{H}(z_{q},z_{p})$ vanishes when $(z_{q},z_{p}) = (\pm \sqrt{\pi},\pm\sqrt{\pi})$. This means that at these points, the system is in the other magic state $|-H_{\textrm{gkp}}^{(\textrm{sq})}\rangle = \hat{Y}_{\textrm{gkp}} |H_{\textrm{gkp}}^{(\textrm{sq})}\rangle$, which is orthogonal to $|H_{\textrm{gkp}}^{(\textrm{sq})}\rangle$. Note that this other magic state is equivalent to the H-type magic state $|H_{\textrm{gkp}}^{(\textrm{sq})}\rangle$ up to a Clifford operation. Thus, the output states $| \psi(\pm \sqrt{\pi},\pm \sqrt{\pi})\rangle$ are as resourceful as the ideal output state $| \psi(0,0)\rangle$. Hence, the magic fidelity $F_{H}(z_{q},z_{p})$ can be misleading as it gives an impression that the states $| \psi(\pm \sqrt{\pi},\pm \sqrt{\pi})\rangle$ are completely useless.          

Note that there are $12$ different magic states that are equivalent to the H-type magic state $|H\rangle$ up to a Clifford operation. These $12$ magic states are explicitly given by 
\begin{align}
\hat{S}^{n}|H\rangle &= \cos\Big{(} \frac{\pi}{8} \Big{)} |0\rangle + i^{n} \sin\Big{(} \frac{\pi}{8} \Big{)} |1\rangle, 
\nonumber\\
\hat{S}^{n}\hat{X}|H\rangle &= \sin\Big{(} \frac{\pi}{8} \Big{)} |0\rangle + i^{n} \cos\Big{(} \frac{\pi}{8} \Big{)} |1\rangle, 
\nonumber\\
\hat{S}^{n}\hat{H}\hat{S}^{\dagger}|H\rangle &= \frac{1}{\sqrt{2}} e^{-i\frac{\pi}{8}} |0\rangle +  \frac{i^{n}}{\sqrt{2}}  e^{i\frac{\pi}{8}} |1\rangle , \,\,\, \textrm{where}\,\,\, n \in \lbrace 0,1,2,3 \rbrace. \label{eq:H-type magic state list}
\end{align}    
In Fig.\ \ref{fig:GKP magic fidelity}(b), we plot the maximum fidelity between the output state $| \psi(z_{q},z_{p})\rangle$ and the $12$ equivalent H-type magic states in Eq.\ \eqref{eq:H-type magic state list} and thereby reproduce Fig.\ 2(a) in Ref.\ \cite{Baragiola2019} (i.e., the maximum magic fidelity $F_{H}^{\textrm{max}}(z_{q},z_{p})$). In this case, the maximum magic fidelity is always larger than $\frac{1}{2}(1+\frac{1}{\sqrt{2}}) = 0.8535\cdots$ for all values of $(z_{q},z_{p})\in\mathbb{R}^{2}$. Note that any state that has a fidelity larger than $\frac{1}{2}(1+\frac{1}{\sqrt{2}})$ with one of the $12$ H-type magic states can be distilled to an ideal H-type magic state by using Clifford operations and Pauli measurements via a magic state distillation protocol \cite{Reichardt2005}. This means that we do not need to discard any measurement outcome $(z_{q},z_{p})$ because all the output states $| \psi(z_{q},z_{p})\rangle$ are distillable to an ideal magic state by using Clifford operations and Pauli measurements which, in the case of the GKP code, can be implemented by using Gaussian operations and homodyne measurements.

\subsection{Generalizations of GKP codes} 
\label{subsection:Generalization of GKP codes}

We have so far discussed a single-mode GKP code that has a square-lattice structure and encodes a qubit. As was realized in Ref.\ \cite{Gottesman2001} and further explored in Ref.\ \cite{Harrington2001,Harrington2004}, it is possible to define a single-mode GKP code that has a different lattice structure (most importantly, a hexagonal-lattice structure) than the square-lattice structure. Furthermore, it is also possible to encode a $d$-dimensional qudit into an oscillator as well as to define a multi-mode GKP code that encodes logical qudits collectively over multiple bosonic modes. We will review these generalizations below.   

\subsubsection{Generalized single-mode GKP codes}

Recall that the stabilizers of the square-lattice GKP code are given by
\begin{align}
\hat{S}_{q} &= \exp[ i 2\sqrt{\pi} \hat{q} ], 
\nonumber\\
\hat{S}_{p} &= \exp[ -i 2\sqrt{\pi} \hat{p} ] . 
\end{align}
More generally, we can consider the following form of the stabilizers: 
\begin{align}
\hat{S}_{q}^{(\boldsymbol{S})} &= \exp[ i\sqrt{2\pi} (S_{qq} \hat{q} + S_{qp} \hat{p}  ) ] , 
\nonumber\\
\hat{S}_{p}^{(\boldsymbol{S})} &= \exp[ - i\sqrt{2\pi} (S_{pq} \hat{q} + S_{pp} \hat{p}  ) ] , 
\end{align}
where the $2\times 2$ matrix $\boldsymbol{S}$ is defined as follows 
\begin{align}
\boldsymbol{S} &\equiv \begin{bmatrix}
S_{qq} & S_{qp} \\
S_{pq} & S_{pp} \\
\end{bmatrix}. 
\end{align}
In the case of the square-lattice GKP code, we have 
\begin{align}
\boldsymbol{S} = \boldsymbol{S}^{(\textrm{sq})} &\equiv \sqrt{2} \boldsymbol{I}_{2}, 
\end{align}
where $\boldsymbol{I}_{n}$ is an $n\times n$ identity matrix. 

In the general case, we need to understand under which condition on $\boldsymbol{S}$ the two stabilizers $\hat{S}_{q}^{(\boldsymbol{S})} $ and $\hat{S}_{p}^{(\boldsymbol{S})} $ commute with each other. To do so, recall the Baker–Campbell–Hausdorff formula, i.e., 
\begin{align}
e^{\hat{A}}e^{\hat{B}}e^{-\hat{A}} = \exp\Big{[} \hat{B} + [\hat{A},\hat{B}] + \frac{1}{2!}[\hat{A},[\hat{A},\hat{B}]] + \frac{1}{3!}[\hat{A},[\hat{A},[\hat{A},\hat{B}]]] \cdots \Big{]} , \label{eq:BCH formula}
\end{align}
and note that 
\begin{align}
\hat{S}_{q}^{(\boldsymbol{S})}  \hat{S}_{p}^{(\boldsymbol{S})}   ( \hat{S}_{q}^{(\boldsymbol{S})}  )^{\dagger} &= \exp\Big{[} - i\sqrt{2\pi} (S_{pq} \hat{q} + S_{pp} \hat{p}  ) + [ i\sqrt{2\pi} (S_{qq} \hat{q} + S_{qp} \hat{p}  ),  - i\sqrt{2\pi} (S_{pq} \hat{q} + S_{pp} \hat{p}  ) ] \Big{]} 
\nonumber\\
&= \exp\Big{[} - i\sqrt{2\pi} (S_{pq} \hat{q} + S_{pp} \hat{p}  ) + 2\pi i ( S_{qq}S_{pp} - S_{qp}S_{pq} ) ] \Big{]} 
\nonumber\\
&= \hat{S}_{p}^{(\boldsymbol{S})}  \exp[ i2\pi \cdot  \textrm{det}(\boldsymbol{S}) ] , 
\end{align}
or equivalently, 
\begin{align}
\hat{S}_{q}^{(\boldsymbol{S})}  \hat{S}_{p}^{(\boldsymbol{S})}   &=  \hat{S}_{p}^{(\boldsymbol{S})}  \hat{S}_{q}^{(\boldsymbol{S})}  \exp[ i2\pi \cdot  \textrm{det}(\boldsymbol{S}) ]. 
\end{align}
Hence, for the two stabilizers $\hat{S}_{q}^{(\boldsymbol{S})}$ and $\hat{S}_{p}^{(\boldsymbol{S})} $ to commute with each other, the matrix $\boldsymbol{S}$ should satisfy 
\begin{align}
\textrm{det}( \boldsymbol{S} ) \in \mathbb{Z}. \label{eq:condition on the matrix single-mode GKP}
\end{align}
For instance, in the case of the square-lattice GKP code, $\boldsymbol{S}^{(\textrm{sq})}$ satisfies
\begin{align}
\textrm{det}( \boldsymbol{S}^{(\textrm{sq})} ) = \textrm{det}( \sqrt{2} \boldsymbol{I}_{2} ) = 2 \in \mathbb{Z}. 
\end{align}
As will be made clear below, it is not a coincidence that the integer $\textrm{det}( \boldsymbol{S}^{(\textrm{sq})} ) = 2$ equals the dimension of the code space of the square-lattice GKP code $\mathcal{C}_{\textrm{gkp}}^{(\textrm{sq})}$. 

Let us first consider the case with $\textrm{det}( \boldsymbol{S} ) = 1$. Even more specifically, consider the simplest case $\boldsymbol{S} = \boldsymbol{I}_{2}$. Then, the two stabilizers are given by 
\begin{align}
\hat{S}_{q}^{ (\boldsymbol{I}_{2}) } &= e^{i\sqrt{2\pi} \hat{q} } , 
\nonumber\\  
\hat{S}_{p}^{ (\boldsymbol{I}_{2}) } &= e^{i\sqrt{2\pi} \hat{p} }. 
\end{align}
We will show that the following state 
\begin{align}
|\textrm{GKP}\rangle &= \sum_{n\in \mathbb{Z} } |\hat{q} = \sqrt{2\pi}n\rangle 
\end{align}
is the unique state (up to an overall phase and normalization) that is stabilized by the two stabilizers $\hat{S}_{q}^{ (\boldsymbol{I}_{2}) } $ and $\hat{S}_{p}^{ (\boldsymbol{I}_{2}) } $. Note that if a state $|\psi\rangle = \int_{-\infty}^{\infty} dq \psi(q) |\hat{q} =q\rangle$ is stabilized by $\hat{S}_{q}^{ (\boldsymbol{I}_{2}) } $, i.e., $\hat{S}_{q}^{ (\boldsymbol{I}_{2}) } |\psi\rangle = |\psi\rangle$, we have 
\begin{align}
\psi(q) = \langle \hat{q} = q | \hat{S}_{q}^{ (\boldsymbol{I}_{2}) }   |\psi\rangle &= \langle \hat{q} = q  | e^{i\sqrt{2\pi} \hat{q} }  \int_{-\infty}^{\infty} dq' \psi(q') |\hat{q} =q'\rangle = \psi(q)e^{i\sqrt{2\pi} q}. 
\end{align}
Thus, for any $q$ such that $e^{i\sqrt{2\pi} q} \neq 1$ (or equivalently, for any $\hat{q}$ such that $\hat{q}  \neq 0 $ mod $\sqrt{\pi}$), $\psi(q)$ has to vanish. Hence, we are left with 
\begin{align}
|\psi\rangle &= \sum_{n\in\mathbb{Z}} \psi( \sqrt{2\pi}n ) |\hat{q} = \sqrt{2\pi}n\rangle = \sum_{n\in\mathbb{Z}} \psi_{n} |\hat{q} = \sqrt{2\pi}n\rangle, 
\end{align}
where $\psi_{n} \equiv \psi( \sqrt{2\pi}n )$. Further requiring that the state $|\psi\rangle $ should be stabilized by the other stabilizer $\hat{S}_{p}^{ (\boldsymbol{I}_{2}) } = e^{i\sqrt{2\pi} \hat{p} }$, we have 
\begin{align}
\psi_{n} &= \langle \hat{q} = \sqrt{2\pi} n |\psi\rangle 
\nonumber\\
&= \langle \hat{q}= \sqrt{2\pi} n | \hat{S}_{p}^{(\boldsymbol{I}_{2})} |\psi\rangle 
\nonumber\\
&= \langle \hat{q} = \sqrt{2\pi}n  | e^{-i\sqrt{2\pi} \hat{p} } \sum_{m\in\mathbb{Z}} \psi_{m} |\hat{q} = \sqrt{2\pi}m\rangle 
\nonumber\\
&= \sum_{m\in\mathbb{Z}}  \psi_{m} \langle \hat{q} = \sqrt{2\pi}n   |\hat{q} = \sqrt{2\pi}(m+1)\rangle  = \psi_{n-1}, \,\,\, \textrm{for all}\,\,\, n\in\mathbb{Z}. 
\end{align}  
We can therefore conclude that $\psi_{n}$ is independent on $n$ and the state $|\psi\rangle$ has to be equivalent to the state $|\textrm{GKP}\rangle$ up to an overall phase and normalization, i.e., 
\begin{align}
|\psi\rangle &= \sum_{n\in\mathbb{Z}} \psi_{n} |\hat{q} = \sqrt{2\pi}n\rangle \propto  \sum_{n\in\mathbb{Z}} |\hat{q} = \sqrt{2\pi}n\rangle = |\textrm{GKP}\rangle. 
\end{align}
Also, we will refer to the state $|\textrm{GKP}\rangle$ as the canonical GKP state. 

Let us move on the the case where $\textrm{det}( \boldsymbol{S} ) = 1$, but $\boldsymbol{S}$ is not necessarily given by the identity matrix $\boldsymbol{I}_{2}$. Similarly as above, the state that is stabilized by the stabilizers $\hat{S}_{q}^{(\boldsymbol{S})}$ and $\hat{S}_{p}^{(\boldsymbol{S})}$ is unique up to an overall phase and normalization. In particular, we will show that the unique state $|\textrm{GKP}_{\boldsymbol{S} }\rangle$ can be obtained by applying a Gaussian operation $\hat{U}_{\boldsymbol{S}^{-1}}$ to the canonical GKP state $|\textrm{GKP}\rangle$, i.e., 
\begin{align}
|\textrm{GKP}_{\boldsymbol{S} }\rangle &= \hat{U}_{\boldsymbol{S}^{-1}} |\textrm{GKP}\rangle . \label{eq:general GKP state from the canonical GKP state}
\end{align}   
This is a very desirable property because it means that the only non-Gaussian resource needed to prepare the state $|\textrm{GKP}_{\boldsymbol{S} }\rangle$ is the preparation of the canonical GKP state. Everything else can be done by using a Gaussian operation. To see why Eq.\ \eqref{eq:general GKP state from the canonical GKP state} holds, recall that a Gaussian operation $\hat{U}_{\boldsymbol{S}}$ satisfies the following property: 
\begin{align}
\hat{U}_{\boldsymbol{S}}^{\dagger} \boldsymbol{\hat{x}} \hat{U}_{\boldsymbol{S}} = \boldsymbol{S} \boldsymbol{\hat{x}}, \label{eq:property of the Gaussian operation with S}
\end{align}
where $\boldsymbol{\hat{x}} \equiv (\hat{q},\hat{p})^{T}$ (see Appendix \ref{appendix:Gaussian states, unitaries, and channels}). Thus, we can see that 
\begin{align}
\hat{S}_{q}^{(\boldsymbol{S})} |\textrm{GKP}_{\boldsymbol{S} }\rangle = \hat{S}_{q}^{(\boldsymbol{S})} \hat{U}_{\boldsymbol{S}^{-1}} |\textrm{GKP}\rangle &= \hat{U}_{\boldsymbol{S}^{-1}} \hat{U}_{\boldsymbol{S}^{-1}}^{\dagger} \hat{S}_{q}^{(\boldsymbol{S})} \hat{U}_{\boldsymbol{S}^{-1}} |\textrm{GKP}\rangle 
\nonumber\\
&= \hat{U}_{\boldsymbol{S}^{-1}} \hat{U}_{\boldsymbol{S}^{-1}}^{\dagger} \exp[ i\sqrt{\pi} (\boldsymbol{S}\boldsymbol{\hat{x}})_{1} ]\hat{U}_{\boldsymbol{S}^{-1}} |\textrm{GKP}\rangle 
\nonumber\\
&= \hat{U}_{\boldsymbol{S}^{-1}}  \exp[ i\sqrt{\pi} (\boldsymbol{S}  \boldsymbol{S}^{-1} \boldsymbol{\hat{x}})_{1} ]  |\textrm{GKP}\rangle 
\nonumber\\
&= \hat{U}_{\boldsymbol{S}^{-1}}    \hat{S}_{q}^{(\boldsymbol{I}_{2})}  |\textrm{GKP}\rangle = \hat{U}_{\boldsymbol{S}^{-1}}    |\textrm{GKP}\rangle = |\textrm{GKP}_{\boldsymbol{S} }\rangle, 
\end{align}  
that is, the state $|\textrm{GKP}_{\boldsymbol{S} }\rangle$ is stabilized by the stabilizer $\hat{S}_{q}^{(\boldsymbol{S})}$. Note that we used Eq.\ \eqref{eq:property of the Gaussian operation with S} to derive the fourth equality, and the fact that the canonical GKP state $ |\textrm{GKP}\rangle$ is stabilized by $ \hat{S}_{q}^{(\boldsymbol{I}_{2})}$ to derive the sixth equality. Similarly, one can also show that the state $|\textrm{GKP}_{\boldsymbol{S} }\rangle$ is stabilized by the other stabilizer $\hat{S}_{p}^{(\boldsymbol{S})}$. Hence, $|\textrm{GKP}_{\boldsymbol{S} }\rangle = \hat{U}_{\boldsymbol{S}^{-1}} |\textrm{GKP}\rangle$ is indeed the unique state (up to an overall phase and normalization) that is stabilized by the stabilizers $\hat{S}_{q}^{(\boldsymbol{S})}$ and $\hat{S}_{p}^{(\boldsymbol{S})}$. Note that the cases with $\textrm{det}( \boldsymbol{S} ) = 1$ are trivial in the context of quantum error correction since there is only one logical state encoded in the code space $\mathcal{C}_{\textrm{gkp}}^{( \boldsymbol{S} )}$. However, the basic properties we have discussed so far will be useful for understanding more interesting cases with $\textrm{dim}(\mathcal{C}_{\textrm{gkp}}^{( \boldsymbol{S} )}) \ge 2$.  

With these basic properties in our hands, we are now ready to analyze the most general case with $\textrm{det}( \boldsymbol{S} ) = d$, where $d$ is an integer such that $d\ge 2$. Consider the following operators
\begin{align}
\hat{Z}_{\textrm{gkp}}^{(\boldsymbol{S})} &\equiv ( \hat{S}_{q}^{(\boldsymbol{S})} )^{\frac{1}{d}}   =  \exp\Big{[} i\frac{ \sqrt{2\pi}}{d} (S_{qq} \hat{q} + S_{qp} \hat{p}  ) \Big{]} , 
\nonumber\\ 
\hat{X}_{\textrm{gkp}}^{(\boldsymbol{S})} &\equiv ( \hat{S}_{p}^{(\boldsymbol{S})} )^{\frac{1}{d}}   =  \exp\Big{[} i\frac{ \sqrt{2\pi}}{d} (S_{pq} \hat{q} + S_{pp} \hat{p}  ) \Big{]} .  
\end{align}
Here, we will show that the dimension of the code space $\mathcal{C}_{\textrm{gkp}}^{( \boldsymbol{S} )}$ equals $d$ when $\textrm{det}( \boldsymbol{S} ) = d$. Furthermore, we will show that the operators $\hat{Z}_{\textrm{gkp}}^{(\boldsymbol{S})}$ and $\hat{X}_{\textrm{gkp}}^{(\boldsymbol{S})}$ act as the logical Pauli Z and X operators (for a qudit) on the encoded states, respectively.  

Note that the two operators $\hat{Z}_{\textrm{gkp}}^{(\boldsymbol{S})}$ and $\hat{X}_{\textrm{gkp}}^{(\boldsymbol{S})}$ commute with the stabilizers $\hat{S}_{q}^{(\boldsymbol{S})}$ and $\hat{S}_{p}^{(\boldsymbol{S})}$, as can be verified by using the BCH formula (see Eq.\ \eqref{eq:BCH formula}). Thus, it makes sense to consider a state $| 0_{\textrm{gkp}}^{( \boldsymbol{S} )}\rangle$ that is simultaneously stabilized by the stabilizer $\hat{S}_{p}^{(\boldsymbol{S})}$ and the operator $\hat{Z}_{\textrm{gkp}}^{(\boldsymbol{S})} = ( \hat{S}_{q}^{(\boldsymbol{S})} )^{\frac{1}{d}} $, i.e., 
\begin{align}
\hat{S}_{p}^{(\boldsymbol{S})} | 0_{\textrm{gkp}}^{( \boldsymbol{S} )}\rangle =\hat{Z}_{\textrm{gkp}}^{(\boldsymbol{S})}  | 0_{\textrm{gkp}}^{( \boldsymbol{S} )}\rangle =  | 0_{\textrm{gkp}}^{( \boldsymbol{S} )}\rangle.  
\end{align}
The state $| 0_{\textrm{gkp}}^{( \boldsymbol{S} )}\rangle$ can also be regarded as a state that is stabilized by $\hat{S}_{q}^{(\boldsymbol{S'})} = \hat{Z}_{\textrm{gkp}}^{(\boldsymbol{S})}$ and $\hat{S}_{p}^{(\boldsymbol{S'})} = \hat{S}_{p}^{(\boldsymbol{S})}$, where $\boldsymbol{S'}$ is given by
\begin{align}
\boldsymbol{S'} = \textrm{diag}\Big{(}\frac{1}{d},1\Big{)} \cdot \boldsymbol{S} 
\end{align}
Since $\textrm{det}( \boldsymbol{S'} ) = \frac{1}{d} \textrm{det}( \boldsymbol{S} ) = 1$, the state $| 0_{\textrm{gkp}}^{( \boldsymbol{S} )}\rangle$ is unique up to an overall phase and normalization, as discussed above. Furthermore, since $\hat{S}_{q}^{(\boldsymbol{S})} = (\hat{Z}_{\textrm{gkp}}^{(\boldsymbol{S})} )^{d}$, the state $| 0_{\textrm{gkp}}^{( \boldsymbol{S} )}\rangle$ is also stabilized by the other stabilizer $\hat{S}_{q}^{(\boldsymbol{S})}$ of the original code space. Thus, the state $| 0_{\textrm{gkp}}^{( \boldsymbol{S} )}\rangle$ is a valid logical state, i.e., $| 0_{\textrm{gkp}}^{( \boldsymbol{S} )}\rangle \in \mathcal{C}_{\textrm{gkp}}^{( \boldsymbol{S} )}$. From now on, we will refer to the state $| 0_{\textrm{gkp}}^{( \boldsymbol{S} )}\rangle$ as the (encoded) computational zero state and the operator $\hat{Z}_{\textrm{gkp}}^{(\boldsymbol{S})}$ as the logical Z operator.  

Note that while the two operators $\hat{Z}_{\textrm{gkp}}^{(\boldsymbol{S})}$ and $\hat{X}_{\textrm{gkp}}^{(\boldsymbol{S})}$ commute with the stabilizers $\hat{S}_{q}^{(\boldsymbol{S})}$ and $\hat{S}_{p}^{(\boldsymbol{S})}$, they do not commute with each other.
\begin{align}
\hat{Z}_{\textrm{gkp}}^{(\boldsymbol{S})} \hat{X}_{\textrm{gkp}}^{(\boldsymbol{S})} = \hat{X}_{\textrm{gkp}}^{(\boldsymbol{S})} \hat{Z}_{\textrm{gkp}}^{(\boldsymbol{S})} \exp\Big{[} i\frac{2\pi}{d^{2}} \textrm{det}(\boldsymbol{S}) \Big{]} = \hat{X}_{\textrm{gkp}}^{(\boldsymbol{S})} \hat{Z}_{\textrm{gkp}}^{(\boldsymbol{S})} \exp\Big{[} i\frac{2\pi}{d} \Big{]}. \label{eq:commutation relation between qudit GKP Pauli operators}
\end{align}
In fact, they satisfy the same commutation relations that the qudit Pauli Z and the Pauli X operators satisfy. For this reason, we refer to the operator $\hat{X}_{\textrm{gkp}}^{(\boldsymbol{S})}$ as the logical Pauli X operator. Also, we can define other computational basis states as follows: 
\begin{align}
|\mu_{\textrm{gkp}}^{ ( \boldsymbol{S} ) } \rangle &\equiv  (\hat{X}_{\textrm{gkp}}^{(\boldsymbol{S})})^{\mu} |0_{\textrm{gkp}}^{ ( \boldsymbol{S} ) } \rangle,  \,\,\, \textrm{where}\,\,\, \mu \in \mathbb{Z}_{d} = \lbrace 0, \cdots, d-1\rbrace. 
\end{align}
Since the logical Pauli X operator $\hat{X}_{\textrm{gkp}}^{(\boldsymbol{S})}$ commutes with the stabilizers, the state $|\mu_{\textrm{gkp}}^{ ( \boldsymbol{S} ) } \rangle$ is also stabilized by the stabilizers and therefore is a valid logical state. Furthermore, these states indeed behave the same way as the usual computational basis states for a qudit. In particular, $|\mu_{\textrm{gkp}}^{ ( \boldsymbol{S} ) } \rangle$ is an eigenstate of the logical Pauli Z operator $\hat{Z}_{\textrm{gkp}}^{(\boldsymbol{S})}$ with an eigenvalue $\hat{Z}_{\textrm{gkp}}^{(\boldsymbol{S})} = e^{i\frac{2\pi}{d}\mu}$. 
\begin{align}
\hat{Z}_{\textrm{gkp}}^{(\boldsymbol{S})} |\mu_{\textrm{gkp}}^{ ( \boldsymbol{S} ) } \rangle &= \hat{Z}_{\textrm{gkp}}^{(\boldsymbol{S})}  (\hat{X}_{\textrm{gkp}}^{(\boldsymbol{S})})^{\mu} |0_{\textrm{gkp}}^{ ( \boldsymbol{S} ) } \rangle 
\nonumber\\
&= e^{i\frac{2\pi}{d}\mu} (\hat{X}_{\textrm{gkp}}^{(\boldsymbol{S})})^{\mu} \hat{Z}_{\textrm{gkp}}^{(\boldsymbol{S})}   |0_{\textrm{gkp}}^{ ( \boldsymbol{S} ) } \rangle 
\nonumber\\
&=  e^{i\frac{2\pi}{d}\mu} (\hat{X}_{\textrm{gkp}}^{(\boldsymbol{S})})^{\mu}   |0_{\textrm{gkp}}^{ ( \boldsymbol{S} ) } \rangle  
\nonumber\\
&= e^{i\frac{2\pi}{d}\mu}|\mu_{\textrm{gkp}}^{ ( \boldsymbol{S} ) } \rangle. 
\end{align} 
The code space $\mathcal{C}_{\textrm{gkp}}^{ (\boldsymbol{S}) }$ is thus the span of the computational basis set $\lbrace |0_{\textrm{gkp}}^{ ( \boldsymbol{S} ) } \rangle, \cdots, |(d-1)_{\textrm{gkp}}^{ ( \boldsymbol{S} ) } \rangle \rbrace$ and is $d$-dimensional. 

\subsubsection{Maximum likelihood decoding of a general single-mode GKP code} 

Let us now analyze the error-correcting capability of the general single-mode GKP code $\mathcal{C}_{\textrm{gkp}}^{( \boldsymbol{S} )}$ encoding a qudit into an oscillator, i.e., $\textrm{det}( \boldsymbol{S} ) = d$. Recall that the stabilizers are given by
\begin{align}
\hat{S}_{q}^{(\boldsymbol{S})} &= \exp[ i\sqrt{2\pi} (S_{qq} \hat{q} + S_{qp} \hat{p}  ) ] = \exp[ i\sqrt{2\pi} (\boldsymbol{S}\boldsymbol{\hat{x}} )_{1} ] , 
\nonumber\\
\hat{S}_{p}^{(\boldsymbol{S})} &= \exp[ - i\sqrt{2\pi} (S_{pq} \hat{q} + S_{pp} \hat{p}  ) ]  = \exp[ i\sqrt{2\pi} (\boldsymbol{S}\boldsymbol{\hat{x}} )_{2} ], 
\end{align}
where $\boldsymbol{\hat{x}} = (\hat{q},\hat{p})^{T}$. Thus, by measuring the stabilizers in the error correction protocol, we can measure the two quadrature operators $(\boldsymbol{S}\boldsymbol{\hat{x}} )_{1} $ and $(\boldsymbol{S}\boldsymbol{\hat{x}} )_{2} $ modulo $\sqrt{2\pi}$. Let us assume that the system is initially in the code space and thus satisfies $\boldsymbol{S}\boldsymbol{\hat{x}} =(0,0)^{T}$ modulo $\sqrt{2\pi}$. Then, consider the Gaussian random shift error $\mathcal{N}_{B_{2}}[\sigma]$, i.e., 
\begin{align}
\boldsymbol{\hat{x}'} = \boldsymbol{\hat{x}}  + \boldsymbol{\xi} , 
\end{align}
where the stochastic noise $\boldsymbol{\xi} =(\xi_{q},\xi_{p})^{T} $ is drawn from an independent and identically distributed Gaussian distribution $(\xi_{q},\xi_{p})\sim_{\textrm{iid}} \mathcal{N}(0, \sigma^{2})$. By measuring the two stabilizers, we can measure $\boldsymbol{S}\boldsymbol{\hat{x}'}$ modulo $\sqrt{2\pi}$. Let us assume that the measurement outcome is $\boldsymbol{z} = (z_{q},z_{p})^{T}$ modulo $\sqrt{2\pi}$. Then, we have
\begin{align}
\boldsymbol{S}\boldsymbol{\hat{x}'}  =  \boldsymbol{S}( \boldsymbol{\hat{x}} + \boldsymbol{\xi} )  = \boldsymbol{S}\boldsymbol{\xi} = \boldsymbol{z} - \sqrt{2\pi}\boldsymbol{n}, 
\end{align}
for some $\boldsymbol{n} = (n_{q},n_{p})^{T} \in \mathbb{Z}^{2}$. The second equality is due to the fact that $\boldsymbol{S}\boldsymbol{\hat{x}} = (0,0)^{T}$ modulo $\sqrt{2\pi}$. As a result, we can conclude that the random shift $\boldsymbol{\xi}$ is given by 
\begin{align}
\boldsymbol{\xi} = \boldsymbol{S}^{-1}\boldsymbol{z} - \sqrt{2\pi}\boldsymbol{S}^{-1}\boldsymbol{n}, 
\end{align}
for some $\boldsymbol{n} = (n_{q},n_{p})^{T} \in \mathbb{Z}^{2}$. In the case of Gaussian random shift errors, smaller shifts are more likely to occur than larger shifts. Thus, we infer that the random noise $\boldsymbol{\xi}$ is the one that has the smallest length among all possible error candidates that are compatible with the stabilizer measurement outcomes. That is, we infer that the random noise is 
\begin{align}
\boldsymbol{\bar{\xi}} = \boldsymbol{S}^{-1}\boldsymbol{z} - \sqrt{2\pi}\boldsymbol{S}^{-1}\boldsymbol{n}^{\star}( \boldsymbol{z} ), 
\end{align}
where $\boldsymbol{n}^{\star}( \boldsymbol{z} )$ is defined as 
\begin{align}
\boldsymbol{n}^{\star}( \boldsymbol{z} ) &\equiv \textrm{argmin}_{\boldsymbol{n}  \in \mathbb{Z}^{2} } | \boldsymbol{S}^{-1}\boldsymbol{z} - \sqrt{2\pi}\boldsymbol{S}^{-1}\boldsymbol{n} |. \label{eq:closest vector decoding}
\end{align}

Note that the optimization in Eq.\ \eqref{eq:closest vector decoding} is equivalent to finding a lattice point characterized by $\boldsymbol{n}^{\star}( \boldsymbol{z} )$ that is closest to the given vector $\boldsymbol{S}^{-1}\boldsymbol{z} $ in a $2$-dimensional lattice generated by the generator matrix $\sqrt{2\pi}\boldsymbol{S}^{-1}$, i.e., the closest vector problem. For this reason, it is useful to consider the notion of the Voronoi cell associated with the lattice generated by $\boldsymbol{L}$ and a lattice point $\boldsymbol{n}$: 
\begin{align}
\textrm{Vor}[ \boldsymbol{L} ] ( \boldsymbol{n} ) \equiv \lbrace \boldsymbol{r} \in \mathbb{R}^{\textrm{dim}(\boldsymbol{L})}  : |\boldsymbol{r} - \boldsymbol{L}\boldsymbol{n} | \le |\boldsymbol{r} - \boldsymbol{L}\boldsymbol{m}| \textrm{ for all }\boldsymbol{m} \in \mathbb{Z}^{\textrm{dim}(\boldsymbol{L})} \rbrace  . 
\end{align}
Thus, the Voronoi cell $\textrm{Vor}[ \boldsymbol{L} ] ( \boldsymbol{n} )$ is the set of all vectors whose closest lattice point is characterized by $\boldsymbol{n}$. Let us assume that the true random noise $\boldsymbol{\xi}$ is in the Voronoi cell $\textrm{Vor}[ \sqrt{2\pi} \boldsymbol{S}^{-1} ] ( \boldsymbol{n} )$, i.e., 
\begin{align}
\boldsymbol{\xi} \in \textrm{Vor}[ \sqrt{2\pi} \boldsymbol{S}^{-1} ] ( \boldsymbol{n} ). 
\end{align}
In this case, from the maximum likelihood decoding (or the closest vector decoding), we will estimate that the random noise is given by $\boldsymbol{\xi}_{\textrm{est}} = \boldsymbol{\xi} - \sqrt{2\pi} \boldsymbol{S}^{-1} \boldsymbol{n}$. Thus, it means that if the true noise $\boldsymbol{\xi}$ is certainly correctable if it is in the Voronoi cell associated with the origin, i.e., 
\begin{align}
\boldsymbol{\xi}_{\textrm{est}} = \boldsymbol{\xi} \,\,\, \textrm{for any}\,\,\, \boldsymbol{\xi} \in \textrm{Vor}[ \sqrt{2\pi} \boldsymbol{S}^{-1} ] ( \boldsymbol{0} ). 
\end{align}
Thus, we can optimize the design of the single-mode GKP code by choosing an optimal lattice $\boldsymbol{S}$ such that the probability that the noise $\boldsymbol{\xi}$ lies in the Voronoi cell $\textrm{Vor}[ \sqrt{2\pi} \boldsymbol{S}^{-1} ] ( \boldsymbol{0} )$ is maximized, i.e., 
\begin{align}
\boldsymbol{S}^{\star} = \textrm{argmax}_{ \boldsymbol{S}\in \mathbb{R}^{2\times 2}: \textrm{det}( \boldsymbol{S} ) = d }   \int_{\boldsymbol{\xi} \in \textrm{Vor}[ \sqrt{2\pi} \boldsymbol{S}^{-1} ] ( \boldsymbol{0} ) }  d^{2}\boldsymbol{\xi} \frac{1}{2\pi\sigma^{2}} \exp\Big{[} -\frac{ |\boldsymbol{\xi}|^{2} }{2\sigma^{2}} \Big{]}, \label{eq:optimal lattice single mode GKP code}
\end{align} 
Note that we assumed the Gaussian random shift error $\mathcal{N}_{B_{2}}[\sigma]$ to derive Eq.\ \eqref{eq:optimal lattice single mode GKP code}.  

In the case of the square-lattice GKP code encoding a qubit into an oscillator (i.e., $\boldsymbol{S}^{(\textrm{sq})} = \sqrt{2}\boldsymbol{I}_{2}$ and $d=2$), the lattice generator matrix is given by $\sqrt{2\pi} ( \boldsymbol{S}^{(\textrm{sq})} )^{-1} = \sqrt{\pi}\boldsymbol{I}_{2}$. Thus, the Voronoi cell assoicated with the origin is given by 
\begin{align}
\textrm{Vor}[ \sqrt{2\pi} ( \boldsymbol{S}^{(\textrm{sq})} )^{-1} ] (\boldsymbol{0}) = \textrm{Vor}[ \sqrt{\pi}\boldsymbol{I}_{2} ] (\boldsymbol{0}) = \Big{\lbrace} \boldsymbol{\xi} = (\xi_{q},\xi_{p}) : |\xi_{q}|,|\xi_{p}|< \frac{\sqrt{\pi}}{2} \Big{\rbrace} . 
\end{align} 
This is consistent with the fact that any small shift errors contained in the square $|\xi_{q}|,|\xi_{p}|< \frac{\sqrt{\pi}}{2}$ can be corrected for the square-lattice GKP code. Note also that $\frac{\sqrt{\pi}}{2}$ is the maximum radius of a circle that can be contained within the Voronoi cell $\textrm{Vor}[ \sqrt{2\pi} ( \boldsymbol{S}^{(\textrm{sq})} )^{-1} ] (\boldsymbol{0})$. In other words, the maximum radius of the correctable shift for the square-lattice GKP code is given by
\begin{align}
r_{\textrm{c}}^{( \textrm{sq} )} = \frac{\sqrt{\pi}}{2}. 
\end{align}

\subsubsection{The single-mode hexagonal-lattice GKP code}

Recall that in the case of the square-lattice GKP code, the logical X and Z error rates are the same but the logical Y error rate is much smaller. This is because the logical X and Z errors occur due to large shifts in either one of the position and the momentum directions, i.e., 
\begin{align}
&|\xi_{q}| < \frac{\sqrt{\pi}}{2}  \,\,\, \textrm{and}\,\,\,   \frac{\sqrt{\pi}}{2} < |\xi_{p}| < \frac{\sqrt{3\pi}}{2}  \rightarrow \textrm{Logical Z error}, 
\nonumber\\
&  \frac{\sqrt{\pi}}{2} < |\xi_{q}| < \frac{3\sqrt{\pi}}{2}  \,\,\, \textrm{and}\,\,\,   |\xi_{p}| < \frac{\sqrt{\pi}}{2}  \rightarrow \textrm{Logical X error}.  
\end{align}  
On the other hand, the logical Y error happens only both the position and the momentum shifts are large and thus is much less likely to occur than the logical X and Z errors (see also Fig.\ \ref{fig:Voronoi cell square verses hexagonal}(a) for an illustration): 
\begin{align}
&  \frac{\sqrt{\pi}}{2} < |\xi_{q}| < \frac{3\sqrt{\pi}}{2}  \,\,\, \textrm{and}\,\,\,  \frac{\sqrt{\pi}}{2} < |\xi_{p}| < \frac{3\sqrt{\pi}}{2}  \rightarrow \textrm{Logical Y error}.  
\end{align}
This indicates that the square-lattice GKP code is not in fact using the allowed area in the phase space in the most efficient way. In other words, it would have been better if we could sacrifice the extremely low Y error rate and instead improve both the X and Z error rates, which are the weakest links in the entire scheme. Below, we will show that the hexagonal-lattice structure allows us to realize this idea.    

The stabilizers of the hexagonal-lattice GKP state encoding a qubit into an oscillator (i.e., $d=2$) are given by 
\begin{align}
\hat{S}_{q}^{(\textrm{hex})} &= \exp\Big{[} i 2 \sqrt{\pi} \Big{(} \frac{2}{\sqrt{3}} \Big{)}^{\frac{1}{2}} \hat{q} \Big{]} , 
\nonumber\\
\hat{S}_{p}^{(\textrm{hex})} &= \exp\Big{[} -i 2 \sqrt{\pi} \Big{(} \frac{2}{\sqrt{3}} \Big{)}^{\frac{1}{2}}  \Big{(} \frac{1}{2} \hat{q} + \frac{\sqrt{3}}{2} \hat{p} \Big{)} \Big{]}.  
\end{align}  
Hence, the matrix $\boldsymbol{S}^{( \textrm{hex} )}$ associated with these stabilizers is given by 
\begin{align}
\boldsymbol{S}^{( \textrm{hex} )} &= \sqrt{2} \Big{(} \frac{2}{\sqrt{3}} \Big{)}^{\frac{1}{2}}  \begin{bmatrix}
1 & 0 \\
\frac{1}{2} & \frac{\sqrt{3}}{2}
\end{bmatrix} . 
\end{align} 
Note also that $\textrm{det}( \boldsymbol{S}^{( \textrm{hex} )} ) = 2$ and thus the dimension of the code space is indeed two. Following the general construction provided above, the logical Pauli operators of the hexagonal-lattice GKP code are given by 
\begin{align}
\hat{Z}_{\textrm{gkp}}^{ (\textrm{hex}) } &= ( \hat{S}_{q}^{(\textrm{hex})} )^{\frac{1}{2}} = \exp\Big{[} i  \sqrt{\pi} \Big{(} \frac{2}{\sqrt{3}} \Big{)}^{\frac{1}{2}} \hat{q} \Big{]} , 
\nonumber\\
\hat{X}_{\textrm{gkp}}^{ (\textrm{hex}) } &= (\hat{S}_{p}^{(\textrm{hex})})^{\frac{1}{2}} = \exp\Big{[} -i \sqrt{\pi} \Big{(} \frac{2}{\sqrt{3}} \Big{)}^{\frac{1}{2}}  \Big{(} \frac{1}{2} \hat{q} + \frac{\sqrt{3}}{2} \hat{p} \Big{)} \Big{]}.  
\end{align}
For the hexagonal-lattice GKP code, the relevant lattice generator $\sqrt{2\pi} ( \boldsymbol{S}^{(\textrm{hex})} )^{-1}$ for the performance analysis is given by  
\begin{align}
\sqrt{2\pi} ( \boldsymbol{S}^{(\textrm{hex})} )^{-1} &= \sqrt{\pi} \Big{(} \frac{2}{\sqrt{3}} \Big{)}^{\frac{1}{2}}  \begin{bmatrix}
\frac{\sqrt{3}}{2} & 0\\
-\frac{1}{2} & 1
\end{bmatrix} ,  
\end{align}
which generates another hexagonal lattice.

\begin{figure}[t!]
\centering
\includegraphics[width=5.0in]{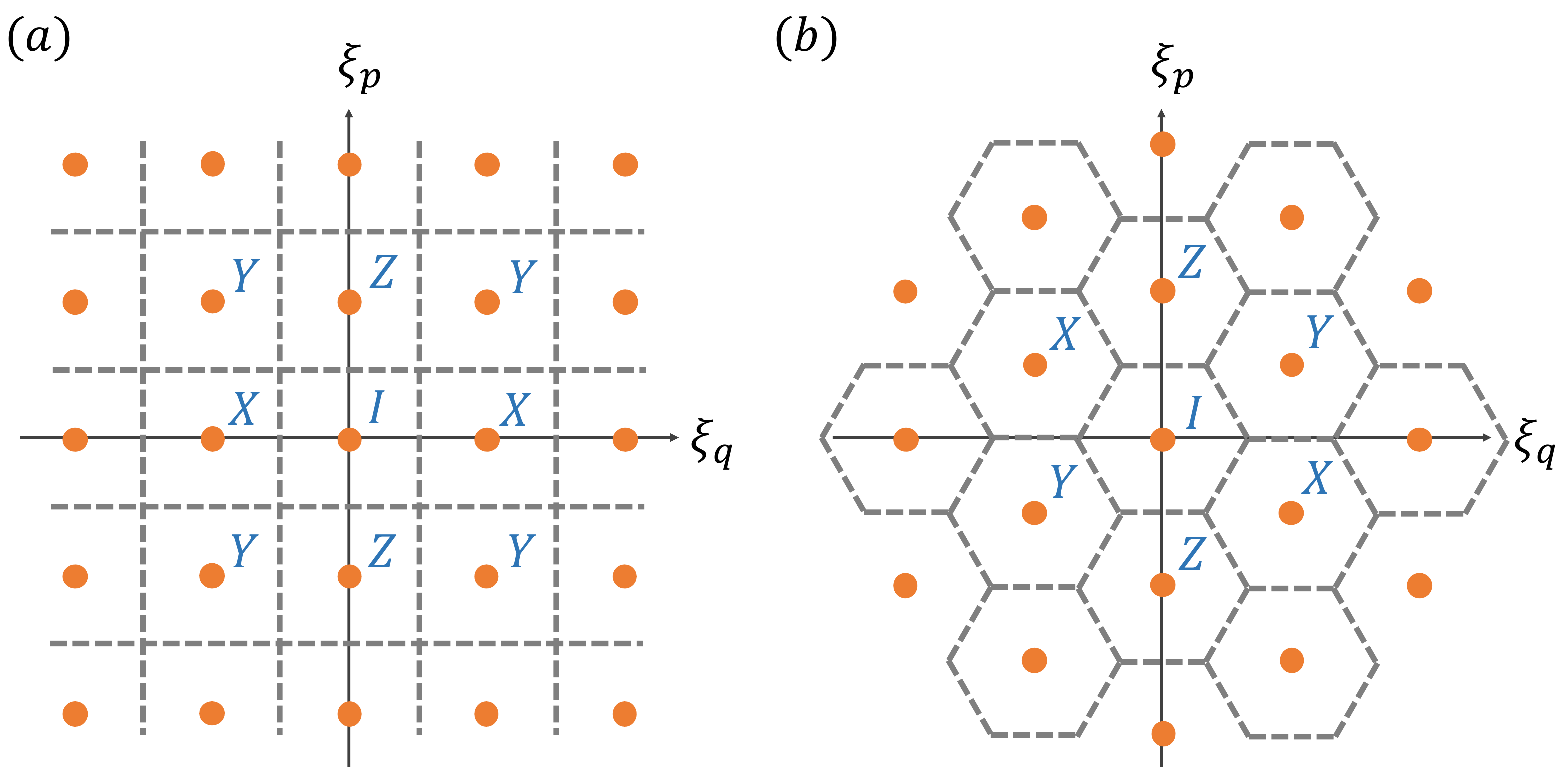}
\caption{Relevant Voronoi cells for the error analysis of (a) the square-lattice GKP code and (b) the hexagonal-lattice GKP code.}
\label{fig:Voronoi cell square verses hexagonal}
\end{figure}

In Fig.\ \ref{fig:Voronoi cell square verses hexagonal} (b), we visualized the Voronoi cells associated with the lattice $\sqrt{2\pi} ( \boldsymbol{S}^{(\textrm{hex})} )^{-1}$. Note that each Voronoi cell is given by a hexagon. Therefore, although each Voronoi cell has the same area $\sqrt{\pi}$ both in the case of the square-lattice GKP code and the hexagonal-lattice GKP code, the latter can contain a circle with a larger radius. In particular, the maximum radius of the circle that can be contained within the correctable Voronoi cell $\textrm{Vor}[ \sqrt{2\pi} ( \boldsymbol{S}^{(\textrm{hex})} )^{-1} ](\boldsymbol{0})$ is given by $\frac{\sqrt{\pi}}{2} ( \frac{2}{\sqrt{3}} )^{\frac{1}{2}} $. Thus, the maximum radius of the correctable shift for the hexagonal-lattice GKP code is given by 
\begin{align}
r_{\textrm{c}}^{ (\textrm{hex}) } = \frac{\sqrt{\pi}}{2} \Big{(} \frac{2}{\sqrt{3}} \Big{)}^{\frac{1}{2}} \simeq 1.07 r_{\textrm{c}}^{ (\textrm{sq}) } , 
\end{align}
and is about $1.07$ times larger than that of the square-lattice GKP code. This means that the hexagonal-lattice GKP code uses the allowed area in the phase space in a more efficient way than the square-lattice GKP code. Note also that in the case of the hexagonal-lattice GKP code, the logical X, Y, and Z error rates are all identical to each other, as desired. We also remark that the hexagonal lattice allows the densest sphere packing in the $2$-dimensional Euclidean space \cite{Fejes1942}.   

Lastly, we estimate the failure probability of the single-mode GKP code assuming the Gaussian random shift error $\mathcal{N}_{B_{2}}[ \sigma ]$ and the maximum correctable radius $r_{\textrm{c}}$ of the code. The success probability of the GKP error correction is lower bounded by
\begin{align}
p_{\textrm{succ}}(\sigma) &\ge \int_{ |\boldsymbol{\xi}|  < r_{\textrm{c}} } d^{2}\boldsymbol{\xi} \frac{1}{2\pi \sigma^{2}} \exp\Big{[} -\frac{|\boldsymbol{\xi}|^{2}}{2\sigma^{2}} \Big{]} 
\nonumber\\
&= \int_{0}^{r_{\textrm{c}}} 2rdr \frac{1}{2\sigma^{2}} \exp\Big{[} -\frac{r^{2}}{2\sigma^{2}} \Big{]} = \int_{0}^{r_{\textrm{c}}^{2}} dx \frac{1}{2\sigma^{2}} \exp\Big{[} -\frac{x}{2\sigma^{2}} \Big{]} = 1 - \exp\Big{[} -\frac{r_{\textrm{c}}^{2}}{2\sigma^{2}} \Big{]} . 
\end{align}
The failure probability is then upper bounded by 
\begin{align}
p_{\textrm{fail}}(\sigma) = 1-p_{\textrm{succ}}(\sigma) \le \exp\Big{[} -\frac{r_{\textrm{c}}^{2}}{2\sigma^{2}} \Big{]}. \label{eq:upper bound of failure probability against random shift single mode general}
\end{align}
Specializing Eq.\ \eqref{eq:upper bound of failure probability against random shift single mode general} to the cases of the square-lattice GKP code ($r_{\textrm{c}}^{(\textrm{sq})} = \frac{\sqrt{\pi}}{2}$) and the hexagonal-lattice GKP code ($r_{\textrm{c}}^{(\textrm{hex})} = \frac{\sqrt{\pi}}{2} ( \frac{2}{\sqrt{3}} )^{\frac{1}{2}}$), we find
\begin{align}
p_{\textrm{fail}}^{(\textrm{sq})}(\sigma) &\le \exp\Big{[} -\frac{(r_{\textrm{c}}^{(\textrm{sq})})^{2}}{2\sigma^{2}} \Big{]}  = \exp\Big{[} -\frac{\pi}{8\sigma^{2}} \Big{]}  , 
\nonumber\\
p_{\textrm{fail}}^{(\textrm{hex})}(\sigma) &\le \exp\Big{[} -\frac{(r_{\textrm{c}}^{(\textrm{hex})})^{2}}{2\sigma^{2}} \Big{]}  = \exp\Big{[} -\frac{\pi}{4\sqrt{3}\sigma^{2}} \Big{]}  . 
\end{align} 
We also remark if one encodes $d$ logical states  where $d$ is an integer such that $d\ge 2$, the above bounds are modified and we have 
\begin{align}
p_{\textrm{fail}}^{(\textrm{sq})}(\sigma;d) &\le  \exp\Big{[} -\frac{\pi}{4d \cdot \sigma^{2}} \Big{]}  , 
\nonumber\\
p_{\textrm{fail}}^{(\textrm{hex})}(\sigma;d) &\le \exp\Big{[} -\frac{\pi}{2\sqrt{3}d \cdot \sigma^{2}} \Big{]}  . \label{eq:performance of GKP codes against a Gaussian random shift error}
\end{align}

\subsubsection{Generalized multi-mode GKP codes}

The generalization of the single-mode GKP code to any lattice generators $\boldsymbol{S}$ (such that $\textrm{det}( \boldsymbol{S} ) \in \mathbb{Z}$) already exhibits the flexibility of the GKP code. One simple way to generalize the single-mode GKP codes to the multi-mode cases is to concatenate the single-mode codes with a conventional multi-qubit error-correcting code, such as the $[[4,1,2]]$ code, the Steane code (or the $[[7,1,3]]$ code), and the surface code, and so on. The concatenation method is certainly practical and we will discuss this in more detail in Chapter \ref{chapter:Fault-tolerant bosonic quantum error correction}. However, concatenation is not the most general approach and indeed it is possible to define a multi-mode GKP code that cannot be decomposed into a single-mode GKP code and a multi-qubit error-correcting code. Here, we review such a most general construction of the multi-mode GKP codes \cite{Gottesman2001,Harrington2001,Harrington2004}. 

In the general $N$-mode case, a multi-mode GKP code $\mathcal{C}_{\textrm{gkp}}^{ (\boldsymbol{S}) }$ is stabilized by $2N$ stabilizers:
\begin{align}
\hat{S}_{j}^{ (\boldsymbol{S}) } &= \exp\Big{[} i\sqrt{2\pi} (-1)^{ I( j>N ) } \sum_{k=1}^{N} \boldsymbol{S}_{jk} \boldsymbol{\hat{x}} _{k} \Big{]}, \,\,\, \textrm{where}\,\,\, j \in \lbrace 1,\cdots, 2N \rbrace. 
\end{align}
Here, $\boldsymbol{S}$ is a $2N\times 2N$ matrix and $\boldsymbol{\hat{x}} = (\hat{q}_{1}, \cdots, \hat{q}_{N}, \hat{p}_{1}, \cdots \hat{p}_{N})^{T}$ is a vector that consists of the $2N$ quadrature operators. $I(C)$ is an indicator function that is given by $1$ if $C$ is true and $0$ otherwise. Note that the quadrature operators satisfy the following commutation relation:
\begin{align}
[ \boldsymbol{\hat{x}}_{j} , \boldsymbol{\hat{x}}_{k}  ] &= i\boldsymbol{\Omega}_{jk},
\end{align}  
where $\boldsymbol{\Omega}$ is a $2N\times 2N$ matrix defined as 
\begin{align}
\boldsymbol{\Omega} \equiv \begin{bmatrix}
0 & \boldsymbol{I}_{N}  \\
-\boldsymbol{I}_{N} & 0 
\end{bmatrix} . 
\end{align}
Note that the convention for $\boldsymbol{\hat{x}}$ and $\boldsymbol{\Omega} $ used here is not the same as the convention used in Appendix \ref{appendix:Gaussian states, unitaries, and channels} and throughout the rest of the thesis. We use a different convention here because it is particularly suited for analyzing the multi-mode GKP code $\mathcal{C}_{\textrm{gkp}}^{ (\boldsymbol{S}) }$. 

By using the BCH formula (see Eq.\ \eqref{eq:BCH formula}), we find 
\begin{align}
\hat{S}_{j}^{ (\boldsymbol{S}) }\hat{S}_{k}^{ (\boldsymbol{S}) } &= \hat{S}_{k}^{ (\boldsymbol{S}) }\hat{S}_{j}^{ (\boldsymbol{S}) } \exp\Big{[}  [ i\sqrt{2\pi} (-1)^{ I(j>N) } \sum_{l=1}^{N} \boldsymbol{S}_{jl} \boldsymbol{\hat{x}} _{l} , i\sqrt{2\pi} (-1)^{ I(k>N) } \sum_{m=1}^{N} \boldsymbol{S}_{km} \boldsymbol{\hat{x}} _{m} ] \Big{]}  
\nonumber\\
&= \hat{S}_{k}^{ (\boldsymbol{S}) }\hat{S}_{j}^{ (\boldsymbol{S}) } \exp\Big{[} -2\pi (-1)^{ I(j>N) + I(k>N) }   \sum_{l,m=1}^{N} \boldsymbol{S}_{jl} [ \boldsymbol{\hat{x}} _{l} ,   \boldsymbol{\hat{x}} _{m} ]  \boldsymbol{S}_{km} \Big{]}  
\nonumber\\
&= \hat{S}_{k}^{ (\boldsymbol{S}) }\hat{S}_{j}^{ (\boldsymbol{S}) } \exp\Big{[} -2\pi i (-1)^{  I(j>N) + I(k>N) }   \sum_{l,m=1}^{N} \boldsymbol{S}_{jl} \boldsymbol{\Omega}_{lm} \boldsymbol{S}_{km} \Big{]}  
\nonumber\\
&= \hat{S}_{k}^{ (\boldsymbol{S}) }\hat{S}_{j}^{ (\boldsymbol{S}) } \exp\Big{[} -2\pi i (-1)^{  I(j>N) + I(k>N) }  ( \boldsymbol{S} \boldsymbol{\Omega} \boldsymbol{S}^{T} )_{jk}  \Big{]}   . 
\end{align} 
Therefore, for the stabilizers $\hat{S}_{j}^{ (\boldsymbol{S}) }$ and $\hat{S}_{k}^{ (\boldsymbol{S}) }$ to commute with each other for all $j,k\in\lbrace 1,\cdots, 2N \rbrace$, all the matrix elements of the $2N\times 2N$ matrix $\boldsymbol{S} \boldsymbol{\Omega} \boldsymbol{S}^{T}$ should be an integer, i.e.,
\begin{align}
\boldsymbol{S} \boldsymbol{\Omega} \boldsymbol{S}^{T} \in \mathbb{Z}^{2N\times 2N} . \label{eq:condition on the matrix multi-mode GKP}
\end{align}  
Note that the matrix $\boldsymbol{S} \boldsymbol{\Omega} \boldsymbol{S}^{T}$ is anti-symmetric since $\boldsymbol{\Omega}$ is an anti-symmetric matrix. In the single-mode case (i.e., $N=1$), the condition in Eq.\ \eqref{eq:condition on the matrix multi-mode GKP} reduces to 
\begin{align}
\boldsymbol{S} \boldsymbol{\Omega} \boldsymbol{S}^{T} &= \begin{bmatrix}
S_{11} & S_{12} \\
S_{21} & S_{22}
\end{bmatrix} \begin{bmatrix}
0 & 1 \\
-1 & 0
\end{bmatrix}\begin{bmatrix}
S_{11} & S_{21} \\
S_{12} & S_{22}
\end{bmatrix} = \begin{bmatrix}
0 & \textrm{det}(\boldsymbol{S}) \\
-\textrm{det}(\boldsymbol{S}) & 0 
\end{bmatrix} \in \mathbb{Z}^{2\times 2}, 
\end{align}
and thus is equivalent to the condition $\textrm{det}(\boldsymbol{S}) \in \mathbb{Z}$, as we discussed above (see Eq.\ \eqref{eq:condition on the matrix single-mode GKP}). However in the general multi-mode case with $N \ge 2$,  the condition in Eq.\ \eqref{eq:condition on the matrix multi-mode GKP} is not equivalent to the condition $\textrm{det}(\boldsymbol{S}) \in \mathbb{Z}$. 

As shown in Refs.\ \cite{Gottesman2001,Harrington2001,Harrington2004}, given that the condition in Eq.\ \eqref{eq:condition on the matrix multi-mode GKP} is fulfilled, one can always assume without loss of generality the matrix $\boldsymbol{S}$ is in the standard form such that  
\begin{align}
\boldsymbol{A} \equiv \boldsymbol{S} \boldsymbol{\Omega} \boldsymbol{S}^{T}  = \begin{bmatrix}
0 & \boldsymbol{D} \\
-\boldsymbol{D} & 0 
\end{bmatrix} \,\,\, \textrm{and} \,\,\, \boldsymbol{D} = \textrm{diag}(d_{1},\cdots, d_{N}), 
\end{align}
where $\boldsymbol{D}$ is an $N\times N$ diagonal matrix whose diagonal elements are given by a natural number. It will turn out below that the natural numbers $d_{1},\cdots, d_{N}$ are very closely related to the number of logical states encoded in the multi-mode GKP code $\mathcal{C}_{\textrm{gkp}}^{ (\boldsymbol{S}) }$. 

Let us first consider an important special case with $d_{j} = 1$ for all $j\in \lbrace 1,\cdots, N \rbrace$, i.e., $\boldsymbol{D} = \boldsymbol{I}_{N}$. In this case, the matrix $\boldsymbol{S}$ is a $2N\times 2N$ symplectic matrix as it satisfies 
\begin{align}
\boldsymbol{S} \boldsymbol{\Omega} \boldsymbol{S}^{T} &= \boldsymbol{\Omega} . 
\end{align}   
The simplest case is when $\boldsymbol{S} $ is given by an identity matrix $\boldsymbol{S} = \boldsymbol{I}_{2N}$. In this case, the $2N$ stabilizes are given by
\begin{align}
\hat{S}_{j}^{ (\boldsymbol{I}_{2N}) } &= e^{i\sqrt{2\pi}\hat{q}_{j}}, 
\nonumber\\
\hat{S}_{N+j}^{ (\boldsymbol{I}_{2N}) } &= e^{-i\sqrt{2\pi}\hat{p}_{j}}, \label{eq:stabilizers of the canonical N mode GKP state}
\end{align} 
where $j \in \lbrace 1,\cdots, N \rbrace$. In this case, following the same reasoning used for the single-mode case, one can show that the tensor product of the canonical GKP states 
\begin{align}
|\textrm{GKP}\rangle^{\otimes N}
\end{align} 
is the unique state (up to an overall phase and normalization) that is stabilized by the $2N$ stabilizers in Eq.\ \eqref{eq:stabilizers of the canonical N mode GKP state}. 

For a general $2N\times 2N$ symplectic matrix $\boldsymbol{S}$, one can define a Gaussian operation $\hat{U}_{\boldsymbol{S}^{-1}} = \hat{U}_{\boldsymbol{S}}^{\dagger}$ that transforms the quadrature operators as follows: 
\begin{align}
\hat{U}_{\boldsymbol{S}^{-1}}^{\dagger} \boldsymbol{\hat{x}} \hat{U}_{\boldsymbol{S}^{-1}} = \boldsymbol{S}^{-1}\boldsymbol{\hat{x}}. 
\end{align}   
Note that this is a valid transformation precisely because the matrix $\boldsymbol{S}$ and $\boldsymbol{S}^{-1}$ are symplectic (see Appendix \ref{appendix:Gaussian states, unitaries, and channels}). Then, similarly as in the single-mode case, one can see that the state 
\begin{align}
|\textrm{GKP}_{\boldsymbol{S}} \rangle &\equiv \hat{U}_{\boldsymbol{S}^{-1}}  |\textrm{GKP}\rangle^{\otimes N}
\end{align}   
is the unique state (up to an overall phase and normalization) that is stabilized by the stabilizers $\hat{S}_{j}^{ (\boldsymbol{S}) }$ for all $j\in\lbrace 1,\cdots, 2N \rbrace$. Because of the uniqueness, the GKP code $\mathcal{C}_{\textrm{gkp}}^{ (\boldsymbol{S}) }$ is not a very interesting quantum error-correcting code if $\boldsymbol{S}$ is a symplectic matrix, since it encodes only one logical state. However, all these basic facts will be useful for understanding more interesting cases with $(d_{1},\cdots, d_{N}) \neq (1,\cdots, 1)$.  

Let us now move on to the most general case with $\boldsymbol{D} = \textrm{diag}(d_{1},\cdots, d_{N}) \neq \boldsymbol{I}_{N}$. In this case, the matrix $\boldsymbol{S}$ is not symplectic. To understand the structure of the code space $\mathcal{C}_{\textrm{gkp}}^{ (\boldsymbol{S}) }$ in a more fine-grained way, we need to understand the logical operators of the code. To do so, let us first consider the following matrix 
\begin{align}
\boldsymbol{S}^{\perp} &\equiv \boldsymbol{A}^{-1} \boldsymbol{S} = \begin{bmatrix}
0 & -\boldsymbol{D}^{-1} \\
\boldsymbol{D}^{-1} & 0 
\end{bmatrix} \boldsymbol{S} , 
\end{align}   
and the associated $2N$ operators
\begin{align}
\hat{L}_{j}^{(\boldsymbol{S})} \equiv \exp\Big{[} i\sqrt{2\pi} \sum_{k=1}^{N}\boldsymbol{S}^{\perp}_{jk} \boldsymbol{\hat{x}}_{k} \Big{]}, \,\,\, \textrm{where}\,\,\, j\in\lbrace 1,\cdots, 2N \rbrace . 
\end{align}
We will show that these operators are the logical operators of the multi-mode GKP code $\mathcal{C}_{\textrm{gkp}}^{ (\boldsymbol{S}) }$. More specifically, we will show that 
\begin{align}
\hat{X}_{j}^{(\boldsymbol{S})} &\equiv \hat{L}_{j}^{(\boldsymbol{S})}, 
\nonumber\\
\hat{Z}_{j}^{(\boldsymbol{S})} &\equiv \hat{L}_{N+j}^{(\boldsymbol{S})}, 
\end{align} 
(where $j\in\lbrace 1,\cdots, N \rbrace$) act in the same way as the Pauli X and Z operators on the code space. Note that all these operators commute with all the stabilizers, i.e., 
\begin{align}
\hat{L}_{j}^{(\boldsymbol{S})} \hat{S}_{k}^{(\boldsymbol{S})}  &=  \hat{S}_{k}^{(\boldsymbol{S})}\hat{L}_{j}^{(\boldsymbol{S})} \exp \Big{[} -2\pi i (-1)^{k} ( \boldsymbol{S}^{\perp} \boldsymbol{\Omega} \boldsymbol{S}^{T} )_{jk} \Big{]} 
\nonumber\\
&=  \hat{S}_{k}^{(\boldsymbol{S})}\hat{L}_{j}^{(\boldsymbol{S})} \exp \Big{[} -2\pi i (-1)^{ I(k >N) } ( \boldsymbol{A}^{-1} \boldsymbol{S} \boldsymbol{\Omega} \boldsymbol{S}^{T} )_{jk} \Big{]} 
\nonumber\\
&=  \hat{S}_{k}^{(\boldsymbol{S})}\hat{L}_{j}^{(\boldsymbol{S})} \exp \Big{[} -2\pi i (-1)^{ I(k >N) } ( \boldsymbol{A}^{-1} \boldsymbol{A} )_{jk} \Big{]} 
\nonumber\\
&=  \hat{S}_{k}^{(\boldsymbol{S})}\hat{L}_{j}^{(\boldsymbol{S})} \exp \Big{[} -2\pi i (-1)^{ I(k >N) } \delta_{jk} \Big{]} 
\nonumber\\
&=  \hat{S}_{k}^{(\boldsymbol{S})}\hat{L}_{j}^{(\boldsymbol{S})} . 
\end{align}
Thus, it makes sense to consider a state $|\boldsymbol{0}_{\textrm{gkp}}^{ ( \boldsymbol{S} )}\rangle$ (where $\boldsymbol{0} \equiv (0,\cdots, 0)$ is a zero vector with $N$ zeros) that is simultaneously stabilized by the stabilizers $\hat{S}_{N+1}^{ (\boldsymbol{S}) }, \cdots, \hat{S}_{2N}^{ (\boldsymbol{S}) }$ and the operators $\hat{L}_{N+1}^{(\boldsymbol{S})},\cdots, \hat{L}_{2N}^{(\boldsymbol{S})}$. By inspection, one can realize that $|\boldsymbol{0}_{\textrm{gkp}}^{ ( \boldsymbol{S} )}\rangle$ can be regarded as a state stabilized by the stabilizers associated with the matrix 
\begin{align}
\boldsymbol{S}' \equiv \begin{bmatrix}
\boldsymbol{D}^{-1} & 0 \\
0 & \boldsymbol{I}_{N} \\
\end{bmatrix}  \boldsymbol{S}. 
\end{align} 
Note that the matrix $\boldsymbol{S}'$ is symplectic as it satisfies
\begin{align}
\boldsymbol{S}' \boldsymbol{\Omega} \boldsymbol{S}'^{T} = \begin{bmatrix}
\boldsymbol{D}^{-1} & 0 \\
0 & \boldsymbol{I}_{N}  
\end{bmatrix}  \begin{bmatrix}
0 & \boldsymbol{D} \\
-\boldsymbol{D} & 0 
\end{bmatrix}  \begin{bmatrix}
\boldsymbol{D}^{-1} & 0 \\
0 & \boldsymbol{I}_{N} 
\end{bmatrix} = \begin{bmatrix}
0 & \boldsymbol{I}_{N} \\
-\boldsymbol{I}_{N} & 0 
\end{bmatrix} = \boldsymbol{\Omega}. 
\end{align}
Thus, the stabilized state $|\boldsymbol{0}_{\textrm{gkp}}^{ ( \boldsymbol{S} )}\rangle$ is unique up to an overall phase and normalization. We call this unique state the computational zero state.   

To get the other logical states, observe that the operators $\hat{L}_{j}^{ (\boldsymbol{S}) }$ do not commute with each other, while they commute with all the stabilizers, i.e., 
\begin{align}
\hat{L}_{j}^{(\boldsymbol{S})} \hat{L}_{k}^{(\boldsymbol{S})}  &=  \hat{L}_{k}^{(\boldsymbol{S})}\hat{L}_{j}^{(\boldsymbol{S})} \exp \Big{[} -2\pi i  ( \boldsymbol{S}^{\perp} \boldsymbol{\Omega} (\boldsymbol{S}^{\perp} )^{T} )_{jk} \Big{]} 
\nonumber\\
&=  \hat{L}_{k}^{(\boldsymbol{S})}\hat{L}_{j}^{(\boldsymbol{S})} \exp \Big{[} -2\pi i  ( \boldsymbol{A}^{-1} \boldsymbol{S} \boldsymbol{\Omega} \boldsymbol{S}^{T} ( \boldsymbol{A}^{-1} )^{T} )_{jk} \Big{]} 
\nonumber\\
&=  \hat{L}_{k}^{(\boldsymbol{S})}\hat{L}_{j}^{(\boldsymbol{S})} \exp \Big{[} -2\pi i  ( \boldsymbol{A}^{-1} \boldsymbol{A} ( \boldsymbol{A}^{-1} )^{T} )_{jk} \Big{]}  
\nonumber\\
&= \hat{L}_{k}^{(\boldsymbol{S})}\hat{L}_{j}^{(\boldsymbol{S})} \exp \Big{[} -2\pi i  (  ( \boldsymbol{A}^{-1} )^{T} )_{jk} \Big{]}  . \label{eq:multi mode GKP code logical operators commutation relation less explicit}
\end{align}
Since $\boldsymbol{A}^{-1}$ is given by 
\begin{align}
 \boldsymbol{A}^{-1}  = \begin{bmatrix}
0 & \boldsymbol{D} \\
-\boldsymbol{D} & 0 
\end{bmatrix}^{-1} = \begin{bmatrix}
0 & -\boldsymbol{D}^{-1} \\
\boldsymbol{D}^{-1} & 0 
\end{bmatrix}, 
\end{align} 
Eq.\ \eqref{eq:multi mode GKP code logical operators commutation relation less explicit} is explicitly given by 
\begin{align}
\hat{Z}_{j}^{(\boldsymbol{S})} \hat{Z}_{k}^{(\boldsymbol{S})} &= \hat{Z}_{k}^{(\boldsymbol{S})} \hat{Z}_{j}^{(\boldsymbol{S})}, 
\nonumber\\
\hat{X}_{j}^{(\boldsymbol{S})} \hat{X}_{k}^{(\boldsymbol{S})} &= \hat{X}_{k}^{(\boldsymbol{S})} \hat{X}_{j}^{(\boldsymbol{S})}, 
\nonumber\\
\hat{Z}_{j}^{(\boldsymbol{S})} \hat{X}_{k}^{(\boldsymbol{S})} &= \hat{X}_{k}^{(\boldsymbol{S})} \hat{Z}_{j}^{(\boldsymbol{S})}  \exp \Big{[} i\frac{2\pi}{d_{j}} \delta_{jk}\Big{]} , \label{eq:multi mode GKP commutation relation logical operators}
\end{align}
for all $j,k\in \lbrace 1,\cdots, N\rbrace$ where $\hat{X}_{j}^{(\boldsymbol{S})} \equiv \hat{L}_{j}^{(\boldsymbol{S})}$ and $\hat{Z}_{j}^{(\boldsymbol{S})} \equiv \hat{L}_{N+j}^{(\boldsymbol{S})}$. Thus, these operators behave exactly the same way as the Pauli operators. Thus, we refer to $\hat{X}_{j}^{(\boldsymbol{S})}$ and $\hat{Z}_{j}^{(\boldsymbol{S})}$ as the Pauli X and Z operators acting on the $j^{\textrm{th}}$ degree of freedom, respectively (where $j\in\lbrace 1,\cdots, N \rbrace$). 

With all the basic facts ready, we can now construct the other logical states as follows: 
\begin{align}
|\boldsymbol{\mu}_{\textrm{gkp}}^{ (\boldsymbol{S}) } \rangle  &\equiv \prod_{j=1}^{N} ( \hat{X}_{j}^{(\boldsymbol{S})} )^{\mu_{j}}  |\boldsymbol{0}_{\textrm{gkp}}^{ (\boldsymbol{S}) } \rangle, 
\end{align} 
where $\boldsymbol{\mu} = (\mu_{1},\cdots, \mu_{N}) \in \mathbb{Z}_{d_{1}}\times \cdots \times \mathbb{Z}_{d_{N}}$. Using the commutation relation in Eq.\ \eqref{eq:multi mode GKP commutation relation logical operators}, we can show that 
\begin{align}
\hat{Z}_{j}^{(\boldsymbol{S})} |\boldsymbol{\mu}_{\textrm{gkp}}^{ (\boldsymbol{S}) } \rangle &= \exp \Big{[} i\frac{2\pi}{d_{j}} \Big{]} \prod_{j=1}^{N} ( \hat{X}_{j}^{(\boldsymbol{S})} )^{\mu_{j}}  \hat{Z}_{j}^{(\boldsymbol{S})} |\boldsymbol{0}_{\textrm{gkp}}^{ (\boldsymbol{S}) } \rangle
\nonumber\\
&= \exp \Big{[} i\frac{2\pi}{d_{j}} \Big{]} \prod_{j=1}^{N} ( \hat{X}_{j}^{(\boldsymbol{S})} )^{\mu_{j}}   |\boldsymbol{0}_{\textrm{gkp}}^{ (\boldsymbol{S}) } \rangle 
\nonumber\\
&= \exp \Big{[} i\frac{2\pi}{d_{j}} \Big{]} |\boldsymbol{\mu}_{\textrm{gkp}}^{ (\boldsymbol{S}) } \rangle, 
\end{align}
as desired. Note also that the $j^{\textrm{th}}$ degree of freedom encodes $d_{j}$ logical states. Thus, the dimension of the code space $\mathcal{C}_{\textrm{gkp}}^{ (\boldsymbol{S}) }$ is given by
\begin{align}
\textrm{dim} ( \mathcal{C}_{\textrm{gkp}}^{ (\boldsymbol{S}) } ) = \prod_{j=1}^{N} d_{j}. 
\end{align}  

\subsubsection{Symplectic lattice codes} 

Recall that in the standard form, the matrix $\boldsymbol{S}$ satisfies 
\begin{align}
\boldsymbol{A} = \boldsymbol{S} \boldsymbol{\Omega} \boldsymbol{S}^{T} = \begin{bmatrix}
0 & \boldsymbol{D} \\
-\boldsymbol{D} & 0 
\end{bmatrix}, 
\end{align}
where $\boldsymbol{D}$ is a diagonal matrix whose entries are given by a natural number. Moreover, the diagonal entry $d_{j}$ characterize the number of logical states encoded in the $j^{\textrm{th}}$ degree of freedom. Here, we consider a special case where the diagonal matrix $\boldsymbol{D}$ is given by 
\begin{align}
\boldsymbol{D} = d  \boldsymbol{I}_{N}. 
\end{align}   
That is, we consider the case where all the $N$ degrees of freedom uniformly encode $d$ logical states. In this case, the total dimension of the code space is given by
\begin{align}
\textrm{dim} ( \mathcal{C}_{\textrm{gkp}}^{ (\boldsymbol{S}) } )  = d^{N}, 
\end{align}
which grows exponentially as we increase the number of modes $N$, for any $d\ge 2$. Note that in this case, the matrix $\boldsymbol{S}$ satisfies 
\begin{align}
\boldsymbol{S} \boldsymbol{\Omega} \boldsymbol{S}^{T} = d \boldsymbol{\Omega} . 
\end{align}  
Thus, the matrix $\boldsymbol{S}$ can be rescaled to a symplectic matrix
\begin{align}
\boldsymbol{\bar{S}} \equiv \frac{1}{\sqrt{d}} \boldsymbol{S}, 
\end{align}
which does satisfy $\boldsymbol{\bar{S}} \boldsymbol{\Omega} \boldsymbol{\bar{S}}^{T} = \boldsymbol{\Omega} $. The multi-mode GKP codes that are constructed this way based on a symplectic matrix $\boldsymbol{\bar{S}}$ are called the symplectic lattice codes \cite{Harrington2004}. Recall that in the single-mode case, we had a freedom choose any $2$-dimensional lattice to define a single-mode GKP code. Indeed, it turned out that we can improve the performance of the GKP code by using the hexagonal-lattice structure instead of the square-lattice structure. Also in the case of the multi-mode GKP code, we can also optimize the performance of the code by choosing a lattice with a better sphere packing efficiency. In particular, we can freely choose any $2N$-dimensional symplectic lattice (generated by a $2N\times 2N$ symplectic matrix $\boldsymbol{\bar{S}}$) to define an $N$-mode GKP code. This idea is discussed in more detail in Chapter \ref{chapter:Achievable communication rates with bosonic codes}. Examples of interesting higher-dimensional symplectic lattices include the $D_{4}$ lattice (for $N=2$), the $F_{6}$ lattice (for $N=3$), the $E_{8}$ lattice (for $N=4$), the Barnes-Wall lattice $\Lambda_{16}$ (for $N=8$) and the leech lattice $\Lambda_{24}$ (for $N=12$). See Ref.\ \cite{Harrington2004} for more details.

\chapter{Benchmarking and optimizing single-mode bosonic codes}
\label{chapter:Benchmarking and optimizing single-mode bosonic codes}

In this chapter, I will present my contributions to the field of bosonic quantum error correction during the first half of my PhD studies \cite{Albert2018,Noh2019}. Two of the main goals of this chapter are to characterize intrinsic error-correcting capabilities of various single-mode bosonic codes against practically relevant excitation loss errors (Section \ref{section:Benchmarking single-mode bosonic codes}) \cite{Albert2018}, and to search for an optimal single-mode bosonic code via a comprehensive numerical optimization (Section \ref{section:Optimizing single-mode bosonic codes}) \cite{Noh2019}. The work in Ref.\ \cite{Albert2018}, spearheaded by Dr.\ Victor Albert, was a joint project among the groups of Professors Steve Girvin, Barbara Terhal, and Liang Jiang in which I took part. The work in Ref.\ \cite{Noh2019} was done in collaboration with Dr.\ Victor Albert and Professor Liang Jiang.     


In the benchmarking, it turned out that GKP codes significantly outperform many other bosonic codes in correcting excitation loss errors (see Fig.\ \ref{fig:benchmarking single-mode bosonic codes}). Moreover, from the code optimization, the hexagonal-lattice GKP code emerged as an optimal single-mode bosonic code for correcting excitation loss errors from Haar-random initial codes (see Fig.\ \ref{fig:biconvex optimization combined}). These results are surprising because GKP codes were not originally designed to correct excitation loss errors. Instead, they were designed to correct random shift errors in the phase space. In Section \ref{section:Decoding GKP codes subject to excitation loss errors}, I will provide a sub-optimal decoding strategy for GKP codes subject to excitation loss errors, which can be readily implemented in experiments. By doing so, I will explain why GKP codes work well against excitation loss errors as well as random shift errors. In one sentence, the explanation goes as follows: 
\begin{itemize}
\item ``The GKP codes work well against loss errors because loss errors can be converted via an amplification to shift errors, which the GKP codes can correct.''
\end{itemize}
I will conclude the chapter by outlining several open questions in Section \ref{section:Open questions benchmarking and optimization}.    

\section{Benchmarking single-mode bosonic codes}
\label{section:Benchmarking single-mode bosonic codes}

\subsection{Competitors and rules}

\subsubsection{Various bosonic codes}

Here, we will compare the performance of various single-mode bosonic codes. In Fig.\ \ref{fig:Wigner functions of bosonic codes for benchmarking}, we provide the Wigner functions of the maximally mixed code state of the four-component cat code with $\alpha = \sqrt{3}$, the $(1,1)$-binomial code, the square-lattice and the hexagonal-lattice GKP codes with an average excitation number $\bar{n}=3$. We choose to visualize the maximally mixed state of a code space because it is in one-to-one correspondence with the corresponding code space. 

As reviewed in Chapter \ref{chapter:Bosonic quantum error correction}, the four-component cat code and the $(1,1)$-binomial code are designed to correct the single excitation loss events and are rotation-symmetric. Specifically, they are invariant under the $180\degree$ phase rotation $\hat{\Pi}_{2} = e^{i\pi \hat{n}}$ and thus have even number of excitations. The square-lattice and the hexagonal-lattice GKP codes are designed to correct random shift errors in the phase space and are translation-symmetric. In particular, they are invariant under a discrete set of translations in the phase space and thus are stabilized by two displacement operations which generate the square-lattice or the hexagonal-lattice structure. Note that some GKP codes, e.g., the square-lattice GKP code, happen to be invariant under the $180\degree$ rotation. However, such rotational symmetry is not utilized and thus is not so relevant for the GKP codes. 

\begin{figure}[t!]
\centering
\includegraphics[width=4.5in]{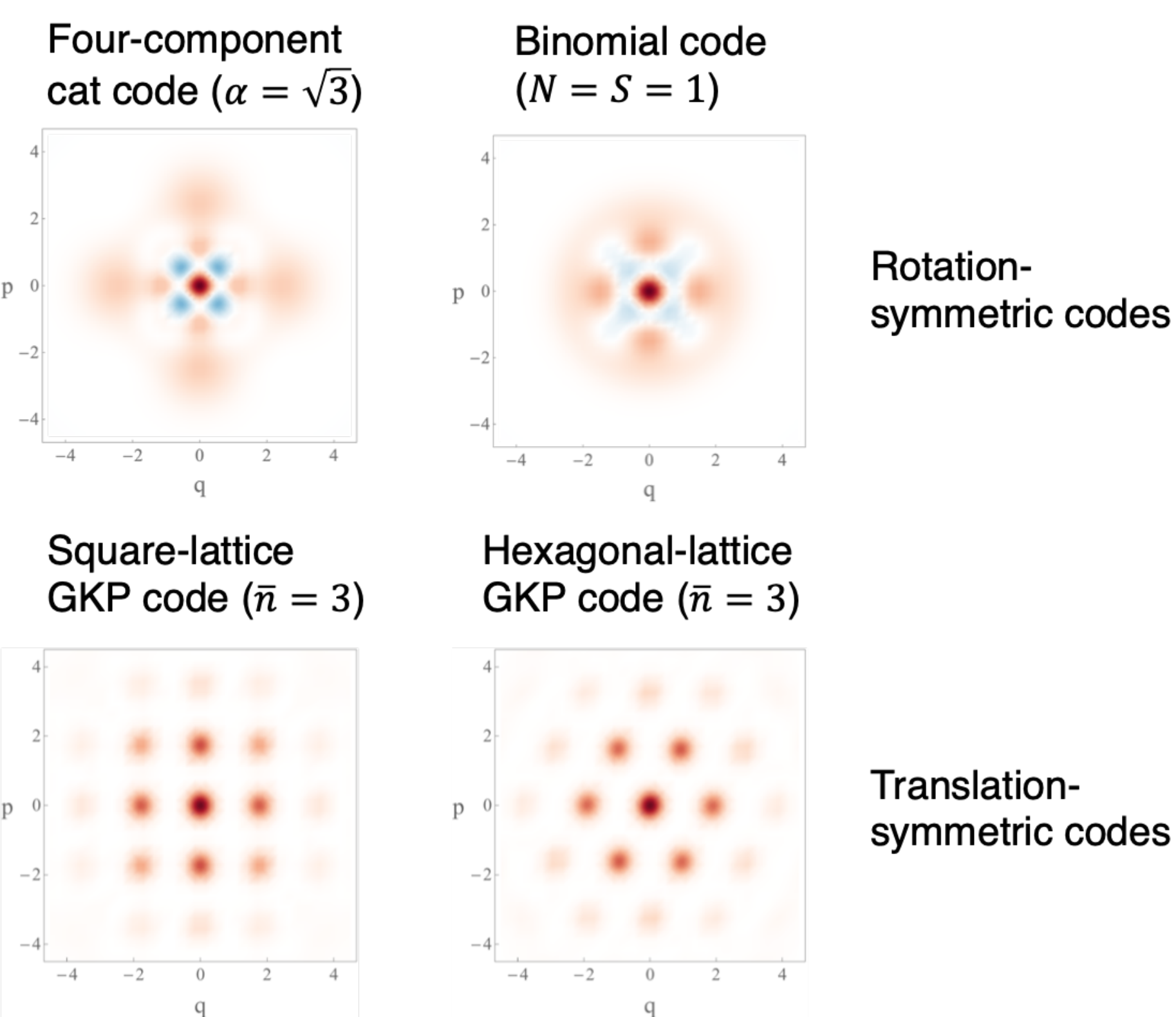}
\caption{Wigner functions of the maximally mixed code state of various single-mode bosonic codes. The four-component cat code and the $(1,1)$-binomial code are designed to correct single excitation loss events and are rotation-symmetric. The square-lattice and the hexagonal-lattice GKP codes are designed to correct random shift errors in the phase space and are translation-symmetric.  }
\label{fig:Wigner functions of bosonic codes for benchmarking}
\end{figure}

Besides the code families listed above, there are also several single-mode bosonic codes that are obtained from a numerical optimization (see the supplemental material of Ref.\ \cite{Albert2018}). An example is the $\sqrt{17}$ code \cite{Michael2016}: 
\begin{align}
|0_{\textrm{num}}^{(\sqrt{17})}\rangle &= \frac{1}{\sqrt{6}} \Big{[} \sqrt{ 7-\sqrt{17} } |0\rangle + \sqrt{ \sqrt{17}-1 }|3\rangle \Big{]}, 
\nonumber\\
|1_{\textrm{num}}^{(\sqrt{17})}\rangle &= \frac{1}{\sqrt{6}} \Big{[} \sqrt{ 9-\sqrt{17} } |1\rangle - \sqrt{ \sqrt{17}-3 }|4\rangle \Big{]} . \label{eq:sqrt 17 code logical states}
\end{align}  
The $\sqrt{17}$ code satisfies the Knill-Laflamme condition for the first-order loss error set $\lbrace \hat{I},\hat{a}  \rbrace$. Also because the Knill-Laflamme condition is satisfied for the error set $\lbrace \hat{I},\hat{a} \rbrace$, the logical states of the $\sqrt{17}$ code have the same average excitation number  
\begin{align}
\bar{n}_{\textrm{num}}^{( \sqrt{17} )} = \frac{\sqrt{17}-1}{2} \simeq 1.562 , 
\end{align}
which is less than those of the $(1,1)$-binomial code, i.e., $\bar{n}_{\textrm{bin}}^{(1,1)} = 2$. Unlike the four-component cat code and the binomial code, however, the $\sqrt{17}$ code does not have an apparent symmetry such as the even excitation number parity. Thus if we use the $\sqrt{17}$-code, we need to perform a general projective measurement (different from the parity measurement) that distinguishes the code space from the error space to look for single-excitation loss errors. 

\subsubsection{Average energy of a bosonic code}

Note that bosonic codes have a notion of ``size''. For instance, we can choose any coherent state amplitude $\alpha$ for the cat codes. As we increase $\alpha$, we can make the cat code have a larger energy and be more robust against bosonic dephasing errors. Also in the case of the binomial codes, we can increase the size of the code by increasing the parameters $N$ and $S$. By doing so, we can deal with higher-order loss events. In the case of the GKP codes, the size of the code is infinite in the ideal case. On the other hand, realistic GKP states have a finite energy. In particular, the size of a GKP state is controlled by the parameter $\Delta$, which characterizes the width of each peak of a GKP state in the phase space. In the ideal case, $\Delta$ vanishes. Similar to other bosonic codes, by allowing the GKP code to have a larger energy (or smaller $\Delta$), we can make the code more robust because any adverse effects due to the finite peaks will become milder. 

As bosonic codes can generally perform better if we allow them to have a larger energy, it is important to control the size of bosonic codes when we compare different code families. Here, we only consider bosonic codes of the qubit-into-an-oscillator type. Consider a qubit-into-an-oscillator bosonic code, i.e., $\mathcal{C} = \textrm{span}\lbrace |0_{\mathcal{C}}\rangle, |1_{\mathcal{C}}\rangle \rbrace$, where $|0_{\mathcal{C}}\rangle$ and $|1_{\mathcal{C}}\rangle$ are two orthonormal logical basis states. Then, we define the average energy (or excitation number) of the code to be the average energy of the maximally mixed code state: 
\begin{align}
\bar{n}_{\mathcal{C}} \equiv  \mathrm{Tr} \Big{[}  \hat{n} \frac{\hat{P}_{\mathcal{C}}}{2}  \Big{]} , 
\end{align} 
where $\hat{P}_{\mathcal{C}} = |0_{\mathcal{C}}\rangle\langle 0_{\mathcal{C}}| + |1_{\mathcal{C}}\rangle\langle 1_{\mathcal{C}}|$ is the projection operator to the code space $\mathcal{C}$. When we compare different code families, we will impose an energy constraint such that each code has an energy less than or equal to a maximum value $\bar{n}_{\textrm{max}}$ (i.e., $\bar{n}_{\mathcal{C}}\le \bar{n}_{\textrm{max}}$ for all codes $\mathcal{C}$).   

A caveat of the above notion of the size of a code is that a certain code might have a thicker tail in the excitation number distribution than some other code even when the two codes have the same average excitation number. For instance, cat and binomial codes respectively follow Poisson and binomial distributions in the excitation number basis and thus have a relatively thin tail. On the other hand, GKP codes follow a geometrical (or a thermal) distribution in the excitation number basis and thus have a thicker tail than those of cat and binomial codes. Despite this caveat, we will use the average energy as a convenient metric for characterizing the size of a bosonic code.    

\subsubsection{Error model}

As discussed in Chapter \ref{chapter:Bosonic quantum error correction}, excitation loss errors are dominant error sources in many realistic bosonic systems. This is especially the case for light modes in optical systems and for microwave cavity modes in circuit QED systems. We will thus focus on excitation loss errors. More specifically, we will consider bosonic pure-loss channels $\mathcal{N}[\eta,0]$ with transmissivity $\eta \in [0,1]$ or loss probability $\gamma = 1-\eta$. Note that the bosonic pure-loss channel $\mathcal{N}[\eta,0]$ is generated by the following Lindblad equation 
\begin{align}
\frac{d\hat{\rho} (t) }{dt} &= \kappa \mathcal{D}[\hat{a}]( \hat{\rho}(t) ) =  \kappa \Big{[}\hat{a}\hat{\rho}(t)\hat{a}^{\dagger}-\frac{1}{2} \lbrace \hat{a}^{\dagger}\hat{a} , \hat{\rho}(t) \rbrace  \Big{]}. 
\end{align}
In particular, we have the following identity 
\begin{align}
\mathcal{N}[\eta  = e^{-\kappa t},0]  &= e^{ \kappa t \mathcal{D}[\hat{a}] }. 
\end{align}  
Thus, the loss probability $\gamma$ is given by 
\begin{align}
\gamma = 1-e^{-\kappa t},
\end{align}
where $\kappa$ is the loss rate and $t$ is the time elapsed. Various other representations of the bosonic pure-loss channel $\mathcal{N}[\eta,0]$ are reviewed in Chapter \ref{chapter:Bosonic quantum error correction} and summarized in Table \ref{table:excitation loss errors}. In the benchmark to be presented below, we compare various bosonic codes against a bosonic pure-loss channel $\mathcal{N}[\eta = 1-\gamma,0]$ for various values of the loss probability $\gamma$. 

\subsubsection{Recovery operation} 

It is also important to realize that choosing an appropriate error recovery (or decoding) operation is crucial when evaluating the performance of a bosonic code. This is because even if a code has an excellent intrinsic error-correcting capability against a certain error model, the code will not perform well against the error if we use a poorly designed decoding scheme. Thus, it is essential to use a well-performing recovery operation when comparing different bosonic code families. In our benchmark below, to be fair to all code families, we report the best performance of each bosonic code by using an optimal recovery operation of the code. This way, we can focus on comparing the intrinsic error-correcting capabilities of various bosonic codes. 

We remark that in practice, even our attempts to correct for errors can be erroneous. For example, the recovery operation itself can fail and add undesirable noise to the system. Moreover, the encoding process can be noisy as well. For these reasons, if a code has a more complicated structure and cannot be prepared efficiently than other codes, the encoding and the error recovery processes for this code can fail with higher probability than those for other codes with a simpler structure. However, we do not address such realistic imperfections here when we compare various bosonic codes. Instead, we assume that the encoding and the error recovery processes can be implemented noiselessly. The reason for this simplistic assumption is again because we want to compare the ultimate error-correcting capability of various bosonic codes. 

We also remark that once we start worrying about realistic imperfections in experimental realizations, we should unavoidably perform a case-by-case study for each code family and for each physical architecture. This is because strategies to deal with realistic imperfections vary largely depending on the structure of the code and the type of the error in the physical system that hosts the code. An overview of such case-by-case considerations on fault-tolerance for various bosonic codes will indeed be provided In Chapter \ref{chapter:Fault-tolerant bosonic quantum error correction}. In particular, we will dive deeper into realistic imperfections of the GKP code and present a tailored method to deal with such imperfections to realize fault-tolerant bosonic quantum error correction with the GKP code. In this chapter, on the other hand, to focus more on the intrinsic error-correcting capability of various bosonic codes, we will work with the simplistic assumption that the error correction processes are noiseless. 

\subsection{Entanglement fidelity as a figure of merit} 

Let us now discuss the figure of merit that we use for benchmarking the performance of various bosonic codes against a bosonic pure-loss channel. Note that for a figure of merit to be useful, it should be sufficiently representative to capture an overall performance of an error-correcting code. Moreover as discussed above, we want to use an optimal recovery operation for each code to make the comparison fair. Thus, it should ideally be straightforward to find an optimal recovery operation that maximizes the chosen figure of merit. 

With these considerations in mind, we choose to use the entanglement fidelity \cite{Schumacher1996a} as a figure of merit. Below, we will explain why the entanglement fidelity is a reasonable choice. In particular, we will show that the entanglement fidelity is sufficiently representative as it is closely related to the average-case fidelity. Also, the entanglement fidelity can be readily maximized over all recovery operations via a convex optimization. 

To make the discussion more concrete, let us consider bosonic codes of the qudit-into-an-oscillator type. Also, we truncate the bosonic Hilbert space and only consider an $n$-dimensional subspace consisting of the $n$ lowest energy states, i.e., $\mathcal{H}_{n} \equiv \textrm{span}\lbrace |0\rangle, \cdots, |n-1\rangle \rbrace$. Note that we can associate a bosonic code $\mathcal{C}$ ($\subseteq \mathcal{H}_{n}$) with an isometry from a hypothetical $d$-dimensional logical Hilbert space $\mathcal{H}' = \textrm{span}\lbrace |0_{\mathcal{H}'}\rangle, \cdots,|(d-1)_{\mathcal{H}'}\rangle  \rbrace$ to the truncated bosonic Hilbert space $\mathcal{H}_{n}$ such that  
\begin{align}
|\mu_{\mathcal{H}'}\rangle\in \mathcal{H}' \rightarrow |\mu_{\mathcal{C}}\rangle \in \mathcal{H}_{n}, \,\,\, \textrm{for all}\,\,\, \mu\in\lbrace 0,\cdots, d-1 \rbrace, 
\end{align}  
where $ |\mu_{\mathcal{C}}\rangle$ is the logical state of the code. Even more generally, we can consider a completely-positive and trace-preserving (CPTP) \cite{Choi1975} encoding map $\mathcal{E} : \mathcal{L}(\mathcal{H}') \rightarrow \mathcal{L}(\mathcal{H}_{n})$ that maps an input density matrix in the hypothetical logical Hilbert space to a density matrix in the physical (truncated) bosonic Hilbert space. Here $\mathcal{L}(\mathcal{H})$ is the space of linear operators acting on the vector space $\mathcal{H}$. For example in the case of the isometry discussed above, we have
\begin{align}
\mathcal{E}( |\mu_{\mathcal{H}'}\rangle\langle \nu_{\mathcal{H}'} |  ) = |\mu_{\mathcal{C}}\rangle\langle \nu_{\mathcal{C}} |, \,\,\,  \textrm{for all}\,\,\, \mu,\nu \in \lbrace 0,\cdots, d-1 \rbrace. 
\end{align}   

Consider an arbitrary pure input state $|\psi_{\mathcal{H}'}\rangle$ in the hypothetical logical space $\mathcal{H}'$ and assume that it is encoded via an encoding map $\mathcal{E}$ into the physical bosonic Hilbert space $\mathcal{H}_{n}$, i.e., 
\begin{align}
|\psi_{\mathcal{H}'}\rangle\langle \psi_{\mathcal{H}'}| \xrightarrow{\mathcal{E}} \mathcal{E} ( |\psi_{\mathcal{H}'}\rangle\langle \psi_{\mathcal{H}'}| ) . 
\end{align}
The encoded state will then undergo a physical error process, which can be described by a CPTP noise map $\mathcal{N}: \mathcal{L}(\mathcal{H}_{n})\rightarrow \mathcal{L}(\mathcal{H}_{n})$. For example in the case of the bosonic pure-loss channel, the noise map $\mathcal{N}$ is given by $\mathcal{N} = \mathcal{N}[ \eta= 1-\gamma ,0 ]$. Upon the action of the noise channel $\mathcal{N}$, the state is further transformed into 
\begin{align}
\mathcal{N} \cdot \mathcal{E} ( |\psi_{\mathcal{H}'}\rangle\langle \psi_{\mathcal{H}'}| ). 
\end{align}  

Lastly, a recovery map is applied to correct for the noise. More specifically, we consider a CPTP recovery map $\mathcal{R}: \mathcal{L}( \mathcal{H}_{n} ) \rightarrow \mathcal{L}(\mathcal{H}')$ that maps the corrupted state in the physical Hilbert space $\mathcal{H}_{n}$ back to the hypothetical logical Hilbert space $\mathcal{H}'$. Then, we are left with the following recovered state: 
\begin{align}
\mathcal{R} \cdot \mathcal{N} \cdot \mathcal{E} ( |\psi_{\mathcal{H}'}\rangle\langle \psi_{\mathcal{H}'}| ).
\end{align} 
Then, the entire encoding, noise, and the recovery process can be summarized by a CPTP map 
\begin{align}
\mathcal{M} &\equiv \mathcal{R} \cdot \mathcal{N} \cdot \mathcal{E}: \mathcal{L}(\mathcal{H'}) \rightarrow \mathcal{L}(\mathcal{H'}). 
\end{align}
Note that the channel $\mathcal{M} $ characterizes how well the logical information is preserved.

Of course in the actual implementation of a bosonic code, the recovery operation happens within the physical Hilbert space. Nevertheless, we consider the recovery map from the physical space to the hypothetical logical space (i.e., $\mathcal{R} : \mathcal{L}( \mathcal{H}_{n} ) \rightarrow \mathcal{L}(\mathcal{H}')$) solely for the purpose of evaluating the error correction scheme. If we want, we can always define a valid physical recovery operation (that maps a corrupted state in the physical Hilbert state to a recovered state in the physical Hilbert space) by applying the encoding map $\mathcal{E}$ after the recovery map $\mathcal{R}$, i.e., $\mathcal{R}' \equiv \mathcal{E} \cdot \mathcal{R}$.             

Now getting back to the evaluation of the error correction scheme, we define the entanglement fidelity of the channel $\mathcal{M} : \mathcal{L}(\mathcal{H}') \rightarrow \mathcal{L}(\mathcal{H}')$ as follows: 
\begin{align}
F_{e}( \mathcal{M} ) \equiv \langle \Phi^{+} | (\mathcal{M} \otimes \textrm{id}_{\mathcal{H}''} ) ( | \Phi^{+}\rangle \langle \Phi^{+} | ) | \Phi^{+}\rangle. 
\end{align}
Here, $|\Phi^{+}\rangle$ is a maximally entangled state between the system $\mathcal{H}'$ and an ancillary system $\mathcal{H}''$ of the same dimension, i.e., 
\begin{align}
|\Phi^{+}\rangle &= \frac{1}{\sqrt{d}} \sum_{\mu =0 }^{d-1} |\mu_{\mathcal{H}'}\rangle|\mu_{\mathcal{H}''}\rangle. 
\end{align}
At glance it might appear that the entanglement fidelity is not representative enough as it quantifies the fidelity just for a single input state $|\Phi^{+}\rangle$. On the other hand, to evaluate the overall performance, we might want to consider a quantity like the average-case fidelity, i.e., 
\begin{align}
F_{\textrm{avg}}(\mathcal{M}) &\equiv \int d\psi \langle \psi| \mathcal{M} (|\psi\rangle\langle\psi| ) |\psi\rangle, 
\end{align}
where the input state $|\psi\rangle$ is drawn uniformly from the hypothetical logical Hilbert space $\mathcal{H}'$. Remarkably, the average-case fidelity $F_{\textrm{avg}}(\mathcal{M})$ and the entanglement fidelity $F_{e}( \mathcal{M} )$ are closely related to each other via the following relation \cite{Horodecki1999,Nielsen2002}: 
\begin{align}
F_{\textrm{avg}}(\mathcal{M}) &= \frac{d F_{e}( \mathcal{M} ) + 1}{d+1},  
\end{align}
where $d$ is the dimension of the hypothetical logical Hilbert space $\mathcal{H}'$. Thus, if we consider qubit-into-an-oscillator bosonic codes, $d$ is given by $d=2$ regardless of the dimension of the physical bosonic Hilbert space. Therefore in practice, the entanglement fidelity $F_{e}( \mathcal{M} )$ is very well correlated with the average-case fidelity $F_{\textrm{avg}}(\mathcal{M})$. In particular, the average-case infidelity $1- F_{\textrm{avg}}(\mathcal{M})$ differs from the entanglement infidelity $1-F_{e}(\mathcal{M})$ merely by a constant factor, i.e.,
\begin{align}
1- F_{\textrm{avg}}(\mathcal{M}) = \frac{d}{d+1} ( 1-F_{e}(\mathcal{M}) ) .  
\end{align} 
In the case of a qubit-into-an-oscillator code ($d=2$), the constant factor is given by $\frac{d}{d+1} = \frac{2}{3}$. Therefore, we can conclude that the entanglement fidelity is as representative as the average-case fidelity.   

\subsection{Maximization of the entanglement fidelity}

Another desirable property of the entanglement fidelity is that we can straightforwardly find an optimal recovery map $\mathcal{R}$ that maximizes the entanglement fidelity, via a convex optimization, for a given encoding scheme $\mathcal{E}$ and a noise map $\mathcal{N}$. More specifically, we can formulate a semidefinite programming (SDP) to find an optimal recovery $\mathcal{R}$ \cite{Reimpell2005,Fletcher2007} (see also Eq.\ \eqref{eq:SDP decoding optimization} in Section \ref{section:Optimizing single-mode bosonic codes}). Thanks to this property, we can readily characterize the intrinsic error-correcting capability of an encoding scheme $\mathcal{E}$ against a given noise model $\mathcal{N}$. Below, we will present the results of the code comparison based on the SDP optimization of the recovery operation. 

\subsection{Results}

Here, we present the benchmarking results of the single-mode bosonic codes against the bosonic pure-loss channels $\mathcal{N}[\eta = 1-\gamma,0]$ \cite{Albert2018}. Specifically, we compare the following code families 
\begin{align}
\lbrace \textrm{single-rail}, \textrm{cat},\textrm{bin},\textrm{num},\textrm{gkps},\textrm{gkp} \rbrace. \label{eq:code families for comparison}
\end{align}
Here, ``single-rail'' represents the trivial encoding scheme based on the vacuum state and the single-photon Fock state $|0_{\textrm{single-rail}}\rangle = |0\rangle$, $|1_{\textrm{single-rail}}\rangle = |1\rangle$. The single-rail code encodes two quantum states with the least energy, it is the best physical bosonic qubit that allows the longest lifetime without error correction. Thus, the performance of the single-rail code serves as a baseline. The ``cat'', ``binomial'', and GKP code families are comprehensively reviewed in Chapter \ref{chapter:Bosonic quantum error correction}. Among the GKP code families, ``gkps'' represents the square-lattice GKP code and ``gkp'' represents a family of general single-mode GKP codes with more general lattice structures than the square-lattice structure. ``num'' represents some numerically optimized codes (including the $\sqrt{17}$ code defined in Eq.\ \eqref{eq:sqrt 17 code logical states}) provided in the supplemental material of Ref.\ \cite{Albert2018}.

\begin{figure}[t!]
\centering
\includegraphics[width=3.7in]{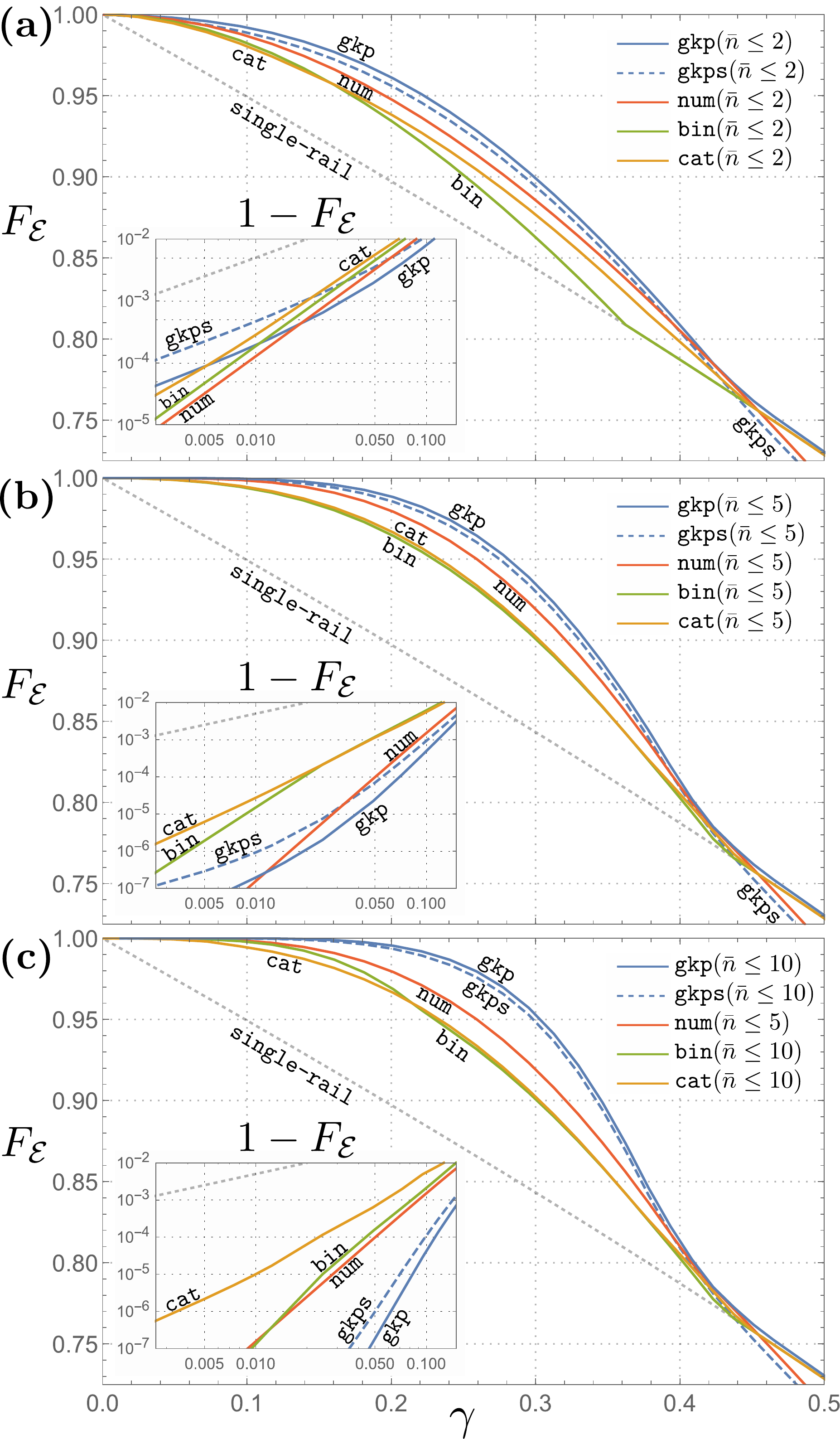}
\caption{[Fig.\ 2 in PRA \textbf{97}, 032346 (2018)] The optimal entanglement fidelity $F_{e}(\mathcal{M}^{\star})$ as a function of the loss probability $\gamma \in [0, \frac{1}{2}]$ for various single-mode code families $\lbrace \textrm{single-rail}, \textrm{cat},\textrm{bin},\textrm{num},\textrm{gkps},\textrm{gkp} \rbrace$ against the bosonic pure-loss channels $\mathcal{N}[\eta = 1-\gamma,0]$. The channel $\mathcal{M}^{\star}$ is defined as $\mathcal{M}^{\star} = \mathcal{R}^{\star} \cdot \mathcal{N} \cdot \mathcal{E}$ and the optimal recovery $\mathcal{R}^{\star}$ is obtained by maximizing the entanglement fidelity via a semidefinite programming, given an encoding map $\mathcal{E}$ (or a code $\mathcal{C}$) and a bosonic pure-loss channel $\mathcal{N} = \mathcal{N}[\eta = 1-\gamma,0]$. Within each code family, we report the performance of the best code optimized over all code parameters subject to an average excitation number constraint (a) $\bar{n}_{\mathcal{C}} \le 2$, (b) $\bar{n}_{\mathcal{C}} \le 5$, and (c) $\bar{n}_{\mathcal{C}} \le 10$. Note that the GKP code families are shown to outperform all the other code families in a wide range of loss parameters.  }
\label{fig:benchmarking single-mode bosonic codes}
\end{figure}

In Fig.\ \ref{fig:benchmarking single-mode bosonic codes} (Fig.\ 2 in Ref.\ \cite{Albert2018}), we plot the optimal entanglement fidelity $F_{e}(\mathcal{M}^{\star})$ for various code families in Eq.\ \eqref{eq:code families for comparison} as a function of the loss probability $\gamma \in [0, \frac{1}{2}]$ against the bosonic pure-loss channels $\mathcal{N}[\eta = 1-\gamma ,0]$. The channel $\mathcal{M}^{\star}$ is defined as $\mathcal{M}^{\star} = \mathcal{R}^{\star} \cdot \mathcal{N} \cdot \mathcal{E}$ and the optimal recovery $\mathcal{R}^{\star}$ is obtained by maximizing the entanglement fidelity via a semidefinite programming, given an encoding map $\mathcal{E}$ (or a code $\mathcal{C}$) and a bosonic pure-loss channel $\mathcal{N} = \mathcal{N}[\eta = 1-\gamma,0]$. Within each code family, we report the performance of the best code optimized over all code parameters subject to an average excitation number constraint (a) $\bar{n}_{\mathcal{C}} \le 2$, (b) $\bar{n}_{\mathcal{C}} \le 5$, and (c) $\bar{n}_{\mathcal{C}} \le 10$. 

First, recall that we generally expect larger bosonic codes to perform better than smaller ones within a fixed code family, because larger codes can correct higher-order error events. This intuition is numerically corroborated in Fig.\ \ref{fig:benchmarking single-mode bosonic codes}. That is, the entanglement fidelity of the effective error channel $\mathcal{M}^{\star} = \mathcal{R}^{\star} \cdot \mathcal{N} \cdot \mathcal{E}$ after the error correction becomes larger as we increase the maximum allowed average excitation number from $\bar{n}_{\textrm{max}} = 2$ to $\bar{n}_{\textrm{max}} = 10$. Also, for any loss probability $\gamma$ smaller than a certain critical value $\gamma_{c}(\bar{n}_{\textrm{max}})$, all the non-trivial code families (i.e., cat, bin, num, gkps, gkp) outperform the single-rail code ($\gamma_{c}( 2) \simeq 0.36 $ and $\gamma_{c}( 5) \simeq  \gamma_{c}(10) \simeq 0.43 $). This clearly shows the advantage of using an error-corrected bosonic qubit with a non-trivial bosonic code over an unprotected bosonic qubit with a trivial encoding scheme. On the other hand, it also indicates that if the error channel $\mathcal{N}$ is too noisy, error correction cannot really reduce the noise in the channel. Indeed using the framework of quantum communication theory (and based on the notion of quantum capacity), we can prove that quantum error correction cannot help for any bosonic pure-loss channel with $\gamma \ge 0.5$ (see Lemma \ref{lemma:Degradability or anti-degradability of bosonic pure-loss channels} and the discussion below). These fundamental communication-theoretic aspects will be discussed in more detail in Chapter \ref{chapter:Quantum capacities of bosonic Gaussian channels}.      

Let us now move on to comparing different code families in the regime where quantum error correction can actually help (i.e., $\gamma \le \gamma_{c}(\bar{n}_{\textrm{max}})$). Overall, we can see that the GKP code families (i.e., gkps and gkp) outperform all the other code families in a wide range of loss parameters. This is a really surprising result because the GKP codes are not designed to correct excitation loss errors, whereas other code families (i.e., cat, bin, num) are specifically designed to correct excitation loss errors. Instead, the GKP codes are designed to correct random shift errors in the phase space. One might be tempted to say that if the loss probability $\gamma$ is small, excitation loss errors can be effectively regarded as a small random shift error in the phase space and this is why the GKP codes work well against the loss errors. This could be a valid explanation in the small energy regime (e.g., for $\bar{n}_{\textrm{max}}=2$). However, as the GKP code gets larger (or $\bar{n}_{\textrm{max}}$ increases), even a small constant fraction of loss can be turned into a huge shift error (uncorrectable by the code), especially in the high-energy sectors of the phase space that are far away from the origin (or the vacuum). Nevertheless, as can be seen from Fig.\ \ref{fig:benchmarking single-mode bosonic codes}(c), even in the case of $\bar{n}_{\textrm{max}} = 10$, the GKP code families outperform all the other code families. Moreover, the performance gap between the GKP code families and other code families becomes even wider. Also, these large GKP codes can now handle even large loss probabilities as well as small ones. For instance, the GKP code families achieve an entanglement fidelity $F_{e}$ as high as $99\%$ (or an entanglement infidelity $1-F_{e}$ as low as $1\%$) starting from a loss channel with $20\%$ loss probability or $\gamma =0.2$. These observations clearly lead us to conclude that the excellent performance of the GKP code families are not merely due to the fact that small loss errors can be understood as small shift errors. Therefore, a more refined explanation that applies to all energy scales is needed to explain the exceptional performance of the GKP codes. We discuss these aspects in more detail in Section \ref{section:Decoding GKP codes subject to excitation loss errors}. 
 
Lastly, let us take a closer look into the small loss probability regime with $\gamma \le 0.1$. In this regime, the cat, bin, and num code families can indeed outperform the GKP code families. For example in the case of $\bar{n}_{\textrm{max}} =2$, the num code family starts to outperform the gkp code family when the loss probability becomes smaller than a critical value, i.e., $\gamma \lesssim  0.025$ (see the inset of Fig.\ \ref{fig:benchmarking single-mode bosonic codes}(a)). For $\bar{n}_{\textrm{max}} =5$, this critical value gets smaller and the num code family outperform the gkp code family when $\gamma \lesssim 0.01$ (see the inset of Fig.\ \ref{fig:benchmarking single-mode bosonic codes}(b)). When $\bar{n}_{\textrm{max}}=10$, the crossing point is not observed in the considered parameter regime and the critical value is clearly less than $0.01$. These results make sense because the cat, bin, and num code families are specifically designed to correct small excitation loss errors in a perturbative manner. However, as we allow the codes to have a larger energy, the relative advantage of the cat, bin, num code families over the GKP code families quickly disappear.   

As we discussed above, it is very important to realize that this performance benchmark captures only the intrinsic error-correcting capability of a bosonic code. Thus, while our results indicate that the GKP code families exhibit an excellent error-correcting capability against excitation loss errors, it does not immediately imply that the GKP code families will outperform all the other code families in practice. Indeed, preparation of a GKP state is generally more challenging than preparation of a cat, binomial, and numerically optimized code state of the same energy. Thus, the GKP error correction schemes may be more sensitive to realistic imperfections than, for instance, the cat code error correction schemes. For these reasons, our performance benchmark does not necessarily discourage the pursuit of cat, bin, and num code families. However, it does encourage a further pursuit of the GKP code families (despite the experimental challenges) as it shows that the GKP code families can exhibit an excellent performance for practically relevant excitation loss errors as well as random shift errors. 

The effects of realistic imperfections in the implementation of the GKP codes will be further discussed in Chapter \ref{chapter:Fault-tolerant bosonic quantum error correction}. Before moving on to the issues related to experimental imperfections, we will further discuss the intrinsic error-correcting capability of various bosonic codes in the rest of this chapter. In particular, we will address the question of whether there exists a code family that is even better than the GKP code families via a brute-force numerical biconvex optimization. We will also discuss the decoding of the GKP code families against excitation loss errors in more detail.     

\section{Optimizing single-mode bosonic codes}
\label{section:Optimizing single-mode bosonic codes}

We have so far compared the known code families, i.e., cat, bin, num, gkps, and gkp code families. One of the key takeaway messages is that the gkp code family outperforms many other known bosonic code families in a wide range of loss parameters. Then, the natural question is whether there exist some other code families that can even outperform the gkp code. Here, we address this question via a numerical optimization and gives a negative answer to the question for the single-mode bosonic codes subject to an average energy constraint. That is, we show that the hexagonal-lattice GKP code emerges as an optimal encoding scheme (of a qudit-into-an-oscillator type) starting from Haar-random initial codes against excitation loss errors. Thus, the numerical optimization results signal that the GKP code family may indeed be the most effective bosonic code family, not just among the families we are aware of, in correcting practically relevant excitation loss errors.    

More specifically, we formulate a biconvex optimization problem to search for an optimal single-mode bosonic code subject to an average energy constraint. Most importantly, we do not assume any structure of the encoding scheme and explore all possible encoding CPTP maps $\mathcal{E}:\mathcal{L}(\mathcal{H}')\rightarrow \mathcal{L}(\mathcal{H}_{n})$ that are allowed by the laws of quantum physics. Similarly as in the case of code comparison, we use the entanglement fidelity as a figure of merit because it is very well correlated with the average-case fidelity. Furthermore, as will be made clear below, the entanglement fidelity can be readily optimized.     

\subsection{Expressing the entanglement fidelity in terms of Choi matrices}

Recall that an error correction scheme consists of 
\begin{align}
\mathcal{E}&:\mathcal{L}(\mathcal{H}')\rightarrow \mathcal{L}(\mathcal{H}_{n}) : \textrm{an encoding map (fixed or to be optimized)}, 
\nonumber\\
\mathcal{N}&:\mathcal{L}(\mathcal{H}_{n})\rightarrow \mathcal{L}(\mathcal{H}_{n}) : \textrm{a noise map (fixed)},  
\nonumber\\
\mathcal{R}&:\mathcal{L}(\mathcal{H}_{n})\rightarrow \mathcal{L}(\mathcal{H}') : \textrm{an error recovery map (to be optimized)}. 
\end{align}
Here, $\mathcal{H}'$ is the hypothetical logical Hilbert space whose dimension is given by the number of logical states $d$. Also, $\mathcal{H}_{n}$ is the physical (truncated) bosonic Hilbert space whose dimension is given by $n$. At the end of the error correction scheme, we are left with an effective channel $\mathcal{M} = \mathcal{R} \cdot \mathcal{N} \cdot \mathcal{E}$ that characterizes how well the logical quantum information is preserved. Then, we evaluate the error correction scheme by using the entanglement fidelity measure: 
\begin{align}
F_{e} ( \mathcal{M} ) \equiv \langle \Phi^{+} | ( \mathcal{M} \otimes \textrm{id}_{\mathcal{H}''} ) ( |\Phi^{+}\rangle\langle \Phi^{+}| ) |\Phi^{+}\rangle,
\end{align}
where $|\Phi^{+}\rangle$ is the maximally entangled state between the system $\mathcal{H}'$ and an ancillary system $\mathcal{H}''$ of the same dimension, i.e., 
\begin{align}
|\Phi^{+}\rangle &= \frac{1}{\sqrt{d}} \sum_{\mu=0}^{d-1} |\mu_{\mathcal{H}'}\rangle|\mu_{\mathcal{H}''}\rangle .
\end{align}
To make the expression more explicit, note that 
\begin{align}
( \mathcal{M} \otimes \textrm{id}_{\mathcal{H}''} ) ( |\Phi^{+}\rangle\langle \Phi^{+}| ) &= \frac{1}{d} \sum_{\mu,\nu=0}^{d-1} \mathcal{M}( |\mu_{\mathcal{H}'}\rangle\langle \nu_{\mathcal{H}'}| ) \otimes |\mu_{\mathcal{H}''}\rangle\langle \nu_{\mathcal{H}''}| 
\nonumber\\
&= \frac{1}{d} \sum_{\mu,\nu,\rho,\sigma=0}^{d-1} ( \hat{X}_{\mathcal{M}} )_{ [ \mu\rho ], [ \nu\sigma ] } |\rho_{\mathcal{H}'}\rangle\langle \sigma_{\mathcal{H}'}|\otimes |\mu_{\mathcal{H}''}\rangle\langle \nu_{\mathcal{H}''}| , 
\end{align}
and thus 
\begin{align}
F_{e} ( \mathcal{M} ) &= \langle \Phi^{+} | ( \mathcal{M} \otimes \textrm{id}_{\mathcal{H}''} ) ( |\Phi^{+}\rangle\langle \Phi^{+}| ) |\Phi^{+}\rangle = \frac{1}{d^{2}} \sum_{\mu,\nu=0}^{d-1} ( \hat{X}_{\mathcal{M}} )_{ [ \mu\mu ], [ \nu\nu ] },   \label{eq:entanglement fidelity in terms of Choi matrix of M}
\end{align}
where $\hat{X}_{\mathcal{M}} \in \mathcal{L}( \mathcal{H}'\otimes \mathcal{H}' )$ is the Choi matrix \cite{Choi1975} of the channel $\mathcal{M}$ whose matrix elements are defined as
\begin{align}
( \hat{X}_{\mathcal{M}} )_{ [ \mu\rho ], [ \nu\sigma ] } &= \langle \rho_{\mathcal{H}'}| \mathcal{M}( |\mu_{\mathcal{H}'}\rangle\langle \nu_{\mathcal{H}'}| )|\sigma_{\mathcal{H}'}\rangle. 
\end{align}
The Choi matrix $\hat{X}_{\mathcal{M}}$ fully characterizes the channel $\mathcal{M}$ as it contains information about where an input basis element $|\mu_{\mathcal{H}'}\rangle\langle \nu_{\mathcal{H}'}|$ is mapped to for all basis elements (i.e., for all $\mu,\nu\in\lbrace 0,\cdots, d-1 \rbrace$). Similarly, we can define the Choi matrices of the encoding, noise, and recovery maps: 
\begin{align}
( \hat{X}_{\mathcal{E}} )_{ [ \mu\rho ], [ \nu\sigma ] } &= \langle \rho_{\mathcal{H}_{n}}| \mathcal{E}( |\mu_{\mathcal{H}'}\rangle\langle \nu_{\mathcal{H}'}| )|\sigma_{\mathcal{H}_{n}}\rangle, 
\nonumber\\
& ( \mu,\nu\in \lbrace 0, \cdots, d-1\rbrace \textrm{ and } \rho,\sigma\in \lbrace 0, \cdots, n-1\rbrace  ), 
\nonumber\\
( \hat{X}_{\mathcal{N}} )_{ [ \mu\rho ], [ \nu\sigma ] } &= \langle \rho_{\mathcal{H}_{n}}| \mathcal{N}( |\mu_{\mathcal{H}_{n}}\rangle\langle \nu_{\mathcal{H}_{n}}| )|\sigma_{\mathcal{H}_{n}}\rangle, 
\nonumber\\
& ( \mu,\nu\in \lbrace 0, \cdots, n-1\rbrace \textrm{ and } \rho,\sigma\in \lbrace 0, \cdots, n-1\rbrace  ), 
\nonumber\\
( \hat{X}_{\mathcal{R}} )_{ [ \mu\rho ], [ \nu\sigma ] } &= \langle \rho_{\mathcal{H}'}| \mathcal{R}( |\mu_{\mathcal{H}_{n}}\rangle\langle \nu_{\mathcal{H}_{n}}| )|\sigma_{\mathcal{H}'}\rangle, 
\nonumber\\
& ( \mu,\nu\in \lbrace 0, \cdots, n-1\rbrace \textrm{ and } \rho,\sigma\in \lbrace 0, \cdots, d-1\rbrace  ). 
\end{align}    

A nice property of the Choi matrix $\hat{X}_{\mathcal{A}}$ is that it can be used to directly check whether a map $\mathcal{A} : \mathcal{L}(\mathcal{H}_{1}) \rightarrow \mathcal{L}(\mathcal{H}_{2})$ is physically realizable or not. In general, it is known that a map $\mathcal{A}$ corresponds to a physically realizable quantum operation if and only if it is a completely-positive (CP) and trace-preserving (TP) map (i.e., a CPTP map) \cite{Nielsen2000}. The CP condition can be easily checked by inspecting whether the associated Choi matrix is positive semidefinite or not, i.e.,
\begin{align}
\mathcal{A} \textrm{ is completely positive } \leftrightarrow \hat{X}_{\mathcal{A}} \succeq 0. \label{eq:Choi matrix CP condition}
\end{align}    
Moreover, the TP condition can also be checked as follows: 
\begin{align}
\mathcal{A} \textrm{ is trace-preserving } \leftrightarrow \mathrm{Tr}_{\mathcal{H}_{2}} [ \hat{X}_{\mathcal{A}} ] \equiv \sum_{\rho = 0}^{\textrm{dim}(\mathcal{H}_{2}) } ( \hat{X}_{\mathcal{A}} )_{ [ \mu\rho ], [ \nu\sigma ] }  |\mu_{\mathcal{H}_{1}}\rangle\langle \nu_{\mathcal{H}_{1}}| = \hat{I}_{\mathcal{H}_{1}}, \label{eq:Choi matrix TP condition}
\end{align} 
where $\hat{I}_{\mathcal{H}_{1}}$ is the identity operation on the Hilbert space $\mathcal{H}_{1}$. Note that the two conditions in Eqs.\ \eqref{eq:Choi matrix CP condition} and \eqref{eq:Choi matrix TP condition} are convex.  

Because of this nice property, it is desirable to work with the Choi matrices when we optimize an objective function over all possible physical operations. In particular, if the objective function (to be minimized) is convex in the input Choi matrix, we can use an efficient convex optimization method \cite{Boyd2004} for the optimization because the constraints on the Choi matrix are convex as well, as shown above. For these reasons, it is desirable to break down the expression in Eq.\ \eqref{eq:entanglement fidelity in terms of Choi matrix of M} in terms of the Choi matrices of the encoding and the recovery maps $\mathcal{E}$ and $\mathcal{R}$ that we wish to optimize. 

To do so, let us consider the superoperator $\hat{T}_{\mathcal{A}}$ of a quantum map $\mathcal{A} : \mathcal{L}( \mathcal{H}_{1}) \rightarrow \mathcal{L}( \mathcal{H}_{2})$. Matrix elements of the the superoperator $\hat{T}_{\mathcal{A}}$ is defined as $( \hat{T}_{\mathcal{A}} )_{\rho\sigma,\mu\nu} \equiv ( \hat{X}_{\mathcal{A}} )_{[ \mu\rho ] , [ \nu\sigma ]}$. A nice property of the superoperators is that the superoperator of a composite channel $\mathcal{B}\cdot \mathcal{A}$ is given by the matrix multiplication of the superoperators of its constituting channels, i.e., 
\begin{align}
\hat{T}_{ \mathcal{B}\cdot  \mathcal{A}} =\hat{T}_{ \mathcal{B} }  \hat{T}_{ \mathcal{A}}  . 
\end{align}

With all the facts ready, let us now get back to the expression of the entanglement fidelity given in Eq.\ \eqref{eq:entanglement fidelity in terms of Choi matrix of M}. Note that
\begin{align}
F_{e}(\mathcal{M}) &=\frac{1}{d^{2}} \sum_{i,i'=0}^{d-1}(\hat{X}_{\mathcal{M}})_{[ii],[i'i']} 
\nonumber\\
&=  \frac{1}{d^{2}} \sum_{i,i'=0}^{d-1}(\hat{X}_{\mathcal{R}\cdot\mathcal{N}\cdot\mathcal{E}})_{[ii],[i'i']} =  \frac{1}{d^{2}} \sum_{i,i'=0}^{d-1}(\hat{T}_{\mathcal{R}\cdot\mathcal{N}\cdot\mathcal{E}})_{ii',ii'} = \frac{1}{d^{2}}\mathrm{Tr}[\hat{T}_{\mathcal{R}\cdot\mathcal{N}\cdot\mathcal{E}}], 
\end{align}
where $\hat{T}_{\mathcal{R}\cdot\mathcal{N}\cdot\mathcal{E}}$ is the superoperator of $ \mathcal{M} = \mathcal{R}\cdot\mathcal{N}\cdot\mathcal{E}$. Note that $\hat{T}_{\mathcal{R}\cdot\mathcal{N}\cdot\mathcal{E}}$ can be decomposed into $\hat{T}_{\mathcal{R}\cdot\mathcal{N}\cdot\mathcal{E}}= \hat{T}_{\mathcal{R}} \hat{T}_{\mathcal{N}} \hat{T}_{\mathcal{E}}$ and thus we have 
\begin{align} 
(\hat{T}_{\mathcal{R}\cdot\mathcal{N}\cdot\mathcal{E}})_{jj',ii'}  =\sum_{k,k',l,l'=0}^{n-1} (\hat{X}_{\mathcal{R}})_{[lj],[l'j']} (\hat{X}_{\mathcal{N}})_{[kl],[k'l']}  (\hat{X}_{\mathcal{E}})_{[ik],[i'k']}  \label{eq:superoperator of composite channel in terms of constituting Choi matrices}
\end{align}
for $i,i',j,j'\in \lbrace 0,\cdots,d-1 \rbrace$, where $\hat{X}_{\mathcal{R}}\in\mathcal{L}(\mathcal{H}'\otimes \mathcal{H}_{n})$, $\hat{X}_{\mathcal{N}}\in\mathcal{L}(\mathcal{H}_{n}\otimes \mathcal{H}_{n})$ and $\hat{X}_{\mathcal{E}}\in\mathcal{L}(\mathcal{H}_{n}\otimes \mathcal{H}')$ are the Choi matrices of the recovery map $\mathcal{R}$, the noise channel $\mathcal{N}$, and the encoding map $\mathcal{E}$, respectively. 

Thus, we can see at this point that the entanglement fidelity $F_{e}(\mathcal{M})$ is a bi-linear function of $\hat{X}_{\mathcal{R}}$ and $\hat{X}_{\mathcal{E}}$, since $\hat{T}_{\mathcal{R}\cdot\mathcal{N}\cdot\mathcal{E}}$ is bi-linear in $\hat{X}_{\mathcal{R}}$ and $\hat{X}_{\mathcal{E}}$, as can be seen from Eq.\ \eqref{eq:superoperator of composite channel in terms of constituting Choi matrices}. To make the bi-linearity more evident, we define a linear map $f_{\mathcal{N}}:\mathcal{L}(\mathcal{H}_{n}\otimes \mathcal{H}')\rightarrow \mathcal{L}(\mathcal{H}'\otimes \mathcal{H}_{n})$ such that 
\begin{equation}
\big{(}f_{\mathcal{N}}(\hat{X})\big{)}_{[l'i'],[li]} \equiv \sum_{k,k'=0}^{n-1} (\hat{X}_{\mathcal{N}})_{[kl],[k'l']}  (\hat{X})_{[ik],[i'k']}, 
\end{equation}
where $l,l'\in\lbrace 0, \cdots,n-1\rbrace$. The entanglement fidelity $F_{e}(\mathcal{M})$ is then given by 
\begin{equation}
F_{e}(\mathcal{M}  ) = F_{e} ( \mathcal{R} \cdot \mathcal{N} \cdot \mathcal{E} )   = \frac{1}{d^{2}}\mathrm{Tr}\Big{[}\hat{X}_{\mathcal{R}} f_{\mathcal{N}}(\hat{X}_{\mathcal{E}}) \Big{]},  
\label{eq:entanglement fidelity in terms of Choi matrices} 
\end{equation}
which is apparently bi-linear in $\hat{X}_{\mathcal{E}}$ and $\hat{X}_{\mathcal{R}}$. This is precisely a property that we wanted because linear functions are convex.   

\subsection{Convex optimization of error recovery operations} 

Suppose that we are given with an encoding map $\mathcal{E} = \bar{\mathcal{E}}$ and we want to understand its intrinsic error-correcting capability against a given noise channel $\mathcal{N}$. This was precisely the case in Section \ref{section:Benchmarking single-mode bosonic codes}. By now it is clear that finding an optimal recovery operation $\mathcal{R}^{\star}$ that maximizes the entanglement fidelity $F_{e}(\mathcal{M})$ is a semidefinite programming (SDP) if the encoding map is fixed, i.e., $\mathcal{E}=\bar{\mathcal{E}}$ \cite{Reimpell2005,Fletcher2007}:
\begin{alignat}{2}
&\max_{\hat{X}_{\mathcal{R}}}&& \frac{1}{d^{2}} \mathrm{Tr}[ \hat{X}_{\mathcal{R}}(f_{\mathcal{N}}(\hat{X}_{\bar{\mathcal{E}}}))]
\nonumber\\
&\,\,\,\textrm{s.t.}&& \hat{X}_{\mathcal{R}}=\hat{X}_{\mathcal{R}}^{\dagger}\succeq 0,\,\, \mathrm{Tr}_{\mathcal{H}'}\hat{X}_{\mathcal{R}} = \hat{I}_{\mathcal{H}_{n}}.  \label{eq:SDP decoding optimization}
\end{alignat} 
Here, the constraints are due to the CPTP nature of the recovery operation $\mathcal{R}$. The optimal recovery operations $\mathcal{R}^{\star}$ used in Fig.\ \ref{fig:benchmarking single-mode bosonic codes} are obtained by solving the SDP in Eq.\ \eqref{eq:SDP decoding optimization}. To solve each instance of SDP, we used CVX, a package for specifying and solving convex programs \cite{CVX,Grant2008}      

\subsection{Biconvex optimization of single-mode bosonic codes}

Similarly, optimizing an encoding map $\mathcal{E}$ for a given recovery operation $\mathcal{R} = \bar{\mathcal{R}}$ is also a semidefinite programming. Thus, the entire problem of optimizing the set of encoding and recovery maps is a biconvex optimization problem. This idea was used to optimize multi-qubit error-correcting codes in Ref.\ \cite{Kosut2009}. Note that in the context of bosonic quantum error correction, it is also important to impose an energy constraint to the error-correcting codes to make a fair comparison between different bosonic code families. Furthermore in the optimization perspective, it is essential to impose the energy constraint while still preserving the bi-convexity of the problem.  

With these issues in minds, let us consider an energy observable $\hat{E}\in \mathcal{L}(\mathcal{H}_{n})$ and let $\mathrm{Tr}_{\mathcal{H}_{n}}[\hat{E}\hat{\rho}_{\mathcal{E}}]$ be the average energy of the encoding map $\mathcal{E}$. Here, 
\begin{align}
\hat{\rho}_{\mathcal{E}}\equiv \mathcal{E}\Big{(} \frac{1}{d}\sum_{i=0}^{d-1}|i_{\mathcal{H}'}\rangle\langle i_{\mathcal{H}'}| \Big{)}=\frac{1}{d} \mathrm{Tr}_{\mathcal{H}'} \hat{X}_{\mathcal{E}}
\end{align}
is the state resulting from applying $\mathcal{E}$ to the maximally mixed state in $\mathcal{H}'$. Then, the energy constraint is explicitly given by
\begin{align}
\mathrm{Tr}_{\mathcal{H}_{n}}[\hat{E}\hat{\rho}_{\mathcal{E}}] &= \frac{1}{d}\mathrm{Tr}[(\hat{E}\otimes \hat{I}_{\mathcal{H}'}) \hat{X}_{\mathcal{E}}] \le  \bar{E} .  
\end{align}
Therefore, we end up with the following energy-constrained biconvex encoding and decoding optimization \cite{Noh2019}: 
\begin{alignat}{2}
&\max_{\hat{X}_{\mathcal{E}},\hat{X}_{\mathcal{R}}}&& \mathrm{Tr}[\hat{X}_{\mathcal{R}}^{\dagger}f_{\mathcal{N}}(\hat{X}_{\mathcal{E}})], \,\,
\nonumber\\
&\quad \textrm{s.t.}&& \hat{X}_{\mathcal{R}}=\hat{X}_{\mathcal{R}}^{\dagger}\succeq 0,\,\, \mathrm{Tr}_{\mathcal{H}'}\hat{X}_{\mathcal{R}} = \hat{I}_{\mathcal{H}_{n}}, 
\nonumber\\
&&& \hat{X}_{\mathcal{E}}=\hat{X}_{\mathcal{E}}^{\dagger}\succeq 0,\,\, \mathrm{Tr}_{\mathcal{H}_{n}}\hat{X}_{\mathcal{E}} = \hat{I}_{\mathcal{H}'},  \,\,\, \textrm{and}\,\,\, \mathrm{Tr}[(\hat{E}\otimes \hat{I}_{\mathcal{H}'}) \hat{X}_{\mathcal{E}}] \le  \bar{E}d . 
\label{eq:biconvex optimization formulation} 
\end{alignat}
Note that the last constraint is due to the average energy constraint to the encoding maps.

\subsection{Results} 

In principle, a global optimal solution of Eq.\ \eqref{eq:biconvex optimization formulation} can be deterministically found by a global optimization algorithm outlined in Ref.\ \cite{Floudas1990}. To implement the algorithm, however, one should in general solve exponentially many convex sub-problems in the number of complicating variables (responsible for non-convexity of the problem; see Ref.\ \cite{Floudas1990} and also Ref.\ \cite{Huber2019} for more details), which is intractable in our application below. Thus, we instead solve Eq.\ \eqref{eq:biconvex optimization formulation} heuristically by alternating between encoding and recovery optimization (i.e., SDP sub-problems) starting from a random initial encoding map. We generate the random initial code by taking the first $d$ columns of an $n\times n$ Haar random unitary matrix. To solve each SDP sub-problem we used CVX, a package for specifying and solving convex programs \cite{CVX,Grant2008}.  
 
\begin{figure}[t!]
\centering
\includegraphics[width=5.8in]{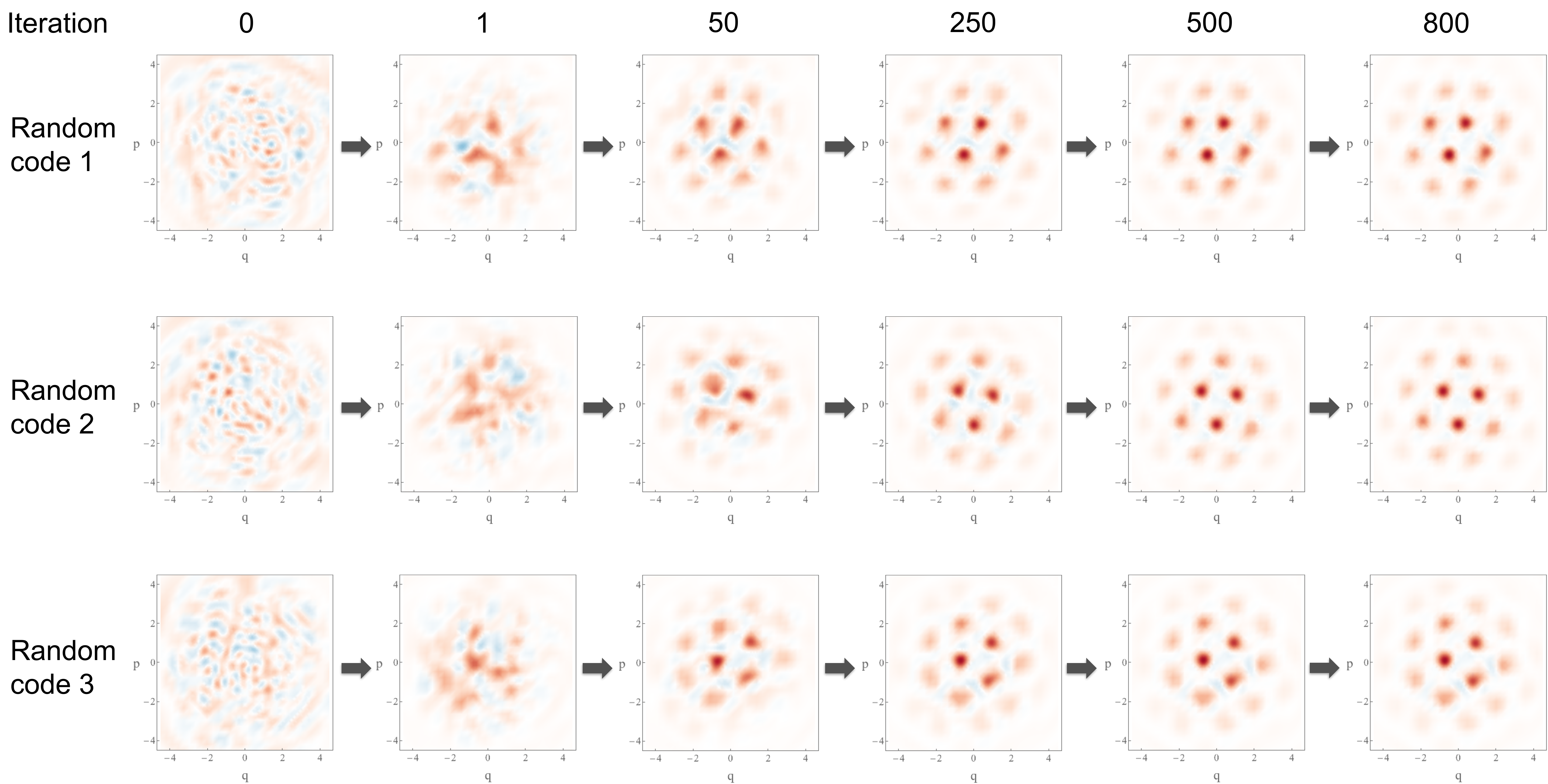}
\caption{[Fig.\ 4 in IEEE Trans. Info. Theory \textbf{65}, 2563–-2582 (2019)] Biconvex optimization of the encoding and recovery maps $\mathcal{E}$ and $\mathcal{R}$ for the bosonic pure-loss channel $\mathcal{N}[\eta,0]$ with $\eta=0.9$, or $\gamma = 1-\eta=  0.1$. We chose $n=20$ and $d=2$ and imposed an average photon number constraint $\mathrm{Tr}[\hat{n}\hat{\rho}_{\mathcal{E}}] \le  \ 3$, where $\hat{\rho}_{\mathcal{E}} = (1/d) \mathrm{Tr}_{\mathcal{H}'}\hat{X}_{\mathcal{E}}$ is the maximally mixed code state. The first column of each row represents the Wigner function of $\hat{\rho}_{\mathcal{E}}$ for a randomly generated encoding map $\mathcal{E}$. From the second to the sixth columns represent the updated code spaces after $1,50,250,500$ and $800$ iterations of the alternating semidefinite programming. }
\label{fig:biconvex optimization combined}
\end{figure}

Let us now specialize Eq.\ \eqref{eq:biconvex optimization formulation} to $\mathcal{N}=\mathcal{N}[\eta,0]$ and $d=2$ to find an optimal qubit-into-an-oscillator code for a bosonic pure-loss channel, subject to an average photon number constraint $\mathrm{Tr}[\hat{n}\hat{\rho}_{\mathcal{E}}] \le \bar{n}_{\textrm{max}}$. To make the optimization tractable, we confine the bosonic Hilbert space to a truncated subspace $\mathcal{H}_{n}\equiv\textrm{span}\lbrace |0\rangle,\cdots,|n-1\rangle \rbrace$ and choose $n\gg \bar{n}_{\textrm{max}}$ to avoid any artifacts caused by truncation. In particular, we use the Kraus representation of a bosonic pure-loss channel $\mathcal{N}[\eta,0](\hat{\rho}) = \sum_{\ell=0}^{n-1}\hat{N}_{\ell}\hat{\rho}\hat{N}_{\ell}^{\dagger}$, where the Kraus operators $\hat{N}_{\ell}$ are given by \cite{Ueda1989,Lee1994,Chuang1997} (see also Table \ref{table:excitation loss errors})  
\begin{equation}
\hat{N}_{\ell} = \sqrt{\frac{(1-\eta)^{\ell}}{\ell!}} \eta^{\frac{\hat{n}}{2}}\hat{a}^{\ell}. 
\end{equation}

In Fig.\ \ref{fig:biconvex optimization combined}, we took $\eta=0.9$, $n=20$, $d=2$ and $\bar{n}_{\textrm{max}}=3$ and plot the Wigner function of the maximally mixed code states of the numerically optimized codes (last column), starting from three different random Haar initial codes (first column). In all instances, the obtained codes are given by a hexagonal-lattice GKP code (see Fig.\ \ref{fig:Wigner functions of bosonic codes for benchmarking}), up to an overall displacement. The optimized code in the second row exhibits the best performance (i.e., $1-F_{\mathcal{N}}^{\star}=0.002092$). 

We emphasize that the biconvex optimization in Eq.\ \eqref{eq:biconvex optimization formulation} explores the most general form of CPTP encoding maps, including the ones involving mixed state encoding. However, from the numerical optimization, we only obtained a pure-state encoding (i.e., $\mathcal{E}(\hat{\rho})=\hat{V}\hat{\rho}\hat{V}^{\dagger}$, where $\hat{V}:\mathcal{H}'\rightarrow\mathcal{H}_{n}$ is an isometry $\hat{V}^{\dagger}\hat{V}=\hat{I}_{\mathcal{H}'}$) as an optimal solution at all iterations of SDP sub-problems. However, we also stress that the alternating semidefinite programming method is not guaranteed to yield a global optimal solution. Despite the latter caveat, the numerical results shown in Fig.\ \ref{fig:biconvex optimization combined} indicate that a hexagonal-lattice GKP code is consistently obtained from independent Haar-random initial codes. Thus, the optimization results suggest that the GKP code family may indeed be the most effective bosonic code family in correcting excitation loss errors. The numerical results also demonstrate the advantage of using an optimal lattice structure that allows the densest sphere packing which, in the case of the $2$-dimensional Euclidean space, is given by the hexagonal-lattice structure \cite{Fejes1942}.

\section{Decoding GKP codes subject to excitation loss errors}
\label{section:Decoding GKP codes subject to excitation loss errors}

It is clear by now that the GKP codes exhibit excellent performance against excitation loss errors and may as well be the optimal codes for this purpose. Again, these results are surprising because the GKP codes are not designed to correct loss errors. Instead, the GKP codes are designed to correct random shift errors. However, this also means that the error recovery schemes that we use for the GKP codes to recover from random shift errors may not really work when they are used to correct for the excitation loss errors. Therefore, to fully take advantage of the excellent error-correcting capability of the GKP codes against loss errors, it is very important to understand how precisely the GKP codes work against excitation loss errors. In other words, we need to understand what the numerically optimized error recovery map $\mathcal{R}^{\star}$ (obtained from an SDP in Eq.\ \eqref{eq:SDP decoding optimization}) does for the GKP codes, ideally in terms of simple operations that we are aware of already.         

As discussed above, it is often said that small loss errors can be regarded as small shift errors and this is why the GKP codes work well against excitation loss errors because the GKP codes can correct small shift errors. This argument may explain the performance of small GKP codes with a small average energy for small loss parameters. However, this argument does not apply for large GKP codes because in that case, even a small fraction of loss can cause a huge shift error in the high-energy domain that is not correctable by the GKP codes. On the other hand, the numerical results in Fig.\ \ref{fig:benchmarking single-mode bosonic codes}(b) and (c) show that the GKP codes perform well even if in the case of large average energy and even for large loss errors (e.g., $\gamma \simeq 0.2$). These numerical results clearly indicate that there are more things going on than ``small loss errors equal small shifts errors''. Below, we will provide an alternative explanation that applies to large GKP codes as well. In one sentence, our explanation is ``the GKP codes work well against loss errors because loss errors can be converted via an amplification to shift errors, which the GKP codes can correct''.

\subsection{Transforming a loss error into a random shift error}

Here, we will show that a bosonic pure-loss channel can be converted via a quantum-limited amplification into a Gaussian random shift error. More precisely, we formulate the following theorem:  

\begin{theorem}[Pure-loss + Amplification = Random shift \cite{Albert2018}]
Let $\mathcal{N}[\eta,0]$ be a bosonic pure-loss channel with a transmissivity $\eta\in[0,1]$. Let $\mathcal{A}[1/\eta , 0]$ be a quantum-limited amplification channel (see Definition \ref{definition:quantum limited amplification}) with gain $G=1/\eta$. Then, we have 
\begin{equation}
\mathcal{A}\Big{[} \frac{1}{\eta} ,0 \Big{]} \cdot \mathcal{N}[\eta,0] = \mathcal{N}_{B_{2}}[\sigma_{\eta,0}], 
\end{equation} 
where the noise variance $(\sigma_{\eta,0})^{2}$ is given by 
\begin{equation}
(\sigma_{\eta,0})^{2} \equiv \frac{1-\eta}{\eta}  = \frac{\gamma}{1-\gamma}.  \label{eq:effective variance of displacement post-amp} 
\end{equation}
Here, $\gamma = 1-\eta$ is the loss probability. See Fig.\ \ref{fig:loss + amp = displacement visualization} for a schematic illustration. \label{theorem:loss plus amplification is displacement post-amplification}
\end{theorem}   

\begin{figure}[t!]
\centering
\includegraphics[width=5.0in]{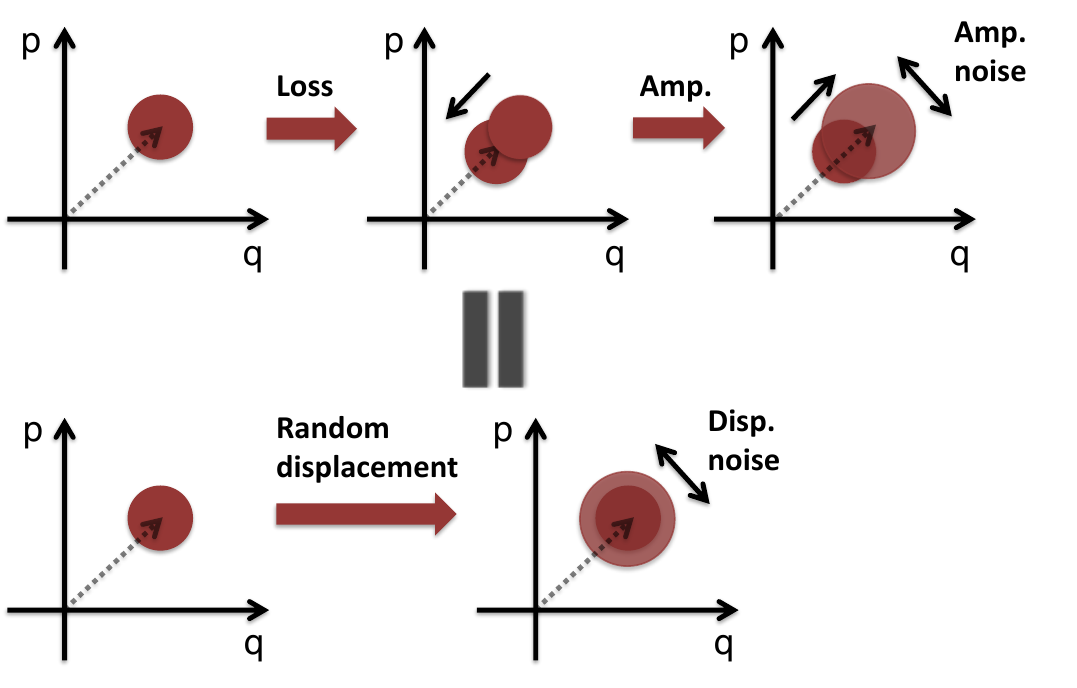}
\caption{Schematic illustration of ``Loss + Amplification = Displacement'' (Theorem \ref{theorem:loss plus amplification is displacement post-amplification}). }
\label{fig:loss + amp = displacement visualization}
\end{figure}

\begin{proof}
The most convenient way to prove Theorem \ref{theorem:loss plus amplification is displacement post-amplification} is to use the notion of Gaussian channels and their characterization. Recall that a Gaussian channel $\mathcal{N}$ is unambiguously characterized by its action on the first two moments of a state, i.e., 
\begin{align}
\boldsymbol{\bar{x}} ' &= \boldsymbol{T} \boldsymbol{\bar{x}}  +\boldsymbol{d}, 
\nonumber\\
\boldsymbol{V} ' &= \boldsymbol{T}  \boldsymbol{V} \boldsymbol{T}^{T} + \boldsymbol{N}.  
\end{align}     
Here, $\boldsymbol{\bar{x}}$ is the mean value of the quadrature operator of a state and $\boldsymbol{V}$ is the covariance matrix of the state. Thus, a Gaussian channel $\mathcal{N}$ is fully characterized by $(\boldsymbol{T}, \boldsymbol{N}, \boldsymbol{d})$. As shown in Section \ref{section:Relevant error models in bosonic systems} (see Table \ref{table:excitation loss errors}), the bosonic pure-loss channel $\mathcal{N}[\eta,0]$ is characterized by 
\begin{align}
\boldsymbol{T}_{1} &= \sqrt{\eta} \boldsymbol{I}_{2}, \quad \boldsymbol{N}_{1} = \frac{1}{2}\sqrt{1-\eta} \boldsymbol{I}_{2}, \quad 
\boldsymbol{d}_{1} =0, 
\end{align}
where $\boldsymbol{I}_{2}$ is the $2\times 2$ matrix. This framework gives a clear picture that the quadrature operators of a state contracts by a factor of $\sqrt{\eta}$ due to the pure-loss error. Moreover, a fraction (i.e., $1-\eta$) of the environmental vacuum noise $\frac{1}{2}\boldsymbol{I}_{2}$ is transferred via the loss error to the system and thus adds noise with a covariance matrix $\frac{1}{2}(1-\eta)\boldsymbol{I}_{2}$.  

The quantum-limited amplification channel $\mathcal{A}[G,0]$ is defined in Definition \ref{definition:quantum limited amplification} and is characterized by 
\begin{align}
\boldsymbol{T}_{2} &= \sqrt{G} \boldsymbol{I}_{2}, \quad \boldsymbol{N}_{2} = \frac{1}{2}\sqrt{G-1} \boldsymbol{I}_{2}, \quad 
\boldsymbol{d}_{2} =0, 
\end{align}
where $G (\ge 1)$ is the gain of the amplifier. Thus, an amplification channel does exactly the opposite to the quadrature operators of a quantum state compared to a loss channel. That is, it amplifies the mean quadrature operators by a factor of $\sqrt{G}$ and in the expense adds a transferred vacuum noise with a covariance matrix $\frac{1}{2}(G-1)\boldsymbol{I}_{2}$. Therefore, by choosing the gain $G$ properly, i.e., $G= \frac{1}{\eta}$, we can compensate the contractive effects of the bosonic pure-loss channel $\mathcal{N}[\eta,0]$ at least at the level of the mean quadrature.  

Putting all these together, we fine that the first two moments of a quantum state $\boldsymbol{\bar{x}}$ and $\boldsymbol{V}$ are transformed via the pure-loss channel $\mathcal{N}[\eta,0]$ into  
\begin{align}
\boldsymbol{\bar{x}} ' &= \sqrt{\eta} \boldsymbol{\bar{x}}  , 
\nonumber\\
\boldsymbol{V} ' &= \eta \boldsymbol{V}  + \frac{1}{2}(1-\eta)\boldsymbol{I}_{2}.  
\end{align}
Then these moments are further transformed via the amplification channel $\mathcal{A}[\frac{1}{\eta},0]$ into 
\begin{align}
\boldsymbol{\bar{x}}'' &= \sqrt{ \frac{1}{\eta} } \boldsymbol{\bar{x}}'    = \boldsymbol{\bar{x}} , 
\nonumber\\
\boldsymbol{V} '' &= \frac{1}{\eta} \boldsymbol{V}'  + \frac{1}{2}\Big{(} \frac{1}{\eta}-1\Big{)}\boldsymbol{I}_{2} 
\nonumber\\
& = \frac{1}{\eta} \Big{(}   \eta \boldsymbol{V}  + \frac{1}{2}(1-\eta)\boldsymbol{I}_{2} \Big{)} + \frac{1}{2}\Big{(} \frac{1}{\eta}-1\Big{)}\boldsymbol{I}_{2}  = \boldsymbol{V} + \Big{(} \frac{1-\eta}{\eta} \Big{)} \boldsymbol{I}_{2}. 
\end{align}
Thus, the net effect is trivial on the mean quadrature as expected. However, there is an added noise to the covariance matrix by $(\frac{1-\eta}{\eta})\boldsymbol{I}_{2}$. Note that this is precisely what a Gaussian random shift error $\mathcal{N}_{B_{2}}[\sigma]$ does to the quadrature operators of a quantum state. In particular, the added noise covariance matrix is given by $\boldsymbol{N} =\sigma^{2} \boldsymbol{I}_{2}$ for the random shift error $\mathcal{N}_{B_{2}}[\sigma]$ (see Table \ref{table:random shift errors}). Thus, we can conclude that the actions of the pure-loss channel $\mathcal{N}[\eta,0]$ followed by the amplification channel $\mathcal{A}[\frac{1}{\eta},0]$ is equivalent to the action of a Gaussian random shift error $\mathcal{N}_{B_{2}}[\sigma_{\eta,0}]$ with the noise variance 
\begin{align}
( \sigma_{\eta,0} )^{2} = \frac{1-\eta}{\eta}. 
\end{align} 
and thus the theorem follows. 
\end{proof}

\subsection{Amplification decoding}

Based on Theorem \ref{theorem:loss plus amplification is displacement post-amplification}, we can immediately come up with a sensible decoding scheme for the GKP codes against excitation loss errors: Given a GKP code and a bosonic pure-loss channel $\mathcal{N}[\eta,0]$ which we want to correct, we can simply apply a quantum-limited amplification $\mathcal{A}[\frac{1}{\eta},0]$ to convert the loss channel into a Gaussian random shift error $\mathcal{N}_{B_{2}}[ \sigma_{\eta,0} ]$ with $\sigma_{\eta,0} = \sqrt{ \frac{1-\eta}{\eta} } = \sqrt{ \frac{\gamma}{1-\gamma} }$. Then, we can use the conventional GKP error correction schemes discussed in detailed in Chapter \ref{chapter:Bosonic quantum error correction} (see Figs.\ \ref{fig:GKP stabilizer measurement position} and \ref{fig:GKP stabilizer measurement momentum}) to correct the resulting random shift error. Then, by using the bounds obtained in Eq.\ \eqref{eq:performance of GKP codes against a Gaussian random shift error}, we find that the failure probability of this error recovery scheme is given by 
\begin{align}
p_{\textrm{fail}}^{(\textrm{sq})}(\eta,0;d) &\le  \exp\Big{[} -\frac{\pi}{4d \cdot (\sigma_{\eta,0})^{2}} \Big{]}  = \exp\Big{[} -\frac{\pi}{4d } \Big{(} \frac{1-\gamma}{\gamma} \Big{)} \Big{]}    , 
\nonumber\\
p_{\textrm{fail}}^{(\textrm{hex})}(\eta,0;d) &\le \exp\Big{[} -\frac{\pi}{2\sqrt{3}d \cdot (\sigma_{\eta,0})^{2}} \Big{]} =  \exp\Big{[} -\frac{\pi}{2\sqrt{3}d}  \Big{(} \frac{1-\gamma}{\gamma} \Big{)} \Big{]} , \label{eq:performance of the single mode GKP codes against loss errors}
\end{align} 
where $d$ is the dimension of the code space (e.g., $d=2$ for qubit-into-an-oscillator codes). Note that the failure probability decrease very rapidly in a non-analytic way as the loss probability $\gamma$ decreases. 

We emphasize that the amplification decoding (i.e., amplification followed by the convention GKP error correction) works for all energy scales. This is because the amplification compensates the large shifts caused by loss errors in the high-energy regime. In particular, the resulting added noise variance $(\sigma_{\eta,0})^{2} = \frac{1-\eta}{\eta}  =\frac{\gamma}{1-\gamma} $ depends only on the loss probability $\gamma$ but not on the size of the GKP codes. This clearly explains why the GKP codes work well against the excitation loss errors even in the high-loss and the large-energy regimes. 

\begin{table}[t!]
  \centering
     \def\arraystretch{1.5}
  \begin{tabular}{ V{3} c V{1.5} c V{1.5} c V{3} }
   \hlineB{3}  
     $\gamma = 0.05$ & Amplification decoding $\mathcal{R}_{\textrm{amp}}$ & Optimal decoding $\mathcal{R}^{\star}$  \\  \hlineB{3} 
     $\bar{n}_{\textrm{max}} = 2$ & N/A & $1.9\times 10^{-3}$  \\  \hlineB{1.5} 
     $\bar{n}_{\textrm{max}} = 5$ & N/A & $2.2\times 10^{-5}$  \\  \hlineB{1.5} 
     $\bar{n}_{\textrm{max}} = 10$ & N/A & $1.5\times 10^{-7}$  \\  \hlineB{1.5} 
     $\bar{n}_{\textrm{max}} \rightarrow \infty$  & $\simeq 1.8\times 10^{-4}$ & N/A  \\  \hlineB{3} 
  \end{tabular}
  \caption{Comparison of the entanglement infidelity $1-F_{e}(\mathcal{M})$ (where $\mathcal{M} = \mathcal{R}\cdot\mathcal{N}\cdot\mathcal{E}$) of the amplification decoding $\mathcal{R} = \mathcal{R}_{\textrm{amp}}$ and the numerically-optimized decoding $\mathcal{R} = \mathcal{R}^{\star}$ for the GKP codes against the bosonic pure-loss channel $\mathcal{N} = \mathcal{N}[\eta,0]$ with a loss probability $\gamma = 1-\eta = 0.05$. For the amplification decoding, the bound on the failure probability for the hexagonal-lattice GKP code is used as an estimate of the entanglement infidelity. For the optimal decoding, the values from the Fig.\ \ref{fig:benchmarking single-mode bosonic codes} are presented.    }
\label{table:amplification decoding verses optimal decoding}
\end{table}

While the amplification decoding can explain all the qualitative features of the excellent performance of the GKP codes numerically observed in Fig.\ \ref{fig:benchmarking single-mode bosonic codes}, it still remains to be answered whether the amplification decoding is indeed identical to the numerically optimized decoding (from an SDP) quantitatively. The answer is unfortunately negative. To see why this is the case, note that the bounds in Eq.\ \eqref{eq:performance of the single mode GKP codes against loss errors} apply to the ideal GKP codes with an infinite energy and therefore will become less favorable as we decrease the allowed energy of the GKP codes. To compare these analytic bounds with the numerical results in Fig.\ \ref{fig:benchmarking single-mode bosonic codes}, let us specialize the bounds in Eq.\ \eqref{eq:performance of the single mode GKP codes against loss errors} to the qubit-into-an-oscillator case (i.e., $d=2$): 
\begin{align}
p_{\textrm{fail}}^{(\textrm{sq})}(\eta,0;d) &\le  \exp\Big{[} -\frac{\pi}{4d \cdot (\sigma_{\eta,0})^{2}} \Big{]}  = \exp\Big{[} -\frac{\pi}{8 } \Big{(} \frac{1-\gamma}{\gamma} \Big{)} \Big{]}    , 
\nonumber\\
p_{\textrm{fail}}^{(\textrm{hex})}(\eta,0;d) &\le \exp\Big{[} -\frac{\pi}{2\sqrt{3}d \cdot (\sigma_{\eta,0})^{2}} \Big{]} =  \exp\Big{[} -\frac{\pi}{4\sqrt{3}}  \Big{(} \frac{1-\gamma}{\gamma} \Big{)} \Big{]} . \label{eq:performance of the single mode GKP codes against loss errors d=2}
\end{align}      
For the hexagonal-lattice GKP code, the above bound yields $\exp[-\frac{\pi}{4\sqrt{3}}  ( \frac{1-\gamma}{\gamma} ) ]  =1.8\times 10^{-4}$ when $\gamma = 0.05$. Again, this bound is obtained by assuming $\bar{n}_{\textrm{max}} \rightarrow \infty$ and the amplification decoding. On the other hand, as shown in Table \ref{table:amplification decoding verses optimal decoding}, the entanglement infidelities of the gkp code family (based on an optimal decoding) are given by $1.9\times 10^{-3}$ ($\bar{n}_{\textrm{max}} =2$), $2.2\times 10^{-5}$ ($\bar{n}_{\textrm{max}} =5$), and $1.5\times 10^{-7}$ ($\bar{n}_{\textrm{max}} =10$) at the same loss parameter $\gamma = 0.05$. That is, the optimal decoding achieves an even lower error rate for a finite-size GKP code than the amplification decoding does for an infinitely large GKP code. That is, the amplification decoding is significantly outperformed by the optimal decoding. While this may be disappointing, it is great news as well because the performance of the amplification decoding is already excellent (as indicated by the rapid non-analytic decrease of the failure probability as the loss probability $\gamma$ decreases) and the optimal decoding will perform even better. 

We remark that the amplification decoding is feasible because the quantum-limited amplification can be implemented by using a vacuum state and a two-mode squeezing operation, which are Gaussian. The only non-Gaussian resource needed in the scheme is the preparation of a computational basis state of the GKP code for implementing the GKP stabilizer measurements (see Figs.\ \ref{fig:GKP stabilizer measurement position} and \ref{fig:GKP stabilizer measurement momentum}). For this reason, we use the amplification decoding scheme when we study the fault-tolerance properties of the GKP codes in Chapter \ref{chapter:Fault-tolerant bosonic quantum error correction}. In particular, we show that it is possible to correct excitation loss errors fault-tolerantly by using the amplification decoding when we concatenate the square-lattice GKP code with the surface code. On the other hand, it is unclear yet whether there is an efficient and structured way to implement the optimal decoding. Moreover, it remains to be answered whether the optimal decoding can be made fault-tolerant.

\section{Open questions}
\label{section:Open questions benchmarking and optimization}

Recall that the alternating SDP method used in Section \ref{section:Optimizing single-mode bosonic codes} is not guaranteed to yield a global optimal solution to the biconvex optimization problem in Eq.\ \eqref{eq:biconvex optimization formulation}. Thus, it remains to be seen whether the hexagonal-lattice GKP code emerges as a global optimal solution when we use optimization methods such as the one in Ref.\ \cite{Huber2019} that are guaranteed to provide a global optimal solution. Moreover, it will be interesting to see in the multi-mode case whether the GKP codes (with a lattice structure allowing the most efficient sphere packing) emerge as an optimal solution from the biconvex optimization.      

Note that while the amplification decoding introduced in Section \ref{section:Decoding GKP codes subject to excitation loss errors} yields a desirable logical error rate that decreases very rapidly as the loss probability $\gamma$ decreases, it is still significantly outperformed by the numerically optimized decoding strategies used in Sections \ref{section:Benchmarking single-mode bosonic codes} and \ref{section:Optimizing single-mode bosonic codes} (see also Table \ref{table:amplification decoding verses optimal decoding}). While the numerically optimized decoding exhibits an excellent performance, it is not clear how this decoding should be implemented in practice. Thus, it will be an interesting research direction to look for an explicit and structured decoding strategy that performs nearly as well as the numerically optimized decoding. Ideally, such a decoding operation should be decomposed into simple operations that can be readily implemented in experiments. 

Note that the quantum-limited amplification used in the amplification decoding is universally applicable to any bosonic states and thus is not specifically tailored to the GKP codes. In the expense of the versatility, the quantum-limited amplification always comes with an inevitable added noise \cite{Heffner1962}. Thus, one possible way to search for such a better decoding scheme is to replace the quantum-limited amplification channel by another amplification channel that is more tailored to the GKP codes.      

\chapter{Fault-tolerant bosonic quantum error correction}
\label{chapter:Fault-tolerant bosonic quantum error correction}

In this chapter, I will discuss fault-tolerant quantum error correction in bosonic systems. This chapter is based on my work on the surface-GKP code in Ref.\ \cite{Noh2020}, which was done in collaboration with Dr.\ Christopher Chamberland. This project was conceived and completed during an internship at IBM T. J. Watson Research Center in the summer of 2019.  


The main goal of this chapter is to investigate the performance of the GKP code in a realistic situation wherein even our attempts to correct for errors can be erroneous. The reasons for choosing to focus on the GKP code are as follows: 
\begin{itemize}
\item As shown in the previous chapter, the GKP code exhibits an excellent, if not optimal, performance against practically relevant excitation loss errors under a set of idealized assumptions. 
\item The GKP code shares many structural similarities with the conventional multi-qubit stabilizer codes. This allows the GKP code to work in concert with multi-qubit stabilizer codes in a seamless way.  
\item Preparation of a GKP state is the only non-Gaussian resource needed to perform a universal set of gates on the GKP code (see Section \ref{section:Translation-symmetric bosonic codes}). Thus, one can simply focus on preparing high-quality GKP states offline. Online non-Gaussian operations (e.g., cubic phase gate or Kerr nonlinearity) are not needed. 
\end{itemize}
Notably, the last property clearly shows that GKP states are a useful non-Gaussian resource. However, ideal GKP states have an infinite energy and thus are not strictly feasible. This means that GKP states are inevitably noisy in realistic situations. It is therefore very important to address the question of whether such noisy GKP states can be used to realize large-scald and fault-tolerant quantum information processing. In this chapter, I will give an affirmative answer in the context of fault-tolerant quantum error correction.


In Section \ref{section:Concatenation of a bosonic code with a multi-qubit code}, I will provide a general overview of concatenation of a bosonic code (even-parity codes, the two-component cat code, and the GKP code) with a multi-qubit code. Then, I will provide an in-depth description of the concatenation of the GKP code with a multi-qubit error-correcting code. Specifically, I will explain how the additional information gathered during the GKP error correction can boost the performance of the next layer of the multi-qubit error correction (see Subsection \ref{subsection:Concatenation of the GKP code with a multi-qubit code}). 

In Section \ref{section:Fault-tolerant bosonic quantum error correction with the surface-GKP code}, I will consider the concatenation of the GKP code with the surface code, i.e., the surface-GKP code, and study the performance of the surface-GKP code, assuming a detailed circuit-level noise model \cite{Noh2020}. In particular, I will demonstrate that fault-tolerant bosonic quantum error correction is possible with the surface-GKP code as long as the noise parameters are smaller than certain fault-tolerance thresholds. The main results and the fault-tolerance thresholds are given in Fig.\ \ref{fig:main results}. Comparison with previous related works is given in Subsection \ref{subsection:Comparison with previous works}. I will conclude the chapter by outlining several open questions in Section \ref{section:Open questions fault-tolerant bosonic quantum error correction}.  

\section{Concatenation of a bosonic code with a multi-qubit code}
\label{section:Concatenation of a bosonic code with a multi-qubit code}

Here, we give a general overview of the concatenation of a single-mode bosonic code with a multi-qubit error-correcting code. 

\subsection{Concatenation of an even-parity code with a multi-qubit code}  

Let us first consider bosonic codes that have an even excitation number parity. Codes of this type include the four-component cat code and the $(1,1)$-binomial code (see Section \ref{section:Rotation-symmetric bosonic codes}). Error recovery processes for these codes are based on the excitation number parity measurement. Such bosonic error correction schemes are hardware efficient because they can be implemented by using a single bosonic mode and an ancilla qubit (e.g., a microwave cavity mode coupled to a transmon qubit in circuit QED systems). In particular, the ancilla qubit is used to measure the excitation number parity.  

One important thing to realize here is that the logical error rates of the even-parity codes cannot be suppressed to an arbitrarily small value by using this minimal architecture (i.e., a single bosonic mode plus an ancilla qubit). Moreover, it is generally not expected either that this minimal error-correction scheme can achieve a logical error rate that is sufficiently low enough to reliably execute a non-trivial quantum algorithm. Thus at some point, it is essential to concatenate such single-mode bosonic codes with a conventional multi-qubit error-correcting code in order to correct errors that are left uncorrected at the bosonic QEC level. Then, it might appear that we need conventional multi-qubit error correction schemes anyway and thus bosonic QEC is not so useful. However, this is certainly not true because, if bosonic QEC is successfully implemented, the error-corrected bosonic qubits will have a lower error rate than that of the best available physical qubit without error correction. If such error-corrected bosonic qubits are used to implement a multi-qubit QEC scheme, the required resource overhead associated with the use of the multi-qubit QEC can be significantly reduced. This is a general idea that applies to any bosonic codes that can correct dominant physical error sources such as excitation loss errors.         

Let us now consider the specifics of the even-parity codes. Suppose that we want to concatenate an even-parity bosonic code with a multi-qubit stabilizer code. Then, we should be able to perform some Clifford operations on the bosonic qubits to measure the stabilizers of the outer multi-qubit code. In circuit QED systems, it is relatively straightforward to implement certain logical gates on the even-parity codes by using a SNAP gate \cite{Krastanov2015,Heeres2015,Reinhold2019,Ma2019} or its variants. Examples of such gates include a single-qubit rotation along the Z axis and a two-qubit gate such as controlled-Z gate (see Subsection \ref{subsection:Logical gates on rotation-symmetric bosonic codes} for more details). However, it is more challenging to implement the logical Hadamard gate on the even-parity bosonic codes \cite{Grimsmo2019}. Hadamard gates are generally essential for multi-qubit QEC schemes. Hence, for the even-parity bosonic codes to be successfully combined with a multi-qubit code, it will be crucial to have a robust and tailored scheme to implement the logical Hadamard gates on the even-parity codes. 

We remark that there is a versatile scheme for implementing a universal set of gates on bosonic codes based on eSWAP gates \cite{Lau2016,Gao2019}. However, this versatile scheme is not necessarily tailored to the even-parity codes Also, this scheme requires clustering of two or four bosonic codes to define a single protected qubit, and we need many of these protected qubits to further concatenate with a multi-qubit QEC scheme. On the other hand, if one can show that eSWAP gates can be implemented robustly by taking advantage of special structures of the gates, the eSWAP-based method will prove to be useful. 

To summarize, while even-parity bosonic qubits may have low idling error rates \cite{Ofek2016,Rosenblum2018,Hu2019,Ma2019a}, it is unclear yet how we can implement all the necessary logical Clifford operations on these bosonic qubits in a robust way.



\subsection{Concatenation of a cat code with a multi-qubit code} 
\label{subsection:Concatenation of a cat code with a multi-qubit code}

Recently, there has been growing interest in concatenating the two-component cat code (instead of the four-component cat code) with a multi-qubit code. The key motivations behind the use of the two-component cat code are different from the ones for the even-parity codes which we discussed above. One clear advantage of the two-component cat code is that it can be implemented more easily than the four-component cat code. That is, low-order nonlinearity suffice to stabilize the two-component cat code manifold. For instance, the two-component cat code can be autonomously stabilized by using an engineered two-photon dissipation 
\begin{align}
\mathcal{D}[ \hat{a}^{2}-\alpha^{2} ],  
\end{align}
which can be realized by coupling the mode $\hat{a}$ with a fast-decaying ancilla mode $\hat{b}$ via a third-order nonlinearity $g\hat{a}^{2}\hat{b}^{\dagger} + g^{*}(\hat{a}^{\dagger})^{2}\hat{b}$ \cite{Mirrahimi2014,Leghtas2015,Touzard2018}. Alternatively, one can also stabilize the two-component cat code by engineering an Hamiltonian of the following form \cite{Puri2017,Grimm2019}
\begin{align}
\hat{H} = -K( \hat{a}^{\dagger} )^{2}\hat{a}^{2} + ( \epsilon_{p}(\hat{a}^{\dagger})^{2} + \epsilon_{p}^{*}\hat{a}^{2} ) , 
\end{align}
and then cooling the system to the ground-state manifold, which is given by the two-component cat code (see Subsection \ref{subsection:Cat codes} for more details). 

Although the two-component cat code can be implemented in a relatively easier way than the four-component cat code, it also has its own drawbacks. That is, the two-component cat code is robust against only the bosonic dephasing errors, but not against the excitation loss errors. This is certainly an issue if we want to construct a bosonic qubit with a low error rate because excitation loss errors are ubiquitous in many realistic bosonic systems. However, it has to be remembered that we will at some point need to concatenate a bosonic code with a multi-qubit code to achieve a sufficiently low error rate for executing a non-trivial quantum algorithm. Therefore, it is not essential to suppress all types of errors at the bosonic QEC level, because errors that are left uncorrected during the cat code QEC can later be taken care of by an outer multi-qubit code.   

With this consideration in mind, let us now take a closer look into the error-correcting capability of the two-component cat code. Recall that a single-excitation loss event causes a logical bit-flip error (or a logical Pauli X error), i.e., 
\begin{align}
\hat{a}|0_{2-\textrm{cat}}^{(\alpha)}\rangle &\propto \hat{a} ( |\alpha\rangle+|-\alpha ) = \alpha ( |\alpha\rangle-|-\alpha ) \propto |1_{2-\textrm{cat}}^{(\alpha)}\rangle, 
\nonumber\\
\hat{a}|1_{2-\textrm{cat}}^{(\alpha)}\rangle &\propto \hat{a} ( |\alpha\rangle-|-\alpha ) = \alpha ( |\alpha\rangle+|-\alpha ) \propto |0_{2-\textrm{cat}}^{(\alpha)}\rangle. 
\end{align}    
Such logical bit-flip errors should later be corrected by using an outer multi-qubit code. If the excitation loss rate is given by $\kappa$, the logical bit-flip error rate is given by 
\begin{align}
\gamma_{\textrm{bit-flip}} = \gamma_{X} \simeq \kappa |\alpha|^{2}, 
\end{align}
where $\bar{n}_{\mathcal{C}_{2-\textrm{cat}}^{(\alpha)}  } \simeq  |\alpha|^{2}$ is the average excitation number of the code \cite{Mirrahimi2014}. On the other hand, the two-component cat code is robust against bosonic dephasing errors and therefore its logical phase-flip rate decreases exponentially as we increase the size of the code $\alpha$: 
\begin{align}
\gamma_{\textrm{phase-flip}} = \gamma_{Z} \simeq 2\kappa_{\phi} |\alpha|^{2} e^{-2|\alpha|^{2}}, 
\end{align}
if the engineered dissipation rate is much higher than the bosonic dephasing rate $\kappa_{\phi}$. See Ref.\ \cite{Mirrahimi2014} for more details. All these suggest that the cat qubits that are constructed using the two-component cat code will have a highly biased noise towards the bit-flip errors (or the Pauli X errors). Again, this clearly illustrates the inability of the two-component cat code in correcting loss errors. On the other hand, the fact that the noise is biased is good news as well because the outer multi-qubit code does not need to correct the Pauli Z errors for most of the time. Instead, they can be more dedicated to correcting the Pauli X errors on the cat qubits that occur due to the excitation loss errors.   

In the past few years, there have been various proposals for taking advantage of the special biased-noise structure of cat qubits \cite{Cohen2017a,Tuckett2018,Tuckett2019,Tuckett2019a,
Guillaud2019,Puri2019}. For instance, In Refs.\ \cite{Cohen2017a,Guillaud2019}, concatenation of the two-component cat code with the repetition code was explored. The repetition code was chosen for the outer multi-qubit code because it can correct Pauli X errors and there is no urgent need to correct Pauli Z errors due to the noise bias. Another approach is to tailor the surface code to biased-noise models. For instance, it has been shown in Refs.\ \cite{Tuckett2018,Tuckett2019,Tuckett2019a} that the fault-tolerance thresholds of the surface code can be relaxed if the noise is biased and the decoder is tailored to the biased noise. In all these schemes, to fully take advantage of the noise bias, it is very important to maintain the noise bias even when logical gates (for implementing an outer multi-qubit code) are being applied to the cat qubits. Indeed, Refs.\ \cite{Guillaud2019,Puri2019} provided schemes for implementing bias-preserving gates on the cat qubits.   

All these recent works on biased-noise cat qubits lead us to a very interesting point. That is, reduction in the logical error rates may not be the only benefit that bosonic QEC provides. The recent progress on the biased-noise cat qubits clearly suggests that bosonic QEC can provide a unique advantage by imposing a special structure on bosonic qubits. In particular, such a special structure can then be used to boost the performance of the outer multi-qubit codes. Thus, bosonic QEC have a great potential in reducing the required resource overhead associated with the use of conventional multi-qubit fault-tolerance schemes.

\subsection{Concatenation of the GKP code with a multi-qubit code}  
\label{subsection:Concatenation of the GKP code with a multi-qubit code}

It is clear by now that bosonic QEC is not just about reducing logical error rates, but also about giving an additional structure to bosonic qubits that can be used to improve the performance of an outer multi-qubit code. This is precisely the case for the GKP codes as well. Here, we will provide a detailed introduction to the concatenation of the square-lattice GKP code with a multi-qubit error-correcting code. Most importantly, we will explain how the additional information gained during the GKP error correction can be used to boost the performance of the outer multi-qubit code \cite{Menicucci2014,Wang2017,Fukui2017,Fukui2018a,Fukui2018b,
Vuillot2019,Fukui2019,Noh2020}.     

Before moving on to any details, we emphasize that the GKP code shares many similarities with multi-qubit stabilizer codes. For instance, the stabilizers of the GKP code are given by displacement operators (analogous to Pauli operators for qubits). Similarly, stabilizers of a multi-qubit stabilizer code are given by a string of Pauli operators. Moreover, any logical Clifford operations on the GKP code can be implemented by using only Gaussian operations (analogous to Clifford operations for qubits). Analogously, any logical Clifford operations on a multi-qubit stabilizer code can be implemented by using only physical Clifford operations. These structural similarities between the GKP code and the conventional stabilizer codes allow them to work in concert very easily.  

Throughout this chapter, we will restrict ourselves to the square-lattice GKP code (encoding a qubit-into-an-oscillator) and will simply refer to it as the GKP code. We will also use the term ``GKP qubit'' to refer to a qubit that is made out of the GKP code. To illustrate that the GKP code can be concatenated with a multi-qubit stabilizer code in a natural way, we review the concatenation of the square-lattice GKP code with the $[[4,1,2]]$ code \cite{Fukui2017,Fukui2018b}. Recall that the stabilizers and the logical Pauli operators of the GKP code are given by
\begin{alignat}{2}
\hat{S}_{q} &= e^{i2\sqrt{\pi}\hat{q}}, & \quad   \hat{Z}_{\textrm{gkp}} &= e^{i\sqrt{\pi}\hat{q}}
\nonumber\\
\hat{S}_{p} &= e^{-i2\sqrt{\pi}\hat{p}} , & \quad  \hat{X}_{\textrm{gkp}} &=   e^{-i\sqrt{\pi}\hat{p}} . 
\end{alignat}
Also, the stabilizers and the logical Pauli operators of the $[[4,1,2]]$ code are given by 
\begin{alignat}{2}
\hat{S}_{Z}^{[1]} &= \hat{Z}_{1}\hat{Z}_{2},  & \quad  \hat{Z}_{L} &=  \hat{Z}_{1}\hat{Z}_{3}, 
\nonumber\\
\hat{S}_{Z}^{[2]} &= \hat{Z}_{3}\hat{Z}_{4},  & \quad \hat{X}_{L} &=  \hat{X}_{1}\hat{X}_{2}, 
\nonumber\\
\hat{S}_{X}^{[1]} &= \hat{X}_{1}\hat{X}_{2}\hat{X}_{3}\hat{X}_{4}.  & & 
\end{alignat}
The logical states of the $[[4,1,2]]$ code are explicitly given by
\begin{align}
|0_{[[4,1,2]]}\rangle &= \frac{1}{\sqrt{2}} ( |0000\rangle + |1111\rangle ), 
\nonumber\\
|1_{[[4,1,2]]}\rangle &= \frac{1}{\sqrt{2}} ( |0011\rangle + |1100\rangle ). 
\end{align}
One can readily check that these logical states are stabilized by the stabilizers $\hat{S}_{Z}^{[1]}$, $\hat{S}_{Z}^{[2]}$, $\hat{S}_{X}^{[1]}$, and are transformed in a desired way by the logical Pauli operators $\hat{Z}_{L}$ and $\hat{X}_{L}$. Note that the $[[4,1,2]]$ code encode one logical qubit as it consists of $4$ physical qubits and has $3$ stabilizers.   

Now to concatenate the GKP code with the $[[4,1,2]]$ code, we need four bosonic modes to construct four GKP qubits. The stabilizers and the logical Pauli operators of these four GKP qubits are given by 
\begin{alignat}{2}
\hat{S}_{q}^{(j)} &= e^{i2\sqrt{\pi}\hat{q}_{j}}, & \quad   \hat{Z}_{\textrm{gkp}}^{(j)} &= e^{i\sqrt{\pi}\hat{q}_{j}}
\nonumber\\
\hat{S}_{p}^{(j)} &= e^{-i2\sqrt{\pi}\hat{p}_{j}} , & \quad  \hat{X}_{\textrm{gkp}}^{(j)} &=   e^{-i\sqrt{\pi}\hat{p}_{j}}, 
\end{alignat}
where $j\in\lbrace 1,2,3,4 \rbrace$ and $\hat{q}_{j}$ and $\hat{p}_{j}$ are the quadrature operators of the $j^{\textrm{th}}$ bosonic mode hosting the $j^{\textrm{th}}$ GKP qubit. Then, the remaining three stabilizers of the concatenated $[[4,1,2]]$-GKP code are given by
\begin{align}
\hat{S}_{Z}^{[1]} &= \hat{Z}_{\textrm{gkp}}^{(1)}\hat{Z}_{\textrm{gkp}}^{(2)} =  e^{i\sqrt{\pi} ( \hat{q}_{1} +  \hat{q}_{2} ) },  
\nonumber\\
\hat{S}_{Z}^{[2]} &= \hat{Z}_{\textrm{gkp}}^{(3)}\hat{Z}_{\textrm{gkp}}^{(4)} =  e^{i\sqrt{\pi} ( \hat{q}_{3} +  \hat{q}_{4} ) },  
\nonumber\\
\hat{S}_{X}^{[1]} &= \hat{X}_{\textrm{gkp}}^{(1)}\hat{X}_{\textrm{gkp}}^{(2)}\hat{X}_{\textrm{gkp}}^{(3)}\hat{X}_{\textrm{gkp}}^{(4)} =  e^{-i\sqrt{\pi} ( \hat{p}_{1} +  \hat{p}_{2} + \hat{p}_{3} + \hat{p}_{4} ) }, 
\end{align}  
corresponding to the three stabilizers of the $[[4,1,2]]$ code. Thus, we need to measure all the quadrature operators 
\begin{align}
\hat{q}_{j} \textrm{ and } \hat{p}_{j} \textrm{ mod }\sqrt{\pi}
\end{align}
for all $j\in\lbrace 1,2,3,4 \rbrace$ to stabilize each bosonic Hilbert space to the GKP code space. Then, we also need to measure 
\begin{align}
\hat{q}_{1} + \hat{q}_{2} , \,\,\, \hat{q}_{3} + \hat{q}_{4}, \textrm{ and } \hat{p}_{1}+\hat{p}_{2}+\hat{p}_{3}+\hat{p}_{4} \textrm{ mod } 2\sqrt{\pi}
\end{align} 
to further stabilizer the four GKP qubits to the $[[4,1,2]]$-GKP code space. We discussed in detail how to measure the quadrature operators $\hat{q}_{j}$ and $\hat{p}_{j}$ modulo $\sqrt{\pi}$ and what to do with the measurement outcomes in Section \ref{section:Translation-symmetric bosonic codes}. Later in the chapter, we will discuss in detail how the multi-GKP-qubit stabilizers (e.g., $\hat{S}_{Z}^{[1]}$, $\hat{S}_{Z}^{[2]}$, $\hat{S}_{X}^{[1]}$ for the $[[4,1,2]]$-GKP code) can be measured when we discuss the surface-GKP code (see, e.g., Fig.\ \ref{fig:Surface code stabilizer measurement}). Here, we will focus instead on what to do with the obtained measurement outcomes.     

In the case of the usual $[[4,1,2]]$ code, a stabilizer measurement yields a binary outcome (i.e., $+1$ or $-1$)
\begin{align}
\hat{S}_{Z}^{[1]} = \pm 1, \quad \hat{S}_{Z}^{[2]} = \pm 1, \quad \hat{S}_{X}^{[1]} = \pm 1, 
\end{align}
because $( \hat{S}_{Z}^{[1]}  )^{2} = ( \hat{S}_{Z}^{[2]}  )^{2} = ( \hat{S}_{X}^{[1]}  )^{2} = \hat{I}$, where $\hat{I}$ is the identity operator. Note that in the case of the $[[4,1,2]]$-GKP code, these binary measurement outcomes correspond to 
\begin{align}
& \hat{q}_{1} + \hat{q}_{2} = \begin{cases}
0 \textrm{ mod } 2\sqrt{\pi} & \hat{S}_{Z}^{[1]} = +1 \\
\sqrt{\pi}  \textrm{ mod } 2\sqrt{\pi} & \hat{S}_{Z}^{[1]} = -1
\end{cases} , 
\nonumber\\
&\hat{q}_{3} + \hat{q}_{4} = \begin{cases}
0 \textrm{ mod } 2\sqrt{\pi} & \hat{S}_{Z}^{[2]} = +1 \\
\sqrt{\pi}  \textrm{ mod } 2\sqrt{\pi} & \hat{S}_{Z}^{[1]} = -1
\end{cases} , 
\nonumber\\
& \hat{p}_{1} + \hat{p}_{2} + \hat{p}_{3} + \hat{p}_{4} = \begin{cases}
0 \textrm{ mod }2\sqrt{\pi} & \hat{S}_{X}^{[1]} = +1 \\
\sqrt{\pi}  \textrm{ mod } 2\sqrt{\pi} & \hat{S}_{X}^{[1]} = -1
\end{cases} . 
\end{align} 
However in the case of the $[[4,1,2]]$-GKP code, the relevant quadrature operators may not take a value that is an integer multiple of $\sqrt{\pi}$. This is especially the case if the measurements are noisy. Nevertheless, we eventually assign a binary value to each stabilizer. To be more specific, we assign $+1$ (or $-1$) to the stabilizer if the measurement outcome of the relevant quadrature is close to an even (or odd) multiple of $\sqrt{\pi}$. For example, given a measurement outcome of the quadrature operator $\hat{q}_{1} + \hat{q}_{2}$ modulo $\sqrt{\pi}$, we assign the value of the stabilizer $\hat{S}_{Z}^{[1]}$ as follows: 
\begin{align}
\hat{S}_{Z}^{[1]} \leftarrow \begin{cases}
+1 &  (n-\frac{1}{2})\sqrt{\pi}< \hat{q}_{1} + \hat{q}_{2} <  (n+\frac{1}{2})\sqrt{\pi} \,\,\, \textrm{for an even }n  \\
-1 &  (n-\frac{1}{2})\sqrt{\pi}< \hat{q}_{1} + \hat{q}_{2} <  (n+\frac{1}{2})\sqrt{\pi} \,\,\, \textrm{for an odd }n 
\end{cases}. 
\end{align} 
Thus, we can see that the multi-GKP-qubit stabilizer measurements are robust against small shift errors, similarly as in the case of the Pauli measurements of the GKP qubits (see Section \ref{section:Translation-symmetric bosonic codes}).   

To get some more intuition on the $[[4,1,2]]$-GKP code, let us consider an explicit error instance. Suppose that the four bosonic modes are initially in the $[[4,1,2]]$-GKP code space. That is, the quadrature operators satisfy 
\begin{align}
&\hat{q}_{j} = \hat{p}_{j} = 0 \textrm{ mod }\sqrt{\pi}, \,\,\, \textrm{for all}\,\,\, j\in\lbrace 1,2,3,4 \rbrace, 
\nonumber\\
&\hat{q}_{1}+\hat{q}_{2} = 0 \textrm{ mod }2\sqrt{\pi},
\nonumber\\
&\hat{q}_{3}+\hat{q}_{4} = 0 \textrm{ mod }2\sqrt{\pi},
\nonumber\\
&\hat{p}_{1} + \hat{p}_{2} + \hat{p}_{3} + \hat{p}_{4} = 0 \textrm{ mod }2\sqrt{\pi} . \label{eq:422-GKP stabilized}
\end{align} 
Then, assume that an independent and identically distributed Gaussian random shift errors are applied to the four bosonic modes, i.e., 
\begin{align}
\hat{q}'_{j} &=\hat{q}_{j} + \xi_{q}^{(j)}, 
\nonumber\\
\hat{p}'_{j} &=\hat{p}_{j} + \xi_{p}^{(j)}, 
\end{align}
where $j\in\lbrace 1,2,3,4 \rbrace$. Here, $\xi_{q}^{(j)}$ and $\xi_{p}^{(j)}$ are the position and the momentum quadrature noise added to the $j^{\textrm{th}}$ mode that are drawn from the distribution $(\xi_{q}^{(1)}, \xi_{p}^{(1)}, \cdots , \xi_{q}^{(4)}, \xi_{p}^{(4)} ) \sim_{\textrm{iid}} \mathcal{N}(0,\sigma^{2})$. Also, $\sigma$ is the standard deviation of the random noise. 

To get the key idea, let us assume that all the stabilizer measurements are noiseless. Noisy stabilizer measurements will be considered below when we discuss the surface-GKP code. Then, consider a specific instance where all the random shifts are small but only the position shift in the first mode is large, i.e., 
\begin{align}
\frac{\sqrt{\pi}}{2} < \xi_{q}^{(1)}  < \frac{3\sqrt{\pi}}{2} \,\,\, \textrm{and}\,\,\,   -\frac{\sqrt{\pi}}{2} < \xi_{q}^{(2)},\cdots, \xi_{q}^{(4)},\xi_{p}^{(1)}, \cdots, \xi_{p}^{(4)}  < \frac{\sqrt{\pi}}{2} . 
\end{align}
Then following the same reasoning given in Section \ref{section:Translation-symmetric bosonic codes}, all the shift errors will be removed except for the position shift in the first mode, which will be incorrectly estimated as
\begin{align}
R_{\sqrt{\pi}}( \xi_{q}^{(1)} ) = \xi_{q}^{(1)}  - \sqrt{\pi}. 
\end{align}  
The definition of the function $R_{s}( z )$ is given in Eq.\ \eqref{eq:definition of the R function}. Thus, this shift error will be under-corrected and thus we are left with a Pauli X error in the first GKP qubit: 
\begin{align}
\exp \Big{[} -i \big{(} \xi_{q}^{(1)} - R_{\sqrt{\pi}} ( \xi_{q}^{(1)} ) \big{)} \hat{p}_{1} \Big{]} = e^{-i\sqrt{\pi}\hat{p}_{1}} = \hat{X}_{\textrm{gkp}}^{(1)}. 
\end{align}
In the case of the usual $[[4,1,2]]$ code, such an X error will be detected by measuring the stabilizer $\hat{S}_{Z}^{[1]} = \hat{Z}_{1}\hat{Z}_{2}$ because it anti-commutes with $\hat{X}_{1}$. We show that the same thing happens in the case of the $[[4,1,2]]$-GKP code as well. To do so, recall that after the GKP error correction, we are left with 
\begin{align}
\hat{q}''_{1} &= \hat{q}'_{1} - R_{\sqrt{\pi}} ( \xi_{q}^{(1)} ) = \hat{q}_{1} + \sqrt{\pi}, 
\nonumber\\
\hat{p}''_{1} &= \hat{p}'_{1} - R_{\sqrt{\pi}} ( \xi_{p}^{(1)} ) = \hat{p}_{1}, 
\nonumber\\
\hat{q}''_{2} &= \hat{q}'_{2} - R_{\sqrt{\pi}} ( \xi_{q}^{(2)} ) = \hat{q}_{2} , 
\nonumber\\
&\,\,\, \vdots
\nonumber\\  
\hat{p}''_{4} &= \hat{p}'_{4} - R_{\sqrt{\pi}} ( \xi_{p}^{(4)} ) = \hat{p}_{4} . 
\end{align}
Then, if we measure the relevant quadrature operators for the $[[4,1,2]]$-code stabilizers, we get
\begin{alignat}{2}
&\hat{q}''_{1} + \hat{q}''_{2} =  \sqrt{\pi} \textrm{ mod }2\sqrt{\pi} &\quad  &\rightarrow \hat{S}_{Z}^{[1]}  =  e^{i\sqrt{\pi} ( \hat{q}''_{1} + \hat{q}''_{2} ) } = -1
\nonumber\\
&\hat{q}''_{3} + \hat{q}''_{4} = 0 \textrm{ mod }2\sqrt{\pi}  &\quad  &\rightarrow  \hat{S}_{Z}^{[2]}  =  e^{i\sqrt{\pi} ( \hat{q}''_{3} + \hat{q}''_{4} ) } = +1
\nonumber\\
&\hat{p}''_{1} + \hat{p}''_{2}  +  \hat{p}''_{3}  + \hat{p}''_{4} = 0 \textrm{ mod }2\sqrt{\pi}  &\quad  &\rightarrow \hat{S}_{X}^{[1]}  =  e^{-i\sqrt{\pi} ( \hat{p}''_{1} + \hat{p}''_{2} + \hat{p}''_{3} + \hat{p}''_{4} ) } = +1 . 
\end{alignat}
Note that the terms like $\hat{q}_{1}+\hat{q}_{2}$ are not presented above because they are $0$ modulo $2\sqrt{\pi}$ (see Eq.\ \eqref{eq:422-GKP stabilized}). Thus, the $[[4,1,2]]$-GKP code can detect the logical Pauli X error on the first GKP qubit $\hat{X}_{\textrm{gkp}}^{(1)}$ via the $\hat{S}_{Z}^{[1]}$ stabilizer measurement, similarly to the usual $[[4,1,2]]$ code.    

Lastly, we discuss the unique advantage of the $[[4,1,2]]$-GKP code over the usual $[[4,1,2]]$ code. That is, we explain how the additional information from the GKP error correction can help the outer $[[4,1,2]]$ code perform better. Note that the usual $[[4,1,2]]$ code cannot correct single-qubit errors although they can detect them: Suppose, for example, that the stabilizer measurement outcomes are given by 
\begin{align}
\hat{S}_{Z}^{[1]} = -1, \quad \hat{S}_{Z}^{[2]} = +1, \quad \hat{S}_{X}^{[1]} = +1 . \label{eq:422 syndrome pattern example X1 vs X2}
\end{align}
As shown above, this might be due to a Pauli X error on the first qubit $\hat{X}_{1}$. However, it may as well be the case that this is due to a Pauli X error on the second qubit $\hat{X}_{2}$ as it produces the same syndrome measurement outcomes. Therefore, we cannot distinguish these two error events especially when all the qubits are equally noisy. This is problematic because if we had a Pauli X error on the first qubit $\hat{X}_{1}$ but believe it was an error on the second qubit $\hat{X}_{2}$ and correct for it, we are causing a logical X error on the $[[4,1,2]]$ code space, i.e., $\hat{X}_{L} = \hat{X}_{1}\hat{X}_{2}$. Therefore if we randomly guess between $\hat{X}_{1}$ and $\hat{X}_{2}$ with an equal probability,      such a logical X error happens with $50\%$ probability given the syndrome pattern in Eq.\ \eqref{eq:422 syndrome pattern example X1 vs X2}. Hence, it is better to discard this specific error instance than to randomly guess. However, such a non-deterministic error detection scheme is clearly not scalable as the success probability decreases exponentially as we repeat the error detection cycles more and more.     

Let us now explain that we are in a better situation if we use GKP qubits instead of usual bare qubits (i.e., the $[[4,1,2]]$-GKP code instead of the usual $[[4,1,2]]$ code). This is because Pauli errors on the GKP qubits are not just given to us from the environment. Instead, we are participating in the Pauli error generation process and therefore get more information about its inner workings. That is, the natural random shift errors do not immediately cause a Pauli error on a GKP qubit. It is only after we incorrectly estimated a large shift error (and then under-correct it) that we have a Pauli error on the GKP qubit. Below, we will see why this is really crucial.       

Recall the example where we have a large position shift error in the first mode, but all the other shift errors are small enough to be corrected by the GKP code. To make the discussion really simple, consider an extreme case where we have
\begin{align}
\xi_{q}^{(1)} = 0.51\sqrt{\pi}, \,\,\, \textrm{and}\,\,\, \xi_{q}^{(2)} = \cdots =\xi_{p}^{(4)} =0.  
\end{align} 
Through the GKP stabilizer measurements, we can measure these shift errors only modulo $\sqrt{\pi}$. Thus, we are given with the following error candidates that are compatible with the GKP stabilizer measurement outcomes: 
\begin{align}
&\xi_{q}^{(1)} \in \lbrace z| z=-0.49\sqrt{\pi}\textrm{ mod }\sqrt{\pi} \rbrace  = \lbrace \cdots , -1.49\sqrt{\pi}, \textcolor{red}{ -0.49\sqrt{\pi} } ,0.51\sqrt{\pi}, 1.51\sqrt{\pi}, \cdots \rbrace, 
\nonumber\\
&\xi_{q}^{(2)} , \cdots, \xi_{p}^{(4)} \in \lbrace z|z=0\textrm{ mod }\sqrt{\pi} \rbrace =  \lbrace \cdots, -\sqrt{\pi}, \textcolor{red}{ 0} ,\sqrt{\pi},\cdots \rbrace . 
\end{align} 
Then, since we perform a maximum likelihood decoding (or a smallest shift decoding), we incorrectly identify that the position shift error in the first mode is $R_{\sqrt{\pi}}( \xi_{q}^{(1)} ) = -0.49\sqrt{\pi}$ as marked in red above. Similarly, we correctly infer that all the other shift errors are $0$. Then, as explained above, this will results in a Pauli X error on the first GKP qubit $\hat{X}_{\textrm{gkp}}^{(1)}$. Of course in practice, we would not know this with certainty. Instead, we are given with the $[[4,1,2]]$-code stabilizer measurement outcomes $\hat{S}_{Z}^{[1]} = -1$, $\hat{S}_{Z}^{[2]} = +1$, and $\hat{S}_{X}^{[1]} = +1$. Then, as explained above, we know that either the first or the second GKP qubit experienced a Pauli X error, and we need to determine which one to correct. 

Again, in the case of the usual $[[4,1,2]]$ code, there is no way for us to tell between the two possibilities $\hat{X}_{1}$ and $\hat{X}_{2}$ in an informed way. On the other hand, this is not the case in the case of the $[[4,1,2]]$-GKP code because we are given with an additional information on the shift errors, i.e., $\xi_{q}^{(1)} = - 0.49\sqrt{\pi}$ mod $\sqrt{\pi}$ and $\xi_{q}^{(2)}=\cdots = \xi_{p}^{(4)} =0$ mod $\sqrt{\pi}$. Given this additional information, we can evaluate the conditional Pauli error probabilities. That is, since each random shift noise follows a Gaussian distribution $\mathcal{N}(0,\sigma^{2})$, we have 
\begin{align}
&\textrm{Pr}[ \hat{X}_{\textrm{gkp}}^{(1)}  \textrm{ happens}\, |\, R_{\sqrt{\pi}}( \xi_{q}^{(1)} ) = -0.49\sqrt{\pi} ] 
\nonumber\\
&=  \frac{ \cdots +  p[\sigma]( -1.49\sqrt{\pi} ) + p[\sigma]( 0.51\sqrt{\pi} )  + \cdots  }{ \cdots +  p[\sigma]( -1.49\sqrt{\pi} ) + p[\sigma]( -0.49\sqrt{\pi} ) + p[\sigma]( 0.51\sqrt{\pi} ) + p[\sigma]( 1.51\sqrt{\pi} )  + \cdots }
\nonumber\\
&\simeq  \frac{  p[\sigma]( 0.51\sqrt{\pi} )    }{  p[\sigma]( -0.49\sqrt{\pi} ) + p[\sigma]( 0.51\sqrt{\pi} )   }  \simeq \frac{1}{2},  
\end{align}
and 
\begin{align}
\textrm{Pr}[ \hat{X}_{\textrm{gkp}}^{(2)} \textrm{ happens}\, |\, R_{\sqrt{\pi}}( \xi_{q}^{(2)} ) =  0  ] &=  \frac{ \cdots +  p[\sigma]( -\sqrt{\pi} ) + p[\sigma]( \sqrt{\pi} )   + \cdots  }{\cdots +  p[\sigma]( -\sqrt{\pi} )  + p[\sigma]( 0 )  + p[\sigma]( \sqrt{\pi} )   + \cdots }
\nonumber\\
&\simeq \frac{  p[\sigma]( -\sqrt{\pi} ) + p[\sigma]( \sqrt{\pi} )   }{ p[\sigma]( 0 )  } \simeq 0. 
\end{align}
Here, $p[\sigma](z) \equiv \frac{1}{\sqrt{2\pi \sigma^{2}}} \exp[ -\frac{z^{2}}{2\sigma^{2}}] $ is the probability density function of the Gaussian distribution $\mathcal{N}(0,\sigma^{2})$. These conditional probabilities indicate that given the shift error information $R_{\sqrt{\pi}}( \xi_{q}^{(1)} ) = -0.49\sqrt{\pi}$, it is quite likely that the first GKP qubit underwent a Pauli X error. On the other hand, given the information $R_{\sqrt{\pi}}( \xi_{q}^{(2)} ) =  0$, it is extremely unlikely that the second GKP qubit went through a Pauli X error. Thus, we should clearly correct for the Pauli X error on the first GKP qubit $\hat{X}_{\textrm{gkp}}^{(1)}$. The chance that this inference is incorrect is extremely low as the conditional probability that the second GKP qubit had a Pauli X error nearly vanishes. This way, we can make a much more informed decision with the $[[4,1,2]]$-GKP code, than in the case of the usual $[[4,1,2]]$ code. It clearly illustrate the unique advantage of the GKP qubits over the usual bare qubits.    

Note that the reason why the first GKP qubit is more likely to have a Pauli error is because the decision on the shift correction is made near the decision boundary $\pm \frac{\sqrt{\pi}}{2}$. Intuitively, we can expect that decisions that are made close to the decision boundary are very sensitive to even slight perturbations. On the other hand, the second GKP qubit is less likely to have a Pauli error because the decision on the shift correction is made deep inside the bulk. In this case, we can intuitively expect that the decision is not significantly influenced by small perturbations. These intuitions will be quantitatively justified below.  

The implication of the above intuition in the context of the $[[4,1,2]]$-GKP code is as follows: When we are given with multiple Pauli error candidates that are compatible with the $[[4,1,2]]$-code stabilizer measurement outcomes, we should correct for the Pauli error on the GKP qubit where the decision on the shift correction is made closest to the decision boundary. Note that if the standard deviation of the shift errors $\sigma$ is sufficiently small (so that the failure probability of the GKP code is tiny), most uncorrectable large shifts occur near the decision boundary $\pm \frac{\sqrt{\pi}}{2}$. Therefore, the corresponding erroneous GKP qubit will be likely (and correctly) identified as the most unreliable GKP qubit. Thus, the $[[4,1,2]]$-GKP code can correct many single-GKP-qubit error events despite the fact that the usual $[[4,1,2]]$ code is not capable of correcting single-qubit errors. Later in the chapter, we will demonstrate an analogous performance improvement using the surface code as an outer multi-qubit code.         

\subsection{Conditional Pauli error probability on a GKP qubit}

To quantitatively justify the intuition that decisions made closer to the decision boundary $\pm \frac{\sqrt{\pi}}{2}$, we explicitly evaluate the conditional Pauli error probability given a shift error information. To do so, let us consider a single GKP qubit subject to a Gaussian random shift error channel $\mathcal{N}_{B_{2}}[\sigma]$, i.e., 
\begin{align}
\hat{q}&\rightarrow \hat{q}+\xi_{q}, 
\nonumber\\
\hat{p}&\rightarrow \hat{p}+\xi_{p},  
\end{align}
where $\xi_{q}$ and $\xi_{p}$ follow a Gaussian random distribution with zero mean and standard deviation $\sigma$, i.e., $\xi_{q},\xi_{p}\sim \mathcal{N}(0,\sigma)$. As discussed in Section \ref{section:Translation-symmetric bosonic codes}, if the random shift $\xi_{q}$ (or $\xi_{p}$) lies in the range $|\xi_{q}-n\sqrt{\pi}|< \sqrt{\pi}/2$ (or $|\xi_{p}-n\sqrt{\pi}|< \sqrt{\pi}/2$) for an odd integer $n$, the GKP error correction protocol results in a Pauli $X$ (or $Z$) error on the GKP qubit. Note that this happens with probability $p_{\textrm{err}}(\sigma)$, where $p_{\textrm{err}}(\sigma)$ is defined as 
\begin{align}
p_{\textrm{err}}(\sigma) &\equiv  \sum_{n\in\mathbb{Z}} \frac{1}{\sqrt{2\pi\sigma^{2}}} \int_{(2n+\frac{1}{2} ) \sqrt{\pi} }^{ (2n+\frac{3}{2})\sqrt{\pi} } d\xi  \exp\Big{[} -\frac{\xi^{2}}{2\sigma^{2}} \Big{]}. \label{eq:definition of perr}
\end{align}
Now, consider a specific instance where, for example, the $\hat{S}_{q}$ stabilizer measurement (i.e., the position measurement modulo $\sqrt{\pi}$) informs us that $\xi_{q}$ is given by $\xi_{q} = z + n\sqrt{\pi}$ for some interger $n$ and $|z| < \sqrt{\pi}/2$. Then, since odd $n$ corresponds to a Pauli $X$ error and even $n$ corresponds to the no error case, we can infer that, given the measured value $z$, there is a Pauli $X$ error with probability $p[\sigma](z)$ where $p[\sigma](z)$ is defined as  
\begin{align}
p[\sigma](z) &\equiv \frac{ \sum_{n\in\mathbb{Z}} \exp[- (z-(2n+1)\sqrt{\pi})^{2}   / (2\sigma^{2})  ] }{ \sum_{n\in\mathbb{Z}} \exp[- (z-n\sqrt{\pi})^{2}  / (2\sigma^{2}) ]  }.  \label{eq:Definition of the p function}
\end{align}

\begin{figure}[t!]
\centering
\includegraphics[width=4.5in]{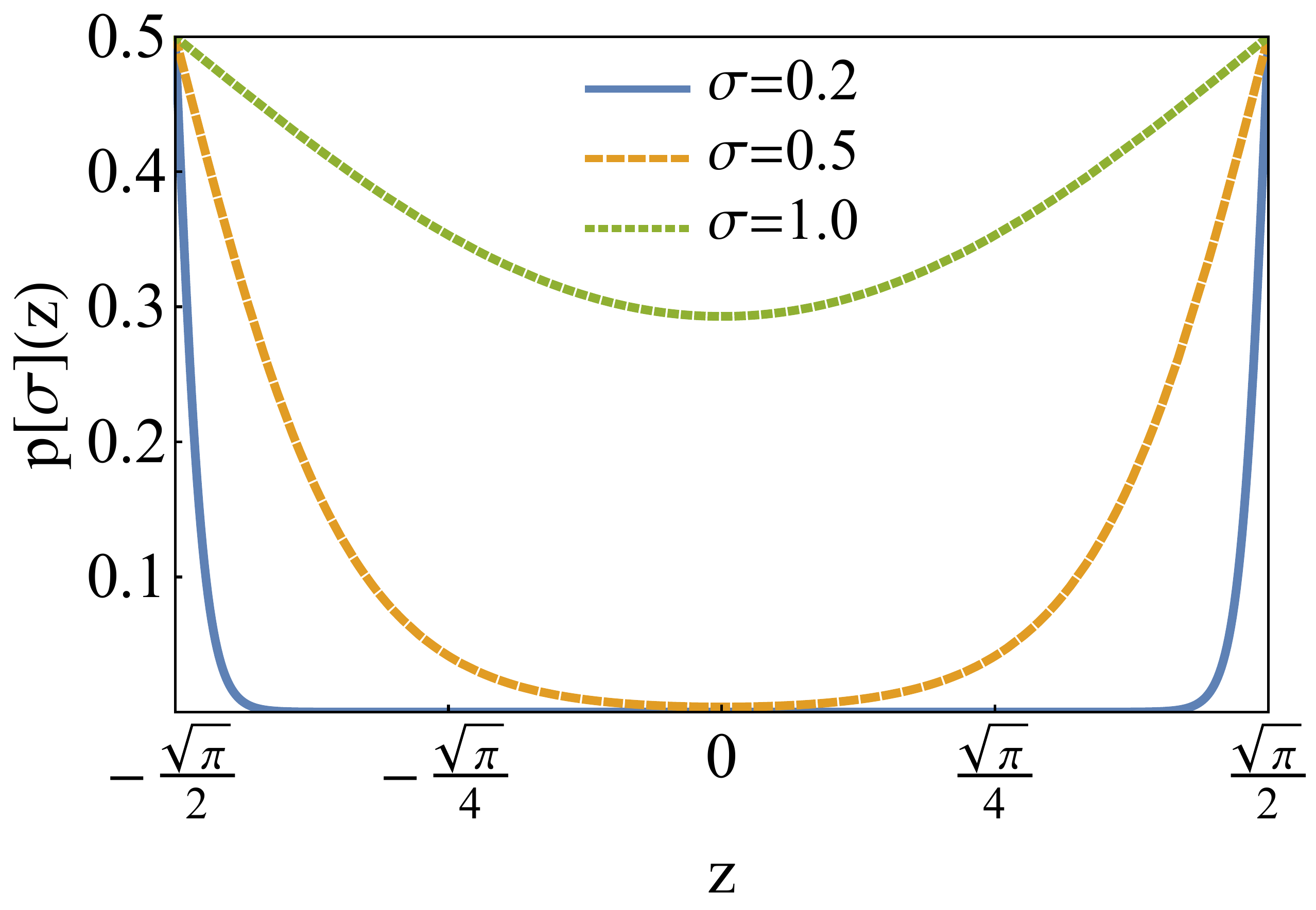}
\caption{[Fig.\ 2 in PRA \textbf{101}, 012316 (2020)] $p[\sigma](z)$ for $\sigma = 0.2$, $0.5$, and $1$. $p[\sigma](z)$ is defined in Eq.\ \eqref{eq:Definition of the p function} and represents the conditional probability of having a Pauli $X$ (or $Z$) error, given the measurement outcome $\xi_{q}=z+n\sqrt{\pi}$ (or $\xi_{p}=z+n\sqrt{\pi}$) for some integer $n$.  }
\label{fig:conditional probability p}
\end{figure}

As shown in Fig.\ \ref{fig:conditional probability p}, the conditional probability $p[\sigma](z)$ becomes larger as $z$ gets closer to the decision boundary $\pm \sqrt{\pi}/2$. Therefore, if the measured shift value modulo $\sqrt{\pi}$ is close to $\pm\sqrt{\pi}/2$, we know that this specific instance of the GKP error correction is less reliable. In other words, the corresponding GKP qubit is more likely to have experienced a Pauli error. This way, the GKP error correction protocol not only corrects the small shift errors but also informs us how reliable the correction is. Various ways of incorporating this additional information in the next level of concatenated error correction have been studied in Refs.\ \cite{Wang2017,Fukui2017,Fukui2018a,Fukui2018b,Fukui2019,Vuillot2019}. In the rest of this chapter, we will show how this additional information for the GKP qubits can be used to boost the performance of the surface code, assuming a detailed circuit-level noise model \cite{Noh2020}.

\section{Fault-tolerant bosonic quantum error correction with the surface-GKP code}
\label{section:Fault-tolerant bosonic quantum error correction with the surface-GKP code}

Here, we demonstrate that fault-tolerant bosonic quantum error correction is possible with the surface-GKP code. Furthermore, we establish fault-tolerance thresholds assuming a detailed circuit-level noise model. The materials in this section are based on Ref.\ \cite{Noh2020}. 

\subsection{GKP qubit}

A pedagogic introduction to the GKP code is provided in Section \ref{section:Translation-symmetric bosonic codes}. Here, we recall several facts about a GKP qubit that are essential for understanding the results presented in this section. Note that we only consider the square-lattice GKP code in this section. Therefore, we will simply refer to the square-lattice GKP code as the GKP code and drop all the superscripts ${}^{(\textrm{sq})}$ that specify the square-lattice structure of the code. Also, we refer to a qubit that is made out of the GKP code as a GKP qubit. 

Stabilizers of the GKP code are given by 
\begin{align}
\hat{S}_{q} \equiv \exp[ i2\sqrt{\pi}\hat{q} ] , \quad \hat{S}_{p} \equiv \exp[ -i2\sqrt{\pi}\hat{p} ] . 
\end{align}
Measuring these two commuting stabilizers is equivalent to measuring the position and momentum operators $\hat{q}$ and $\hat{p}$ modulo $\sqrt{\pi}$. Therefore, any phase space shift error $\exp[i(\xi_{p}\hat{q} - \xi_{q}\hat{p})]$ acting on the ideal GKP qubit can be detected and corrected as long as $|\xi_{q}|,|\xi_{p}| < \sqrt{\pi}/2$. Explicitly, the computational basis states of the ideal GKP qubit are given by 
\begin{align}
|0_{\textrm{gkp}}\rangle &= \sum_{n\in\mathbb{Z}} |\hat{q} = 2n\sqrt{\pi} \rangle,
\nonumber\\
|1_{\textrm{gkp}}\rangle &= \sum_{n\in\mathbb{Z}} |\hat{q} = (2n+1)\sqrt{\pi} \rangle. 
\end{align}
Also, the complementary basis states $|\pm_{\textrm{gkp}}\rangle \equiv \frac{1}{\sqrt{2}}(|0_{\textrm{gkp}}\rangle \pm|1_{\textrm{gkp}}\rangle)$ are given by 
\begin{align}
|+_{\textrm{gkp}}\rangle &=  \sum_{n\in\mathbb{Z}} |\hat{p} = 2n\sqrt{\pi} \rangle,
\nonumber\\
|-_{\textrm{gkp}}\rangle &=  \sum_{n\in\mathbb{Z}} |\hat{p} = (2n+1)\sqrt{\pi} \rangle. 
\end{align}
Clearly, all these basis states have $\hat{q} = \hat{p} =0$ modulo $\sqrt{\pi}$ and thus are stabilized by $\hat{S}_{q}$ and $\hat{S}_{p}$. Pauli operators of the GKP qubit are given by the square root of the stabilizers, i.e., 
\begin{align}
\hat{Z}_{\textrm{gkp}} &= (\hat{S}_{q})^{\frac{1}{2}} = \exp[ i\sqrt{\pi}\hat{q}] , 
\nonumber\\
\hat{X}_{\textrm{gkp}} &= (\hat{S}_{p})^{\frac{1}{2}} = \exp [ -i\sqrt{\pi}\hat{p}] . 
\end{align}
Clifford operations \cite{Gottesman1999} on the GKP qubits can be implemented by using only Gaussian operations. More explicitly, generators of the Clifford group, $\hat{S}_{\textrm{gkp}}, \hat{H}_{\textrm{gkp}}$ and $\textrm{CNOT}_{\textrm{gkp}}^{j\rightarrow k}$ are given by 
\begin{align}
\hat{S}_{\textrm{gkp}} &= \exp\Big{[} i\frac{\hat{q}^{2}}{2} \Big{]} , 
\nonumber\\
\hat{H}_{\textrm{gkp}} &= \exp\Big{[} i \frac{\pi}{2}\hat{a}^{\dagger}\hat{a} \Big{]}  , 
\nonumber\\
\textrm{CNOT}_{\textrm{gkp}}^{j\rightarrow k} &= \textrm{SUM}_{j\rightarrow k} \equiv  \exp[ - i \hat{q}_{j}\hat{p}_{k}], \label{eq:Clifford gates for GKP qubits}
\end{align}   

\begin{figure}[t!]
\centering
\includegraphics[width=5.9in]{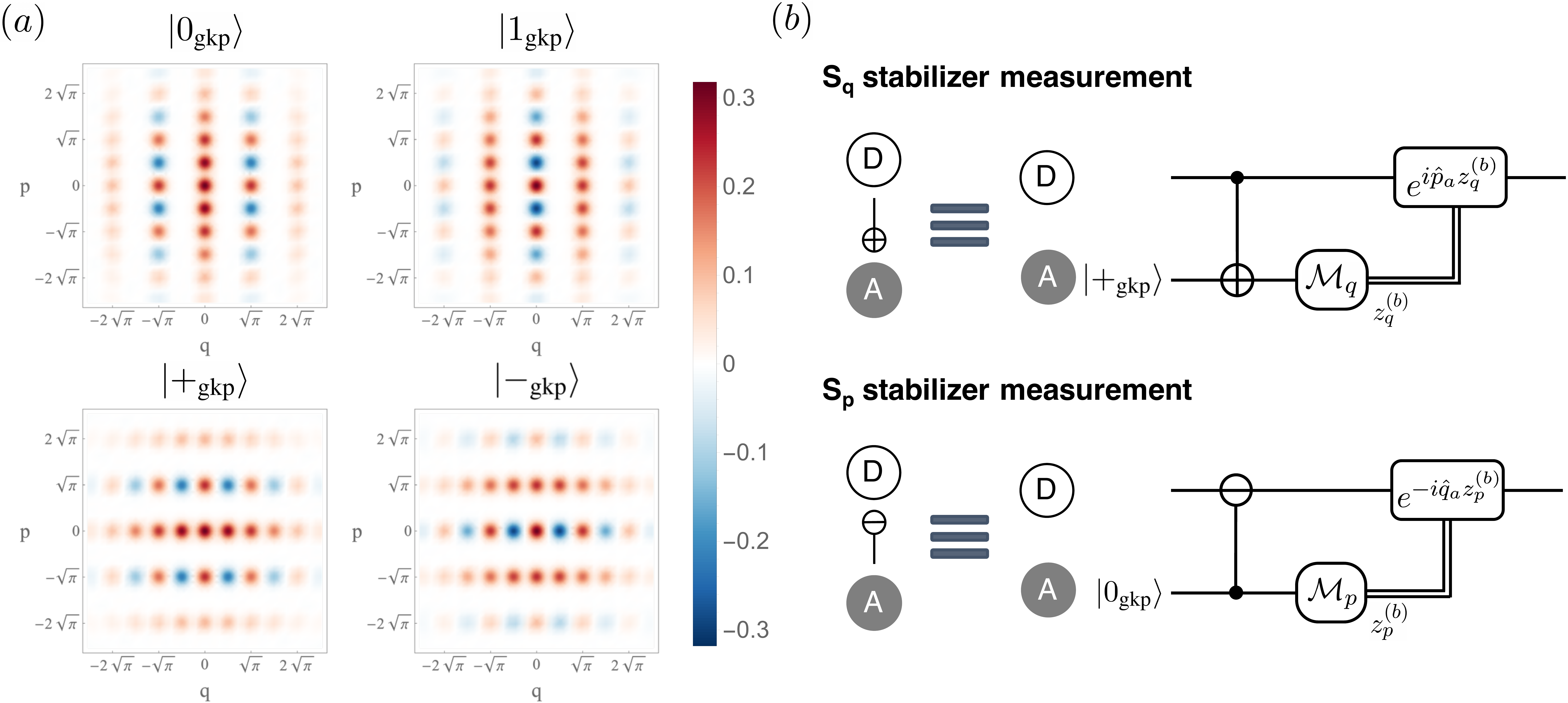}
\caption{[Fig.\ 1 in PRA \textbf{101}, 012316 (2020)] (a) Computational basis states ($|0_{\textrm{gkp}}\rangle$, $|1_{\textrm{gkp}}\rangle$) and complementary basis states ($|+_{\textrm{gkp}}\rangle$, $|-_{\textrm{gkp}}\rangle$) of an approximate GKP qubit with an average photon number $\bar{n}=5$. (b) Circuits for measuring the $\hat{S}_{q}$ and $\hat{S}_{p}$ stabilizers. $\mathcal{M}_{q}$ and $\mathcal{M}_{p}$ represent the homodyne measurement of the position and momentum operators, respectively. Also, the controlled-$\oplus$ symbol represents the SUM gate and similarly the controlled-$\ominus$ symbol represents the inverse-SUM gate (see Eq.\ \eqref{eq:Clifford gates for GKP qubits}). Note that the size of the correction shifts $\exp[i\hat{p}_{a}z_{q}^{(b)}]$ and $\exp[-i\hat{q}_{a}z_{p}^{(b)}]$ in the $\hat{S}_{q}$ and $\hat{S}_{p}$ stabilizer measurements are determined by the homodyne measurement outcomes $z_{q}^{(b)}$ and $z_{p}^{(b)}$.   }
\label{fig:GKP qubit fundamentals}
\end{figure}

The measurements of the GKP stabilizers $\hat{S}_{q}$ and $\hat{S}_{p}$ can be respectively performed by preparing an ancilla GKP state $|+_{\textrm{gkp}}\rangle$ or $|0_{\textrm{gkp}}\rangle$, and then applying the $\textrm{SUM}_{D \rightarrow A}$ or $\textrm{SUM}_{A\rightarrow D}^{\dagger}$ gate, and finally measuring the position or the momentum operator of the ancilla mode via a homodyne detection (see Fig.\ \ref{fig:GKP qubit fundamentals}(b)). Here, $D$ refers to the data mode and $A$ refers to the ancilla mode. Note that the only non-Gaussian resources required for the GKP-stabilizer measurements are the ancilla GKP states $|0_{\textrm{gkp}}\rangle$ and $|+_{\textrm{gkp}}\rangle$.

\subsection{The surface code with GKP qubits}

Recall that shift errors of size larger than $\sqrt{\pi}/2$ cannot be corrected by the single-mode GKP code. Here, to correct arbitrarily large shift errors, we consider the concatenation of the GKP code with the surface code \cite{Bravyi1998,Dennis2002,Fowler2012}, namely, the surface-GKP code. Specifically, we use the family of rotated surface codes \cite{Bombin2007,Tomita2014} that only requires $d^{2}$ data qubits and $d^{2}-1$ syndrome qubits to get a distance-$d$ code. Note that the distance-$d$ surface code can correct arbitrary qubit errors of weight less than or equal to $\lfloor \frac{d-1}{2} \rfloor$.  

\begin{figure}[t!]
\centering
\includegraphics[width=5.9in]{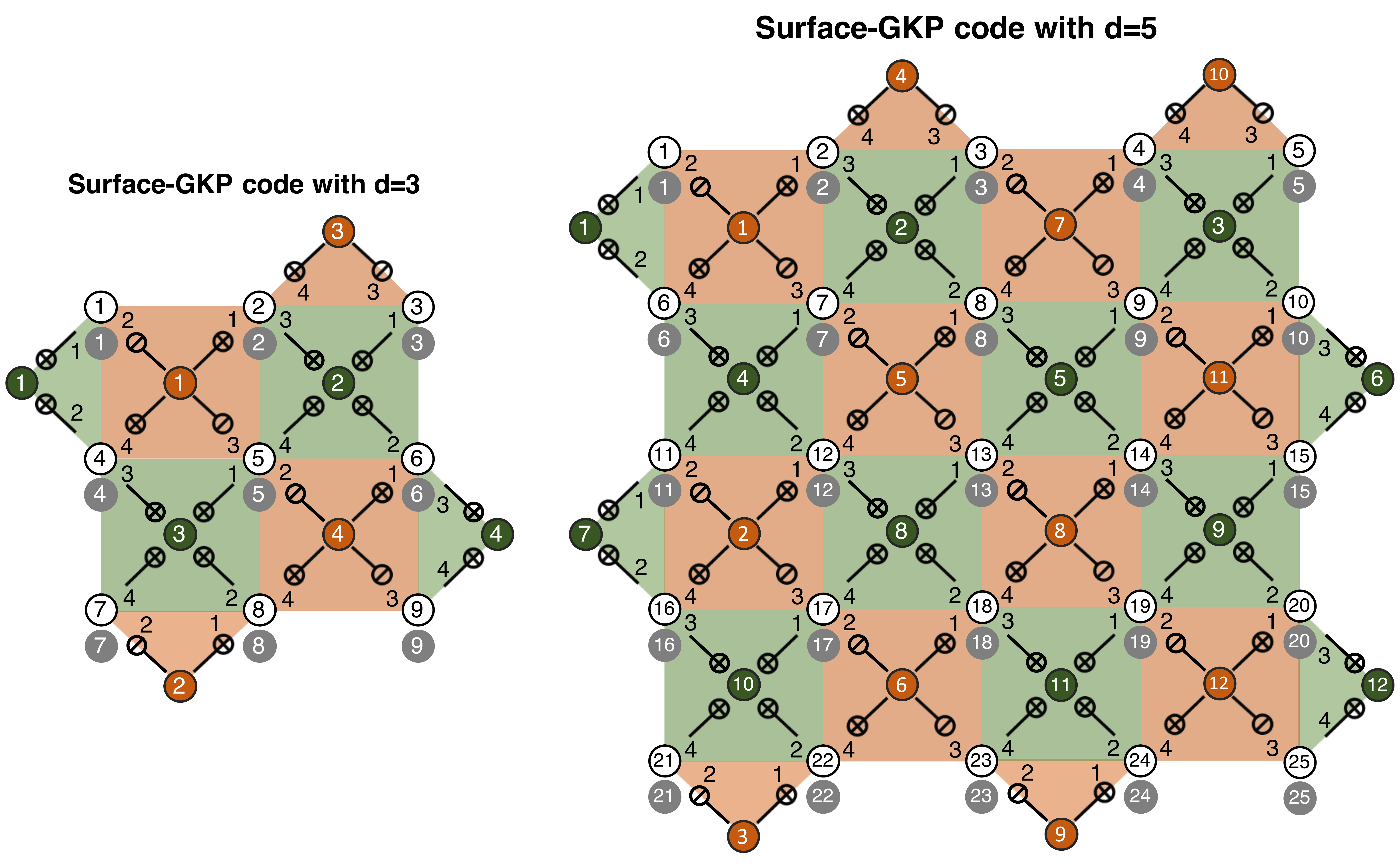}
\caption{[Fig.\ 3 in PRA \textbf{101}, 012316 (2020)] The surface-GKP codes with $d=3$ and $d=5$. White circles represent the data GKP qubits and grey circles represent the ancilla GKP qubits that are used to measure GKP stabilizers of each data GKP qubit. Green and orange circles represent the syndrome GKP qubits that are used to measure the $Z$-type and $X$-type surface code stabilizers of the data GKP qubits, respectively. In general, there are $d^{2}$ data GKP qubits and $(d^{2}-1)/2$ $Z$-type and $X$-type syndrome GKP qubits. See also Fig.\ \ref{fig:Surface code error propagation} for the reason behind our choice of inverse-SUM gates in the $X$-type stabilizer measurements.    }
\label{fig:Surface-GKP codes}
\end{figure}

The layout for the data and ancilla qubits of the surface-GKP code is given in Fig.\ \ref{fig:Surface-GKP codes}. Each of the $d^{2}$ data qubits (white circles in Fig.\ \ref{fig:Surface-GKP codes}) corresponds to a GKP qubit. That is, the distance-$d$ surface-GKP code is stabilized by the following $2d^{2}$ GKP stabilizers 
\begin{align}
\hat{S}_{q}^{(k)} \equiv \exp[ i2\sqrt{\pi}\hat{q}_{k} ], \quad \hat{S}_{p}^{(k)} \equiv \exp[ -i2\sqrt{\pi}\hat{p}_{k} ], 
\end{align}
for $k\in \lbrace 1,\cdots, d^{2}  \rbrace$. These GKP stabilizers are measured by $d^{2}$ ancilla GKP qubits (grey circles in Fig.\ \ref{fig:Surface-GKP codes}) using the circuits given in Fig.\ \ref{fig:GKP qubit fundamentals}(b). Moreover, the data GKP qubits are further stabilized by the $d^{2}-1$ surface code stabilizers. For example, in the $d=3$ case, the $8$ surface code stabilizers are explicitly given by    
\begin{alignat}{2}
\hat{S}_{Z}^{[1]} &= \hat{Z}_{\textrm{gkp}}^{(1)}\hat{Z}_{\textrm{gkp}}^{(4)},\quad & \hat{S}_{Z}^{[2]} &= \hat{Z}_{\textrm{gkp}}^{(2)}\hat{Z}_{\textrm{gkp}}^{(3)}\hat{Z}_{\textrm{gkp}}^{(5)}\hat{Z}_{\textrm{gkp}}^{(6)}, 
\nonumber\\
\hat{S}_{Z}^{[3]} &= \hat{Z}_{\textrm{gkp}}^{(4)}\hat{Z}_{\textrm{gkp}}^{(5)}\hat{Z}_{\textrm{gkp}}^{(7)}\hat{Z}_{\textrm{gkp}}^{(8)}, \quad & \hat{S}_{Z}^{[4]} &= \hat{Z}_{\textrm{gkp}}^{(6)}\hat{Z}_{\textrm{gkp}}^{(9)}, \label{eq:surface code stabilizers Z type d=3}
\end{alignat}
and
\begin{align}
\hat{S}_{X}^{[1]} &= (\hat{X}_{\textrm{gkp}}^{(1)})^{\dagger}\hat{X}_{\textrm{gkp}}^{(2)}\hat{X}_{\textrm{gkp}}^{(4)}(\hat{X}_{\textrm{gkp}}^{(5)})^{\dagger},\quad  \hat{S}_{X}^{[2]} = (\hat{X}_{\textrm{gkp}}^{(7)})^{\dagger}\hat{X}_{\textrm{gkp}}^{(8)},  
\nonumber\\
\hat{S}_{X}^{[3]} &= \hat{X}_{\textrm{gkp}}^{(2)}(\hat{X}_{\textrm{gkp}}^{(3)})^{\dagger}, \quad  \hat{S}_{X}^{[4]} = (\hat{X}_{\textrm{gkp}}^{(5)})^{\dagger}\hat{X}_{\textrm{gkp}}^{(6)}\hat{X}_{\textrm{gkp}}^{(8)}(\hat{X}_{\textrm{gkp}}^{(9)})^{\dagger}, \label{eq:surface code stabilizers X type d=3}
\end{align}
where $\hat{Z}_{\textrm{gkp}}^{(k)} \equiv \exp[i\sqrt{\pi}\hat{q}_{k}]$ and $\hat{X}_{\textrm{gkp}}^{(k)} \equiv \exp[-i\sqrt{\pi}\hat{p}_{k}]$ (see Fig.\ \ref{fig:Surface-GKP codes}).

As shown in Fig.\ \ref{fig:Surface code stabilizer measurement}, the $Z$-type surface code stabilizers are measured by the $Z$-type GKP syndrome qubits (green circles in Fig.\ \ref{fig:Surface-GKP codes}) by using the SUM gates $\textrm{SUM}_{a\rightarrow e},\cdots,\textrm{SUM}_{d\rightarrow e}$ and the position homodyne measurement $\mathcal{M}_{q}$. Similarly, the $X$-type surface code stabilizers are measured by the $X$-type GKP syndrome qubits (orange circles in Fig.\ \ref{fig:Surface-GKP codes}) by using the SUM and the inverse-SUM gates $\textrm{SUM}_{e\rightarrow a}^{\dagger},\textrm{SUM}_{e\rightarrow b},\textrm{SUM}_{e\rightarrow c},\textrm{SUM}_{e\rightarrow d}^{\dagger}$ and the momentum homodyne measurement $\mathcal{M}_{p}$. Note that all the $Z$-type and $X$-type surface code stabilizers can be measured in parallel without conflicting with each other, if the SUM and the inverse-SUM gates are executed in an order that is specified in Figs.\ \ref{fig:Surface-GKP codes} and \ref{fig:Surface code stabilizer measurement}.

\begin{figure}[t!]
\centering
\includegraphics[width=4.0in]{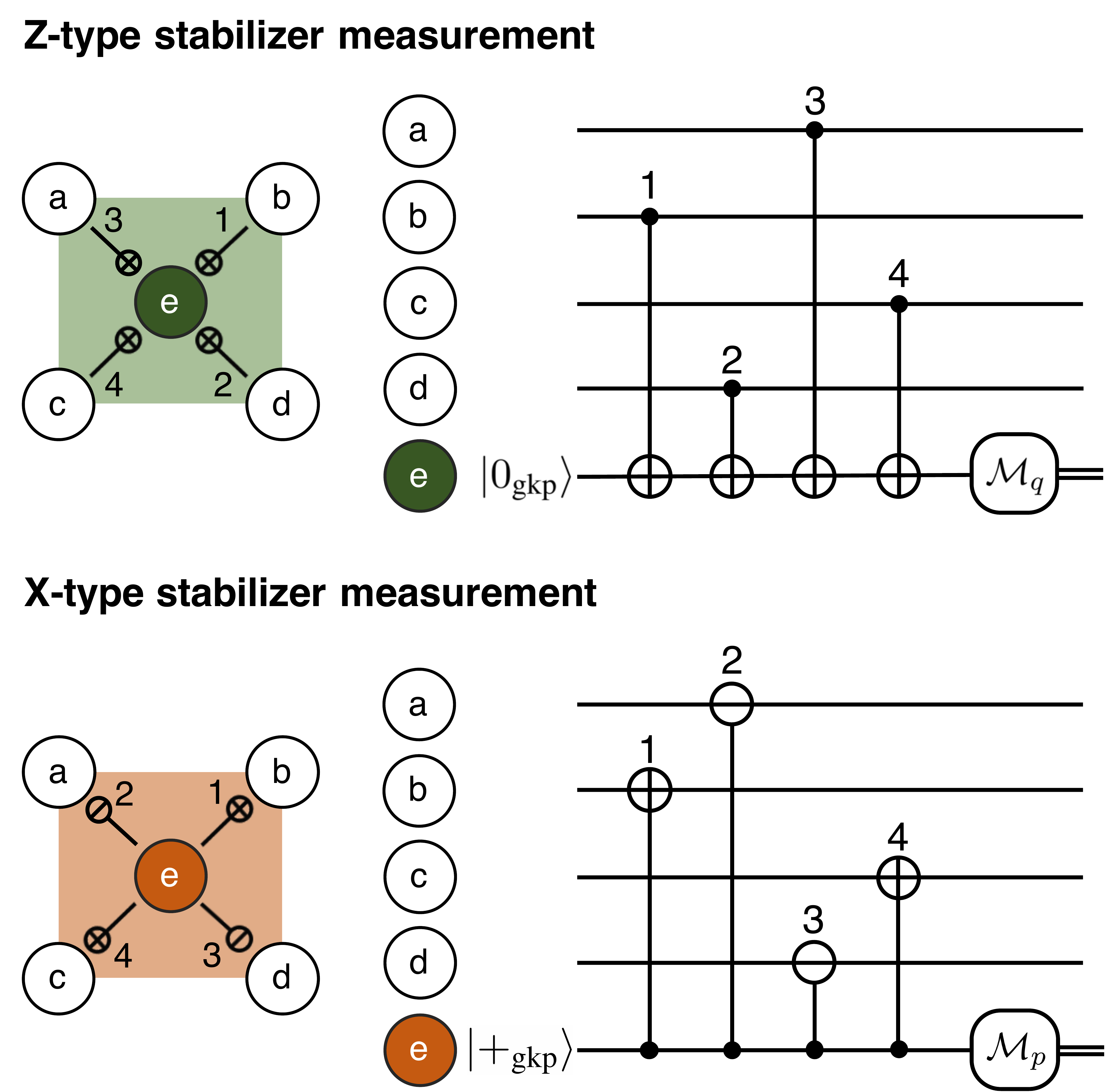}
\caption{[Fig.\ 4 in PRA \textbf{101}, 012316 (2020)] Circuits for surface code stabilizer measurements.   }
\label{fig:Surface code stabilizer measurement}
\end{figure}

We remark that in the usual case where the surface code is implemented with bare qubits (such as transmons \cite{Koch2007,Schreier2008}), it makes no difference to replace, for example, $\hat{S}_{X}^{[1]} = (\hat{X}^{(1)})^{\dagger}\hat{X}^{(2)}\hat{X}^{(4)}(\hat{X}^{(5)})^{\dagger}$ by $\hat{S}_{X}^{[1]} = \hat{X}^{(1)}\hat{X}^{(2)}\hat{X}^{(4)}\hat{X}^{(5)}$ since the Pauli operators are hermitian. Similarly, the action of $(\hat{X}_{\textrm{gkp}}^{(k)})^{\dagger}$ on the GKP qubit subspace is identical to that of $\hat{X}_{\textrm{gkp}}^{(k)}$ and therefore measuring $\hat{S}_{X}^{[1]} = (\hat{X}_{\textrm{gkp}}^{(1)})^{\dagger}\hat{X}_{\textrm{gkp}}^{(2)}\hat{X}_{\textrm{gkp}}^{(4)}(\hat{X}_{\textrm{gkp}}^{(5)})^{\dagger}$ is equivalent to measuring $\hat{S}_{X}^{[1]} = \hat{X}_{\textrm{gkp}}^{(1)}\hat{X}_{\textrm{gkp}}^{(2)}\hat{X}_{\textrm{gkp}}^{(4)}\hat{X}_{\textrm{gkp}}^{(5)}$ in the case of the surface-GKP code if the syndrome measurements are noiseless. 

It is important to note, however, that the actions of $(\hat{X}_{\textrm{gkp}}^{(k)})^{\dagger}$ and $\hat{X}_{\textrm{gkp}}^{(k)}$ are not the same outside of the GKP qubit subspace. Therefore, it does make a difference to choose $(\hat{X}_{\textrm{gkp}}^{(k)})^{\dagger}$ instead of $\hat{X}_{\textrm{gkp}}^{(k)}$ in the noisy measurement case, since shift errors propagate differently depending on the choice. For example, we illustrate in Fig.\ \ref{fig:Surface code error propagation} how the initial position shift error in the fourth $X$-type syndrome GKP qubit (X4 qubit) propagates to the second $Z$-type syndrome GKP qubit (Z2 qubit) through the fifth and the sixth data GKP qubits (D5 and D6 	qubits). Note that an initial random position shift in the X4 qubit (represented by the red lightning symbol) is propagated to the D6 qubit via the SUM gate $\textrm{SUM}_{X4\rightarrow D6}$ and then to the Z2 qubit via $\textrm{SUM}_{D6\rightarrow Z2}$. Additionally, it is also propagated to the D5 qubit via the inverse-SUM gate $\textrm{SUM}^{\dagger}_{X4\rightarrow D5}$ with its sign flipped and then the flipped shift is further propagated to the Z2 qubit via $\textrm{SUM}_{D5\rightarrow Z2}$. Thus, the propagated shift errors eventually cancel out each other at the Z2 qubit (visualized by the empty lightning symbol) due to the sign flip during the inverse-SUM gate.

\begin{figure}[t!]
\centering
\includegraphics[width=5.4in]{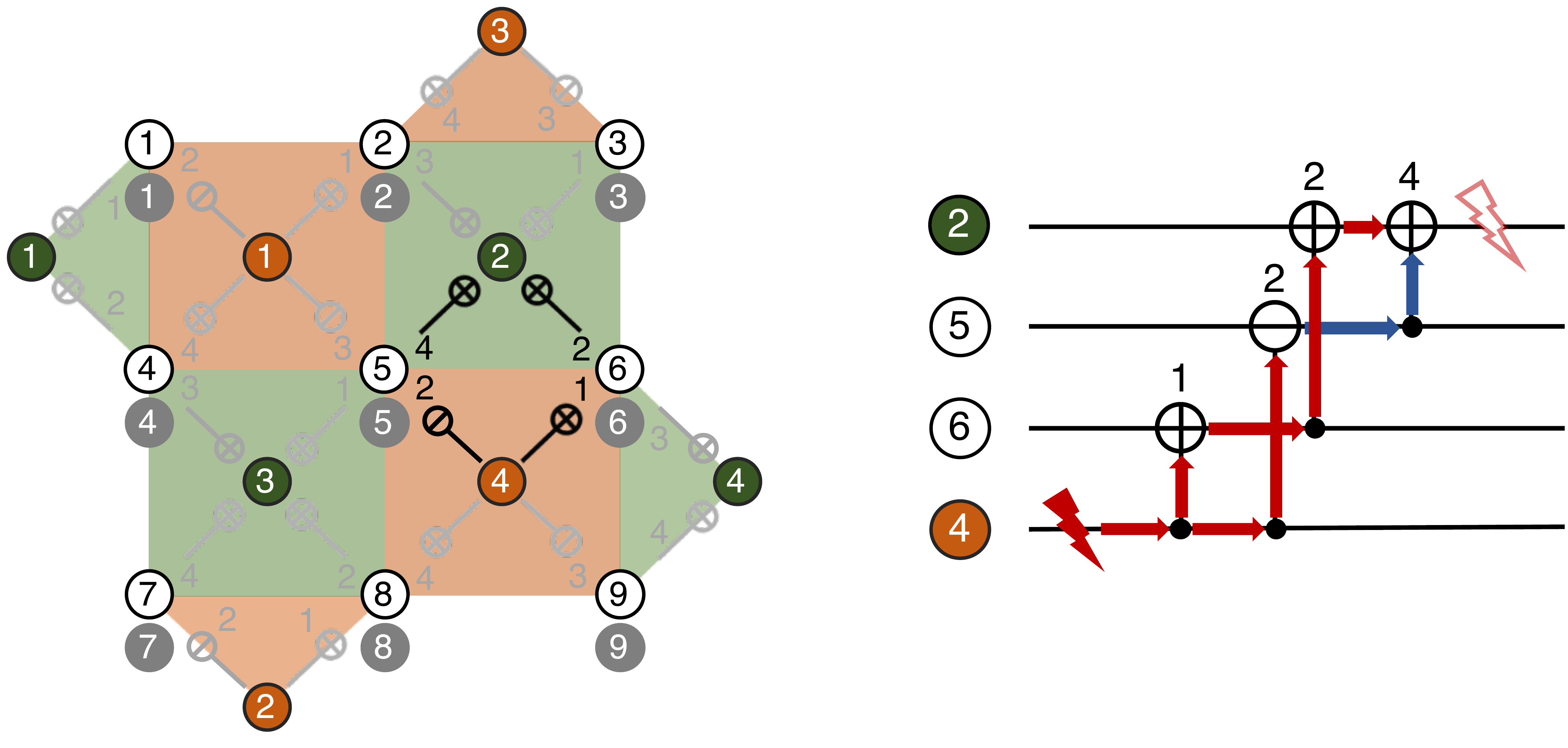}
\caption{[Fig.\ 5 in PRA \textbf{101}, 012316 (2020)] Noise propagation from the X4 qubit to the Z2 qubit during surface code stabilizer measurements. The red lightening symbol represents the initial position of a shift error on the qubit X4. During the propagation of the shift error to the qubit Z2, the sign of the shift error is flipped by the inverse-SUM gate $\textrm{SUM}^{\dagger}_{X4\rightarrow D5}$. This sign flip then results in cancellations of the propagated shift errors on the qubit Z2 (empty lightening symbol).       }
\label{fig:Surface code error propagation}
\end{figure}

Note that if the SUM gate $\textrm{SUM}_{X4\rightarrow D5}$ were used instead of the inverse-SUM gate $\textrm{SUM}^{\dagger}_{X4\rightarrow D5}$, the propagated shift errors would add together and therefore be amplified by a factor of $2$. In this regard, we emphasize that we have carefully chosen the specific pattern of the SUM and the inverse-SUM gates in Fig.\ \ref{fig:Surface-GKP codes} to avoid such noise amplifications.

\subsection{Circuit-level noise model}

Here, we discuss the noise model that we use to simulate the full error correction protocol with the surface-GKP code. To be more specific, the surface-GKP error correction protocol is implemented by repeatedly measuring the $\hat{S}_{q}$ and $\hat{S}_{p}$ GKP stabilizers for each data GKP qubit by using the circuits in Fig.\ \ref{fig:GKP qubit fundamentals}(b), and then measuring the surface code stabilizers shown in Figs.\ \ref{fig:Surface-GKP codes} and \ref{fig:Surface code stabilizer measurement}. Note that the required resources for these measurements are as follows:  
\begin{itemize}
\item Preparation of the GKP states $|0_{\textrm{gkp}}\rangle$ and $|+_{\textrm{gkp}}\rangle$. 
\item SUM and inverse-SUM gates.
\item Position and momentum homodyne measurements.
\item Displacement operations for error correction.
\end{itemize}
We assume that all these components can be noisy except for the displacement operations since in most experimental platforms, the errors associated with the displacement operations are negligible compared to the other errors. Moreover, note that displacement operations are only needed for error correction. Thus, they need not be implemented physically in practice since they can be kept track of by using a Pauli frame \cite{Knill2005, DiVincenzo2007,Terhal2015,Chamberland2018}. Below, we describe the noise model for each component in more detail.

\subsubsection{Finitely-squeezed GKP states}

Let us recall that realistic GKP states have a finite average excitation number, or a finite squeezing. As discussed in Section \ref{section:Translation-symmetric bosonic codes}, a finite-size GKP state can be modeled by applying a Gaussian envelope operator $\exp[-\Delta^{2}\hat{n}]$ to an ideal GKP state, i.e., $|\psi_{\textrm{gkp}}^{\Delta}\rangle\propto \exp[-\Delta^{2}\hat{n}]|\psi_{\textrm{gkp}}\rangle$. Expanding the envelope operator in terms of displacement operators \cite{Cahill1969}, we can write 
\begin{align}
|\psi_{\textrm{gkp}}^{\Delta}\rangle &\propto  \int  \frac{d^{2}\alpha}{\pi}  \mathrm{Tr}\big{[} \exp[-\Delta^{2}\hat{n}] \hat{D}^{\dagger}(\alpha)\big{]} \hat{D}(\alpha) |\psi_{\textrm{gkp}} \rangle 
\nonumber\\
&\propto \int d^{2}\alpha \exp\Big{[} -\frac{|\alpha|^{2}}{ 2\sigma_{\textrm{gkp}}^{2} }  \Big{]} \hat{D}(\alpha) |\psi_{\textrm{gkp}} \rangle , \label{eq:finite GKP expansion in displacements}
\end{align} 
where $\sigma_{\textrm{gkp}}^{2} = (1-e^{-\Delta^{2}}) / (1+e^{-\Delta^{2}}) \xrightarrow{\Delta\ll 1} \Delta^{2} /2$ (see Eq.\ \eqref{eq:approximate GKP laguerre proof}). That is, an approximate GKP state can be understood as the state that results from applying coherent superpositions of displacement operations with a Gaussian envelope to an ideal GKP state. More details about the approximate GKP codes can be found in \cite{Terhal2016,Shi2019,Pantaleoni2019,Matsuura2019}. 

To simplify our analysis of the surface-GKP code, we consider noisy GKP states corrupted by an incoherent mixture of displacement operations, instead of the coherent superposition as in Eq.\ \eqref{eq:finite GKP expansion in displacements}. That is, whenever a fresh GKP state $|0_{\textrm{gkp}}\rangle$ or $|+_{\textrm{gkp}}\rangle$ is supplied to the error correction chain, we assume that a noisy GKP state
\begin{align}
|0_{\textrm{gkp} } \rangle &\rightarrow \mathcal{N}_{B_{2}}[\sigma_{\textrm{gkp}}](|0_{\textrm{gkp} } \rangle\langle 0_{\textrm{gkp}} | ), \textrm{ or}
\nonumber\\
|+_{\textrm{gkp} } \rangle &\rightarrow \mathcal{N}_{B_{2}}[\sigma_{\textrm{gkp}}](|+_{\textrm{gkp} } \rangle\langle +_{\textrm{gkp}} | ) \label{eq:noisy GKP state incoherent displacement error}
\end{align}
is supplied. See Table \ref{table:random shift errors} for the definition of Gaussian random shift error $\mathcal{N}_{B_{2}}[\sigma]$. Note that $\mathcal{N}_{B_{2}}[\sigma]$ models an incoherent mixture of random displacement errors. We remark that the noisy GKP states corrupted by an incoherent displacement error (as in Eq.\ \eqref{eq:noisy GKP state incoherent displacement error}) are noisier than the noisy GKP states corrupted by a coherent displacement error (as in Eq.\ \eqref{eq:finite GKP expansion in displacements}), because the former can be obtained from the latter by applying a technique similar to Pauli twirling \cite{Emerson2007} (see Appendix A in Ref.\ \cite{Noh2020} for more details). In this sense, by adopting the incoherent noise model, we make a conservative assumption about the GKP noise while simplifying the analysis. 

We define the squeezing $s_{\textrm{gkp}}$ of a noisy GKP state $\mathcal{N}_{B_{2}}[\sigma_{\textrm{gkp}}](|\psi_{\textrm{gkp} } \rangle\langle \psi_{\textrm{gkp}} | )$ as $s_{\textrm{gkp}}\equiv -10\log_{10}(2\sigma_{\textrm{gkp}}^{2})$ (aligning our notation with those in Refs.\ \cite{Menicucci2014,Fukui2018a,Fukui2019}), where the unit of $s_{\textrm{gkp}}$ is in dB. We also assume that idling modes are undergoing independent Gaussian random displacement errors $\mathcal{N}_{B_{2}}[\sigma_{p}]$ with variance $\sigma_{p}^{2} = \kappa \Delta t_{p}$ during the GKP state preparation, where $\kappa$ is the photon loss and heating rate (see below) and $\Delta t_{p}$ is the time needed to prepare the GKP states. 

\subsubsection{Noisy SUM and inverse-SUM gates}

We assume that photon loss errors occur continuously during the execution of the SUM or the inverse-SUM gates. To be more specific, we assume that SUM gates are implemented by letting the system evolve under the Hamiltonian $\hat{H} = g\hat{q}_{1}\hat{p}_{2}$ for $\Delta t = 1/g$ (the first mode is the control mode and the second mode is the target mode), during which independent photon loss errors occur continuously in both the control and the target mode. That is, we replace the unitary SUM gate $\textrm{SUM}_{1\rightarrow 2} = \exp[-i\hat{q}_{1}\hat{p}_{2}]$ (or the inverse-SUM gate $\textrm{SUM}_{1\rightarrow 2}^{\dagger} = \exp[i\hat{q}_{1}\hat{p}_{2}]$) by a completely positive and trace-preserving (CPTP) map \cite{Choi1975} $\exp[\mathcal{L}_{+}\Delta t]$ (or $\exp[\mathcal{L}_{-}\Delta t]$) with $\Delta t = 1/g$, where $g$ is the coupling strength and the Lindbladian generator $\mathcal{L}_{\pm}$ is given by 
\begin{align}
\mathcal{L}_{\pm }(\hat{\rho}) = \mp ig[ \hat{q}_{1}\hat{p}_{2}, \hat{\rho}  ] + \kappa \big{(} \mathcal{D}[\hat{a}_{1}] + \mathcal{D}[\hat{a}_{2}] \big{)}\hat{\rho}. \label{eq:noisy SUM or inverse-SUM gates Lindbladian}
\end{align}
Here, $\mathcal{D}[\hat{A}](\hat{\rho}) \equiv \hat{A}\hat{\rho}\hat{A}^{\dagger} - \frac{1}{2}\lbrace \hat{A}^{\dagger}\hat{A},\hat{\rho} \rbrace$, and $\kappa$ is the photon loss rate. 

In a similar spirit as above, we make a more conservative assumption about the gate error to make the analysis more tractable. That is, we make the noisy gate $\exp[\mathcal{L}_{\pm}\Delta t]$ noisier by adding heating errors $\kappa( \mathcal{D}[\hat{a}_{1}^{\dagger}] + \mathcal{D}[\hat{a}_{2}^{\dagger}] )$ to the Lindbladian $\mathcal{L}_{\pm}$, i.e., 
\begin{align}
\mathcal{L}'_{\pm} \equiv  \mathcal{L}_{\pm} + \kappa\big{(} \mathcal{D}[\hat{a}_{1}^{\dagger}] + \mathcal{D}[\hat{a}_{2}^{\dagger}] \big{)} , 
\end{align}
where the heating rate $\kappa$ is the same as the photon loss rate. This is to convert the loss errors into random displacement errors (see Refs.\ \cite{Albert2018,Noh2019}). Indeed, the noisy SUM or the inverse-SUM gate $\exp[\mathcal{L}'_{\pm}\Delta t]$ is equivalent to the ideal SUM or the inverse-SUM gate followed by a correlated Gaussian random displacement error $\hat{q}_{k}\rightarrow \hat{q}_{k} + \xi_{q}^{(k)}$ and $\hat{p}_{k}\rightarrow \hat{p}_{k} + \xi_{p}^{(k)}$ for $k\in \lbrace 1,2 \rbrace$, where the additive shift errors are drawn from bivariate Gaussian distributions $(\xi_{q}^{(1)},\xi_{q}^{(2)}) \sim \mathcal{N}(0,\boldsymbol{N}_{q}^{\pm})$ and $(\xi_{p}^{(1)},\xi_{p}^{(2)}) \sim \mathcal{N}(0,\boldsymbol{N}_{p}^{\pm})$ with the noise covariance matrices 
\begin{align}
\boldsymbol{N}_{q}^{\pm} = \sigma_{c}^{2}\begin{bmatrix}
1 & \pm 1/2\\
\pm 1/2 & 4/3
\end{bmatrix},  \,\,  \boldsymbol{N}_{p}^{\pm} = \sigma_{c}^{2}\begin{bmatrix}
4/3 & \mp 1/2\\
\mp 1/2 & 1
\end{bmatrix}. \label{eq:noise covariance matrix SUM or inverse-SUM}
\end{align}
(See Appendix A in Ref.\ \cite{Noh2020} for more details.) Here, the variance $\sigma_{c}^{2}$ is given by $\sigma_{c}^{2} = \kappa \Delta t = \kappa /g$. The noise covariance matrices $\boldsymbol{N}_{q}^{+}$ and $\boldsymbol{N}_{p}^{+}$ are used for the SUM gate and $\boldsymbol{N}_{q}^{-}$ and $\boldsymbol{N}_{p}^{-}$ are used for the inverse-SUM gate. If there are idling modes during the application of the SUM or the inverse-SUM gates on some other pairs of modes, we assume that the idling modes undergo independent Gaussian random displacement errors $\mathcal{N}_{B_{2}}[\sigma_{c}]$ of the same variance $\sigma_{c}^{2} = \kappa \Delta t = \kappa/g$, because they should wait for the same amount of time until the gates are completed.

\subsubsection{Noisy homodyne measurements} 

Lastly, we model errors in position and momentum homodyne measurements by adding independent Gaussian random displacement errors $\mathcal{N}_{B_{2}}[\sigma_{m}]$ of the variance $\sigma_{m}^{2} = \kappa \Delta t_{m}$ before the ideal homodyne measurements. Here, $\Delta t_{m}$ is the time needed to implement the homodyne measurements. Also, during the homodyne measurements, we assume that idling modes are undergoing independent Gaussian random displacement errors of the same variance $\sigma_{m}^{2} = \kappa \Delta t_{m}$. 

\subsection{Main results}

Let us now rigorously analyze the performance of the surface-GKP code by simulating the full error correction protocol assuming the noise model described so far. We focus on the case $\sigma_{p}=\sigma_{c}=\sigma_{m}\equiv \sigma$ where all circuit elements are comparably noisy. However, we assume that the noise afflicting GKP states $\sigma_{\textrm{gkp}}$ is independent of the circuit noise. Since we have two independent noise parameters $\sigma_{\textrm{gkp}}$ and $\sigma$, the fault-tolerance thresholds would form a curve instead of a single number. Therefore, instead of exhaustively investigating the entire parameter space, we consider the following three representative scenarios:  
\begin{enumerate}[label={Case \Roman*},wide =\parindent]
\item \!\!: $\sigma_{\textrm{gkp}} \neq 0$ and $\sigma=0$ \label{case:1}
\item \!\!: $\sigma_{\textrm{gkp}} = 0$ and $\sigma\neq 0$ \label{case:2}
\item \!\!: $\sigma_{\textrm{gkp}} = \sigma \neq  0$ \label{case:3}
\end{enumerate}  
Then, we find the threshold values for $\sigma_{\textrm{gkp}}$ (\ref{case:1}), $\sigma$ (\ref{case:2}), and $\sigma_{\textrm{gkp}}=\sigma$ (\ref{case:3}), under which fault-tolerant quantum error correction is possible with the surface-GKP code. Specifically, we take the distance $d$ surface-GKP code and repeat the (noisy) stabilizer measurements $d$ times. Then, we construct 3D space-time graphs based on the stabilizer measurement outcomes and apply a minimum-weight perfect matching decoding algorithm \cite{Edmonds1965,Edmonds1965b} to perform error correction. Specifically, we use a simple method to compute the renormalized edge weights of the 3D matching graphs, based on the information obtained during GKP-stabilizer measurements. Such graphs are then used to perform MWPM.  A detailed description of our method is given in Subsection \ref{subsection:Simulation details}. Below, we report the logical $X$ error rates, which are the same as the logical $Z$ error rates. Logical $Y$ error rates are not shown since they are much smaller than the logical $X$ and $Z$ error rates.

\begin{figure}[t!]
\centering
\includegraphics[width=4.8in]{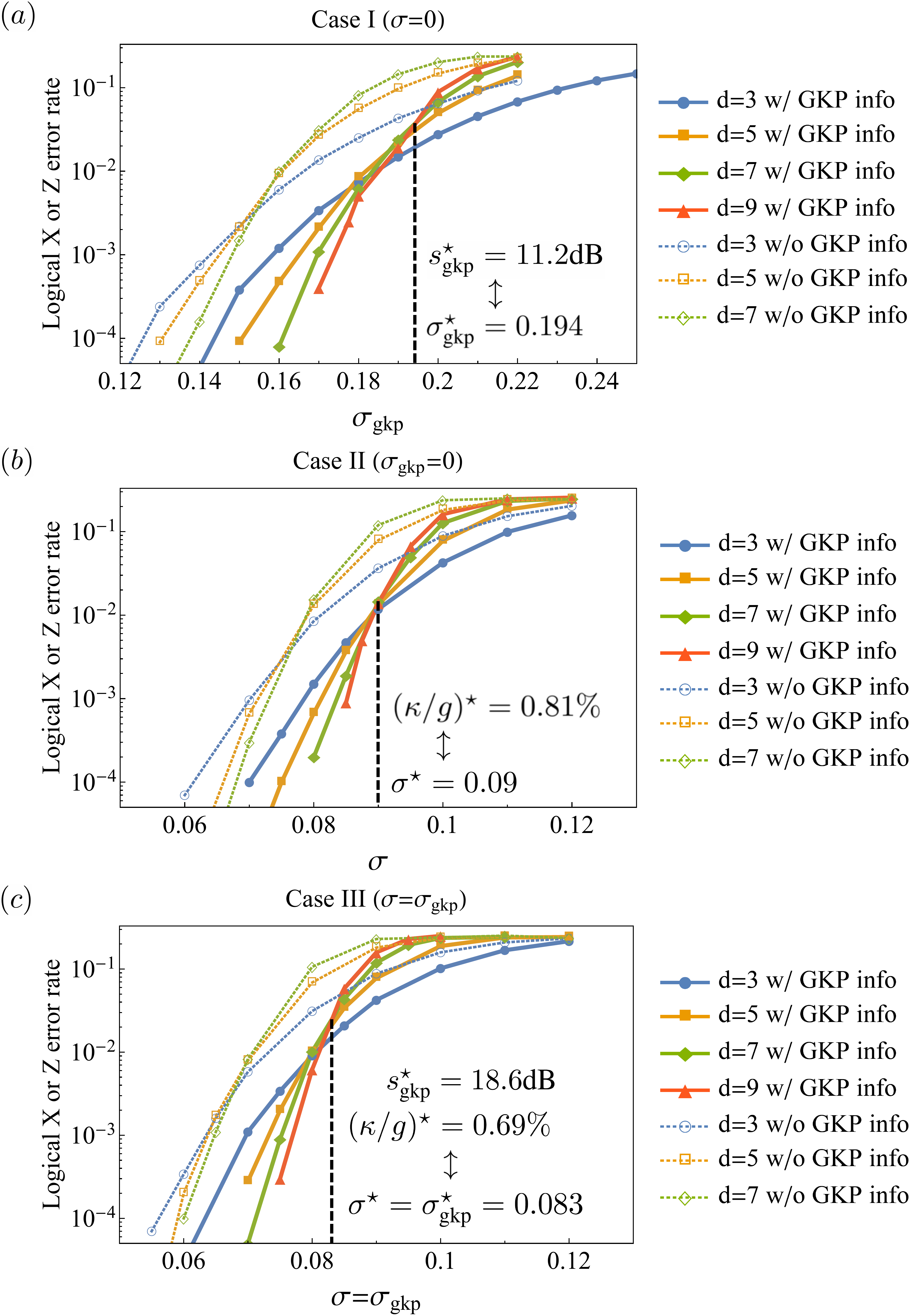}
\caption{[Fig.\ 6 in PRA \textbf{101}, 012316 (2020)] The logical $X$ error rate of the surface-GKP code for various $d$ when (a) $\sigma=0$ (\ref{case:1}), (b) $\sigma_{\textrm{gkp}}=0$ (\ref{case:2}), and (c) $\sigma = \sigma_{\textrm{gkp}}$ (\ref{case:3}), which is the same as the logical $Z$ error rate. The solid lines represent logical error rates when information from the GKP-stabilizer measurements is used to renormalize edge weights in the matching graphs. The dotted lines correspond to the case when information from GKP-stabilizer measurements is ignored. In all cases, given that $\sigma_{\textrm{gkp}}$ and $\sigma$ are below certain fault-tolerance thresholds, the logical $X$ or $Z$ error rates are suppressed to an arbitrarily small value as we increase the code distance $d$.     }
\label{fig:main results}
\end{figure}

In Fig.\ \ref{fig:main results}(a), we consider the case where GKP states are the only noisy components in the scheme, i.e., $\sigma=0$ (\ref{case:1}). We show the performance of the surface-GKP code when both the additional information from GKP-stabilizer measurements is incorporated and when it is ignored. When the additional information is incorporated, the logical $X$ error rate (same as the logical $Z$ error rate) decreases as we increase the code distance $d$ if $\sigma_{\textrm{gkp}}$ is smaller than the threshold value $\sigma_{\textrm{gkp}}^{\star} = 0.194$ (or if the squeezing of the noisy GKP state $s_{\textrm{gkp}}$ is higher than the threshold value $s_{\textrm{gkp}}^{\star} = 11.2$dB). That is, in this case, fault-tolerant error correction is possible with the surface-GKP code if the squeezing of the GKP states is above $11.2$dB. Note that if the additional information from GKP-stabilizer measurements is ignored, the threshold squeezing value decreases and logical error rates can range from one to several orders of magnitude larger for a given $\sigma_{\textrm{gkp}}$.      

In Fig.\ \ref{fig:main results}(b), we consider the case where GKP states are noiseless but the other circuit elements are noisy, i.e., $\sigma_{\textrm{gkp}}=0$ (\ref{case:2}). In this case, if the additional information from the GKP error correction protocol is incorporated, we can suppress the logical $X$ error rate (same as the logical $Z$ error rate) to any desired small value by choosing a sufficiently large code distance $d$ as long as $\sigma$ is smaller than the threshold value $\sigma^{\star} = 0.09$. Note that since $\sigma^{2} = \kappa / g$ the threshold value $\sigma^{\star} = 0.09$ corresponds to $(\kappa / g)^{\star} = 8.1\times 10^{-3} = 0.81\%$, where $\kappa$ is the photon loss rate and $g$ is the coupling strength of the SUM or the inverse-SUM gates. That is, fault-tolerant error correction with the surface-GKP code is possible if the SUM or the inverse-SUM gates can be implemented roughly $120$ times faster than the photon loss processes. Note that if the additional information from GKP-stabilizer measurements is ignored, the threshold value becomes smaller and logical error rates can range from one to several orders of magnitude larger for a given $\sigma$.  

Finally in Fig.\ \ref{fig:main results}(c), we consider the case where the GKP states and the other circuit elements are comparably noisy, i.e., $\sigma = \sigma_{\textrm{gkp}}$ (\ref{case:3}). In this case, fault-tolerant error correction is possible if $\sigma = \sigma_{\textrm{gkp}}$ is smaller than the threshold value $\sigma^{\star}=\sigma^{\star}_{\textrm{gkp}} = 0.083$. This threshold value corresponds to the GKP squeezing $s_{\textrm{gkp}}^{\star}=18.6$dB and $\kappa / g = 6.9\times 10^{-3} = 0.69\%$. Similarly, as in the previous cases, if the additional information from GKP-stabilizer measurements is ignored, the threshold value becomes smaller and logical error rates can range from one to several orders of magnitude larger for a given noise parameter $\sigma = \sigma_{\textrm{gkp}}$. 

For all three cases, we clearly observe that fault-tolerant quantum error correction with the surface-GKP code is possible despite noisy GKP states and noisy circuit elements, given that the noise parameters are below certain fault-tolerance thresholds. Recent state-of-the-art experiments have demonstrated the capability to prepare GKP states of squeezing between $5.5$dB and $9.5$dB \cite{Campagne2019,Fluhmann2018,Fluhmann2019,Fluhmann2019b}, approaching the established squeezing threshold values $s_{\textrm{gkp}}^{\star} \ge 11.2$dB.

In circuit QED systems, beam-splitter interactions between two high-Q cavity modes have been implemented experimentally with $\kappa/ g \sim 10^{-2}$, where $g$ is the relevant coupling strength and $\kappa$ is the photon loss rate \cite{Gao2018}. While the same scheme (based on four-wave mixing processes) may be adapted to realize the SUM or the inverse-SUM gates between two high-Q cavity modes \cite{Zhang2019}, this scheme will induce non-negligible Kerr nonlinearities and thus may not be compatible with the GKP qubits which should be operated in the regime where Kerr nonlinearities are negligible \cite{Campagne2019}. On the other hand, by using three-wave mixing elements \cite{Frattini2017}, it would be possible to implement the SUM or the inverse-SUM gates between two high-Q cavity modes in a way that is not significantly limited by Kerr nonlinearities.

Let us now compare the performance of the surface-GKP code with the usual rotated surface code implemented by bare qubits such as transmon qubits. Assuming a full circuit-level depolarizing noise (both for single- and two-qubit gates), it was numerically demonstrated that fault-tolerant quantum error correction is possible with the rotated surface code if the physical error rate is below the threshold $p^{\star} = 1.2\%$ \cite{Wang2011}. Note that such a high threshold value was obtained by introducing 3D space-time correlated edges (see Figs.\ 3 and 4 in Ref.\ \cite{Wang2011}) and fully optimizing the renormalized edge weights based on the noise parameters.  

Our circuit-level noise model (in terms of shift errors) is quite different from the depolarizing noise model considered in typical qubit-based fault-tolerant error correction schemes. Moreover, we also introduce non-Gaussian resources, i.e., GKP states in our scheme. Therefore, our results cannot be directly compared with the results in Ref.\ \cite{Wang2011}. We nevertheless point out that we obtain comparable threshold values $(\kappa/ g)^{\star} = 0.81\%$ (\ref{case:2}) and $(\kappa/ g)^{\star} = 0.69\%$ (\ref{case:3}) where $\kappa$ is the photon loss rate and $g$ is the coupling strength of the two-mode gates. We stress that we do not introduce 3D space-time correlated edges and provide a simple method for computing the renormalized edge weights. In particular, 3D space-time correlated edges are not necessary in our case with the surface-GKP code. This is because any shift errors that are correlated due to two-mode gates will not cause any Pauli errors to GKP qubits nor trigger syndrome GKP qubits incorrectly, as long as the size of the correlated shifts is smaller than $\sqrt{\pi}/2$, which is the case below the fault-tolerance thresholds computed above.

We also point out that in general, topological codes without leakage reduction units \cite{Aliferis2007} are not robust against leakage errors that occur when a bare qubit state is excited and falls out of its desired two-level subspace \cite{Aliferis2007,Fowler2013,Suchara2015,Brown2019}. In the case of the surface-GKP code, leakage errors do occur as well because each bosonic mode may not be in the desired two-level GKP code subspace. However, the surface-GKP code is inherently resilient to such leakage errors (and thus does not require leakage reduction units) since GKP-stabilizer measurements will detect and correct such events. Indeed, in our simulation of the surface-GKP code, leakage errors continuously occur due to shift errors, but the established fault-tolerance thresholds are nevertheless still favorable since GKP-stabilizer measurements prevent the leakage errors from propagating further. 

We lastly remark that the logical $X$ or $Z$ error rates in Fig.\ \ref{fig:main results} decrease very rapidly as $\sigma_{\textrm{gkp}}$ and $\sigma$ approach zero in the case of the surface-GKP code. This is again because the GKP code can correct any shift errors of size less than $\sqrt{\pi}/2$ and therefore the probability that a Pauli error occurs in a GKP qubit (at the end of GKP-stabilizer measurements) becomes exponentially small as $\sigma_{\textrm{gkp}}$ and $\sigma$ approach zero. More precisely, at the end of each GKP-stabilizer measurement, a bulk data GKP qubit undergoes a Pauli $X$ or $Z$ error with probability 
\begin{align}
p_{\textrm{err}}\Big{(} \sqrt{5\sigma_{\textrm{gkp}}^{2} + \frac{59}{3}\sigma^{2} } \Big{)}, 
\end{align}
where $p_{\textrm{err}}(\sigma)$ is defined in Eq.\ \eqref{eq:definition of perr}. Here, the variance $5\sigma_{\textrm{gkp}}^{2} +  (59/3)\sigma^{2}$ was carefully determined by thoroughly keeping track of how circuit-level noise propagates during stabilizer measurements (see also Subsection \ref{subsection:Simulation details}). As can be seen from Fig.\ \ref{fig:perr}, $p_{\textrm{err}}(\sigma)$ agrees well with the asymptotic expression $p_{\textrm{asy}}(\sigma) = (\sqrt{8}\sigma^{2} / \pi ) \exp[ -\pi/ (8\sigma^{2} )  ]$ in the $\sigma\ll 1$ limit. Thus, $p_{\textrm{err}}(\sigma)$ decreases exponentially as $\sigma$ goes to zero. 

\begin{figure}[t!]
\centering
\includegraphics[width=4.0in]{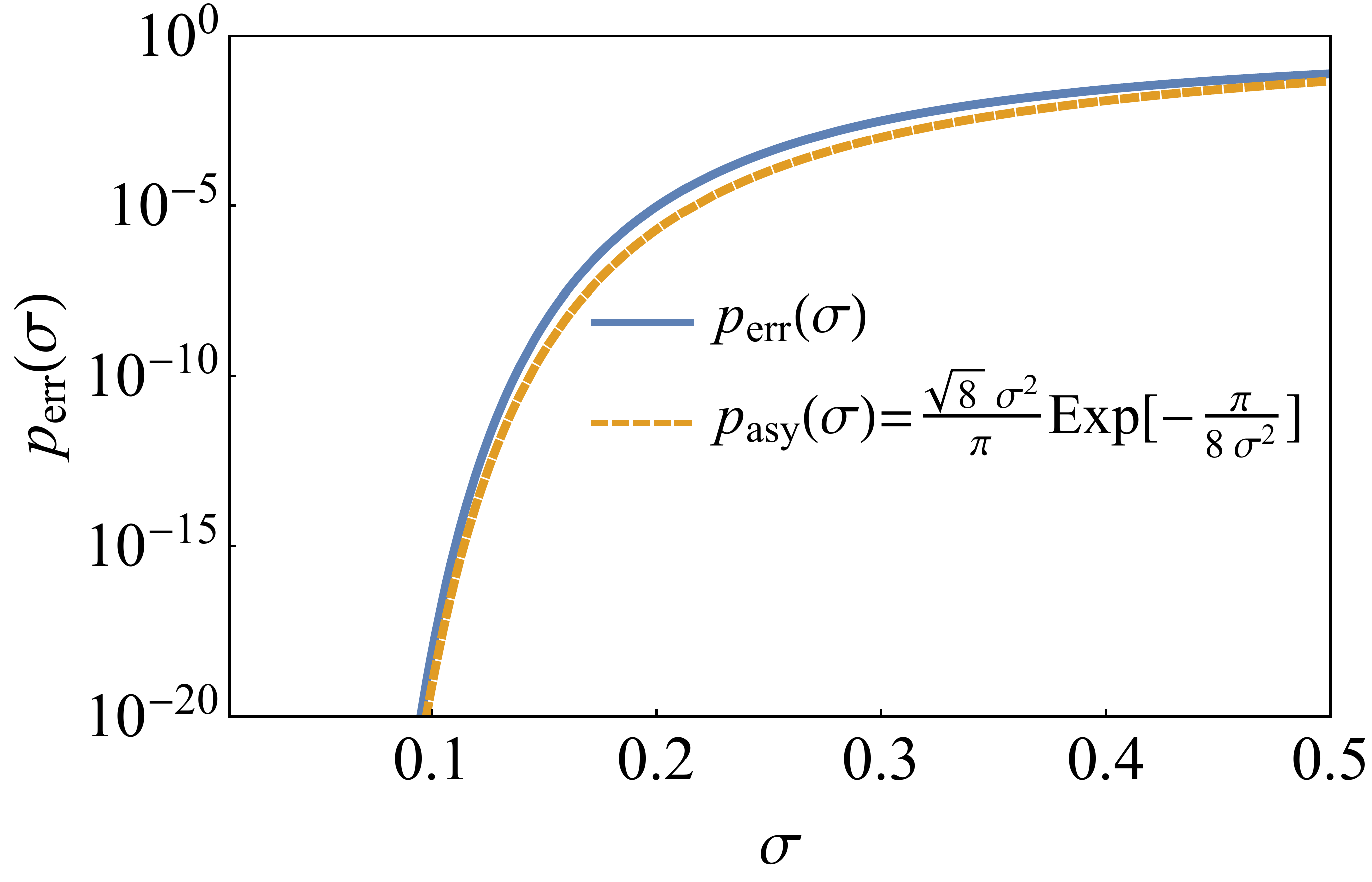}
\caption{[Fig.\ 7 in PRA \textbf{101}, 012316 (2020)] Visualization of the function $p_{\textrm{err}}(\sigma)$ (blue). The asymptotic expression $p_{\textrm{asy}}(\sigma) = (\sqrt{8}\sigma^{2} / \pi ) \exp[ -\pi/ (8\sigma^{2} )  ]$ is represented by the yellow dashed line. $p_{\textrm{err}}(\sigma)$ and $p_{\textrm{asy}}(\sigma)$ agree well with each other in the $\sigma\ll 1$ limit.   }
\label{fig:perr}
\end{figure}

Similarly, the probability that a bulk surface code stabilizer measurement yields an incorrect measurement outcome is given by
\begin{align}
p_{\textrm{err}}\Big{(} \sqrt{7\sigma_{\textrm{gkp}}^{2} + \frac{116}{3}\sigma^{2} } \Big{)}, 
\end{align}
and decays exponentially as $\sigma_{\textrm{gkp}}$ and $\sigma$ approach zero. Therefore, if the circuit-level noise of the physical bosonic modes is very small to begin with, GKP codes will locally provide a significant noise reduction. In this case, the overall resource overhead associated with the next level of global encoding will be modest since a small-distance surface code would suffice. Therefore in this regime, the surface-GKP code may be able to achieve the same target logical error rate in a more hardware-efficient way than the usual surface code. However, since this regime requires high quality GKP states, the additional resource overhead associated with the preparation of such high quality GKP states should also be taken into account for a comprehensive resource estimate. We leave such an analysis to future work.      

\subsection{Comparison with previous works}
\label{subsection:Comparison with previous works}

Here, we compare the results obtained in our work with previous works in Refs.\ \cite{Menicucci2014,Wang2017,Fukui2018a,Fukui2019,Vuillot2019}. Firstly, Refs.\ \cite{Wang2017,Vuillot2019} considered the toric-GKP code and computed fault-tolerance thresholds for both code capacity and phenomenological noise models. In particular, the phenomenological noise models used in these works describe faulty syndrome extraction procedures (due to finitely-squeezed ancilla GKP states) in a way that does not take into account the propagation of the relevant shift errors. More specifically, in Figs.\ 1, 2, and 7 in Ref.\ \cite{Vuillot2019}, shift errors are manually added in the beginning of each stabilizer measurement measurement and right before each homodyne measurement. Therefore, this phenomenological noise model can be understood as a model for homodyne detection inefficiencies while assuming ideal ancilla GKP states. In other words, the fault-tolerance threshold values established in Ref.\ \cite{Vuillot2019} (i.e., $\sigma_{0}^{\star} = 0.235$ and $\sigma_{0}^{\star} = 0.243$; see Fig.\ 12 therein) do not accurately represent the tolerable noise in the ancilla GKP states since the noise propagation was not thoroughly taken into account. Thus, these threshold values can only be taken as a rough upper bound on $\sigma_{\textrm{gkp}}^{\star}$ and cannot be directly compared with the threshold values obtained in our work. Note also that the threshold values in Ref.\ \cite{Vuillot2019} were computed for the toric code which has a different threshold compared to the rotated surface code \cite{Fowler2012}.

On the other hand, in our work we assume that every GKP state supplied to the error correction chain has a finite squeezing and we comprehensively take into account the propagation of such shift errors through the entire error correction circuit. By doing so, we accurately estimate the tolerable noise in the finitely-squeezed ancilla GKP states by computing $\sigma_{\textrm{gkp}}^{\star}$. Related, we stress that when the noise propagation is taken into account, detailed scheduling and design of the syndrome extraction circuits become very crucial and we carefully designed the circuits in a way that mitigates the adverse effects of the noise propagation (see Fig.\ \ref{fig:Surface code error propagation}).

Moreover, we also consider photon loss and heating errors occurring continuously during the implementation of the SUM and inverse-SUM gates.
Thus, we establish fault-tolerance thresholds for the strength of the two-mode coupling relative to the photon loss rate and demonstrate that fault-tolerant quantum error correction with the surface-GKP code is possible in more general scenarios. We also remark that Ref.\ \cite{Vuillot2019} used a minimum-energy decoder based on statistical-mechanical methods in the noisy regime whereas we provide a simple method for computing renormalized edge weights to be used in a MWPM decoder. 

\begin{table}[t!]
  \centering
  \def\arraystretch{2}
  \begin{tabular}{ V{3} c V{1.5} c V{1.5}  c V{1.5} c V{1.5} c V{3} }
   \hlineB{3}  
    \ref{case:1} ($\sigma=0$) & Method & \,\,\, $\sigma_{\textrm{gkp}}^{\star}$ \,\,\, & \,\,\, $s_{\textrm{gkp}}^{\star}$ \,\,\, &  Post-selection? \\  \hlineB{3} 
    Ref.\ \cite{Menicucci2014} & Concatenated codes (MB) & $0.067$ & $20.5$dB &   NO  \\ \hlineB{1.5}  
    Ref.\ \cite{Fukui2018a} & 3D cluster state (MB) & $0.228$ & $9.8$dB &  YES \\ \hlineB{1.5}  
    Ref.\ \cite{Fukui2019} & 3D cluster state (MB) & $0.273$ & $8.3$dB &  YES  \\ \hlineB{1.5}  
    Refs.\ \cite{Wang2017,Vuillot2019} & Toric-GKP code (GB) & N/A  & N/A &  NO   \\ \hlineB{1.5}  
    Our work & Surface-GKP code (GB) & $0.194$ & $11.2$dB &  NO   \\ \hlineB{3}  
  \end{tabular}
   \caption{[Table 1 in PRA \textbf{101}, 012316 (2020)] Threshold values for the squeezing of GKP states for fault-tolerant quantum error correction. Here, we compare the established threshold values obtained by assuming that 
   GKP states are the only noisy components in the error correction circuit (i.e., \ref{case:1}). MB stands for measurement-based and GB stands for gate-based. $d$ is the distance of the code. For the results on the toric-GKP code, $\sigma_{\textrm{gkp}}^{\star}$ and $s_{\textrm{gkp}}^{\star}$ are not available because in the results were obtained by assuming a phenomenological noise model that does not take into account the propagation of shift errors through the entire error correction circuit. That is, the threshold values established in [Phys. Rev. A \textbf{99}, 032344 (2019)] using the toric-GKP code (i.e., $\sigma_{0}^{\star} = 0.235$ and $\sigma_{0}^{\star} = 0.243$ see Fig.\ 12 therein) do not accurately quantify the tolerable noise in the ancilla GKP states. Instead, $\sigma_{0}^{\star}$ can only be taken as a rough upper bound on $\sigma_{\textrm{gkp}}^{\star}$ (see the main text for more details).  }
   \label{table:comparison with other works}
\end{table}

Secondly, Refs.\ \cite{Menicucci2014,Fukui2018a,Fukui2019} considered measurement-based quantum computing with GKP qubits and did establish fault-tolerance thresholds for the squeezing of the GKP states. Assuming that GKP states are the only noisy components (i.e., \ref{case:1}), Ref.\ \cite{Menicucci2014} found the squeezing threshold value $s_{\textrm{gkp}}^{\star} = 20.5$dB, and Refs.\ \cite{Fukui2018a} and \cite{Fukui2019} later brought the value down to $s_{\textrm{gkp}}^{\star} = 9.8$dB and $s_{\textrm{gkp}}^{\star} = 8.3$dB, respectively. Notably, the squeezing thresholds found in Refs.\ \cite{Fukui2018a,Fukui2019} are more favorable than the squeezing threshold found in our work, i.e., $s_{\textrm{gkp}}^{\star} = 11.2$dB (see Fig.\ \ref{fig:main results}(a)). In this regard, we remark that the favorable threshold values obtained in Refs.\ \cite{Fukui2018a,Fukui2019} rely on the use of post-selection. That is, each GKP measurement succeeds with probability strictly less than unity and thus the overall success probability would decrease exponentially as the system size $d$ increases. On the other hand, we do not discard any measurement outcomes and thus our scheme succeeds with unit probability for any distance $d$. Therefore, our scheme with the surface-GKP code deterministically suppresses errors exponentially with the code distance as long as $\sigma_{\textrm{gkp}}$ and $\sigma$ are below the threshold values. The differences between our work and the previous works are summarized in Table \ref{table:comparison with other works}.

\subsection{Simulation details}
\label{subsection:Simulation details}

Here, we provide detailed step-by-step descriptions of the simulation of the surface-GKP code that we used to obtain the main results in Fig.\ \ref{fig:main results}.  

\subsubsection{GKP stabilizer measurements}

Consider the distance-$d$ surface-GKP code consisting of $d^{2}$ data GKP qubits. Each data GKP qubit is stabilized by the two GKP stabilizers $\hat{S}_{q}^{(k)} = \exp[i2\sqrt{\pi} \hat{q}_{k} ]$ and $\hat{S}_{p}^{(k)} =  \exp[- i2\sqrt{\pi} \hat{p}_{k} ]$ where $k\in \lbrace 1, \cdots, d^{2} \rbrace$. In the first step of GKP-stabilizer measurements (left in Fig.\ \ref{fig:GKP stabilizer measurement}), $\hat{S}_{q}^{(k)}$ ($\hat{S}_{p}^{(k)}$) stabilizers are measured for odd (even) $k$. In the second step (right in Fig.\ \ref{fig:GKP stabilizer measurement}), on the other hand, $\hat{S}_{p}^{(k)}$ ($\hat{S}_{q}^{(k)}$) stabilizers are measured for odd (even) $k$. Note that we alternate between $\hat{S}_{q}$ and $\hat{S}_{p}$ measurements in a checkerboard pattern in order to balance the position and momentum quadrature noise.

\begin{figure}[t!]
\centering
\includegraphics[width=5.8in]{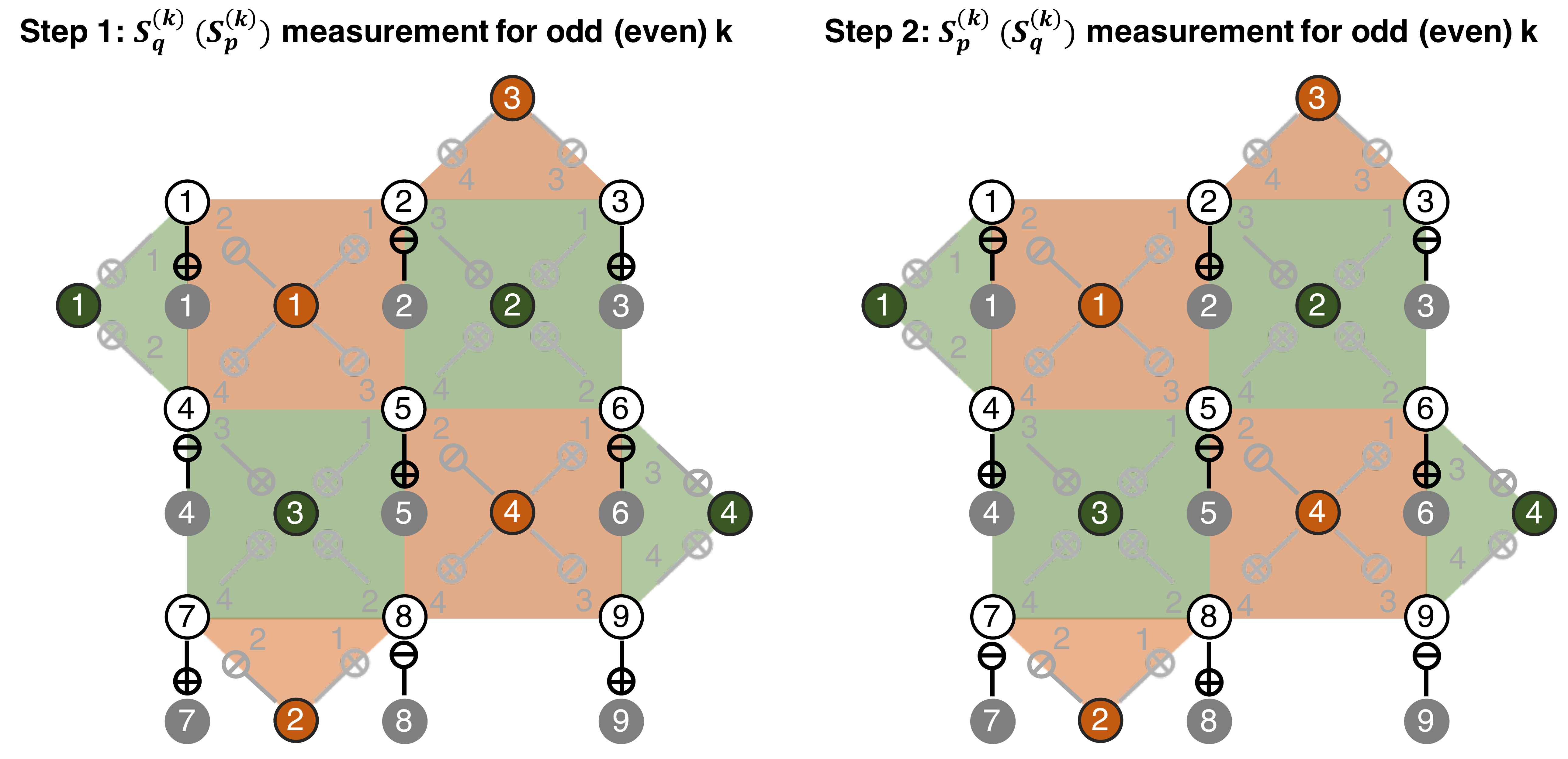}
\caption{[Fig.\ 8 in PRA \textbf{101}, 012316 (2020)] Measurement of the GKP stabilizers for $d=3$. See also Fig.\ \ref{fig:GKP qubit fundamentals}(b) and the caption for the definition of each graphical symbol.   }
\label{fig:GKP stabilizer measurement}
\end{figure}

Let $\xi_{q}^{D}$ and $\xi_{p}^{D}$ ($\xi_{q}^{A}$ and $\xi_{p}^{A}$) be the data (ancilla) position and momentum quadrature noise, where 
\begin{align}
\xi_{q}^{D} &= (\xi_{q}^{(D1)},\cdots, \xi_{q}^{(Dd^{2})}),
\nonumber\\
\xi_{p}^{D} &= (\xi_{p}^{(D1)},\cdots, \xi_{p}^{(Dd^{2})}), 
\nonumber\\
\xi_{q}^{A} &= (\xi_{q}^{(A1)},\cdots, \xi_{q}^{(Ad^{2})}),
\nonumber\\
\xi_{p}^{A} &= (\xi_{p}^{(A1)},\cdots, \xi_{p}^{(Ad^{2})}). 
\end{align}
In Step 1, we add random shift errors occurring during the GKP state preparation as follows: 
\begin{align}
\xi_{q}^{(Dk)} &\leftarrow \xi_{q}^{(Dk)}  + \textrm{randG}(\sigma^{2}), 
\nonumber\\
\xi_{p}^{(Dk)} &\leftarrow \xi_{p}^{(Dk)}  + \textrm{randG}(\sigma^{2}),
\nonumber\\
\xi_{q}^{(Ak)} &\leftarrow  \textrm{randG}(\sigma_{\textrm{gkp}}^{2} ), 
\nonumber\\
\xi_{p}^{(Ak)} &\leftarrow  \textrm{randG}(\sigma_{\textrm{gkp}}^{2} ), \label{eq:noise during gkp preparation} 
\end{align}
for $k \in\lbrace 1,\cdots, d^{2} \rbrace$ where $\textrm{randG}(\boldsymbol{V})$ generates a random vector sampled from a multivariate Gaussian distribution $\mathcal{N}(0, \boldsymbol{V})$ with zero mean and the covariance matrix $\boldsymbol{V}$. Then, due to the SUM and the inverse-SUM gates, the quadrature noise vectors are updated as follows. 
\begin{align}
(\xi_{q}^{(Dk)} , \xi_{q}^{(Ak)} ) &\leftarrow  (\xi_{q}^{(Dk)} , \xi_{q}^{(Ak)} +\xi_{q}^{(Dk)} ) + \textrm{randG}\Big{(} \sigma^{2}\begin{bmatrix}
1 & 1/2 \\
1/2 & 4/3
\end{bmatrix} \Big{)} , 
\nonumber\\
(\xi_{p}^{(Dk)} , \xi_{p}^{(Ak)} ) &\leftarrow  (\xi_{p}^{(Dk)} - \xi_{p}^{(Ak)} , \xi_{p}^{(Ak)} ) + \textrm{randG}\Big{(} \sigma^{2}\begin{bmatrix}
4/3 & -1/2 \\
-1/2 & 1
\end{bmatrix} \Big{)} , \label{eq:SUM between data and ancilla}
\end{align} 
for odd $k$ ($\hat{S}_{q}^{(k)}$ stabilizer measurement) and 
\begin{align}
(\xi_{q}^{(Dk)} , \xi_{q}^{(Ak)} ) &\leftarrow  (\xi_{q}^{(Dk)} - \xi_{q}^{(Ak)} , \xi_{q}^{(Ak)} ) + \textrm{randG}\Big{(} \sigma^{2}\begin{bmatrix}
4/3 & -1/2 \\
-1/2 & 1
\end{bmatrix} \Big{)} , 
\nonumber\\
(\xi_{p}^{(Dk)} , \xi_{p}^{(Ak)} ) &\leftarrow  (\xi_{p}^{(Dk)} , \xi_{p}^{(Ak)} + \xi_{p}^{(Dk)} ) + \textrm{randG}\Big{(} \sigma^{2}\begin{bmatrix}
1 & 1/2 \\
1/2 & 4/3
\end{bmatrix} \Big{)} ,  \label{eq:inverse-SUM between data and ancilla}
\end{align} 
for even $k$ ($\hat{S}_{p}^{(k)}$ stabilizer measurement). Due to the noise before (or during) the homodyne measurement, the noise vectors are updated as 
 \begin{align}
\xi_{q}^{(Dk)} &\leftarrow \xi_{q}^{(Dk)}  + \textrm{randG}(\sigma^{2}), 
\nonumber\\
\xi_{p}^{(Dk)} &\leftarrow \xi_{p}^{(Dk)}  + \textrm{randG}(\sigma^{2}),
\nonumber\\
\xi_{q}^{(Ak)} &\leftarrow \xi_{q}^{(Ak)}  + \textrm{randG}(\sigma^{2} ), 
\nonumber\\
\xi_{p}^{(Ak)} &\leftarrow \xi_{p}^{(Ak)}  + \textrm{randG}(\sigma^{2} ), \label{eq:noise during homodyne measurement}
\end{align}
for all $k\in \lbrace 1,\cdots, d^{2} \rbrace$. Then, through the homodyne measurement and the error correction process, the data noise vectors are transformed as 
\begin{align}
&\xi_{q}^{(Dk)} \leftarrow \xi_{q}^{(Dk)}  - R_{\sqrt{\pi}} \big{(} \xi_{q}^{(Ak)} \big{)}, \label{eq:position correction}
\\
&\xi_{p}^{(Dk)} \leftarrow \xi_{p}^{(Dk)}  - R_{\sqrt{\pi}} \big{(} \xi_{p}^{(Ak)} \big{)} ,  \label{eq:momentum correction}
\end{align}
for odd $k$ (Eq.\ \eqref{eq:position correction}) and even $k$ (Eq.\ \eqref{eq:momentum correction}), respectively. $R_{s}(z)$ is defined as 
\begin{align}
R_{s}(z) \equiv z - s \Big{ \lfloor} \frac{z}{s}+ \frac{1}{2} \Big{\rfloor} . \label{eq:definition of the R function app}
\end{align}

In Step 2, $\hat{S}_{p}^{(k)}$ ($\hat{S}_{q}^{(k)}$) stabilizers are measured for odd (even) $k$ instead of $\hat{S}_{q}^{(k)}$ ($\hat{S}_{p}^{(k)}$). Thus, the noise vectors are updated similarly as in Eqs.\ \eqref{eq:noise during gkp preparation}--\eqref{eq:momentum correction}, except that Eqs.\ \eqref{eq:SUM between data and ancilla} and \eqref{eq:position correction} (Eqs, \eqref{eq:inverse-SUM between data and ancilla} and \eqref{eq:momentum correction}) are applied when $k$ is even (odd) instead of when $k$ is odd (even).

\subsubsection{Surface code stabilizer measurements}

Recall that there are $d'\equiv (d^{2}-1)/2$ $Z$-type and $X$-type syndrome GKP qubits that are used to measure the surface code stabilizers. Let $\xi_{q}^{Z}$ and $\xi_{p}^{Z}$ ($\xi_{q}^{X}$ and $\xi_{p}^{X}$) be the position and momentum noise vectors of the $Z$-type ($X$-type) syndrome GKP qubits, where 
\begin{align}
\xi_{q}^{Z} &= (\xi_{q}^{(Z1)},\cdots, \xi_{q}^{(Zd')}),
\nonumber\\
\xi_{p}^{Z} &= (\xi_{p}^{(Z1)},\cdots, \xi_{p}^{(Zd')}), 
\nonumber\\
\xi_{q}^{X} &= (\xi_{q}^{(X1)},\cdots, \xi_{q}^{(Xd')}),
\nonumber\\
\xi_{p}^{X} &= (\xi_{p}^{(X1)},\cdots, \xi_{p}^{(Xd')}). 
\end{align}
Note that the SUM and the inverse-SUM gates for the syndrome extraction are executed in four time steps (see Steps 3,4,5,6 in Fig.\ \ref{fig:surface code stabilizer measurement four time steps}). Let $Z_{1}(k),\cdots, Z_{4}(k)$ ($X_{1}(k),\cdots, X_{4}(k)$) be the label of the data GKP qubit that the $k^{\textrm{th}}$ $Z$-type ($X$-type) syndrome GKP qubit is coupled with in Steps $3,\cdots ,6$. (If the syndrome GKP qubit is idling, the value is set to be zero). For example when $d=3$, $Z_{1}(k)$ and $X_{1}(k)$ are given by 
\begin{alignat}{4}
&Z_{1}(1) = 1,\,\,  & Z_{1}(2) &= 3,\,\, & Z_{1}(3) &= 5,\,\, & Z_{1}(4) &= 0, 
\nonumber\\
&X_{1}(1) = 2,\,\,  & X_{1}(2) &= 0,\,\, & X_{1}(3) &= 8,\,\, & X_{1}(4) &= 6, 
\end{alignat}
representing the connectivity between the syndrome and the data GKP qubits in Step 3.  

\begin{figure}[t!]
\centering
\includegraphics[width=5.8in]{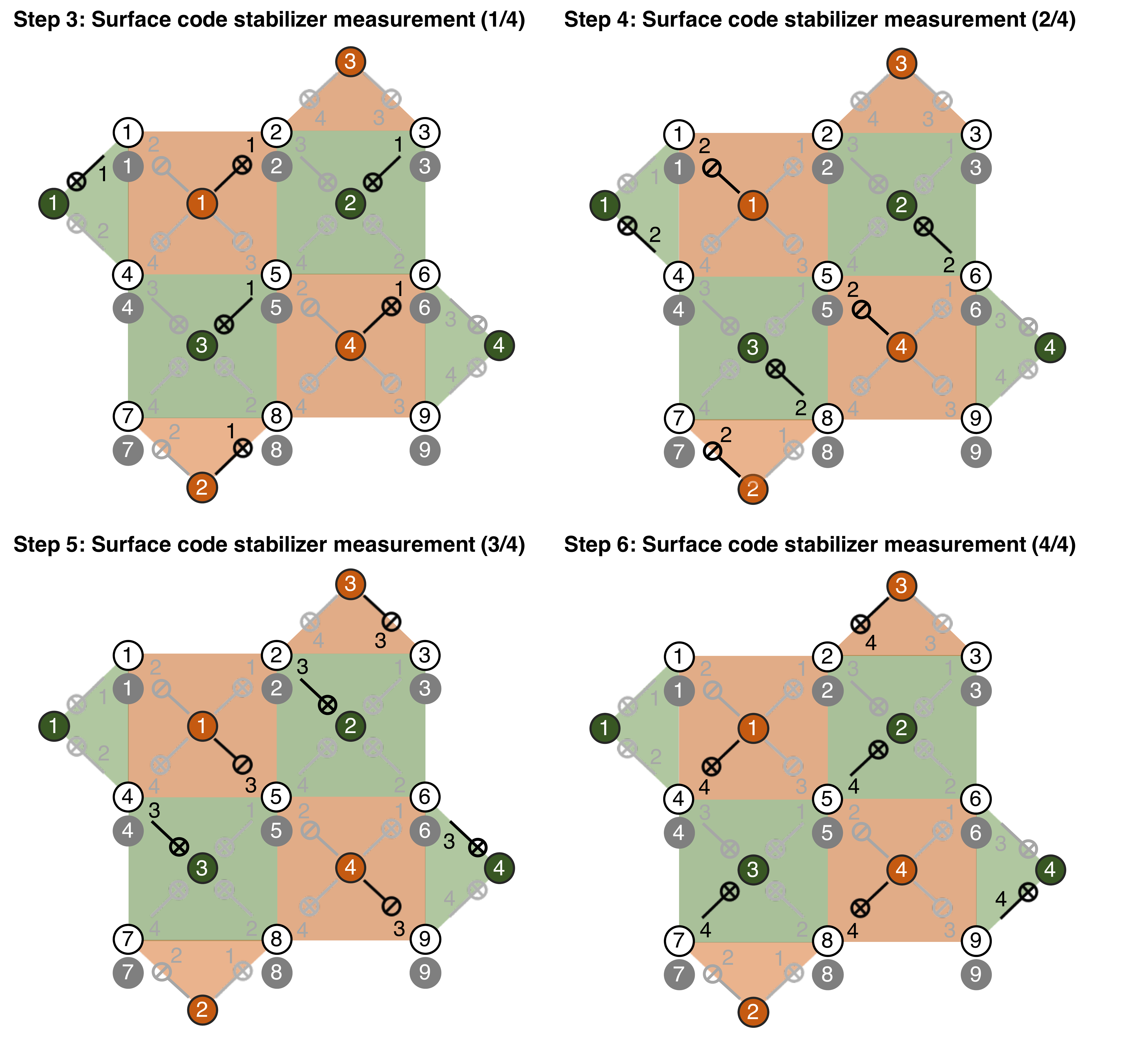}
\caption{[Fig.\ 9 in PRA \textbf{101}, 012316 (2020)] Measurement of the surface code stabilizers for $d=3$.  }
\label{fig:surface code stabilizer measurement four time steps}
\end{figure}

Due to the shift errors occurring during the preparation of GKP states, the noise vectors are updated as follows: 
\begin{align}
\xi_{q}^{(Dk)} &\leftarrow \xi_{q}^{(Dk)}  + \textrm{randG}(\sigma^{2}), 
\nonumber\\
\xi_{p}^{(Dk)} &\leftarrow \xi_{p}^{(Dk)}  + \textrm{randG}(\sigma^{2}),
\nonumber\\
\xi_{q}^{(Z\ell )} &\leftarrow  \textrm{randG}(\sigma_{\textrm{gkp}}^{2} ), 
\nonumber\\
\xi_{p}^{(Z\ell )} &\leftarrow  \textrm{randG}(\sigma_{\textrm{gkp}}^{2} ),
\nonumber\\
\xi_{q}^{(X\ell )} &\leftarrow  \textrm{randG}(\sigma_{\textrm{gkp}}^{2} ), 
\nonumber\\
\xi_{p}^{(X\ell )} &\leftarrow  \textrm{randG}(\sigma_{\textrm{gkp}}^{2} ), \label{eq:noise during gkp preparation surface code} 
\end{align}
for $k \in\lbrace 1,\cdots, d^{2} \rbrace$ and $\ell \in \lbrace 1,\cdots, d'  \rbrace$. In Step 3, the SUM gates transform the noise vectors as 
\begin{align}
(\xi_{q}^{(DZ_{1}(\ell))} , \xi_{q}^{(Z\ell)} ) &\leftarrow  (\xi_{q}^{(DZ_{1}(\ell))} , \xi_{q}^{(Z\ell)} +\xi_{q}^{(DZ_{1}(\ell))} ) + \textrm{randG}\Big{(} \sigma^{2}\begin{bmatrix}
1 & 1/2 \\
1/2 & 4/3
\end{bmatrix} \Big{)} , 
\nonumber\\
(\xi_{p}^{(DZ_{1}(\ell))} , \xi_{p}^{(Z\ell)} ) &\leftarrow  (\xi_{p}^{(DZ_{1}(\ell))} - \xi_{p}^{(Z\ell)} , \xi_{p}^{(Z\ell)} )  + \textrm{randG}\Big{(} \sigma^{2}\begin{bmatrix}
4/3 & -1/2 \\
-1/2 & 1
\end{bmatrix} \Big{)} , \label{eq:SUM gates between Z-type and data step 3}
\end{align}
for all $\ell \in \lbrace 1,\cdots, d' \rbrace$ if $Z_{1}(\ell)\neq 0$ and 
\begin{align}
\xi_{q}^{(Z\ell)} &\leftarrow \xi_{q}^{(Z\ell)} + \textrm{randG}(\sigma^{2}),
\nonumber\\
\xi_{p}^{(Z\ell)} &\leftarrow \xi_{p}^{(Z\ell)} + \textrm{randG}(\sigma^{2}), \label{eq:SUM gates between Z-type and data step 3 idling Z-type}
\end{align}
if $Z_{1}(\ell) =  0$. Similarly, 
\begin{align}
(\xi_{q}^{(DX_{1}(\ell))} , \xi_{q}^{(X\ell)} ) &\leftarrow  (\xi_{q}^{(DX_{1}(\ell))} + \xi_{q}^{(X\ell)}  , \xi_{q}^{(X\ell)}  )  + \textrm{randG}\Big{(} \sigma^{2}\begin{bmatrix}
4/3 & 1/2 \\
1/2 & 1
\end{bmatrix} \Big{)} , 
\nonumber\\
(\xi_{p}^{(DX_{1}(\ell))} , \xi_{p}^{(X\ell)} ) &\leftarrow  (\xi_{p}^{(DX_{1}(\ell))}  , \xi_{p}^{(X\ell)} -\xi_{p}^{(DX_{1}(\ell))} )  + \textrm{randG}\Big{(} \sigma^{2}\begin{bmatrix}
1 & -1/2 \\
-1/2 & 4/3
\end{bmatrix} \Big{)} , \label{eq:SUM gates between $X$-type and data step 3}
\end{align}
for all $\ell \in \lbrace 1,\cdots, d' \rbrace$ if $X_{1}(\ell)\neq 0$ and 
\begin{align}
\xi_{q}^{(X\ell)} &\leftarrow \xi_{q}^{(X\ell)} + \textrm{randG}(\sigma^{2}),
\nonumber\\
\xi_{p}^{(X\ell)} &\leftarrow \xi_{p}^{(X\ell)} + \textrm{randG}(\sigma^{2}), \label{eq:SUM gates between $X$-type and data step 3 idling $X$-type}
\end{align}
if $X_{1}(\ell) = 0$. Since there are idling data GKP qubits, the data noise vectors are updated as 
\begin{align}
\xi_{q}^{(Dk)} &\leftarrow \xi_{q}^{(Dk)} + \textrm{randG}(\sigma^{2}), 
\nonumber\\
\xi_{p}^{(Dk)} &\leftarrow \xi_{p}^{(Dk)} + \textrm{randG}(\sigma^{2}), \label{eq:idling data step 3}
\end{align}
only for $k$ such that $Z_{1}(\ell)\neq k$ and $X_{1}(\ell)\neq k$ for all $\ell \in \lbrace 1,\cdots ,d' \rbrace$.  

In Step 4, the SUM gates between the $Z$-type syndrome GKP qubits and data GKP qubits transform the noise vectors in the same way as in Eqs.\ \eqref{eq:SUM gates between Z-type and data step 3} and \eqref{eq:SUM gates between Z-type and data step 3 idling Z-type} except that $Z_{1}(\ell)$ is replaced by $Z_{2}(\ell)$. However, since the $X$-type syndrome GKP qubits are coupled with the data GKP qubits through inverse-SUM gates instead of SUM gates, the noise vectors are then updated as 
\begin{align}
(\xi_{q}^{(DX_{2}(\ell))} , \xi_{q}^{(X\ell)} ) &\leftarrow  (\xi_{q}^{(DX_{2}(\ell))} - \xi_{q}^{(X\ell)}  , \xi_{q}^{(X\ell)}  ) + \textrm{randG}\Big{(} \sigma^{2}\begin{bmatrix}
4/3 & -1/2 \\
-1/2 & 1
\end{bmatrix} \Big{)} , 
\nonumber\\
(\xi_{p}^{(DX_{2}(\ell))} , \xi_{p}^{(X\ell)} ) &\leftarrow  (\xi_{p}^{(DX_{2}(\ell))}  , \xi_{p}^{(X\ell)} +\xi_{p}^{(DX_{2}(\ell))} )  + \textrm{randG}\Big{(} \sigma^{2}\begin{bmatrix}
1 & 1/2 \\
1/2 & 4/3
\end{bmatrix} \Big{)} , \label{eq:inverse-SUM gates between $X$-type and data step 4}
\end{align}
for all $\ell \in \lbrace 1,\cdots, d' \rbrace$ if $X_{2}(\ell)\neq 0$ and 
\begin{align}
\xi_{q}^{(X\ell)} &\leftarrow \xi_{q}^{(X\ell)} + \textrm{randG}(\sigma^{2}),
\nonumber\\
\xi_{p}^{(X\ell)} &\leftarrow \xi_{p}^{(X\ell)} + \textrm{randG}(\sigma^{2}), \label{eq:inverse-SUM gates between $X$-type and data step 4 idling Z-type}
\end{align}
if $X_{2}(\ell) = 0$, instead of as in Eqs.\ \eqref{eq:SUM gates between $X$-type and data step 3} and \eqref{eq:SUM gates between $X$-type and data step 3 idling $X$-type}. Due to the idling data GKP qubits, the noise vectors are further updated as in Eq.\ \eqref{eq:idling data step 3} only for $k$ such that $Z_{2}(\ell)\neq k$ and $X_{2}(\ell)\neq k$ for all $\ell \in \lbrace 1,\cdots ,d' \rbrace$. 

Note that in Step 5 and Step 6, the $X$-type syndrome GKP qubits are coupled with the data GKP qubits via inverse-SUM gates and SUM gates, respectively. Therefore, in Step 5, the noise vectors are updated in the same way as in Step 4, except that $Z_{2}(\ell)$ and $Z_{2}(\ell)$ are replaced by $Z_{3}(\ell)$ and $X_{3}(\ell)$. On the other hand, in Step 6, the noise vectors are updated in the same way as in Step 3, except that $Z_{1}(\ell)$ and $X_{1}(\ell)$ are replaced by $Z_{4}(\ell)$ and $X_{4}(\ell)$. Due to the noise before (or during) the homodyne measurement, the noise vectors are updated as 
\begin{align}
\xi_{q}^{(Dk)} &\leftarrow \xi_{q}^{(Dk)}  + \textrm{randG}(\sigma^{2}), 
\nonumber\\
\xi_{p}^{(Dk)} &\leftarrow \xi_{p}^{(Dk)}  + \textrm{randG}(\sigma^{2}),
\nonumber\\
\xi_{q}^{(Z\ell )} &\leftarrow \xi_{q}^{(Z\ell )} +  \textrm{randG}(\sigma^{2} ), 
\nonumber\\
\xi_{p}^{(Z\ell )} &\leftarrow \xi_{p}^{(Z\ell )} +  \textrm{randG}(\sigma^{2} ), 
\nonumber\\
\xi_{q}^{(X\ell )} &\leftarrow \xi_{q}^{(X\ell )} +  \textrm{randG}(\sigma^{2} ), 
\nonumber\\
\xi_{p}^{(X\ell )} &\leftarrow \xi_{p}^{(X\ell )} +  \textrm{randG}(\sigma^{2} ),  \label{eq:noise during homodyne measurement surface code}
\end{align}
for all $k\in\lbrace 1,\cdots, d^{2} \rbrace$ and $\ell\in\lbrace 1,\cdots, d' \rbrace$. Then, through the homodyne measurement, we measure $\xi_{q}^{(Z\ell )}$ and $\xi_{p}^{(X\ell )}$ modulo $2\sqrt{\pi}$ and assign stabilizer values as 
\begin{align}
\hat{S}_{Z}^{(\ell)} &\leftarrow \begin{cases}
+1 & | R_{\sqrt{2\pi}}( \xi_{q}^{(Z\ell )} )  |  \le \sqrt{\pi}/2  \\
-1 & | R_{\sqrt{2\pi}}( \xi_{q}^{(Z\ell )} )  |  > \sqrt{\pi}/2 
\end{cases}, 
\nonumber\\
\hat{S}_{X}^{(\ell)} &\leftarrow \begin{cases}
+1 & | R_{\sqrt{2\pi}}( \xi_{p}^{(X\ell )} )  |  \le \sqrt{\pi}/2  \\
-1 & | R_{\sqrt{2\pi}}( \xi_{p}^{(X\ell )} )  |  > \sqrt{\pi}/2 
\end{cases}, \label{eq:stabilizer value assignment}
\end{align}
for all $\ell \in \lbrace 1,\cdots, d' \rbrace$. $R_{s}(z)$ is defined in Eq.\ \eqref{eq:definition of the R function app}.

\subsubsection{Construction of three-dimensional space-time graphs}

Now we construct 3D space-time graphs to which we will apply a minimum-weight perfect matching decoding algorithm. The overall structure is as follows: Since each stabilizer measurement can be faulty, we repeat the noisy stabilizer measurement cycle $d$ times. Then, we perform another round of ideal stabilizer measurement cycle assuming that all circuit elements and supplied GKP states are noiseless. The reason for adding the extra noiseless measurement cycle is to ensure that the noisy states are restored back to the code space so we can later conveniently determine whether the error correction succeed or not. Then, the $Z$-type and the $X$-type 3D space-time graphs are constructed to represent the outcomes of $d+1$ rounds of stabilizer measurement cycles. These space-time graphs will then be used to decode the $Z$-type and the $X$-type syndrome measurement outcomes.

\begin{figure}[t!]
\centering
\includegraphics[width=5.8in]{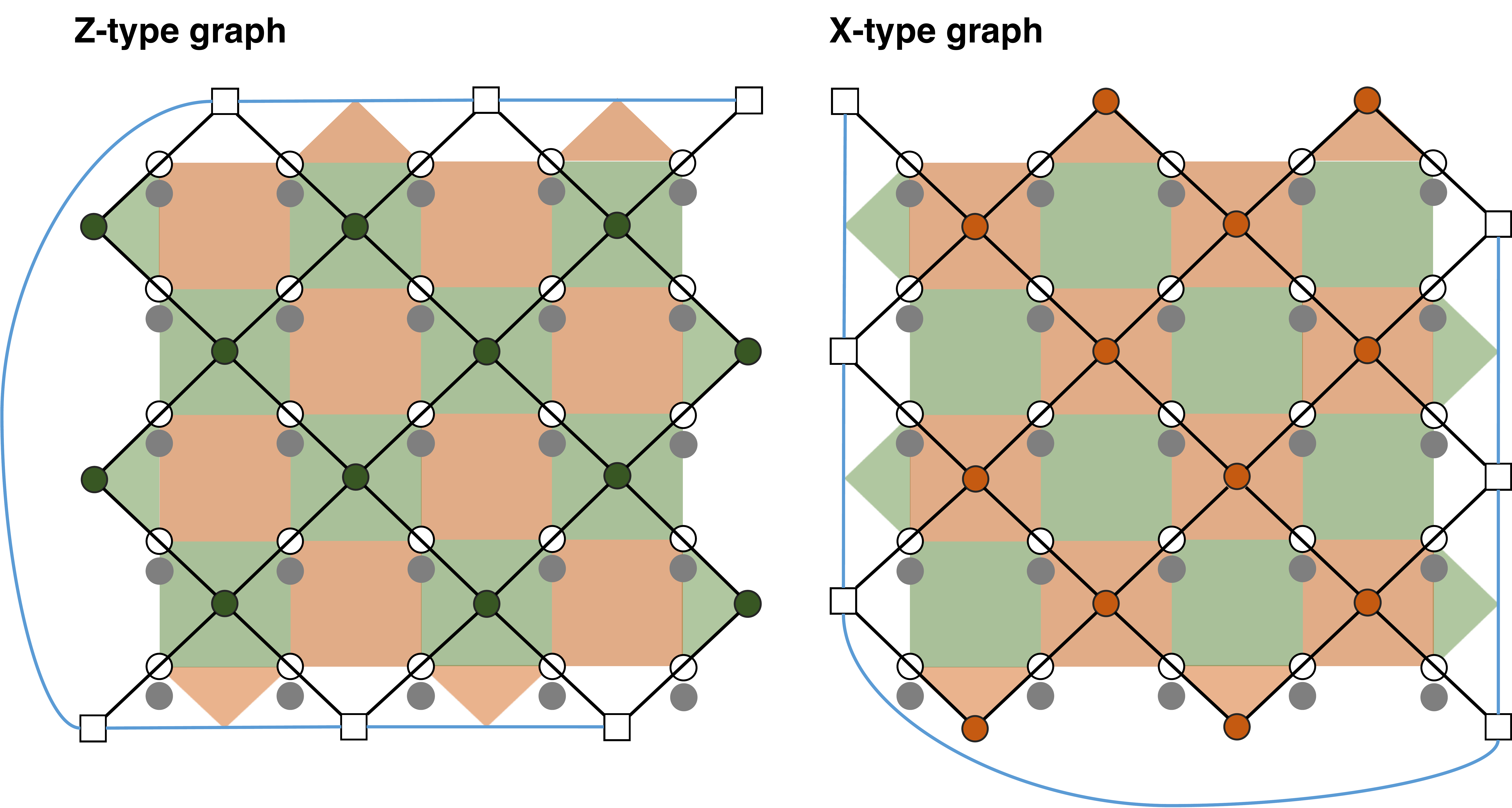}
\caption{[Fig.\ 10 in PRA \textbf{101}, 012316 (2020)] $Z$-type and $X$-type 2D space graphs for the surface-GKP code with $d=5$. These 2D graphs will be stacked up to construct $Z$-type and $X$-type 3D space-time graphs. }
\label{fig:MWPM graphs} 
\end{figure}

We first construct the $Z$-type and $X$-type 2D space graphs as in Fig.\ \ref{fig:MWPM graphs}. Each bulk vertex of the 2D space graph corresponds to a syndrome GKP qubit and each bulk edge corresponds to a data GKP qubit. Note also that there are boundary vertices (squares in Fig.\ \ref{fig:MWPM graphs}) that do not correspond to any syndrome GKP qubits and the corresponding boundary edges (blue lines in Fig.\ \ref{fig:MWPM graphs}) that are not associated with any data GKP qubits. Therefore, the boundary edge weighs are always set to be zero. 

Then, we associate each 2D space graph with one round of stabilizer measurement cycle. So, there are $d+1$ 2D space graphs and these 2D space graphs are stacked up together by introducing vertical edges that connect the same vertices in two adjacent 2D space graphs (corresponding to two adjacent stabilizer measurement rounds). Below, we discuss in detail how the bulk edge weights are assigned. 

We start by initializing the data position and momentum noise vectors to a zero vector: 
\begin{align}
\xi_{q}^{D} &= (\xi_{q}^{(D1)}, \cdots, \xi_{q}^{(Dd^{2})}) = (0,\cdots, 0), 
\nonumber\\
\xi_{p}^{D} &= (\xi_{p}^{(D1)}, \cdots, \xi_{p}^{(Dd^{2})}) = (0,\cdots, 0). 
\end{align} 
These data noise vectors are fed into Step 1 of GKP-stabilizer measurement as described in Eqs.\ \eqref{eq:noise during gkp preparation}--\eqref{eq:noise during homodyne measurement}. Let $w_{Z}^{H}(k)$ and $w_{Z}^{H}(k)$ be the horizontal edge weights of the $Z$-type and $X$-type graphs corresponding to the $k^{\textrm{th}}$ data GKP qubit ($k\in\lbrace 1,\cdots, d^{2} \rbrace$). Then, while updating the data position and momentum noise vectors as prescribed in Eqs.\ \eqref{eq:position correction} and \eqref{eq:momentum correction}, we assign the horizontal edge weights as
\begin{align}
w_{Z}^{H}(k) &\leftarrow \begin{cases}
-\log_{2}\big{(} p[ \sqrt{\sigma_{\textrm{gkp}}^{2} + \frac{10}{3}\sigma^{2}} ]\big{(} R_{\sqrt{\pi}}( \xi_{q}^{(Ak)} ) \big{)} \big{)} & \textrm{round }1 \\
-\log_{2}\big{(} p[ \sigma_{Z}^{H}(k;d) ]\big{(} R_{\sqrt{\pi}}( \xi_{q}^{(Ak)} ) \big{)} \big{)} &  \textrm{round }2\textrm{ to round }d\\
-\log_{2}\big{(} p[ \sqrt{ (\sigma_{Z}^{H}(k;d) )^{2} - \sigma_{\textrm{gkp}}^{2} - \frac{10}{3}\sigma^{2}  } ]\big{(} R_{\sqrt{\pi}}( \xi_{q}^{(Ak)} )   \big{)} \big{)} & \textrm{round }d+1
\end{cases}, 
\end{align}
for odd $k$ and 
\begin{align}
w_{X}^{H}(k) &\leftarrow \begin{cases}
-\log_{2}\big{(} p[ \sqrt{\sigma_{\textrm{gkp}}^{2} + \frac{10}{3}\sigma^{2}} ]\big{(} R_{\sqrt{\pi}}( \xi_{p}^{(Ak)} ) \big{)} \big{)} & \textrm{round }1 \\
-\log_{2}\big{(} p[ \sigma_{X}^{H}(k;d) ]\big{(} R_{\sqrt{\pi}}( \xi_{p}^{(Ak)} ) \big{)} \big{)} &  \textrm{round }2\textrm{ to round }d\\
-\log_{2}\big{(} p[ \sqrt{ (\sigma_{X}^{H}(k;d) )^{2} - \sigma_{\textrm{gkp}}^{2} - \frac{10}{3}\sigma^{2}  } ]\big{(} R_{\sqrt{\pi}}( \xi_{p}^{(Ak)} )   \big{)} \big{)} & \textrm{round }d+1
\end{cases}, 
\end{align}
for even $k$ if the additional GKP information is used. Here, we use $\xi_{q}^{(Ak)}$ and $\xi_{p}^{(Ak)}$ that are obtained after applying Eq.\ \eqref{eq:noise during homodyne measurement}. $p_{\textrm{err}}(\sigma)$ and $p[\sigma](z)$ are defined in Eqs.\ \eqref{eq:definition of perr} and \eqref{eq:Definition of the p function} and $R_{s}(z)$ is defined in Eq.\ \eqref{eq:definition of the R function app}. On the other hand, if the additional GKP information is not used, we assign the horizontal edge weights as 
\begin{align}
w_{Z}^{H}(k) &\leftarrow \begin{cases}
-\log_{2}\big{(} p_{\textrm{err}} \big{(}  \sqrt{\sigma_{\textrm{gkp}}^{2} + \frac{10}{3}\sigma^{2}} \big{)} \big{)} & \textrm{round }1 \\
-\log_{2}\big{(} p_{\textrm{err}} \big{(} \sigma_{Z}^{H}(k;d) \big{)} \big{)} &  \textrm{round }2\textrm{ to round }d\\
-\log_{2}\big{(} p_{\textrm{err}} \big{(} \sqrt{ (\sigma_{Z}^{H}(k;d) )^{2} - \sigma_{\textrm{gkp}}^{2} - \frac{10}{3}\sigma^{2}  } \big{)} \big{)} & \textrm{round }d+1
\end{cases}, 
\end{align}
for odd $k$ and 
\begin{align}
w_{X}^{H}(k) &\leftarrow \begin{cases}
-\log_{2}\big{(} p_{\textrm{err}} \big{(} \sqrt{\sigma_{\textrm{gkp}}^{2} + \frac{10}{3}\sigma^{2}} \big{)} \big{)}  & \textrm{round }1 \\
-\log_{2}\big{(} p_{\textrm{err}} \big{(} \sigma_{X}^{H}(k;d) \big{)} \big{)}  &  \textrm{round }2\textrm{ to round }d\\
-\log_{2}\big{(} p_{\textrm{err}} \big{(} \sqrt{ (\sigma_{X}^{H}(k;d) )^{2} - \sigma_{\textrm{gkp}}^{2} - \frac{10}{3}\sigma^{2}  } \big{)} \big{)} & \textrm{round }d+1
\end{cases}, 
\end{align}
for even $k$. Here, $\sigma_{Z}^{H}(k;d)$ and $\sigma_{X}^{H}(k;d)$ are defined as 
\begin{align}
\sigma_{Z}^{H}(k;d) &\equiv \begin{cases}
\begin{cases}
\sqrt{  4\sigma_{\textrm{gkp}}^{2} + \frac{52}{3}\sigma^{2} }  & \frac{k-1}{d}\in 2\mathbb{Z} \\
\sqrt{  4\sigma_{\textrm{gkp}}^{2} + \frac{58}{3}\sigma^{2} }  & \frac{k-1}{d}\in 2\mathbb{Z}+1
\end{cases} & k\in d\mathbb{Z} +1 \\
\begin{cases}
\sqrt{  4\sigma_{\textrm{gkp}}^{2} + \frac{55}{3}\sigma^{2} }  & \frac{k}{d}\in 2\mathbb{Z}+1 \\
\sqrt{  4\sigma_{\textrm{gkp}}^{2} + \frac{49}{3}\sigma^{2} }  & \frac{k}{d}\in 2\mathbb{Z}
\end{cases} & k\in d\mathbb{Z} \\
\sqrt{  5\sigma_{\textrm{gkp}}^{2} + \frac{59}{3}\sigma^{2} } & \textrm{otherwise}
\end{cases}, 
\nonumber\\
\sigma_{X}^{H}(k;d) &\equiv \begin{cases}
\begin{cases}
\sqrt{  4\sigma_{\textrm{gkp}}^{2} + \frac{49}{3}\sigma^{2} } & k\in 2\mathbb{Z} + 1  \\
\sqrt{  4\sigma_{\textrm{gkp}}^{2} + \frac{55}{3}\sigma^{2} } & k\in 2\mathbb{Z} 
\end{cases} &  k\in \lbrace 1,\cdots, d \rbrace \\
\begin{cases}
\sqrt{  4\sigma_{\textrm{gkp}}^{2} + \frac{58}{3}\sigma^{2} } & k\in 2\mathbb{Z} + 1  \\
\sqrt{  4\sigma_{\textrm{gkp}}^{2} + \frac{52}{3}\sigma^{2} } & k\in 2\mathbb{Z} 
\end{cases} &  k\in \lbrace d^{2}-d+1 ,\cdots, d^{2} \rbrace \\
\sqrt{  5\sigma_{\textrm{gkp}}^{2} + \frac{59}{3}\sigma^{2} }&\textrm{otherwise}
\end{cases}. 
\end{align}
We remark that we have carefully determined $\sigma_{Z}^{H}(k;d)$ and $\sigma_{X}^{H}(k;d)$ by thoroughly keeping tracking of how the circuit-level noise propagates   

Then, moving on to Step 2 of GKP-stabilizer measurement, we update the noise vectors as described in Eq.\ \eqref{eq:noise during gkp preparation}--\eqref{eq:noise during homodyne measurement}, except that Eqs.\ \eqref{eq:SUM between data and ancilla} and \eqref{eq:inverse-SUM between data and ancilla} are applied for even and odd $k$ (instead of odd and even $k$), respectively. Similarly as above, while updating the data position and momentum noise vectors as prescribed in Eqs.\ \eqref{eq:position correction} and \eqref{eq:momentum correction}, we assign the horizontal edge weights as 
\begin{align}
w_{Z}^{H}(k) &\leftarrow \begin{cases}
-\log_{2}\big{(} p[ \sqrt{\sigma_{2\textrm{gkp}}^{2} + \frac{20}{3}\sigma^{2}} ]\big{(} R_{\sqrt{\pi}}( \xi_{q}^{(Ak)} ) \big{)} \big{)} & \textrm{round }1 \\
-\log_{2}\big{(} p[ \sigma_{Z}^{H}(k;d) ]\big{(} R_{\sqrt{\pi}}( \xi_{q}^{(Ak)} ) \big{)} \big{)} &  \textrm{round }2\textrm{ to round }d\\
-\log_{2}\big{(} p[ \sqrt{ (\sigma_{Z}^{H}(k;d) )^{2} - 2\sigma_{\textrm{gkp}}^{2} - \frac{20}{3}\sigma^{2}  } ]\big{(} R_{\sqrt{\pi}}( \xi_{q}^{(Ak)} )   \big{)} \big{)} & \textrm{round }d+1
\end{cases}, 
\end{align}
for even $k$ and 
\begin{align}
w_{X}^{H}(k) &\leftarrow \begin{cases}
-\log_{2}\big{(} p[ \sqrt{2\sigma_{\textrm{gkp}}^{2} + \frac{20}{3}\sigma^{2}} ]\big{(} R_{\sqrt{\pi}}( \xi_{p}^{(Ak)} ) \big{)} \big{)} & \textrm{round }1 \\
-\log_{2}\big{(} p[ \sigma_{X}^{H}(k;d) ]\big{(} R_{\sqrt{\pi}}( \xi_{p}^{(Ak)} ) \big{)} \big{)} &  \textrm{round }2\textrm{ to round }d\\
-\log_{2}\big{(} p[ \sqrt{ (\sigma_{X}^{H}(k;d) )^{2} - 2\sigma_{\textrm{gkp}}^{2} - \frac{20}{3}\sigma^{2}  } ]\big{(} R_{\sqrt{\pi}}( \xi_{p}^{(Ak)} )   \big{)} \big{)} & \textrm{round }d+1
\end{cases}, 
\end{align}
for odd $k$ if the additional GKP information is used. Here, we use $\xi_{q}^{(Ak)}$ and $\xi_{p}^{(Ak)}$ that are obtained after applying Eq.\ \eqref{eq:noise during homodyne measurement}. If on the other hand the additional GKP information is not used, we assign the horizontal edge weights as 
\begin{align}
w_{Z}^{H}(k) &\leftarrow \begin{cases}
-\log_{2}\big{(} p_{\textrm{err}} \big{(}  \sqrt{2\sigma_{\textrm{gkp}}^{2} + \frac{20}{3}\sigma^{2}} \big{)} \big{)} & \textrm{round }1 \\
-\log_{2}\big{(} p_{\textrm{err}} \big{(} \sigma_{Z}^{H}(k;d) \big{)} \big{)} &  \textrm{round }2\textrm{ to round }d\\
-\log_{2}\big{(} p_{\textrm{err}} \big{(} \sqrt{ (\sigma_{Z}^{H}(k;d) )^{2} - 2\sigma_{\textrm{gkp}}^{2} - \frac{20}{3}\sigma^{2}  } \big{)} \big{)} & \textrm{round }d+1
\end{cases}, 
\end{align}
for even $k$ and 
\begin{align}
w_{X}^{H}(k) &\leftarrow \begin{cases}
-\log_{2}\big{(} p_{\textrm{err}} \big{(} \sqrt{2\sigma_{\textrm{gkp}}^{2} + \frac{20}{3}\sigma^{2}} \big{)} \big{)}  & \textrm{round }1 \\
-\log_{2}\big{(} p_{\textrm{err}} \big{(} \sigma_{X}^{H}(k;d) \big{)} \big{)} &  \textrm{round }2\textrm{ to round }d\\
-\log_{2}\big{(} p_{\textrm{err}} \big{(} \sqrt{ (\sigma_{X}^{H}(k;d) )^{2} - 2\sigma_{\textrm{gkp}}^{2} - \frac{20}{3}\sigma^{2}  } \big{)} \big{)} & \textrm{round }d+1
\end{cases}, 
\end{align}
for odd $k$. This way, all the horizontal edge weights are assigned. 

Vertical edge weights are assigned during surface code stabilizer measurements: We follow Steps 3--6 of surface code stabilizer measurements and update the noise vectors as described in Eqs.\ \eqref{eq:noise during gkp preparation surface code}--\eqref{eq:noise during homodyne measurement surface code}. Let $w_{Z}^{V}(\ell)$ and $w_{X}^{V}(\ell)$ be the vertical edge weights of the $Z$-type and $X$-type 3D space-time graphs corresponding to the $\ell^{\textrm{th}}$ $Z$-type and $X$-type syndrome qubit. Then, after assigning the stabilizer values as in Eq.\ \eqref{eq:stabilizer value assignment}, we further assign the vertical edge weights as follows:  
\begin{align}
w_{Z}^{V}(\ell) &\leftarrow  -\log_{2}\big{(} p[\sigma_{Z}^{V}(\ell;d)] \big{(} R_{\sqrt{\pi}}( \xi_{q}^{(Zk)} ) \big{)} \big{)}, 
\nonumber\\
w_{X}^{V}(\ell) &\leftarrow  -\log_{2}\big{(} p[\sigma_{X}^{V}(\ell;d)] \big{(} R_{\sqrt{\pi}}( \xi_{p}^{(Xk)} ) \big{)} \big{)}, 
\end{align}
while in rounds $1$ to $d$ for all $\ell \in\lbrace 1,\cdots, d' = (d^{2}-1)/2 \rbrace$, if the additional GKP information is used. Here, we use $\xi_{q}^{(Zk)}$ and $\xi_{p}^{(Xk)}$ that are obtained after applying Eq.\ \eqref{eq:noise during homodyne measurement surface code} and $\sigma_{Z}^{V}(\ell;d)$ and $\sigma_{X}^{V}(\ell;d)$ are defined as 
\begin{align}
\sigma_{Z}^{V}(\ell;d)&= \begin{cases}
\sqrt{ 4\sigma_{\textrm{gkp}}^{2} + \frac{56}{3}\sigma^{2}  }  & \ell\in 2d'' \mathbb{Z} + 1 \\
\sqrt{ 7\sigma_{\textrm{gkp}}^{2} + \frac{107}{3}\sigma^{2}  }  & \ell\in 2d'' \mathbb{Z} + d''+ 1 \\
\sqrt{ 4\sigma_{\textrm{gkp}}^{2} + \frac{73}{3}\sigma^{2}  }  & \ell\in 2d'' \mathbb{Z} \\
\sqrt{ 7\sigma_{\textrm{gkp}}^{2} + \frac{116}{3}\sigma^{2}  }&\textrm{otherwise}
\end{cases}, 
\nonumber\\
\sigma_{X}^{V}(\ell;d)&= \begin{cases}
\sqrt{  4\sigma_{\textrm{gkp}}^{2} + \frac{56}{3}\sigma^{2}  } & \ell\in 2d''\mathbb{Z} + d''  \\
\sqrt{  4\sigma_{\textrm{gkp}}^{2} + \frac{73}{3}\sigma^{2}  } & \ell\in 2d''\mathbb{Z} +d''+ 1 \\
\sqrt{  7\sigma_{\textrm{gkp}}^{2} + \frac{107}{3}\sigma^{2}  } & \ell\in 2d''\mathbb{Z}  \\
\sqrt{  7\sigma_{\textrm{gkp}}^{2} + \frac{116}{3}\sigma^{2}  } & \textrm{otherwise}\\
\end{cases}. 
\end{align}
Similarly as above, we have carefully determined $\sigma_{Z}^{V}(\ell;d)$ and $\sigma_{X}^{V}(\ell;d)$ by thoroughly keeping track of how the circuit-level noise propagates. If on the other hand the additional GKP information is not used, we assign the vertical edge weights as 
\begin{align}
w_{Z}^{V}(\ell) &\leftarrow  -\log_{2}\big{(} p_{\textrm{err}} \big{(} \sigma_{Z}^{V}(\ell;d) \big{)} \big{)}, 
\nonumber\\
w_{X}^{V}(\ell) &\leftarrow  -\log_{2}\big{(} p_{\textrm{err}} \big{(} \sigma_{X}^{V}(\ell;d) \big{)}  \big{)}. 
\end{align}
This way, all the vertical edge weights are assigned and thus we are left with the complete $Z$-type and $X$-type 3D space-time graphs with all the horizontal and vertical edge weights assigned.  

\subsubsection{Minimum-weight perfect matching} 

Now, given the 3D space-time graphs, the correction is determined by using a minimum-weight perfect matching decoding algorithm. More specifically, we do the following: 
\begin{enumerate}
\item Simulate $d$ rounds of noisy stabilizer measurements followed by one round of ideal stabilizer measurements and construct the $Z$-type and $X$-type 3D space-time graphs as described above. 
\item Highlight all vertices whose assigned stabilizer value is changed from the previous round. If the number of highlighted vertices is odd, highlight a boundary vertex. Thus, the number of highlighted vertices is always even. 
\item For all pairs of highlighted $Z$-type ($X$-type) vertices, find the path with the minimum total weight. Then, save the minimum total weight and all edges in the path. Then, we are left with a $Z$-type ($X$-type) complete graph of highlighted vertices, where the weight of the edge $(v,w)$ is given by the minimum total weight of the path that connects $v$ and $w$. 
\item Apply the minimum-weight perfect matching algorithm \cite{Edmonds1965,Edmonds1965b} on the $Z$-type ($X$-type) complete graph of highlighted vertices. For all matched pairs of $Z$-type ($X$-type) vertices, highlight all the $Z$-type ($X$-type) edges contained in the path that connects the matched vertices. 
\item Suppress all vertical edges and project the $Z$-type ($X$-type) 3D space-time graph onto the 2D plane. For each $Z$-type ($X$-type) horizontal edge, count how many times it was highlighted. If it is highlighted even times, do nothing. Otherwise, apply the Pauli correction operator $\hat{X}_{\textrm{gkp}}$ ($\hat{Z}_{\textrm{gkp}}$) to the corresponding data GKP qubit. Equivalently, update the quadrature noise as $\xi_{q}^{(Dk)}\leftarrow \xi_{q}^{(Dk)} +\sqrt{\pi}$ ($\xi_{p}^{(Dk)}\leftarrow \xi_{p}^{(Dk)} +\sqrt{\pi}$). 
\end{enumerate}

Once the correction is done, we are left with the data noise vectors $\xi_{q}^{D} = (\xi_{q}^{(D1)},\cdots ,\xi_{q}^{(Dd^{2})})$ and $\xi_{p}^{D}=(\xi_{p}^{(D1)},\cdots ,\xi_{p}^{(Dd^{2})})$. Define 
\begin{align}
\textrm{total}(\xi_{q}^{D}) &\equiv \frac{1}{\sqrt{\pi}} \sum_{k=1}^{d^{2}}\xi_{q}^{(Dk)} , 
\nonumber\\
\textrm{total}(\xi_{p}^{D}) &\equiv \frac{1}{\sqrt{\pi}} \sum_{k=1}^{d^{2}}\xi_{p}^{(Dk)} . 
\end{align}
Then, we determine that there is 
\begin{align}
\begin{cases}
\textrm{logical }X & \textrm{total}(\xi_{q}^{D}) = \textrm{odd }\& \textrm{ total}(\xi_{p}^{D}) =\textrm{even}\\
\textrm{logical }Z & \textrm{total}(\xi_{q}^{D})= \textrm{even }\& \textrm{ total}(\xi_{p}^{D}) =\textrm{odd}\\
\textrm{logical }Y & \textrm{total}(\xi_{q}^{D}) = \textrm{odd }\& \textrm{ total}(\xi_{p}^{D}) =\textrm{odd}
\end{cases}
\end{align} 
error. Otherwise if both $\textrm{total}(\xi_{q}^{D})$ and $\textrm{total}(\xi_{p}^{D})$ are even, there is no logical error. 

We use the Monte Carlo method to compute the logical $X,Y,Z$ error probability. In Fig.\ \ref{fig:main results}, we plot the logical $X$ error probability obtained from 10,000--100,000 samples, which is the same as the logical $Z$ error probability. The number of samples is determined such that statistical fluctuations are negligible.

\section{Open questions}
\label{section:Open questions fault-tolerant bosonic quantum error correction}

Note that we modeled noisy GKP states by applying an incoherent random shift error $\mathcal{N}[B_{2}](\sigma_{\textrm{gkp}})$, similarly as in Refs.\ \cite{Menicucci2014,Fukui2018a,Fukui2019}. While we use this noise model for theoretical convenience and justify it by using a twirling argument (similar to the justification of a depolarizing error model for multi-qubit QEC), we remark that it is not practical to use the twirling operation in realistic situations. This is because
the twirling operation increases the average photon number of the GKP states, whereas in practice it is desirable to keep the photon number bounded below a certain cutoff. Therefore, an interesting direction for future work would be to see if one can implement the stabilizer measurements in Figs.\ \ref{fig:GKP qubit fundamentals}, \ref{fig:Surface-GKP codes}, and \ref{fig:Surface code stabilizer measurement} in a manner that prevents the average photon number from diverging as we repeat the stabilizer measurements. It will be especially crucial to keep the average photon number under control when each bosonic mode suffers from dephasing errors and/or undesired nonlinear interactions such as Kerr nonlinearities. 

Given such stabilizer measurement schemes, it will be ideal to analyze the performance of the surface-GKP code by assuming the noise model with a coherent random shift errors, i.e., 
\begin{align}
|\psi_{\textrm{gkp}}^{\Delta}\rangle &\propto \exp[ -\Delta^{2} \hat{n} ] |\psi_{\textrm{gkp}}\rangle . 
\end{align} 
A modular subsystem decomposition formalism \cite{Pantaleoni2019} has recently been proposed as an efficient way to analyze finitely-squeezed approximate GKP codes. While this formalism has been used to analyze various single-mode schemes involving finitely-squeezed GKP states \cite{Tzitrin2019,Wan2019}, it has not been applied to analyze large-scale and fault-tolerant concatenated GKP codes. Thus, it will be interesting to see if the modular subsystem decomposition formalism proves to be useful for analyzing fault-tolerance properties of the GKP codes at scale. 

We also remark that minimum-weight perfect matching is not as good as the maximum likelihood estimation \cite{Dennis2002}. In fact, the work in Ref.\ \cite{Vuillot2019} used the maximum-likelihood decoding instead of the minimum-weight perfect matching decoding to analyze the performance of the toric-GKP code. However, the maximum-likelihood decoding algorithm was not applied to a full circuit-level noise model in Ref.\ \cite{Vuillot2019}. In our work \cite{Noh2020}, on the other hand, we have used a suboptimal minimum-weight perfect matching decoding algorithm and considered a full circuit-level noise model. Thus, it will be interesting to see if the maximum-likelihood decoding algorithm can be adapted to the surface-GKP code for a full circuit-level noise model.   

Lastly, recall that concatenating the GKP with a multi-qubit code is not the most general approach. Instead, it is possible to define a multi-mode GKP code based on a general symplectic lattice that cannot be decomposed into a single-mode GKP code and a multi-qubit error-correcting code \cite{Gottesman2001,Harrington2001,Harrington2004}. Hence, it will be an interesting research direction to see if we can benefit from such a fundamentally different approach for scaling up the GKP code.

\chapter{Quantum capacity of Gaussian thermal-loss channels}
\label{chapter:Quantum capacities of bosonic Gaussian channels}

In this chapter, I will present my contributions to the field of quantum communication theory. The main quantities of interest are quantum capacities of Gaussian thermal-loss channels. This chapter will be based on my works on upper \cite{Noh2019} and lower \cite{Noh2020a} bounds of the Gaussian thermal-loss channel capacity. The work in Ref.\ \cite{Noh2019} was done in collaboration with Dr.\ Victor Albert and Professor Liang Jiang, and the work in Ref.\ \cite{Noh2020a} was done in collaboration with Professors Stefano Pirandola and Liang Jiang.

Recall the benchmarking results for various single-mode bosonic codes in Fig.\ \ref{fig:benchmarking single-mode bosonic codes} and note that the advantage of using an error-corrected bosonic qubit disappears if the excitation loss probability becomes too large ($\gamma \gtrsim 0.44$). This shows a possibility that there may be errors that are too noisy to be corrected no matter how good an error correction scheme is. Even more generally, it indicates that there will be some fundamental limits on the efficiency of quantum error correction schemes, which are determined solely by how noisy an error channel is to begin with. The framework of quantum communication theory provides a way to examine such fundamental aspects of quantum error correction. Specifically, the notion of quantum capacity is particularly useful for understanding the fundamental performance limits of quantum error correction. 

The main goal of this chapter is to investigate the fundamental aspects of bosonic quantum error correction via quantum communication-theoretic tools. I will focus on Gaussian thermal-loss channels because they are good models of realistic quantum communication channels. In particular, I will study their quantum capacity. 

In Section \ref{section:Quantum capacity and quantum error correction}, I will review the close relation between quantum error correction and the notion of quantum capacity. In Section \ref{section:Quantum capacity of bosonic pure-loss channels}, I will review previous results on the quantum capacity of bosonic pure-loss channels, a subclass of Gaussian thermal-loss channels. In Section \ref{section:Upper bound of the Gaussian thermal-loss channel capacity}, I will present an improved upper bound of the Gaussian thermal-loss channel capacity using a data-processing argument \cite{Noh2019}. In Section \ref{section:Lower bound of the Gaussian thermal-loss channel capacity}, I will provide the tightest lower bound of the Gaussian thermal-loss channel capacity and show that higher quantum communication rates can be achieved than previously believed \cite{Noh2020a}. I will conclude the chapter by outlining related open questions in Section \ref{section:Open questions quantum capacity}.

\section{Quantum capacity and quantum error correction}
\label{section:Quantum capacity and quantum error correction}

Here, we will review several known quantum communication-theoretic results on the notion of quantum capacity and its relation to quantum error correction. For a comprehensive and pedagogic introduction to quantum communication theory, see Refs.\ \cite{Holevo2012,Wilde2013,Hayashi2016,Watrous2018}.    

\subsection{Achievable quantum state transmission rate}

Let us consider a general noisy quantum communication channel 
\begin{align}
\mathcal{N}_{A\rightarrow B} : \mathcal{L}( \mathcal{H}_{A} ) \rightarrow \mathcal{L}( \mathcal{H}_{B} ) ,  
\end{align}
that transmits a quantum state of an information sender A (or Alice) to a receiver B (or Bob). $\mathcal{H}_{A}$ and $\mathcal{H}_{B}$ are the Hilbert spaces which Alice and Bob have access to, respectively. We assume that Alice and Bob are far away from each other and therefore it is not possible to directly implement an entangling quantum operation on the joint system of Alice and Bob $\mathcal{H}_{A}\otimes \mathcal{H}_{B}$. Instead, we assume that these two systems are connected only via the noisy channel $\mathcal{N}_{A\rightarrow B}$. On the other hand, we assume that Alice and Bob can implement any local quantum operations on their own Hilbert spaces that they have access to. 

Note that a quantum state that Alice sends will be corrupted by the noisy channel $\mathcal{N}$ and thus Bob will receive a noisy state that does not faithfully carry the quantum information Alice intended to transmit. Therefore, it is essential to use quantum error correction if Alice and Bob want to achieve reliable quantum communication despite the channel noise. To implement quantum error correction, Alice has to encode her quantum state through an error-correcting code, or an encoding map   
\begin{align}
\mathcal{E}_{A_{0} \rightarrow A^{N} } : \mathcal{L}( \mathcal{H}_{A_{0}} ) \rightarrow \mathcal{L}(  \mathcal{H}_{A}^{\otimes N} ) .  
\end{align}
Here, $\mathcal{H}_{A}$ is the Hilbert space associated with Alice's local quantum memory which hosts the quantum states that Alice wants to faithfully transmit to Bob. Note that we are considering a general encoding scheme where Alice encodes her quantum information collectively to $N$ channel inputs $ \mathcal{H}_{A}^{\otimes N} \equiv \mathcal{H}_{A_{1}}\otimes \cdots \otimes \mathcal{H}_{A_{N}}$ that she has access to. $A^{N}$ is an abbreviation for $A_{1}\cdots A_{N}$. Thus, if Alice wants to send a state $\hat{\rho} \in \mathcal{D}(\mathcal{H}_{A})$, she inputs an encoded  
\begin{align}
\mathcal{E}_{A_{0} \rightarrow A^{N} }( \hat{\rho} )
\end{align}
to her channel inputs. Here, $\mathcal{D}(\mathcal{H}) \equiv \lbrace \hat{\rho} \in \mathcal{L}(\mathcal{H}) | \hat{\rho}=\hat{\rho}^{\dagger} \succeq 0, \mathrm{Tr}[\hat{\rho}]  = 1 \rbrace $ is the space of density matrices associated with the Hilbert space $\mathcal{H}$.  

Then, the encoded state $\mathcal{E}_{A_{0} \rightarrow A^{N} }( \hat{\rho} )$ is sent to Bob through $N$ noisy channels 
\begin{align}
( \mathcal{N}^{\otimes N} )_{A^{N} \rightarrow B^{N} }  :  \mathcal{L}(  \mathcal{H}_{A}^{\otimes N}  ) \rightarrow \mathcal{L}(  \mathcal{H}_{B}^{\otimes N} ) . 
\end{align}
Thus, Bob receives a noisy state 
\begin{align}
( \mathcal{N}^{\otimes N} )_{A^{N} \rightarrow B^{N} } \cdot \mathcal{E}_{A_{0} \rightarrow A^{N} }( \hat{\rho} ). 
\end{align}
To recover the quantum information that Alice intended to send, Bob has to perform a recovery map 
\begin{align}
\mathcal{R}_{B^{N} \rightarrow B_{0} } : \mathcal{L}(  \mathcal{H}_{B}^{\otimes N} ) \rightarrow \mathcal{L}( \mathcal{H}_{B_{0}} )
\end{align} 
when he registers the received state in the $N$ channel outputs $\mathcal{H}_{B}^{\otimes N}$ to his local quantum memory $\mathcal{H}_{B_{0}}$. Then, Bob is left with a state 
\begin{align}
\hat{\sigma}  \equiv \mathcal{R}_{B^{N} \rightarrow B_{0} } \cdot ( \mathcal{N}^{\otimes N} )_{A^{N} \rightarrow B^{N} } \cdot \mathcal{E}_{A_{0} \rightarrow A^{N} }( \hat{\rho} ) . 
\end{align}
Thus, through the encoding and the recovery maps $\mathcal{E}$ and $\mathcal{R}$, Alice and Bob perform quantum error correction. 

Ideally, the final state $\hat{\sigma}$ should be identical to the input state $\hat{\rho}$ for any $\hat{\rho} \in \mathcal{D}( \mathcal{H}_{A} )$. Suppose that this is indeed the case and assume $\textrm{dim}( \mathcal{H}_{A} ) = \textrm{dim}( \mathcal{H}_{B} ) = d$. Then, we say that Alice and Bob achieved a quantum state transmission rate 
\begin{align}
R = \frac{1}{N}\log_{2} d . 
\end{align}
This rate quantifies the number of qubits per channel use that are reliably transmitted from Alice to Bob. In general, however, the final state $\hat{\sigma}$ is not exactly the same as the input state $\hat{\rho}$, and there is always some small residual error $\epsilon$ for any finite number of channel uses $N < \infty$. If the encoding and recovery maps are well designed, this residual error is suppressed to an arbitrarily small value by increasing the number of channel uses $N$. Incorporating this general case, achievable quantum state transmission rate of a pair of an encoding map and a recovery map $(\mathcal{E},\mathcal{R})$ is defined as follows: 

\begin{definition}[Achievable quantum state transmission rate]
Consider a noisy channel $\mathcal{N}_{A \rightarrow B} : \mathcal{L}( \mathcal{H}_{A} ) \rightarrow \mathcal{L}( \mathcal{H}_{B} )$ from Alice to Bob, and a pair of an encoding map and a recovery map $(\mathcal{E},\mathcal{R})$ where 
\begin{align}
&\mathcal{E}_{A_{0} \rightarrow A^{N} } : \mathcal{L}( \mathcal{H}_{A_{0}} ) \rightarrow \mathcal{L}(  \mathcal{H}_{A}^{\otimes N} ), 
\nonumber\\
&\mathcal{R}_{B^{N} \rightarrow B_{0} } : \mathcal{L}(  \mathcal{H}_{B}^{\otimes N} ) \rightarrow \mathcal{L}( \mathcal{H}_{B_{0}} ) . 
\end{align}   
Let $d$ be the dimension of the Hilbert spaces $\mathcal{H}_{A_{0}}$ and $\mathcal{H}_{A_{0}}$, i.e., $d \equiv \textrm{dim}( \mathcal{H}_{A} ) = \textrm{dim}( \mathcal{H}_{B} )$. Then, we say that the pair of the encoding and the recovery operations $(\mathcal{E},\mathcal{R})$ is an $(N, R-\delta , \epsilon)$ quantum communication code if the following conditions are satisfied:   
\begin{align}
&\frac{1}{N}\log_{2}d  = R-\delta, 
\nonumber\\
&|\!| \hat{\sigma} - \hat{\rho}  |\!|_{1} \le \epsilon \,\,\, \textrm{for all}\,\,\, \hat{\rho} \in \mathcal{H}_{A_{0}}, 
\end{align}
where $\hat{\sigma} $ is defined as
\begin{align}
\hat{\sigma} \equiv \mathcal{R}_{B^{N} \rightarrow B_{0} } \cdot ( \mathcal{N}^{\otimes N} )_{A^{N} \rightarrow B^{N} } \cdot \mathcal{E}_{A_{0} \rightarrow A^{N} }( \hat{\rho} ) . 
\end{align}
Then we say the pair of the encoding and the recovery maps $(\mathcal{E},\mathcal{R})$ achieves a quantum state transmission rate $R$ against a channel $\mathcal{N}$, if $\delta$ and $\epsilon$ can be made arbitrarily small as we increase the number of channel uses indefinitely, i.e., as $N\rightarrow \infty$. This way, the rate $R$ quantifies the number of qubits per channel use that are reliably sent from Alice to Bob.   
\end{definition}

\subsection{Quantum capacity}

Note that the achievable quantum state transmission rate $R$ depends on the encoding and the recovery maps $\mathcal{E}$ and $\mathcal{R}$, and the noisy channel $\mathcal{N}$, i.e., 
\begin{align}
R = R( \mathcal{E}, \mathcal{R} ; \mathcal{N}) . 
\end{align}
Given a noisy channel $\mathcal{N}$, we will achieve a higher rate if we choose a better pair of encoding and the recovery maps, or a better quantum error correction scheme. Thus, the achievable rate $R( \mathcal{E}, \mathcal{R} ; \mathcal{N})$ characterizes the efficiency of a quantum error correction scheme $(\mathcal{E}, \mathcal{R})$ against a noisy channel $\mathcal{N}$. Then, we define the quantum capacity of a channel $\mathcal{N}$ as follows:  

\begin{definition}[Quantum capacity of a quantum channel]
Let $\mathcal{N}$ be a quantum channel. The quantum capacity $Q(\mathcal{N})$ of a quantum channel $\mathcal{N}$ is defined as 
\begin{align}
C_{Q}(\mathcal{N}) \equiv \max_{ \mathcal{E}, \mathcal{R} } R( \mathcal{E}, \mathcal{R} ; \mathcal{N}), 
\end{align}
the maximum achievable quantum state transmission rate by using an optimal choice of an encoding map and a recovery map. Note that the quantum capacity $Q(\mathcal{N})$ depends only on the channel $\mathcal{N}$. Thus, the quantum capacity of a quantum channel characterizes the channel's intrinsic information-transmission capability.     
\end{definition}

\subsection{Coherent information and quantum capacity} 

While the operational meaning of the quantum capacity is clear, it is not yet clear how we can evaluate this quantity. We can in principle compute the quantum capacity by comprehensively optimizing the achievable rate $R$ over all possible quantum error correction schemes $(\mathcal{E}, \mathcal{R})$. However, such an optimization is not generally feasible in practice. Luckily, however, it is possible to quantify the quantum capacity of a quantum channel without having to optimize over all possible encoding and recovery maps. More specifically, it is possible to characterize the quantum capacity based solely on an entropic quantity, i.e., regularized coherent information \cite{Schumacher1996,Lloyd1997,Devetak2005}, which does not depend on the encoding and recovery maps. It is one of the biggest achievements of quantum communication theory.  

Generalizing the Shannon entropy in the classical information theory \cite{Shannon1948}, entropy of a quantum state is defined as
\begin{align}
S( \hat{\rho} ) = -\mathrm{Tr} [ \hat{\rho}\log_{2}\hat{\rho} ] . 
\end{align}
$S( \hat{\rho} )$ is also called the von Neumann entropy of a quantum state $\hat{\rho}$. Then, we define a complementary channel $\mathcal{N}^{c}$ of a channel $\mathcal{N}$ as follows: 

\begin{definition}[Complementary channel]
Let $\mathcal{N} = \mathcal{N}_{A\rightarrow B}$ be a quantum channel. The channel $\mathcal{N}_{A\rightarrow B}$ can be dilated as 
\begin{align}
\mathcal{N}_{A\rightarrow B}(\hat{\rho}_{A}) = \mathrm{Tr}_{E}[ \hat{U}_{AE\rightarrow BE} ( \hat{\rho}_{A} \otimes |0\rangle\langle 0|_{E} ) \hat{U}_{AE\rightarrow BE}^{\dagger}  ], 
\end{align} 
where $E$ is an environment, $|0\rangle\langle 0|_{E}$ is a pure state, and $\hat{U}_{AE\rightarrow BE}$ is a unitary operation. Then, a complementary channel $\mathcal{N}^{c} = \mathcal{N}^{c}_{A\rightarrow E}$ of a channel $\mathcal{N}$ is defined as 
\begin{align}
\mathcal{N}^{c}_{A\rightarrow E}(\hat{\rho}_{A}) \equiv  \mathrm{Tr}_{B}[ \hat{U}_{AE\rightarrow BE} ( \hat{\rho}_{A} \otimes |0\rangle\langle 0|_{E} ) \hat{U}_{AE\rightarrow BE}^{\dagger}  ]. 
\end{align} 
\end{definition}

Note that there may be multiple unitary operations $\hat{U}_{AE\rightarrow BE}$ that give rise to the same channel $\mathcal{N}_{A\rightarrow B}$ and thus complementary channels are not unique. However, any two complementary channels of $\mathcal{N}_{A\rightarrow B}$ are equivalent to each other up to a unitary operation on the environment $E$. Therefore, the von Nuemann entropy of an output state of any complementary channel of $\mathcal{N}$ is the same because the von Nuemann entropy is invariant under a unitary operation. This means that
\begin{align}
S( \mathcal{N}^{c}_{A\rightarrow E}(\hat{\rho}_{A}) ) 
\end{align}
is well-defined. Then, we define coherent information as follows: 

\begin{definition}[Coherent information of a quantum channel]
The coherent information of a channel $\mathcal{N}$ with respect to an input state $\hat{\rho}$ is defined as
\begin{align}
I_{c}( \mathcal{N} , \hat{\rho} ) &\equiv S( \mathcal{N}(\hat{\rho}) )  - S( \mathcal{N}^{c}(\hat{\rho}) ) . 
\end{align}
Then, the one-shot coherent information of a channel $\mathcal{N}$ is defined as
\begin{align}
Q( \mathcal{N} ) &\equiv \max_{\hat{\rho} } I_{c}( \mathcal{N} , \hat{\rho} ), 
\end{align}
i.e., maximization of the coherent information $I_{c}( \mathcal{N} , \hat{\rho} ) $ over all input states $\hat{\rho}$. Lastly, the regularized coherent information of a channel $\mathcal{N}$ is defined as
\begin{align}
Q_{\reg}( \mathcal{N} ) &\equiv \lim_{N\rightarrow \infty} \frac{1}{N} Q( \mathcal{N}^{\otimes N} ) = \lim_{N\rightarrow \infty} \frac{1}{N} \max_{\hat{\rho} } I_{c}( \mathcal{N}^{\otimes N} , \hat{\rho} ) . 
\end{align}
Here, the state $\hat{\rho}$ should be optimized over all possible input states to the $N$ channels $\mathcal{N}^{\otimes N}$. 
\end{definition}

Remarkably, Refs.\ \cite{Schumacher1996,Lloyd1997,Devetak2005} established that the quantum capacity of a quantum channel equals the channel's regularized coherent information. 

\begin{theorem}[Quantum capacity equals regularized coherent information \cite{Schumacher1996,Lloyd1997,Devetak2005}]
The quantum capacity $C_{Q}(\mathcal{N})$ of a quantum channel $\mathcal{N}$ equals the regularized coherent information $Q_{\reg}( \mathcal{N} )$ of the channel: 
\begin{align}
C_{Q}(\mathcal{N}) &= Q_{\reg}( \mathcal{N} ) . 
\end{align}
\end{theorem}
Note that the regularized coherent information is a purely entropic quantity that does not depend on any encoding and recovery operations. Nevertheless, it does provide a fundamental limit on the ultimate efficiency of quantum error correction schemes. On the other hand, it is also important to realize that the evaluation of the regularized coherent information involves optimization over all input state states $\hat{\rho}$ to the $N$ channels. In particular, we should take the limit of infinitely many channel uses, i.e., $N\rightarrow \infty$. Therefore, evaluation of the quantum capacity is still intractable in the most general case \cite{Wolf2011,Oskouei2018}. However, if a channel satisfies a certain special property, evaluation of the channel's quantum capacity can be made tractable. For example, it is possible to efficiently compute the quantum capacity of a quantum channel if the channel is degradable of anti-degradable \cite{Devetak2005C,Yard2008,Caruso2006}.  

\begin{definition}[Degradability or anti-degradability of a quantum channel]
A quantum channel $\mathcal{N} = \mathcal{N}_{A\rightarrow B}$ is called degradable if there is a degrading channel $D_{B\rightarrow E}$ such that 
\begin{align}
\mathcal{N}^{c}_{A\rightarrow E} &= D_{B\rightarrow E} \cdot \mathcal{N}_{A\rightarrow B}, 
\end{align}
i.e., if Bob can simulate the environment's complementary channel output. Conversely, a channel is called anti-degradable if there is a degrading channel $\mathcal{D}_{E \rightarrow B }$ such that 
\begin{align}
\mathcal{N}_{A\rightarrow B} &= D_{E\rightarrow B} \cdot \mathcal{N}^{c}_{A\rightarrow E}, 
\end{align}
i.e., if the environment can simulate Bob's channel output. 
\end{definition} 

Then for degradable channels, we have the following desirable properties: 

\begin{theorem}[Additivity of the coherent information of degradable channels \cite{Devetak2005C}]
The coherent information of a degradable channel is additive, i.e.,
\begin{align}
C_{Q}(\mathcal{N}) &= Q_{\reg}( \mathcal{N} ) = Q(\mathcal{N}), 
\end{align}
and thus the regularization (i.e., $N\rightarrow \infty$) is not needed to evaluate the channel's quantum capacity.  \label{theorem:Additivity of the coherent information of degradable channels}
\end{theorem} 
 
\begin{theorem}[Concavity of the coherent information of degradable channels \cite{Yard2008}]
The coherent information  of a degradable channel $\mathcal{N}$ is concave in the input states $\hat{\rho}$. That is, for any degradable channel $\mathcal{N}$, we have
\begin{align}
I_{c}\Big{(} \mathcal{N} , \sum_{x}p_{x}\hat{\rho}_{x} \Big{)} \ge \sum_{x}p_{x} I_{c}(\mathcal{N} , \hat{\rho}_{x}), 
\end{align}
for any $p_{x}$ such that $p_{x}\ge 0$ and $\sum_{x}p_{x}=1$. \label{theorem:Concavity of the coherent information of degradable channels}
\end{theorem}

Theorem \ref{theorem:Additivity of the coherent information of degradable channels} shows that evaluation of the one-shot coherent information suffices for computing the quantum capacity of degradable channels. To evaluate the one-shot coherent information, we still need to perform the maximization of coherent information over all input states $\hat{\rho}$ (to a single channel). Theorem \ref{theorem:Additivity of the coherent information of degradable channels} shows that this a convex minimization problem which can be solved efficiently \cite{Boyd2004}. This is because the set of density matrices is convex and the objective function $I_{c}(\mathcal{N},\hat{\rho})$ (to be maximized) is concave in the input state $\hat{\rho}$. Note that concave maximization is equivalent to convex minimization. Thus, the quantum capacity of degradable channels can be efficiently evaluated. In the case of anti-degradable channels, the situation is even simpler: 

\begin{theorem}[Quantum capacity of anti-degradable channels \cite{Caruso2006}]
The quantum capacity of an anti-degradable channel vanishes, i.e., 
\begin{align}
C_{Q}(\mathcal{N}) =0 ,
\end{align} 
for any anti-degradable channel $\mathcal{N}$. \label{theorem:Quantum capacity of anti-degradable channels}
\end{theorem}    

\subsection{Superadditivity of coherent information}

It is possible that a quantum channel is neither degradable nor anti-degradable. In this case, the one-shot coherent information does not necessarily equal the quantum capacity. Indeed, it has been demonstrated that the coherent information may be superadditive for many non-degradable channels and thus the regularization is essential \cite{DiVincenzo1998,Smith2007,Smith2008,Fern2008,Smith2011,Cubitt2015,
Lim2018,Leditzky2018,Lim2019,Bausch2020}. For these channels, only lower and upper bounds of the quantum capacity are known. We remark that such a superadditivity is a unique feature of quantum communication theory, not present in its classical counterpart. In Section \ref{section:Lower bound of the Gaussian thermal-loss channel capacity}, we will demonstrate that similar superadditive behavior is also observed for the practically relevant Gaussian thermal-loss channels. In particular, we will show that the coherent information of Gaussian thermal-loss channels is superadditive with respect to Gaussian input states \cite{Noh2020a}.

\subsection{Energy-constrained quantum capacity}

In the rest of this chapter, we will focus on Gaussian thermal-loss channels which act on bosonic Hilbert spaces. As previously discussed in the context of benchmarking various bosonic codes (see Section \ref{section:Benchmarking single-mode bosonic codes}), it is important to control the allowed average energy because bosonic codes with a larger energy generally perform better than the ones with a smaller energy. This is also true in the context of quantum communication. As we allow bosonic channels to support more energy (or excitations), the amount of information that the channels can carry increase. In practice, however, realistic quantum communication channels are only able to support states that have an average energy that is smaller than a certain critical value. Thus, if we evaluate the quantum capacity of these channels without imposing an energy constraint, we will overestimate the channel's information-transmission capability. Therefore, to understand the limitations coming from an energy constraint, it is important to generalize the notion of quantum capacity to energy-constrained scenarios.

In essence, the energy-constrained quantum capacity of a quantum channel equals the channel's energy-constrained regularized coherent information. That is, one can compute the energy-constrained quantum capacity by replacing the optimization over all input states in the regularized coherent information with an optimization over all input states that satisfy a desired energy constraint. For more details see Ref.\ \cite{Wilde2018}.

\section{Quantum capacity of bosonic pure-loss channels}
\label{section:Quantum capacity of bosonic pure-loss channels}

With all the necessary facts ready, let us now consider the quantum capacity of a bosonic pure-loss channel $\mathcal{N}[\eta,0]$ (see Definition \ref{definition:Bosonic pure loss channels} for the definition of bosonic pure-loss channels and Table \ref{table:excitation loss errors} for their various other representations). Note that we used bosonic pure-loss channels for benchmarking and optimizing single-mode bosonic codes in Chapter \ref{chapter:Benchmarking and optimizing single-mode bosonic codes} due to their experimental relevance. Note also that bosonic pure-loss channels are a subclass of Gaussian thermal-loss channels $\mathcal{N}[\eta,\nth]$ (defined in Definition \ref{definition:Gaussian loss channels}) with $\nth =0$. The quantum capacity of a general Gaussian thermal-loss channel with $\nth\neq 0$ will be considered in Sections \ref{section:Upper bound of the Gaussian thermal-loss channel capacity} and \ref{section:Lower bound of the Gaussian thermal-loss channel capacity}. Here, we will review the known results on the bosonic pure-loss channel capacity.     

First, we show that a bosonic pure-loss channel is degradable or anti-degradable. 

\begin{lemma}[Degradability or anti-degradability of bosonic pure-loss channels \cite{Caruso2006}]
A bosonic pure-loss channel $\mathcal{N}[\eta,0]$ with a transmissivity $\eta \in (\frac{1}{2},1]$ (or $\eta \in [0, \frac{1}{2}]$) is degardable (or anti-degradable). \label{lemma:Degradability or anti-degradability of bosonic pure-loss channels} 
\end{lemma}
\begin{proof}
A complementary channel $\mathcal{N}^{c}[\eta,0]$ of the bosonic pure-loss channel $\mathcal{N}[\eta,0]$ is given by 
\begin{align}
\mathcal{N}^{c}[\eta,0] &= \mathcal{N}[1-\eta,0] . 
\end{align}
Thus for $\eta \in (\frac{1}{2},1]$, one can degrade the channel by a degrading map $\mathcal{D} = \mathcal{N}[\frac{1-\eta}{\eta},0]$ to get the complementary channel, i.e.,
\begin{align}
\mathcal{N}^{c}[\eta,0] &= \mathcal{N}[1-\eta,0]  = \mathcal{N}\Big{[} \frac{1-\eta}{\eta}, 0 \Big{]} \cdot \mathcal{N}[\eta,0],  
\end{align}
and thus the bosonic pure-loss channel $\mathcal{N}[\eta,0]$ is degradable if $\eta \in (\frac{1}{2},1]$. On the other hand if $\eta \in [0,\frac{1}{2}]$, one can degrade the complementary channel by a degrading map $\mathcal{N}[ \frac{\eta}{1-\eta},0 ]$ to get the channel, i.e., 
\begin{align}
\mathcal{N}[\eta,0] &= \mathcal{N}\Big{[}\frac{\eta}{1-\eta},0\Big{]} \cdot \mathcal{N}[1-\eta,0] , 
\end{align}
and thus the bosonic pure-loss channel $\mathcal{N}[\eta,0]$ is anti-degradable if $\eta \in [0,\frac{1}{2}]$. 
\end{proof}

For $\eta\in [0,\frac{1}{2}]$, the quantum capacity of the bosonic pure-loss channel $\mathcal{N}[\eta,0]$ vanishes because the channel is anti-degradable (see Theorem \ref{theorem:Quantum capacity of anti-degradable channels}). Thus for any loss probability $\gamma = 1-\eta \ge \frac{1}{2}$, the bosonic pure-loss channel $\mathcal{N}[\eta = 1-\gamma,0]$ does not have any information-transmission capability. This is related to the fact that the single-mode bosonic codes stop being useful as the loss probability approaches $50\%$ (see Fig.\ \ref{fig:benchmarking single-mode bosonic codes}). 

For $\eta \in (\frac{1}{2},1]$, the bosonic pure-loss channel is degradable so the one-shot coherent information of the channel equals the channel's quantum capacity (see Theorem \ref{theorem:Additivity of the coherent information of degradable channels}). However, we still need to maximize the coherent information $I_{c}(\mathcal{N}[\eta,0],\hat{\rho})$ by optimizing the input states $\hat{\rho}$.    

\begin{lemma}[Optimality of Gaussian input states \cite{Wolf2007}]
The coherent information of a bosonic pure-loss channel $\mathcal{N}[\eta,0]$ is maximized by a Gaussian state. \label{lemma:Optimality of Gaussian input states}
\end{lemma}

Furthermore, by using the concavity of the coherent information for degradable channels, we can show that a diagonal state in the Fock basis maximizes the coherent information: 

\begin{lemma}[Optimality of diagonal states in the Fock basis \cite{Noh2019}]
Let $\mathcal{N}[\eta,0]$ be a bosonic pure-loss channel and $\hat{\rho}=\sum_{m,n=0}^{\infty}\rho_{mn}|m\rangle\langle n|$ be an arbitrary bosonic state represented in the Fock basis. Then, we have
\begin{equation}
I_{c}(\mathcal{N}[\eta,0],\hat{\rho}) \le I_{c}\Big{(} \mathcal{N}[\eta,0], \sum_{n=0}^{\infty}\rho_{nn}|n\rangle\langle n|\Big{)}.  
\end{equation}
Thus, the coherent information of a bosonic pure-loss channel is maximized by a diagonal state in the Fock basis. \label{lemma:Optimality of diagonal states in the Fock basis}
\end{lemma}
\begin{proof}
Define $\hat{\rho}_{\theta} \equiv \mathcal{U}[\theta]\hat{\rho} = e^{i\theta\hat{n}}\hat{\rho}e^{-i\theta\hat{n}}$ and let $p(\theta)$ be a probability density function defined over $\theta\in [0,2\pi)$. Since a bosonic pure-loss channel is degradable, its coherent information $I_{c}(\mathcal{N}[\eta,0],\hat{\rho})$ is concave in the input state (see Theorem \ref{theorem:Concavity of the coherent information of degradable channels}): 
\begin{equation}
\int_{0}^{2\pi} d\theta p(\theta)I_{c}(\mathcal{N}[\eta,0],\hat{\rho}_{\theta}) \le I_{c}\Big{(}\mathcal{N}[\eta,0] , \int_{0}^{2\pi} d\theta p(\theta)\hat{\rho}_{\theta}\Big{)}.  \label{eq:thermal sufficiency proof intermediate}
\end{equation}
The rotational invariance of bosonic pure-loss channels implies $\mathcal{N}[\eta,0](\hat{\rho}_{\theta}) = \mathcal{U}[\theta]\cdot \mathcal{N}[\eta,0](\hat{\rho})$ and similarly $\mathcal{N}^{c}[\eta,0](\hat{\rho}_{\theta}) = \mathcal{U}[\theta]\cdot \mathcal{N}^{c}[\eta,0](\hat{\rho})$. Since quantum entropy is invariant under a unitary transformation (i.e., $H(\hat{\rho}) = H(\hat{U}\hat{\rho}\hat{U}^{\dagger})$ for a unitary $\hat{U}$), we have $I_{c}(\mathcal{N}[\eta,0],\hat{\rho}_{\theta}) = I_{c}(\mathcal{N}[\eta,0],\hat{\rho})$. The left hand side of Eq.\ \eqref{eq:thermal sufficiency proof intermediate} is then given by $I_{c}(\mathcal{N}[\eta,0],\hat{\rho})$ since $\int_{0}^{2\pi}p(\theta)=1$. Choosing $p(\theta)$ to be a flat distribution $p(\theta) = 1/(2\pi)$, we find 
\begin{align}
\int_{0}^{2\pi} d\theta p(\theta) \hat{\rho}_{\theta} &= \sum_{m,n=0}^{\infty}\frac{1}{2\pi} \int_{0}^{2\pi}  d\theta e^{i\theta(m-n)} \rho_{mn}|m\rangle\langle n| = \sum_{n=0}^{\infty} \rho_{nn}|n\rangle\langle n|, \label{eq:rotation twirling} 
\end{align}
where we used $\int_{0}^{2\pi}  d\theta e^{i\theta(m-n)}=2\pi \delta_{mn}$ to derive the last equality. Plugging Eq.\ \eqref{eq:rotation twirling} into the right hand side of Eq.\ \eqref{eq:thermal sufficiency proof intermediate}, the lemma follows. 
\end{proof}

Putting all these facts together, we can finally determine the energy-constrained quantum capacity of bosonic pure-loss channels. 

\begin{theorem}[Energy-constrained quantum capacity of bosonic pure-loss channels] 
The quantum capacity of a bosonic pure-loss channel $\mathcal{N}[\eta,0]$ subject to an energy constraint $\mathrm{Tr}[\hat{\rho}\hat{n}] \le \bar{n}$ is given by
\begin{align}
C_{Q}^{n \le \bar{n}}( \mathcal{N}[\eta,0] ) = \max\Big{[}  g(\eta \bar{n}) - g((1-\eta)\bar{n}) , 0 \Big{]},  
\end{align}
where $g(x)$ is the von Neumann entropy of a thermal state $\hat{\tau}(x) \equiv \sum_{n=0}^{\infty} \frac{x^{n}}{(x+1)^{n+1}}|n\rangle\langle n|$:
\begin{align}
g(x) \equiv S(\hat{\tau}(x)) =  (x+1)\log_{2}(x+1) - x\log_{2}x. 
\end{align}
In the infinite-energy limit (i.e., $\bar{n}\rightarrow \infty$), we have 
\begin{align}
C_{Q}( \mathcal{N}[\eta,0] )  = \lim_{\bar{n}\rightarrow \infty} C_{Q}^{ n \le \bar{n}}( \mathcal{N}[\eta,0] ) = \max\Big{[}  \log_{2}\Big{(} \frac{\eta}{1-\eta} \Big{)} , 0  \Big{]} . 
\end{align} \label{theorem:Energy-constrained quantum capacity of bosonic pure-loss channels}
\end{theorem} 

\begin{proof}
For $\eta \in [0,\frac{1}{2}]$, the quantum capacity of the bosonic pure-loss channel $\mathcal{N}[\eta,0]$ vanishes. For $\eta \in (\frac{1}{2},1]$, combining the results in Lemmas \ref{lemma:Optimality of Gaussian input states} and \ref{lemma:Optimality of diagonal states in the Fock basis}, we find the optimal input state that maximizes the coherent information of the bosonic pure-loss channel should be a thermal state, i.e., a Gaussian state that is diagonal in the Fock basis. Thus, we have 
\begin{align}
C_{Q}^{n \le \bar{n}}( \mathcal{N}[\eta,0] ) &= \max_{0\le x\le\bar{n}} I_{c}( \mathcal{N}[\eta,0], \hat{\tau}(x) ) 
\nonumber\\
&= \max_{0\le x\le\bar{n}} ( g(\eta x) - g((1-\eta)x) ) = g(\eta \bar{n}) - g((1-\eta)\bar{n}) .  
\end{align}  
See, e.g., Refs.\ \cite{Bombelli1986,Holevo2001} for the second equality. The last equality follows from the fact that $g(\eta x) - g((1-\eta)x)$ monotonically increases in $x$ for any $\eta\in(\frac{1}{2},1]$. In the energy-unconstrained case, by using $g(x) = \log_{2}(ex) + \mathcal{O}(\frac{1}{x})$, we find
\begin{align}
C_{Q}( \mathcal{N}[\eta,0] ) &= \lim_{\bar{n}\rightarrow \infty} C_{Q}^{n \le \bar{n}}( \mathcal{N}[\eta,0] ) 
\nonumber\\
&= \lim_{\bar{n}\rightarrow \infty}  \Big{[} g(\eta \bar{n}) - g((1-\eta)\bar{n}) \Big{]}
\nonumber\\
&= \lim_{\bar{n}\rightarrow \infty}  \Big{[}  \log_{2}\Big{(} \frac{\eta}{1-\eta} \Big{)} + \mathcal{O}\Big{(}\frac{1}{\bar{n}}\Big{)} \Big{]} = \log_{2}\Big{(} \frac{\eta}{1-\eta} \Big{)}. 
\end{align} 
Combining these results with the fact that the quantum capacity vanishes when $\eta\in[0,\frac{1}{2}]$, the theorem follows. 
\end{proof}

\section{Upper bounds of the Gaussian thermal-loss channel capacity}
\label{section:Upper bound of the Gaussian thermal-loss channel capacity}

As shown above, the quantum capacity of a bosonic pure-loss channel is analytically determined thanks to the channel's degradability or anti-degradability. However, a Gaussian thermal-loss channel $\mathcal{N}[\eta,\nth]$ with $\nth \neq 0$ is neither degradable nor anti-degradable. Thus, its coherent information is not necessarily additive and its quantum capacity has not been analytically determined. Various upper bounds of the Gaussian thermal-loss channel capacity have been established \cite{Holevo2001,Pirandola2017,Sharma2018,Rosati2018,Noh2019}. Here, we present the upper bounds obtained by using variations of data-processing arguments \cite{Sharma2018,Rosati2018,Noh2019}.  

A key step towards establishing the data-processing upper bounds is to decompose a Gaussian thermal-loss channel $\mathcal{N}[\eta,\nth]$ in terms of a bosonic pure-loss channel and a quantum-limited amplification channel. This decomposition is useful as it allows us to relate the quantum capacity of a Gaussian thermal-loss channel (which we want to evaluate) with the quantum capacity of a bosonic pure-loss channel (which we understand already).    

\begin{lemma}[Thermal-loss  = Amplification + Pure-loss \cite{Sharma2018,Noh2019}]
A Gaussian thermal-loss channel $\mathcal{N}[\eta,\nth]$ can be decomposed into a bosonic pure-loss channel and a quantum-limited amplification channel as follows: 
\begin{align}
\mathcal{N}[\eta,\nth] = \mathcal{A}[ G',0 ] \cdot \mathcal{N}[\eta', 0 ] ,  
\end{align}
where $G'$ and $\eta'$ are given by 
\begin{align}
G' &= (1-\eta)\nth + 1 \,\,\, \textrm{and} \,\,\,  \eta' = \frac{\eta}{G'} = \frac{\eta}{(1-\eta)\nth + 1}. 
\end{align} \label{lemma:Thermal-loss  = Amplification + Pure-loss}
\end{lemma}

Then, a data-processing argument lets us to upper bound the quantum capacity of the Gaussian thermal-loss channel $\mathcal{N}[\eta,\nth]$ by the quantum capacity of the bosonic pure-loss channel $\mathcal{N}[\eta', 0 ] $. 

\begin{theorem}[Data-processing bound of the Gaussian thermal-loss channel capacity \cite{Sharma2018}]
The quantum capacity of a Gaussian thermal-loss channel $\mathcal{N}[\eta,\nth]$ subject to an energy constraint $\mathrm{Tr}[\hat{\rho} \hat{n}] \le \bar{n}$ is upper bounded by the following data-processing bound $Q^{n\le \bar{n}}_{\scriptsize{\textrm{DP}}}(\eta,\nth)$:
\begin{align}
C_{Q}^{n \le \bar{n}}(\mathcal{N}[\eta,\nth]) &\le Q^{n \le \bar{n}}_{\scriptsize{\textrm{DP}}}(\eta,\nth) \equiv C_{Q}^{ n \le \bar{n}}(\mathcal{N}[\eta',0]),   \label{eq:data processing upper bound energy constrained}
\end{align}
where $\eta' = \frac{\eta}{(1-\eta)\nthtiny+1}$. More explicitly, $Q^{n\le \bar{n}}_{\scriptsize{\textrm{DP}}}(\eta,\nth)$ is given by
\begin{align}
Q^{n\le \bar{n}}_{\scriptsize{\textrm{DP}}}(\eta,\nth) &= \max\Big{[} g\Big{(} \frac{\eta \bar{n}}{(1-\eta)\nth+1} \Big{)} - g\Big{(}  \frac{ (1-\eta)(\nth+1)\bar{n} }{(1-\eta)\nth+1} \Big{)} , 0  \Big{]} . 
\end{align}
In the energy-unconstrained case (i.e., $\bar{n}\rightarrow\infty$), we have
\begin{align}
C_{Q}(\mathcal{N}[\eta,\nth]) &\le \lim_{\bar{n}\rightarrow \infty}  Q^{n\le \bar{n}}_{\scriptsize{\textrm{DP}}}(\eta,\nth) = \max\Big{[} \log_{2}\Big{(}  \frac{\eta }{ (1-\eta)(\nth+1) } \Big{)} , 0 \Big{]} .  \label{eq:data processing upper bound energy unconstrained}
\end{align} \label{theorem:data processing bound}
\end{theorem}
\begin{proof}
Since the quantum capacity $C_{Q}(\mathcal{N}[\eta,\nth])$ is the maximum achievable quantum state transmission rate, there exists a set of encoding and recovery channels, denoted by $\lbrace \mathcal{E},\mathcal{R}\rbrace$, which achieves a rate $R = C_{Q}(\mathcal{N}[\eta,\nth])$ for the Gaussian thermal-loss channel $\mathcal{N}[\eta,\nth]$. Since $\mathcal{N}[\eta,\nth] = \mathcal{A}[G',0]\cdot\mathcal{N}[\eta',0]$ (see Lemma \ref{lemma:Thermal-loss  = Amplification + Pure-loss}), this implies that the following pair of encoding and recovery maps $\lbrace \mathcal{E} , \mathcal{R}\cdot\mathcal{A}[G',0] \rbrace$ achieves a rate $R = C_{Q}(\mathcal{N}[\eta,\nth])$ for the bosonic pure-loss channel $\mathcal{N}[\eta',0]$. Since an achievable rate $R$ is upper bounded by the quantum capacity $C_{Q}^{ n \le \bar{n}}(\mathcal{N}[\eta',0])$, Eq.\ \eqref{eq:data processing upper bound energy constrained} follows. Eq.\ \eqref{eq:data processing upper bound energy unconstrained} is derived by taking the $\bar{n}\rightarrow \infty$ limit and using $g(x) = \log_{2}(ex)  + \mathcal{O}(\frac{1}{x})$.   
\end{proof}

It has been realized in Refs.\ \cite{Rosati2018,Sharma2018,Noh2019} that the data-processing bound in Theorem \ref{theorem:data processing bound} can be improved by using another decomposition of Gaussian thermal-loss channels with a twisted order or pure-loss and amplification channels.  

\begin{lemma}[Thermal-loss = Pure-loss + Amplification \cite{Rosati2018,Sharma2018,Noh2019}]
A Gaussian thermal-loss channel $\mathcal{N}[\eta,\nth]$ can be decomposed into a quantum limited amplification channel and a bosonic pure-loss channel as follows: 
\begin{align}
\mathcal{N}[\eta,\nth] &= \mathcal{N}[\tilde{\eta}',0] \cdot \mathcal{A}[ \tilde{G}',0 ] , 
\end{align}
where $\tilde{G}'$ and $\tilde{\eta}'$ are given by 
\begin{align}
\tilde{G}' &= \frac{\eta}{ \eta - (1-\eta)\nth } \,\,\,\textrm{and}\,\,\, \tilde{\eta}' = \frac{\eta}{\tilde{G}'} = \eta - (1-\eta)\nth . 
\end{align}
Note that this decomposition is only applicable if $ \eta'\ge 0 \leftrightarrow \eta \ge  \frac{\nth}{\nth+1}$, i.e., when the Gaussian thermal-loss channel $\mathcal{N}[\eta,\nth]$ is not entanglement-breaking \cite{Holevo2008}.  \label{lemma:Thermal-loss = Pure-loss + Amplification}
\end{lemma}

Then using a data-processing argument, we can similarly relate the quantum capacity of the Gaussian thermal-loss channel $\mathcal{N}[\eta,\nth]$ with the quantum capacity of the bosonic pure-loss channel $\mathcal{N}[\tilde{\eta}',0] $. 

\begin{theorem}[Improved data-processing bound of the Gaussian thermal-loss channel capacity \cite{Rosati2018,Sharma2018,Noh2019}]
The quantum capacity of a Gaussian thermal-loss channel $\mathcal{N}[\eta,\nth]$ subject to an energy constraint $\mathrm{Tr}[\hat{\rho} \hat{n}] \le \bar{n}$ is upper bounded by the following improved data-processing bound $Q^{n\le \bar{n}}_{\scriptsize{\textrm{IDP}}}(\eta,\nth)$:
\begin{align}
C_{Q}^{n\le \bar{n}}(\mathcal{N}[\eta,\nth]) &\le Q^{n\le \bar{n}}_{\scriptsize{\textrm{IDP}}}(\eta,\nth)  \equiv C_{Q}^{n\le \tilde{G}'\bar{n}+(\tilde{G}'-1)}(\mathcal{N}[\tilde{\eta}',0]),  
\end{align}
where $\tilde{\eta}' =  \eta - (1-\eta)\nth $ and $\tilde{G}' = \frac{\eta}{ \eta - (1-\eta)\nth }$. More explicitly, $Q^{n\le \bar{n}}_{\scriptsize{\textrm{IDP}}}(\eta,\nth)$ is given by
\begin{align}
Q^{n\le \bar{n}}_{\scriptsize{\textrm{IDP}}}(\eta,\nth) &= \max\Big{[}g(\eta \bar{n}+(1-\eta)\nth)  - g\Big{(} \frac{ (1-\eta)(\nthtiny+1)(\eta\bar{n}+(1-\eta)\nthtiny)   }{\eta - (1-\eta)\nthtiny} \Big{)} , 0\Big{]}. 
\end{align}   
In the energy-unconstrained case (i.e., $\bar{n}\rightarrow\infty$), we have 
\begin{align}
C_{Q}(\mathcal{N}[\eta,\nth]) &\le \lim_{\bar{n}\rightarrow \infty}  Q^{n\le \bar{n}}_{\scriptsize{\textrm{IDP}}}(\eta,\nth) = \max\Big{[} \log_{2}\Big{(} \frac{\eta - (1-\eta)\nth }{ (1-\eta)(\nth+1) } \Big{)} ,0 \Big{]} . 
\end{align} \label{theorem:improved data processing bound}
\end{theorem}
\begin{proof}
Let $\lbrace \mathcal{E}^{n\le \bar{n}},\mathcal{D} \rbrace$ be the set of encoding and decoding which achieves a rate $R=C_{Q}^{n\le \bar{n}}(\mathcal{N}[\eta,\nth])$ for the Gaussian thermal-loss channel $\mathcal{N}[\eta,\nth]$. Then, the encoding and decoding set $\lbrace \mathcal{A}[\tilde{G}',0]\cdot\mathcal{E}^{n\le \bar{n}} , \mathcal{D} \rbrace$ achieves a rate $R = C_{Q}^{n\le \bar{n}}(\mathcal{N}[\eta,\nth])$ for the bosonic pure-loss channel $\mathcal{N}[\tilde{\eta}',0]$. Since the new encoding $\mathcal{E}' \equiv \mathcal{A}[\tilde{G}']\cdot\mathcal{E}^{n\le \bar{n}}$ has an average photon number  
\begin{align}
\bar{n}' = \mathrm{Tr}[ \mathcal{A}[ \tilde{G}',0 ](\hat{\rho}) ] = \tilde{G}'\bar{n}+(\tilde{G}'-1) , 
\end{align}
the rate $R$ should be less than $C_{Q}^{n\le \tilde{G}'\bar{n}+(\tilde{G'}-1)}(\mathcal{N}[\tilde{\eta}',0])$. The result for the energy-unconstrained case is obtained by taking $\bar{n}\rightarrow \infty$ and using $g(x) = \log_{2}(ex) + \mathcal{O}(\frac{1}{x})$.  
\end{proof}

\begin{figure}[t!]
\centering
\includegraphics[width=4.5in]{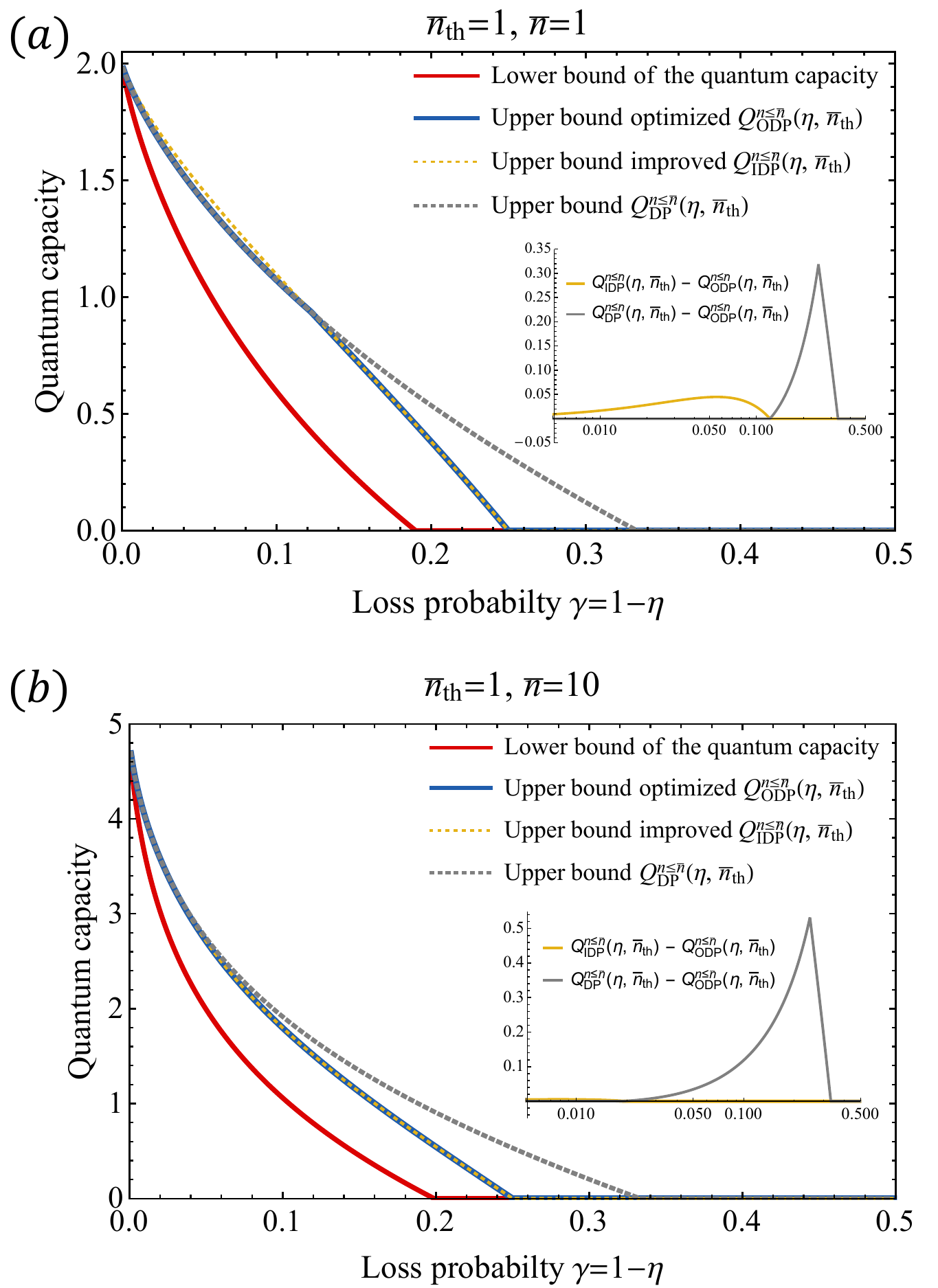}
\caption{[Fig.\ 1 in IEEE Trans. Info. Theory \textbf{65}, 2563--2582 (2019)] Bounds of the quantum capacity of a Gaussian thermal-loss channel $\mathcal{N}[\eta,\nthtiny]$ for (a) $(\nth,\bar{n})=(1,1)$ and (b) $(\nth,\bar{n})=(1,10)$. The solid red lines represent a lower bound of the Gaussian thermal-loss channel capacity which is obtained by evaluating the coherent information with respect to a thermal input state (see Eq.\ \eqref{eq:lower bound Gaussian thermal loss channel HW}). The dashed yellow line and the dashed grey line represent the improved data-processing bound and the data-processing bound, respectively. The solid blue line represents the optimized data-processing bound $Q^{n\le \bar{n}}_{\tiny{\textrm{ODP}}}(\eta,\nthtiny) \equiv \max[ Q^{n\le \bar{n}}_{\tiny{\textrm{IDP}}}(\eta,\nthtiny), Q^{n\le \bar{n}}_{\tiny{\textrm{DP}}}(\eta,\nthtiny) ]$. The improved data-processing bound $Q^{n\le \bar{n}}_{\tiny{\textrm{IDP}}}(\eta,\nthtiny)$ is identical to the optimized data-processing bound $Q^{n\le \bar{n}}_{\tiny{\textrm{ODP}}}(\eta,\nthtiny)$ in a wide range of parameter space, and is very close to the optimal one even when it is not optimal. The data-processing bound $Q^{n\le \bar{n}}_{\tiny{\textrm{DP}}}(\eta,\nthtiny)$ is optimal when $\eta \ge  \eta^{\star}(\nth,\bar{n})$ for some $\eta^{\star}(\nth,\bar{n})$ (e.g., $\eta^{\star}(1,1) = 0.8775\cdots$).  }
\label{fig:Gaussian thermal-loss channel capacity upper bound}
\end{figure}

Let us now compare the improved data-processing bound with the data-processing bound. Note that in the energy-unconstrained case (i.e., $\bar{n}\rightarrow\infty$), we have
\begin{equation}
Q_{\scriptsize{\textrm{IDP}}}(\eta,\nth) < Q_{\scriptsize{\textrm{DP}}}(\eta,\nth)\equiv\lim_{\bar{n}\rightarrow\infty} Q^{n\le \bar{n}}_{\scriptsize{\textrm{DP}}}(\eta,\nth) \label{eq:IDP better than DP in energy-unconstrained}
\end{equation}
for all $\eta\in[0,1)$, since 
\begin{align}
\eta -(1-\eta)\nth & <  \eta.  \label{eq:comparison of etas in data-processing bound}
\end{align}
Thus, the improved data-processing bound $Q_{\scriptsize{\textrm{IDP}}}(\eta,\nth)$ is a strictly tighter upper bound of the Gaussian thermal-loss channel capacity than the data-processing bound $Q_{\scriptsize{\textrm{DP}}}(\eta,\nth)$. On the other hand in the energy-constrained case, the improved data-processing bound $Q^{n\le \bar{n}}_{\scriptsize{\textrm{IDP}}}(\eta,\nth)$ is not always tighter than $Q^{n\le \bar{n}}_{\scriptsize{\textrm{DP}}}(\eta,\nth)$ (see the insets in Fig.\ \ref{fig:Gaussian thermal-loss channel capacity upper bound}). Physically, this is because the increased encoding energy due to the pre-amplification allows a larger quantum capacity, which is crucial if the allowed average photon number in the encoding is small, i.e., $\bar{n}\ll 1$. Thus, we take the maximum between $Q^{n\le \bar{n}}_{\scriptsize{\textrm{DP}}}(\eta,\nth)$ and $Q^{n\le \bar{n}}_{\scriptsize{\textrm{IDP}}}(\eta,\nth)$ to optimize the bound: 
\begin{align}
Q^{n\le \bar{n}}_{\tiny{\textrm{ODP}}}(\eta,\nthtiny) \equiv \max[ Q^{n\le \bar{n}}_{\tiny{\textrm{IDP}}}(\eta,\nthtiny), Q^{n\le \bar{n}}_{\tiny{\textrm{DP}}}(\eta,\nthtiny) ]. \label{eq:ODP}
\end{align}  
We numerically observe that 
\begin{align}
Q^{n\le \bar{n}}_{\scriptsize{\textrm{ODP}}}(\eta,\nth) = \begin{cases}
Q^{n\le \bar{n}}_{\scriptsize{\textrm{DP}}}(\eta,\nth) & \eta \ge \eta^{\star}(\nth,n)\\
Q^{n\le \bar{n}}_{\scriptsize{\textrm{DP}}}(\eta,\nth) & \eta < \eta^{\star}(\nth,n)
\end{cases} . 
\end{align}
for some $\eta^{\star}(\nth,n)$. In the energy-unconstrained case (i.e., $\bar{n}\rightarrow \infty$), we observe that $\lim_{\bar{n}\rightarrow\infty }\eta^{\star}(\nth,\bar{n}) =1$ and thus $Q^{n\le \bar{n}}_{\scriptsize{\textrm{ODP}}}(\eta,\nth)=Q^{n\le \bar{n}}_{\scriptsize{\textrm{IDP}}}(\eta,\nth)$ for all $\eta\in [0,1]$, which is consistent with Eq.\ \eqref{eq:IDP better than DP in energy-unconstrained}. We refer to Ref.\ \cite{Sharma2018} (e.g., Fig.\ 6 therein) for a more comprehensive comparison of the existing upper bounds, also including approximate degradability bounds \cite{Sutter2017}.

\section{Lower bounds of the Gaussian thermal-loss channel capacity}
\label{section:Lower bound of the Gaussian thermal-loss channel capacity}

Let us now move on to lower bounds of the Gaussian thermal-loss channel capacity. Recall that for a bosonic pure-loss channel $\mathcal{N}[\eta,0]$, its quantum capacity can be found by evaluating its one-shot coherent information with respect to an input thermal state $\hat{\tau}(\bar{n})$. For a Gaussian thermal-loss channel $\mathcal{N}[\eta,\nth]$ with $\nth\neq 0$, however, the thermal state $\hat{\tau}(\bar{n})$ is not necessarily an optimal input state that maximizes the coherent information. Nevertheless, the coherent information with respect to an input thermal state is a valid lower bound to the quantum capacity of a Gaussian thermal-loss channel capacity subject to an average energy constraint $\mathrm{Tr}[\hat{\rho}\hat{n}]\le \bar{n}$, i.e., 
\begin{align}
C_{Q}^{n\le \bar{n}}( \mathcal{N}[\eta,\nth] ) \ge I_{c}( \mathcal{N}[\eta,\nth] , \hat{\tau}(\bar{n}) ) &= g(\eta\bar{n}+(1-\eta)\nth)
\nonumber\\
&\quad -g\Big{(} \frac{D+(1-\eta)(\bar{n}-\nth)-1}{2} \Big{)}
\nonumber\\
&\quad -g\Big{(} \frac{D-(1-\eta)(\bar{n}-\nth)-1}{2} \Big{)} , \label{eq:lower bound Gaussian thermal loss channel HW}
\end{align}
where $D$ is defined as $D \equiv \sqrt{((1+\eta)\bar{n}+(1-\eta)\nth+1)^{2}-4\eta\bar{n}(\bar{n}+1)}$ \cite{Holevo2001}. This bound has been the best known lower bound for the past two decades. Here, we present an improved lower bound by showing that there is a non-trivial multi-channel strategy that can outperform the single-channel strategy with a thermal input state \cite{Noh2020a}.  

\subsection{Correlated multi-mode thermal states}

We first construct a family of Gaussian multi-mode states, called correlated multi-mode thermal states, which is the key ingredient for improving the lower bound of the Gaussian thermal-loss channel capacity. Recall that $\hat{\tau}(\bar{n})$ denotes the single-mode thermal state with an average photon number $\mathrm{Tr}[\hat{n} \hat{\tau}(\bar{n})] = \bar{n}$, i.e., 
\begin{align}
\hat{\tau}(\bar{n}) \equiv \sum_{n=0}^{\infty}\frac{\bar{n}^{n}}{(1+\bar{n})^{n+1}}|n\rangle\langle
n|, 
\end{align}
where $|n\rangle$ is a Fock state. Uncorrelated multi-mode
thermal states would then simply be given by a tensor product of
single-mode thermal states
$\big{\lbrace}\hat{\tau}(\bar{n})\big{\rbrace}^{\otimes N}$. Now
we define correlated multi-mode thermal states as follows:
\begin{align}
\hat{\mathcal{T}}(\vec{N},\vec{n}) \equiv \hat{U}_{\textrm{GFT}}^{(N)} \Big{[} \big{\lbrace}\hat{\tau}(\bar{n}_{1})\big{\rbrace}^{\otimes N_{1}}\otimes \cdots \otimes \big{\lbrace}\hat{\tau}(\bar{n}_{r})\big{\rbrace}^{\otimes N_{r}} \Big{]} \big{(} \hat{U}_{\textrm{GFT}}^{(N)} \big{)}^{\dagger}.
\end{align}
Here, $\vec{N} = (N_{1},\cdots, N_{r})$ such that $\sum_{k=1}^{r}N_{k} = N$ and $\vec{n} = (\bar{n}_{1},\cdots ,\bar{n}_{r})$. $\hat{U}_{\textrm{GFT}}^{(N)}$ is the $N$-mode Gaussian Fourier transformation whose action on the $j^{\textrm{th}}$ annihilation operator $\hat{a}_{j}$ is given by
\begin{align}
\big{(} \hat{U}_{\textrm{GFT}}^{(N)} \big{)}^{\dagger} \hat{a}_{j} \hat{U}_{\textrm{GFT}}^{(N)} &= \frac{1}{\sqrt{N}}\sum_{k=1}^{N}e^{i\frac{2\pi}{N}(j-1)(k-1)}\hat{a}_{k}, \label{eq:Gaussian Fourier transforamtion}
\end{align}
for all $j\in\lbrace 1,\cdots, N \rbrace$. Hence, the correlated multi-mode thermal state $\hat{\mathcal{T}}(\vec{N},\vec{n})$ is a collection of single-mode thermal states (where each of the first $N_{1}$ modes supports on average $\bar{n}_{1}$ photons, each of the next $N_{2}$ modes supports on average $\bar{n}_{2}$ photons and so on) which are uniformly mixed by the Gaussian Fourier transformation $\hat{U}_{\textrm{GFT}}^{(N)}$ (see Fig.\ \ref{fig:correlated multi-mode thermal states}). We remark that each mode in the correlated $N$-mode thermal state $\hat{\mathcal{T}}(\vec{N},\vec{n})$ supports on average $\bar{n} = \frac{1}{N}\sum_{k=1}^{r} N_{k}\bar{n}_{k}$ photons. 

\begin{figure}[t!]
\centering
\includegraphics[width=3.0in]{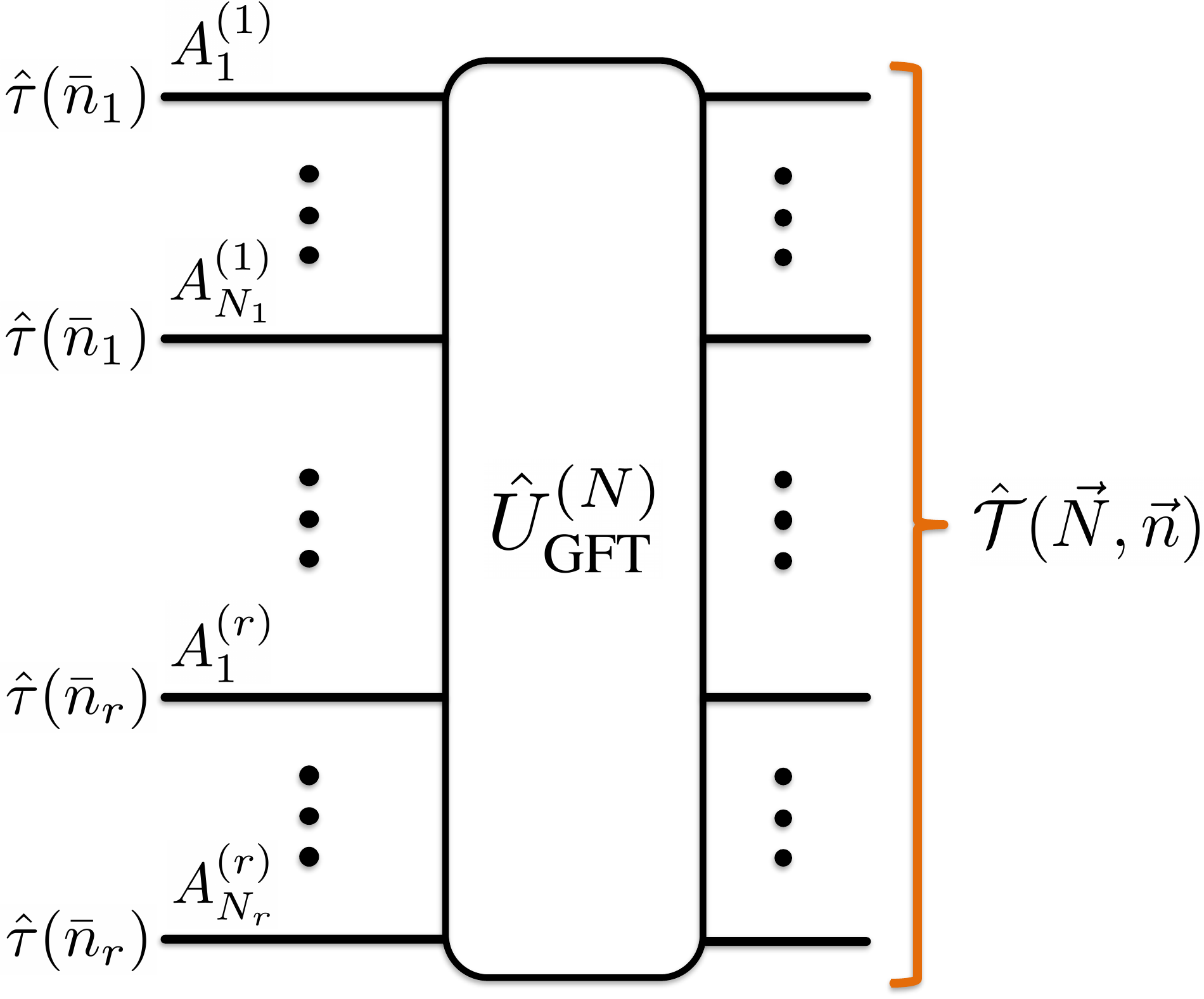}
\caption{[Fig.\ 1 in Nature Communications \textbf{11}, 457 (2020)] A correlated multi-mode thermal state $\hat{\mathcal{T}}(\vec{N},\vec{n})$ with $\vec{N} = (N_{1},\cdots, N_{r})$ and $\vec{n} = (\bar{n}_{1},\cdots ,\bar{n}_{r})$ such that $\sum_{k=1}^{r}N_{k} = N$. $\hat{\mathcal{T}}(\vec{N},\vec{n})$ can be generated by applying the $N$-mode Gaussian Fourier transformation $\hat{U}_{\textrm{GFT}}^{(N)}$ to an uncorrelated thermal state $\big{\lbrace}\hat{\tau}(\bar{n}_{1})\big{\rbrace}^{\otimes N_{1}}\otimes \cdots \otimes \big{\lbrace}\hat{\tau}(\bar{n}_{r})\big{\rbrace}^{\otimes N_{r}}$. }
\label{fig:correlated multi-mode thermal states}
\end{figure}

A simple non-trivial example of correlated multi-mode thermal
states would be $\hat{\mathcal{T}}(\vec{N},\vec{n})$ with $\vec{N}
= (1,N-1)$ and $\vec{n} = (N\bar{n},0)$ and its covariance
matrix is given by
\begin{align}
\boldsymbol{V}  = \begin{bmatrix}
(\bar{n}+\frac{1}{2})\boldsymbol{I}_{2}& \bar{n}\boldsymbol{I}_{2} & \cdots & \bar{n}\boldsymbol{I}_{2}\\
\bar{n}\boldsymbol{I}_{2}&(\bar{n}+\frac{1}{2})\boldsymbol{I}_{2}&\cdots& \bar{n}\boldsymbol{I}_{2}\\
\vdots&\vdots&\ddots&\vdots\\
\bar{n}\boldsymbol{I}_{2}&\bar{n}\boldsymbol{I}_{2}&\cdots &(\bar{n}+\frac{1}{2})\boldsymbol{I}_{2}
\end{bmatrix},
\end{align}
where $\boldsymbol{I}_{2}$ is the $2\times 2$ identity matrix. As can be seen from the diagonal elements of the covariance matrix, every mode supports on average $\bar{n}$ photons. Therefore, the reduced density matrix of each mode is given by a single-mode thermal state $\hat{\tau}(\bar{n})$. On the other hand, the off-diagonal elements of the covariance matrix indicate that the position (or the momentum) quadratures of every pair of modes are positively correlated: This is what distinguishes $\hat{\mathcal{T}}(\vec{N},\vec{n})$ from the uncorrelated $N$-mode thermal state $\big{\lbrace}\hat{\tau}(\bar{n})\big{\rbrace}^{\otimes N}$ and why we call it a correlated multi-mode thermal state. We remark that correlated multi-mode thermal states can be efficiently prepared because the Gaussian Fourier transformation $\hat{U}_{\textrm{GFT}}^{(N)}$ can be implemented efficiently by using a variant of the fast Fourier transform technique \cite{Cooley1965}.

\subsection{Superadditivity with respect to Gaussian input states}

Now we present the main result: 

\begin{theorem}
Consider a correlated $N$-mode thermal state $\hat{\mathcal{T}}(\vec{N},\vec{n})$ with $\vec{N} = (M,N-M)$ and $\vec{n} = (\frac{N}{M}\bar{n},0)$ and let $x = \frac{M}{N}$, where $M\in \lbrace 1,\cdots, N \rbrace$. Then, the coherent information with respect to the input state $\hat{\mathcal{T}}(\vec{N},\vec{n})$ is given by
\begin{align}
\frac{1}{N}I_{\textrm{c}}\big{(} \mathcal{N}[\eta,\nth]^{\otimes N} , \hat{\mathcal{T}}(\vec{N},\vec{n}) \big{)} = x I_{\textrm{c}}(\mathcal{N}[\eta,\nth], \hat{\tau}\Big{(}\frac{\bar{n}}{x}\Big{)} ). \label{eq:thermal loss correlated multi mode thermal state coherent information}
\end{align}
Since $x$ can be any rational number in $(0,1]$ and the set of rational numbers is a dense subset of the set of real numbers, we have the following improved lower bound of the quantum capacity of the Gaussian thermal-loss channels.
\begin{align}
C_{\textrm{Q}}^{n \le \bar{n}}(\mathcal{N}[\eta,\nth]) &\ge \max_{0< x\le 1}x I_{\textrm{c}}(\mathcal{N}[\eta,\nth], \hat{\tau}\Big{(}\frac{\bar{n}}{x}\Big{)} ) .  \label{eq:thermal loss KN lower bound}
\end{align}
\label{theorem:Quantum capacity thermal loss}
\end{theorem}

\begin{proof}
Let $\mathcal{U}_{\textrm{GFT}}^{(N)}(\hat{\rho})\equiv \hat{U}_{\textrm{GFT}}^{(N)} \hat{\rho} (\hat{U}_{\textrm{GFT}}^{(N)})^{\dagger}$ be the unitary quantum channel associated with the $N$-mode Gaussian Fourier transformation. Then, $\mathcal{U}_{\textrm{GFT}}^{(N)}$ commutes with the tensor product of Gaussian thermal-loss channels, i.e.,
\begin{align}
\mathcal{U}_{\textrm{GFT}}^{(N)} \mathcal{N}[\eta,\nth]^{\otimes N} = \mathcal{N}[\eta,\nth]^{\otimes N} \mathcal{U}_{\textrm{GFT}}^{(N)}. \label{method_eq:commutation between GFT and thermal loss}
\end{align}
This is a direct consequence of the fact that the $N$-mode Gaussian Fourier transformation is a passive linear optical operation with an orthogonal transformation matrix $T$. Now, recall that the correlated multi-mode thermal state $\hat{\mathcal{T}}(\vec{N},\vec{n})$ with $\vec{N}=(N_{1},\cdots,N_{r})$ and $\vec{n} = (\bar{n}_{1},\cdots, \bar{n}_{r})$ is defined as
\begin{align}
\hat{\mathcal{T}}(\vec{N},\vec{n}) = \mathcal{U}_{\textrm{GFT}}^{(N)} \Big{(}  \big{\lbrace}\hat{\tau}(\bar{n}_{1})\big{\rbrace}^{\otimes N_{1}}\otimes \cdots \otimes \big{\lbrace}\hat{\tau}(\bar{n}_{r})\big{\rbrace}^{\otimes N_{r}} \Big{)} . \label{method_eq:correlated multi-mode thermal state recalled}
\end{align}

Combining Eq.\ \eqref{method_eq:commutation between GFT and thermal loss} and Eq.\ \eqref{method_eq:correlated multi-mode thermal state recalled}, one can see that sending the correlated multi-mode thermal state $\hat{\mathcal{T}}(\vec{N},\vec{n})$ to the $N$ copies of Gaussian thermal-loss channels is equivalent to sending a collection of thermal states $\big{\lbrace}\hat{\tau}(\bar{n}_{1})\big{\rbrace}^{\otimes N_{1}}\otimes \cdots \otimes \big{\lbrace}\hat{\tau}(\bar{n}_{r})\big{\rbrace}^{\otimes N_{r}}$ to the Gaussian thermal-loss channels and then the receiver performing the Gaussian Fourier transformation. Since any local operations are assumed to be free, the achievable communication rates with the correlated multi-mode thermal state $\hat{\mathcal{T}}(\vec{N},\vec{n})$ is the same as the rates achievable with the collection of thermal states $\big{\lbrace}\hat{\tau}(\bar{n}_{1})\big{\rbrace}^{\otimes N_{1}}\otimes \cdots \otimes \big{\lbrace}\hat{\tau}(\bar{n}_{r})\big{\rbrace}^{\otimes N_{r}}$. 

Recall that coherent information is an achievable quantum state transmission rate. Since
\begin{align}
I_{\textrm{c}} ( \mathcal{N}[\eta,\nth]^{\otimes N}, \hat{\mathcal{T}}(\vec{N},\vec{n}) ) &=I_{\textrm{c}}\Big{(} \mathcal{N}[\eta,\nth]^{\otimes N}, \big{\lbrace}\hat{\tau}(\bar{n}_{1})\big{\rbrace}^{\otimes N_{1}}\otimes \cdots \otimes \big{\lbrace}\hat{\tau}(\bar{n}_{r})\big{\rbrace}^{\otimes N_{r}} \Big{)}
\nonumber\\
&= \sum_{k=1}^{r} N_{k} I_{\textrm{c}}( \mathcal{N}[\eta,\nth] , \hat{\tau}(\bar{n}_{k}) ) , \label{eq:theorem 1 proof intermediate}
\end{align}
the correlated multi-mode thermal state $\hat{\mathcal{T}}(\vec{N},\vec{n})$ can achieve the quantum state transmission rate
\begin{align}
\frac{1}{N}\sum_{k=1}^{r} N_{k} I_{\textrm{c}}( \mathcal{N}[\eta,\nth] , \hat{\tau}(\bar{n}_{k}) )
\end{align}
per channel use. Specializing this to $\vec{N} = (M,N-M)$ and $\vec{n} = (\frac{N}{M}\bar{n},0)$, we get the rate
\begin{align}
\frac{M}{N}I_{\textrm{c}}(\mathcal{N}[\eta,\nth], \hat{\tau}\Big{(} \frac{N}{M}\bar{n} \Big{)} ) = xI_{\textrm{c}}( \mathcal{N}[\eta,\nth], \hat{\tau}\Big{(} \frac{\bar{n}}{x} \Big{)} )
\end{align}
as stated in Eq.\ \eqref{eq:thermal loss correlated multi mode thermal state coherent information} in Theorem \ref{theorem:Quantum capacity thermal loss}, where $x\equiv M/N$. Following the rest of the arguments given in Theorem \ref{theorem:Quantum capacity thermal loss}, the theorem follows.

Note that it might appear that the use of Gaussian Fourier transformation is not necessary because as shown in Eq.\ \eqref{eq:theorem 1 proof intermediate}, the coherent information of the correlated multi-mode thermal state $\hat{\mathcal{T}}(\vec{N},\vec{n})$ is the same as the coherent information of the uncorrelated multi-mode thermal state $\big{\lbrace}\hat{\tau}(\bar{n}_{1})\big{\rbrace}^{\otimes N_{1}}\otimes \cdots \otimes \big{\lbrace}\hat{\tau}(\bar{n}_{r})\big{\rbrace}^{\otimes N_{r}}$. It is nevertheless essential to use the Gaussian Fourier transformation because it uniformly spreads the excessive photons in the uncorrelated multi-mode thermal state across all modes such that the energy constraint is fulfilled (see also the discussion below Eq.\ \eqref{eq:Gaussian Fourier transforamtion}). 
\end{proof}


\begin{figure}
\centering
\includegraphics[width=4.0in]{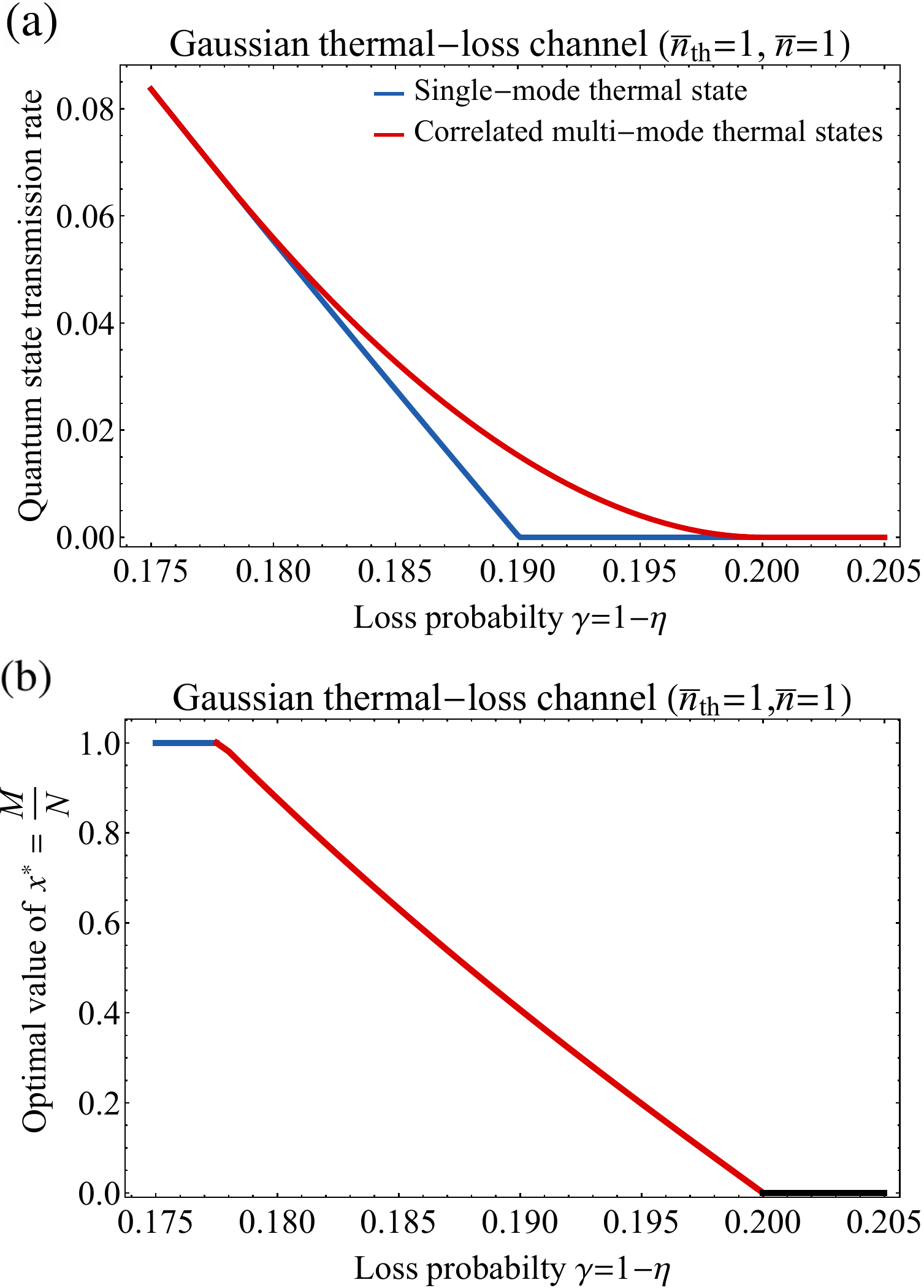}
\caption{[Fig.\ 2 in Nature Communications \textbf{11}, 457 (2020)] (a) Quantum state transmission rate of Gaussian thermal-loss channels $\mathcal{N}[\eta,\nth=1]$ as a function of the loss probability $\gamma=1-\eta$ achievable with the single-mode thermal state $\hat{\tau}(\bar{n})$ (blue, Eq.\ \eqref{eq:lower bound Gaussian thermal loss channel HW}) and with a correlated multi-mode thermal state $\hat{\mathcal{T}}(\vec{N},\vec{n})$ (red, Eq.\ \eqref{eq:thermal loss KN lower bound}) subject to the maximum allowed average photon number $\bar{n}=1$ per channel use. For the correlated multi-mode thermal states, the achievable rate was evaluated by taking $M/N = x^{\star}= \textrm{argmax}_{0 < x\le 1} x I_{\textrm{c}}(\mathcal{N}[\eta,\nth], \hat{\tau}(\bar{n}/x) )$, where $\vec{N} = (M,N-M)$ and $\vec{n} =(\frac{N}{M}\bar{n},0)$. (b) The optimal value of $x^{\star} = M/N$ at each $\gamma$, color-coded by the type of optimizer (blue: single-mode thermal states; red: correlated multi-mode thermal states), that yields the maximum quantum state transmission rate. We set $x^{\star}=M/N=0$ when all the states we consider yield vanishing quantum state transmission rate (black). }
\label{fig:Quantum capacity thermal loss}
\end{figure}

Our new bound in Eq.\ \eqref{eq:thermal loss KN lower bound} is at least as tight as the previous bound in Eq.\ \eqref{eq:lower bound Gaussian thermal loss channel HW} since the previous bound can be recovered by plugging in $x=1$ to the objective function. To demonstrate that our new bound can be strictly tighter than the previous bound, we take a family of Gaussian thermal-loss channels $\mathcal{N}[\eta,\nth]$ with $\nth=1$ and compute the new bound in Eq.\ \eqref{eq:thermal loss KN lower bound} for each $\eta=1-\gamma$, assuming that the maximum allowed average photon number per channel is $\bar{n}=1$. In Fig.\ \ref{fig:Quantum capacity thermal loss}(a), we plot the quantum state transmission rates achievable with the single-mode thermal state $\hat{\tau}(\bar{n}=1)$ and with the correlated multi-mode thermal states $\hat{\mathcal{T}}(\vec{N},\vec{n})$. When the loss probability is low (i.e., $\gamma \le 0.1775$), the single-mode thermal state yields the largest coherent information. However, when the loss probability is higher ($\gamma\ge 0.1775$), there exists a correlated multi-mode thermal state that outperforms the single-mode thermal state. Thus, we established a tighter lower bound to the quantum capacity of Gaussian thermal-loss channels than previously known \cite{Holevo2001}. In Fig.\ \ref{fig:Quantum capacity thermal loss}(b), we plot the optimal value of $M/N$ as a function of $\gamma$ that allows such a higher communication rate. It is important to note that only a finite number of modes is required if the optimal value of $x$ is a rational number. For example, $x^{\star}=3/8$ corresponds to the correlated $8$-mode thermal state $\hat{\mathcal{T}}(\vec{N},\vec{n})$ with $\vec{N}=(M,N-M)=(3,5)$ and $\vec{n} = (8\bar{n}/3,0)$. On the other hand, if $x^{\star}$ is irrational, one needs infinitely many modes to accurately obtain the rate $x I_{\textrm{c}}(\mathcal{N}[\eta,\nth], \hat{\tau} ( \bar{n}/x )  )|_{x=x^{\star}}$. 

\subsection{Convexity of coherent information and superadditivity}   
\label{subsection:Convexity of coherent information and superadditivity}

We now explain the non-trivial behavior shown in Fig.\ \ref{fig:Quantum capacity thermal loss} (i.e., $x^{\star} < 1$) in an intuitive way. Specifically, we relate the observed non-trivial behavior with the convexity of the coherent information $I_{\textrm{c}}(\mathcal{N}[\eta,\nth],\hat{\tau}(\bar{n}))$ in the allowed average photon number $\bar{n}$ for fixed values of $\eta$ and $\nth$. For concreteness, we take the Gaussian thermal-loss channel $\mathcal{N}[\eta,\nth]$ with $\eta = 0.81$ (or $\gamma = 0.19$) and $\nth = 1$ and plot its coherent information $I_{\textrm{c}}(\mathcal{N}[\eta,\nth],\hat{\tau}(\bar{n}))$ with respect to single-mode thermal states $\hat{\tau}(\bar{n})$ as a function of $\bar{n}$. As can be seen from the solid blue line in Fig.\ \ref{fig:achievable rates verses nbar}a, the coherent information $I_{\textrm{c}}(\mathcal{N}[\eta,\nth],\hat{\tau}(\bar{n}))$ is convex in $\bar{n}$ for small $\bar{n}$ and concave for large $\bar{n}$. Consider the region of rates achievable by the single-mode thermal states $A^{(1)}_{\eta,\nth} \equiv \lbrace (\bar{n},R) | \bar{n}\ge 0  \textrm{ and } R \le I_{\textrm{c}}(\mathcal{N}[\eta,\nth],\hat{\tau}(\bar{n}))  \rbrace$ (shaded blue region in Fig.\ \ref{fig:achievable rates verses nbar}a) and also its convex hull $A^{(\infty)}_{\eta,\nth} \equiv \textrm{ConvexHull}(A^{(1)}_{\eta,\nth})$ (shaded red and blue regions in Fig.\ \ref{fig:achievable rates verses nbar}a). We observe that the region $A^{(\infty)}_{\eta,\nth}$ is achievable by correlated multi-mode thermal states: Consider a generic convex combination of $r$ points in $A^{(1)}_{\eta,\nth}$, i.e.,
\begin{align}
\sum_{k=1}^{r} \lambda_{k} \Big{(} \bar{n}_{k} , I_{\textrm{c}}( \mathcal{N}[\eta,\nth],\hat{\tau}(\bar{n}_{k}) )  \Big{)}, \label{eq:rate convex combination}
\end{align}
where $\lambda_{k}\ge 0$ for all $k\in\lbrace 1,\cdots, r\rbrace$ and $\sum_{k=1}^{r} \lambda_{k}=1$. Then, the rate in Eq.\ \eqref{eq:rate convex combination} can be achieved by a correlated multi-mode thermal state $\hat{\mathcal{T}}(\boldsymbol{\mathrm{N}},\boldsymbol{\mathrm{n}})$ with $\boldsymbol{\mathrm{N}} = (N_{1},\cdots, N_{r})$ and $\boldsymbol{\mathrm{n}} = (\bar{n}_{1},\cdots, \bar{n}_{r})$ such that $\lambda_{k} = N_{k}/N$ for all $k\in \lbrace1,\cdots, r\rbrace$ where $N=\sum_{k=1}^{r}N_{k}$. Note that $\lambda_{k}$ should be a rational number. Similarly as above, however, by choosing a sufficiently large $N$ one can approximate any irrational $\lambda_{k}$ to a desired accuracy which can be arbitrarily small.    

\begin{figure}[t!]
\centering
\includegraphics[width=4.0in]{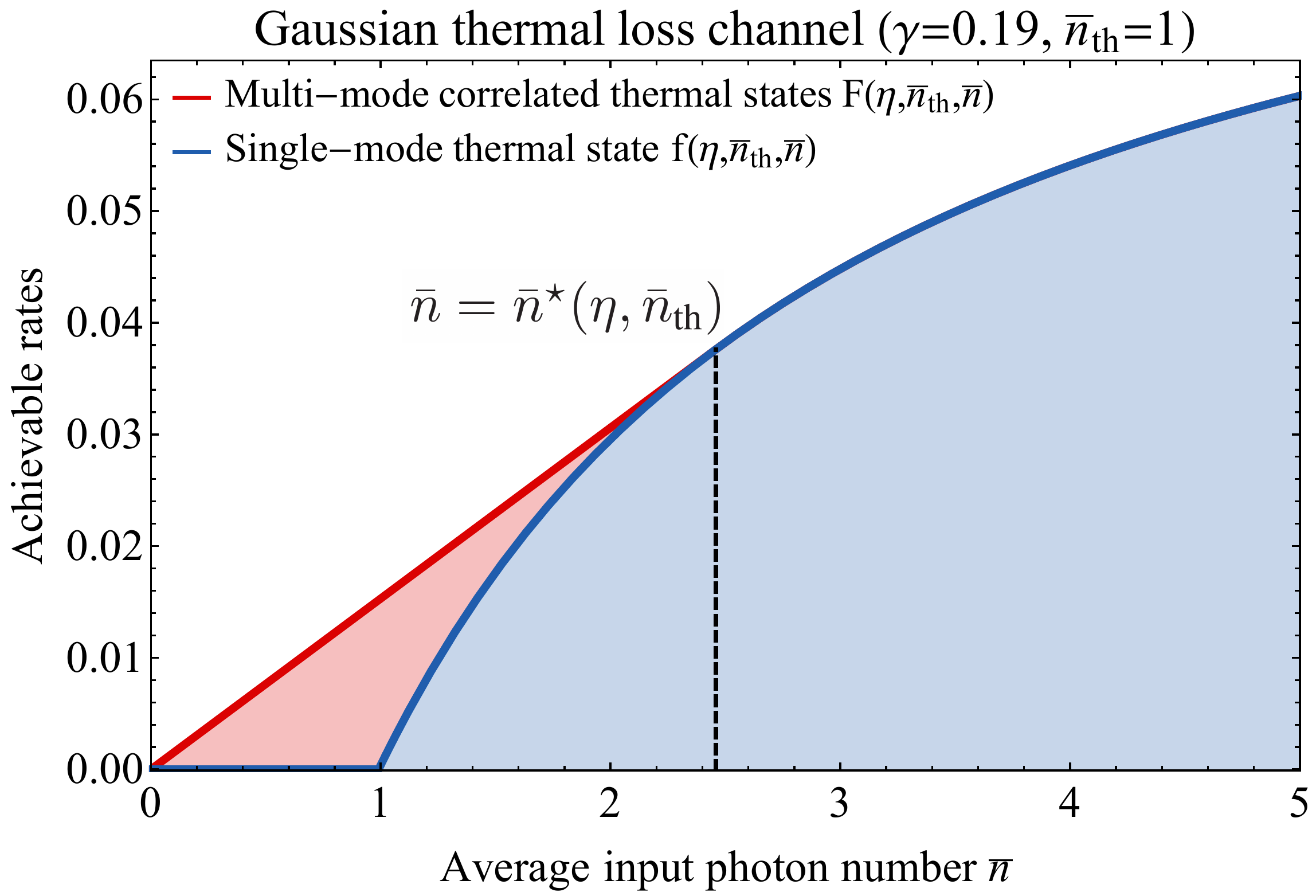}
\caption{[Fig.\ 3(a) in Nature Communications \textbf{11}, 457 (2020)] Achievable quantum state transmission rate of the single-mode (blue) and correlated multi-mode (red) thermal states as a function of $\bar{n}$ for the Gaussian thermal-loss channel $\mathcal{N}[\eta=0.81,\nth = 1]$. Note that the first-order contact point is given by $\bar{n}^{\star}(\eta=0.81,\nth=1) = 2.458$ which corresponds to $x^{\star} = 1/2.458 = 0.407$. This value agrees with $x^{\star}=0.407$ which is independently obtained in Fig.\ \ref{fig:Quantum capacity thermal loss}(b) for $\gamma=0.19$ and $\nth = 1$ (see also the main text).  }
\label{fig:achievable rates verses nbar}
\end{figure}

Importantly, due to the convexity of the coherent information $I_{\textrm{c}}(\mathcal{N}[\eta,\nth],\hat{\tau}(\bar{n}))$ in the small $\bar{n}$ regime, the region $A^{(\infty)}_{\eta,\nth}$ properly contains the region $A^{(1)}_{\eta,\nth}$, as indicated by the shaded red region in Fig.\ \ref{fig:achievable rates verses nbar}. This is why correlated multi-mode thermal states outperform single-mode thermal states in the noisy channel regime. In particular, the highest achievable rate can be obtained by taking the convex combination of the origin $(0,0)$ and the first-order contact point $(\bar{n}^{\star}(\eta,\nth) ,  I_{\textrm{c}}(\mathcal{N}[\eta,\nth], \hat{\tau} ( \bar{n}^{\star}(\eta,\nth) ) )$ with some weights $\lambda$ and $1-\lambda$, respectively (see the solid red line in Fig.\ \ref{fig:achievable rates verses nbar}). Note that the rate $x I_{\textrm{c}}(\mathcal{N}[\eta,\nth], \hat{\tau} ( \bar{n}/x )  )$ in Eq.\ \eqref{eq:thermal loss KN lower bound} can be understood as the one that is derived from such a convex combination with $1-\lambda = x = \bar{n}/\bar{n}^{\star}(\eta,\nth)$. For example, in the case of the Gaussian thermal-loss channel $\mathcal{N}[\eta,\nth]$ with $\eta = 0.81$ (or $\gamma = 0.19$) and $\nth =1$, the first-order contact point is given by $\bar{n}^{\star}(\eta,\nth) = 2.458$ (see Fig.\ \ref{fig:achievable rates verses nbar}) which corresponds to $x = 0.407$ for $\bar{n}=1$: This agrees with the optimal value $x^{\star} = 0.407$ in Fig.\ \ref{fig:Quantum capacity thermal loss}(b) for $\eta = 0.81$ (or $\gamma = 0.19$), $\nth=1$, and $\bar{n}=1$.

In a related work \cite{Lupo2009}, it was shown that a global encoding scheme with a correlated Gaussian input state can yield larger coherent information than a local encoding scheme with an uncorrelated Gaussian input state for lossy bosonic channels with correlated environmental noise. We remark that our work differs from this previous work in that we show a correlated Gaussian input state can outperform its uncorrelated counterpart even for the usual thermal-loss channels with uncorrelated environmental noise. Note that the loss model with uncorrelated environmental noise which we consider here has greater practical relevance because noise in realistic optical and microwave communication channels is well approximated by thermal-loss channels with uncorrelated environmental thermal noise \cite{Xiang2017,Axline2018}.  

\section{Open questions}
\label{section:Open questions quantum capacity} 

Note that our result in Theorem \ref{theorem:Quantum capacity thermal loss} can be understood as the establishment of the superadditivity of the coherent information of Gaussian thermal-loss channels with respect to Gaussian input states: As shown in Ref.\ \cite{Holevo2001}, the single-mode thermal state $\hat{\tau}(\bar{n})$ is the optimal single-mode Gaussian input state for the coherent information of Gaussian thermal-loss channels. Since we show that multi-mode correlated thermal states (which are Gaussian) sometimes outperform the single-mode thermal state, it means that the coherent information of Gaussian thermal-loss channels is superadditive with respect to Gaussian input states. On the other hand, it is still unclear whether the coherent information of Gaussian thermal-loss channels is genuinely superadditive with respect to all input states. This is because technically there is still a possibility that some non-Gaussian input state may outperform all Gaussian input states. It will thus be interesting to see if this is the case, and thus whether our result in Theorem \ref{theorem:Quantum capacity thermal loss} implies the genuine superadditivity of Gaussian thermal-loss channels.  

Another interesting open question is whether the convexity argument presented in Subsection \ref{subsection:Convexity of coherent information and superadditivity} can be adapted to explain the known superadditivity behavior of the qubit
depolarization \cite{DiVincenzo1998,Smith2007,Fern2008} and dephrasure \cite{Leditzky2018,Pirandola2019c,Bausch2020} channels. To contrast, we remark that the coherent information of a degradable channel is concave with respect to input states and its quantum capacity is additive  \cite{Devetak2005C,Caruso2006,Yard2008} (see also Ref.\ \cite{Watanabe2012}). 

We also remark that our improvement of the lower bounds is not strong enough to close the gap between the lower bound and the best-known upper bounds of the energy-constrained quantum quantum capacity of Gaussian thermal-loss channels \cite{Sharma2018,Rosati2018,Noh2019}. It will thus be interesting to see whether it is possible to further improve the lower and upper bounds to get a better understanding of the quantum capacity of Gaussian thermal-loss channels.   

\chapter{Achievable quantum state transmission rates with bosonic codes}
\label{chapter:Achievable communication rates with bosonic codes}

In this chapter, I will present results in Ref.\ \cite{Noh2019} on the achievable quantum state transmission rates with bosonic codes against Gaussian thermal-loss channels. While the quantum capacity is an achievable rate, evaluation of the quantum capacity does not lend explicit error correction strategies that achieve the quantum capacity. The main goal of this chapter is to provide explicit bosonic quantum error correction schemes that nearly achieve the fundamental limits set by the quantum communication theory.   

In Section \ref{section:Achievable rates of multi-mode GKP codes}, I will show that in the energy-unconstrained case, there exists a class of multi-mode GKP codes that achieves the quantum capacity of Gaussian thermal-loss channels up to at most a constant gap from the optimized data-processing upper bound established in the previous chapter. In Section \ref{section:Achievable rates of the numerically optimized single-mode codes}, I will apply the biconvex optimization technique which I introduced in Chapter \ref{chapter:Benchmarking and optimizing single-mode bosonic codes} to find optimal qudit-into-an-oscillator bosonic codes. Then, I will compute the achievable rates of these numerically optimized codes and demonstrate that the optimized single-mode codes achieve the energy-constrained quantum capacity of Gaussian thermal-loss channels up to at most a constant gap which is smaller than that of the energy-unconstrained case. I will conclude the chapter by outlining several open questions in Section \ref{section:Open questions achievable rates}.   

\section{Achievable rates of multi-mode GKP codes}
\label{section:Achievable rates of multi-mode GKP codes}

Recall Section \ref{section:Translation-symmetric bosonic codes} and note that it is possible to increase the size of the correctable shifts by using the hexagonal-lattice GKP code instead of the square-lattice GKP code. The reason why the hexagonal-lattice GKP code outperforms the square-lattice GKP code is because the hexagonal lattice allows more efficient circle packing than the square lattice. Furthermore, we can improve the performance of the hexagonal-lattice GKP code by using multiple (say $N$) modes collectively and using a $2N$-dimensional symplectic lattice allowing more efficient sphere packing than the $2$-dimensional hexagonal lattice. 

It is known that there exists a $2N$-dimensional lattice in the Euclidean space allowing $d_{\min}\ge \sqrt{N/(\pi e)}$ \cite{Hlawka1943} and a stronger statement was proven in Ref.\ \cite{Buser1994} that the same holds also for symplectic lattices. Choosing such a lattice to define the GKP code, one can correct all random displacement errors within the radius $r\le   \sqrt{N/( 2 e d)}$. For the Gaussian random displacement channel $\mathcal{N}_{B_{2}}[\sigma]$, the probability of a displacement with radius larger than $ \sqrt{2N}\sigma$ occurring vanishes in the limit of infinitely many modes $N\rightarrow\infty$. Thus, if $\sqrt{N/( 2 e d)}\ge  \sqrt{2N}\sigma$ is satisfied, i.e., 
\begin{equation}
d\le d_{\sigma} \equiv  \frac{1}{4e\sigma^{2}}, 
\end{equation}
encoded information can be transmitted faithfully with an asymptotically vanishing decoding error probability as $N\rightarrow\infty$. Then, it follows that a communication rate 
\begin{align}
R=\log_{2} \lfloor d_{\sigma} \rfloor = \log_{2} \Big{\lfloor} \frac{1}{4e\sigma^{2}} \Big{\rfloor} 
\end{align}
can be achieved for the Gaussian random displacement channel $\mathcal{N}_{B_{2}}[\sigma]$ (see Eq.\ (55) in Ref.\ \cite{Harrington2001}; the floor function is due to the fact that $d$ can only be an integer).

Note that the above estimation is overly conservative since we did not take into account correctable displacements outside the correctable sphere. With an improved estimation of the decoding error probability, the following statement was ultimately established: 

\begin{lemma}[Eq.\ (66) in \cite{Harrington2001}]  
Let $\mathcal{N}_{B_{2}}[\sigma]$ be a Gaussian random shift channel. Then, there exists a family of symplectic lattices generated by symplectic matrices $\boldsymbol{S}$ such that the corresponding GKP code family achieves the following rate for the Gaussian random displacement channel $\mathcal{N}_{B_{2}}[\sigma]$ in the $N\rightarrow\infty$ limit.  
\begin{equation}
R = \max \Big{(} \log_{2} \Big{\lfloor} \frac{1}{e\sigma^{2}} \Big{\rfloor} , 0  \Big{)}.  
\end{equation}
Here, $\lfloor x \rfloor$ is the floor function, due to the fact that the dimension of the GKP code space $d$ can only be an integer.  \label{lemma:achievable GKP communication rate for Gaussian random displacement channel}
\end{lemma}  

In Section \ref{section:Decoding GKP codes subject to excitation loss errors}, we showed that the GKP codes work well against excitation loss errors because we can convert an excitation loss error via an amplification to a random shift error, and then use the conventional decoding strategies for the GKP codes as described in Section \ref{section:Translation-symmetric bosonic codes}. Here, we generalize this observation to the Gaussian thermal-loss channels. 

\begin{lemma}[Thermal-loss + Amplification = Random Shift \cite{Noh2019}]
Let $\mathcal{N}[\eta,\nth]$ be a Gaussian thermal-loss channel and $\mathcal{A}[\frac{1}{\eta},0]$ be a quantum-limited amplification channel with gain $\frac{1}{\eta}$. Then, we have 
\begin{align}
\mathcal{N}[\eta,\nth] \cdot \mathcal{A}\Big{[}\frac{1}{\eta},0\Big{]} =  \mathcal{N}_{B_{2}}[\tilde{\sigma}_{\eta,\nth}], 
\end{align}
where the noise standard variance $(\tilde{\sigma}_{\eta,\nth})^{2}$ is given by
\begin{align}
(\tilde{\sigma}_{\eta,\nth})^{2} = (1-\eta)(\nth+1). \label{eq}
\end{align} \label{lemma:Thermal-loss + Amplification = Random Shift}
\end{lemma}

Note that we have reversed the order of the amplification and the loss channel when compared with the earlier channel conversion scheme in Theorem \ref{theorem:loss plus amplification is displacement post-amplification}. Thus, the resulting noise variance $(\tilde{\sigma}_{\eta,\nth})^{2}$ is smaller than the noise variance $(\sigma_{\eta,0})^{2}$ obtained in Theorem \ref{theorem:loss plus amplification is displacement post-amplification} when $\nth =0 $, i.e., 
\begin{align}
(\tilde{\sigma}_{\eta,0})^{2} = (1-\eta) < \frac{1-\eta}{\eta} = (\sigma_{\eta,0})^{2}, 
\end{align}
for all $\eta \in [0,1)$. The reason for the smaller noise variance is that we have applied the amplification channel before sending a quantum state to the thermal-loss channel. This way, the added noise from the amplification channel is reduced by a factor of $\eta$ due to the loss channel and the added noise from the loss channel is not amplified. On the other hand, if we send a quantum state to the loss channel and then amplify it later on the receiver side, the added noise from the loss channel is amplified by a factor of $\frac{1}{\eta}$ due to the amplification channel and the added noise from the amplification channel is not reduced. 

In Section \ref{section:Decoding GKP codes subject to excitation loss errors}, we did not consider this twisted order of the amplification and the loss channel because we restricted ourselves to situations that any decoding attempts have to be made after the noisy channel is applied. However, in the context of quantum communication, the encoder (or information sender) can pre-amplify the transmitted quantum state so the channel conversion strategy in Lemma \ref{lemma:Thermal-loss + Amplification = Random Shift} is allowed. Now, based on this lemma, we establish achievable quantum state transmission rates of the GKP codes against Gaussian thermal-loss channels: 

\begin{theorem}[Achievable rates of the GKP codes against Gaussian thermal-loss channels \cite{Noh2019}]
Let $\mathcal{N}[\eta,\nth]$ be a Gaussian thermal-loss channel. Then, there exists a family of symplectic lattices generated by symplectic matrices $\boldsymbol{S}$ such that the corresponding GKP code family achieves the following rate for the Gaussian random displacement channel $\mathcal{N}[\eta,\nth]$ in the $N\rightarrow\infty$ limit.  
\begin{equation}
R = \max \Big{(} \log_{2} \Big{\lfloor} \frac{1}{e\tilde{\sigma}_{\eta,\nth}^{2}} \Big{\rfloor} , 0  \Big{)} =  \max \Big{(} \log_{2} \Big{\lfloor} \frac{1}{e(1-\eta)(\nth+1)} \Big{\rfloor} , 0  \Big{)} . 
\end{equation}
Here, $\lfloor x \rfloor$ is the floor function, due to the fact that the dimension of the GKP code space $d$ can only be an integer. \label{theorem:achievable rate of the GKP codes for Gaussian loss channel}
\end{theorem} 
\begin{proof}
Lemma \ref{lemma:Thermal-loss + Amplification = Random Shift} states that a Gaussian thermal-loss channel $\mathcal{N}[\eta,\nth]$ can be converted via a quantum-limited amplification $\mathcal{A}[1/\eta]$ into a Gaussian random shift channel $\mathcal{N}_{B_{2}}[\tilde{\sigma}^{2}_{\eta,\nthtiny}]$, where $\tilde{\sigma}^{2}_{\eta,\nthtiny} = (1-\eta)(\nth+1)$. Combining this with Lemma \ref{lemma:achievable GKP communication rate for Gaussian random displacement channel}, the theorem follows. 
\end{proof}

\begin{figure}[!t]
\centering
\includegraphics[width=4.3in]{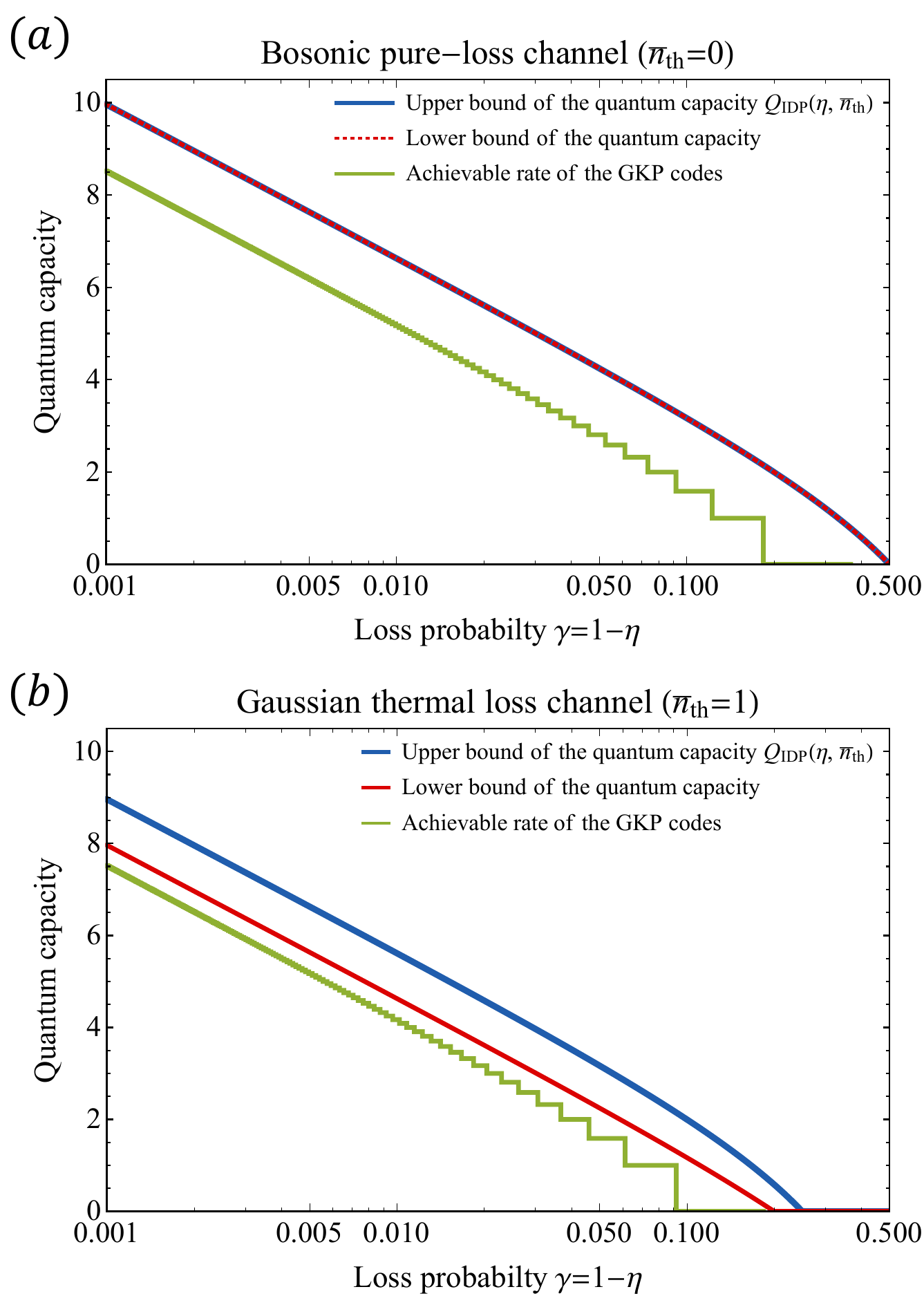}
\caption{[Fig.\ 3 in IEEE Trans. Info. Theory \textbf{65}, 2563--2582 (2019)] Achievable quantum state transmission rate of the GKP codes (Theorem \eqref{theorem:achievable rate of the GKP codes for Gaussian loss channel}; green) compared with a lower bound (Theorem \ref{theorem:Energy-constrained quantum capacity of bosonic pure-loss channels} and Eq.\ \eqref{eq:lower bound Gaussian thermal loss channel HW}; red) and an upper bound (Theorem \ref{theorem:improved data processing bound}; blue) of the quantum capacity of (a) the bosonic pure-loss channels $\mathcal{N}[\eta,\nthtiny = 0]$ and (b) the Gaussian thermal-loss channels $\mathcal{N}[\eta,\nthtiny = 1]$. The red line in the (a) panel overlaps with the blue line, since for the bosonic pure-loss channels, the lower and upper bounds coincide with each other and the quantum capacity is analytically determined.}
\label{fig:Achievable rate of GKP codes compared with capacity}
\end{figure} 

Recall Theorem \ref{theorem:improved data processing bound} and note that $C_{Q}(\mathcal{N}[\eta,\nth]) \le Q_{\scriptsize{\textrm{IDP}}}(\eta,\nth)$ where
\begin{equation}
Q_{\scriptsize{\textrm{IDP}}}(\eta,\nth) = \max \Big{[}\log_{2}\Big{(} \frac{\eta}{(1-\eta)(\nth+1)} \Big{)}, 0\Big{]} < \max \Big{[}\log_{2}\Big{(} \frac{1}{(1-\eta)(\nth+1)} \Big{)}, 0\Big{]} . 
\end{equation}
Comparing this with the rate established in Theorem \ref{theorem:achievable rate of the GKP codes for Gaussian loss channel}, we find 
\begin{align}
C_{Q}(\mathcal{N}[\eta,\nth])- R \lesssim \log_{2} e = 1.44269\cdots , 
\end{align}
where $\sim$ is due to the floor function. Thus, a family of the GKP code defined over an optimal symplectic lattice achieves the quantum capacity of Gaussian thermal-loss channels up to at most a constant gap from an upper bound of the quantum capacity (see Fig.\ \ref{fig:Achievable rate of GKP codes compared with capacity} for an illustration).

The established rate in Theorem \ref{theorem:achievable rate of the GKP codes for Gaussian loss channel} relies on the existence of a symplectic lattice in higher dimensions satisfying a certain desired condition (see Eqs.\ (56) and (57) in \cite{Harrington2001}). In this regard, we remark that the $E_{8}$ lattice and the Leech lattice $\Lambda_{24}$ (both symplectic; see appendix of \cite{Buser1994}) were recently shown to support the densest sphere packing in 8 and 24 dimensional Euclidean spaces, respectively \cite{Viazovska2017,Cohn2017}, and can be used to define a $4$-mode and a $12$-mode GKP code, respectively.

\section{Achievable rates of the numerically optimized single-mode codes}
\label{section:Achievable rates of the numerically optimized single-mode codes}

Here, we look for explicit bosonic codes that near achieve the quantum capacity of Gaussian thermal-loss channels in the energy-constrained case. To do so, we apply the biconvex optimization method developed in Section \ref{section:Optimizing single-mode bosonic codes} to find an optimal qudit-into-an-oscillator code with $d\in\lbrace 2,3,4,5\rbrace$ for the bosonic pure-loss channel $\mathcal{N}[\eta,\nth]=\mathcal{N}[0.9,0]$ and also for the Gaussian thermal-loss channel $\mathcal{N}[\eta,\nth]=\mathcal{N}[0.9,1]$. Then, we estimate the achievable rates of these numerically optimized codes. 

To represent the Gaussian thermal loss channel $\mathcal{N}[\eta,\nth]$ in the Fock basis, we use the decomposition in Lemma \ref{lemma:Thermal-loss  = Amplification + Pure-loss}, i.e., 
\begin{align}
\mathcal{N}[\eta,\nth] = \mathcal{A}[G']\mathcal{N}[\eta',0], \,\,\, \textrm{where}\,\,\, G'=(1-\eta)\nth +1  = \frac{\eta}{\eta'}  , 
\end{align}
and the Kraus representation of the quantum-limited amplification $\mathcal{A}[G](\hat{\rho}) = \sum_{g=0}^{\infty}\hat{A}_{g}\hat{\rho}\hat{A}_{g}^{\dagger}$ where \cite{Ivan2011}  
\begin{align}
\hat{A}_{g} &=  \sqrt{ \frac{1}{g!} \Big{(}1-\frac{1}{G}\Big{)}^{g} } (\hat{a}^{\dagger})^{g} \Big{(} \frac{1}{G} \Big{)}^{\frac{\hat{n}+1}{2}}. 
\end{align} 
Note that the photon gain parameter $g$ can take any non-negative integer values even when we consider a truncated bosonic space $\mathcal{H}_{n}$ and thus we should also truncate $g$ at some sufficiently large $g_{\max}$. 

\begin{figure}[!t]
\centering
\includegraphics[width=5.9in]{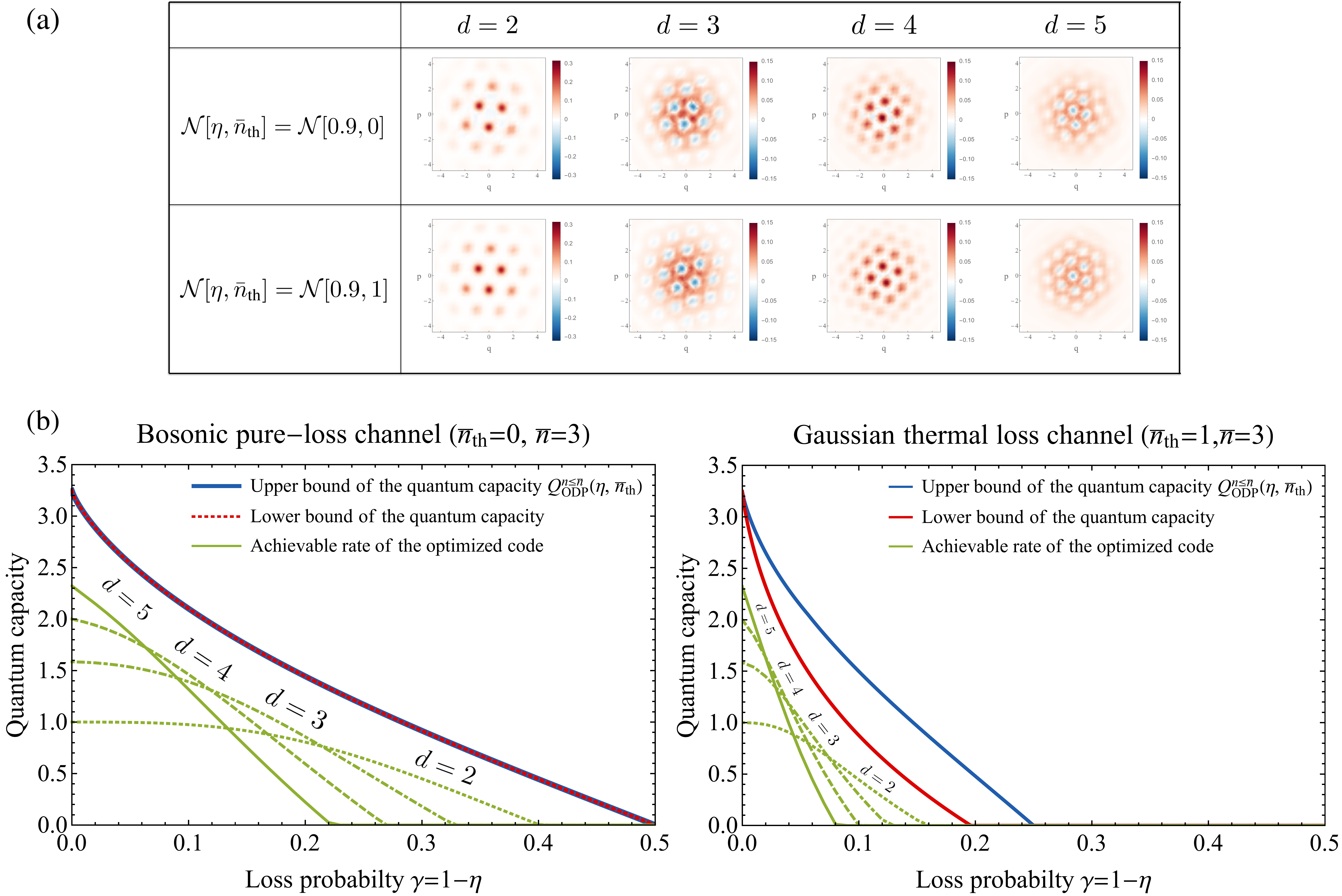}
\caption{[Fig.\ 5 in IEEE Trans. Info. Theory \textbf{65}, 2563--2582 (2019)] (a) Wigner function of the maximally mixed code state $\hat{\rho}_{\mathcal{E}}= (1/d) \mathrm{Tr}_{\mathcal{H}'}\hat{X}_{\mathcal{E}}$ of an optimized qudit-into-an-oscillator code ($d\in\lbrace 2,3,4,5 \rbrace$) for the bosonic pure-loss channel $\mathcal{N}[0.9,0]$ (top) and the Gaussian thermal-loss channel $\mathcal{N}[0.9,1]$ (bottom), subject to an average photon number constraint $\mathrm{Tr}[\hat{n}\hat{\rho}_{\mathcal{E}}] \le  \bar{n}= 3$, obtained by the alternating SDP method (Section \ref{section:Optimizing single-mode bosonic codes}). (b) Lower bounds of the achievable rates of the numerically optimized one-mode qudit-into-an-oscillator codes (at $\eta=0.9$) for the bosonic pure-loss channels $\mathcal{N}[\eta,0]$ (left) and the Gaussian thermal-loss channels $\mathcal{N}[\eta,1]$ (right) for each $d\in\lbrace 2,3,4,5\rbrace$, compared with lower (Eq.\ \eqref{eq:lower bound Gaussian thermal loss channel HW}) and upper (Eq.\ \eqref{eq:ODP}) bounds of the energy-constrained quantum capacity of the Gaussian thermal-loss channels.      }
\label{fig:biconvex optimization qudit into an oscillator}
\end{figure}

In Fig.\ \ref{fig:biconvex optimization qudit into an oscillator}(a), we took $\eta=0.9$, $\nth\in\lbrace 0,1\rbrace$, $n=30$, $d\in\lbrace 2,3,4,5 \rbrace$, $g_{\max}=15$ (for $\nth=1$) and plot the Wigner function of the maximally mixed code state $\hat{\rho}_{\mathcal{E}}$ of the optimized qudit-into-an-oscillator codes, subject to an average photon number constraint $\mathrm{Tr}[\hat{n}\hat{\rho}_{\mathcal{E}}] \le \bar{n}= 3$. Similarly as in the qubit-into-an-oscillator case, the alternating SDP method (starting from a Haar random initial code) yields a hexagonal-lattice GKP code (up to an overall displacement) as an optimal solution of the biconvex optimization for all values of $d\in\lbrace 2,3,4,5 \rbrace$ and $\nth\in \lbrace 0,1\rbrace$.

To characterize the achievable rates of the numerically optimized codes, we take the obtained codes (optimized at $\eta=0.9$) for each $(d,\nth)$ and compute the optimal entanglement fidelity $F_{e}^{\star}(\eta)$ by optimizing the decoding map for bosonic pure-loss channels $\mathcal{N}[\eta,0]$ and for Gaussian thermal loss channels $\mathcal{N}[\eta,1]$ for $0.5\le \eta \le 1$ via a semidefinite program (see Eq.\ \eqref{eq:SDP decoding optimization}). Note that any bipartite state $\hat{\rho}\in \mathcal{D}(\mathcal{H}'\otimes \mathcal{H}')$ ($\textrm{dim}(\mathcal{H}')=d$) with an entanglement fidelity $F_{e}$ can be converted into an Werner state, 
\begin{equation}
\hat{W}(F_{e},d) \equiv F_{e}|\Phi^{+}\rangle\langle \Phi^{+}|  + \frac{(1-F_{e})}{d^{2}-1} (\hat{I}-|\Phi^{+}\rangle\langle \Phi^{+}|),  
\end{equation}   
by an LOCC protocol based on the isotropic twirling $\int d\hat{U} \hat{U}\otimes \hat{U}^{*} \hat{\rho} (\hat{U}\otimes \hat{U}^{*})^{\dagger}$, where $\hat{U}^{*}$ is the complex conjugate of $\hat{U}$ \cite{Werner1989}. Thus, the coherent information of the Werner state 
\begin{equation}
R(F_{e},d) = \log d +F_{e}\log F_{e} + (1-F_{e})\log \Big{(} \frac{1-F_{e}}{d^{2}-1} \Big{)}
\end{equation} 
is a lower bound of the achievable rate of the state $\hat{\rho}$ with $\langle \Phi^{+}|\hat{\rho}|\Phi^{+}\rangle = F_{e}$.

In Fig.\ \ref{fig:biconvex optimization qudit into an oscillator}(b), we plot the lower bound of the achievable rate of the numerically optimized qudit-into-an-oscillator codes (obtained by evaluating $R(F_{e}^{\star},d)$ for the optimized entanglement fidelity $F_{e}^{\star}$) for bosonic pure-loss channels $\mathcal{N}[\eta,0]$ and for Gaussian thermal loss channels $\mathcal{N}[\eta,1]$ and compare it with the lower and upper bounds of the energy constrained quantum capacity (with $\bar{n}=3$). The achievable rate deviates from an upper bound of the quantum capacity at most by $0.923$ and $1.141$ qubits per channel use in the case of $\mathcal{N}[\eta,0]$ and $\mathcal{N}[\eta,1]$ respectively, which are smaller than the gap $\log e  = 1.44269\cdots$ we found in Section \ref{section:Achievable rates of multi-mode GKP codes}. This is because the amplification decoding we considered in Section \ref{section:Achievable rates of multi-mode GKP codes} is not the optimal decoding. We remark that the gap between the achievable rate and the upper bound in Fig.\ \ref{fig:biconvex optimization qudit into an oscillator}(b) may be further reduced if we use a qudit-into-$N$-oscillators encoding (such as the GKP code defined over an optimal $2N$-dimensional symplectic lattice for $N\ge 2$), instead of qudit-into-an-oscillator encoding.

\section{Open questions}
\label{section:Open questions achievable rates}

Recall that we used the amplification decoding in Section \ref{section:Decoding GKP codes subject to excitation loss errors} to establish the achievable rates in Theorem \ref{theorem:achievable rate of the GKP codes for Gaussian loss channel} for the multi-mode GKP codes in the energy-unconstrained case. Similarly as in Chapter \ref{chapter:Benchmarking and optimizing single-mode bosonic codes}, an immediate open question is whether a higher quantum communication rate can be achieved if we use an optimal decoding strategy instead of the sub-optimal amplification decoding scheme. It will be especially interesting to see if such an optimal decoding can be used to make the GKP code achieve a higher quantum state transmission rate than the lower bound of the quantum capacity given in Eq.\ \eqref{eq:lower bound Gaussian thermal loss channel HW}. Note also that the numerical biconvex optimization was applied only to single-mode bosonic codes. It will thus be interesting to see if the gap between the achievable rates and the upper bounds of the quantum capacity in Fig.\ \ref{fig:biconvex optimization qudit into an oscillator} can be reduced as we apply the same analysis to multi-mode bosonic codes.     

\chapter{Non-Gaussian resources for bosonic quantum information processing} 
\label{chapter:Non-Gaussian resources for bosonic quantum information processing}

In this chapter, I will discuss the importance of non-Gaussian resources for continuous-variable quantum information processing. This chapter is based on my work on oscillator encoding in Ref.\ \cite{Noh2019a} and an unpublished negative result on cubic phase state distillation. The work in Ref.\ \cite{Noh2019a} was done in collaboration with Professors Steve Girvin and Liang Jiang. 

Gaussian states, operation, and measurements can be efficiently simulated by using a classical computer \cite{Weedbrook2012}. Thus, non-Gaussian resources \cite{Zhuang2018,Takagi2018} are essential for realizing any non-trivial quantum computation beyond the reach of classical computation. Examples of non-Gaussian resources include the single-photon Fock state and photon-number-resolving measurements \cite{Knill2001,Aaronson2011}, Kerr nonlinearities \cite{Lloyd1999}, cubic phase state and gate \cite{Gottesman2001}, SNAP gate \cite{Krastanov2015}, Schr\"odinger cat states \cite{Cochrane1999}, and GKP states \cite{Gottesman2001,Baragiola2019}. Non-Gaussian resources are also crucial for bosonic quantum error correction. This is due to the established no-go results \cite{Eisert2002,Niset2009,Vuillot2019} which state that Gaussian errors cannot be corrected by using only Gaussian operations. Since Gaussian errors are ubiquitous in many realistic bosonic systems, these no-go results set a hard limit on the practical utility of the Gaussian QEC schemes.  

In Chapters \ref{chapter:Bosonic quantum error correction}, \ref{chapter:Benchmarking and optimizing single-mode bosonic codes}, \ref{chapter:Fault-tolerant bosonic quantum error correction}, \ref{chapter:Achievable communication rates with bosonic codes}, it has been shown that GKP states are valuable resources for realizing error-corrected discrete-variable quantum information processing. Here, I will show that GKP states are also a valuable non-Gaussian resource for implementing error-corrected continuous-variable quantum information processing. 

In Section \ref{section:GKP state as a non-Gaussian resource}, I will circumvent the no-go results on Gaussian QEC \cite{Eisert2002,Niset2009,Vuillot2019} and provide a non-Gaussian oscillator-into-oscillators encoding scheme that can correct practically relevant Gaussian errors such as random shift errors and excitation loss errors. The only non-Gaussian resource needed in the scheme is preparation of the canonical GKP states. I will also discuss adverse effects of the finite-squeezing in approximate GKP states. In Section \ref{section:Cubic phase state as a non-Gaussian resource}, I will discuss cubic phase states \cite{Gottesman2001} which are analogous to magic states \cite{Bravyi2005} for the conventional discrete-variable quantum computation. In particular, I will present some of my failed attempts on cubic phase state distillation.

\section{GKP state as a non-Gaussian resource}
\label{section:GKP state as a non-Gaussian resource}

\subsection{Canonical GKP state and modular quadrature measurement} 

The canonical GKP state was introduced in Section \ref{section:Translation-symmetric bosonic codes}. Here, we summarize the properties of the canonical GKP states that are referenced in this section. The Heisenberg uncertainty principle states that the position and momentum operators $\hat{q} \equiv (\hat{a}^{\dagger}+\hat{a})/\sqrt{2}$ and $\hat{p}\equiv i(\hat{a}^{\dagger}-\hat{a})/\sqrt{2}$ cannot be measured simultaneously because they do not commute with each other (i.e., $[\hat{q},\hat{p}]=i\neq 0$). Despite the Heisenberg uncertainty principle, however, the following displacement operators
\begin{align}
\hat{S}_{q}\equiv e^{i\sqrt{2\pi}\hat{q}} \,\,\,\textrm{and}\,\,\,\hat{S}_{p} \equiv e^{-i\sqrt{2\pi}\hat{p}}
\end{align}
do commute with each other and therefore can be measured simultaneously \cite{Gottesman2001}. Note that measuring $\hat{S}_{q} = \exp[i\sqrt{2\pi}\hat{q}]$ and $\hat{S}_{p} \equiv \exp[-i\sqrt{2\pi}\hat{p}]$ is equivalent to measuring their exponents (or phase angles) $i\sqrt{2\pi}\hat{q}$ and $-i\sqrt{2\pi}\hat{p}$ modulo $2\pi i$. Thus, the commutativity of $\hat{S}_{q}$ and $\hat{S}_{p}$ implies that the position and momentum operators can indeed be measured simultaneously if they are measured modulo $\sqrt{2\pi}$. The canonical GKP state (or the grid state) \cite{Gottesman2001,Duivenvoorden2017} is then defined as the unique (up to an overall phase) simultaneous eigenstate of the two commuting displacement operators $\hat{S}_{q}$ and $\hat{S}_{p}$ with unit eigenvalues. Explicitly, the canonical GKP state is given by
\begin{align}
|\textrm{GKP}\rangle \propto \sum_{n\in\mathbb{Z}}|\hat{q}=\sqrt{2\pi}n\rangle \propto \sum_{n\in\mathbb{Z}}|\hat{p}=\sqrt{2\pi}n\rangle. 
\end{align}     
Clearly, the canonical GKP state has definite values of both the position and momentum operators modulo $\sqrt{2\pi}$, i.e., $\hat{q}=\hat{p}=0$ mod $\sqrt{2\pi}$.

\begin{figure}[t!]
\centering
\includegraphics[width=5.8in]{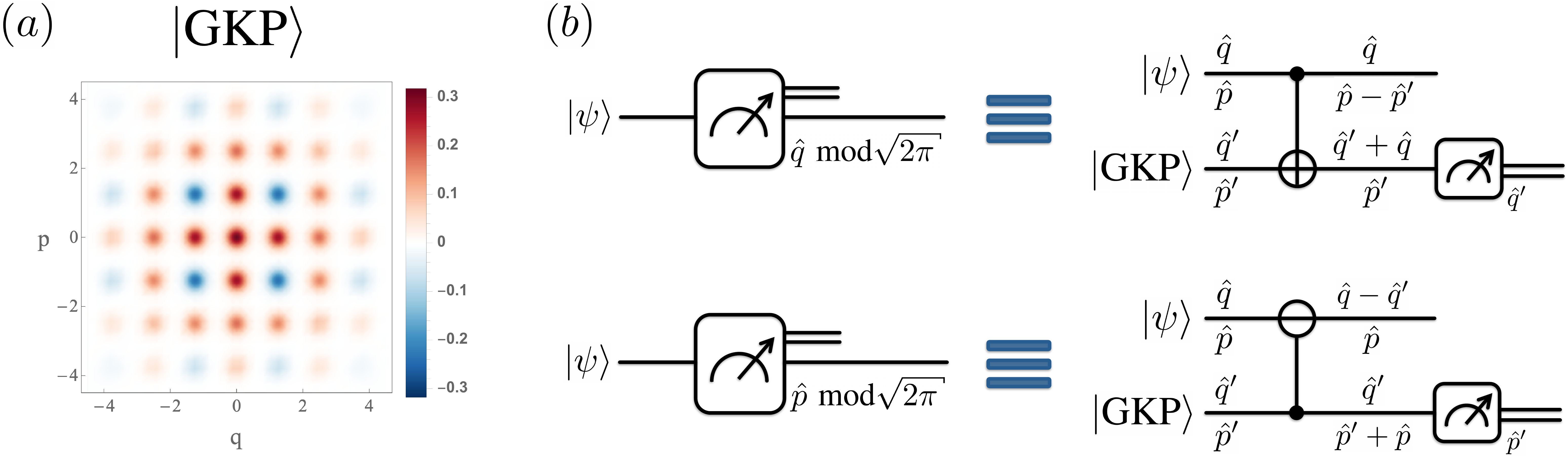}
\caption{[Fig.\ 1 in 	arXiv:1903.12615 (2019)] (a) An approximate GKP state with an average photon number $\bar{n}=5$. (b) Measurement of the position or momentum operator modulo $\sqrt{2\pi}$. The controlled-$\oplus$ and $\ominus$ symbols respectively represent the SUM and the inverse-SUM gates. }
\label{fig:GKP state and measurements}
\end{figure}

Observe that the canonical GKP state has an infinite average photon number as it is superpositions of infinitely many ($\sum_{n\in\mathbb{Z}}$) infinitely squeezed states ($|\hat{q}=\sqrt{2\pi}n\rangle$ or $|\hat{p}=\sqrt{2\pi}n\rangle$). Thus, the canonical GKP state is unphysical. However, one can define an approximate GKP state with a finite average photon number (or a finite squeezing) by applying a non-unitary operator $\exp[-\Delta^{2} \hat{n}]$ to the canonical GKP state and then normalizing the output state: $|\textrm{GKP}_{\Delta}\rangle \propto \exp[-\Delta \hat{n}]|\textrm{GKP}\rangle$ \cite{Gottesman2001}. In Fig.\ \ref{fig:GKP state and measurements}(a), we plot the Wigner function of the canonical GKP state with an average photon number $\bar{n}=5$. Note that negative peaks in the Wigner function indicate that the canonical GKP state is a non-Gaussian state \cite{Takagi2018}. In the interest of clarity, we only consider the ideal canonical GKP state when we present our main results. Adverse effects of the finite squeezing will be discussed later in the section.  

Clearly, the ability to measure the position and momentum operators modulo $\sqrt{2\pi}$ allows us to prepare the canonical GKP state. Remarkably, the converse is also true. That is, we can measure the quadrature operators modulo $\sqrt{2\pi}$ given GKP states and Gaussian operations as resources: As shown in Fig.\ \ref{fig:GKP state and measurements}(b), one can measure the position (momentum) operator modulo $\sqrt{2\pi}$ by using a canonical GKP state, the SUM (inverse-SUM) gate and the homodyne measurement of the position (momentum) operator. The SUM gate is a Gaussian operation and is defined as $\textrm{SUM}_{j\rightarrow k} \equiv \exp[ -i\hat{q}_{j}\hat{p}_{k}]$, which maps $\hat{q}_{k}$ to $\hat{q}_{k}+\hat{q}_{j}$. The inverse-SUM gate is defined as the inverse of the SUM gate. The canonical GKP state and the modulo simultaneous quadrature measurement are the key non-Gaussian resources of our oscillator-into-oscillators encoding schemes which we introduce below.    

\subsection{GKP-two-mode-squeezing code}

Here, we construct a non-Gaussian oscillator-into-oscillators code, namely the GKP-two-mode-squeezing code. Let $|\psi\rangle =\int dq\psi(q)|\hat{q}_{1} =q \rangle$ be an arbitrary bosonic state which we want to encode into two bosonic modes. We define the encoded state of the GKP-two-mode-squeezing code as follows:
\begin{align}
|\psi_{L}\rangle &=  \textrm{TS}_{1,2}(G) |\psi\rangle\otimes |\textrm{GKP}\rangle. \label{eq:logical states two-mode GKP-squeezed-repetition code}
\end{align}
Here, $|\textrm{GKP}\rangle$ is the canonical GKP state in the second mode and $\textrm{TS}_{1,2}(G)$ is the two-mode squeezing operation acting on the modes $1$ and $2$ with a gain $G\ge 1$ (hence the name of the code; see Fig.\ \ref{fig:GKP-two-mode-squeezing code}(a)). Since the logical information is encoded in the first mode before the two-mode squeezing, we refer to the first mode as the data mode and the second mode as the ancilla mode. In the Heisenberg picture, the two-mode squeezing operation $\textrm{TS}_{1,2}(G)$ transforms the quadrature operator $\boldsymbol{x}=(\hat{q}_{1},\hat{p}_{1},\hat{q}_{2},\hat{p}_{2})^{T}$ into $\boldsymbol{x'}=(\hat{q}'_{1},\hat{p}'_{1},\hat{q}'_{2},\hat{p}'_{2})^{T} = \boldsymbol{S_{\textrm{TS}}}(G)\boldsymbol{x}$, where the $4\times 4$ symplectic matrix $\boldsymbol{S_{\textrm{TS}}}(G)$ associated with $\textrm{TS}_{1,2}(G)$ is given by   
\begin{align}
\boldsymbol{S_{\textrm{TS}}}(G) = \begin{bmatrix}
\sqrt{G}\boldsymbol{I}_{2}&\sqrt{G-1}\boldsymbol{Z}_{2}\\
\sqrt{G-1}\boldsymbol{Z}_{2}&\sqrt{G}\boldsymbol{I}_{2}
\end{bmatrix}. 
\end{align}  
Here, $\boldsymbol{I}_{2}=\textrm{diag}(1,1)$ is the $2\times 2$ identity matrix and $\boldsymbol{Z}_{2}=\textrm{diag}(1,-1)$ is the Pauli Z matrix. Note that the gain $G$ can be chosen at will to optimize the performance of the error correction scheme.

\begin{figure}[t!]
\centering
\includegraphics[width=5.8in]{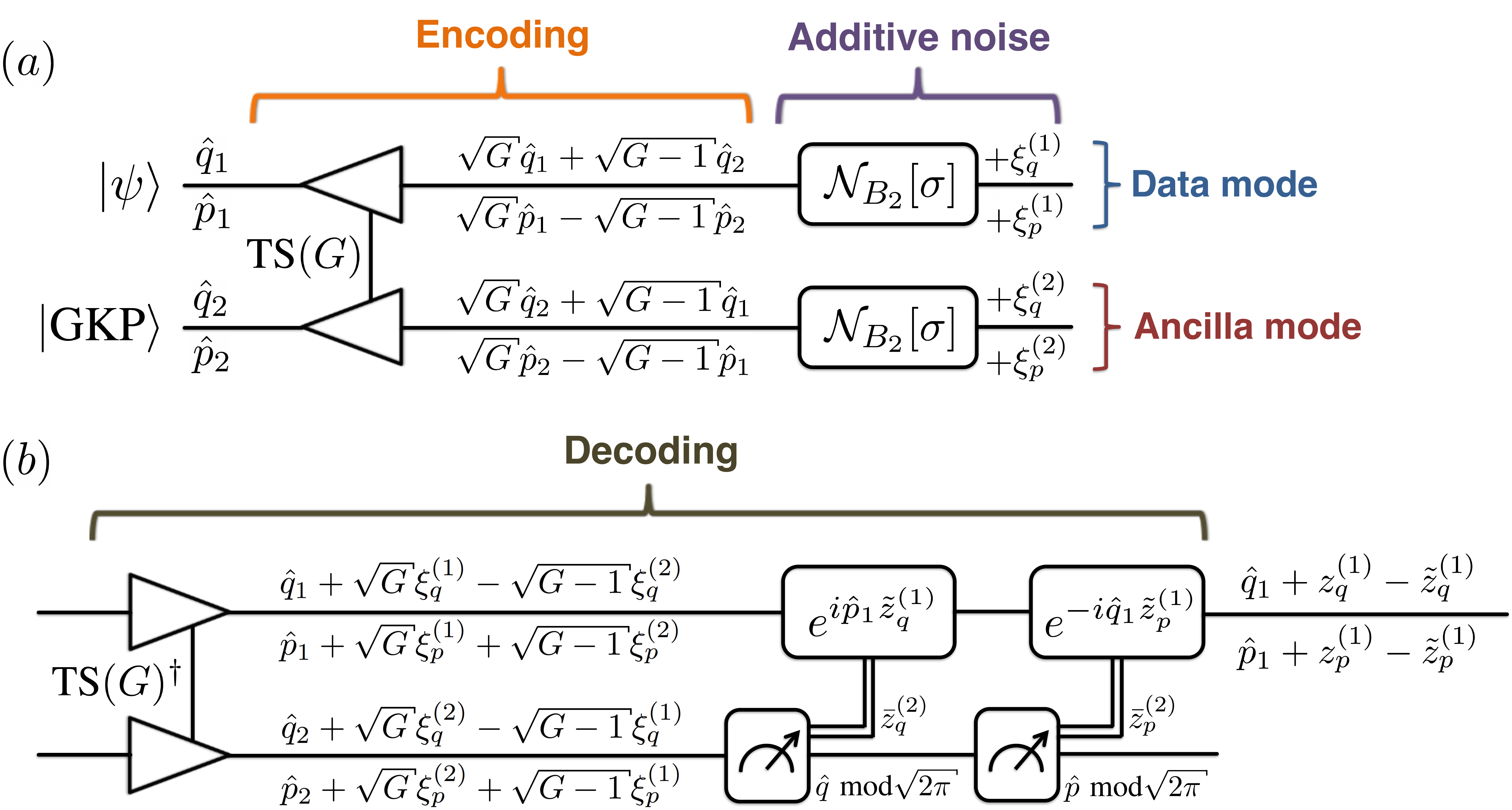}
\caption{[Fig.\ 4 in 	arXiv:1903.12615 (2019)] (a) Encoding circuit of the GKP-two-mode-squeezing code subject to independent and identically distributed additive Gaussian noise errors. (b) Decoding circuit of the GKP-two-mode-squeezing code. Note that the circuits for the measurements of the position and the momentum operators modulo $\sqrt{2\pi}$ (at the end of the decoding) are defined in Fig.\ \ref{fig:GKP state and measurements}(b).  }
\label{fig:GKP-two-mode-squeezing code}
\end{figure}

For the noise model, we consider the independent and identically distributed random shift errors, i.e., $\mathcal{N}^{(1)}_{B_{2}}[\sigma] \otimes \mathcal{N}^{(2)}_{B_{2}}[\sigma]$. In the Heisenberg picture, $\mathcal{N}^{(k)}_{B_{2}}[\sigma]$ adds Gaussian random noise $\xi_{q}^{(k)}$ and $\xi_{p}^{(k)}$ to the position and the momentum quadrature of the $k^{\textrm{th}}$ mode. Thus, the quadrature operator $\boldsymbol{\hat{x}'}$ is further transformed via the additive Gaussian noise error into $\boldsymbol{\hat{x}''} = \boldsymbol{\hat{x}'} + \boldsymbol{\xi}$, where $\boldsymbol{\xi} = (\xi_{q}^{(1)},\xi_{p}^{(1)},\xi_{q}^{(2)},\xi_{p}^{(2)})^{T}$ is the quadrature noise vector obeying $(\xi_{q}^{(1)},\xi_{p}^{(1)},\xi_{q}^{(2)},\xi_{p}^{(2)})\sim_{\textrm{iid}}\mathcal{N}(0,\sigma^{2})$ (see Fig.\ \ref{fig:GKP-two-mode-squeezing code}(a)).  

The decoding procedure (shown in Fig.\ \ref{fig:GKP-two-mode-squeezing code}(b)) starts with an application of the inverse of the encoding circuit $(\textrm{TS}_{1,2}(G))^{\dagger}$. Then, the quadrature operator is transformed into $
\boldsymbol{\hat{x}'''} = \boldsymbol{(S_{\textrm{TS}}}(G)\boldsymbol{)^{-1}}\boldsymbol{\hat{x}''} = \boldsymbol{\hat{x}} + \boldsymbol{z}$, where $\boldsymbol{z}\equiv (z_{q}^{(1)},z_{p}^{(1)},z_{q}^{(2)},z_{p}^{(2)})^{T}$ is the reshaped quadrature noise vector which is given by
\begin{align}
\boldsymbol{z} &= \boldsymbol{(S_{\textrm{TS}}}(G)\boldsymbol{)^{-1}}\boldsymbol{\xi}  = \begin{bmatrix}
\sqrt{G}\xi_{q}^{(1)}-\sqrt{G-1}\xi_{q}^{(2)} \\
\sqrt{G}\xi_{p}^{(1)}+\sqrt{G-1}\xi_{p}^{(2)} \\
\sqrt{G}\xi_{q}^{(2)}-\sqrt{G-1}\xi_{q}^{(1)}  \\
\sqrt{G}\xi_{p}^{(2)}+\sqrt{G-1}\xi_{p}^{(1)}
\end{bmatrix} \equiv \begin{bmatrix}
z_{q}^{(1)}\\
z_{p}^{(1)}\\
z_{q}^{(2)}\\
z_{p}^{(2)}
\end{bmatrix}. \label{eq:reshaped quadrature noises two-mode GKP-squeezed-repetition code}
\end{align}
The role of the two-mode squeezing operations in the encoding and the decoding circuits is clear by now: They transform uncorrelated additive noise $\boldsymbol{
\xi}$ into correlated additive noise $\boldsymbol{z}$. This means that after noise reshaping, we can extract useful information about the reshaped data quadrature noise $z_{q}^{(1)}$, $z_{p}^{(1)}$ by measuring only the reshaped ancilla quadrature noise $z_{q}^{(2)}$, $z_{p}^{(2)}$. Importantly, the encoded logical information in the data mode is not revealed through this process because the data mode needs not be measured.       

Note that we need to measure both the position and momentum noise in the ancilla mode. This is precisely the reason why we measure both the position and momentum operators of the ancilla mode modulo $\sqrt{2\pi}$ at the end of the decoding circuit (see Fig.\ \ref{fig:GKP-two-mode-squeezing code}(b)). Note that $\hat{q}_{2}=\hat{p}_{2}=0$ modulo $\sqrt{2\pi}$ holds because the ancilla mode is initialized to the canonical GKP state. Thus, measuring the output quadrature operators $\hat{q}_{2}''' = \hat{q}_{2}+z_{q}^{(2)}$ and $\hat{p}_{2}'''=\hat{p}_{2} + z_{p}^{(2)}$ modulo $\sqrt{2\pi}$ is equivalent to measuring just the reshaped ancilla quadrature noise $z_{q}^{(2)}$ and $z_{p}^{(2)}$ modulo $\sqrt{2\pi}$. 

Based on the outcomes of the simultaneous measurement of the ancilla quadrature noise modulo $\sqrt{2\pi}$, we estimate that the ancilla quadrature noise $z_{q}^{(2)}$ and $z_{p}^{(2)}$ are the smallest ones that are compatible with the modular measurement outcomes, i.e.,  
\begin{align}
\bar{z}_{q}^{(2)} &= R_{\sqrt{2\pi}}(z_{q}^{(2)}) \,\,\,\textrm{and}\,\,\,\bar{z}_{p}^{(2)} = R_{\sqrt{2\pi}}(z_{p}^{(2)}),  \label{eq:estimates of the ancilla quadrature noises two-mode GKP-squeezed-repetition code}
\end{align}  
where $R_{s}(z)\equiv z-n^{\star}(z)s$ and $n^{\star}(z)\equiv \textrm{argmin}_{n\in\mathbb{Z}}|z-ns|$. More explicitly, $R_{s}(z)$ equals a displaced sawtooth function with an amplitude and period $s$ and is given by $R_{s}(z) = z$ if $z\in[-s/2,s/2]$. We then further estimate that the reshaped data quadrature noise $z_{q}^{(1)}$ and $z_{p}^{(1)}$ are 
\begin{align}
\tilde{z}_{q}^{(1)} &=  - \frac{2\sqrt{G(G-1)}}{2G-1} \bar{z}_{q}^{(2)}, 
\nonumber\\
\tilde{z}_{p}^{(1)} &=   \frac{2\sqrt{G(G-1)}}{2G-1} \bar{z}_{p}^{(2)},  \label{eq:estimates of the data quadrature noises two-mode GKP-squeezed-repetition code}
\end{align}   
which are obtained from the maximum likelihood estimation as detailed in Subsection \ref{subsection:Detailed analysis of the GKP-two-mode squeezing code}. 

The decoding operation is simply to remove the estimated noise in the data mode by applying the counter displacement operations $\exp[i\hat{p}_{1}\tilde{z}_{q}^{(1)}]$ and $\exp[-i\hat{q}_{1}\tilde{z}_{p}^{(1)}]$ to the data mode (see Fig.\ \ref{fig:GKP-two-mode-squeezing code}(b)). As a result, we are left with the following logical position and momentum quadrature noise
\begin{align}
\xi_{q} &\equiv z_{q}^{(1)}-\tilde{z}_{q}^{(1)} = z_{q}^{(1)} +\frac{2\sqrt{G(G-1)}}{2G-1} R_{\sqrt{2\pi}}(z_{q}^{(2)}), 
\nonumber\\
\xi_{p} &\equiv z_{p}^{(1)}-\tilde{z}_{p}^{(1)} = z_{p}^{(1)} - \frac{2\sqrt{G(G-1)}}{2G-1} R_{\sqrt{2\pi}}(z_{p}^{(2)}). \label{eq:logical quadrature noises two-mode GKP-squeezed-repetition code}
\end{align} 
Then, the variance $(\sigma_{q})^{2} = (\sigma_{p})^{2} = (\sigma_{L})^{2}$ of the output logical quadrature noise $\xi_{q}$ and $\xi_{p}$ is given by 
\begin{align}
(\sigma_{L})^{2} = \frac{\sigma^{2}}{2G-1} + \sum_{n\in\mathbb{Z}} \frac{ 4G(G-1) }{ (2G-1)^{2} } 2\pi n^{2} \times q_{n}(\sigma). 
\end{align}   
(See Subsection \ref{subsection:Detailed analysis of the GKP-two-mode squeezing code} for the proof.) Here, $q_{n}(\sigma)$ is defined as
\begin{align}
q_{n}(\sigma) &\equiv \int_{(n-\frac{1}{2})\sqrt{2\pi}}^{(n+\frac{1}{2})\sqrt{2\pi}} dz p[\sqrt{2G-1}\sigma](z) 
\end{align}
where $p[\sigma](z) \equiv \frac{1}{\sqrt{2\pi\sigma^{2}}} \exp[ -\frac{z^{2}}{2\sigma^{2}} ]$ is the probability density function of the Gaussian normal distribution $\mathcal{N}( 0, \sigma^{2} )$.  

Recall that we can freely choose the gain $G$ to optimize the performance of the GKP-two-mode-squeezing code. Here, we choose $G$ such that the standard deviation of the output logical quadrature noise $\sigma_{L}$ is minimized. In Fig.\ \ref{fig:performance of the two-mode GKP-squeezed-repetition code}, we plot the minimum standard deviation of the output logical quadrature noise $\sigma_{L}^{\star}$ (see Fig.\ \ref{fig:performance of the two-mode GKP-squeezed-repetition code}(a)) and the optimal gain $G^{\star}$ (see Fig.\ \ref{fig:performance of the two-mode GKP-squeezed-repetition code}(b)) that achieves the minimum output standard deviation given an input noise standard deviation $\sigma$. These optimal values are obtained via a brute-force numerical optimization. Note that in Fig.\ \ref{fig:performance of the two-mode GKP-squeezed-repetition code}(b), we show the strength of the required single-mode squeezing operations to achieve the optimal gain $G^{\star}$ in the unit of decibel (i.e., $20\log_{10}\lambda^{\star}$ where $\lambda^{\star} = \sqrt{G^{\star}}+\sqrt{G^{\star}-1}$; see Subsection \ref{subsection:Detailed analysis of the GKP-two-mode squeezing code}). 

\begin{figure}[t!]
\centering
\includegraphics[width=4.2in]{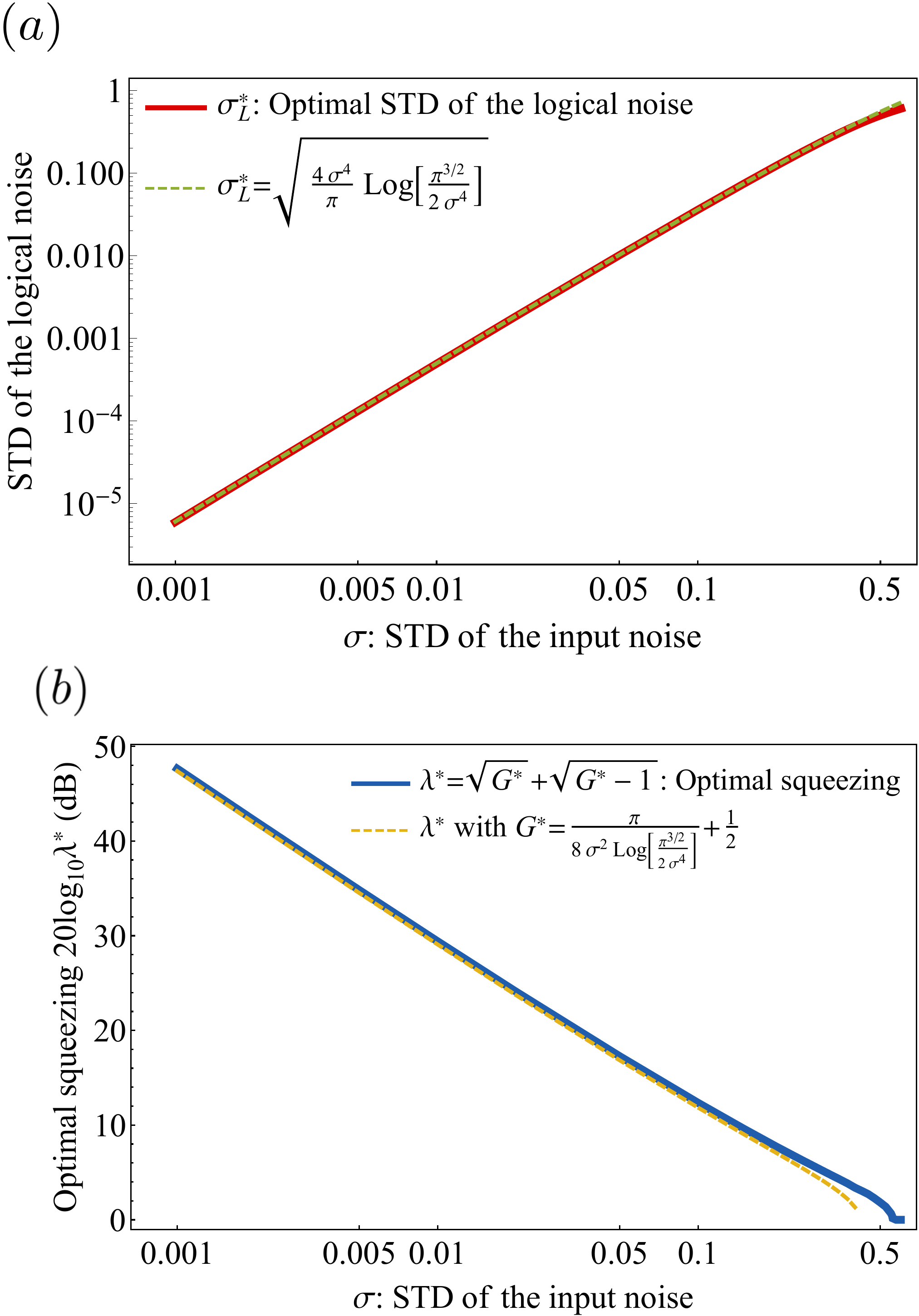}
\caption{[Fig.\ 5 in 	arXiv:1903.12615 (2019)] (a) The minimum standard deviations of the output logical quadrature noise $\sigma_{q}=\sigma_{p}=\sigma_{L}^{\star}$ as a function of the input standard deviation $\sigma$ for the GKP-two-mode-squeezing code and (b) the optimal two-mode squeezing gain $G^{\star}$ that achieves $\sigma_{L}^{\star}$, translated to the required single-mode squeezing in the unit of decibel $20\log_{10}\lambda^{\star}$ where $\lambda^{\star} \equiv \sqrt{G^{\star}} + \sqrt{G^{\star}-1}$. The green dashed line in (a) represents $\sigma_{L}^{\star} = \frac{2\sigma^{2}}{\sqrt{\pi}}\sqrt{\log_{e}[\frac{\pi^{3/2}}{2\sigma^{4}}]}$ and the yellow dashed line in (b) represents $G^{\star} = \frac{\pi}{8\sigma^{2}}(\log_{e}[\frac{\pi^{3/2}}{2\sigma^{4}}])^{-1} +\frac{1}{2}$.   }
\label{fig:performance of the two-mode GKP-squeezed-repetition code}
\end{figure}

From the numerical optimization, we find that for $\sigma\ge \sigma_{c} \equiv 0.558$, the optimal gain $G^{\star}$ is trivially given by $G^{\star}=1$ and thus the GKP-two-mode-squeezing code cannot reduce the noise standard deviation: $\sigma_{L}^{\star} = \sigma$. On the other hand, if the input noise is small enough, i.e., $\sigma < \sigma_{c}= 0.558$, the optimal gain $G^{\star}$ is strictly larger than $1$ and the minimum standard deviation of the output quadrature noise $\sigma_{L}$ can be made smaller than the standard deviation of the input quadrature noise $\sigma$: $\sigma_{L}^{\star}< \sigma$. Moreover in the $\sigma \ll 1$ regime, we analytically find that the optimal gain $G^{\star}$ is asymptotically given by  
\begin{align}
G^{\star} \xrightarrow{\sigma\ll 1} \frac{\pi}{8\sigma^{2}} \Big{(} \log_{e}\Big{[}\frac{\pi^{3/2}}{2\sigma^{4}}\Big{]} \Big{)}^{-1} +\frac{1}{2},  \label{eq:optimal G asymptotic}
\end{align} 
and the optimal standard deviation of the logical quadrature noise $\sigma_{L}^{\star}$ is given by
\begin{align}
\sigma_{L}^{\star} \xrightarrow{\sigma \ll 1} \frac{2\sigma^{2}}{\sqrt{\pi}} \sqrt{ \log_{e}\Big{[}\frac{\pi^{3/2}}{2\sigma^{4}}\Big{]} }.  \label{eq:optimal STD asymptotic}
\end{align}
(See Subsection \ref{subsection:Detailed analysis of the GKP-two-mode squeezing code} for the detailed derivation.) As shown in Fig.\ \ref{fig:performance of the two-mode GKP-squeezed-repetition code}, these asymptotic expressions agree well with the exact numerical results in the small $\sigma$ regime. Note that the asymptotic expression in Eq.\ \eqref{eq:optimal STD asymptotic} implies that the minimum standard deviation of the output quadrature noise $\sigma_{L}^{\star}$ decreases quadratically as $\sigma$ decreases (i.e., $\sigma_{L}^{\star}\propto \sigma^{2}$) modulo a small logarithmic correction.

Excitation loss errors with external thermal noise are described by Gaussian thermal-loss channels. In general, Gaussian thermal-loss channels can be converted via a quantum-limited amplification to an additive Gaussian noise channel \cite{Albert2018,Noh2019}. For instance, the bosonic pure-loss channel with loss probability $\gamma \in [0,1]$ can be converted to an additive Gaussian noise channel $\mathcal{N}_{B_{2}}[\sigma]$ with $\sigma = \sqrt{\gamma}$. Hence, the GKP-two-mode-squeezing code can also handle the excitation loss errors because we can simply convert the loss errors into the additive noise errors and then apply the same decoding scheme presented above. (Note, however, that this amplification decoding may not be the optimal decoding strategy. See Chapters \ref{chapter:Benchmarking and optimizing single-mode bosonic codes} and \ref{chapter:Achievable communication rates with bosonic codes} for more details.) 

Assuming the amplification decoding, the critical value of the standard deviation $\sigma_{c} = 0.558$ corresponds to the critical loss probability $\gamma_{c} = (\sigma_{c})^{2} = 0.311$ in the case of the pure-excitation loss channel. Thus, the GKP-two-mode-squeezing code helps when the loss probability is below $31.1\%$. For example, consider the pure-loss channel with $1\%$ loss probability (i.e., $\gamma = 0.01$ and $\sigma= \sqrt{\gamma} =0.1$). Then, the optimal gain is given by $G^{\star} = 4.806$ which requires $20\log_{10}\lambda^{\star} = 12.35$dB single-mode squeezing operations. Also in this case, the resulting standard deviation of the output noise is given by $\sigma_{L}^{\star} = 0.036$ which corresponds to the loss probability $0.13\%$. This corresponds to a QEC ``gain'' for the protocol of $1/0.13\simeq 7.7$ in terms of the loss probability and $0.1/0.036\simeq 2.8$ in terms of displacement errors.  

\subsection{Adverse effects due to finite GKP squeezing} 

Here, we discuss experimental realization of the GKP-two-mode-squeezing code and the effects of realistic imperfections. Note that the only required non-Gaussian resource for implementing the GKP-two-mode-squeezing code is preparation of the canonical GKP states. While Gaussian operations are readily available in many realistic bosonic systems, preparing a canonical GKP state is not strictly possible because it would require infinite squeezing. Recently, however, finitely-squeezed approximate GKP states have been realized in a trapped ion system \cite{Fluhmann2018,Fluhmann2019,Fluhmann2019b} by using a heralded preparation scheme with post-selection \cite{Travaglione2002} and in a circuit QED system by using a deterministic scheme \cite{Terhal2016,Campagne2019}. Thus, the GKP-two-mode-squeezing code can in principle be implemented in the state-of-the-art quantum computing platforms.  

Imperfections in realistic GKP states such as finite squeezing will add additional quadrature noise to the system. Therefore in near-term experiments, the performance of the GKP-two-mode-squeezing code will be mainly limited by the finite squeezing of the approximate GKP states. Indeed, we show below that a non-trivial QEC gain $\sigma^{2} / (\sigma_{L}^{\star})^{2} > 1$ is achievable with the GKP-two-mode-squeezing code only when the supplied GKP states have a squeezing larger than $11.0$dB. On the other hand, the squeezing of the experimentally realized GKP states ranges from $5.5$dB to $9.5$dB \cite{Fluhmann2019,Campagne2019}. In this regard, we emphasize that our oscillator encoding scheme is compatible with non-deterministic GKP state preparation schemes. This is because the required GKP states can be prepared offline and then supplied to the error correction circuit in the middle of the decoding procedure (similar to the magic state injection for the qubit-based universal quantum computation \cite{Bravyi2005}). Thus in near-term experiments, it will be more advantageous to sacrifice the success probability of the GKP state preparation schemes and aim to prepare a GKP state of higher quality (with a squeezing larger than the critical value $11.0$dB) by using post-selection.

In general, the imperfections in GKP states may be especially detrimental to a GKP-two-mode-squeezing code involving a large squeezing parameter. This is because such imperfections may be significantly amplified by the large squeezing operations. Indeed, the optimal gain $G^{\star}$ is asymptotically given by $G^{\star} \propto 1/\sigma^{2}$ in the $\sigma\ll 1$ limit. Therefore, if the standard deviation of the input noise is very small, we indeed have a huge gain parameter $G^{\star} \gg 1$ (or $\lambda^{\star} = \sqrt{G^{\star}}+\sqrt{G^{\star}-1} \gg 1$). However, we explain in detail below that we have designed the GKP-two-mode-squeezing code very carefully so that any imperfections in the GKP states are not amplified by the large squeezing operations. 

With these potential issues in mind, let us now analyze the adverse effects of the finite squeezing in a rigorous and quantitative manner. Recall that an approximate GKP state with a finite squeezing can be modeled by $|\textrm{GKP}_{\Delta}\rangle \propto \exp[-\Delta^{2} \hat{n}]|\textrm{GKP}\rangle$. As shown in Ref.\ \cite{Noh2020} (see also Section \ref{section:Translation-symmetric bosonic codes}), one can convert the finitely-squeezed GKP state $|\textrm{GKP}_{\Delta}\rangle$ via a noise twirling into 
\begin{align}
\mathcal{N}_{B_{2}}[\sigma_{\textrm{gkp}}] ( |\textrm{GKP}\rangle\langle \textrm{GKP}| ), \label{eq:noisy GKP state after the twirling} 
\end{align}
i.e., an ideal canonical GKP state corrupted by an incoherent random shift error $\mathcal{N}_{B_{2}}[\sigma_{\textrm{gkp}}]$. Here, $\sigma_{\textrm{gkp}}^{2} = (1-e^{-\Delta^{2}}) / (1+e^{-\Delta^{2}})$ is the variance of the additive noise associated with the finite GKP squeezing. The noise standard deviation $\sigma_{\textrm{gkp}}$ characterizes the width of each peak in the Wigner function of an approximate GKP state. The GKP squeezing is then defined as $s_{\textrm{gkp}}\equiv -10\log_{10}(2\sigma_{\textrm{gkp}}^{2})$. The GKP squeezing $s_{\textrm{gkp}}$ quantifies how much an approximate GKP state is squeezed in both the position and the momentum quadrature in comparison to the vacuum noise variance $1/2$.  

\begin{figure}[t!]
\centering
\includegraphics[width=5.8in]{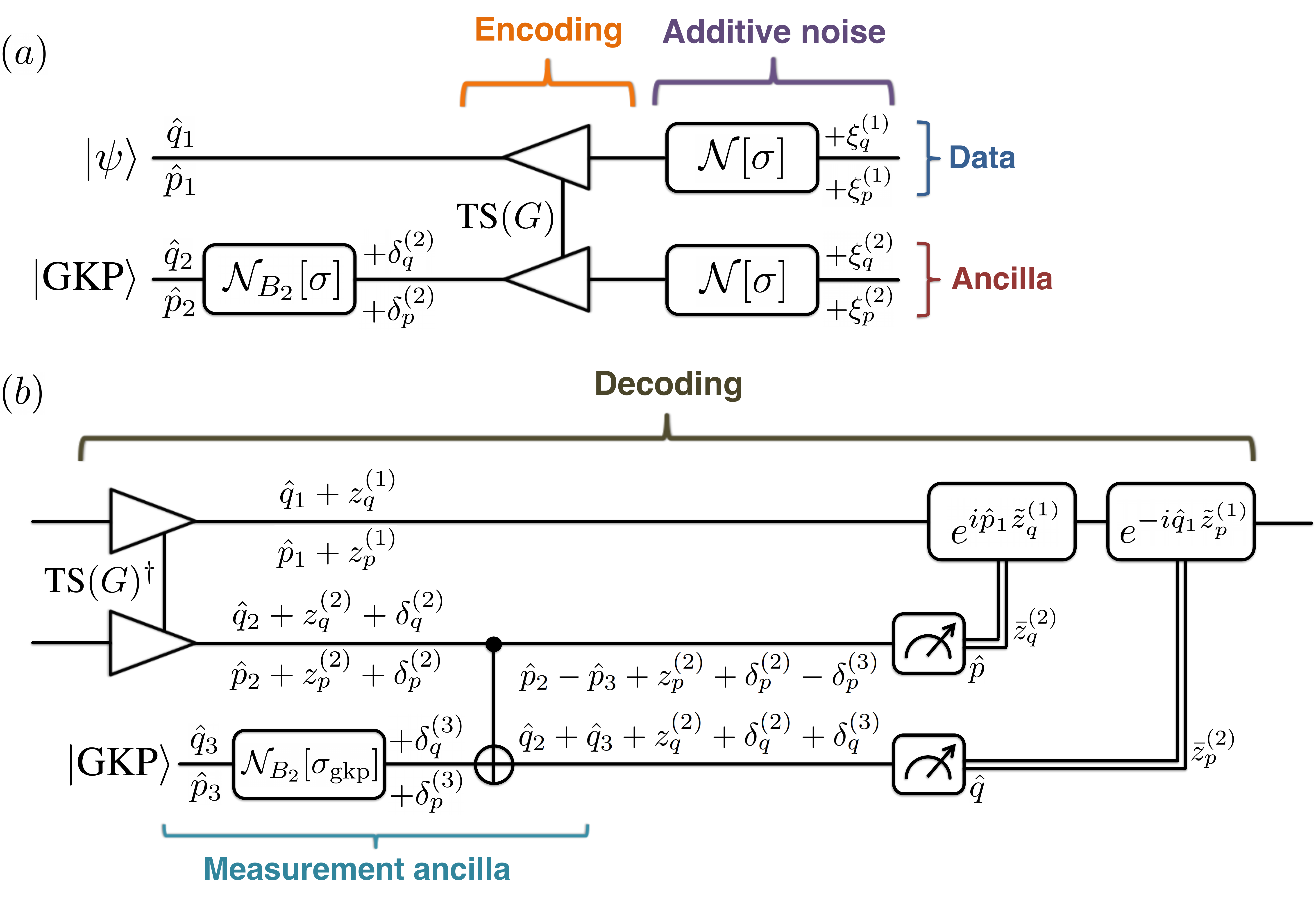}
\caption{[Fig.\ 8 in 	arXiv:1903.12615 (2019)] (a) Encoding circuit of the GKP-two-mode-squeezing code subject to independent and identically distributed additive Gaussian noise errors. The input GKP state in the ancilla mode is assumed to be a noisy canonical GKP state with a noise standard deviation $\sigma_{\textrm{gkp}}$. (b) Decoding circuit of the GKP-two-mode-squeezing code. Note that the simultaneous position and momentum quadrature measurement modulo $\sqrt{2\pi}$ is implemented by using a third ancilla mode (i.e., the measurement ancilla mode) initialized to a noisy GKP state with a noise standard deviation $\sigma_{\textrm{gkp}}$.} 
\label{fig:GKP-two-mode-squeezing code with noisy GKP states}
\end{figure}

In Fig.\ \ref{fig:GKP-two-mode-squeezing code with noisy GKP states}, we present the full circuit for the implementation of the GKP-two-mode-squeezing code. Note that the third mode (or the measurement mode) in the decoding scheme is introduced to simultaneously measure the position and momentum operators of the ancilla mode modulo $\sqrt{2\pi}$. That is, we consume one GKP state to perform the simultaneous and modular position and momentum measurements. We remark that we would have consumed two GKP states for the simultaneous and modular quadrature measurements if we were to use the measurement circuits in Fig.\ \ref{fig:GKP state and measurements}(b) (i.e., one for the modular position measurement and the other for the modular momentum measurement). While this scheme certainly works, it is not the most efficient strategy. This is because the measurement circuits in Fig.\ \ref{fig:GKP state and measurements}(b) are for non-destructive measurements. While the first measurement (e.g., the modular position measurement) has to be performed in a non-destructive way, the following measurement (e.g., the modular momentum measurement) can be done in a destructive way since we no longer need the quantum state in the ancilla mode and instead only need the classical measurement outcomes $\bar{z}_{q}^{(2)}$ and $\bar{z}_{p}^{(2)}$. This is the reason why we simply measure the momentum quadrature of the second mode (i.e., the ancilla mode) in a destructive way after the modular position measurement (see Fig.\ \ref{fig:GKP-two-mode-squeezing code with noisy GKP states}(b)). Such a non-Gaussian resource saving is especially crucial in the regime where the finite squeezing of an approximate GKP state is the limiting factor.

Thanks to the resource saving described above, we only need to supply two GKP states to implement the GKP-two-mode-squeezing code (one in the input of the ancilla mode and the other for the simultaneous and modular ancilla quadrature measurements). We assume that these two GKP states are corrupted by an additive Gaussian noise channel $\mathcal{N}[\sigma_{\textrm{gkp}}]$, i.e., $(\delta_{q}^{(2)},\delta_{p}^{(2)},\delta_{q}^{(3)},\delta_{p}^{(3)}) \sim_{\textrm{iid}}\mathcal{N}(0, \sigma_{\textrm{gkp}}^{2})$ (see Eq.\ \eqref{eq:noisy GKP state after the twirling} and Fig.\ \ref{fig:GKP-two-mode-squeezing code with noisy GKP states}). Due to this additional noise associated with the finite squeezing of the GKP states, the estimated reshaped ancilla quadrature noise is corrupted as follows: 
\begin{align}
\bar{z}_{q}^{(2)} &= R_{\sqrt{2\pi}}( z_{q}^{(2)} + \xi_{q}^{(\textrm{gkp})} ) , 
\nonumber\\
\bar{z}_{p}^{(2)} &= R_{\sqrt{2\pi}}( z_{p}^{(2)} + \xi_{p}^{(\textrm{gkp})} ) , \label{eq:estimated ancilla quadrature noise noisy GKP states}
\end{align}
Here, $\xi_{q}^{(\textrm{gkp})} \equiv\delta_{q}^{(2)} + \delta_{q}^{(3)} $ and $\xi_{p}^{(\textrm{gkp})} \equiv \delta_{p}^{(2)} - \delta_{p}^{(3)}$ are the additional noise due to the finite GKP squeezing and follow $(\xi_{q}^{(\textrm{gkp})}, \xi_{p}^{(\textrm{gkp})} ) \sim_{\textrm{iid}} \mathcal{N}( 0, 2\sigma_{\textrm{gkp}}^{2} )$. Such additional noise will then be propagated to the data mode through the miscalibrated counter displacement operations based on noisy estimates. In the presence of additional GKP noise, the sizes of the optimal counter displacements $\exp[i\hat{p}_{1}\tilde{z}_{q}^{(1)}]$ and $\exp[-i\hat{q}_{1}\tilde{z}_{p}^{(1)}]$ are given by 
\begin{align}
\tilde{z}_{q}^{(1)} &=  - \frac{2\sqrt{G(G-1)} \sigma^{2} }{(2G-1)\sigma^{2} + 2\sigma_{\textrm{gkp}}^{2} } \bar{z}_{q}^{(2)} \xrightarrow{G\gg 1} -\bar{z}_{q}^{(2)}, 
\nonumber\\
\tilde{z}_{p}^{(1)} &=   \frac{2\sqrt{G(G-1)} \sigma^{2} }{(2G-1)\sigma^{2} + 2\sigma_{\textrm{gkp}}^{2} } \bar{z}_{p}^{(2)} \xrightarrow{G\gg 1} \bar{z}_{p}^{(2)} 
\end{align}
and do not explicitly depend on $G$ in the $G \gg 1$ limit. Therefore, the additional GKP noise $\xi_{q}^{(\textrm{gkp})}$ and $\xi_{p}^{(\textrm{gkp})}$ will simply be added to the data quadrature operators without being amplified by the large gain parameter $G \gg 1$. This absence of the noise amplification is a critically important feature of our scheme and is generally not the case for a generic GKP-stabilizer code involving large squeezing operations. 

\begin{figure}[t!]
\centering
\includegraphics[width=4.2in]{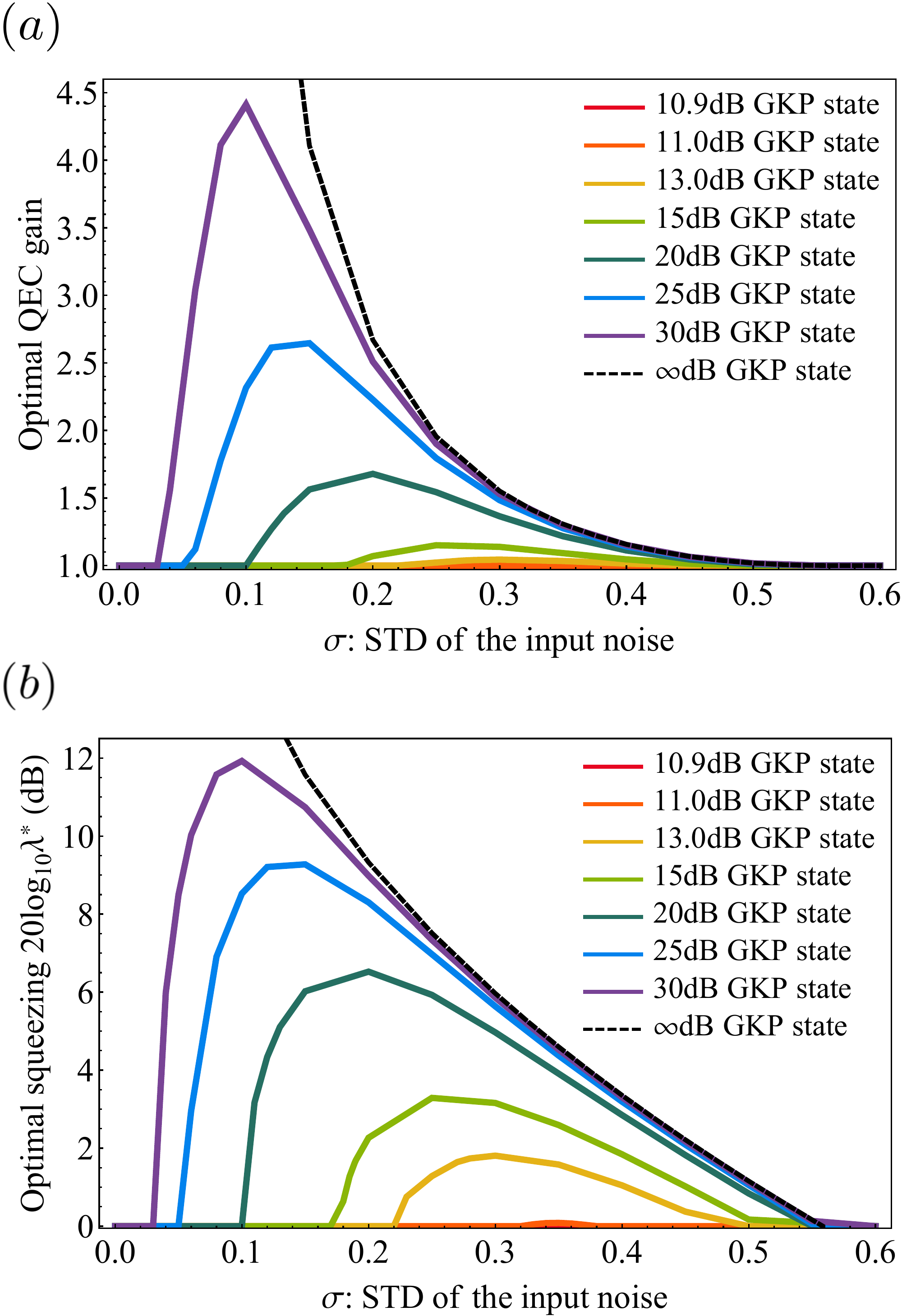}
\caption{[Fig.\ 9 in 	arXiv:1903.12615 (2019)] (a) The optimal QEC gain $\sigma^{2} / (\sigma_{L}^{\star})^{2}$ as a function of input noise standard deviation $\sigma$ for various values of the GKP squeezing $s_{\textrm{gkp}} = -10\log_{10}(2\sigma_{\textrm{gkp}}^{2})$ ranging from $12.8$dB to $30$dB and (b) the optimal two-mode squeezing gain $G^{\star}$, translated to the required single-mode squeezing in the unit of decibel $20\log_{10}\lambda^{\star}$ where $\lambda^{\star} \equiv \sqrt{G^{\star}} +\sqrt{G^{\star}-1 }$. The non-trivial QEC gain $\sigma^{2} / (\sigma_{L}^{\star})^{2}>1$ is achievable only when the squeezing of the supplied approximate GKP states is larger than the critical squeezing $11.0$dB. The dashed black lines represent the asymptotic results for the infinitely squeezed canonical GKP states which are shown in Fig.\ \ref{fig:performance of the two-mode GKP-squeezed-repetition code}.  }
\label{fig:performance of the two-mode GKP-two-mode-squeezing code noisy GKP states}
\end{figure}

In Subsection \ref{subsection:Detailed analysis of the GKP-two-mode squeezing code}, we provide a through analysis of the adverse effects of the finitely-squeezed GKP states. In particular, we derive the variance of the output logical quadrature noise $(\sigma_{L})^{2}$ as a function of the input noise standard deviation $\sigma$, the GKP noise standard deviation $\sigma_{\textrm{gkp}}$, and the gain of the two-mode squeezing $G$. Similarly as above, we optimize the gain $G$ of the two-mode squeezing to minimize the output logical noise standard deviation $\sigma_{L}$ for given $\sigma$ and $\sigma_{\textrm{gkp}}$.   

In Fig.\ \ref{fig:performance of the two-mode GKP-two-mode-squeezing code noisy GKP states}(a), we plot the maximum achievable QEC gain $\sigma^{2} / (\sigma_{L}^{\star})^{2}$ as a function of the input noise standard deviation $\sigma$ for various values of the GKP squeezing ranging from $10.9$dB to $30$dB. In Fig.\ \ref{fig:performance of the two-mode GKP-two-mode-squeezing code noisy GKP states}(b), we plot the optimal gain $G^{\star}$ of the two-mode squeezing, translated to the required single-mode squeezing in the unit of decibel. We observe that the non-trivial QEC gain $\sigma^{2} / (\sigma_{L}^{\star})^{2}>1$ can be achieved only when the supplied GKP states have a squeezing larger than the critical value $11.0$dB. Also, when the supplied GKP states have a squeezing of $30$dB, the maximum QEC gain is given by $\sigma^{2} / (\sigma_{L}^{\star})^{2} = 4.41$, which is achieved when $\sigma = 0.1$. For comparison, the QEC gain at the same input noise standard deviation $\sigma=0.1$ is given by $\sigma^{2} / (\sigma_{L}^{\star})^{2} = 7.7$ when the ideal canonical GKP states are used to implement the GKP-two-mode-squeezing code (see above). The fact that these two values ($4.41$ verses $7.7$) are close to each other is an indicative of the fact that the additional GKP noise is not catastrophically amplified by the large ($12.3$dB) single-mode squeezing operations needed in this regime.    

\subsection{Detailed analysis of the GKP-two-mode squeezing code}
\label{subsection:Detailed analysis of the GKP-two-mode squeezing code}

Here, we first explain the underlying reasons behind our choice of the estimates $\tilde{z}_{q}^{(1)}$ and $\tilde{z}_{p}^{(1)}$ in Eq.\ \eqref{eq:estimates of the data quadrature noises two-mode GKP-squeezed-repetition code} for the GKP-two-mode-squeezing code. Note that the covariance matrix of the reshaped noise $\boldsymbol{z} = (z_{q}^{(1)},z_{p}^{(1)},z_{q}^{(2)},z_{p}^{(2)})^{T}$ is given by 
\begin{align}
V_{\boldsymbol{z}} = \sigma^{2}\begin{bmatrix}
(2G-1)\boldsymbol{I}_{2}&-2\sqrt{G(G-1)}\boldsymbol{Z}_{2}\\
-2\sqrt{G(G-1)}\boldsymbol{Z}_{2}&(2G-1)\boldsymbol{I}_{2}
\end{bmatrix}. 
\end{align}
For now, let us ignore the fact that we can measure $z_{q}^{(2)}$ and $z_{p}^{(2)}$ only modulo $\sqrt{2\pi}$ and instead assume that we know their exact values. Note that $z_{q}^{(1)}$ is only correlated with $z_{q}^{(2)}$, whereas $z_{p}^{(1)}$ is only correlated with $z_{p}^{(2)}$. Consider the estimates of the form $\bar{z}_{q}^{(1)} = c_{q}z_{q}^{(2)}$ and $\bar{z}_{p}^{(1)} = c_{p}z_{p}^{(2)}$, where $c_{q}$ and $c_{p}$ are constants. We choose $c_{q}$ and $c_{p}$ such that the variances of $z_{q}^{(1)}-\bar{z}_{q}^{(1)}$ and $z_{p}^{(1)}-\bar{z}_{p}^{(1)}$ are minimized: Since $\textrm{Var}(z_{q}^{(1)}-\bar{z}_{q}^{(1)})$ and $\textrm{Var}(z_{p}^{(1)}-\bar{z}_{p}^{(1)})$ are given by
\begin{align}
\textrm{Var}(z_{q}^{(1)}-\bar{z}_{q}^{(1)}) &= \textrm{Var}(z_{q}^{(1)}) - 2c_{q} \cdot \textrm{Cov}(z_{q}^{(1)},z_{q}^{(2)})  + c_{q}^{2} \textrm{Var}(z_{q}^{(2)}) , 
\nonumber\\
\textrm{Var}(z_{p}^{(1)}-\bar{z}_{p}^{(1)}) &= \textrm{Var}(z_{p}^{(1)}) - 2c_{p} \cdot \textrm{Cov}(z_{p}^{(1)},z_{p}^{(2)})    + c_{p}^{2} \textrm{Var}(z_{p}^{(2)}),  
\end{align} 
they are minimized when 
\begin{align}
c_{q} &= \frac{\textrm{Cov}(z_{q}^{(1)},z_{q}^{(2)})}{\textrm{Var}(z_{q}^{(2)})} = -\frac{2\sqrt{G(G-1)}}{2G-1} ,
\nonumber\\
c_{q} &= \frac{\textrm{Cov}(z_{p}^{(1)},z_{p}^{(2)})}{\textrm{Var}(z_{p}^{(2)})} = \frac{2\sqrt{G(G-1)}}{2G-1}. 
\end{align} 
Therefore, if both $z_{q}^{(2)}$ and $z_{p}^{(2)}$ are precisely known, the optimal estimates of $z_{q}^{(1)}$ and $z_{p}^{(1)}$ are given by 
\begin{align}
\bar{z}_{q}^{(1)} &=  -\frac{2\sqrt{G(G-1)}}{2G-1} z_{q}^{(2)}, 
\nonumber\\
\bar{z}_{p}^{(1)} &=  \frac{2\sqrt{G(G-1)}}{2G-1} z_{p}^{(2)}. 
\end{align}  
Since, however, we can only measure $z_{q}^{(2)}$ and $z_{p}^{(2)}$ modulo $\sqrt{2\pi}$, we replace $z_{q}^{(2)}$ and $z_{p}^{(2)}$ by $\bar{z}_{q}^{(2)} = R_{\sqrt{2\pi}}(z_{q}^{(2)})$ and $\bar{z}_{p}^{(2)} = R_{\sqrt{2\pi}}(z_{p}^{(2)})$ and get the estimates $\tilde{z}_{q}^{(1)}$ and $\tilde{z}_{p}^{(1)}$ in Eq.\ \eqref{eq:estimates of the data quadrature noises two-mode GKP-squeezed-repetition code}. 

Now we provide explicit expression for the probability density functions of the logical quadrature noise $\xi_{q}$ and $\xi_{p}$ for the GKP-two-mode-squeezing code. Recall Eq.\ \eqref{eq:logical quadrature noises two-mode GKP-squeezed-repetition code}: 
\begin{align}
\xi_{q} & = z_{q}^{(1)} - \tilde{z}_{q}^{(1)} = z_{q}^{(1)} + \frac{2\sqrt{G(G-1)}}{2G-1}R_{\sqrt{2\pi}}(z_{q}^{(2)}),
\nonumber\\
\xi_{p} & = z_{p}^{(1)} - \tilde{z}_{p}^{(1)} = z_{p}^{(1)} - \frac{2\sqrt{G(G-1)}}{2G-1}R_{\sqrt{2\pi}}(z_{p}^{(2)}), 
\end{align}
where $\boldsymbol{z} = (z_{q}^{(1)},z_{p}^{(1)},z_{q}^{(2)},z_{p}^{(2)})^{T}$ follows a Gaussian distribution with zero means and the covariance matrix $V_{\boldsymbol{z}}$. By using 
\begin{align}
R_{s}(z) \equiv \sum_{n\in\mathbb{Z}} (z-ns)\cdot  I\Big{\lbrace} z\in\Big{[}\Big{(}n-\frac{1}{2}\Big{)}s,\Big{(}n+\frac{1}{2}\Big{)}s\Big{]} \Big{\rbrace},  \label{supp_eq:identity on Rs function}
\end{align}
we find that the probability density functions of the quadrature noise are given by  
\begin{align}
Q(\xi_{q})&\equiv \int_{-\infty}^{\infty}dz_{q}^{(1)}\int_{-\infty}^{\infty}dz_{q}^{(2)} \delta\Big{(} \xi_{q}-z_{q}^{(1)} - \frac{2\sqrt{G(G-1)}}{2G-1}R_{\sqrt{2\pi}}(z_{q}^{(2)}) \Big{)} 
\nonumber\\
&\qquad\qquad\qquad\times  \frac{1}{2\pi\sigma^{2}}  \exp\Big{[} -\frac{(2G-1)}{2\sigma^{2}}\big{(}(z_{q}^{(1)})^{2}+(z_{q}^{(2)})^{2}\big{)} - \frac{2\sqrt{G(G-1)}}{\sigma^{2}}z_{q}^{(1)}z_{q}^{(2)} \Big{]}
\nonumber\\
&= \sum_{n\in\mathbb{Z}}\int_{-\infty}^{\infty}dz_{q}^{(1)}\int_{-\infty}^{\infty}dz_{q}^{(2)}   \delta\Big{(} \xi_{q}-z_{q}^{(1)} - \frac{2\sqrt{G(G-1)}}{2G-1}(z_{q}^{(2)}-\sqrt{2\pi}n) \Big{)}
\nonumber\\
&\qquad \times I\Big{\lbrace} z_{q}^{(2)}\in \Big{[} \Big{(} n-\frac{1}{2} \Big{)}\sqrt{2\pi},\Big{(} n+\frac{1}{2} \Big{)}\sqrt{2\pi} \Big{]} \Big{\rbrace} 
\nonumber\\
&\qquad\times  \frac{1}{2\pi\sigma^{2}}\exp\Big{[} -\frac{(2G-1)}{2\sigma^{2}}\Big{(} z_{q}^{(1)} + \frac{2\sqrt{G(G-1)}}{2G-1}z_{q}^{(2)} \Big{)}^{2} - \frac{1}{2(2G-1)\sigma^{2}}(z_{q}^{(2)})^{2}\Big{]}
\nonumber\\
&= \sum_{n\in\mathbb{Z}} \int_{(n-\frac{1}{2})\sqrt{2\pi}}^{(n+\frac{1}{2})\sqrt{2\pi}}dz_{q}^{(2)} \frac{1}{\sqrt{2\pi(2G-1)\sigma^{2}}}\exp\Big{[} - \frac{1}{2(2G-1)\sigma^{2}}(z_{q}^{(2)})^{2}\Big{]}
\nonumber\\
&\qquad\qquad\qquad\qquad\qquad\times \sqrt{\frac{2G-1}{2\pi\sigma^{2}}}\exp\Big{[} -\frac{(2G-1)}{2\sigma^{2}}\Big{(} \xi_{q} + \frac{2\sqrt{G(G-1)}}{2G-1}\sqrt{2\pi}n \Big{)}^{2} \Big{]}
\nonumber\\
&= \sum_{n\in\mathbb{Z}}q_{n}\cdot p\Big{[}\frac{\sigma}{\sqrt{2G-1}}\Big{]}(\xi_{q}-\mu_{n}) , \label{appendix_eq:probability density functions of the position quadrature noise two-mode GKP-squeezed-repetition code without GKP noise}
\end{align}
and similarly
\begin{align}
P(\xi_{p})&\equiv \int_{-\infty}^{\infty}dz_{p}^{(1)}\int_{-\infty}^{\infty}dz_{p}^{(2)}  \delta\Big{(} \xi_{p}-z_{p}^{(1)} + \frac{2\sqrt{G(G-1)}}{2G-1}R_{\sqrt{2\pi}}(z_{p}^{(2)}) \Big{)} 
\nonumber\\
&\qquad\qquad\qquad\times  \frac{1}{2\pi\sigma^{2}}\exp\Big{[} -\frac{(2G-1)}{2\sigma^{2}}\big{(}(z_{p}^{(1)})^{2}+(z_{p}^{(2)})^{2}\big{)} + \frac{2\sqrt{G(G-1)}}{\sigma^{2}}z_{p}^{(1)}z_{p}^{(2)} \Big{]}
\nonumber\\
&= \sum_{n\in\mathbb{Z}}q_{n}\cdot p\Big{[}\frac{\sigma}{\sqrt{2G-1}}\Big{]}(\xi_{q}-\mu_{n}),    \label{appendix_eq:probability density functions of the momentum quadrature noise two-mode GKP-squeezed-repetition code without GKP noise}
\end{align}
where $q_{n}(=q_{-n})$ and $\mu_{n}$ are as defined as
\begin{align}
q_{n} &= \int_{(n-\frac{1}{2})\sqrt{2\pi}}^{(n+\frac{1}{2})\sqrt{2\pi}} dz p[\sqrt{2G-1}\sigma](z) \,\,\,\textrm{and}\,\,\,  \mu_{n} = \frac{2\sqrt{G(G-1)}}{2G-1}\sqrt{2\pi}n. \label{eq:qn and mun}
\end{align} 

We derive the asymptotic expressions for the optimal gain $G^{\star}$ and the minimum standard deviation $\sigma_{L}^{\star}$ given in Eqs.\ \eqref{eq:optimal G asymptotic} and \eqref{eq:optimal STD asymptotic}. By assuming $\sqrt{2G-1}\sigma \ll 1$, we find 
\begin{align}
(\sigma_{L})^{2} \simeq \frac{\sigma^{2}}{2G-1} + \frac{8\pi\sqrt{G(G-1)}}{(2G-1)^{2}}\textrm{erfc}\Big{(} \frac{\sqrt{\pi}}{2\sqrt{2G-1}\sigma} \Big{)}. \label{appendix_eq:variance of the logical quadrature noise asymptotic} 
\end{align}
Assuming further that $G\gg 1$ (which is relevant when $\sigma\ll 1$) and using the asymptotic formula for the complementary error function, i.e.,  
\begin{align}
\textrm{erfc}(x) \xrightarrow{x\rightarrow \infty} \frac{1}{x\sqrt{\pi}}\exp[-x^{2}], 
\end{align}
we can simplify Eq.\ \eqref{appendix_eq:variance of the logical quadrature noise asymptotic} as 
\begin{align}
(\sigma_{L})^{2} \simeq \sigma^{2}x + \frac{4\sigma}{\sqrt{x} }\exp\Big{[} -\frac{\pi}{4\sigma^{2}}x \Big{]} \equiv f(x), 
\end{align}
where $x\equiv 1/(2G-1)$. The optimum $x^{\star}$ can be found by solving $f'(x^{\star})=0$. Note that $f'(x)$ is given by
\begin{align}
f'(x) = \sigma^{2} - \Big{(} \frac{\pi}{\sigma \sqrt{x}} + \frac{2\sigma}{\sqrt{x^{3}}} \Big{)}\exp\Big{[} -\frac{\pi}{4\sigma^{2}}x \Big{]} . 
\end{align} 
Thus, $x^{\star}$ should satisfy 
\begin{align}
x^{\star} &= \frac{4\sigma^{2}}{\pi} \log_{e}\Big{(} \frac{\pi}{\sigma^{3}\sqrt{x^{\star}}} + \frac{2}{\sigma \sqrt{(x^{\star})^{3}}} \Big{)}
\nonumber\\
&= \frac{4\sigma^{2}}{\pi} \log_{e}\Big{(} \frac{\pi^{3/2}}{2\sigma^{4}\sqrt{\log_{e}(\cdots)}} + \frac{\pi^{3/2}}{4\sigma^{4} \sqrt{(\log_{e}(\cdots))^{3}}} \Big{)} . 
\end{align} 
where we iteratively plugged in the first equation into itself to get the second equation. Since $\log_{e}(\cdots) \gg 1$, we can disregard the second term in the second line. By further neglecting the logarithmic factor $\sqrt{\log_{e}(\cdots)}$, we get 
\begin{align}
x^{\star} \simeq \frac{4\sigma^{2}}{\pi} \log_{e}\Big{(} \frac{\pi^{3/2}}{2\sigma^{4}} \Big{)}. 
\end{align}
Since $G^{\star} = \frac{1}{2x^{\star}} + \frac{1}{2}$, Eq.\ \eqref{eq:optimal G asymptotic} follows: 
\begin{align}
G^{\star}  \xrightarrow{\sigma \ll 1} \frac{\pi}{8\sigma^{2}} \Big{(} \log_{e}\Big{[} \frac{\pi^{3/2}}{2\sigma^{4}} \Big{]} \Big{)}^{-1} + \frac{1}{2}. 
\end{align}
The optimal value $(\sigma^{\star}_{L})^{2} = f(x^{\star})$ is then approximately given by 
\begin{align}
(\sigma^{\star}_{L})^{2} \simeq \frac{4\sigma^{4}}{\pi}\log_{e}\Big{[} \frac{\pi^{3/2}}{2\sigma^{4}} \Big{]}  + \frac{4\sigma^{4}}{\pi}  \Big{(}\log_{e}\Big{[} \frac{\pi^{3/2}}{2\sigma^{4}} \Big{]}\Big{)}^{-\frac{1}{2}}
\end{align}
Since $\log_{e}(\pi^{3/2}/(2\sigma^{4}))\gg 1$, we can disregard the second term and obtain Eq.\ \eqref{eq:optimal STD asymptotic}: 
\begin{align}
\sigma^{\star}_{L} \xrightarrow{\sigma \ll 1} \frac{2\sigma^{2}}{\sqrt{\pi}}\Big{(}\log_{e}\Big{[} \frac{\pi^{3/2}}{2\sigma^{4}} \Big{]} \Big{)}^{\frac{1}{2}}. 
\end{align}  

Let us now consider the case with noisy GKP states (see Fig.\ \ref{fig:GKP-two-mode-squeezing code with noisy GKP states} and Eq.\ \eqref{eq:estimated ancilla quadrature noise noisy GKP states}). Then, we have 
\begin{align}
\xi_{q} &\equiv z_{q}^{(1)} +\frac{2\sqrt{G(G-1)} \sigma^{2} }{ (2G-1)\sigma^{2} + 2\sigma_{\textrm{gkp}}^{2} } R_{\sqrt{2\pi}}(z_{q}^{(2)} +\xi_{q}^{(\textrm{gkp})} ) , 
\end{align}
where the GKP noise $\xi_{q}^{(\textrm{gkp})}$ is independent of $z_{q}^{(1)}$ and $z_{q}^{(2)}$ and follows $\xi_{q}^{(\textrm{gkp})}\sim \mathcal{N}( 0, 2\sigma_{\textrm{gkp}}^{2} )$. Then, the probability density function $Q(\xi_{q})$ is given by
\begin{align}
Q(\xi_{q})&\equiv \int_{ \mathbb{R}^{3} }dz_{q}^{(1)}dz_{q}^{(2)}  d\xi_{q}^{(\textrm{gkp})}  \delta\Big{(} \xi_{q}-z_{q}^{(1)} - \frac{2\sqrt{G(G-1)} \sigma^{2} }{ (2G-1)\sigma^{2} +2\sigma_{\textrm{gkp}}^{2} }R_{\sqrt{2\pi}}(z_{q}^{(2)} + \xi_{q}^{(\textrm{gkp})} ) \Big{)}  
\nonumber\\
&\qquad \times  \frac{1}{\sqrt{ 16\pi^{3}\sigma^{4}\sigma_{\textrm{gkp}}^{2} }} \exp\Big{[} -\frac{(2G-1)}{2\sigma^{2}}\big{(}(z_{q}^{(1)})^{2}+(z_{q}^{(2)})^{2}\big{)} - \frac{2\sqrt{G(G-1)}  }{\sigma^{2}}z_{q}^{(1)}z_{q}^{(2)}   \Big{]} 
\nonumber\\
&\qquad\times  \exp\Big{[} - \frac{1}{4\sigma_{\textrm{gkp}}^{2}}  ( \xi_{q}^{(\textrm{gkp})} )^{2}  \Big{]}
\nonumber\\
&=\sum_{n\in\mathbb{Z}} \int_{ \mathbb{R}^{3} }dz_{q}^{(1)}dz_{q}^{(2)}  d\xi_{q}^{(\textrm{gkp})}  \delta\Big{(} \xi_{q}-z_{q}^{(1)} - \frac{2\sqrt{G(G-1)} \sigma^{2} }{ (2G-1)\sigma^{2} + 2\sigma_{\textrm{gkp}}^{2} }(z_{q}^{(2)} + \xi_{q}^{(\textrm{gkp})} -\sqrt{2\pi}n) \Big{)} 
\nonumber\\
&\qquad\times  \frac{1}{\sqrt{ 16\pi^{3}\sigma^{4}\sigma_{\textrm{gkp}}^{2} }} \exp\Big{[} -\frac{(2G-1)}{2\sigma^{2}}\Big{(} z_{q}^{(1)} + \frac{2\sqrt{G(G-1)}}{2G-1}z_{q}^{(2)} \Big{)}^{2} - \frac{1}{2(2G-1)\sigma^{2}}(z_{q}^{(2)})^{2}\Big{]} 
\nonumber\\
&\qquad\times \exp\Big{[} - \frac{1}{4\sigma_{\textrm{gkp}}^{2}}  ( \xi_{q}^{(\textrm{gkp})} )^{2}  \Big{]}  I\Big{\lbrace} z_{q}^{(2)} + \xi_{q}^{(\textrm{gkp})} \in \Big{[} \Big{(} n-\frac{1}{2} \Big{)}\sqrt{2\pi},\Big{(} n+\frac{1}{2} \Big{)}\sqrt{2\pi} \Big{]} \Big{\rbrace} 
\nonumber\\
&= \sum_{n\in\mathbb{Z}}\int_{\mathbb{R}^{2}} dz_{q}^{(2)}  d\xi_{q}^{(\textrm{gkp})}  \frac{1}{\sqrt{ 16 \pi^{3}\sigma^{4}\sigma_{\textrm{gkp}}^{2} }} \exp\Big{[} - \frac{1}{2(2G-1)\sigma^{2}}(z_{q}^{(2)})^{2}\Big{]} \exp\Big{[} - \frac{1}{4\sigma_{\textrm{gkp}}^{2}}  ( \xi_{q}^{(\textrm{gkp})} )^{2}  \Big{]}  
\nonumber\\
&\qquad\times \exp\Big{[} -\frac{(2G-1)}{2\sigma^{2}}\Big{(}  \xi_{q} - \frac{2\sqrt{G(G-1)}\sigma^{2}}{ (2G-1)\sigma^{2} +2\sigma_{\textrm{gkp}}^{2} }( \xi_{q}^{(\textrm{gkp})} -\sqrt{2\pi}n) 
\nonumber\\
&\qquad\qquad\qquad\qquad\qquad\qquad + \frac{2\sqrt{G(G-1)} 2\sigma_{\textrm{gkp}}^{2} }{ (2G-1)( (2G-1)\sigma^{2} + 2\sigma_{\textrm{gkp}}^{2} ) } z_{q}^{(2)}  \Big{)}^{2} \Big{]}
\nonumber\\
&\qquad \times I\Big{\lbrace} z_{q}^{(2)} +\xi_{q}^{(\textrm{gkp})} \in \Big{[} \Big{(} n-\frac{1}{2} \Big{)}\sqrt{2\pi},\Big{(} n+\frac{1}{2} \Big{)}\sqrt{2\pi} \Big{]} \Big{\rbrace} . 
\end{align}
Thus, the variance of the output logical quadrature noise $(\sigma_{L})^{2} = \textrm{Var}[\xi_{q}] = \mathbb{E}[ (\xi_{q})^{2} ]$ is given by 
\begin{align}
(\sigma_{L})^{2} &= \sum_{n\in\mathbb{Z}}\int_{\mathbb{R}^{3}} dz_{q}^{(2)}  d\xi_{q}^{(\textrm{gkp})}  d\xi_{q}  \frac{1}{\sqrt{ 16\pi^{3}\sigma^{4}\sigma_{\textrm{gkp}}^{2} }}  \exp\Big{[} - \frac{1}{2(2G-1)\sigma^{2}}(z_{q}^{(2)})^{2}\Big{]} 
\nonumber\\
&\qquad \times (\xi_{q})^{2} \exp\Big{[} -\frac{(2G-1)}{2\sigma^{2}}\Big{(}  \xi_{q} - \frac{2\sqrt{G(G-1)}\sigma^{2}}{ (2G-1)\sigma^{2} +2\sigma_{\textrm{gkp}}^{2} }( \xi_{q}^{(\textrm{gkp})} -\sqrt{2\pi}n) 
\nonumber\\
&\qquad\qquad\qquad\qquad\qquad\qquad\qquad + \frac{2\sqrt{G(G-1)} 2\sigma_{\textrm{gkp}}^{2} }{ (2G-1)( (2G-1)\sigma^{2} + 2\sigma_{\textrm{gkp}}^{2} ) } z_{q}^{(2)}  \Big{)}^{2} \Big{]}
\nonumber\\
&\qquad \times\exp\Big{[} - \frac{1}{4\sigma_{\textrm{gkp}}^{2}}  ( \xi_{q}^{(\textrm{gkp})} )^{2}  \Big{]}   I\Big{\lbrace} z_{q}^{(2)} +\xi_{q}^{(\textrm{gkp})} \in \Big{[} \Big{(} n-\frac{1}{2} \Big{)}\sqrt{2\pi},\Big{(} n+\frac{1}{2} \Big{)}\sqrt{2\pi} \Big{]} \Big{\rbrace} 
\nonumber\\
&= \sum_{n\in\mathbb{Z}}\int_{\mathbb{R}^{2}} dz_{q}^{(2)}  d\xi_{q}^{(\textrm{gkp})}    p [  \sqrt{2G-1}\sigma ] ( z_{q}^{(2)} ) \cdot p[\sqrt{2}\sigma_{\textrm{gkp}}]( \xi_{q}^{(\textrm{gkp})} ) 
\nonumber\\
&\qquad\times \Big{[} \frac{\sigma^{2}}{2G-1} +  \Big{(}\frac{2\sqrt{G(G-1)}\sigma^{2}}{ (2G-1)\sigma^{2} +2\sigma_{\textrm{gkp}}^{2} }( \xi_{q}^{(\textrm{gkp})} -\sqrt{2\pi}n) 
\nonumber\\
&\qquad\qquad\qquad\qquad\qquad- \frac{2\sqrt{G(G-1)} 2\sigma_{\textrm{gkp}}^{2} }{ (2G-1)( (2G-1)\sigma^{2} + 2\sigma_{\textrm{gkp}}^{2} ) } z_{q}^{(2)} \Big{)}^{2}  \Big{]}
\nonumber\\
&\qquad \times I\Big{\lbrace} z_{q}^{(2)} +\xi_{q}^{(\textrm{gkp})} \in \Big{[} \Big{(} n-\frac{1}{2} \Big{)}\sqrt{2\pi},\Big{(} n+\frac{1}{2} \Big{)}\sqrt{2\pi} \Big{]} \Big{\rbrace}.  \label{eq:variance of the output logical quadrature noise noisy GKP states better}
\end{align}
Fig.\ \ref{fig:performance of the two-mode GKP-two-mode-squeezing code noisy GKP states} was obtained by optimizing the gain $G$ to minimize $(\sigma_{L})^{2}$ in Eq.\ \eqref{eq:variance of the output logical quadrature noise noisy GKP states better}.  

\section{Cubic phase state as a non-Gaussian resource}
\label{section:Cubic phase state as a non-Gaussian resource}

Here, we review the cubic phase state and gate \cite{Gottesman2001} and discuss distillation of noisy cubic phase states. 

\subsection{Cubic phase state and gate}

The cubic phase gate $\hat{V}_{\gamma}$ \cite{Gottesman2001} is defined as 
\begin{align}
\hat{V}_{\gamma} \equiv \exp[ i\gamma \hat{q}^{3}] , 
\end{align}
and is a non-Gaussian operation analogous to the $T$ gate $\hat{T} = |0\rangle\langle 0| + \exp[i\pi/4]|1\rangle\langle 1|$ for the multi-qubit universal quantum computation. The cubic phase state $|\gamma\rangle$ is defined as 
\begin{align}
|\gamma\rangle \equiv \hat{V}_{\gamma} |\hat{p}=0\rangle  = \frac{1}{\sqrt{2\pi}} \int_{-\infty}^{\infty}dq e^{i\gamma q^{3}}  |\hat{q}=q\rangle, 
\end{align}
and is analogous to the magic state $|T\rangle = \hat{T} |+\rangle \propto |0\rangle + \exp[i\pi/4]|1\rangle$ for qubits.  

We make the following important observation: The cubic phase state $|\gamma\rangle$ is an eigenstate of a Gaussian operator $\hat{V}_{\gamma}\hat{p}\hat{V}_{\gamma}^{\dagger} = \hat{p}+[i\gamma\hat{q}^{3},\hat{p}] = \hat{p}-3\gamma \hat{q}^{2}$, i.e., 
\begin{align}
(\hat{p}-3\gamma \hat{q}^{2}) |\gamma\rangle = \hat{V}_{\gamma}\hat{p}\hat{V}_{\gamma}^{\dagger}|\gamma\rangle  = \hat{V}_{\gamma}\hat{p}|\hat{p}=0\rangle = 0. 
\end{align}   
Note that the displaced cubic phase state 
\begin{align}
|\gamma,p\rangle &\equiv e^{ip\hat{q}} |\gamma\rangle = \hat{V}_{\gamma}|\hat{p}=p\rangle
\end{align}
is also an eigenstate of the Gaussian operator $\hat{p}-3\gamma \hat{q}^{2}$ with an eigenvalue $p$ 
\begin{align}
(\hat{p}-3\gamma \hat{q}^{2})|\gamma,p\rangle = \hat{V}_{\gamma}\hat{p}\hat{V}_{\gamma}^{\dagger}|\gamma,p\rangle =  \hat{V}_{\gamma} \hat{p}|\hat{p}=p\rangle = p\hat{V}_{\gamma}|\hat{p}=p\rangle = p|\gamma,p\rangle,   \label{eq:magic state being an eigenstate of a Gaussian operator}
\end{align} 
where $p$ can be any real number ($|\gamma,p=0\rangle = |\gamma\rangle$). Therefore, the set of displaced cubic phase states $\lbrace |\gamma,p\rangle | p\in\mathbb{R} \rbrace$ spans the entire bosonic Hilbert space. That is, any state $\hat{\rho}$ can be expressed as 
\begin{align}
\hat{\rho} = \int dpdp'P(p,p')|\gamma,p\rangle\langle \gamma,p'|, \label{eq:state in terms of magic basis} 
\end{align}     
where $P(p,p') = (P(p',p))^{*}$ and $\int_{-\infty}^{\infty}dp P(p,p) =1$.

Now, we introduce the magic twirling (consisting only of Gaussian operations) and show that it can transform any state $\hat{\rho}$ into a mixture of displaced cubic phase states 
\begin{align}
\hat{\rho}' = \int_{-\infty}^{\infty}dp P(p,p)|\gamma,p\rangle\langle \gamma,p|.  
\end{align}

We define the magic twirling as a random application of Gaussian operations $\exp[i(\hat{p}-3\gamma \hat{q}^{2})x]$ with $x \sim \mathcal{N}(0,\Sigma^{2}\rightarrow \infty)$, i.e., 
\begin{align}
\mathcal{T}(\hat{\rho}) &\equiv \lim_{\Sigma\rightarrow\infty} \int_{-\infty}^{\infty} dx \frac{1}{\sqrt{2\pi\Sigma^{2}}}e^{-\frac{x^{2}}{2\Sigma^{2}}}  e^{i(\hat{p}-3\gamma \hat{q}^{2})x} \hat{\rho} e^{-i(\hat{p}-3\gamma \hat{q}^{2})x}. \label{eq:magic twirling definition}
\end{align} 
Plugging in Eq.\ \eqref{eq:state in terms of magic basis} to Eq.\ \eqref{eq:magic twirling definition} and using $(\hat{p}-3\gamma\hat{q}^{2})|\gamma,p\rangle = p|\gamma,p\rangle$ (see Eq.\ \eqref{eq:magic state being an eigenstate of a Gaussian operator}), we get the desired result: 
\begin{align}
\mathcal{T}(\hat{\rho}) &=\int dpdp' P(p,p') \Big{[}  \lim_{\Sigma\rightarrow\infty}  \int_{-\infty}^{\infty} dx \frac{1}{\sqrt{2\pi\Sigma^{2}}}e^{-\frac{x^{2}}{2\Sigma^{2}}}  e^{i(p-p')x} \Big{]} |\gamma,p\rangle\langle\gamma,p'| 
\nonumber\\
&=  \int dpdp' P(p,p') \Big{[} \lim_{\Sigma\rightarrow\infty} e^{-\frac{\Sigma^{2}}{2}(p-p')^{2}} \Big{]}  |\gamma,p\rangle\langle\gamma,p'| 
\nonumber\\
&=  \int dpdp' P(p,p') \delta_{pp'}  |\gamma,p\rangle\langle\gamma,p'|
\nonumber\\
&=  \frac{1}{\delta(0)} \int dpdp' P(p,p') \delta(p-p')  |\gamma,p\rangle\langle\gamma,p'|
\nonumber\\
&= \frac{1}{\delta(0)} \int_{-\infty}^{\infty} dp P(p,p)  |\gamma,p\rangle\langle\gamma,p|. \label{eq:state after magic twirling}
\end{align}
Note that the normalization constant $1/\delta(0)$ is due to the fact that displaced cubic phase states $|\gamma,p\rangle$ are orthonormalized by the Dirac delta function (i.e., $\langle \gamma,p|\gamma,p'\rangle = \delta(p-p')$, yielding $\langle \gamma|\gamma\rangle = \delta(0)$), whereas the state $\hat{\rho}$ is normalized to unity (i.e., $\mathrm{Tr}[\hat{\rho}]=1$).     

Since $|\gamma,p\rangle = \exp[ip\hat{q}]|\gamma\rangle$, one understand the mixture of displaced cubic phase states in Eq.\ \eqref{eq:state after magic twirling} as the ideal cubic phase state $|\gamma\rangle\langle \gamma|$ corrupted by a random displacement error in the momentum direction, i.e., 
\begin{align}
\mathcal{T}(\hat{\rho}) &\propto \mathcal{N}_{P}(|\gamma\rangle\langle\gamma|) \,\,\,\textrm{where}\,\,\, \mathcal{N}_{P}(\hat{\rho}) \equiv \int_{-\infty}^{\infty} dp P(p,p) e^{ip\hat{q}} \hat{\rho} e^{-ip\hat{q}}. 
\end{align} 
The purpose of the cubic phase state distillation will then be to reduce, e.g., the variance of the random displacement error $\mathcal{N}_{P}$ by using only Gaussian states, operations, and homodyne measurements. Before moving on to the cubic phase state distillation, we will discuss the cubic phase state injection procedure (in the presence of the random displacement noise $\mathcal{N}_{P}$) in the next section.

Let us now show that the ability to perform a noisy cubic phase gate $\mathcal{N}_{P}(\hat{V}_{\gamma} \bullet \hat{V}_{\gamma}^{\dagger})$ is equivalent to the ability to prepare a noisy cubic phase state $\mathcal{N}_{P}(|\gamma\rangle\langle\gamma|)$ if we assume Gaussian states, operations, and homodyne measurements are free. 

\begin{figure}[t!]
\centering
\includegraphics[width=5.6in]{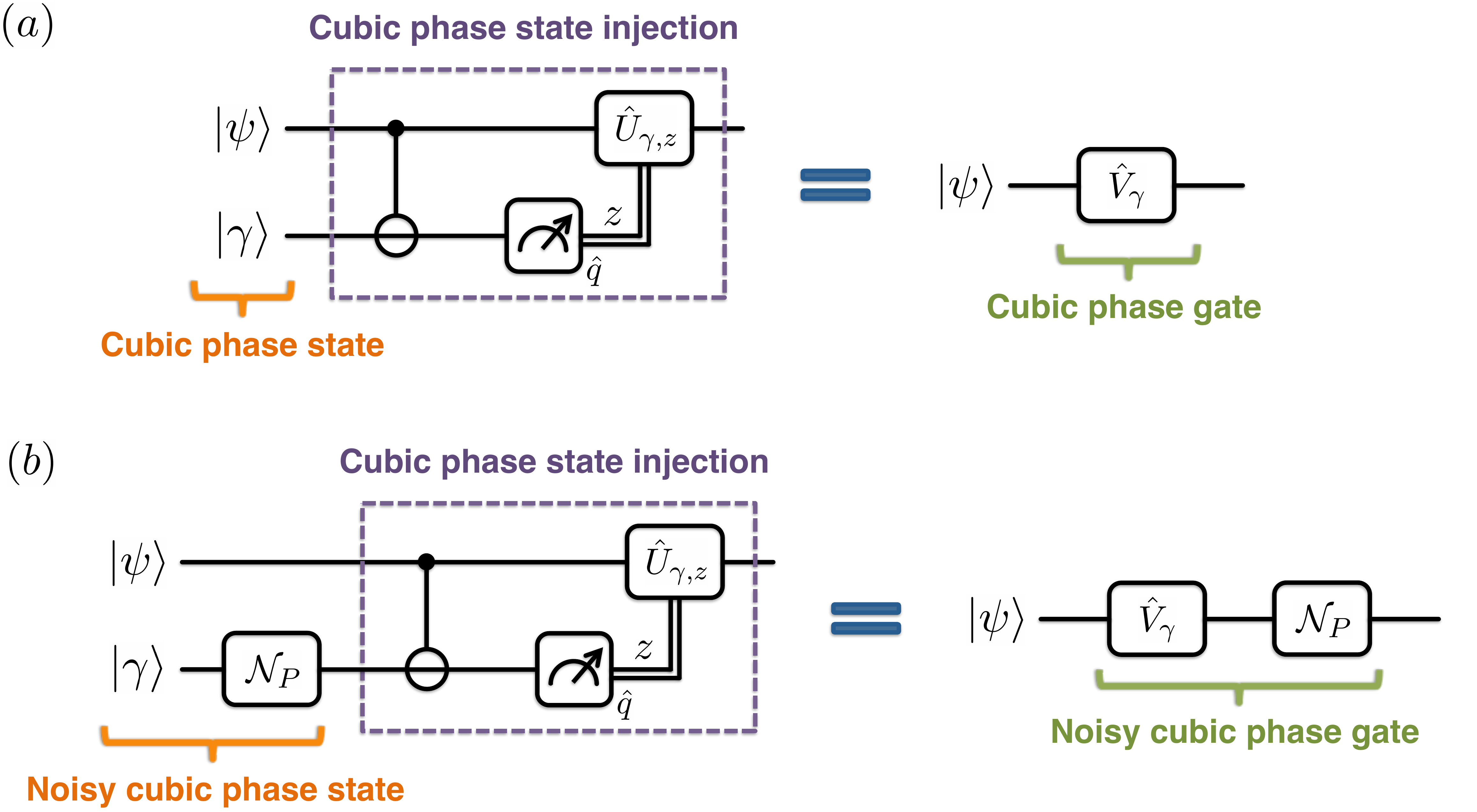}
\caption{Cubic phase state injection procedure (a) with an ideal cubic phase state $|\gamma\rangle$ and (b) with a noisy cubic phase state $\mathcal{N}_{P}(|\gamma\rangle\langle\gamma|)$. $\hat{U}_{\gamma,z}$ is given by $\hat{U}_{\gamma,z} = \exp[-3i\gamma(z\hat{q}^{2}+z^{2}\hat{q})]$.    }
\label{fig:cubic phase state injection}
\end{figure}

It is clear that the ability to perform the noisy cubic phase gate $\mathcal{N}_{P}(\hat{V}_{\gamma} \bullet \hat{V}_{\gamma}^{\dagger})$ allows us to prepare the noisy cubic phase state $\mathcal{N}_{P}(|\gamma\rangle\langle\gamma|)$ because one can simply apply $\mathcal{N}_{P}(\hat{V}_{\gamma} \bullet \hat{V}_{\gamma}^{\dagger})$ to the position eigenstate $|\hat{p}=0\rangle \langle \hat{p}=0|$: 
\begin{align}
\mathcal{N}_{P}(|\gamma\rangle\langle\gamma|)  &= \mathcal{N}_{P}\Big{(}\hat{V}_{\gamma} |\hat{p}=0\rangle \langle \hat{p}=0| \hat{V}_{\gamma}^{\dagger}\Big{)} . 
\end{align}        
Now, we show that the converse is also true. That is, we show that one can perform the noisy cubic phase gate $\mathcal{N}_{P}(\hat{V}_{\gamma} \bullet \hat{V}_{\gamma}^{\dagger})$ via the cubic phase state injection procedure by consuming a single noisy cubic phase state $\mathcal{N}_{P}(|\gamma\rangle\langle\gamma|)$. Below we recall the cubic phase state injection procedure introduced in Ref.\ \cite{Gottesman2001} (see also Fig.\ \ref{fig:cubic phase state injection}(a)) and show that it converts a displaced cubic phase state $|\gamma,p\rangle$ into the cubic phase gate $\hat{V}_{\gamma}$ followed by an unwanted shift $e^{ip\hat{q}}$:    
\begin{align}
&\textrm{\textbf{Cubic phase state injection procedure}}
\nonumber\\
&\qquad\quad\textrm{Input:}
\nonumber\\
&\qquad\qquad\qquad |\psi\rangle \otimes |\gamma,p\rangle = \int_{-\infty}^{\infty}dq \psi(q)|\hat{q}_{1}=q\rangle \otimes  \int_{-\infty}^{\infty}dq' e^{i\gamma q'^{3}} e^{ipq'} |\hat{q}_{2}=q'\rangle 
\nonumber\\
&\qquad\quad\textrm{Step 1: Apply }e^{i\hat{q}_{1}\hat{p}_{2}}
\nonumber\\
&\qquad\qquad\qquad \rightarrow   \int_{-\infty}^{\infty}dq\int_{-\infty}^{\infty}dq' \psi(q)e^{i\gamma q'^{3}} e^{ipq'} |\hat{q}_{1}=q\rangle \otimes |\hat{q}_{2}=q'-q\rangle
\nonumber\\
&\qquad\qquad\qquad =\int_{-\infty}^{\infty}dq\int_{-\infty}^{\infty}dz \psi(q)e^{i\gamma (z+q)^{3}} e^{ip(z+q)} |\hat{q}_{1}=q\rangle \otimes |\hat{q}_{2}=z\rangle 
\nonumber\\
&\qquad\quad\textrm{Step 2: Measure }\hat{q}_{2} 
\nonumber\\
&\qquad\qquad\qquad \rightarrow  \int_{-\infty}^{\infty}dq \psi(q)e^{i\gamma (z+q)^{3}} e^{ip(z+q)} |\hat{q}_{1}=q\rangle \,\,\, +\textrm{ classical side information }z
\nonumber\\
&\qquad\quad\textrm{Step 3: Apply }\hat{U}_{\gamma,z} \equiv e^{-3i\gamma(z\hat{q}_{1}^{2}+z^{2}\hat{q}_{1})}
\nonumber\\
&\qquad\qquad\qquad \rightarrow  \int_{-\infty}^{\infty}dq \psi(q)e^{i\gamma q^{3}} e^{ip q } |\hat{q}_{1}=q\rangle 
\nonumber\\
&\qquad\quad\textrm{Output:}
\nonumber\\
&\qquad\qquad\qquad = e^{ip\hat{q}} \hat{V}_{\gamma} |\psi\rangle 
\end{align}
Thus, the cubic phase injection procedure converts a noisy cubic phase state $\mathcal{N}_{P}(|\gamma\rangle\langle\gamma|)$ into a noisy cubic phase gate $\mathcal{N}_{P}(\hat{V}_{\gamma}\bullet \hat{V}_{\gamma}^{\dagger})$, i.e., 
\begin{align}
\hat{\rho} \otimes \mathcal{N}_{P}(|\gamma\rangle\langle\gamma|) \xrightarrow{\textrm{CPS injection}} \mathcal{N}_{P}(\hat{V}_{\gamma}\hat{\rho} \hat{V}_{\gamma}^{\dagger}). 
\end{align} 
See also Fig.\ \ref{fig:cubic phase state injection}(b).

\subsection{A general set up for cubic phase state distillation}

Finally, we discuss distillation of noisy cubic phase states: We want to convert many noisy cubic phase states into fewer but less noisy cubic phase states. In other words, we aim to develop continuous-variable (CV) analog of magic state distillation introduced in Ref.\ \cite{Bravyi2005}. We present a general setup for cubic phase state distillation in Fig.\ \ref{fig:cubic phase distillation setup}. Our goal is to consume $n$ noisy cubic phase states and distill $k(<n)$ less noisy cubic phase states by using only Gaussian states, operations, and homodyne measurements.    

\begin{figure}[t!]
\centering
\includegraphics[width=5.6in]{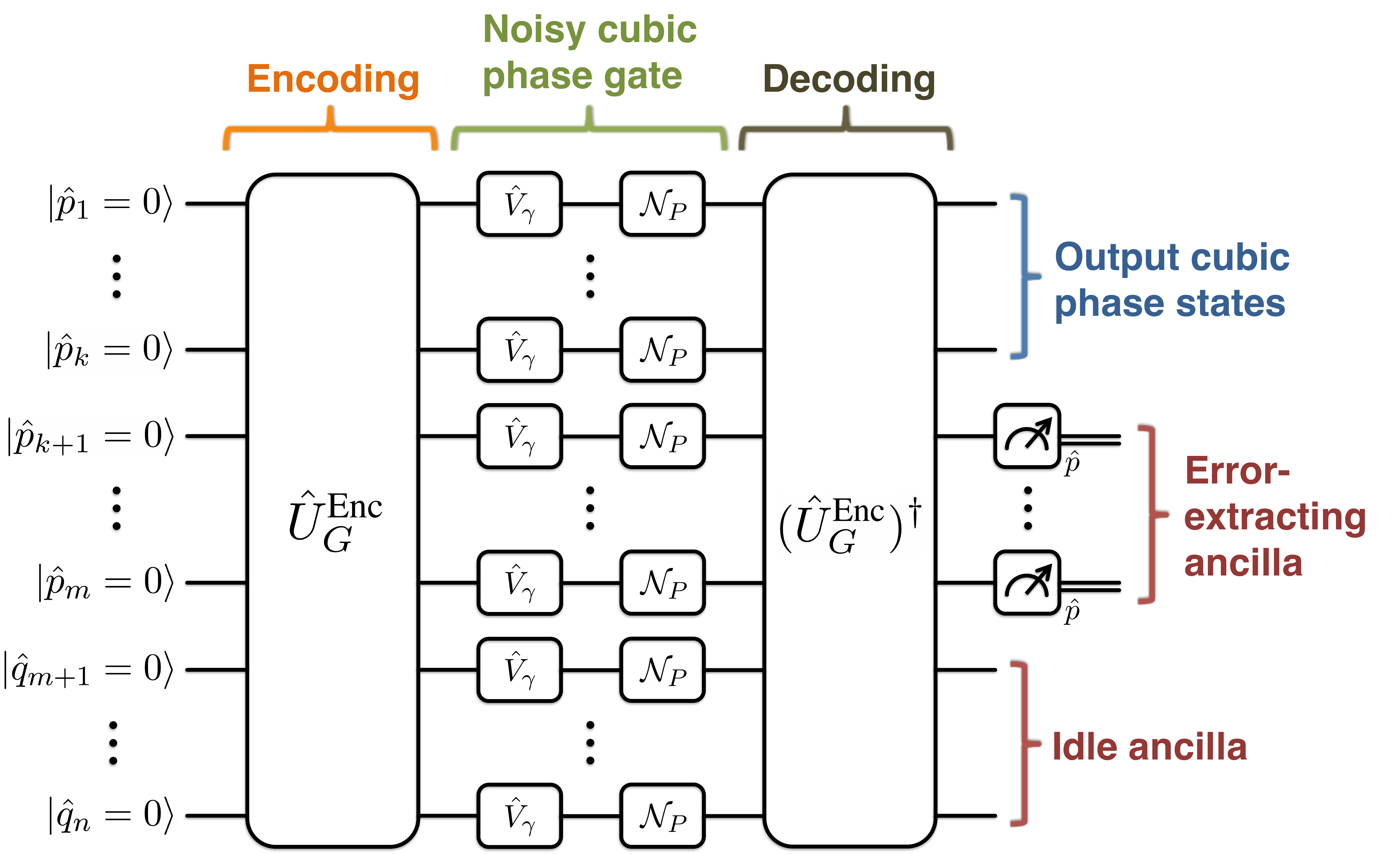}
\caption{A general setup for distilling $k$ less noisy cubic phase states out of $n$ noisier cubic phase states. }
\label{fig:cubic phase distillation setup}
\end{figure}

Let us consider $n$ bosonic modes where the first $m$ modes are initialized to momentum eigenstates $|\hat{p}_{1}=0\rangle,\cdots|\hat{p}_{m}=0\rangle$ and the last $n-m$ modes are initialized to position eigenstates $|\hat{q}_{m+1}=0\rangle,\cdots,|\hat{q}_{n}=0\rangle$. Note that these initial states are Gaussian states. Then, we apply an encoding Gaussian circuit $\hat{U}_{G}^{\textrm{Enc}}$ and then apply noisy cubic phase gates $\mathcal{N}_{P}(\hat{V}_{\gamma}\bullet\hat{V}_{\gamma}^{\dagger})$ transversally to all bosonic modes. Ideally, we want that the transversal cubic phase gates interleaved with the encoding and the inverse of the encoding Gaussian circuit, i.e., $(\hat{U}_{G}^{\textrm{Enc}})^{\dagger}\cdot (\bigotimes_{j=1}^{n}\hat{V}_{\gamma}^{(j)}) \cdot \hat{U}_{G}^{\textrm{Enc}}$ to implement cubic phase gates to the first $k$ modes and identity operations to the next $m-k$ modes which were initialized to momentum eigenstates. That is, we want 
\begin{align}
&(\hat{U}_{G}^{\textrm{Enc}})^{\dagger}\cdot \Big{(}\bigotimes_{j=1}^{n}\hat{V}_{\gamma}^{(j)}\Big{)} \cdot \hat{U}_{G}^{\textrm{Enc}} \Big{(} \bigotimes_{j=1}^{m} |\hat{p}_{j}=0\rangle \Big{)}\Big{(} \bigotimes_{j=m+1}^{n} |\hat{q}_{j}=0\rangle \Big{)} 
\nonumber\\
&\quad = \Big{(} \bigotimes_{j=1}^{k} \hat{V}_{\gamma}^{(j)}|\hat{p}_{j}=0\rangle \Big{)}  \Big{(} \bigotimes_{j=k+1}^{m} |\hat{p}_{j}=0\rangle \Big{)} \Big{(} \bigotimes_{j=m+1}^{n} |\hat{q}_{j}=0\rangle \Big{)}, \label{eq:transversality} 
\end{align}      
where $\hat{V}_{\gamma}^{(j)}$ is the cubic phase gate acting on the $j^{\textrm{th}}$ mode, i.e., $\hat{V}_{\gamma}^{(j)}\equiv \exp[i\gamma\hat{q}_{j}^{3}]$. Note that in this desired case the first $k$ modes support cubic phase states $\hat{V}_{\gamma}^{(j)}|\hat{p}_{j}=0\rangle$, while the next $m-k$ modes are in the momentum eigenstates $\bigotimes_{j=1}^{m} |\hat{p}_{j}=0\rangle$.     

In the next subsection, we discuss under which condition (on the encoding circuit $\hat{U}_{G}^{\textrm{Enc}}$) we have the desired property in Eq.\ \eqref{eq:transversality}. For now, we focus on what to do next: Since the modes $k+1,\cdots,m$ are supposed to be in the momentum eigenstates $\bigotimes_{j=k+1}^{m} |\hat{p}_{j}=0\rangle$ if there were no errors, we can use these modes to detect momentum quadrature noises. In particular, the inverse of the encoding Gaussian circuit $(\hat{U}_{G}^{\textrm{Enc}})^{\dagger}$ transforms the uncorrelated momentum quadrature noises $\bigotimes_{j=1}^{k}\mathcal{N}_{P}^{(j)}$ into correlated momentum quadrature noises $(\mathcal{U}_{G}^{\textrm{Enc}})^{\dagger}\cdot (\bigotimes_{j=1}^{n}\mathcal{N}_{P}^{(j)})\cdot \mathcal{U}_{G}^{\textrm{Enc}}$, where $\mathcal{U}_{G}^{\textrm{Enc}} \equiv \hat{U}_{G}^{\textrm{Enc}}\bullet (\hat{U}_{G}^{\textrm{Enc}})^{\dagger}$ and $\mathcal{N}_{P}^{(j)}$ is the momentum quadrature noise channel acting on the $j^{\textrm{th}}$ mode, i.e., 
\begin{align}
\mathcal{N}_{P}^{(j)}(\hat{\rho}) \equiv \int_{-\infty}^{\infty}dp P(p,p)e^{ip\hat{q}_{j}}\hat{\rho}e^{-ip\hat{q}_{j}}. 
\end{align}
Therefore, the noises in the modes $k+1,\cdots,m$ will be correlated with the noises in the modes $1,\cdots, k$. If the noise correlation is strong enough, we can infer noises in the output cubic phase states in the modes $1,\cdots,k$ indirectly by extracting the noises in the modes $k+1,\cdots,m$ via momentum homodyne measurements. This information can then be used to reduce the noises in the modes $1,\cdots,k$ and ideally we end up with less noisy cubic phase states. 

Note that the modes $m+1,\cdots,n$ are in position eigenstates and therefore cannot be used to detect momentum quadrature noises. Nevertheless these idle modes are essential to ensure the desired property in Eq.\ \eqref{eq:transversality}.

\subsection{Triorthogonality and transversality} 

We now discuss under which condition (on $\hat{U}_{G}^{\textrm{Enc}}$) the desired property in Eq.\ \eqref{eq:transversality} holds. For simplicity, we restrict ourselves to CSS-type encoding circuit, i.e., 
\begin{align}
\begin{bmatrix}
\boldsymbol{\hat{q}'}\\
\boldsymbol{\hat{p}'}
\end{bmatrix} &= \begin{bmatrix}
\boldsymbol{A} & 0 \\
0 & (\boldsymbol{A}^{T})^{-1}
\end{bmatrix}\begin{bmatrix}
\boldsymbol{\hat{q}}\\
\boldsymbol{\hat{p}}
\end{bmatrix},  
\end{align}   
where $\boldsymbol{\hat{q}} \equiv (\hat{q}_{1},\cdots, \hat{q}_{n})^{T}$ and $\boldsymbol{\hat{p}} \equiv (\hat{p}_{1},\cdots, \hat{p}_{n})^{T}$ are the input quadrature operators and $\boldsymbol{\hat{q}'} \equiv ((\hat{U}_{G}^{\textrm{Enc}})^{\dagger}\hat{q}_{1}\hat{U}_{G}^{\textrm{Enc}},\cdots, (\hat{U}_{G}^{\textrm{Enc}})^{\dagger}\hat{q}_{n}\hat{U}_{G}^{\textrm{Enc}})^{T}$ and $\boldsymbol{\hat{p}'} \equiv ((\hat{U}_{G}^{\textrm{Enc}})^{\dagger}\hat{p}_{1}\hat{U}_{G}^{\textrm{Enc}},\cdots, (\hat{U}_{G}^{\textrm{Enc}})^{\dagger}\hat{p}_{n}\hat{U}_{G}^{\textrm{Enc}})^{T}$ are the output quadrature operators. The only constraint on the $n\times n$ matrix $\boldsymbol{A}$ is that it should be invertible. Note that the transformation matrix $(\boldsymbol{A}^{T})^{-1}$ for the momentum quadrature operators was chosen to make the matrix $\boldsymbol{S}=\textrm{diag}( \boldsymbol{A}, (\boldsymbol{A}^{T})^{-1})$ symplectic, i.e., 
\begin{align}
\boldsymbol{S}\boldsymbol{\Omega}\boldsymbol{S}^{T} = \boldsymbol{\Omega}\,\,\, \textrm{where}\,\,\,\boldsymbol{\Omega} \equiv \begin{bmatrix}
0 & \boldsymbol{I}_{n}\\
-\boldsymbol{I}_{n} & 0
\end{bmatrix}, 
\end{align}   
and $\boldsymbol{I}_{n}$ is the $n\times n$ identity matrix. 

Let us now recall the left hand side of the desired relation in Eq.\ \eqref{eq:transversality}. In particular, we need to inspect $(\hat{U}_{G}^{\textrm{Enc}})^{\dagger}\cdot (\bigotimes_{j=1}^{n}\hat{V}_{\gamma}^{(j)}) \cdot \hat{U}_{G}^{\textrm{Enc}}$. Note that 
\begin{align}
(\hat{U}_{G}^{\textrm{Enc}})^{\dagger}\cdot \Big{(}\bigotimes_{j=1}^{n}\hat{V}_{\gamma}^{(j)}\Big{)} \cdot \hat{U}_{G}^{\textrm{Enc}} &= (\hat{U}_{G}^{\textrm{Enc}})^{\dagger} \cdot \exp\Big{[} \sum_{j=1}^{n} i\gamma \hat{q}_{j}^{3}\Big{]} \cdot \hat{U}_{G}^{\textrm{Enc}} 
\nonumber\\
&=  \exp\Big{[} \sum_{a,b,c=1}^{n} i\gamma \Big{(} \sum_{j=1}^{n}A_{j,a}A_{j,b}A_{j,c} \Big{)} \hat{q}_{a}\hat{q}_{b}\hat{q}_{c}\Big{]} . 
\end{align}
If we act this operation on the state $( \bigotimes_{j=1}^{m} |\hat{p}_{j}=0\rangle )( \bigotimes_{j=m+1}^{n} |\hat{q}_{j}=0\rangle )$, we get  
\begin{align}
&(\hat{U}_{G}^{\textrm{Enc}})^{\dagger}\cdot \Big{(}\bigotimes_{j=1}^{n}\hat{V}_{\gamma}^{(j)}\Big{)} \cdot \hat{U}_{G}^{\textrm{Enc}} \Big{(} \bigotimes_{j=1}^{m} |\hat{p}_{j}=0\rangle \Big{)}\Big{(} \bigotimes_{j=m+1}^{n} |\hat{q}_{j}=0\rangle \Big{)} 
\nonumber\\
&\quad= \Big{(} \exp\Big{[} \sum_{a,b,c=1}^{m} i\gamma \Big{(} \sum_{j=1}^{n}A_{j,a}A_{j,b}A_{j,c} \Big{)} \hat{q}_{a}\hat{q}_{b}\hat{q}_{c}\Big{]} \bigotimes_{j=1}^{m} |\hat{p}_{j}=0\rangle \Big{)}\Big{(} \bigotimes_{j=m+1}^{n} |\hat{q}_{j}=0\rangle \Big{)}, 
\end{align}
because $\hat{q}_{a}\hat{q}_{b}\hat{q}_{c}=0$ unless $1\le a,b,c\le m$.

\begin{figure}[t!]
\centering
\includegraphics[width=5.9in]{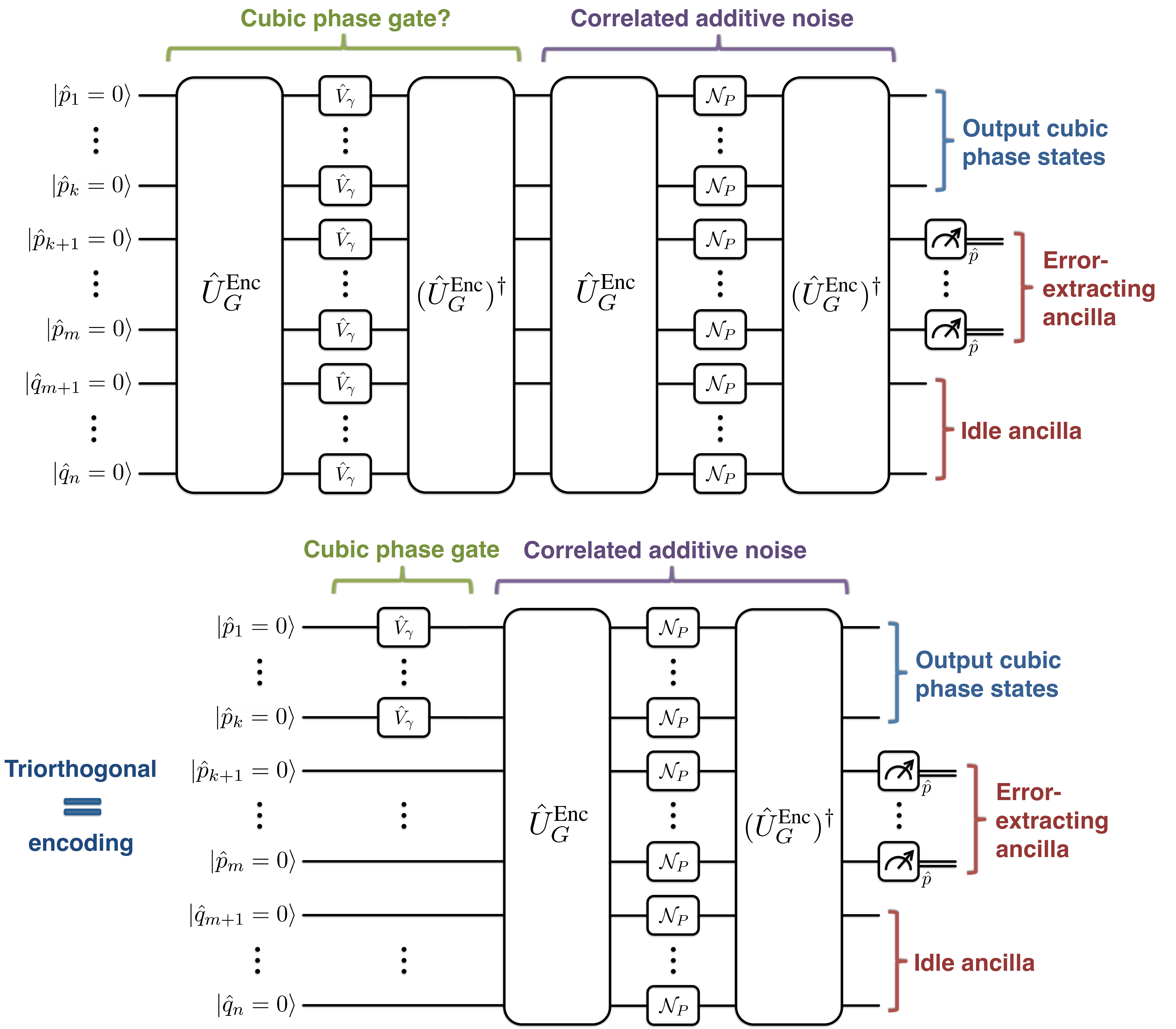}
\caption{Cubic phase state distillation with a $(n,m,k)$-triorthogonal encoding circuit $\hat{U}_{G}^{\textrm{Enc}}$. }
\label{fig:cubic phase state distillation triorthogonality}
\end{figure}

Let us consider a $n\times m$ submatrix
\begin{align}
\boldsymbol{\bar{A}} = \begin{bmatrix}
A_{11} & \cdots & A_{1m}\\
\vdots & \ddots & \vdots \\
A_{n1} & \cdots & A_{nm} 
\end{bmatrix} 
\end{align}   
of the $n\times n$ matrix $\boldsymbol{A}$. We say that the submatrix $\boldsymbol{\bar{A}}$ is $(n,m,k)$-triorthogonal if it is full-rank and its elements satisfy
\begin{align}
\sum_{j=1}^{n}A_{j,a}A_{j,b}A_{j,c} = \begin{cases}
1 & (a,b,c)=(1,1,1),\cdots,(k,k,k) \\
0 & \textrm{otherwise}
\end{cases}, \label{eq:triorthogonality}
\end{align}  
for all $1\le a,b,c\le m$. Therefore, if we choose $\hat{U}_{G}^{\textrm{Enc}}$ such that the submatrix $\boldsymbol{\bar{A}}$ is $(n,m,k)$-triorthogonal, we have the desired property in Eq.\ \eqref{eq:transversality}. 
\begin{align}
&(\hat{U}_{G}^{\textrm{Enc}})^{\dagger}\cdot \Big{(}\bigotimes_{j=1}^{n}\hat{V}_{\gamma}^{(j)}\Big{)} \cdot \hat{U}_{G}^{\textrm{Enc}} \Big{(} \bigotimes_{j=1}^{m} |\hat{p}_{j}=0\rangle \Big{)}\Big{(} \bigotimes_{j=m+1}^{n} |\hat{q}_{j}=0\rangle \Big{)} 
\nonumber\\
&\quad= \Big{(} \exp\Big{[} \sum_{a,b,c=1}^{m} i\gamma \Big{(} \sum_{j=1}^{n}A_{j,a}A_{j,b}A_{j,c} \Big{)} \hat{q}_{a}\hat{q}_{b}\hat{q}_{c}\Big{]} \bigotimes_{j=1}^{m} |\hat{p}_{j}=0\rangle \Big{)}\Big{(} \bigotimes_{j=m+1}^{n} |\hat{q}_{j}=0\rangle \Big{)}
\nonumber\\
&\quad = \underbrace{ \Big{(} \bigotimes_{j=1}^{k} \exp[i\gamma\hat{q}_{j}^{3}]|\hat{p}_{j}=0\rangle \Big{)} }_{ \textrm{Output cubic phase states} } \underbrace{ \Big{(} \bigotimes_{j=k+1}^{m} |\hat{p}_{j}=0\rangle \Big{)} }_{ \textrm{Error-extracting ancilla} } \underbrace{  \Big{(} \bigotimes_{j=m+1}^{n} |\hat{q}_{j}=0\rangle \Big{)}  }_{ \textrm{Idle ancilla} }. 
\end{align} 
Note that the triorthogonality we introduced in Eq.\ \eqref{eq:triorthogonality} is analogous to the triorthogonality introduced in Ref.\ \cite{Bravyi2012} (defined with respect to modulo $2$ arithmetic) for magic state distillation. Consequences of the triorthogonality is summarized graphically in Fig.\ \ref{fig:cubic phase state distillation triorthogonality}.

We now give several examples of triorthogonal matrices inspired by the triorthogonal matrices for magic state distillation introduced in Refs.\ \cite{Bravyi2005,Bravyi2012}. First, by carefully modifying the triorthogonal matrix that corresponds to the 15-qubit Reed-Muller code originally considered in Ref.\ \cite{Bravyi2005}, we find 
\begin{align}
\boldsymbol{\bar{A}}_{\textrm{RM}}^{\textrm{mod }2} = \begin{bmatrix}
1 & 1 & 0 & 0 & 0 \\
1 & 0 & 1 & 0 & 0 \\
1 & 1 & 1 & 0 & 0 \\
1 & 0 & 0 & 1 & 0 \\
1 & 1 & 0 & 1 & 0 \\
1 & 0 & 1 & 1 & 0 \\
1 & 1 & 1 & 1 & 0 \\
1 & 0 & 0 & 0 & 1 \\
1 & 1 & 0 & 0 & 1 \\
1 & 0 & 1 & 0 & 1 \\
1 & 1 & 1 & 0 & 1 \\
1 & 0 & 0 & 1 & 1 \\
1 & 1 & 0 & 1 & 1 \\
1 & 0 & 1 & 1 & 1 \\
1 & 1 & 1 & 1 & 1 
\end{bmatrix} \rightarrow \boldsymbol{\bar{A}}_{\textrm{RM}} \equiv  \begin{bmatrix}
1 & 1 & 0 & 0 & 0 \\
1 & 0 & 1 & 0 & 0 \\
-1 & -1 & -1 & 0 & 0 \\
1 & 0 & 0 & 1 & 0 \\
-1 & -1 & 0 & -1 & 0 \\
-1 & 0 & -1 & -1 & 0 \\
1 & 1 & 1 & 1 & 0 \\
1 & 0 & 0 & 0 & 1 \\
-1 & -1 & 0 & 0 & -1 \\
-1 & 0 & -1 & 0 & -1 \\
1 & 1 & 1 & 0 & 1 \\
-1 & 0 & 0 & -1 & -1 \\
1 & 1 & 0 & 1 & 1 \\
1 & 0 & 1 & 1 & 1 \\
-1 & -1 & -1 & -1 & -1 
\end{bmatrix}. \label{eq:triorthogonal matrix RM}
\end{align}  
One can check that $\boldsymbol{\bar{A}}_{\textrm{RM}}$ is $(15,5,1)$-triorthogonal. Similarly, by modifying the triorthogonal matrix constructed in Ref.\ \cite{Bravyi2012}, we find
\begin{align}
\boldsymbol{\bar{A}}_{\textrm{BH}}^{\textrm{mod }2}(k=2) = \begin{bmatrix}
0 & 0 & 0 & 0 & 1 \\
0 & 0 & 1 & 0 & 1 \\
0 & 0 & 0 & 1 & 1 \\
0 & 0 & 1 & 1 & 1 \\
1 & 1 & 0 & 0 & 1 \\
1 & 1 & 1 & 0 & 1 \\
1 & 1 & 0 & 1 & 1 \\
1 & 1 & 1 & 1 & 1 \\
1 & 0 & 1 & 0 & 0 \\
1 & 0 & 0 & 1 & 0 \\
1 & 0 & 1 & 1 & 0 \\
0 & 1 & 1 & 0 & 0 \\
0 & 1 & 0 & 1 & 0 \\
0 & 1 & 1 & 1 & 0 \\
\end{bmatrix}  \rightarrow \boldsymbol{\bar{A}}_{\textrm{BH}}(k=2) \equiv \begin{bmatrix}
0 & 0 & 0 & 0 & -1 \\
0 & 0 & 1 & 0 & 1 \\
0 & 0 & 0 & 1 & 1 \\
0 & 0 & -1 & -1 & -1 \\
1 & 1 & 0 & 0 & 1 \\
-1 & -1 & -1 & 0 & -1 \\
-1 & -1 & 0 & -1 & -1 \\
1 & 1 & 1 & 1 & 1 \\
1 & 0 & -1 & 0 & 0 \\
1 & 0 & 0 & -1 & 0 \\
-1 & 0 & 1 & 1 & 0 \\
0 & 1 & 1 & 0 & 0 \\
0 & 1 & 0 & 1 & 0 \\
0 & -1 & -1 & -1 & 0 \\
\end{bmatrix}.  
\end{align} 
Similarly, one can check that $\boldsymbol{\bar{A}}_{\textrm{BH}}(k=2)$ is $(14,5,2)$-triorthogonal.

\subsection{Failed attempts on cubic phase distillation} 

Here, we analyze the performance of an error-correcting scheme with a $(n,m,k)$-triorthogonal matrix $\boldsymbol{\bar{A}}$. For simplicity, we restrict ourselves to the $k=1$ case. That is, we aim to consume $n$ noisy cubic phase states and distill one less noisier cubic phase state.

Recall that each cubic phase gate is corrupted by a random displacement channel
\begin{align}
\mathcal{N}_{P}(\hat{\rho}) \equiv \int_{-\infty}^{\infty} dp P(p,p) e^{ip\hat{q}} \hat{\rho} e^{-ip\hat{q}}. 
\end{align}       
In the Heisenberg picture, this channel adds a random shift $\xi_{p}$ to the momentum quadrature, i.e., 
\begin{align}
\hat{q} &\rightarrow \hat{q}'\equiv \hat{q},
\nonumber\\
\hat{p} &\rightarrow \hat{p}'\equiv \hat{p} + \xi_{p},
\end{align}
where the random variable $\xi_{p}$ follows a probability distribution $\textrm{PDF}(\xi_{p})\equiv P(\xi_{p},\xi_{p})$. Similarly, the iid random displacement error $\bigotimes_{j=1}^{n}\mathcal{N}_{P}^{(j)}$ adds a random shift $\xi_{q}^{(j)}$ to the $j^{\textrm{th}}$ mode, where the shifts $\boldsymbol{\xi_{p}}\equiv (\xi_{q}^{(1)},\cdots,\xi_{q}^{(n)})^{T}$ follow the probability distribution $\textrm{PDF}(\xi_{p}^{(1)},\cdots,\xi_{p}^{(n)}) \equiv \bigotimes_{j=1}^{n}P(\xi_{p}^{(j)},\xi_{p}^{(j)})$. 

Let us now consider a specific instance where $\bigotimes_{j=1}^{n}\mathcal{N}_{P}^{(j)}$ generates shifts $\boldsymbol{\xi_{p}} = (\xi_{p}^{(1)},\cdots,\xi_{p}^{(n)})^{T}$. These shifts are described by a displacement operator $\exp[i\sum_{j=1}^{n}\xi_{p}^{(j)}\hat{q}_{j}]$. Note that the displacement error interleaved with the encoding and the inverse of the encoding circuit is given by
\begin{align}
(\hat{U}_{G}^{\textrm{Enc}})^{\dagger} \cdot  \exp\Big{[}i\sum_{j=1}^{n}\xi_{p}^{(j)}\hat{q}_{j}\Big{]} \cdot \hat{U}_{G}^{\textrm{Enc}} &=  \exp\Big{[} \sum_{a=1}^{n} i \Big{(} \sum_{j=1}^{n}A_{ja}\xi_{p}^{(j)} \Big{)} \hat{q}_{a}\Big{]} 
\nonumber\\
&= \exp\Big{[} \sum_{a=1}^{n} i (\boldsymbol{A}^{T}\cdot \boldsymbol{\xi_{p}})_{a} \hat{q}_{a}\Big{]} , 
\end{align}
where we used $(\hat{U}_{G}^{\textrm{Enc}})^{\dagger}\hat{q}_{j}\hat{U}_{G}^{\textrm{Enc}} = \sum_{a=1}^{n}A_{ja}\hat{q}_{a}$. That is, the encoding and the inverse of the encoding circuit converts the noise $\boldsymbol{\xi_{p}}$ into $\boldsymbol{z_{p}} \equiv \boldsymbol{A}^{T}\cdot\boldsymbol{\xi_{p}} = (z_{p}^{(1)},\cdots,z_{p}^{(n)})^{T}$.   

Since the last $n-m$ modes are initialized to the position eigenstates $|\hat{q}_{m+1}=0\rangle,\cdots,|\hat{q}_{n}=0\rangle$, the momentum shifts act trivially on these modes, i.e., $\exp[iz_{p}^{(j)}\hat{q}_{j}]|\hat{q}_{j}=0\rangle =0$ for all $j\in\lbrace m+1,\cdots,n \rbrace$. In other words, the shifts $z_{p}^{(m+1)},\cdots,z_{p}^{(n)}$ do not matter and we only need to keep track of the shifts in the first $m$ modes
\begin{align}
\boldsymbol{\bar{z}_{p}} &= \boldsymbol{\bar{A}}^{T}\cdot \boldsymbol{\xi_{p}} \equiv (z_{p}^{(1)},\cdots,z_{p}^{(m)})^{T},  
\end{align}
where $\boldsymbol{\bar{A}}$ is a $(n,m,1)$-triorthogonal matrix. Note that the shifts $z_{p}^{(1)},\cdots,z_{p}^{(m)}$ are generally correlated. To see this, consider the covariance matrices of the original shift vector $\boldsymbol{\xi_{p}}$ and the reshaped shift vector $\boldsymbol{\bar{z}_{p}}$. Since the original shifts are not mutually correlated, the covariance matrix $\boldsymbol{V_{\xi}}$ of the original shifts is proportional to the $n\times n$ identity matrix, i.e., $\boldsymbol{V_{\xi}} = \sigma^{2}\boldsymbol{I}_{n}$, where $\sigma^{2}$ is the variance of the probability distribution $\textrm{PDF}(\xi_{p}) = P(\xi_{p},\xi_{p})$. Also, the covariance matrix $\boldsymbol{V_{\bar{z}}}$ of the reshaped shifts is given by 
\begin{align}
\boldsymbol{V_{\bar{z}}} &= \boldsymbol{\bar{A}^{T}} \cdot \boldsymbol{V_{\xi}} \cdot \boldsymbol{\bar{A}} = \sigma^{2}\boldsymbol{\bar{A}^{T}} \cdot \boldsymbol{\bar{A}}  , 
\end{align} 
and is generally not proportional to the $m\times m$ identity matrix. 

The goal of the remaining distillation procedure is to extract information about the shift in the first mode $z_{p}^{(1)}$ based on the extracted shifts in the next $m-1$ modes $z_{p}^{(2)},\cdots,z_{p}^{(m)}$. Consider the estimate of the following form: 
\begin{align}
\tilde{z}_{p}^{(1)}  = \sum_{j=2}^{m}c_{j}z_{p}^{(j)}. 
\end{align}  
The constants $c_{j}$ should be chosen such that the variance of the corrected shift $z_{p}^{(1)} - \tilde{z}_{p}^{(1)}$ is minimized. Note that  
\begin{align}
\textrm{Var}(z_{p}^{(1)} - \tilde{z}_{p}^{(1)}) &= \textrm{Cov}(z_{p}^{(1)},z_{p}^{(1)}) - \sum_{a=2}^{m}c_{a}\textrm{Cov}(z_{p}^{(1)},z_{p}^{(a)}) 
\nonumber\\
&\quad - \sum_{a=2}^{m}c_{a}\textrm{Cov}(z_{p}^{(a)},z_{p}^{(1)}) + \sum_{a,b=2}^{m} c_{a}c_{b}\textrm{Cov}(z_{p}^{(a)},z_{p}^{(b)})  
\nonumber\\
&= \boldsymbol{( V_{\bar{z}} )_{\textrm{ul}}}  - \boldsymbol{c}^{T}\cdot  \boldsymbol{( V_{\bar{z}} )_{\textrm{ll}}} -  \boldsymbol{( V_{\bar{z}} )_{\textrm{ur}}}  \cdot \boldsymbol{c} + \boldsymbol{c}^{T}\cdot  \boldsymbol{( V_{\bar{z}} )_{\textrm{lr}}} \cdot   \boldsymbol{c} ,  
\end{align}
where $\boldsymbol{c} \equiv (c_{2},\cdots,c_{m})^{T}$, 
\begin{align}
\boldsymbol{V_{\bar{z}}} &= \begin{bmatrix}
\boldsymbol{( V_{\bar{z}} )_{\textrm{ul}}} & \boldsymbol{( V_{\bar{z}} )_{\textrm{ur}}}\\
\boldsymbol{( V_{\bar{z}} )_{\textrm{ll}}} & \boldsymbol{( V_{\bar{z}} )_{\textrm{lr}}}
\end{bmatrix}, 
\end{align}
and $ \boldsymbol{( V_{\bar{z}} )_{\textrm{ur}}} =  \boldsymbol{( V_{\bar{z}} )}_{\textrm{ll}}^{T}$. Therefore, the optimal constant vector $\boldsymbol{c}^{\star}$ is given by 
\begin{align}
\boldsymbol{c}^{\star} &= \boldsymbol{( V_{\bar{z}} )^{-1}_{\textrm{lr}}} \cdot \boldsymbol{( V_{\bar{z}} )_{\textrm{ll}}} , 
\end{align}
and the optimal output variance $\Sigma^{2}\equiv \textrm{Var}(z_{p}^{(1)} - \tilde{z}_{p}^{(1)})$ is given by
\begin{align}
\Sigma^{2} &= \boldsymbol{( V_{\bar{z}} )_{\textrm{ul}}} - \boldsymbol{( V_{\bar{z}} )^{T}_{\textrm{ll}}} \cdot  \boldsymbol{( V_{\bar{z}} )^{-1}_{\textrm{lr}}} \cdot \boldsymbol{( V_{\bar{z}} )_{\textrm{ll}}}. 
\end{align}

Recall the $(15,5,1)$-triorthogonal matrix $\boldsymbol{\bar{A}}_{\textrm{RM}}$ in Eq.\ \eqref{eq:triorthogonal matrix RM}. The covariance matrix $\boldsymbol{ V_{\bar{z}} }$ is given by 
\begin{align}
\boldsymbol{ V_{\bar{z}} } &= \sigma^{2} \boldsymbol{\bar{A}}^{T}_{\textrm{RM}} \boldsymbol{\bar{A}}_{\textrm{RM}} = \sigma^{2}\begin{bmatrix}
15&8&8&8&8\\
8&8&4&4&4\\
8&4&8&4&4\\
8&4&4&8&4\\
8&4&4&4&8
\end{bmatrix}, 
\end{align}
and thus
\begin{align}
\boldsymbol{( V_{\bar{z}} )_{\textrm{ul}}} = 15\sigma^{2}, \quad \boldsymbol{( V_{\bar{z}} )_{\textrm{ll}}} =\sigma^{2} \begin{bmatrix}
8\\
8\\
8\\
8
\end{bmatrix}, \quad \boldsymbol{( V_{\bar{z}} )_{\textrm{lr}}} =  \sigma^{2} \begin{bmatrix}
8&4&4&4\\
4&8&4&4\\
4&4&8&4\\
4&4&4&8
\end{bmatrix}. 
\end{align}
Unfortunately, the optimal output variance $\Sigma^{2}$ is given by 
\begin{align}
\Sigma^{2} = 15\sigma^{2} -\sigma^{2} \begin{bmatrix}
8&8&8&8
\end{bmatrix} \begin{bmatrix}
8&4&4&4\\
4&8&4&4\\
4&4&8&4\\
4&4&4&8
\end{bmatrix}^{-1}   \begin{bmatrix}
8\\
8\\
8\\
8
\end{bmatrix} = \frac{11}{5}\sigma^{2}
\end{align}
which is strictly larger than the input variance $\sigma^{2}$. That is, the Reed-Muller code that works well for multi-qubit magic state distillation does not work for cubic phase state distillation. Therefore, we should search for another $(n,m,1)$-triorthogonal matrix $\boldsymbol{\bar{A}}$ that can indeed yield $\Sigma^{2} < \sigma^{2}$.  

\subsection{Partial no-go result on cubic phase state distillation}
\label{subsection:Partial no-go result on cubic phase state distillation}

Here, we show that we cannot have a $(n,m,1)$-triorthogonal matrix $\boldsymbol{\bar{A}}$ that yields $\Sigma^{2} < \sigma^{2}$. Suppose that $\boldsymbol{\bar{A}}$ is a $(n,m,1)$-triorthogonal matrix. Consider the transformation 
\begin{align}
\boldsymbol{\bar{A}} \rightarrow \boldsymbol{\bar{A}'} \equiv \boldsymbol{\bar{A}} \begin{bmatrix}
1 & 0 \\
\boldsymbol{\lambda} & \boldsymbol{\Lambda}
\end{bmatrix} , \label{eq:transformation between triorthogonal matrices}
\end{align}
where $\boldsymbol{\lambda}$ is an $(m-1)\times 1$ column vector and $\boldsymbol{\Lambda}$ is an $(m-1)\times (m-1)$ invertible matrix. One can see that the transformed matrix $\boldsymbol{\bar{A}'}$ is also a $(n,m,1)$-triorthogonal matrix if $\boldsymbol{\bar{A}}$ is. Now, we show that the new triorthogonal matrix $\boldsymbol{\bar{A}'}$ performs identically to the old one $\boldsymbol{\bar{A}}$. Recall that $\boldsymbol{V_{\bar{z}}} \equiv \sigma^{2}\boldsymbol{\bar{A}^{T}} \boldsymbol{\bar{A}}$ and consider 
\begin{align}
\boldsymbol{V'_{\bar{z}}} &\equiv \sigma^{2} \boldsymbol{\bar{A}'^{T}} \boldsymbol{\bar{A}'} 
\nonumber\\
&= \begin{bmatrix}
1 & \boldsymbol{\lambda^{T}} \\
0 & \boldsymbol{\Lambda^{T}} 
\end{bmatrix} \boldsymbol{V_{\bar{z}}}  \begin{bmatrix}
1 & 0 \\
\boldsymbol{\lambda} & \boldsymbol{\Lambda}
\end{bmatrix} = \begin{bmatrix}
1 & \boldsymbol{\lambda^{T}} \\
0 & \boldsymbol{\Lambda^{T}} 
\end{bmatrix}
\begin{bmatrix}
\boldsymbol{( V_{\bar{z}} )_{\textrm{ul}}} & \boldsymbol{( V_{\bar{z}} )_{\textrm{ur}}}\\
\boldsymbol{( V_{\bar{z}} )_{\textrm{ll}}} & \boldsymbol{( V_{\bar{z}} )_{\textrm{lr}}}
\end{bmatrix}
  \begin{bmatrix}
1 & 0 \\
\boldsymbol{\lambda} & \boldsymbol{\Lambda}
\end{bmatrix} 
\nonumber\\
&= \begin{bmatrix}
\boldsymbol{( V_{\bar{z}} )_{\textrm{ul}}} + 2\boldsymbol{( V_{\bar{z}} )_{\textrm{ll}}^{T}} \boldsymbol{\lambda} + \boldsymbol{\lambda^{T}} \boldsymbol{( V_{\bar{z}} )_{\textrm{lr}}} \boldsymbol{\lambda} & \big{(} \boldsymbol{( V_{\bar{z}} )_{\textrm{ll}}^{T}} + \boldsymbol{\lambda^{T}}\boldsymbol{( V_{\bar{z}} )_{\textrm{lr}}^{T}}  \big{)}\boldsymbol{\Lambda} \\
\boldsymbol{\Lambda^{T}}\big{(} \boldsymbol{( V_{\bar{z}} )_{\textrm{ll}}} + \boldsymbol{( V_{\bar{z}} )_{\textrm{lr}}} \boldsymbol{\lambda} \big{)}& \boldsymbol{\Lambda^{T}} \boldsymbol{( V_{\bar{z}} )_{\textrm{lr}}} \boldsymbol{\Lambda} 
\end{bmatrix}. 
\end{align}
We used $\boldsymbol{( V_{\bar{z}} )_{\textrm{ur}}} = \boldsymbol{( V_{\bar{z}} )_{\textrm{ll}}^{T}}$ to get the last equality. The output variance $\Sigma'^{2}$ associated with the covariance matrix $\boldsymbol{V'_{\bar{z}}}$ is then given by 
\begin{align}
\Sigma'^{2} &= \boldsymbol{( V_{\bar{z}} )_{\textrm{ul}}} + 2\boldsymbol{( V_{\bar{z}} )_{\textrm{ll}}^{T}} \boldsymbol{\lambda} + \boldsymbol{\lambda^{T}} \boldsymbol{( V_{\bar{z}} )_{\textrm{lr}}} \boldsymbol{\lambda} 
\nonumber\\
&\quad - \Big{(} \boldsymbol{( V_{\bar{z}} )_{\textrm{ll}}^{T}} + \boldsymbol{\lambda^{T}}\boldsymbol{( V_{\bar{z}} )_{\textrm{lr}}^{T}}  \Big{)}\boldsymbol{\Lambda} \Big{(} \boldsymbol{\Lambda^{T}} \boldsymbol{( V_{\bar{z}} )_{\textrm{lr}}} \boldsymbol{\Lambda}  \Big{)}^{-1} \boldsymbol{\Lambda^{T}}\Big{(} \boldsymbol{( V_{\bar{z}} )_{\textrm{ll}}} + \boldsymbol{( V_{\bar{z}} )_{\textrm{lr}}} \boldsymbol{\lambda} \Big{)}
\nonumber\\
&= \boldsymbol{( V_{\bar{z}} )_{\textrm{ul}}} - \boldsymbol{( V_{\bar{z}} )_{\textrm{ll}}^{T}} \boldsymbol{( V_{\bar{z}} )_{\textrm{lr}}^{-1}} \boldsymbol{( V_{\bar{z}} )_{\textrm{ll}} } 
\nonumber\\
&= \Sigma^{2},  
\end{align}
which is the same as the one associated with the original covariance matrix $\boldsymbol{V_{\bar{z}}}$. That is, the new triorthogonal matrix $\boldsymbol{\bar{A}'}$ performs identically to the old one $\boldsymbol{\bar{A}}$. 

Since the performance is invariant under the transformation given in Eq.\ \eqref{eq:transformation between triorthogonal matrices}, we can restrict ourselves to $(n,m,1)$-triorthogonal matrices of the form 
\begin{align}
\boldsymbol{\bar{A}} &= \begin{bmatrix}
\uparrow & & \uparrow \\
\vec{v}_{1} & \cdots & \vec{v}_{m} \\
\downarrow & & \downarrow 
\end{bmatrix}, \,\,\, \textrm{where} \,\,\, \vec{v}_{1}^{T} \vec{v}_{a} =0 \,\,\,\textrm{for all}\,\,\, a\in \lbrace 2,\cdots, m\rbrace, 
\end{align}
without loss of generality. (One can always transform a $(n,m,1)$-triorthogonal matrix into the above form by setting $\boldsymbol{\lambda} = \boldsymbol{( V_{\bar{z}} )_{\textrm{ll}}^{-1}} \boldsymbol{( V_{\bar{z}} )_{\textrm{ll}}}$ to eliminate the off-diagonal blocks in $\boldsymbol{V'_{\bar{z}}}$). In this case, the output variance is simply given by
\begin{align}
\Sigma^{2} = \sigma^{2} \vec{v}_{1}^{T} \vec{v}_{1} = \sigma^{2} \sum_{a=1}^{n} (\vec{v}_{1})_{a}^{2}. 
\end{align}
Meanwhile, due to the $(n,m,1)$-triorthogonality, the vector $\vec{v}_{1}$ should satisfy 
\begin{align}
\sum_{a=1}^{n} (\vec{v}_{1})_{a}^{3 }=1. 
\end{align} 
One can then prove that any vector $\vec{v}_{1}$ that satisfies $\sum_{a=1}^{n} (\vec{v}_{1})_{a}^{3 }=1$ should have $\sum_{a=1}^{n} (\vec{v}_{1})_{a}^{2} \ge 1$. This implies that $\Sigma^{2} \ge \sigma^{2}$ for any $(n,m,1)$-triorthogonal matrix and thus there cannot be a $(n,m,1)$-triorthogonal matrix that yields $\Sigma^{2} < \sigma^{2}$. 

\subsection{Magic variance}

Here, we consider the magic variance of Gaussian states. Recall that we can convert any state $\hat{\rho}$ into a magic-diagonal state $\mathcal{T}(\hat{\rho}) =  \int_{-\infty}^{\infty} dp P(p,p) |\gamma,p\rangle\langle \gamma,p|$, where $P(p,p)$ is given by 
\begin{align}
P(p,p) &= \langle \gamma,p| \hat{\rho} |\gamma,p\rangle. 
\end{align}  
We define the magic mean and magic variance of a state $\hat{\rho}$ as the mean and variance of the magic distribution $P(p,p)$. That is, 
\begin{align}
\textrm{Magic mean} &\equiv \int_{-\infty}^{\infty} dp \cdot p  P(p.p) 
\nonumber\\
&= \int_{-\infty}^{\infty} dp \cdot p  \mathrm{Tr}\Big{[} \hat{\rho}|\gamma,p\rangle\langle \gamma,p| \Big{]}  
\nonumber\\
&= \int_{-\infty}^{\infty} dp \mathrm{Tr}\Big{[} \hat{\rho}e^{i\gamma \hat{q}^{3}} p |\hat{p}=p\rangle \langle \hat{p}=p| e^{-i\gamma\hat{q}^{3} } \Big{]} 
\nonumber\\
&=  \mathrm{Tr}\Big{[} \hat{\rho}e^{i\gamma \hat{q}^{3}} \hat{p} e^{-i\gamma\hat{q}^{3} } \Big{]} 
\nonumber\\
&= \langle \hat{p}-3\gamma\hat{q}^{2} \rangle , 
\end{align}
and 
\begin{align}
\textrm{Magic variance} &\equiv  \int_{-\infty}^{\infty} dp \cdot p^{2}  P(p.p) - (\textrm{magic mean})^{2} =  \big{\langle} (\hat{p}-3\gamma\hat{q}^{2})^{2} \big{\rangle} - \big{\langle} \hat{p}-3\gamma\hat{q}^{2} \big{\rangle}^{2}. 
\end{align}
Apparently, displaced cubic phase states $|\gamma,p\rangle$ have vanishing magic variance. Now, we consider the minimum possible magic variance of single mode Gaussian states. To do so, we need to compute 
\begin{align}
\langle \hat{p}\rangle, \,\,\,  \langle \hat{q}^{2} \rangle,  \,\,\,  \langle \hat{q}^{2}\hat{p}\rangle = \langle  \hat{p}\hat{q}^{2} \rangle^{*}, \,\,\, \langle \hat{q}^{4}\rangle . 
\end{align}

Note that a Gaussian state is fully characterized by its first two moments $\boldsymbol{\bar{x}}$ and $\boldsymbol{V}$. Specifically, the characteristic function of a Gaussian state is given by  
\begin{align}
\chi( \boldsymbol{\xi} ) \equiv \mathrm{Tr}\Big{[} \hat{\rho} \exp[i\boldsymbol{\hat{x}^{T}} \boldsymbol{\Omega} \boldsymbol{\xi} ] \Big{]} = \exp\Big{[} -\frac{1}{2}\boldsymbol{\xi^{T}} (\boldsymbol{\Omega V \Omega^{T} }) \boldsymbol{\xi} -i(\boldsymbol{\Omega\bar{x}})^{T} \boldsymbol{\xi}  \Big{]}. 
\end{align}
We now make these expressions more explicit for single-mode Gaussian states. Note that 
\begin{align}
\exp[i\boldsymbol{\hat{x}^{T}} \boldsymbol{\Omega} \boldsymbol{\xi} ] &= e^{ i \xi_{p}\hat{q} -i\xi_{q}\hat{p} } = e^{ i \xi_{p}\hat{q}} e^{-i\xi_{q}\hat{p}} e^{-\frac{i}{2}\xi_{q}\xi_{p} }, 
\end{align}
and 
\begin{align}
-\frac{1}{2}\boldsymbol{\xi^{T}} (\boldsymbol{\Omega V \Omega^{T} }) \boldsymbol{\xi} -i(\boldsymbol{\Omega\bar{x}})^{T} \boldsymbol{\xi}  &= -\frac{1}{2}\Big{[} V_{pp}\xi_{q}^{2} - 2V_{qp} \xi_{q}\xi_{p} + V_{qq}\xi_{p}^{2} \Big{]} -i\bar{p} \xi_{q} + i\bar{q}\xi_{p}. 
\end{align}
Thus, we have
\begin{alignat}{2}
\frac{\partial \chi}{\partial \xi_{q}} \Big{|}_{\xi_{q},\xi_{p}=0} &= \langle -i\hat{p} \rangle  = -i\bar{p}  &&\rightarrow \langle \hat{p} \rangle = \bar{p}, 
\nonumber\\
\frac{\partial \chi}{\partial \xi_{p}} \Big{|}_{\xi_{q},\xi_{p}=0} &= \langle i\hat{q} \rangle  = i\bar{q} &&\rightarrow \langle \hat{q} \rangle = \bar{q}, 
\nonumber\\
\frac{\partial^{2} \chi}{\partial \xi_{q}^{2} } \Big{|}_{\xi_{q},\xi_{p}=0} &= \langle -\hat{p}^{2} \rangle = -V_{pp} - \bar{p}^{2} \quad && \rightarrow \langle \hat{p}^{2} \rangle = V_{pp}+\bar{p}^{2}, 
\nonumber\\
\frac{\partial^{2} \chi}{\partial \xi_{q} \partial \xi_{p} } \Big{|}_{\xi_{q},\xi_{p}=0} &= \langle \hat{q}\hat{p} \rangle - \frac{i}{2} = V_{qp} + \bar{q}\bar{p}  \quad && \rightarrow \langle \hat{q}\hat{p} \rangle = V_{qp}+\bar{q}\bar{p} + \frac{i}{2},  
\nonumber\\
\frac{\partial^{2} \chi}{\partial \xi_{p}^{2} } \Big{|}_{\xi_{q},\xi_{p}=0} &= \langle -\hat{q}^{2} \rangle = -V_{qq} - \bar{q}^{2} \quad && \rightarrow \langle \hat{q}^{2} \rangle = V_{qq}+\bar{q}^{2}, 
\end{alignat}
and more importantly, 
\begin{alignat}{2}
\frac{\partial^{3} \chi}{ \partial \xi_{q} \partial \xi_{p}^{2} } \Big{|}_{\xi_{q},\xi_{p}=0} &= \langle i\hat{q}^{2}\hat{p} +\frac{1}{2}\hat{q}  \rangle = i\bar{p}V_{qq} +2i\bar{q}V_{qp} +i\bar{q}^{2}\bar{p}  && \rightarrow \langle \hat{q}^{2}\hat{p} \rangle = V_{qq}\bar{p} + 2V_{qp}\bar{q} + \bar{q}^{2}\bar{p} +\frac{i}{2} \bar{q} , 
\nonumber\\
\frac{\partial^{4} \chi}{ \partial \xi_{p}^{4} } \Big{|}_{\xi_{q},\xi_{p}=0} &= \langle \hat{q}^{4}  \rangle = 3V_{qq}^{2} + 6V_{qq}\bar{q}^{2} + \bar{q}^{4}  && \rightarrow \langle \hat{q}^{4} \rangle =   3V_{qq}^{2} + 6V_{qq}\bar{q}^{2} + \bar{q}^{4}. 
\end{alignat}
Equipped with all these, let us now compute the magic variance for single-mode Gaussian states: 
\begin{align}
\textrm{Magic variance} &= \big{\langle} \hat{p}^{2} -3\gamma\hat{q}^{2}\hat{p} -3\gamma\hat{p}\hat{q}^{2} + 9\gamma^{2}\hat{q}^{4} \big{\rangle} - \big{(} \langle \hat{p}\rangle - 3\gamma\langle \hat{q}^{2}\rangle \big{)}^{2} 
\nonumber\\
&= V_{pp}+\bar{p}^{2} -6\gamma ( V_{qq}\bar{p}+2V_{qp}\bar{q} + \bar{q}^{2}\bar{p}  ) + 9\gamma^{2} (3V_{qq}^{2}+6V_{qq}\bar{q}^{2} + \bar{q}^{4})
\nonumber\\
&\quad - \bar{p}^{2} + 6\gamma(V_{qq}+\bar{q}^{2})\bar{p} - 9\gamma^{2} (V_{qq}^{2}+2V_{qq}\bar{q}^{2} + \bar{q}^{4}) 
\nonumber\\
&= V_{pp} + 18 \gamma^{2}V_{qq}^{2} + 36\gamma^{2}V_{qq} \Big{[} \Big{(}\bar{q} - \frac{ V_{qp}}{6\gamma V_{qq}}  \Big{)}^{2} - \frac{V_{qp}^{2}}{36\gamma^{2} V_{qq}^{2}}  \Big{]}
\nonumber\\
&\ge V_{pp} - \frac{V_{qp}^{2}}{V_{qq}} + 18\gamma^{2}V_{qq}^{2}. 
\end{align}
Note that Heisenberg uncertainty principle implies $V_{qq}V_{pp} \ge V_{qp}^{2} + \frac{1}{4}$. Therefore, we have 
\begin{align}
\textrm{Magic variance} &\ge \frac{1}{4V_{qq}} + 18\gamma^{2}V_{qq}^{2}  \ge 3 \Big{(} \frac{1}{8V_{qq}} \frac{1}{8V_{qq}} 18\gamma^{2}V_{qq}^{2} \Big{)}^{\frac{1}{3}} = \frac{3}{2} \Big{(} \frac{3\gamma}{2} \Big{)}^{\frac{2}{3}}. 
\end{align}
Note that the inequality is saturated when 
\begin{align}
\bar{q} &= \frac{ V_{qp}}{6\gamma V_{qq}} ,  \quad V_{qp}=0, \quad V_{qq} = \Big{(} \frac{1}{72\gamma^{2}} \Big{)}^{\frac{1}{3}}. \label{eq:magic variance saturation}
\end{align}
Thus, it is impossible to have a cubic phase state distillation scheme that works for all $\sigma^{2} \le \sigma_{c}^{2}$ for some $\sigma_{c}^{2} \ge  \frac{2}{3}( \frac{3\gamma}{2} )^{\frac{2}{3}}$. This is because then we can start from the Gaussian state defined by the parameters in Eq.\ \eqref{eq:magic variance saturation} and reach the ideal cubic phase state which is non-Gaussian. However, it is still technically possible that there is a cubic phase distillation scheme that works if $\sigma^{2} < \frac{2}{3}( \frac{3\gamma}{2} )^{\frac{2}{3}}$. 

\section{Open questions}

The biggest open question is whether there exists a family of noisy non-Gaussian states that can be distilled indefinitely to a noiseless non-Gaussian state by using only Gaussian operations. In Ref.\ \cite{Takagi2018}, it was shown that noisy cubic phase states can be distilled to a less noisy cubic phase state. However, the distillation scheme presented in Ref.\ \cite{Takagi2018} does not reduce the noise in cubic phase states to an arbitrarily small value. The results given in Subsection \ref{subsection:Partial no-go result on cubic phase state distillation} show that the direct translation of the conventional magic state distillation schemes to bosonic systems do not work either. Thus, it will be interesting to look for (with optimism) a completely different approach towards cubic phase state distillation and see if cubic phase state distillation is possible. Conversely, it will also be interesting to see if (with pessimism) a resource-theoretic framework \cite{Chitambar2019} can be used to provide an argument against the possibility of cubic phase state distillation, or more generally distillation of any non-Gaussian resources (e.g., GKP states and Fock states) by using only Gaussian operations.







\appendix

\chapter{Gaussian states, unitaries, and channels}
\label{appendix:Gaussian states, unitaries, and channels}

\section{Bosonic modes}

Let $\mathcal{H}$ denote an infinite-dimensional Hilbert space. Quantum states of $N$ bosonic modes are in a tensor product of $N$ such Hilbert spaces $\mathcal{H}^{\otimes N}$. Each bosonic mode is associated with an annihilation and a creation operator $\hat{a}_{k}$ and $\hat{a}^{\dagger}_{k}$, satisfying the bosonic communication relation 
\begin{align}
[\hat{a}_{i},\hat{a}_{j}] = [\hat{a}^{\dagger}_{i},\hat{a}^{\dagger}_{j}] = 0, \quad [\hat{a}_{i},\hat{a}^{\dagger}_{j}] =\delta_{ij}, 
\end{align}
where $\delta_{ij}$ is the Kronecker delta function and $[\hat{A},\hat{B}]\equiv \hat{A}\hat{B}-\hat{B}\hat{A}$. The Hilbert space $\mathcal{H}$ is spanned by the eigenstates of the excitation number operator $\hat{n}\equiv \hat{a}^{\dagger}\hat{a}$. That is, $\mathcal{H}=\textrm{span}\lbrace |n\rangle \rbrace_{n=0}^{\infty}$ where $\hat{n}|n\rangle= n|n\rangle$. In the excitation number basis (or the Fock basis), the annihilation and creation operators are given by 
\begin{align}
\hat{a} = \sum_{n=1}^{\infty}\sqrt{n}|n-1\rangle\langle n|, \quad \hat{a}^{\dagger} = \sum_{n=0}^{\infty}\sqrt{n+1}|n+1\rangle\langle n|. 
\end{align}

A coherent state $|\alpha\rangle$ is an eigenstate of the annihilation operator $\hat{a}$ with a complex eigenvalue $\alpha$, i.e., $\hat{a}|\alpha\rangle = \alpha|\alpha\rangle$. In the Fock basis, $|\alpha\rangle$ is given by $|\alpha\rangle = e^{-\frac{1}{2}|\alpha|^{2}}\sum_{n=0}^{\infty}\frac{\alpha^{n}}{\sqrt{n!}}|n\rangle$. 
Note that the vacuum state $|0\rangle$ is a special case of coherent states with $\alpha=0$. The displacement operator $\hat{D}(\alpha)$ is defined as $\hat{D}(\alpha) \equiv \exp(\alpha\hat{a}^{\dagger}-\alpha^{*}\hat{a})$, and a coherent state $|\alpha\rangle$ can be understood as a displaced vacuum state:
\begin{equation}
|\alpha\rangle = \hat{D}(\alpha)|0\rangle.  \label{eq:coherent state equals to displacec vacuum}
\end{equation}

Quadrature operators are defined as 
\begin{align}
\hat{q}_{k}\equiv \frac{1}{\sqrt{2}}(\hat{a}_{k}+\hat{a}_{k}^{\dagger}), \quad \hat{p}_{k}\equiv \frac{i}{\sqrt{2}}(\hat{a}_{k}^{\dagger}-\hat{a}_{k}), 
\end{align}
and are called position and momentum operator, respectively. Note that we follow the same convention as used for GKP codes \cite{Gottesman2001,Albert2018} which differs from Ref.\ \cite{Weedbrook2012} by a factor of $\sqrt{2}$ in the definition of $\hat{q}_{k}$ and $\hat{p}_{k}$. Define $\boldsymbol{\hat{x}}\equiv  (\hat{q}_{1},\hat{p}_{1},\cdots,\hat{q}_{N},\hat{p}_{N})^{T}$. 
Then, the bosonic commutation relation reads $[\boldsymbol{\hat{x}}_{i},\boldsymbol{\hat{x}}_{j}]=i \boldsymbol{\Omega}_{ij}$, where $\boldsymbol{\Omega}$ is defined as 
\begin{equation}
\boldsymbol{\Omega} \equiv \begin{bmatrix}
\boldsymbol{\omega} &  &  \\
& \ddots & \\
& & \boldsymbol{\omega}
\end{bmatrix} \,\,\,  \textrm{and}\,\,\, \boldsymbol{\omega} \equiv 
\begin{bmatrix}
0 & 1    \\
-1 & 0
\end{bmatrix}. 
\end{equation}
Eigenvalue spectrum of the quadrature operators are continuous, $\hat{q}|q\rangle = q|q\rangle, \quad \hat{p}|p\rangle = p|p\rangle$, where $q,p\in(-\infty,\infty)$. The position and momentum eigenstates $|q\rangle$ and $|p\rangle$ are normalized by the Dirac delta function, i.e., $\langle q|q'\rangle = \delta(q-q')$ and $\langle p|p'\rangle = \delta(p-p')$. Also, $|q\rangle$ and $|p\rangle$ are related by a Fourier transformation $|q\rangle = \frac{1}{\sqrt{2\pi}}\int_{-\infty}^{\infty}dp e^{-i q p}|p\rangle$, $|p\rangle = \frac{1}{\sqrt{2\pi}}\int_{-\infty}^{\infty}dp e^{i q p}|q\rangle$. Note that the position and momentum eigenstates are transformed via the displacement operator as follows: 
\begin{alignat}{2} 
\hat{D}(\xi_{1}/\sqrt{2})|q\rangle &= e^{-i\xi_{1}\hat{p}}|q\rangle\, &&= |q+\xi_{1}\rangle , 
\nonumber\\
\hat{D}(i\xi_{2}/\sqrt{2})|p\rangle &= e^{i\xi_{2}\hat{q}}|p\rangle &&= |p+\xi_{2}\rangle . 
\end{alignat}

\section{Gaussian states}

Let $\mathcal{L}(\mathcal{H})$ be the space of linear operators on the Hilbert space $\mathcal{H}$. A general quantum state (pure or mixed) is described by a density operator $\hat{\rho}\in\mathcal{D}(\mathcal{H})$, where $\mathcal{D}(\mathcal{H})\equiv \lbrace \hat{\rho}\in\mathcal{L}(\mathcal{H})\, |\, \hat{\rho}^{\dagger}=\hat{\rho}\succeq 0, \mathrm{Tr}[\hat{\rho}]=1\rbrace$. The expectation value of an observable $\hat{E}$ of a state $\hat{\rho}$ is given by $\langle \hat{E} \rangle = \mathrm{Tr}[\hat{\rho}\hat{E}]$. The Wigner characteristic function $\chi(\boldsymbol{\xi})$ is defined as 
\begin{align}
\chi(\boldsymbol{\xi})\equiv \mathrm{Tr}[\hat{\rho}\exp(i\boldsymbol{\hat{x}^{T}\Omega\xi})], 
\end{align}
where $\hat{\rho}\in\mathcal{D}(\mathcal{H}^{\otimes N})$ and $\boldsymbol{\xi} = (\xi_{1},\cdots,\xi_{2N})^{T}$. The Weyl operator $\exp(i\boldsymbol{\hat{x}^{T}\Omega \xi})$ is of the form of a displacement operator and satisfies the orthogonality relation
\begin{align}
\mathrm{Tr}[\exp(-i\boldsymbol{\hat{x}^{T}\Omega \xi})\exp(i\boldsymbol{\hat{x}^{T}\Omega \xi'})] = (2\pi)^{N}\delta(\boldsymbol{\xi}-\boldsymbol{\xi'}). 
\end{align}
The Wigner characteristic function $\chi(\boldsymbol{\xi})$ is in one-to-one correspondence with a state $\hat{\rho}$ and the inverse function is explicitly given by
\begin{align}
\hat{\rho} = \frac{1}{(2\pi)^{N}}\int d^{2N}\boldsymbol{\xi} \chi(\boldsymbol{\xi})  \exp(-i \boldsymbol{\hat{x}^{T}\Omega\xi}). 
\end{align}
The Wigner function $W(\boldsymbol{x})$ is the Fourier transformation of $\chi(\boldsymbol{\xi})$, i.e., 
\begin{equation}
W(\boldsymbol{x}) = \frac{1}{(2\pi)^{2N}}\int d^{2N}\boldsymbol{\xi} \chi(\boldsymbol{\xi})  \exp(-i \boldsymbol{x^{T}\Omega \xi}), \label{eq:definition of Wigner function}
\end{equation}  
where $\boldsymbol{x} = (x_{1},\cdots,x_{N})^{T}$ is the eigenvalue of the quadrature operator $\boldsymbol{\hat{x}}$. 

A quantum state $\hat{\rho}$ is called a Gaussian state if its Wigner characteristic function and Wigner function are Gaussian \cite{Weedbrook2012}:
\begin{align} 
\chi(\boldsymbol{\xi}) &= \exp\big{[} -\frac{1}{2}\boldsymbol{\xi^{T}(\Omega V \Omega^{T}) \xi} -i \boldsymbol{(\Omega \bar{x}) \xi} \big{]},  
\nonumber\\
W(\boldsymbol{x}) &= \frac{\exp\big{[} -\frac{1}{2}(\boldsymbol{x}-\boldsymbol{\bar{x}})^{T}\boldsymbol{V^{-1}}(\boldsymbol{x}-\boldsymbol{\bar{x}})\big{]}}{(2\pi)^{N}\sqrt{\textrm{det}\boldsymbol{V}}}. \label{eq:Wigner characteristic function for Gaussian states}
\end{align}
Here, $\boldsymbol{\bar{x}}$ and $\boldsymbol{V}$ are the first and second moments of the state $\hat{\rho}$: 
\begin{equation}
\boldsymbol{\bar{x}} \equiv  \langle \boldsymbol{\hat{x}} \rangle  = \mathrm{Tr}[\hat{\rho}\boldsymbol{\hat{x}}], \quad \boldsymbol{V}_{ij} \equiv \frac{1}{2}\langle \lbrace \boldsymbol{\hat{x}}_{i}-\boldsymbol{\bar{x}}_{i} , \boldsymbol{\hat{x}}_{j}-\boldsymbol{\bar{x}}_{j}  \rbrace \rangle , 
\end{equation} 
where $\lbrace \hat{A},\hat{B}\rbrace\equiv \hat{A}\hat{B}+\hat{B}\hat{A}$. Thus, a Gaussian state is fully characterized by its first two moments, i.e., $\hat{\rho}=\hat{\rho}_{G}(\boldsymbol{\bar{x}},\boldsymbol{V})$. The Heisenberg uncertainty relation reads 
\begin{align}
\boldsymbol{V}+\frac{i}{2}\boldsymbol{\Omega} \succeq 0
\end{align}
and implies $V(\hat{q}_{k})V(\hat{p}_{k})\ge \frac{1}{4}$ for all $k\in\lbrace 1,\cdots,N \rbrace$, where $V(\hat{x}_{i}) \equiv V_{ii}$. 

The vacuum state $|0\rangle\langle 0|$ is the simplest example of a one-mode Gaussian state with $\boldsymbol{\bar{x}}=\boldsymbol{0}$ and $\boldsymbol{V} = \frac{\boldsymbol{I}_{2}}{2}$, where $\boldsymbol{I}_{n}$ is defined as the $n\times n$ identity matrix. A coherent state $|\alpha\rangle\langle\alpha|$ is also a Gaussian state: $|\alpha\rangle\langle\alpha| = \hat{\rho}_{G}(\boldsymbol{\bar{x}}_{\alpha},\frac{\boldsymbol{I}_{2}}{2})$ with $\boldsymbol{\bar{x}}_{\alpha}\equiv \sqrt{2}(\alpha_{R},\alpha_{I})^{T}$ and $\alpha= \alpha_{R}+i\alpha_{I}$. Coherent states (including the vacuum state) saturate the uncertainty relation and thus have the minimum uncertainty. A thermal state is an example of a Gaussian mixed state and is given by
\begin{equation}
\hat{\tau}(\nth) \equiv \sum_{n=0}^{\infty} \frac{(\nth)^{n}}{(\nth+1)^{n+1}} |n\rangle\langle n| =\hat{\rho}_{G}\Big{(}\boldsymbol{0},\Big{(}\nth+\frac{1}{2}\Big{)}\boldsymbol{I}_{2}\Big{)} ,   \label{eq:thermal state definition and Gaussian specification} 
\end{equation} 
in the Fock basis. Here, $\nth$ is the average photon number, i.e., $\nth = \mathrm{Tr}[\hat{\rho}_{\nthtiny}\hat{n}]$. The quantum von Neumann entropy of a state $\hat{\rho}$ is defined as $H(\hat{\rho})\equiv -\mathrm{Tr}[\hat{\rho}\log_{2}\hat{\rho}]$. The entropy of a thermal state $\hat{\tau}(\nth)$ is given by 
\begin{equation}
H(\hat{\rho}_{\nthtiny}) = g(\nth), \label{eq:thermal state entropy}  
\end{equation}
where $g(x) \equiv (x +1) \log_{2}(x+1)-x\log_{2} x$. Since a thermal state is a mixed state, we have $H(\hat{\tau}(\nth))\ge 0$ where the equality holds only when the state is the vacuum, i.e., $\nth=0$.  

\section{Gaussian unitaries}

A unitary operation that maps a Gaussian state to another Gaussian state is called a Gaussian unitary operation. A Gaussian unitary operation is generated by a second-order polynomial of $\boldsymbol{\hat{a}} = (\hat{a}_{1},\cdots,\hat{a}_{N})^{T}$ and $\boldsymbol{\hat{a}^{\dagger}} = (\hat{a}_{1}^{\dagger},\cdots,\hat{a}_{N}^{\dagger})^{T}$, i.e., $\hat{U}_{G} = \exp(-i\hat{H})$ with $\hat{H} = i (\boldsymbol{\alpha^{T}\hat{a}^{\dagger}} + \boldsymbol{\hat{a}^{\dagger}F\hat{a}}+ \boldsymbol{\hat{a}^{\dagger} G \hat{a}^{\dagger} }) + \textrm{h.c.}$, where $\boldsymbol{\alpha^{T}} = (\alpha_{1},\cdots,\alpha_{N})$ and $\boldsymbol{F},\boldsymbol{G}$ are $N\times N$ complex matrices. In the Heisenberg picture, the annihilation operator $\boldsymbol{\hat{a}}$ is transformed into $\hat{U}_{G}^{\dagger} \boldsymbol{\hat{a}} \hat{U}_{G}  = \boldsymbol{A \hat{a}} + \boldsymbol{B \hat{a}^{\dagger}} +\boldsymbol{\alpha}$, where $N\times N$ complex matrices $\boldsymbol{A},\boldsymbol{B}$ (determined by $\boldsymbol{F},\boldsymbol{G}$) satisfy $\boldsymbol{AB^{T}} = \boldsymbol{BA^{T}}$ and $\boldsymbol{A A^{\dagger}}=\boldsymbol{B B^{\dagger}}+\boldsymbol{I}_{N}$. In terms of the quadrature operators, the transformation reads 
\begin{align}
\boldsymbol{x}\rightarrow \hat{U}_{G}^{\dagger}\boldsymbol{x}\hat{U}_{G} = \boldsymbol{Sx}+ \boldsymbol{d}, 
\end{align}
where $\boldsymbol{d} = (d_{1},\cdots,d_{2N})^{T}=\sqrt{2}(\alpha_{1}^{R},\alpha_{1}^{I},\cdots,\alpha_{N}^{R},\alpha_{N}^{I})^{T}$ and the $2N\times 2N$ matrix $\boldsymbol{S}$ is symplectic: 
\begin{equation}
\boldsymbol{S \Omega S^{T}}  =\boldsymbol{\Omega}. \label{eq:symplectic matrix defining property}
\end{equation}
A Gaussian unitary operation is thus fully characterized by $\boldsymbol{S},\boldsymbol{d}$, and under $\hat{U}_{\boldsymbol{S},\boldsymbol{d}}$ the first two moments of a Gaussian state $\hat{\rho}_{G}(\boldsymbol{\bar{x}},\boldsymbol{V})$ are transformed as 
\begin{equation}
\boldsymbol{\bar{x}} \rightarrow \boldsymbol{S\bar{x}}+\boldsymbol{d}, \quad \boldsymbol{V}\rightarrow \boldsymbol{S V S^{T}}.  \label{eq:transformation of moments of quadrature by Gaussian unitary}
\end{equation}

The displacement operator $\hat{D}(\alpha)$ is a one-mode Gaussian unitary operation with $\boldsymbol{\alpha} = \alpha$ and $\boldsymbol{F}=\boldsymbol{G}=0$, yielding $\boldsymbol{A}=1,\boldsymbol{B}=0$ and $\boldsymbol{S}=\boldsymbol{I}_{2}$, $\boldsymbol{d}=\sqrt{2}(\alpha_{R},\alpha_{I})^{T}$. The squeezing operator $\hat{S}(r)\equiv \exp(\frac{r}{2}(\hat{a}^{2}-\hat{a}^{\dagger 2}) )$ is a one-mode Gaussian unitary operation and transforms quadrature operators by $\hat{q}\rightarrow e^{-r}\hat{q}$ and $\hat{p}\rightarrow e^{r}\hat{p}$, i.e., $\boldsymbol{S}=\textrm{diag}(e^{-r},e^{r})$. A quadrature eigenstate can be understood as an infinitely squeezed state: For example, $|\hat{q}=0\rangle \propto \lim_{r \rightarrow+\infty} \hat{S}(r)|0\rangle$ and $|\hat{p}=0\rangle \propto \lim_{r \rightarrow-\infty} \hat{S}(r)|0\rangle$, where $|0\rangle$ is the vacuum state. The phase rotation operator is defined as $\hat{U}(\theta)\equiv \exp(i\theta\hat{a}^{\dagger}\hat{a})$. Under the phase rotation, quadrature operators are transformed as 
\begin{equation}
\boldsymbol{\hat{x}} \rightarrow \boldsymbol{R}(\theta)\boldsymbol{\hat{x}} \,\,\, \textrm{where} \,\,\,  \boldsymbol{R}(\theta) \equiv \begin{bmatrix}
\cos\theta & -\sin\theta    \\
\sin\theta & \cos\theta
\end{bmatrix}, \label{eq:symplectic transformation phase rotation}
\end{equation}
yielding, e.g., $\hat{U}(\theta)|\alpha\rangle = |\alpha e^{i\theta}\rangle$. 

The beam splitter unitary is a two-mode Gaussian unitary operation generated by a Hamiltonian of the form $\hat{H}\propto i(\hat{a}_{1}^{\dagger}\hat{a}_{2}-\hat{a}_{1}\hat{a}_{2}^{\dagger})$. The beam splitter unitary $\hat{B}(\eta)$ transforms the annihilation operators by $\hat{a}_{1}\rightarrow \sqrt{\eta}\hat{a}_{1}+\sqrt{1-\eta}\hat{a}_{2}$ and $\hat{a}_{2}\rightarrow -\sqrt{1-\eta}\hat{a}_{1}+\sqrt{\eta}\hat{a}_{2}$, where $\eta\in[0,1]$ is called the transmissivity. In terms of the quadrature operator $\boldsymbol{\hat{x}}=(\hat{q}_{1},\hat{p}_{1},\hat{q}_{2},\hat{p}_{2})^{T}$, the transformation reads
\begin{equation}
\boldsymbol{\hat{x}}\rightarrow \boldsymbol{B}(\eta) \boldsymbol{\hat{x}} \,\,\,\textrm{where}\,\,\,  \boldsymbol{B}(\eta)\equiv \begin{bmatrix}
\sqrt{\eta}\boldsymbol{I}_{2} & \sqrt{1-\eta}\boldsymbol{I}_{2}   \\
-\sqrt{1-\eta}\boldsymbol{I}_{2} & \sqrt{\eta}\boldsymbol{I}_{2}
\end{bmatrix}. \label{eq:symplectic transformation beam splitter}
\end{equation}
Another example of the two-mode Gaussian operation is the two-mode squeezing, generated by $\hat{H}\propto i(\hat{a}_{1}\hat{a}_{2}-\hat{a}_{1}^{\dagger}\hat{a}_{2}^{\dagger})$. Under the two-mode squeezing $\hat{S}_{2}(G)$, the annihilation operators are transformed as $\hat{a}_{1}\rightarrow \sqrt{G}\hat{a}_{1}+\sqrt{G-1}\hat{a}_{2}^{\dagger}$ and $\hat{a}_{2}\rightarrow \sqrt{G-1}\hat{a}_{1}^{\dagger}+\sqrt{G}\hat{a}_{2}$, where $G\ge 1$ is the gain of the two-mode squeezing operation. Under the two-mode squeezing $\hat{S}_{2}(G)$, the quadrature operators are transformed as 
\begin{equation}
\boldsymbol{\hat{x}}\rightarrow \boldsymbol{S_{2}}(G) \boldsymbol{\hat{x}} \,\,\,\textrm{where}\,\,\,  \boldsymbol{S_{2}}(G)\equiv \begin{bmatrix}
\sqrt{G}\boldsymbol{I}_{2} & \sqrt{G-1}\boldsymbol{Z}_{2}   \\
\sqrt{G-1}\boldsymbol{Z}_{2} & \sqrt{G}\boldsymbol{I}_{2}
\end{bmatrix},  \label{eq:symplectic transformation two mode squeezer}
\end{equation}
where $\boldsymbol{Z}_{2}\equiv \textrm{diag}(1,-1)$. 

\section{Gaussian channels}
\label{section:Gaussian channels}

A quantum channel $\mathcal{N} : \mathcal{D}(\mathcal{H}^{\otimes N})\rightarrow \mathcal{D}(\mathcal{H}^{\otimes N})$ maps a quantum state $\hat{\rho}\in\mathcal{D}(\mathcal{H}^{\otimes N})$ to another state in $\mathcal{D}(\mathcal{H}^{\otimes N})$ via a completely positive and trace-preserving (CPTP) map \cite{Choi1975}. Gaussian channels map a Gaussian state to another Gaussian state and can be simulated by 
\begin{align}
\mathcal{N}(\hat{\rho}) = \mathrm{Tr}_{E}[\hat{U}_{G}(\hat{\rho}\otimes \hat{\rho}_{E})\hat{U}_{G}^{\dagger}]. 
\end{align}
Here, $\hat{U}_{G}$ is a Gaussian unitary operation on the system plus the environmental modes, $\hat{\rho}_{E}$ is a Gaussian state, and $\mathrm{Tr}_{E}$ is the partial trace with respect to the environmental mode. Let $\boldsymbol{\hat{X}^{T}}=(\boldsymbol{\hat{x}^{T}},\boldsymbol{\hat{y}^{T}})$ be a collection of quadrature operators of the system mode $\boldsymbol{\hat{x}}$ and the environmental mode $\boldsymbol{\hat{y}}$, and assume that the initial system and environmental states are given by $\hat{\rho}_{G}(\boldsymbol{\bar{x}},\boldsymbol{V_{x}})$ and $\hat{\rho}_{G}(\boldsymbol{\bar{y}},\boldsymbol{V_{y}})$, respectively. Here, $\boldsymbol{\bar{x}}$, $\boldsymbol{\bar{y}}$ are the first moments and $\boldsymbol{V_{x}}$, $\boldsymbol{V_{y}}$ are the second moments of the system and the environment. If the Gaussian unitary operation $\hat{U}_{G}$ acting on the joint system is characterized by
\begin{equation}
\boldsymbol{S} = \begin{bmatrix}
\boldsymbol{S_{xx}} & \boldsymbol{S_{xy}}   \\
\boldsymbol{S_{yx}} & \boldsymbol{S_{yy}}
\end{bmatrix} 
\,\,\, \textrm{and} \,\,\,
\boldsymbol{D}= \begin{bmatrix}
\boldsymbol{d_{x}}    \\
\boldsymbol{d_{y}}
\end{bmatrix},  \label{eq:Gaussian unitary characterization for Gaussian channels}
\end{equation}  
the first two moments of the system mode are transformed as $\boldsymbol{\bar{x}}\rightarrow \boldsymbol{S_{xx}}\boldsymbol{\bar{x}}+\boldsymbol{S_{xy}}\boldsymbol{\bar{y}}+\boldsymbol{d_{x}}$ and $\boldsymbol{V_{x}}\rightarrow \boldsymbol{S_{xx}V_{x}S^{T}_{xx}}+\boldsymbol{S_{xy}V_{y}S^{T}_{xy}}$, as can be derived by specializing Eq.\ \eqref{eq:transformation of moments of quadrature by Gaussian unitary} to Eq.\ \eqref{eq:Gaussian unitary characterization for Gaussian channels}. After tracing out the environment, the resulting effective Gaussian channel for the system is characterized by 
\begin{equation}
\boldsymbol{\bar{x}}\rightarrow \boldsymbol{T\bar{x}} + \boldsymbol{d}, \quad \boldsymbol{V_{x}}\rightarrow \boldsymbol{TV_{x}T^{T}}+\boldsymbol{N}, 
\end{equation}
where 
\begin{align}
\boldsymbol{T}&=\boldsymbol{S_{xx}}, 
\nonumber\\
\boldsymbol{N}&=\boldsymbol{S_{xy}V_{y}S^{T}_{xy}}, 
\nonumber\\
\boldsymbol{d}&=\boldsymbol{S_{xy}\bar{y}}+\boldsymbol{d_{x}}. \label{eq:characterization of a Gaussian channel from the underlying Gaussian unitary}
\end{align}
Thus in general, a Gaussian channel is fully characterized by $\boldsymbol{T}$, $\boldsymbol{N}$, and $\boldsymbol{d}$. Below, we introduce three important examples of Gaussian channels, namely, Gaussian thermal-loss channels, noisy amplification channels, and additive Gaussian noise channels.  

\begin{definition}[Gaussian thermal-loss channels]
Let $\hat{B}(\eta)$ be a beam splitter unitary with transmissivity $\eta\in[0,1]$, acting on the modes $1$ and $2$. Then, the Gaussian thermal-loss channel is defined as 
\begin{align}
\mathcal{N}[\eta,\nth](\hat{\rho}_{1}) \equiv \mathrm{Tr}_{2}[\hat{B}(\eta)(\hat{\rho}_{1}\otimes \hat{\tau}(\nth))\hat{B}^{\dagger}(\eta)], 
\end{align}
where $\mathrm{Tr}_{2}$ is the partial trace with respect to the mode $2$ which is initially in a thermal state $\hat{\tau}(\nth)$ with an average photon number $\nth$. The Gaussian thermal-loss channel $\mathcal{N}[\eta,\nth]$ is characterized by 
\begin{align}
\mathcal{N}[\eta,\nth] &\leftrightarrow (\boldsymbol{T},\boldsymbol{N},\boldsymbol{d}) = \Big{(}\sqrt{\eta}\boldsymbol{I}_{2},(1-\eta)\Big{(}\nth+\frac{1}{2}\Big{)}\boldsymbol{I}_{2},0\Big{)}, 
\end{align}
The above characterization can be derived by specializing Eq.\ \eqref{eq:characterization of a Gaussian channel from the underlying Gaussian unitary} to the case of beam splitter unitaries and the initial thermal environmental state. 
\label{definition:Gaussian loss channels} 
\end{definition}

\begin{definition}[Bosonic pure-loss channels]
Bosonic pure-loss channels are a special case of Gaussian thermal-loss channels with $\nth=0$. That is, the bosonic pure-loss channel with transmissivity $\eta\in[0,1]$ is defined as 
\begin{align}
\mathcal{N}[\eta,0](\hat{\rho}_{1}) \equiv \mathrm{Tr}_{2}[\hat{B}(\eta)(\hat{\rho}_{1}\otimes |0\rangle\langle 0|_{2})\hat{B}^{\dagger}(\eta)]. 
\end{align}
The bosonic pure-loss channel $ \mathcal{N}[\eta,0]$ is characterized by 
\begin{align}
\mathcal{N}[\eta,0] &\leftrightarrow (\boldsymbol{T},\boldsymbol{N},\boldsymbol{d}) = \Big{(}\sqrt{\eta}\boldsymbol{I}_{2}, \frac{1}{2} (1-\eta)\boldsymbol{I}_{2},0\Big{)}, 
\end{align}
\label{definition:Bosonic pure loss channels} 
\end{definition} 

Now, we move on to the amplification channels. 

\begin{definition}[Noisy amplification channels]
Let $\hat{S}_{2}(G)$ be a two-mode squeezing unitary operation with gain $G\ge 1$, acting on modes $1$ and $2$. Then, the noisy amplification channel is defined as 
\begin{align}
\mathcal{A}[G,\nth](\hat{\rho}_{1})  \equiv \mathrm{Tr}_{2}[\hat{S}_{2}(G)(\hat{\rho}_{1}\otimes \hat{\tau}(\nth) )\hat{S}_{2}^{\dagger}(G)]. 
\end{align}
The noisy amplification channel $\mathcal{A}[G,\nth]$ is characterized by 
\begin{align}
\mathcal{A}[G,\nth]&\leftrightarrow (\boldsymbol{T},\boldsymbol{N},\boldsymbol{d}) = \Big{(}\sqrt{G}\boldsymbol{I}_{2},(G-1) \Big{(}\nth+\frac{1}{2}\Big{)} \boldsymbol{I}_{2},0\Big{)}. 
\end{align}
The above characterization can be derived by specializing Eq.\ \eqref{eq:characterization of a Gaussian channel from the underlying Gaussian unitary} to the case of two-mode squeezing unitaries and the initial thermal environmental state. 
\label{definition:noisy amplification channels} 
\end{definition}

\begin{definition}[Quantum-limited amplification channels] 
Quantum-limited amplification channels are a special case of noisy amplification channels with $\nth=0$. That is, the quantum limited amplification channel with gain $G\ge 1$ is defined as 
\begin{align}
\mathcal{A}[G,0](\hat{\rho}_{1})  \equiv \mathrm{Tr}_{2}[\hat{S}_{2}(G)(\hat{\rho}_{1}\otimes |0\rangle\langle 0|_{2} )\hat{S}_{2}^{\dagger}(G)] . 
\end{align}
The quantum-limited amplification channel $\mathcal{A}[G,0]$ is characterized by
\begin{align}
\mathcal{A}[G,0]&\leftrightarrow (\boldsymbol{T},\boldsymbol{N},\boldsymbol{d}) = \Big{(}\sqrt{G}\boldsymbol{I}_{2},\frac{1}{2}(G-1)\boldsymbol{I}_{2},0\Big{)}. 
\end{align}
\label{definition:quantum limited amplification} 
\end{definition}

Note that the noise $\boldsymbol{N}=\frac{(G-1)}{2}\boldsymbol{I}_{2}$ is due to the variance of the ancillary vacuum state, transferred to the system via the two-mode squeezing operation. Since the vacuum state has the minimum variance allowed by the Heisenberg uncertainty principle, the quantum-limited amplification incurs the least noise among all linear amplification channels \cite{Heffner1962}.
  
Lastly, we introduce additive Gaussian noise channels which are also sometimes referred to as Gaussian random displacement channels. 
\begin{definition}[Additive Gaussian noise channels]
The additive Gaussian noise channel is defined as 
\begin{equation}
\mathcal{N}_{B_{2}}[\sigma^{2}](\hat{\rho}) \equiv \frac{1}{\pi\sigma^{2}} \int d^{2}\alpha e^{-\frac{|\alpha|^{2}}{\sigma^{2}}} \hat{D}(\alpha)\hat{\rho}\hat{D}^{\dagger}(\alpha),  
\end{equation}
where $\hat{D}(\alpha)$ is the displacement operator and $\sigma^{2}$ is the variance of the random shifts.
The additive Gaussian noise channels belong to the class $B_{2}$ channel \cite{Caruso2006B,Holevo2007} (hence the subscript) and is characterized by 
\begin{align}
\mathcal{N}_{B_{2}}[\sigma^{2}] \leftrightarrow (\boldsymbol{T},\boldsymbol{N},\boldsymbol{d}) = (\boldsymbol{I}_{2},\sigma^{2}\boldsymbol{I}_{2},0). 
\end{align}
\label{definition:Gaussian displacement channel}
\end{definition}

\backmatter

\nocite{apsrev41Control}
\bibliography{KN_dissertation_v5,revtex-custom}

\end{document}